\documentclass[a4paper,12pt,twoside]{report}
\usepackage[left=2cm,right=2cm,top=2cm,bottom=3cm]{geometry}

\usepackage{graphicx}
\usepackage{verbatim}
\usepackage{latexsym}
\usepackage{mathchars}
\usepackage{setspace}

\input{blocked.sty}
\input{uhead.sty}
\input{boxit.sty}
\input{icthesis.sty}

\usepackage{afterpage}

\usepackage{times}	
\usepackage{xcolor} 

\usepackage{cite}
\usepackage{array}
\usepackage{color}
\usepackage{float}
\usepackage[cmex10]{amsmath}
\usepackage{nomencl}						
\usepackage[normalem]{ulem}			
\usepackage{diagbox}						
\usepackage{slashbox}						
\usepackage{colortbl}						
\usepackage{multirow}						
\usepackage{tabularx}
\usepackage{placeins}
\usepackage{booktabs}

\let\OLDthebibliography\thebibliography
\renewcommand\thebibliography[1]{
  \OLDthebibliography{#1}
  \setlength{\parskip}{0pt}
	\setlength{\itemsep}{0pt}
}

\usepackage{subcaption}
\usepackage{graphicx}

\usepackage[cmex10]{amsmath}
\usepackage{nomencl}						
\usepackage[normalem]{ulem}			
\usepackage{diagbox}						
\usepackage{slashbox}						
\usepackage{colortbl}						
\usepackage{placeins}
\usepackage{mathrsfs}
\usepackage{mathdots}
\usepackage{mathtools}
\usepackage{amssymb}
\usepackage{arydshln}
\usepackage{booktabs}

\usepackage{wrapfig}
\usepackage{pdflscape}
\usepackage{lscape}
\usepackage{rotating}
\usepackage{epstopdf}

\usepackage{relsize}
\usepackage{float}
\usepackage{upgreek}
\usepackage{amssymb}
\usepackage{amsfonts}
\usepackage{amsthm}
\usepackage{algorithm}
\usepackage{algorithmic}
\usepackage[toc,page]{appendix}
\usepackage{enumitem}

\usepackage[breakable, theorems, skins]{tcolorbox}
\DeclareRobustCommand{\mybox}[2][gray!20]{%
\begin{tcolorbox}[   
        breakable,
        left=0pt,
        right=0pt,
        top=0pt,
        bottom=0pt,
        colback=#1,
        colframe=#1,
        width=\dimexpr\textwidth\relax, 
        enlarge left by=0mm,
        boxsep=5pt,
        arc=0pt,outer arc=0pt,
        ]
        #2
\end{tcolorbox}
}

\usepackage{url}

\renewcommand\thesubsubsection{\Alph{subsubsection}}		

\newcommand{\NN}{{\sf I\kern-0.14emN}}   
\newcommand{\ZZ}{{\sf Z\kern-0.45emZ}}   
\newcommand{\QQQ}{{\sf C\kern-0.48emQ}}   
\newcommand{\RR}{{\sf I\kern-0.14emR}}   

\newcommand{\normallinespacing}{\renewcommand{\baselinestretch}{1.5} \normalsize}

\newcommand{\syncc}{~\stackrel{\textstyle \rhd\kern-0.57em\lhd}{\scriptstyle L}~}

\newtheorem{theorem}{Theorem}
\newtheorem{lemma}{Lemma}
\newtheorem{proposition}{Proposition}
\newtheorem{corollary}{Corollary}
\newtheorem{definition}{Definition}
\newtheorem{assumption}{Assumption}

\newtheorem{remark}{Remark}
\newtheorem{example}{Example}
\newtheorem{condition}{Condition}

\newcommand{\highlight}[1]{\textcolor{black}{#1}}

\begin{document}


\title{\LARGE {\bf Shapley Value Based Multi-Agent Reinforcement Learning: Theory, Method and Its Application to Energy Network}\\
 \vspace*{6mm}
}

\author{\textbf{Jianhong Wang}}

\submitdate{\highlight{01/01} 2023}

\normallinespacing
\maketitle

\preface

\addcontentsline{toc}{chapter}{Abstract}

\begin{abstract}
    Multi-agent reinforcement learning is an area of rapid advancement in artificial intelligence and machine learning. One of the important questions to be answered is how to conduct credit assignment in a multi-agent system. There have been many schemes designed to conduct credit assignment by multi-agent reinforcement learning algorithms. Although these credit assignment schemes have been proved useful in improving the performance of multi-agent reinforcement learning, most of them are designed heuristically without a rigorous theoretic basis and therefore infeasible to understand how agents cooperate. In this thesis, we aim at investigating the foundation of credit assignment in multi-agent reinforcement learning via cooperative game theory. We first extend a game model called convex game and a payoff distribution scheme called Shapley value in cooperative game theory to Markov decision process, named as Markov convex game and Markov Shapley value respectively. We represent a global reward game as a Markov convex game under the grand coalition. As a result, Markov Shapley value can be reasonably used as a credit assignment scheme in the global reward game. Markov Shapley value possesses the following virtues: (i) efficiency; (ii) identifiability of dummy agents; (iii) reflecting the contribution and (iv) symmetry, which form the fair credit assignment. Based on Markov Shapley value, we propose two multi-agent reinforcement learning algorithms called SHAQ and SQDDPG. To address the direct approximation problem existing in SQDDPG, we also propose SMFPPO. Furthermore, we extend Markov convex game to partial observability to deal with the partially observable problems, named as partially observable Markov convex game. In this game, we propose partially observable Shapley policy iteration and partially observable Shapley value iteration which endow the capability of tackling partially observable scenarios for SHAQ, SQDDPG and SMFPPO. In application, we evaluate SQDDPG and SMFPPO on the real-world problem in energy networks. 

\end{abstract}

\cleardoublepage

\addcontentsline{toc}{chapter}{Acknowledgements}

\begin{acknowledgements}

I would like to thank my parents for supporting me in accomplishing my dream on research and assisting my life with financial and mental support. 

Also, I extremely appreciate my supervisors, Dr. Yunjie Gu, Prof. Tim C. Green, and Dr. Tae-Kyun Kim. Yunjie is an ambitious researcher, who always encouraged me to deal with fundamental problems. Without his advice, it is impossible for me to insist on continuing my research on Shapley value for multi-agent reinforcement learning. I received very valuable advice from Tim on the applications of multi-agent reinforcement learning to power networks. Tae-Kyun gives me important help in pointing out some issues which I might have missed and polishing my idea.

Additionally, I would like to thank my friends, since many interesting ideas were born from the casual communication and entertainment with them. Especially, I would like to thank Wangkun Xu, who always talked to me about snippets of thoughts. Although many of them cannot be realized, some of them were still helpful to my research and stimulated some new ideas. 

I especially thank Yuan Zhang, my collaborator and good friend, who gave me lots of help and insightful comments on my research.

Finally, I would like to thank my girl friend, Dr. Li Guo, who always supported me to chase my dream and provided me with illumination when I was in the dark.

\end{acknowledgements}

\cleardoublepage

\addcontentsline{toc}{chapter}{Statement of Originality}

\begin{statement}
I, Jianhong Wang, confirm that this work is the result of my own endeavor, and all work undertaken by other authors is appropriately referenced. 
\end{statement}

\cleardoublepage

\addcontentsline{toc}{chapter}{Copyright Declaration}

\begin{CopyrightDeclaration}
The copyright of this thesis rests with the author. Unless otherwise indicated, its contents are licensed under a Creative Commons Attribution-Non Commercial 4.0 International Licence (CC BY-NC).

Under this licence, you may copy and redistribute the material in any medium or format. You may also create and distribute modified versions of the work. This is on the condition that: you credit the author and do not use it, or any derivative works, for a commercial purpose. 

When reusing or sharing this work, ensure you make the licence terms clear to others by naming the licence and linking to the licence text. Where a work has been adapted, you should indicate that the work has been changed and describe those changes.

Please seek permission from the copyright holder for uses of this work that are not included in this licence or permitted under UK Copyright Law.
	
\end{CopyrightDeclaration}

\cleardoublepage

\addcontentsline{toc}{chapter}{List of Abbreviations}

\begin{abbreviations}
\begin{tabular}{p{2cm}lp{20cm}l}
\newcommand{\tabincell}[2]{\begin{tabular}{@{}#1@{}}#2\end{tabular}}
AI & \textbf{A}rtificial \textbf{I}ntelligence \\
ML & \textbf{M}achine \textbf{L}earning \\
MAS & \textbf{M}ulti-\textbf{A}gent \textbf{S}ystem \\
RL & \textbf{R}einforcement \textbf{L}earning \\
MARL & \textbf{M}ulti-\textbf{A}gent \textbf{R}einforcement \textbf{L}earning \\
CG & \textbf{C}onvex \textbf{G}ame \\
MCG & \textbf{M}arkov \textbf{C}onvex \textbf{G}ame \\
POMCG & \textbf{P}artially \textbf{O}bservable \textbf{M}arkov \textbf{C}onvex \textbf{G}ame \\
GRG & \textbf{G}lobal \textbf{R}eward \textbf{G}ame \\
PSRO & \textbf{P}olicy-\textbf{S}pace \textbf{R}esponse \textbf{O}racle \\
ARMA & \textbf{A}uto\textbf{R}egressive–\textbf{M}oving-\textbf{A}verage Model \\
MDP & \textbf{M}arkov \textbf{D}ecision \textbf{P}rocess \\
PI & \textbf{P}olicy \textbf{I}teration \\
VI & \textbf{V}alue \textbf{I}teration \\
TD & \textbf{T}emporal \textbf{D}ifference \\
GPI & \textbf{G}eneralized \textbf{P}olicy \textbf{I}teration \\
DPG & \textbf{D}eterministic \textbf{P}olicy \textbf{G}radient \\
DNN & \textbf{D}eep \textbf{N}eural \textbf{N}etwork \\
DQN & \textbf{D}eep \textbf{Q}-\textbf{N}etwork \\
DDPG & \textbf{D}eep \textbf{D}eterministic \textbf{P}olicy \textbf{G}radient \\
TRPO & \textbf{T}rust \textbf{R}egion \textbf{P}olicy \textbf{O}ptimization \\
PPO & \textbf{P}roximal \textbf{P}olicy \textbf{O}ptimization \\
CTDE & \textbf{C}entralized \textbf{T}raining and \textbf{D}ecentralized \textbf{E}xecution \\
IGM & \textbf{I}ndividual-\textbf{G}lobal-\textbf{M}ax \\
SV & \textbf{S}hapley \textbf{V}alue \\
MSV & \textbf{M}arkov \textbf{S}hapley \textbf{V}alue
\end{tabular}

\begin{tabular}{p{2cm}lp{20cm}l}
\newcommand{\tabincell}[2]{\begin{tabular}{@{}#1@{}}#2\end{tabular}}
PV & \textbf{P}hoto\textbf{V}oltaics \\
DSO & \textbf{D}istributed \textbf{S}ystem \textbf{O}perator \\
SC & \textbf{S}hunt \textbf{C}apacitor \\
SVR & \textbf{S}tep \textbf{V}oltage \textbf{R}egulator \\
STATCOM & \textbf{STAT}ic Synchronous \textbf{COM}pensator \\
OPF & \textbf{O}ptimal \textbf{P}ower \textbf{F}low \\
ADMM & \textbf{A}lternating \textbf{D}irection \textbf{M}ethod of \textbf{M}ultipliers \\
SVC & \textbf{S}tatic \textbf{V}ar \textbf{C}ompensator \\
SHAQ & \textbf{SHA}pley \textbf{Q}-Learning \\
SQDDPG & \textbf{S}hapley \textbf{Q}-Value \textbf{D}eep \textbf{D}eterministic \textbf{P}olicy \textbf{G}radient \\
SMFPPO & \textbf{S}hapley Value \textbf{M}odel-\textbf{F}ree \textbf{P}olicy \textbf{G}radient \\
MSQ & \textbf{M}arkov \textbf{S}hapley \textbf{Q}-value \\
SBOE & \textbf{S}hapley-\textbf{B}ellman \textbf{O}ptimality \textbf{E}quation \\
SBO & \textbf{S}hapley-\textbf{B}ellman \textbf{O}perator \\
POMDP & \textbf{P}artially \textbf{O}bservable \textbf{M}arkov \textbf{D}ecision \textbf{P}rocess \\
Dec\text{-}POMDP & \textbf{D}ecentralized \textbf{P}artially \textbf{O}bservable \textbf{M}arkov \textbf{D}ecision \textbf{P}rocess \\
POSPI & \textbf{P}artially \textbf{O}bservable \textbf{S}hapley \textbf{P}olicy \textbf{I}teration \\
POSVI & \textbf{P}artially \textbf{O}bservable \textbf{S}hapley \textbf{V}alue \textbf{I}teration
\end{tabular}

\end{abbreviations}

\cleardoublepage

\addcontentsline{toc}{chapter}{List of Symbols}
\makenomenclature
\renewcommand{\nomname}{List of Symbols}

\renewcommand{\nompreamble}{The next list describes the commonly used symbols that will appear in the thesis.}

\nomenclature{\(\mathcal{S}\)}{State space (a set of states)}
\nomenclature{\(\mathcal{A}\)}{Action space (a set of actions), indicating the joint action space in multi-agent scenarios}
\nomenclature{\(T\)}{Probability transition function}
\nomenclature{\(Pr\)}{Probability measure}
\nomenclature{\(R_{t}\)}{Reward function at time step $t$}
\nomenclature{\(G_{t}\)}{Return from time step $t$}
\nomenclature{\(\pi\)}{Stationary policy function, indicating the joint policy in multi-agent scenarios}
\nomenclature{\(\pi_{i}\)}{Stationary policy of Agent $i$}
\nomenclature{\(\pi_{\theta}\)}{Parametric stationary policy function (indicating the joint policy in multi-agent scenarios)}
\nomenclature{\(\pi_{\theta_{i}}\)}{Parametric stationary policy function of Agent $i$}
\nomenclature{\(\mathbf{S}_{t}\)}{State variable at time step $t$, where $t$ may be ignored}
\nomenclature{\(\mathbf{S}_{t}'\)}{Next state variable at time step $t$, where $t$ may be ignored}
\nomenclature{\(\mathbf{A}_{t}\)}{Action variable at time step $t$, where $t$ may be ignored}
\nomenclature{\(\mathbf{s}\)}{Arbitrary value of state}
\nomenclature{\(\mathbf{s}'\)}{Arbitrary value of next state}
\nomenclature{\(\mathbf{a}\)}{Arbitrary value of action}
\nomenclature{\(V^{\pi}\)}{State value function evaluating policy $\pi$}
\nomenclature{\(Q^{\pi}\)}{Q-value function (action value function) evaluating policy $\pi$}
\nomenclature{\(V^{*}\)}{Optimal state value function}
\nomenclature{\(Q^{*}\)}{Optimal Q-value function}
\nomenclature{\(Q^{\pi^{*}}\)}{Optimal joint Q-value function}
\nomenclature{\(\gamma\)}{Discount factor for cumulative rewards}
\nomenclature{\(\mathbb{E}[\cdot]\)}{Expectation value}
\nomenclature{\(\nabla_{\theta}\)}{Gradient of some function with respect to arbitrary parameter $\theta$}
\nomenclature{\(d^{\pi}(\mathbf{s})\)}{Stationary state distribution under policy $\pi$}
\nomenclature{\(\left(f_{i}\right)_{i=1,2,...}\)}{Tuple of functions indexed by $i$}
\nomenclature{\(\mathcal{N}\)}{Set of agents (sometimes indicating Gaussian distribution for convention)}
\nomenclature{\(\mathcal{I}\)}{Set of agents when $\mathcal{N}$ is used to indicate Gaussian distribution}
\nomenclature{\(V\)}{Coalition value function in cooperative game theory (the set of nodes in the context of power networks for convention)}
\nomenclature{\(\mathcal{CS}\)}{Coalition structure in cooperative game theory}
\nomenclature{\(\mathcal{C}\)}{Coalition in cooperative game theory}
\nomenclature{\(x_{i}\)}{Payoff distribution of Agent $i$ in cooperative game theory}
\nomenclature{\(\mathbf{x}\)}{Tuple of payoff distributions in cooperative game theory}
\nomenclature{\(\delta_{i}^{m}(\mathcal{C})\)}{Marginal contribution of Agent $i$ about coalition $\mathcal{C}$ appeared in permutation $m$, in cooperative game theory}
\nomenclature{\(Sh_{i}\)}{Shapley value of Agent $i$ in cooperative game theory}
\nomenclature{\(\vert \mathcal{M} \vert \)}{Cardinality of an arbitrary set $\mathcal{M}$}
\nomenclature{\(p_{i}\)}{Active power injection at Node $i$}
\nomenclature{\(q_{i}\)}{Reactive power injection at Node $i$}
\nomenclature{\(p_{i}^{\scriptscriptstyle PV}\)}{Active power generated by PV at Node $i$}
\nomenclature{\(q_{i}^{\scriptscriptstyle PV}\)}{Reactive power generated by PV inverter at Node $i$}
\nomenclature{\(p_{i}^{\scriptscriptstyle L}\)}{Active power consumed by load at Node $i$}
\nomenclature{\(q_{i}^{\scriptscriptstyle L}\)}{Reactive power consumed by load at Node $i$}
\nomenclature{\(v_{i}\)}{Voltage at Node $i$}
\nomenclature{\(\theta_{i}\)}{Phase at Node $i$}
\nomenclature{\(\theta_{ij}\)}{Phase difference between Node $i$ and $j$}
\nomenclature{\(s_{i}\)}{Apparent power injection at Node $i$}
\nomenclature{\(g_{ij}\)}{Conductance on branch $(i, j)$}
\nomenclature{\(b_{ij}\)}{Susceptance on branch $(i, j)$}
\nomenclature{\(z_{ij}\)}{Impedance on branch $(i, j)$}
\nomenclature{\(r_{ij}\)}{Resistance on branch $(i, j)$}
\nomenclature{\(x_{ij}\)}{Reactance on branch $(i, j)$}
\nomenclature{\(v_{0}\)}{Voltage at the slack bus (Node $0$)}
\nomenclature{\(v_{\text{ref}}\)}{Voltage at the slack bus (Node $0$)}
\nomenclature{\(\Delta v_{ij}\)}{Voltage drop between Node $i$ and $j$}
\nomenclature{\(p_{0}\)}{Active power at the slack bus (Node $0$)}
\nomenclature{\(\pi_{\scriptscriptstyle \mathcal{C}}\)}{Policy of coalition $\mathcal{C}$}
\nomenclature{\(\mathcal{A}_{\scriptscriptstyle \mathcal{C}}\)}{Coalition action set}
\nomenclature{\(V^{\pi_{\mathcal{C}}}(\mathbf{s})\)}{Coalition value function with respect to coalition policy $\pi_{\scriptscriptstyle \mathcal{C}}$}
\nomenclature{\(Q^{\pi_{\mathcal{C}}}\)}{Coalition Q-value function with respect to coalition policy $\pi_{\scriptscriptstyle \mathcal{C}}$}
\nomenclature{\(\max_{\pi_{i}} x_{i}(\mathbf{s})\)}{Payoff distribution under the optimal joint policy and state $\mathbf{s}$}
\nomenclature{\(Q^{\pi_{\mathcal{D}}^{*}}(\mathbf{s}, \mathbf{a}_{\scriptscriptstyle\mathcal{C}})\)}{Optimal coalition Q-value of $\mathcal{C}$ w.r.t. the optimal policy of sub-coalition $\mathcal{D} \ \mathlarger{\subseteq} \ \mathcal{C}$ and the suboptimal policy of sub-coalition $\mathcal{C} \backslash \mathcal{D}$}
\nomenclature{\(\Phi_{i}(\mathbf{s} \vert \mathcal{C}_{i})\)}{Marginal contribution of an agent $\mathit{i}$ under intermediate coalition $\mathcal{C}_{i}$}
\nomenclature{\(\Upphi_{i}(\mathbf{s}, a_{i} \vert \mathcal{C}_{i})\)}{Action marginal contribution of an agent $\mathit{i}$ under intermediate coalition $\mathcal{C}_{i}$}
\nomenclature{\(V^{\phi}_{i}(\mathbf{s})\)}{Markov Shapley value of Agent $i$ under state $\mathbf{s}$}
\nomenclature{\(Q^{\phi}_{i}(\mathbf{s}, a_{i})\)}{Markov Shapley Q-value of Agent $i$ under state $\mathbf{s}$ and action $a_{i}$}
\nomenclature{\(\mathbb{I}\)}{Indicator function}
\nomenclature{\(\mathcal{O}\)}{Joint observation set in POMCG}
\nomenclature{\(\mathcal{O}_{i}\)}{Observation set of Agent $i$ in POMCG}
\nomenclature{\(\Omega\)}{Observation probability function in POMCG}
\nomenclature{\(\mathbf{o}_{t}\)}{Joint observation at time step $t$ ($t$ is usually ignored) in POMCG}
\nomenclature{\(\mathbf{o}_{i, t}\)}{Observation of Agent $i$ at time step $t$ ($t$ is usually ignored) in POMCG}
\nomenclature{\(h_{i, t}\)}{Agent $i$'s action-observation history (AOH) for time step $t$ in POMCG}
\nomenclature{\(\mathbf{h}_{t}\)}{Joint history at time step $t$ in POMCG}
\nomenclature{\(b_{t}\)}{Joint belief state at time step $t$ in POMCG}
\nomenclature{\(\mathcal{B}\)}{Joint belief state space in POMCG}
\nomenclature{\(R_{b}\)}{Reward function with respect to belief state in POMCG}
\nomenclature{\(T_{b}\)}{Transition function with respect to belief state in POMCG}
\nomenclature{\(\tau(\mathbf{o}_{t+1}, \mathbf{a}_{t}, b_{t} \vert \mathcal{CS})\)}{Function of current belief state, action and next observation to represent next belief state in POMCG}

\printnomenclature

\body

\chapter{Introduction} 
\label{ch:introduction}
    Multi-agent reinforcement learning (MARL) is an area of rapid advancement in artificial intelligence (AI) and machine learning (ML) to solve many realistic decentralised control problems with cooperative structure, e.g. robotic teams for emergent rescues \cite{koes2006constraint,ramchurn2010decentralized}, traffic network control \cite{mannion2016experimental} and energy network control \cite{cao2020reinforcement}. Unlike the conventional supervised learning in an offline learning paradigm, the learning process of MARL discussed in this thesis is completely conducted in an online learning paradigm. Owing to the recent advances of reinforcement learning (RL) (i.e. an area of ML which investigates the long-term optimal control through maximizing the accumulated discounted rewards) \cite{silver2016mastering,heinrich2016deep,schulman2017proximal,haarnoja2018soft}, the concept, model and method in game theory were extended and incorporated with the techniques developed in RL, forming the research area of MARL \cite{tesauro2003extending,lanctot2017unified}. Therefore, the research area of MARL is dedicated to solving more challenging problems that can be modelled as game models from the perspective of game theory, but are difficult to be addressed by the traditional methods in game theory.
    
    Cooperative game is a critical research area in the field of game theory and it can well model the cooperative tasks, where each decentralised controller is seen as an agent with the decentralised control regime and the whole system is named as multi-agent system (MAS). The most severe challenge of cooperative games is how cooperation among agents can be explicitly represented, so as to guide developing theory or methods to search the solution of control schemes. To address the challenge, a global reward is usually designed to give feedback to a team of agents, evaluating whether their decisions satisfy the common goal encoded in the global reward. The cooperative game that is equipped with a global reward is named as global reward game \cite{chang2004all}. Its objective is to achieve a joint policy of agents that reaches a long-term (one-shot) common goal, through maximizing the long-term accumulated global rewards (the one-shot global reward).
    
    The usual approach in MARL to solve global reward game is letting each agent individually maximize the accumulated discounted global rewards (also known as the value). Although it can guarantee the convergence to a stability solution called Nash equilibrium in the non-cooperative game theory under some conditions (e.g., centralised optimization with no conflicting explorations during learning), the contribution of each agent is failed to be identified which will intensively defect the convergence rate \cite{balch1997learning,balch1999reward} and even degrade the final joint policy into a policy for which only one agent performs the task with other dummy (idle) agents to the worst case. For example, when an agent is a dummy that contributes nothing to the whole group, it still receives the same value as other agents, which becomes a misleading signal to optimize the policy during learning. If the above learning process repeats for a while, it is highly likely that the dummy agent would not learn any useful policy to the team, and the joint policy degrades. The signal given to each agent is also called credit, which can be understood as a metric to measure an agent's decision as per its contribution to the team. 
    
    To address the above problem, there appears a branch of MARL called credit assignment which studies the scheme to assign a proper credit to each agent. In this thesis, we focus on investigating a novel theoretical framework to incorporate a credit assignment scheme into global reward game from the perspective of cooperative game theory.

\section{Motivation and Objectives} 
\label{sec:motivation_and_objectives}
    In general, credit assignment cannot be interpreted by a cooperative game guided by the solution concept for non-cooperative game theory such as Nash equilibrium. Specifically, each agent's payoff function is reformulated to be an identical function called potential function (equivalent to a global reward function), which can encourage cooperation among agents. Under such a situation, each agent only maximizes the accumulated global rewards, with no consideration of their own contributions to the global reward. This motivates us to establish a new theoretical framework and a novel solution concept that reasonably involves credit assignment along with searching the optimal joint policy of agents.
        
    Furthermore, most of previous works only raised the motivation of applying credit assignment. For example, although it is possible to manually shape an individual reward function for each agent, however, this often impedes learning performance due to the inaccurate description of reward shaping \cite{foerster2018counterfactual}. Moreover, \cite{SunehagLGCZJLSL18} showed up an example about inefficiency of the policies learned by the global value, i.e., only one agent learns a useful policy with other agents being lazy (i.e., contributing nothing to the team). Unfortunately, these works on credit assignment did not come up with any understanding or interpretation of the credit assigned to each agent during learning. To bridge this gap, we raise the following questions:
    \begin{enumerate}
        \item \textit{Is credit assignment valid in a global reward game?}
        \item \textit{Is it possible to derive theoretically guaranteed MARL algorithms that identify agents' contributions by credit assignment in a global reward game?}
    \end{enumerate}
    
    In this thesis, we attempt to answer these questions by extending the concepts in cooperative game theory, which is for showing the validity of credit assignment in global reward game. Especially, a payoff distribution scheme in cooperative game theory called Shapley value \cite{shapley1953value} is generalised and incorporated into MARL algorithms. Furthermore, we propose new MARL algorithms based on the generalised Shapley value.

\section{Contributions}
\label{sec:contributions}
    The contributions of this thesis will be summarized and discussed in Chapter \ref{chap:conclusion}. We now give a brief summary to help readers have an overall picture of the main contributions of this thesis. The general contribution is to construct a theoretical framework that extends the convex game to Markov decision process, named as Markov convex game, so that it can be used to rationalize the credit assignment in global reward game. Shapley value \cite{shapley1953value} is chosen as a credit assignment scheme that has been well studied in cooperative game theory, generalised to Markov convex game, named as Markov Shapley value. Markov Shapley value is proved to converge to the Markov core proposed in this thesis (i.e. a generalised solution concept extended from convex game to Markov convex game). By the property of Markov convex game, we prove that using Markov Shapley value as a credit assignment scheme leads to maximization of the global value, which matches the objective of the global reward game. In addition, Markov Shapley value inherits the properties of fairness from Shapley value that can well identify and quantify each agent's contribution to the team. The above reasons show up why Markov convex game is a suitable theoretical framework to represent global reward game to solve and validate the credit assignment problem.
    
    Using Markov Shapley value as a credit assignment scheme incorporated into an existing multi-agent reinforcement learning algorithm called MADDPG \cite{lowe2017multi} that belongs to the category of deterministic policy gradient algorithms, we propose SQDDPG. Although the above theoretical results are sufficient to motivate using Markov Shapley value as a credit assignment, it is with no theoretical guarantees on the convergence and the reliability may be risky. This directly impedes the further application to the real-world problems which request the restrict requirements of reliability. To address this problem, we develop the theory of Shapley-Bellman operator that is guaranteed to converge to the optimal Markov Shapley values. Depending on the theory, a multi-agent reinforcement learning algorithm called SHAQ is derived. Moreover, owing to the close relationship between Q-learning \cite{watkins1992q} and deterministic policy gradient \cite{silver2014deterministic}, the reliability of SQDDPG is also guaranteed. Regarding SQDDPG, Markov Shapley value is implemented as the convex combination of learnable marginal contributions (i.e., each marginal contribution is a direct function). Although this implementation can reduce the bias of fitting function, the issue of inconsistent credit assignments (i.e., the credit assignment of each agent could be formed by different classes of coalition values) may be induced, which prevents its application to real-world problems. To resolve the bias-inconsistency problem, we propose to directly learn coalition values and then use the learned coalition values to form marginal contributions and Markov Shapley values. The bias of directly learning coalition values is acceptable in light of our theoretical analysis. Since the difference of two coalition values is difficult to be differentiable, we use PPO as the base algorithm and propose SMFPPO. Due to that many real-world applications are only partially observable, we further extend Markov convex game to adapt to partially observable scenarios, named as partially observable Markov convex game. Furthermore, we propose Shapley policy iteration and Shapley value iteration for this game, which can be applied to instruct the implementations of SQDDPG, SMFPPO and SHAQ for partial observation.
    
    We evaluate performance of SQDDPG and SHAQ on the benchmark tasks from the community of machine learning. Both algorithms perform generally better than the state-of-the-art baselines and exhibit interpretability of credit assignments (i.e., the assigned credits can reflect agents' contributions and can be used as an index to interpret agents' behaviours). Moreover, we apply SQDDPG and SMFPPO to solve problems in energy networks. Amid the trend of decolonisation, an electric power network is the most important pillar of an energy network. An electric power network is also more challenging to control and operate compared to other energy networks (e.g. gas and heat networks) due to its fast and nonlinear behaviour. For this reason, in this thesis, we focus on the control of an electric power network, and more particularly, on the active voltage control problem in electric power distribution networks, as a test bench for the proposed MARL algorithms to solve real-world problems. In simulations, we show that SQDDPG and SMFPPO outperform other state-of-the-art multi-agent reinforcement learning algorithms.
    
    The output of the original work presented in this thesis is listed as follows.
    \begin{enumerate}
        \item Extending convex game in cooperative game theory to Markov decision process, named as Markov convex game and showing that it can represent global reward game, which provides a foundation of applying credit assignment in global reward game;
        \item Generalising a payoff distribution scheme in cooperative game theory called Shapley value to Markov convex game, which is then used as a credit assignment scheme to solve global reward game, named as Markov Shapley value;
        \item Incorporating the Markov Shapley value into the Bellman equation (i.e. the foundation of reinforcement learning) named as Shapley-Bellman equation to formulate a complete theoretical framework for multi-agent reinforcement learning. In more details, the proposed Shapley-Bellman equation is guaranteed to converge to the optimal Markov Shapley values and the optimal joint policy;
        \item Proposing three multi-agent reinforcement learning algorithms such as SHAQ, SQDDPG and SMFPPO, based on the above theoretical framework and practical implementation tricks in deep learning;
        \item Extending the Markov convex game to partial observability named as partially observable Markov convex game, the theoretical results of which guide the implementation of SHAQ, SQDDPG and SMFPPO to solve partially observable tasks;
        \item Formulating the active voltage control problem as a Dec-POMDP, so that applying multi-agent reinforcement learning algorithms to solve the task becomes reasonable;
        \item Releasing a simulator of active voltage control in power distribution networks for multi-agent reinforcement learning and evaluating the performance of SQDDPG and SMFPPO to demonstrate the potential of applying these two algorithms to solve the real-world problems. 
    \end{enumerate}

\section{Thesis Overview}
\label{sec:thesis_overview}
    \paragraph{Chapter 2.} This chapter introduces the necessary background knowledge about MARL, cooperative game theory and voltage control in electric power distribution networks. As for MARL, we start from multi-agent learning for traditional games, followed by the foundation of RL. Then, we introduce MARL based on the definition and knowledge shown in these two sections. Finally, we introduce credit assignment that is a traditional problem existing in MARL and the target problem to be solved in this thesis. About cooperative game theory, we start from clarifying the motivation from non-cooperative game theory, followed by the background knowledge of convex game and Shapley value which are the basic concepts to be studied and generalised in this thesis. 
    
    \paragraph{Chapter 3.} This chapter is the main part of this thesis, involving the theory and the derived methods. More specifically, the validity of applying credit assignment in global reward game is firstly discussed based on cooperative game theory. Then, a theory to incorporate Shapley value from cooperative game theory as a credit assignment scheme into multi-agent reinforcement learning is established. Based on the theory, we propose three MARL algorithms such as SQDDPG, SHAQ and SMFPPO. Moreover, we further extend Markov convex game to partial observability named as partially observable Markov convex game (POMCG), and propose Shapley policy iteration and Shapley value iteration that can solve this problem in theory. Finally, depending on the theory of POMCG, we provide an insight into practical implementation of SQDDPG, SHAQ and SMFPPO on solving partially observable problems.
    
    \paragraph{Chapter 4.} This chapter evaluates the performance of SQDDPG and SHAQ on the popular benchmarks from the community of machine learning. 
    
    \paragraph{Chapter 5.} This chapter evaluates the performance of SQDDPG and SMFPPO on the active voltage control problem in power distribution networks.
    
    \paragraph{Chapter 6.} This chapter summarizes the outcomes of this thesis and points out future works.


\chapter{Literature Review}

\section{The Roadmap of Multi-Agent Reinforcement Learning}
\label{sec:roadmap_of_multiagent_reinforcement_learning}
    In this section, we will review the progress and development influencing multi-agent reinforcement learning (MARL) and its relationship to the traditional realms such as game theory, control theory and reinforcement learning. In brief words, a multi-agent system is a group of interacting autonomous entities called agents sharing an environment, which they perceive by sensors and in which they act with actuators \cite{busoniu2008comprehensive}. There are a wide range of applications that can be modelled as a multi-agent system, such as robotic teams \cite{stone2000multiagent}, economics \cite{wooldridge2009introduction}, power systems \cite{wang2021multi} and so on.
    
    Among all these applications, it can be categorized as three classes of tasks: fully competitive scenarios, fully cooperative scenarios and mixed scenarios. Fully competitive scenarios imply that agents in the environment are opponents to each other (e.g., the sum of agents' returns or payoffs is zero, so called zero-sum game from the game theoretical perspective); fully cooperative scenarios imply that agents coordinate to accomplish a common goal (e.g., maximizing the common return or payoff from the game theoretical perspective); and mixed scenarios lie between these two settings (e.g., the general-sum game from the game theoretical perspective). In this thesis, we mainly concentrate on the fully cooperative scenarios, standing on the perspective of cooperative game theory, i.e. a theoretical framework assuming that there exists a binding agreement among agents about the distribution of payoffs or the choice of strategies \cite{chalkiadakis2011computational}. The payoff distribution is highly related to the credit assignment problem in multi-agent learning \cite{hernandez2019survey}.
    
    To provide a clear logic of the progress of MARL in history, the review commences as per the following topics: multi-agent learning, reinforcement learning, multi-agent reinforcement learning (including global reward game) and credit assignment. The taxonomy of the above topics is shown in Figure \ref{fig:marl_taxonomy_topics}. Note that multi-agent reinforcement learning as the conjunction of multi-agent learning and reinforcement learning, can solve more general games or scenarios (e.g. general-sum Markov games \cite{perolat2017learning}), however, in this thesis we merely concentrate on global reward game and the credit assignment problem.
    \begin{figure}[ht!]
        \centering
        \includegraphics[width=\textwidth]{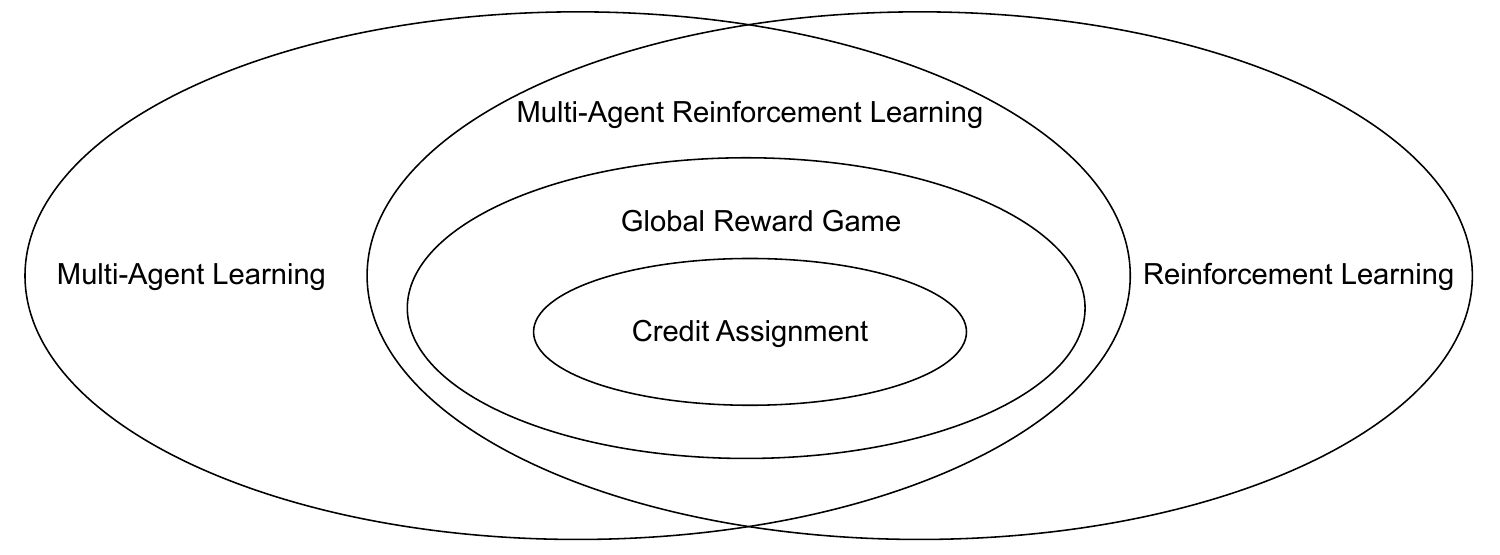}
        \caption{Taxonomy of multi-agent learning, reinforcement learning, multi-agent reinforcement learning, global reward game and the credit assignment problem. More specifically, multi-agent reinforcement learning is an interdiscipline of multi-agent learning and reinforcement learning. Global reward game is a problem formulating the fully cooperative scenarios for multi-agent reinforcement learning. Credit assignment is a critical question to the global reward game.}
    \label{fig:marl_taxonomy_topics}
    \end{figure}
    
    \subsection{Multi-Agent Learning for Traditional Games}
    \label{subsec:multi-agent_learning}
        Although it is possible to manually design the strategies of agents in advance, the increasing complexity and uncertainty appearing in real-world problems suggest automatically learning or searching strategies \cite{busoniu2008comprehensive,sen1999learning}. Multi-agent learning was initially arising and focused on solving static (stateless), repeated and multi-stage games in the literature of game theory \cite{fudenberg1998theory}. These games were mostly modelled to solve some social and economic problems \cite{osborne2004introduction}. 
        
        \paragraph{Game Class.} Static game is a game class that involves no dynamics and each agent just needs to make a one-shot decision. The representative example of static game is bi-matrix game, where there exist two agents and each agent possesses a payoff function described as a matrix influenced by both agents' decisions. Repeated game is an extension of static game and the only difference is that agents repeatedly make decisions on the same static game rather than the one-shot decision. During the process of decision making, each agent is able to collect the history of other agents' decisions that may help its own decision thereafter. Multi-stage game is an extension of repeated game, where the main difference is that the stage game will not be always identical, contrary to repeated game where the stage game is invariant. In both repeated game and multi-stage game, the outcome is a sequence of decisions for multiple stages (that may be infinite) and the payoff of each agent is simply calculated as the accumulated payoffs of stage games (that may be multiplied by discount factor for infinite stages). Both of games can be expressed as game-trees, equivalent to an extensive-form game with imperfect information (since the simultaneous move at each stage). Any extensive-form game can be represented as a static game \cite{osborne2004introduction}.
        
        \paragraph{Solution Concept.} The commonest solution concept to solve these games is called Nash equilibrium which is a stability criterion describing agents' decisions, where no agents would benefit by unilaterally varying its strategy. In other words, it is a joint strategy such that each agent's strategy is the best response to others. Any static game was proved to have a Nash equilibrium, while some games may possess more than one Nash equilibrium. In the context of MARL, Nash equilibrium is usually employed as a learning objective and a rule to update policies.
        
        \paragraph{Learning Algorithm.} The most famous and commonest algorithm to automatically solve these games (learn strategies) is called fictitious play that was proved to converge to Nash equilibrium in restricted classes of games such as fully cooperative games (also called potential game \cite{monderer1996potential} in game theory) and two-player zero-sum games \cite{monderer1996fictitious}. In a nutshell, at each iteration an agent acts the best response to an empirical model of other agents' history strategies, which can be regarded as a tracking method from the perspective of control theory. The follow-up works such as generalised weakened fictitious play \cite{leslie2006generalised}, joint strategy fictitious play \cite{marden2009joint} and fictitious self-play \cite{heinrich2015fictitious,heinrich2016deep}, generalised fictitious play to the model-free (or sample-based) paradigm without any knowledge of game models and combined it with the modern machine learning techniques, so that fictitious play can be applied in wider scenarios. Policy-space response oracle (PSRO) \cite{lanctot2017unified} generalised fictitious self-play to meta-games, where the strategy became the meta-strategy for choosing strategies.
        
        The multi-agent learning algorithms introduced above mainly aimed at solving the traditional class of games that can be represented as static games. For the dynamic games with uncertainties, these algorithms cannot be directly used. Rather, multi-agent reinforcement learning takes the place of solving such a class of game, called Markov game.\footnote{Some traditional learning methods were also extended to adapt to Markov game.}
    
    \subsection{Foundation of Reinforcement Learning}
    \label{subsec:foundation_of_reinforcement_learning}
        Reinforcement Learning (RL) is a sub-area of machine learning (ML) studying the interactive process between an agent and an environment, during which the agent takes an action at each time step as a reaction to the state it observes in order to maximize the (discounted) cumulative rewards received from the environment. The state is usually measured by sensors and the reward is usually hand-crafted via a real-valued function of state and action variables to encode an objective. In comparison with other two categories of ML algorithms, supervised learning and unsupervised learning, RL receives weak signals as rewards during learning rather than supervised learning with strong signals as labels and unsupervised learning with no explicit signals. More specifically, the RL paradigm can be seen as a sequential decision model (e.g. Markov chain \cite{gagniuc2017markov}, autoregressive-moving-average model (ARMA) \cite{box2015time}, etc.) from the perspective of supervised learning, but with no explicit ground-truth actions. On the other hand, the traditional RL only concentrates on the single-agent problem, however, it can be extended to the multi-agent scenarios with some natural mathematical extensions, which will be introduced in the following section. In the next paragraph, we will see a mathematical model as the foundation of RL called Markov decision process.

        \paragraph{Markov Decision Process.} The interactive process between an agent and an environment is typically described as a mathematical model called Markov decision process (MDP) \cite{bellman1952theory}, which can be specifically written as a 4-tuple $\langle \mathcal{S}, \mathcal{A}, T, R \rangle$. $\mathcal{S}$ is a set of states that indicates the set of events of an environment and $\mathcal{A}$ is a set of actions that the agent can select to interact with the environment. $T: \mathcal{S} \times \mathcal{A} \times \mathcal{S} \rightarrow [0, 1]$ is a (probability) transition function that describes the environmental dynamics, i.e., how the next state will be transited to, given the current state and an action taken by the agent. It is conventional to use a probability function $Pr(\mathbf{s}' | \mathbf{s}, \mathbf{a})$ to denote $T(\mathbf{s}', \mathbf{s}, \mathbf{a})$. $R: \mathcal{S} \times \mathcal{A} \rightarrow \mathbb{R}$ is a reward function that describes the goal of learning (or the desired behaviours of the agents). The purpose of solving MDP is finding an optimal sequence of actions that can maximize the discounted cumulative rewards such that $\max_{\{a_{t}\}_{t=0}^{\infty}} \mathbb{E}_{\{a_{t}\}_{t=0}^{\infty}} \left[ \sum_{t=0}^{\infty} \gamma^{t} R_{t} | \mathbf{S}_{0} \right]$. Usually, we consider a Markov stationary policy to simplify the problem, i.e., $\pi: \mathcal{S} \times \mathcal{A} \rightarrow [0, 1]$. Thereby, the above optimization problem is written as $\max_{\pi} \mathbb{E}_{\pi} \left[ \sum_{t=0}^{\infty} \gamma^{t} R_{t} | \mathbf{S}_{0} \right]$, which is usually written as the optimal state value function denoted as $V^{*}(\mathbf{S}_{0}) = \max_{\pi} V^{\pi}(\mathbf{S}_{0})$, where $V^{\pi}(\mathbf{S}_{0})$ is called state value function. If considering state-action pairs, then $Q^{\pi}(\mathbf{S}_{0}, \mathbf{A}_{0}) = \mathbb{E}_{\pi} \left[ \sum_{t=0}^{\infty} \gamma^{t} R_{t} | \mathbf{S}_{0}, \mathbf{A}_{0} \right]$, which is called action-value function or Q-value function.\footnote{The ``Q'' in Q-value is the abbreviation of the word quality.} It is known that this problem can be solved by an optimality criterion called Bellman optimality equation \cite{sutton2018reinforcement} such that
        \begin{equation}
        \label{eq:bellman_optimality_eq_v}
            V^{*}(\mathbf{s}) = \max_{\mathbf{a}} \sum_{\mathbf{s}' \in \mathcal{S}} Pr(\mathbf{s}' | \mathbf{s}, \mathbf{a}) \left[ R(\mathbf{s}, \mathbf{a}) + \gamma V^{*}(\mathbf{s}') \right], \quad \forall \mathbf{s} \in \mathcal{S}.
        \end{equation}
        
        Since $V^{*}(\mathbf{S}_{t}) = \max_{\mathbf{a}} Q^{*}(\mathbf{S}_{t}, \mathbf{a})$, we can equivalently write Eq.~\ref{eq:bellman_optimality_eq_v} as follows:
        \begin{equation}
        \label{eq:bellman_optimality_eq_q}
            Q^{*}(\mathbf{s}, \mathbf{a}) = \sum_{\mathbf{s}' \in \mathcal{S}} Pr(\mathbf{s}' | \mathbf{s}, \mathbf{a}) \left[ R(\mathbf{s}, \mathbf{a}) + \gamma \max_{\mathbf{a}'} Q^{*}(\mathbf{s}', \mathbf{a}') \right], \quad \forall \mathbf{s} \in \mathcal{S}, \mathbf{a} \in \mathcal{A}.
        \end{equation}
        
        There are two categories of dynamic programming methods to solve the Bellman optimality equation, called policy iteration (PI) and value iteration (VI) respectively. In more details, PI can be described as a process of two stages: policy evaluation and policy improvement, which can be mathematically expressed as follows:
        \begin{equation}
        \label{eq:policy_iteration}
            \begin{split}
                &\textbf{Policy Evaluation: } V(\mathbf{s}) \leftarrow \sum_{\mathbf{s}'} Pr(\mathbf{s}' | \mathbf{s}, \pi(\mathbf{s})) \left[ R + \gamma V(\mathbf{s}') \right], \forall \mathbf{s} \in \mathcal{S}, \\
                &\textbf{Policy Improvement: } \pi(\mathbf{s}) \leftarrow \arg\max_{\mathbf{a}} \sum_{\mathbf{s}'} Pr(\mathbf{s}' | \mathbf{s}, \mathbf{a}) \left[ R + \gamma V(\mathbf{s}') \right], \forall \mathbf{s} \in \mathcal{S},
            \end{split}
        \end{equation}
        where the update operation for policy evaluation is called dynamic programming (or Bellman operator in some other literature) which should nearly converge to a fixed point of $V(\mathbf{s})$, while policy improvement is performed by the argmax operation. These two stages are alternatively performed until the convergence to a nearly-optimal policy and the corresponding optimal values. 
        
        VI is a compact format of the combination of policy evaluation and policy improvement in one simple update operation such that
        \begin{equation}
        \label{eq:value_iteration_v}
            V(\mathbf{s}) \leftarrow \max_{\mathbf{a}} \sum_{\mathbf{s}'} Pr(\mathbf{s}' | \mathbf{s}, \mathbf{a}) \left[ R + \gamma V(\mathbf{s}') \right], \forall \mathbf{s} \in \mathcal{S}.
        \end{equation}
        
        Eq.~\ref{eq:value_iteration_v} is a dynamic programming which is performed recursively until the convergence to a fixed point of values. The fixed point is proved to be unique, denoted as $V^{*}(\mathbf{s})$. Then, the optimal policy can be recovered by performing the following operation such that
        \begin{equation}
        \label{eq:value_iteration_pi}
            \pi(\mathbf{s}) = \arg\max_{\mathbf{a}} \sum_{\mathbf{s}'} Pr(\mathbf{s}' | \mathbf{s}, \mathbf{a}) \left[ R + \gamma V^{*}(\mathbf{s}') \right], \forall \mathbf{s} \in \mathcal{S}.
        \end{equation}
        
        Eq.~\ref{eq:value_iteration_v} is guaranteed to converge to Bellman optimality equation expressed in Eq.~\ref{eq:bellman_optimality_eq_v}. If the Q-value $Q(\mathbf{s}, \mathbf{a})$ is in place of the state value $V(\mathbf{s})$, Eq.~\ref{eq:value_iteration_v} can be rewritten as the following update operation such that
        \begin{equation}
        \label{eq:value_iteration_q}
            Q^{\pi}(\mathbf{s}, \mathbf{a}) \leftarrow \sum_{\mathbf{s}'} Pr(\mathbf{s}' | \mathbf{s}, \mathbf{a}) \left[ R + \gamma \max_{\mathbf{a}'} Q^{\pi}(\mathbf{s}', \mathbf{a}') \right], \forall \mathbf{s} \in \mathcal{S}, \mathbf{a} \in \mathcal{A},
        \end{equation}
        which is proven to converge to the Bellman optimality equation expressed in Eq.~\ref{eq:bellman_optimality_eq_q} \cite{bertsekas2019reinforcement}.
        
        The main difference between PI and VI are discussed as follows. VI do not store the explicit policy function during learning, instead of which, it infers the optimal policy from the optimal values during execution. Policy iteration stores both policy function and value function during learning, so that it can directly make decision by the explicit policy function during execution.
        
        \paragraph{Approximation and Learning.} Solving the Bellman optimality equation by dynamic programming needs not only the access of a transition function (i.e. not always easy to be obtained in real-world applications), but also computational cost due to the curse of dimensionality over state size and action size. To address these two issues, Monte Carlo estimation is applied to form two methods, Monte Carlo control and temporal difference (TD) learning, which learn the optimal policy from the raw experience collected from simulation, without any exact model of environmental dynamics. Besides, we can extend these two methods to the situations where the function approximation is applied and introduce some modern RL algorithms that are used in this thesis such as Q-learning, Actor-Critic method, deterministic policy gradient and etc..
        
        As for Monte Carlo control, the Q-values are estimated by averaging the collected returns of sampled episodes and policy improvement is performed based on the estimated Q-values, following the paradigm that extends from PI called generalized policy iteration (GPI), where the ``generalized'' means the approximation of policy evaluation. To improve the representative capability of policy denoted as $\theta$, it is common to use a technique called function approximation to represent a policy as a parametric function. Therefore, it can be optimized by policy gradient method. The general idea is deriving a proximal gradient with respect to parameters of the policy (since an objective function is usually non-differentiable with respect to parameters of the policy) and employing gradient ascent to update the parameters. Mathematically, it can be expressed as follows:
        \begin{equation}
        \label{eq:policy_gradient}
            \begin{split}
                \nabla_{\theta} J(\theta) = \mathbb{E}_{\pi} \left[ Q^{\pi}(\mathbf{s}, \mathbf{a}) \nabla_{\theta} \log \pi_{\theta}(\mathbf{a} | \mathbf{s}) \right], \\
                \theta \leftarrow \theta + \alpha \nabla_{\theta} J(\theta),
            \end{split}
        \end{equation}
        where $\nabla_{\theta} J(\theta)$ is the policy gradient with respect to the objective function such that
        \begin{equation}
        \label{eq:policy_gradient_objective}
            J(\theta) = \sum_{\textbf{s} \in \mathcal{S}} d^\pi(\textbf{s}) V^\pi(\textbf{s}) = \sum_{\textbf{s} \in \mathcal{S}} d^\pi(\textbf{s}) \sum_{\textbf{a} \in \mathcal{A}} \pi_\theta(\textbf{a} \vert \textbf{s}) Q^\pi(\textbf{s}, \textbf{a}),
        \end{equation}
        where $d^\pi(\textbf{s})$ denotes the stationary distribution of a Markov chain involving $\pi_{\theta}(\mathbf{a} | \mathbf{s})$. 
        
        Based on the Monte Carlo estimation of Q-values (i.e., estimating the Q-values by the returns of sampled episodes), $\nabla_{\theta} J(\theta)$ can be written as $\mathbb{E}_{\pi} \left[ G_{t} \nabla_{\theta} \log \pi_{\theta}(\mathbf{A}_{t} | \mathbf{S}_{t}) \right]$, $G_{t}$ called return indicates the discounted cumulative rewards starting from timestep $t$. By sampling only one episode of experience to approximate the expected return, Eq.~\ref{eq:policy_gradient} becomes the update operation as follows:
        \begin{equation}
        \label{eq:reinforce}
            \theta \leftarrow \theta + \alpha G_{t} \nabla_{\theta} \log \pi_{\theta}(\mathbf{A}_{t} | \mathbf{S}_{t}),
        \end{equation}
        which is called REINFORCE \cite{williams1992simple}.
        
        As for TD learning, the general idea is using Monte Carlo estimation to approximate the dynamic programming. In other words, it estimates the value of the current state based on the estimation of a value of another successor state. This idea is also called bootstrapping. To ease the understanding of readers, we only introduce the basic off-policy TD(0) control. For exploring all state-action pairs in case of being trapped in the suboptimal policy, the raw experience is collected by an $\epsilon$-greedy policy that is called behaviour policy (i.e., with the probability of $\epsilon > 0$ a random action is selected, while with the probability of $1 - \epsilon$ the optimal policy is selected). Then, we use Monte Carlo estimation to approximate the dynamic programming as value iteration (see Eq.~\ref{eq:value_iteration_q}) and update Q-values with the collected raw experience such that
        \begin{equation}
        \label{eq:q_learning}
            Q(\mathbf{S}, \mathbf{A}) \leftarrow Q(\mathbf{S}, \mathbf{A}) + \alpha \left[ R + \gamma \max_{\textbf{a}} Q(\mathbf{S}', \mathbf{a}) - Q(\mathbf{S}, \mathbf{A}) \right],
        \end{equation}
        where the target policy (i.e., the policy we aim to learn) is a greedy policy and $R + \gamma \max_{\textbf{a}} Q(\mathbf{S}', \mathbf{a})$ is usually called TD target. Eq.~\ref{eq:q_learning} is also known as Q-learning, which is proved to converge to the optimal Q-values given some necessary conditions \cite{melo2001convergence}. With the precondition of greedy policy, we have the relationship such that $V^{\pi}(\mathbf{S}) = Q^{\pi}(\mathbf{S}, \mathbf{A})$. By the policy evaluation in Eq.~\ref{eq:policy_iteration} and Monte Carlo estimation, we can rewrite the policy evaluation as follows:
        \begin{equation}
        \label{eq:td_prediction}
            Q(\mathbf{S}, \mathbf{A}) \leftarrow Q(\mathbf{S}, \mathbf{A}) + \alpha \left[ R + \gamma Q(\mathbf{S}', \mathbf{A}) - Q(\mathbf{S}, \mathbf{A}) \right].
        \end{equation}
        
        Eq.~\ref{eq:td_prediction} is called TD(0) prediction that can be used to estimate Q-values directly. Replacing the return in Monte Carlo control by TD(0) prediction with the parametric Q-value function denoted as $Q_{\omega}(\mathbf{S}, \mathbf{A})$ in REINFORCE, a new algorithm called actor-critic method \cite{konda2000actor} is yielded such that
        \begin{equation}
            \begin{split}
                \omega \leftarrow \omega + \alpha^{\omega} \left[ R + \gamma Q_{\omega}(\mathbf{S}', \mathbf{A}) - Q_{\omega}(\mathbf{S}, \mathbf{A}) \right] \nabla_{\omega} Q_{\omega}(\mathbf{S}, \mathbf{A}), \\
                \theta \leftarrow \theta + \alpha^{\theta} Q_{\omega}(\mathbf{S}, \mathbf{A}) \nabla_{\theta} \log \pi_{\theta}(\mathbf{A} | \mathbf{S}).
            \end{split}
        \end{equation}
        
        The upper operation is called critic that stands for updating parameters of the action-value function, whereas the lower operation is called actor that stands for updating parameters of the policy function.
        
        If actions are continuous and that $\mathbf{a} = \pi_{\theta}(\mathbf{s})$, Eq.~\ref{eq:policy_gradient_objective} becomes the following equation such that
        \begin{equation}
        \label{eq:deterministic_policy_gradient_objective}
            J(\theta) = \sum_{\textbf{s} \in \mathcal{S}} d^\pi(\textbf{s}) Q^\pi(\textbf{s}, \textbf{a}).
        \end{equation}
        
        Since the action in $Q^{\pi}(\mathbf{s}, \mathbf{a})$ is now parametric, we can directly derive its gradient by the chain rule in calculus such that
        \begin{equation}
        \label{eq:deterministic_policy_gradient}
            \nabla_{\theta} J(\theta) = \mathbb{E}_{\mathbf{s} \sim d^\pi(\textbf{s})} \left[ \nabla_{\mathbf{a}} Q^{\pi}(\mathbf{s}, \mathbf{a}) \nabla_{\theta} \pi_{\theta}(\mathbf{s}) |_{\mathbf{a} = \pi_{\theta}(\mathbf{s}) } \right],
        \end{equation}
        which is called deterministic policy gradient (DPG) \cite{silver2014deterministic}. Using Monte Carlo estimation, the operation to update parameters of the policy is written as follows:
        \begin{equation}
            \theta \leftarrow \theta + \alpha^{\theta} \nabla_{\mathbf{a}} Q^{\pi}(\mathbf{s}, \mathbf{a}) \nabla_{\theta} \pi_{\theta}(\mathbf{s}) |_{\mathbf{a} = \pi_{\theta}(\mathbf{s})}.
        \end{equation}
        
        \paragraph{Practical Tricks and Implementations.} As the growing of deep learning, RL is combined with several tricks and deep neural networks (DNNs) to improve performance. The first improvement is using the DNNs to approximate the Q-value function to increase the capability of representation, however, Q-learning may suffer from instability and divergence when combined with the approximation of a nonlinear Q-value function and bootstrapping. To mitigate the issues, \cite{mnih2015human} proposed experience replay and periodically updated target. The main idea of experience replay is constructing a replay buffer $D = \left\{ e_{1}, e_{2}, ..., e_{T} \right\}$ that stores the raw experience of transitions $e_{t} = \left( \mathbf{S}_{t}, \mathbf{A}_{t}, R_{t}, \mathbf{S}_{t+1} \right)$ collected by $\epsilon$-greedy policy. During each update of Q-learning, one or a batch of transitions are randomly sampled from the replay buffer and any samples in the replay buffer can be repeatedly used to improve data efficiency, remove correlations in the observation sequences, and smooth over changes in the data distribution. The main idea of periodically updated target is that the Q-value function in the TD target is periodically updated, which makes training more stable as it overcomes the short-term oscillations. Accordingly, the loss function to update the Q-value function is as follows:
        \begin{equation}
        \label{eq:fitted_q_learning}
            \min_{\omega} \mathbb{E}_{\left( \mathbf{s}, \mathbf{a}, R, \mathbf{s}' \right) \sim U(D)} \left[ \left( R + \max_{\mathbf{a}'} Q_{\omega-}(\mathbf{s}, \mathbf{a}') -  Q_{\omega}(\mathbf{s}, \mathbf{a}) \right) \right].
        \end{equation}
        
        Eq.~\ref{eq:fitted_q_learning} is the form of fitted Q-learning \cite{riedmiller2005neural} that tackles the continuous function approximation of Q-values (due to the approximation using DNN here). $U(D)$ indicates uniform sampling from the replay buffer. $Q_{\omega-}(\mathbf{s}, \mathbf{a}')$ is updated periodically (i.e., more slowly than the update frequency of $Q_{\omega}(\mathbf{s}, \mathbf{a}')$) by cloning $Q_{\omega}(\mathbf{s}, \mathbf{a}')$. Combined these two tricks and the function approximation using DNNs with Q-learning, it forms deep Q-network (DQN) \cite{mnih2015human}.
        
        Following the same philosophy, the algorithm of deep deterministic policy gradient (DDPG) \cite{lillicrap2015continuous} was proposed. Different from using the $\epsilon$-greedy policy as a behaviour policy for DQN, it uses an extra Gaussian noise over a continuous action to conduct exploration such that $\pi'(\mathbf{s}) = \pi(\mathbf{s}) + \mathcal{N}$, where $\mathcal{N}$ indicates a unit Gaussian distribution. Another modification is that in DDPG the target Q-value function is updated softly, i.e., $\omega^{-} \leftarrow \tau \omega + (1 - \tau) \omega^{-}$, where $\tau$ is the learning rate to periodically update parameters of the target Q-value function $\omega^{-}$. Besides, parameters of the target policy is updated similarly as the target Q-value function.
        
        The rapid change of policy led by the updates of parameters may cause training instability. To mitigate this issue, the update of policy should not be too much at each step, based on which the trust region policy optimization (TRPO) \cite{schulman2015trust} was proposed. Before looking into TRPO, let us see the off-policy model in Eq.~\ref{eq:policy_gradient_objective}, i.e., optimizing the policy $\pi$ based on the raw experience collected by another policy $\beta$. The mismatch between the training data distribution and the true policy state distribution is compensated by importance sampling \cite{kloek1978bayesian}. This off-policy model can be described as follows:
        \begin{equation}
        \label{eq:off-policy_model}
            \begin{split}
                J(\theta) &= \sum_{\textbf{s} \in \mathcal{S}} d^{\pi_{\theta_{\text{old}}}}(\textbf{s}) \sum_{\textbf{a} \in \mathcal{A}} \pi_\theta(\textbf{a} \vert \textbf{s}) \hat{Q}^{\pi_{\theta_{\text{old}}}}(\textbf{s}, \textbf{a}) \\
                &= \sum_{\textbf{s} \in \mathcal{S}} d^{\pi_{\theta_{\text{old}}}}(\textbf{s}) \sum_{\textbf{a} \in \mathcal{A}} \beta(\mathbf{a} | \mathbf{s}) \frac{\pi_\theta(\textbf{a} \vert \textbf{s})}{\beta(\mathbf{a} | \mathbf{s})} \hat{Q}^{\pi_{\theta_{\text{old}}}}(\textbf{s}, \textbf{a}) \\
                &= \mathbb{E}_{\mathbf{s} \sim d^{\pi_{\theta_{\text{old}}}}, \mathbf{a} \sim \beta} \left[ \frac{\pi_\theta(\textbf{a} \vert \textbf{s})}{\beta(\mathbf{a} | \mathbf{s})} \hat{Q}^{\pi_{\theta_{\text{old}}}}(\textbf{s}, \textbf{a}) \right],
            \end{split}
        \end{equation}
        where $\pi_{\theta_{\text{old}}}$ indicates parameters of the policy before the update and $\hat{Q}^{\pi_{\theta_{\text{old}}}}(\textbf{s}, \textbf{a})$ indicates the estimated Q-value function. Theoretically, during an on-policy training process, the behaviour policy should be consistent with the target policy. Nevertheless, when the rollout workers and optimization are processed asynchronously, the behaviour policy could be old-dated. TRPO captures this phenomenon and models the behaviour policy as $\pi_{\theta_{\text{old}}}$ and Eq.~\ref{eq:off-policy_model} becomes the following equation such that
        \begin{equation}
        \label{eq:trpo_objective}
            J^{\text{TRPO}}(\theta) = \mathbb{E}_{\mathbf{s} \sim d^{\pi_{\theta_{\text{old}}}}, \mathbf{a} \sim \pi_{\theta_{\text{old}}}} \left[ \frac{\pi_\theta(\textbf{a} \vert \textbf{s})}{\pi_{\theta_{\text{old}}}(\mathbf{a} | \mathbf{s})} \hat{Q}^{\pi_{\theta_{\text{old}}}}(\textbf{s}, \textbf{a}) \right].
        \end{equation}
        
        Besides, TRPO also considers the trust region constraint that restricts the update from an old policy to a new policy through the measure of KL divergence \cite{kullback1951information} within some threshold $\delta > 0$ such that
        \begin{equation}
        \label{eq:trpo_constraint}
            \mathbb{E}_{\textbf{s} \sim d^{\pi_{\theta_{\text{old}}}}} \left[ D_\text{KL}(\pi_{\theta_{\text{old}}}( \cdot | \mathbf{s}) \| \pi_\theta(\cdot \vert \textbf{s}) \right] \leq \delta. 
        \end{equation}
        
        TRPO was proved to guarantee a monotonic improvement over PI \cite{schulman2015trust}.
        
        To simplify TRPO, proximal policy optimization (PPO) \cite{schulman2017proximal} was proposed to replace the formulation of TRPO by a clipped surrogate objective. If the ratio between an old policy and a new policy is denoted as that
        \begin{equation}
            r(\theta) = \frac{\pi_\theta(\mathbf{a} \vert \textbf{s})}{\pi_{\theta_{\text{old}}}(\mathbf{a} | \mathbf{s})},
        \end{equation}
        the objective of TRPO becomes the following equation such that
        \begin{equation}
            J^{\text{TRPO}}(\theta) = \mathbb{E}_{\mathbf{s} \sim d^{\pi_{\theta_{\text{old}}}}, \mathbf{a} \sim \pi_{\theta_{\text{old}}}} \left[ r(\theta) \hat{Q}^{\pi_{\theta_{\text{old}}}}(\textbf{s}, \textbf{a}) \right].
        \end{equation}
        
        Since $J^{\text{TRPO}}(\theta)$ will suffer instability if the parameter update is too large, PPO enforces the $r(\theta)$ to stay around 1 by adding the clipping operation such that $clip(r(\theta), 1-\epsilon, 1+\epsilon)$, where $\epsilon > 0$ is small enough and $r(\theta) \in [1-\epsilon, 1+\epsilon]$. By taking the minimum between the original TRPO objective and the clipped objective to avoid the extreme policy update, the objective of PPO is expressed as follows:
        \begin{equation}
            J^{\text{PPO}}(\theta) = \mathbb{E}_{\mathbf{s} \sim d^{\pi_{\theta_{\text{old}}}}, \mathbf{a} \sim \pi_{\theta_{\text{old}}}} \left[ \min \left\{ r(\theta) \hat{Q}^{\pi_{\theta_{\text{old}}}}(\textbf{s}, \textbf{a}), clip(r(\theta), 1-\epsilon, 1+\epsilon) \hat{Q}^{\pi_{\theta_{\text{old}}}}(\textbf{s}, \textbf{a}) \right\} \right].
        \end{equation}
        
        The policy of PPO is modelled as Gaussian distribution for continuous actions (i.e., learning both mean and variance), and categorical distribution for discrete actions. Moreover, maximizing an extra entropy term over the policy is needed to encourage exploration. Some other tricks such as normalization and clipping of rewards and observations could influence training performance.
        
    \subsection{Multi-Agent Reinforcement Learning}
    \label{subsec:multi-agent_rl}
        \paragraph{Markov Game.} We begin with the common setting of multi-agent dynamic system under Markov property called Markov game (also known as stochastic game in some literature). Mathematically, a Markov game can be described as a tuple $\langle \mathcal{I}, \mathcal{S}, \mathcal{A}, T, (R_{i})_{i = 1,...,|\mathcal{I}|} \rangle$. $\mathcal{I}$ is a set of agents. $\mathcal{S}$ is the set of states; $\mathcal{A} = \mathlarger{\mathlarger{\times}}_{i = 1}^{|\mathcal{I}|} \mathcal{A}_{i}$ is a joint action set and $\mathcal{A}_{i}$ is the action set of each agent. $T: \mathcal{S} \times \mathcal{A} \times \mathcal{S} \rightarrow [0, 1]$ is a transition function describing the environmental dynamics (i.e., how the next state is achieved given the current state and a joint action of agents). Similar to MDP, it is conventional to use a probability function $Pr(\mathbf{s}' | \mathbf{s}, \mathbf{a})$ to express $T(\mathbf{s}', \mathbf{s}, \mathbf{a})$. $R_{i}: \mathcal{S} \times \mathcal{A} \rightarrow \mathbb{R}$ is a reward function describing how much an arbitrary agent $i$ can benefit from its action in some state. Each agent's goal is to maximize the accumulated discounted rewards such that $\mathbb{E}\left[ \sum_{t=1}^{\infty} \gamma^{t-1} R_{i} \right]$. If the sum of agents' rewards is equal to zero such that $\sum_{i \in \mathcal{I}} R_{i} = 0$, then the Markov game becomes fully competitive. If each agent's reward function is identical such that $R_{1} = R_{2} = ... = R_{|\mathcal{I}|}$, then the Markov game becomes fully cooperative, named as Global Reward Game (a.k.a. Team Reward Game), which is the scenario considered in this thesis.
        
        \paragraph{Goal of Learning.} From the perspective of learning, in addition to the stability of learning dynamics, some researchers focus on adaption of the dynamic behaviours of other agents \cite{busoniu2008comprehensive} or awareness to opponents \cite{hernandez2019survey}. More specifically, the stability implies the convergence to a stationary joint policy that satisfies some solution concept (e.g. Nash equilibrium and its extension to dynamic games called Markov perfect equilibrium \cite{maskin2001markov}), while the adaption means an agent's performance can be improved or maintained as other agents change their policies. The criteria used to define the adaption include targeted optimality, compatibility, safety \cite{powers2004new} and rationality \cite{bowling2002multiagent,bowling2004convergence}. 
        
        The rationality can be further specified as two schemes such that (1) an agent is required to converge to a best response when other agents remain stationary; and (2) an agent should reach a return that is at least as good as the return of any stationary strategy, and this holds for any set of other agents' strategies, which is usually called no-regret. It is noticed that stability and adaption are two orthogonal goals, and they can be considered either individually or simultaneously when designing a solution concept. In other words, the satisfaction of either stability or adaption does not necessarily implies the satisfaction of another. Nevertheless, there is an exception that a Nash equilibrium is achieved whenever the rationality defined by (1) above is satisfied. This example implies that the goal of learning highly depends on a subjective definition rather than an objective property. Therefore, the meaningfulness of a learning goal is critical, which is an indispensable part of understanding multi-agent behaviours.
        
        In this thesis, we define the learning goal based on a concept in cooperative game theory called core. It is a mathematical description of the extent of cooperation under a specific game model called convex game. In more details, if the outcome of the game lies in the core, it is believable that agents have formed collaboration with the optimal coalition value.
        
        \paragraph{Learning Algorithm Taxonomy.} The most popular taxonomy to categorize MARL algorithms is based on task taxonomy \cite{busoniu2008comprehensive} (i.e., there exist a group of algorithms that can overcome a type of task). Besides, an MARL algorithm can be categorized as per prior knowledge of a task. If an agent knows full or partial information, it is called model-based learning.\footnote{In the modern MARL community, a paradigm that learns a environmental model is also called model-based MARL.} Otherwise, it is called model-free learning. Most of existing works focus on the research on model-free learning, due to its easy implementation. Nevertheless, the increasing requirement of sample efficiency (with less interactions with environments) encourages advances of model-based learning from both theoretical and empirical aspects \cite{bai2020provable,krupnik2020multi,willemsen2021mambpo}.
        
        Another taxonomy of MARL algorithms is based on learning paradigms, by which MARL can be categorized as independent learning, agent awareness and agent tracking (agent modelling). Typically, independent learning is fully targeting on stability regardless of other agents' behaviours, while agent tracking is fully targeting on adaption with no consideration of stability; and agent awareness lies between these two paradigms \cite{bu2008comprehensive}. However, the toughest obstacle of adopting independent learning is the non-stationary environment to each agent when regarding other agents as part of its environment\footnote{In the recent work, \cite{LyuXDA21} proved that independent learning can also converge to the optimal policy (i.e. decentralized critic in Actor-Critic framework) given specific assumptions on the correlation between an agent's and other agents' sample histories.}. To resolve this issue, there are two paradigms such that (1) permitting communication among agents as part of information for decision (during both the learning and the execution phases) and (2) centralized training and decentralized execution (CTDE), whereby the information is only shared among agents (e.g. through communication) during the learning phase. Furthermore, centralized training can be seen as a sort of agent awareness, since it considers other agents' actions as part of its input that is equivalent to tracking agents' decisions. 
        
        In this thesis, we mainly focus on applying the CTDE paradigm to design a novel MARL algorithm. There are two main reasons for this choice: (1) the centralized training can mitigate defect of non-stationary environment problems caused by simultaneous multi-agent decisions, and adapt to variations of environmental dynamics (i.e., adaption to other agents' policies caused by emergencies); and (2) decision is flexible to be conducted when some emergencies (e.g., communication is unavailable in real-world applications such as natural disasters).
        
    \subsection{Credit Assignment in Global Reward Game}
    \label{subsec:credit_assignment}
        Credit assignment is a significant problem that has been studied in global reward game for a long period. The original motivation is solving the problem of unfairness (i.e., each agent's contribution cannot be fairly reflected) \cite{wolpert2002optimal} and addressing negative influence on convergence to the optimal joint policy \cite{balch1997learning,balch1999reward}. The counterpart solution to credit assignment solving global reward game is the shared reward approach, where all agents jointly maximize the cumulative global rewards. A classic example to show the importance of credit assignment is shown in Example \ref{exa:didactic_example}.
        \begin{example}
        \label{exa:didactic_example}
            Suppose that two agents are dummies and they contribute nothing to a group of three agents, while the rest agent performs beneficial actions to the group. By employing the shared reward approach, these two dummy agents receive the same credits as the rest agent. This may lead to a problem that the two dummy agents will believe that their meaningless actions can still optimize the global rewards and continue perform sub-optimal policies. Thereafter, the multi-agent problem will be degraded to a single-agent problem, wherein only one agent truly aims at solving the task. As a result, the mismatched credits will probably impede convergence to the optimal joint policy.
        \end{example}

        The earlier literature incorporated coordination graphs to linearly factorize the global value as a method to implement the credit assignment \cite{guestrin2002coordinated,kok2005using}. In contrast, \cite{chang2004all} attempted using Kalman filter to infer the credit to each agent. \cite{foerster2018counterfactual} and \cite{nguyen2018credit} modelled the marginal contributions inspired by the reward difference \cite{wolpert2002optimal}. VDN \cite{SunehagLGCZJLSL18} was proposed to learn the credit assignment as factorised Q-values, assuming that any global Q-value equals to the sum of factorised Q-values. Nevertheless, this factorisation may limit the representation of the global Q-value. To mitigate this issue, QMIX \cite{RashidSWFFW18} and QTRAN \cite{SonKKHY19} were proposed to represent the global Q-value in a richer class with respect to factorised Q-values, as per an assumption named as Individual-Global-Max (IGM) (see Definition \ref{def:igm}). Another research thread is investigating representation of the credit assignment such as COMA \cite{foerster2018counterfactual}. It is not difficult to observe that the analytical form of COMA shown in Eq.~\ref{eq:coma} naturally satisfies IGM. The theoretical framework of Markov convex game proposed in this thesis can be seen as an realization of the monolithic IGM, but with immense theoretical backgrounds.
        \begin{definition}
        \label{def:igm}
            For a joint Q-value $Q^{\pi}(\mathbf{s}, \mathbf{a})$ with a deterministic policy, if the following equation is assumed to hold such that
            \begin{equation}
                \arg \max_{\mathbf{a}} Q^{\pi}(\mathbf{s}, \mathbf{a}) = \big(\arg\max_{a_{i}} Q_{i}(\mathbf{s}, a_{i}) \big)_{i=1, 2, ..., |\mathcal{N}|},
            \label{eq:igm}
            \end{equation}
            then we say that $\big(\mathit{Q}_{i}(\mathbf{s}, a_{i})\big)_{i=1, 2, ..., |\mathcal{N}|}$ satisfies Individual-Global-Max (IGM) and $\mathit{Q}^{\pi}(\mathbf{s}, \mathbf{a})$ can be factorised by $\big(\mathit{Q}_{i}(\mathbf{s}, a_{i})\big)_{i=1, 2, ..., |\mathcal{N}|}$.
        \end{definition}

        On the other hand, the branch of works discussed above did not take into consideration of equilibrium of the credit assignment (i.e., the credit assignment to a team of agents with the optimal joint policy is unclear). To address this weakness, in this thesis we introduce a solution concept named as Markov core (that can be regarded as a stability criterion) to evaluate the credit assignment corresponding to the optimal joint policy. This bridges the gap between the credit assignment and the optimal global value. 
        
        \paragraph{VDN.} VDN linearly factorises a global value function such that
        \begin{equation}
            Q^{\pi}(\mathbf{s}, \mathbf{a}) = \sum_{i \in \mathcal{N}} Q_{i}(\mathbf{s}, a_{i}),
        \label{eq:vdn}
        \end{equation}
        so that Eq.~\ref{eq:igm} holds.
        
        \paragraph{QMIX.} QMIX learns a monotonic mixing function $\mathit{f}_{\mathbf{s}}: \mathlarger{\mathlarger{\times}}_{i \in \mathcal{N}} Q_{i}(\mathbf{s}, a_{i}) \times \mathbf{s} \mapsto \mathbb{R}$ to implement the factorisation such that
        \begin{equation}
            Q^{\pi}(\mathbf{s}, \mathbf{a}) = f_{\mathbf{s}}\big(Q_{1}(\mathbf{s}, a_{1}), ..., Q_{|\mathcal{N}|}(\mathbf{s}, a_{|\mathcal{N}|})\big),
        \label{eq:qmix}
        \end{equation}
        so that Eq.~\ref{eq:igm} holds. Although QMIX has a richer functional class of factorisation than that of VDN, it meets a problem that $\max_{\mathbf{a}} Q^{\pi}(\mathbf{s}, \mathbf{a}) = \sum_{i \in \mathcal{N}} \max_{a_{i}} Q_{i}(\mathbf{s}, a_{i})$ does not necessarily hold, which may lead to biases on Q-value estimation \cite{SonKKHY19} and corrupt the learning process to achieve the optimal joint policy. Theoretically, VDN does not possess the problem discussed above, however, the functional class of the simply additive factorisation is so restrictive \cite{RashidSWFFW18}.
        
        \paragraph{QTRAN.} QTRAN gives a sufficient condition for value factorisation that satisfies IGM such that
        \begin{equation}
            \sum_{i \in \mathcal{N}} Q_{i}(\mathbf{s}, a_{i}) - Q^{\pi}(\mathbf{s}, \mathbf{a}) + V^{\pi}(\mathbf{s}) = 
            \begin{cases}
                 0 & \mathbf{a} = \mathbf{\bar{a}}, \\
                 \geq 0 & \mathbf{a} \neq \mathbf{\bar{a}},
            \end{cases}
        \label{eq:qtran}
        \end{equation}
        wherein
        \begin{equation*}
            V^{\pi}(\mathbf{s}) = \max_{\mathbf{a}} Q^{\pi}(\mathbf{s}, \mathbf{a}) - \sum_{i \in \mathcal{N}} Q_{i}(\mathbf{s}, \bar{a}_{i}).
        \end{equation*}
        In Eq.~\ref{eq:qtran}, $\mathbf{a} = \mathlarger{\mathlarger{\times}}_{\scriptscriptstyle i \in \mathcal{N}} a_{i}$; and $\mathbf{\bar{a}} = \mathlarger{\mathlarger{\times}}_{\scriptscriptstyle i \in \mathcal{N}} \bar{a}_{i}$ where $\bar{a}_{i} = \arg\max_{a_{i}} Q_{i}(\mathbf{s}, a_{i})$ because of IGM. Additionally, \cite{SonKKHY19} showed that the above condition also holds for affine transformation on $Q_{i}, \forall i \in \mathcal{N}$ such that $w_{i} Q_{i} + b_{i}$. For this reason, an additional transformed global Q-value such that $Q^{\pi'}(\mathbf{s}, \mathbf{a}) = \sum_{i \in \mathcal{N}} Q_{i}(\mathbf{s}, a_{i})$ by setting $w_{i} = 1$ and $\sum_{i \in \mathcal{N}} b_{i} = 0$ is used to represent the value factorisation. It is forced to fit the above condition with a learned global Q-value $Q^{\pi}(\mathbf{s}, \mathbf{a})$ and $V^{\pi}(\mathbf{s})$. \cite{SonKKHY19} argued that finding a factorisation of $Q^{\pi'}(\mathbf{s}, \mathbf{a})$ is equivalent to finding $[Q_{i}]_{i \in \mathcal{N}}$ to satisfy IGM. 
        
        \paragraph{COMA.} COMA \cite{foerster2018counterfactual} was inspired by the idea of difference rewards \cite{wolpert2002optimal} and proposed to subtract the counterfactual baseline that excludes the effect of an arbitrary agent $i$ from the global Q-value to represent the credit assigned to the agent $i$. The mathematical expression is as follows:
        \begin{equation}
            Q_{i}(\mathbf{s}, \mathbf{a}) = Q^{\pi}(\mathbf{s}, \mathbf{a}) - \sum_{a_{i} \in \mathcal{A}_{i}} \pi_{i}(a_{i} | \mathbf{s}) Q^{\pi}\left( \mathbf{s}, (\mathbf{a}_{-i}, a_{i}) \right),
        \label{eq:coma}
        \end{equation}
        where $\sum_{a_{i} \in \mathcal{A}_{i}} \pi_{i}(a_{i} | \mathbf{s}) Q^{\pi}\left( \mathbf{s}, (\mathbf{a}_{-i}, a_{i}) \right)$ is called the counterfactual baseline. COMA is actually a special case of the marginal contribution from the perspective of cooperative game theory and we will discuss it in the next chapter.
        
\section{Cooperative Game Theory}
\label{sec:cooperative_game_theory}
    Game theory aims at studying interactions between self-interested computational entities called agents, where the research object is called game that captures the main attributes of a scenario populated by a set of agents \cite{osborne2004introduction,chalkiadakis2011computational}. Each game is equipped with an outcome to describe the result of the game. A primary concern of a game is investigating the rational outcome. To achieve this goal, game theorists have developed an enormous number of solution concepts to depict rationality. In other words, if an outcome lies in an appropriate solution concept, then it captures the rational outcome of the game. The known challenge is that these solution concepts are not guaranteed to exist in a game, which leads to the emergence of a range of solution concepts (e.g. Nash equilibrium, correlation equilibrium and etc.) representing different rationality. This can be slightly mitigated by restricting the categories of games. From the categories of games, it can be classified to non-cooperative games and cooperative games. Cooperative game theory is a subject that mainly studies the cooperative games modelling collaborations among agents, which is the research interest of this thesis.
    
    \subsection{A Motivation of Cooperative Game Theory}
    \label{subsec:difference_from_non-cooperative_game_theory}
        To understand the motivation of establishing cooperative game theory, it is a good point to see why a non-cooperative game is non-cooperative. We will give an example of a conventional non-cooperative game called Prisoner's Dilemma from \cite{axelrod1981evolution,chalkiadakis2011computational} to elucidate the reason. The Prisoner's Dilemma is usually described as Example \ref{exp:prisoners_delimma} shows.
        \begin{example}
        \label{exp:prisoners_delimma}
            There are two men who were accused of crimes and held in separate cells with no way of communication or meeting, so with no way to make binding agreements. They are told that
            \begin{itemize}
                \item if one confesses and the other does not, the confessor will be freed and the other will be jailed for three years;
                \item if both confess, then each will be jailed for two years;
                \item if neither confesses, then they will be jailed for one year.
            \end{itemize}
        \end{example}
        
        The key question is whether the prisoners decide to cooperate (not confess) or not cooperate (confess). Since in this game both prisoners are symmetric, the thought and decision of either agent should be equal to another. It is not difficult to see that for either agent, regardless of another agent's decision, the optimal choice should be to confess. Therefore, the outcome becomes that both agents confess and they will be jailed for two years. This actually follows a solution concept under the non-cooperative game theory known as dominant strategy equilibrium, which describes the unique rational outcome under this solution concept. Nevertheless, if two agents cooperate, then the outcome will be better for both agents (i.e., each agent only needs to be jailed for one year). The miss of the later outcome is due to the lack of conditions for cooperation (i.e., binding agreements are not possible). In more details, this game is designed as a scenario where no agents can communicate or meet to each other. As a result, they cannot trust one another and can solely maximize their own utilities given an assumption that the other will think in the same way. If communication is available, they could make binding agreements to cooperate and achieve a better outcome than the dominant strategy equilibrium, and the dilemma will naturally disappear. This example shows up the significance of making binding agreements in cooperative games, which leads to the definition of a cooperative game as Definition \ref{def:cooperative_game} shows.
        \begin{definition}[\cite{peleg2007introduction} ]
        \label{def:cooperative_game}
            A game is cooperative if the agents can make binding agreements about the distribution of payoffs or the choice of strategies, even if these agreements are not specified or implied by the rules of the game.
        \end{definition}
        
        Cooperative game theory focuses on the solution concepts that take into consideration of making binding agreements. We would emphasize that such binding agreements are common in the real-life scenarios such as building contracts that forms the commerce and global economy \cite{chalkiadakis2011computational}. This enables this theoretical framework meaningful. We assume that some methods exist for shaping binding agreements when we develop the theory to simplify analysis. 
        
    \subsection{Convex Game}
    \label{subsec:convex_game}
        \paragraph{Problem Definition.} Convex game (CG) \cite{shapley1971cores} is a typical characteristic game belonging to the category of transferable utility games in cooperative game theory. We now introduce the basic definitions referred to a popular textbook \cite{chalkiadakis2011computational}. The CG is formally represented as $\Gamma=\langle \mathcal{N}, V \rangle$, where $\mathcal{N}$ is the set of all agents and $V$ is the value function to measure the profits earned by a coalition. $\mathcal{N}$ also indicates the grand coalition (i.e., the largest coalition which includes all agents). The value function $V: 2^{\scriptscriptstyle |\mathcal{N}|} \rightarrow \mathbb{R}_{\geq0}$ is a characteristic function that maps from the coalition space to real numbers. For all coalitions $\mathcal{C}, \mathcal{D} \subset \mathcal{N}$, the value function satisfies the following property such that
        \begin{equation}
            V(\mathcal{C} \cup \mathcal{D}) + V(\mathcal{C} \cap \mathcal{D}) \geq V(\mathcal{C}) + V(\mathcal{D}).
        \label{eq:convex_function}
        \end{equation}
        
        Eq.~\ref{eq:convex_function} is defined as a supermodular function and a game with a supermodular characteristic function is defined as convex, which is the origin of the name of convex game. Note that the value of an empty coalition is defined as zero, i.e., $V(\emptyset) = 0$. If imposing an additional assumption such that any two coalitions are independent, i.e., $\mathcal{C} \cap \mathcal{D} = \emptyset$, Eq.~\ref{eq:convex_function} is reduced to the following form such that
        \begin{equation}
            V(\mathcal{C} \cup \mathcal{D}) \geq V(\mathcal{C}) + V(\mathcal{D}) - V(\mathcal{C} \cap \mathcal{D}) = V(\mathcal{C}) + V(\mathcal{D}).
        \label{eq:superadditive_function}
        \end{equation}
        
        Eq.~\ref{eq:superadditive_function} is called superadditive. Note that convex game is necessarily superadditive, however, the reverse is not always true. In this thesis, with the condition of independent coalitions, Eq.~\ref{eq:superadditive_function} is used as the definition of convex game in analysis for conciseness.
    
        \paragraph{Solution Concept.} The outcome of CG is a tuple $(\mathcal{CS}, \mathbf{x})$, where $\mathcal{CS} = \{ \mathcal{C}_{1}, \mathcal{C}_{2}, ..., \mathcal{C}_{m} \}$ is a coalition structure, and $\mathbf{x} = (x_{i})_{i \in \mathcal{N}}$ indicates the payoffs distributed to agents which satisfy two conditions: (1) $x_{i} \geq 0, \forall i \in \mathcal{N}$; and (2) $\sum_{i \in \mathcal{C}} x_{i} \leq V(\mathcal{C}), \forall \mathcal{C} \subseteq \mathcal{CS}$. The solution concepts of interests are from two aspects: (1) the fairness of payoff distribution scheme; and (2) the stability of coalition structure.
        
        Core is a stable outcome set of CG that defines the stability of coalition structure. In details, it can be mathematically defined as the following set such that $$\texttt{Core}(\Gamma) = \left\{ (\mathcal{C}, \mathbf{x}) \ \bigg| \ \sum_{i \in \mathcal{C}} x_{i} \geq V(\mathcal{C}), \forall \mathcal{C} \subseteq \mathcal{N} \right\}.$$ It ensures a rational payoff distribution scheme and no subset of agents would have incentives to deviate from its coalition to acquire more profits. This can be easily verifiable by the following didactic example. Suppose that $\sum_{i \in \mathcal{C}} x_{i} < V(\mathcal{C})$ for some $\mathcal{C} \ \mathlarger{\mathlarger{\subseteq}} \ \mathcal{N}$, then the agents considered in $\mathcal{C}$ may intend to form their own coalition $\mathcal{C}$. For instance, an agent $i$ can gain the profit such that $x_{i}' = x_{i} + \frac{V(\mathcal{C}) - \sum_{i \in \mathcal{C}} x_{i}}{|\mathcal{C}|}$, where $x_{i}$ is the payoff from the original outcome and coalition structure without considering that the coalition $\mathcal{C}$ is formed, and $x_{i}'$ is the new payoff after actually forming the coalition $\mathcal{C}$.

    \subsection{Shapley Value}
    \label{subsec:shapley_value}
        Shapley value \cite{shapley1953value} is a solution concept describing the notion of fairness, which is usually considered under the grand coalition. In the next paragraph, we will retrospect the procedure of forming Shapley Value step by step referred to the contents in \cite{chalkiadakis2011computational}. It is intuitive that each agent's payoff should reflect its contribution to accomplish the fairness. To realize this idea, a possible implementation is assigning each agent $i$ the payoff $V(\mathcal{N}) - V(\mathcal{N} \backslash \{i\})$. Nevertheless, the sum of payoff distributions via this payoff distribution scheme could be disparate from the $V(\mathcal{N})$. To address this problem, an ordering (permutation) of agents can be fixed and each agent receives the payoff according to the contribution to its predecessor coalition in this ordering, the payoff scheme of which is called marginal contribution. By this payoff distribution scheme, two agents that play the symmetric role could receive different payoffs, caused by the ordering selection. To eliminate the ordering dependence, Shapley proposed to average over all possible orderings of agents. The whole procedure above forms the final shape of Shapley Value, whose formal definition is shown as follows. Given a cooperative game $\Gamma=(\mathcal{N}, V)$, for any permutation $m \in \Pi(\mathcal{N})$, let $\delta_{i}^{m}(\mathcal{C}) = V(\mathcal{C}_{i}^{m} \cup \{i\}) - V(\mathcal{C}_{i}^{m})$ denote a marginal contribution of agent $i$ in the ordering $m$ to form the grand coalition, then the Shapley value of each agent $i$ can be written as follows:
        \begin{equation}
            \text{Sh}_{i}(\Gamma)= \frac{1}{|\mathcal{N}|!} \sum_{m \in \Pi(\mathcal{N})} \delta_{i}^{m}(\mathcal{C}).
        \label{eq:shapley_value_org_1}
        \end{equation}
        
        For the sake of the assumption of set-valued functions for characteristic value functions, Eq.~\ref{eq:shapley_value_org_1} can be equivalently transformed to the following expression such that
        \begin{equation}
            \begin{split}
                \text{Sh}_{i}(\Gamma) &= \frac{1}{|\mathcal{N}|!} \sum_{m \in \Pi(\mathcal{N})} \delta_{i}^{m}(\mathcal{C}) \\
                &= \frac{1}{|\mathcal{N}|!} \sum_{m \in \Pi(\mathcal{N})} V(\mathcal{C}_{i}^{m} \cup \{i\}) - V(\mathcal{C}_{i}^{m}) \\
                &= \frac{1}{|\mathcal{N}|!} \sum_{\mathcal{C} \subseteq \mathcal{N} \backslash \{i\}} |\mathcal{C}|!(|\mathcal{N}|-|\mathcal{C}|-1)! \left[ V(\mathcal{C} \cup \{i\}) - V(\mathcal{C}) \right] \\
                &\triangleq \sum_{\mathcal{C} \subseteq \mathcal{N} \backslash \{i\}} \frac{|\mathcal{C}|!(|\mathcal{N}|-|\mathcal{C}|-1)!}{|\mathcal{N}|!} \cdot \delta_{i}(\mathcal{C}).
            \end{split}
        \label{eq:shapley_value_org_2}
        \end{equation}
        
        In more details, the transformation is based on the fact that there are multiple repeated coalitions among all possible permutations of agents, which suggests that the Shapley value of agent $i$ is defined based on all its predecessor coalitions, i.e., $\mathcal{C} \subseteq \mathcal{N} \backslash \{i\}$. It is not difficult to check that each coalition $\mathcal{C}$ is repeated $|\mathcal{C}|!(|\mathcal{N}|-|\mathcal{C}|-1)!$ times. By defining the marginal contribution based on an arbitrary coalition $\mathcal{C}$ such that $\delta_{i}(\mathcal{C}) = V(\mathcal{C} \cup \{i\}) - V(\mathcal{C})$, Eq.~\ref{eq:shapley_value_org_2} turns out and it is the commonest form of Shapley value appears in the related literature.
        
        As Eq.~\ref{eq:shapley_value_org_2} shows, Shapley value takes the average of marginal contributions of all possible coalitions, so that it satisfies efficiency and fairness (i.e., sensitivity to dummy agents and symmetry) as Theorem \ref{thm:shapley_value_property} states. If we calculate the Shapley value for an agent, we have to consider $2^{\scriptscriptstyle |\mathcal{N}|}-1$ possible coalitions where the agent could join during the process of forming the grand coalition, which may lead to the computational catastrophe.
        \begin{theorem}[\cite{chalkiadakis2011computational} ]
        \label{thm:shapley_value_property}
            Shapley value is the only payoff distribution scheme with the following properties:
            \begin{itemize}
                \item[(1)] Efficiency: $\sum_{i \in \mathcal{N}} \text{Sh}_{i}(\Gamma) = V(\mathcal{N})$;
                \item[(2)] Identifiability of dummy agents: if agent $i$ is a dummy, then $\text{Sh}_{i}(\Gamma) = 0$;
                \item[(3)] Reflecting the contribution;\footnote{Note that this property was not described in the past introduction of Shapley value \cite{chalkiadakis2011computational}. Nevertheless, it is the underlying insight to construct the Shapley value, so we mention it here as a property.}
                \item[(4)] Symmetry: if agent $i$ and $j$ are symmetric in $\Gamma$, then $\text{Sh}_{i}(\Gamma) = \text{Sh}_{j}(\Gamma)$.
            \end{itemize}
        \end{theorem}
        
\section{Voltage Control in Electric Power Distribution Networks}
\label{sec:voltage_control_in_electric_power_distribution_networks}
    \subsection{Power Distribution Network}
    \label{subsec:power_distribution_network}
        \begin{figure}[ht!]
            \centering
            \includegraphics[width=1.\linewidth]{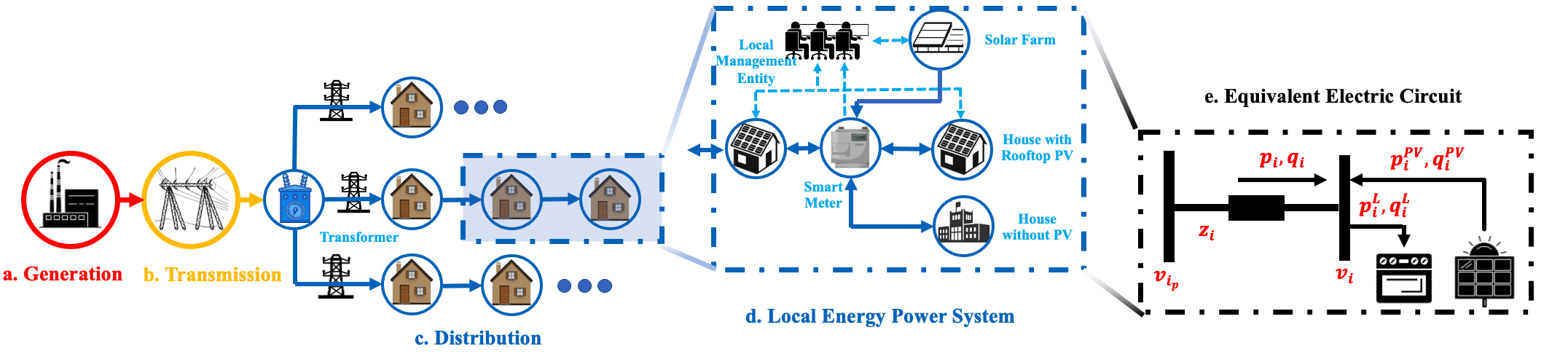}
            \caption{Illustration on distribution network (block a-b-c) under PV penetration. The solid and dotted lines represent the power and information flows respectively. Block d is the detailed version of distribution network and block e is the circuit model of block d.}
        \label{fig:block}
        \end{figure}
        
        An electric power distribution network is illustrated in Figure \ref{fig:block} stage a to c. The electricity is generated from power plant and transmitted through transmission lines. Muti-stage transformers are applied to reduce the voltage levels while the electricity is being delivered to the distribution network. The electricity is then consumed by residential and industrial clients. A typical PV unit consists of PV panels and voltage-source inverters which can be installed either on roof-top or in the solar farm. Conventionally, there exist management entities such as distributed system operator (DSO) monitoring and operating the PV resources through the local communication channels. With emergent PV penetration, distribution network gradually grows to be an active participant in power networks that can deliver power and service to its users and the main grid (see the bidirectional power flows in Figure \ref{fig:block} stage-d). Nevertheless, this also confronts the new challenges to the safety and resilience of power networks, since the active power generated by PV penetration is uncertain due to the weather conditions (i.e., difficult to be modelled in closed form). This motivates us to apply MARL as a solution to mitigate this issue in this thesis. 
    
    \subsection{System Model and Voltage Deviation}
    \label{subsec:system_model_and_voltage_deviation}
        In this thesis, we consider medium (10-24kV) and low (0.23-1kV) voltage distribution networks where PVs are highly penetrated. We model the distribution network in Figure \ref{fig:block} as a tree graph $\mathcal{G} = (V, E)$, where $V = \{0,1,\dots,N\}$ and $E = \{1,2,\dots,N\}$ represent the set of nodes (buses) and edges (branches) respectively \cite{gan2013optimal}. Bus 0 is considered as the connection to the main grid, balancing the active and reactive power in the distribution network. For each bus $i\in V$, let $\mathit{v}_i$ and $\theta_i$ be the magnitude and phase angle of the complex voltage and $\mathit{s}_j = p_i + j q_i$ be the complex power injection. Then the active and reactive power injection can be defined as follows:
        \begin{equation}
            \label{eq:power_injection}
            \begin{split}
                p_{i}^{\scriptscriptstyle{PV}} - p_{i}^{\scriptscriptstyle{L}} = v_i^2\sum_{j\in V_i}g_{i j}
                     -v_{i} \sum_{j\in V_i} v_{j}\left(g_{i j} \cos \theta_{i j}+b_{i j} \sin \theta_{i j}\right), \quad \forall i \in V \setminus \{0\} \\
                q_{i}^{\scriptscriptstyle{PV}} - q_{i}^{\scriptscriptstyle{L}} = -v_i^2\sum_{j\in V_i}b_{i j}
                + v_{i} \sum_{j\in V_i} v_{j}\left(g_{i j} \sin \theta_{i j} + b_{i j} \cos \theta_{i j}\right), \quad \forall i \in V \setminus \{0\}
            \end{split}
        \end{equation}
        where $V_i \coloneqq \{ j \ | \ (i,j) \in E \}$ is the index set of buses connected to bus $\mathit{i}$. $\mathit{g}_{i j}$ and $\mathit{b}_{i j}$ are the conductance and susceptance on branch $(\mathit{i},\mathit{j})$. $\theta_{i j} = \theta_i - \theta_j$ is the phase difference between bus $\mathit{i}$ and $\mathit{j}$. $p_{i}^{\scriptscriptstyle{PV}}$ and $q_{i}^{\scriptscriptstyle{PV}}$ are active power and reactive power of the PV on the bus $i$ (that are zeros if there is no PV on the bus $i$). $p_{i}^{\scriptscriptstyle{L}}$ and $q_{i}^{\scriptscriptstyle{L}}$ are active power and reactive power of the loads on the bus $i$ (that are zeros if there is no loads on the bus $i$). Eq.~\ref{eq:power_injection} can represent the power system dynamics which is essential for solving the power flow problem and active voltage control problem \cite{saadat1999power} (see details in Appendix A.1-A.2). For the safe and optimal operation, $5\%$ voltage deviation is usually allowed, i.e., $v_0 = 1.0$ per unit ($p.u.$) and $0.95 \ p.u. \leq v_i \leq 1.05 \ p.u., \forall i \in V \setminus \{0\}$. When the load is heavy during the nighttime, the end-user voltage could be smaller than $0.95 \ p.u.$ \cite{agalgaonkar2013distribution}. In contrast, to export its power, large penetration of $p_i^{\scriptscriptstyle PV}$ leads to reverse current flow that would increase $v_i$ out of the nominal range (Figure \ref{fig:block}-d) \cite{masters2002voltage, yang2015voltage}.
    
        \paragraph{Two-Bus Network Analysis.} To intuitively show up how voltage is varied by PVs and how PV inverters can participate in the voltage control, we give an example for a two-bus distribution network as shown in Figure \ref{fig:block_v}. In Figure \ref{fig:block_v}, $\mathit{z}_{0i} = r_{0i} + j x_{0i}$ represents the impedance on branch $(0, i)$; $\mathit{r}_{0i}$ and $\mathit{x}_{0i}$ are resistance and reactance on branch $(0, i)$, respectively; $\mathit{p}_i^{\scriptscriptstyle L}$ and $\mathit{q}_i^{\scriptscriptstyle L}$ denote active and reactive power consumption, respectively; $\mathit{p}_i^{\scriptscriptstyle PV}$ and $q_i^{\scriptscriptstyle PV}$ indicate active and reactive PV power generation, respectively. The parent bus voltage $v_{0}$ is set as the reference for the two-bus network.
        \begin{figure}[ht!]
            \centering
            \includegraphics[width = 0.55\linewidth]{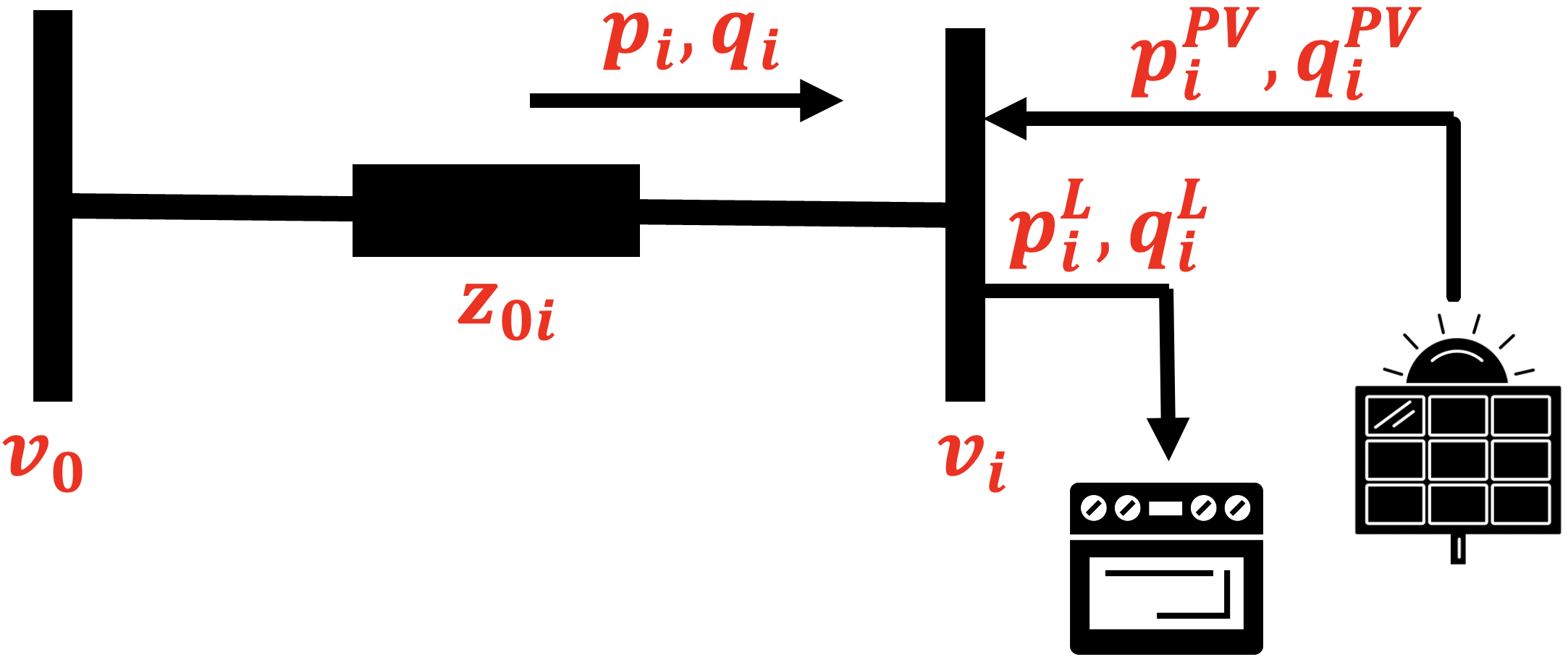}
            \caption{Two-bus electric circuit of the distribution network.}
            \label{fig:block_v}
        \end{figure}
        
        The voltage drop $\Delta v_{0i} = v_{0} - v_i$ in Figure \ref{fig:block_v} can be approximated as follows:
        \begin{equation}
        \label{eq:voltage_dev}
             \Delta v_{0i} = \frac{r_{0i} ( p_i^{\scriptscriptstyle L} - p_i^{\scriptscriptstyle PV} ) + x_{0i} ( q_i^{\scriptscriptstyle L} - q_i^{\scriptscriptstyle PV} ) }{ v_{i} }.
        \end{equation}
        
        The power loss of the 2-bus network in Figure \ref{fig:block_v} can be written as follows:
        \begin{equation}
        \label{eq:power_loss}
            P_{\text{loss}} = \frac{(p_i^{\scriptscriptstyle L} - p_i^{\scriptscriptstyle PV})^2 + (q_i^{\scriptscriptstyle L} - q_i^{\scriptscriptstyle PV})^2}{v_{0}^2} \cdot r_{0i}.
        \end{equation}
        
        \paragraph{Traditional Voltage Control Methods.} Conventionally, PVs are not allowed to participate in the voltage control so that $q_i^{\scriptscriptstyle PV}$ is restricted to 0 by the grid code. To export its power, large penetration of $p_i^{\scriptscriptstyle PV}$ may increase $v_i$ out of its safe range, causing reverse current flow \cite{masters2002voltage, yang2015voltage}. Voltage control devices, such as shunt capacitor (SC) and step voltage regulator (SVR) are usually equipped in the network to maintain the voltage level \cite{senjyu2008optimal}. Nonetheless, these methods cannot respond to intermittent solar radiation, e.g. frequent voltage fluctuation due to cloud cover \cite{singhal2019real}. Additionally, with the rising PV penetration in the network, the operation of traditional regulators would be at their control limit (i.e. runaway condition) \cite{agalgaonkar2014distribution}.
        
        \paragraph{Inverter-based Volt/Var Control.} To adapt to the continually rising PV penetration, grid-support services, such as voltage and reactive power control are required for every new-installed PV by the latest grid code IEEE Std-1547\textsuperscript{\texttrademark}-2018 \cite{ieee2018ieee}. For instance, the PV reactive power can be regulated by the PV inverter under partial static synchronous compensator (STATCOM) mode \cite{varma2018pv}. Depending on the voltage deviation levels, the inverter can inject or absorb different amount of reactive power exceeding its capacity  \cite{varma2009night}. This control method is then named as Volt/Var control, as the reactive power (with unit VAR) is determined by the voltage (with unit Volt). Intuitively by Eq.~\ref{eq:voltage_dev}, when the voltage increases due to large PV penetration in the lunch-time, the PV inverter absorbs reactive power, while during the night-time, the full inverter capacity is used to balance voltage fluctuation caused by increasing load \cite{agalgaonkar2014distribution}. Note that the only control variable in Eq.~\ref{eq:voltage_dev} and Eq.~\ref{eq:power_loss} is $q_i^{\scriptscriptstyle PV}$ which represents reactive power generated by the PV. Based on Eq.~\ref{eq:voltage_dev}, to enforce zero voltage deviation, the reactive power should satisfy the following condition such that
        \begin{equation}
        \label{eq:voltage_regu}
            q_i^{\scriptscriptstyle PV} = \frac{r_{0i}}{x_{0i}}(p_i^{\scriptscriptstyle L} - p_i^{\scriptscriptstyle PV})+q_i^{\scriptscriptstyle L}.
        \end{equation}
        
        Since the ratio $r_{0i}/x_{0i}$ in the distribution network is extremely large, $q_i^{\scriptscriptstyle PV}$ could become negative (i.e. absorbing reactive power) with great magnitude during the period of the peak PV injection (i.e., $p_i^{\scriptscriptstyle PV} \gg p_i^{\scriptscriptstyle L}$). From Eq.~\ref{eq:power_loss}, to achieve the least power loss, $q_i^{\scriptscriptstyle PV}$ needs to be equal to $q_i^{\scriptscriptstyle L}$ (i.e. no reactive power injection). This result may conflict with the voltage control target in Eq.~\ref{eq:voltage_regu}, implying that it is hard to simultaneously maintain safe voltage levels and minimise the power losses, even for the two-bus network. This section only demonstrates a 2-bus network which has linear relationship between voltage deviation and PV reactive power. Although the power systems in the real world are non-linear and more complex, they possess the same phenomenon on the contradiction between voltage control and power loss minimisation.
    
    \subsection{Foundation of Voltage Control in Power System}
    \label{subsec:foundation_of_voltage_control_in_power_system}
        \paragraph{Power Flow.} The power flow problem is designed to find the steady-state operation point of a power system. After measuring power injections $p_{i}^{\scriptscriptstyle PV} - p_{i}^{\scriptscriptstyle L}$ and $q_{i}^{\scriptscriptstyle PV} - q_{i}^{\scriptscriptstyle L}$, the bus voltages $v_i\angle \theta_i$ can be retrieved by iteratively solving Eq.~\ref{eq:power_injection} using Newton-Raphson or Gauss-Seidel method \cite{gomez2018electric}. The power plow serves as the fundamental role in grid planning and security assessment by locating any voltage deviations. It is also used as the system dynamics to generate observations during the runs of MARL.
        
        \paragraph{Optimal Power Flow.} Optimal power flow (OPF) is an optimization problem that aims at minimizing the total power loss subject to the power balance constraints defined in Eq.~\ref{eq:power_injection}, PV reactive power limits, and bus voltage limits \cite{gomez2018electric}. As the centralized OPF has full access to the system topology, measurements, and PV resources, it provides the optimal active voltage control performance and can be used as a benchmark method. However, the performance of the OPF highly depends on the accuracy of the grid model and the optimisation is time-consuming which makes it difficult to be deployed online.
        
        The OPF considered in this paper can be briefly formulated as:
        \begin{equation}
        \label{eq:opf}
            \begin{array}{ll}
            \min_{q_i^{\scriptscriptstyle PV}} & p_0 \\
            \text { s.t. } & \text{Eq.~}\ref{eq:power_injection} \\
            &  |q_i^{\scriptscriptstyle PV}| \leq q_{i,\max}^{\scriptscriptstyle PV},\quad i \in V^{\scriptscriptstyle PV} \\
            & v_{i, \min} \leq v_i \leq v_{i, \max}, \quad i \in V \setminus 0 \\
            & v_0 = v_{\text{ref}}, \\
            \end{array}
        \end{equation}
        where $p_0$ and $v_{0}$ are active power and reference voltage of the slack bus, respectively. $V^{\scriptscriptstyle PV}$ is the index set of the buses equipped with PVs. $p_i^{\scriptscriptstyle PV}$, $q_i^{\scriptscriptstyle PV}$, and $s_i$ are the active power, reactive power, and the capacity of PV at bus $i$,  respectively. In this thesis, each PV inverter is oversized with $s_i = 1.2 \ p_{i,\text{max}}, \forall i \in V^{\scriptscriptstyle PV}$. The maximum PV reactive power is $q_{i, \max}^{\scriptscriptstyle PV} = \sqrt{s_i^2 - ( p_i^{\scriptscriptstyle PV} )^2 }$. Note that the objective of the OPF problem is equivalent to minimizing the overall power loss. Eq.~\ref{eq:opf} may be infeasible due to the large penetration of PVs. In this case, slack variables can be added on the voltage constraint.
        
        \paragraph{Droop Control.} To regulate local voltage deviation, the standard droop control defines a piece-wise linear relationship between PV reactive power generation and voltage deviation at a bus equipped with inverter-based PVs \cite{singhal2019real,ieee2018ieee}. It is a fully decentralised control and ignore both the total voltage divisions and the power loss. The droop control, as recommended by IEEE Std-1547\textsuperscript{\texttrademark}-2018 \cite{ieee2018ieee}, follows the control strategy $q^{\scriptscriptstyle PV}_i=f(v_i)$, where $q^{\scriptscriptstyle PV}_i$ and $v_i$ are the PV reactive power and the voltage measurement of a PV bus $i$. $f(\cdot)$ is piecewise linear as shown in Figure \ref{fig:droop_control}. In more details, $v_{\text{ref}}$ represents the voltage set point (e.g. 1.0 p.u.). $v_\text{a}$ and $v_\text{d}$ represent the saturation regions limited by the PV inverter capacity and the current PV active power. There also exists a dead-band between $v_\text{b}$ and $v_\text{c}$ that does not require any control. For the voltage lower than $v_\text{b}$, the inverter provides reactive power proportional to the voltage deviation against $v_\text{ref}$. If the voltage is higher than $v_\text{c}$, the inverter absorbs reactive power until convergence achieves. The droop control only requires the local voltage measurements, which is simple and efficient to be implemented. However, it cannot directly minimise the power losses nor respond to fast voltage changes. For simplicity, we set $v_{\text{b}} = v_{\text{c}} = v_{\text{ref}}$ in this thesis.
        \begin{figure}[ht!]
            \centering
              \includegraphics[width=0.55\textwidth]{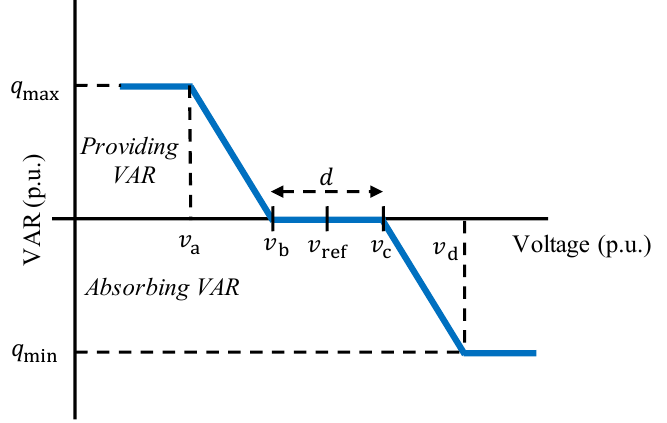}
            \caption{Illustration of the droop control law.}
        \label{fig:droop_control}
        \end{figure}
    
    \subsection{Active Voltage Control}
    \label{subsec:prior_works_on_voltage_control}
        \paragraph{Traditional Methods for Active Voltage Control.} Voltage rising and fluctuation problem in distribution networks has been studied for 20 years \cite{masters2002voltage}. The traditional voltage regulation devices such as OLTC and capacitor banks \cite{senjyu2008optimal} are often installed at substations and therefore may not be effective in regulating voltages at the far end of the line \cite{fusco2021decentralized}. The emergence of distributed generation, such as root-top PVs, introduces new approaches for voltage regulation by the active reactive power control of grid-connected inverters \cite{turistsyn2011options}. The state-of-the-art active voltage control strategies can be roughly classified into two categories: (1) reactive power dispatch based on optimal power flow (OPF) \cite{agalgaonkar2014distribution,vovos2007centralized}; and (2) droop control based on local voltage and power measurements \cite{jahangiri2013distributed,schiffer2016voltage}. Specifically, centralised OPF \cite{gan2013optimal,xu2017multi,anese2014optimal} minimises the power loss while fulfilling voltage constraints (e.g. power flow equation defined in Eq.~\ref{eq:power_injection}); distributed OPF \cite{zheng2016fully,tang2020distributed} used distributed optimization techniques, such as alternating direction method of multipliers (ADMM), to replace the centralised solver. The primary limitation of OPF is the need of exact system model \cite{sun2019review}. Besides, solving constrained optimisation problem is time-consuming, so it is difficult to respond to the rapid change of load profile \cite{singhal2019real}. On the other hand, droop control only depends on its local measurements, but its performance severely relies on the manually-designed parameters, which is often sub-optimal due to the lack of global information \cite{singhal2019real}. It is possible to enhance droop control by distributed algorithms, but extra communications are needed \cite{zeraati2018voltage,fusco2021decentralized}. In this thesis, we investigate the possibility of applying MARL to the active voltage control problem. Compared with the previous works on traditional methods, (1) MARL is model-free, so no exact system model is needed; and (2) the response of MARL is fast to handle the rapid changes of environments (e.g. the intermittency of renewable energy).
        
        \paragraph{Multi-Agent Reinforcement Learning for Active Voltage Control.} We now discuss the previous works that applied MARL to active voltage control problem in the power system community. \cite{cao2020multi,liu2021online} applied MADDPG with reactive power of inverters or static var compensators (SVCs) as control actions. \cite{wang2020data} applied MADDPG with a manually designed voltage inner loop, so that agents set reference voltage (instead of reactive power) as their control actions. \cite{cao2020distributed} applied MATD3 also with reactive power as control actions. \cite{cao2021data} applied MASAC, where both reactive power and the curtailment of active power are used as control actions. In the above works, distribution networks are divided into regions, with each region controlled by a single agent \cite{burger2019restructuring1,burger2019restructuring2}. It is not clear if these MARL approaches scales well for increasing number of agents. In particular, it is not clear if each single inverter in a distribution network can behave as an independent agent. In this thesis, we model the active voltage control problem as a decentralized partially observable Markov decision process (Dec-POMDP) \cite{oliehoek2016concise}, where each inverter is controlled by an agent. We propose Bowl-shape as a barrier function to represent voltage constraint as part of the reward.

\chapter{Theory and Method}
\label{cha:theory_and_method}
    In this chapter, we introduce the main contributions of this thesis. We firstly extend convex game to Markov decision process, named Markov convex game. Then, we justify that applying payoff distribution schemes in the Markov convex game as a credit assignment scheme to solve a global reward game is feasible. Thereby, we have showed that credit assignment is valid in a global reward game from the perspective of cooperative game theory. Next, we extend Shapley value, a payoff distribution scheme in the convex game, to Markov convex game, named Markov Shapley value, and incorporate the Markov Shapley value into the MARL algorithmic framework, forming a generic theoretical framework. Based on the theoretical framework, we derive three Shapley value based MARL algorithms called Shapley Q-learning (SHAQ), Shapley value deep deterministic policy gradient (SQDDPG) and Shapley value model-free proximal policy gradient (SMFPPO). Finally, we extend the theoretical framework to partial observability and show how the above Shapley value based MARL algorithms can be tweaked to solve partially observable tasks in practice.
\section{Markov Convex Game}
\label{sec:markov_convex_game}
    Extending from convex game, we now formally define Markov convex game (MCG) that can be described as a tuple $\Gamma = \left\langle \mathcal{N}, \mathcal{S}, \mathcal{A}, T, \Lambda, \pi, R_{t}, \gamma \right\rangle$. $\mathcal{N}$ is the set of all agents. $\mathcal{S}$ is the set of states and $\mathcal{A} = \mathlarger{\mathlarger{\times}}_{i \in \mathcal{N}} \mathcal{A}_{i}$ is the joint action set of all agents wherein $\mathcal{A}_{i}$ is each agent's action set. $T( \mathbf{s}', \mathbf{s}, \mathbf{a} ) = Pr(\mathbf{s}' | \mathbf{s}, \mathbf{a})$ is defined as the transition probability between two states. $\mathcal{CS} = \left\{ \mathcal{C}_{1}, ..., \mathcal{C}_{n} \right\}$ is a \textit{coalition structure}, where $\mathcal{C}_{i} \ \mathlarger{\subseteq} \ \mathcal{N}$ called a \textit{coalition} is a subset of all agents. $\Lambda$ is a collection of coalition structures. $\emptyset$ and $\mathcal{N}$ are two special cases of coalitions, i.e., the \textit{empty coalition} and the \textit{grand coalition} respectively. Conventionally, it is assumed that $\mathcal{C}_{m} \ \mathlarger{\mathlarger{\cap}} \ \mathcal{C}_{k} = \emptyset, \forall \mathcal{C}_{m}, \mathcal{C}_{k} \ \mathlarger{\mathlarger{\subseteq}} \ \mathcal{N}$. $\pi = \mathlarger{\mathlarger{\times}}_{\scriptscriptstyle i \in \mathcal{N}} \pi_{i}$ is the joint policy of all agents. For any coalition $\mathcal{C}$, a \textit{coalition policy} $\pi_{\scriptscriptstyle\mathcal{C}}(\mathbf{a}_{\scriptscriptstyle\mathcal{C}} | \mathbf{s}) = \mathlarger{\mathlarger{\times}}_{\scriptscriptstyle i \in \mathcal{C}} \pi_{i}(a_{i}|\mathbf{s})$ is defined over the \textit{coalition action set} $\mathcal{A}_{\scriptscriptstyle\mathcal{C}} = \mathlarger{\mathlarger{\times}}_{\scriptscriptstyle i \in \mathcal{C}} \mathcal{A}_{i}$. Therefore, $\pi$ can be seen as the \textit{grand coalition policy}. $R: \mathcal{S} \times \mathcal{A}_{\scriptscriptstyle\mathcal{C}} \rightarrow [0, \infty)$ as a characteristic function is the \textit{coalition reward function}. Accordingly, $R(\mathbf{s}, \mathbf{a})$ is the \textit{grand coalition reward} for $\textbf{S}_{t} = \textbf{s}$ and $\textbf{A}_{t} = \textbf{a}$ at time step $t$, which could be shortened as $\mathit{R}_{t}$ for conciseness in the rest of this thesis. $\gamma \in (0, 1)$ is the discounted factor. The infinite long-term discounted cumulative coalition rewards is defined as $V^{\pi_{\scriptscriptstyle\mathcal{C}}}(\mathbf{s}) = \mathbb{E}_{\pi_{\scriptscriptstyle\mathcal{C}}} \big[ \sum_{t=1}^{\infty} \gamma^{t-1} R_{t}(\mathbf{s}, \mathbf{a}_{\scriptscriptstyle\mathcal{C}}) \ | \ \mathbf{S}_{t}=\mathbf{s} \big] \in [0, \infty)$, called a \textit{coalition value}. Moreover, the empty coalition value $V^{\pi_{\emptyset}}(\mathbf{s}) = 0$ and $V^{\pi}(\mathbf{s})$ denotes the grand coalition value. If the grand coalition value is equivalent to a global reward, then the grand coalition value is equal to a global value. 
    
    The solution to MCG is finding a tuple $\left\langle \mathcal{CS}, \left( \max_{\pi_{i}} x_{i}(\mathbf{s}) \right)_{i \in \mathcal{N}} \right\rangle$, where $\left( \max_{\pi_{i}} x_{i}(\mathbf{s}) \right)_{i \in \mathcal{N}}$ indicates the \textit{payoff distribution} under the optimal joint policy given a coalition structure $\mathcal{CS}$. Under the assumption that $\mathcal{C}_{m} \ \mathlarger{\mathlarger{\cap}} \ \mathcal{C}_{k} = \emptyset, \forall \mathcal{C}_{m}, \mathcal{C}_{k} \ \mathlarger{\subseteq} \ \mathcal{N}$, the condition for MCG becomes the inequality as follows:
    \begin{equation}
        \max_{\pi_{\mathcal{C}_{\cup}}} V^{\pi_{\mathcal{C}_{\cup}}}(\mathbf{s}) \geq
        \max_{\pi_{\mathcal{C}_{m}}} V^{\pi_{\mathcal{C}_{m}}}(\mathbf{s})
        + \max_{\pi_{\mathcal{C}_{k}}} V^{\pi_{\mathcal{C}_{k}}}(\mathbf{s}), \quad
        \forall \mathcal{C}_{m}, \mathcal{C}_{k} \ \mathlarger{\subseteq} \ \mathcal{N}, \mathcal{C}_{\cup}=\mathcal{C}_{m} \ \mathlarger{\cup} \ \mathcal{C}_{k}.
    \label{eq:mcg_assumption}
    \end{equation}
    
    In an MCG with the grand coalition as the coalition structure, i.e., $\mathcal{CS} = \{ \mathcal{N} \}$, \textit{Markov core}, a solution concept describing stability, is defined as a set of payoff distribution schemes by which no agent has incentives to deviate from the grand coalition to gain more profits. Mathematically, Markov core is expressed as follows:
    \begin{equation}
        \texttt{MarkovCore}(\Gamma) = \left\{ \left( \max_{\pi_{i}} x_{i}(\mathbf{s}) \right)_{i \in \mathcal{N}} \ \bigg| \max_{\pi_{\mathcal{C}}} x(\mathbf{s}|\mathcal{C}) \geq
        \max_{\pi_{\mathcal{C}}} V^{\pi_{\mathcal{C}}}(\mathbf{s}),
        \forall \mathcal{C} \ \mathlarger{\subseteq} \ \mathcal{N}, \mathbf{s} \in \mathcal{S} \ \right\},
    \label{eq:epsilon_core}
    \end{equation}
    where $\max_{\pi_{\mathcal{C}}} x(\mathbf{s}|\mathcal{C}) = \sum_{i \in \mathcal{C}} \max_{\pi_{i}} x_{i}(\mathbf{s})$. The objective of the MCG is finding a payoff distribution scheme $\left( x_{i}(\mathbf{s}) \right)_{i \in \mathcal{N}}$ that finally converges to Markov core under the optimal joint policy.
    
    To assist the application on Q-learning, we similarly define \textit{coalition Q-value} as $Q^{\pi_{\mathcal{C}}}(\mathbf{s}, \mathbf{a}_{\scriptscriptstyle\mathcal{C}}) \in [0, +\infty)$ for all coalitions $\mathcal{C} \ \mathlarger{\mathlarger{\subset}} \ \mathcal{N}$. Following the above convention, the grand coalition Q-value (or the global Q-value) can be written as $Q^{\pi}(\mathbf{s}, \mathbf{a})$. Moreover, the optimal coalition Q-value of a coalition $\mathcal{C}$ with respect to the optimal joint policy of another coalition $\mathcal{D} \ \mathlarger{\subseteq} \ \mathcal{C}$ (i.e., $\pi_{\scriptscriptstyle\mathcal{D}}^{*}$) and the sub-optimal joint policy of the coalition $\mathcal{C} \backslash \mathcal{D}$ (i.e., $\pi_{\scriptscriptstyle\mathcal{C} \backslash \mathcal{D}}$) is defined as $Q^{\pi_{\mathcal{D}}^{*}}(\mathbf{s}, \mathbf{a}_{\scriptscriptstyle\mathcal{C}})$. Therefore, the optimal coalition Q-value of a coalition $\mathcal{C}$ with respect to the optimal joint policy of the coalition $\mathcal{C}$ is defined as $Q^{\pi_{\mathcal{C}}^{*}}(\mathbf{s}, \mathbf{a}_{\scriptscriptstyle\mathcal{C}})$. Similarly, the optimal global coalition Q-value with respect to the optimal joint policy of the grand coalition is denoted as $Q^{\pi^{*}}(\mathbf{s}, \mathbf{a})$.

    Eq.~\ref{eq:mcg_assumption} implies a fact existing in most real-life scenarios that a larger coalition results in the greater distributed payoffs (see Remark \ref{rmk:credit_assignment}) and therefore the greater optimal global value in cooperation, which directly increases agents' incentives to join the grand coalition. This interpretation for the dynamic scenario in this thesis is consistent with the static scenario given by \cite{shapley1971cores}, which is also known as the snowball effect.
    \begin{remark}
    \label{rmk:credit_assignment}
        Suppose that there are two coalitions $\mathcal{T}, \mathcal{S}$ such that $\mathcal{T} \ \mathlarger{\mathlarger{\subset}} \ \mathcal{S} \ \mathlarger{\mathlarger{\subset}} \ \mathcal{N}$ and an agent $i \in \mathcal{N} \backslash \mathcal{S}$. For convenience, we denote $\mathcal{C}_{1} = \mathcal{T} \ \mathlarger{\mathlarger{\cup}} \ \{i\}$ and $\mathcal{C}_{2} = \mathcal{S}$, and thus $\mathcal{C}_{\scriptscriptstyle\cap} = \mathcal{C}_{1} \ \mathlarger{\mathlarger{\cap}} \ \mathcal{C}_{2} = (\mathcal{T} \ \mathlarger{\mathlarger{\cup}} \ \{i\}) \ \mathlarger{\mathlarger{\cap}} \ \mathcal{S} = \mathcal{T}$ and $\mathcal{C}_{\scriptscriptstyle\cup} = \mathcal{C}_{1} \ \mathlarger{\mathlarger{\cup}} \ \mathcal{C}_{2} = (\mathcal{T} \ \mathlarger{\mathlarger{\cup}} \ \{i\}) \ \mathlarger{\mathlarger{\cup}} \ \mathcal{S} = \mathcal{S} \ \mathlarger{\mathlarger{\cup}} \ \{i\}$. By Eq.~\ref{eq:mcg_assumption}, we can write down the following inequalities such that
        \begin{equation}
            \begin{split}
                \max_{\pi_{\mathcal{S} \cup \{i\}}} V^{\pi_{\mathcal{S} \cup \{i\}}}(\mathbf{s}) - \max_{\pi_{\mathcal{S}}} V^{\pi_{\mathcal{S}}}(\mathbf{s}) &= \max_{\pi_{\mathcal{\mathcal{C}_{\cup}}}} V^{\pi_{\mathcal{C}_{\cup}}}(\mathbf{s}) - \max_{\pi_{\mathcal{C}_{2}}} V^{\pi_{\mathcal{C}_{2}}}(\mathbf{s}) \\
                &\geq \max_{\pi_{\mathcal{C}_{1}}} V^{\pi_{\mathcal{C}_{1}}}(\mathbf{s}) - \max_{\pi_{\mathcal{\mathcal{C}_{\cap}}}} V^{\pi_{\mathcal{C}_{\cap}}}(\mathbf{s}) \\
                &= \max_{\pi_{\mathcal{T} \cup \{i\}}} V^{\pi_{\mathcal{T} \cup \{i\}}}(\mathbf{s}) - \max_{\pi_{\mathcal{T}}} V^{\pi_{\mathcal{T}}}(\mathbf{s}).
            \end{split}
        \end{equation}
    It is intuitive to see that each agent can gain more payoffs if the size of the coalition grows.
    \end{remark}

\section{Validity of Credit Assignment in Global Reward Game}
\label{sec:validity_of_credit_assignment}
    In this section, we show up the connection between global reward game and Markov convex game, and then show the validity of credit assignment in a global reward game. More specifically, we mainly prove two following results: (1) Finding the payoff distribution under the grand coalition in Markov core also maximizing the social welfare (see Proposition \ref{prop:CFS_objective}); (2) solving a Markov convex game under the grand coalition is equivalent to solving a global reward game (see Theorem \ref{thm:eqv_mcg_grg}). Accordingly, the purpose of solving a global reward game is more evident and a global reward game is able to be represented as a Markov convex game under the grand coalition. Thereafter, it is feasible to solve a global reward game by a payoff distribution scheme defined in the Markov convex game (e.g. Markov Shapley value) as a credit assignment scheme. As a result, we have showed that credit assignment is valid to solve a global reward game from the perspective of cooperative game theory.
    
    \begin{lemma}[\cite{chalkiadakis2011computational} ]
        If an outcome $(\mathcal{CS}, \mathbf{x})$ is in the core of a characteristic function game $\langle \mathcal{N}, V \rangle$, then $V(\mathcal{CS}) \geq V(\mathcal{CS}')$ for every coalition structure $\mathcal{CS}' \in \mathcal{CS}_{\scriptscriptstyle \mathcal{N}}$, where $V(\mathcal{CS}) = \sum_{\mathcal{C} \in \mathcal{CS}} V(\mathcal{C})$ denotes the social welfare under the coalition structure $\mathcal{CS}$ and $\mathcal{CS}_{\scriptscriptstyle \mathcal{N}}$ is the set of all possible coalition structures.
    \label{lem:cfs_core->v}
    \end{lemma}
    
    \begin{proposition}
        For an Markov convex game, if an outcome with the grand coalition as the coalition structure (i.e., $\{ \mathcal{N} \}$) is in the Markov core, then it leads to the maximal social welfare, i.e., $\max_{\pi} V^{\pi}(\{\mathcal{N}\}) \geq \max_{\pi} V^{\pi}(\mathcal{CS}')$ for every coalition structure $\mathcal{CS}' \in \mathcal{CS}_{\scriptscriptstyle \mathcal{N}}$.
    \label{prop:CFS_objective}
    \end{proposition}
    
    \begin{proof}
        We first define $\max_{\pi} V^{\pi}(\mathbf{s}; \mathcal{CS}) = \sum_{\mathcal{C} \in \mathcal{CS}} \max_{\pi_{\mathcal{C}}} V^{\pi_{\mathcal{C}}}(\mathbf{s})$. Since Markov convex game is a characteristic function game and Markov core is an analogue of the core in convex game, if we define $\max_{\pi_{\mathcal{C}}} V^{\pi_{\mathcal{C}}}(\mathbf{s}) = \tilde{V}(\mathcal{C})$ for an arbitrary $\mathbf{s} \in \mathcal{S}$, the result is directly obtained.
    \end{proof}
    
    Proposition \ref{prop:CFS_objective} shows that solving payoff distribution under the grand coalition for satisfying Markov core also maximize the social welfare with respect to the coalition structure. This gives more insight into the reason why we need to use Markov core as the solution concept for solving Markov convex game.
    \begin{theorem}
        If an outcome with the grand coalition as the coalition structure (i.e., $\{ \mathcal{N} \}$) is in the Markov core, then solving a Markov convex game under the grand coalition is equivalent to solving a global reward game.
    \label{thm:eqv_mcg_grg}
    \end{theorem}
    
    \begin{proof}
        If an outcome with the grand coalition as the coalition structure (i.e., $\{ \mathcal{N} \}$) is in the Markov core, we can conclude that $\max_{\pi} V^{\pi}(\mathbf{s}; \{\mathcal{N}\})$ is the resulting objective function which we need to optimize to find the stationary optimal joint policy. Since $\max_{\pi} V^{\pi}(\mathbf{s}; \{\mathcal{N}\}) = \max_{\pi} V^{\pi}(\mathbf{s})$, we can observe that solving the objective function resulting from a Markov convex game is equivalent to solving the objective function of a global reward game.
    \end{proof}

    By Theorem \ref{thm:eqv_mcg_grg}, the result from Proposition \ref{prop:CFS_objective} can be regarded as a reason to clarify the motivation of studying the global reward game from the perspective of cooperative game theory. Furthermore, the payoff distribution scheme belonging to the Markov core obeys a relationship such that $\max_{\pi} V^{\pi}(\mathbf{s}) = \sum_{i \in \mathcal{N}} \max_{\pi_{i}} x_{i}(\mathbf{s})$ by definition. This result provides a justification for applying credit assignment (a.k.a. value decomposition or factorization in some literature \cite{SunehagLGCZJLSL18,RashidSWFFW18}) in a global reward game from the perspective of cooperative game theory, by unifying the concepts of payoff distribution scheme and credit assignment scheme. In the rest of this thesis, we may replace the concept ``payoff distribution'' with ``credit assignment'' for the ease of understanding.
    
    When solving a global reward game as a Markov convex game under the grand coalition, the goal is finding the optimal credit assignment scheme to satisfy the Markov core. This leads to the definition of Markov Shapley value which is the optimal credit assignment scheme reaching the Markov core, and therefore can be used to solve the global reward game. We will introduce the construction of Markov Shapley value in details in Section \ref{sec:shaley_value_for_multi-agent_reinforcement_learning}. Before that, we show the relationship among Markov convex game, global reward game and Markov Shapley value in a diagram as shown in Figure \ref{fig:mcg_grg_msv} for clarity.
    \begin{figure}[ht!]
        \centering
        \includegraphics[width = 0.95\linewidth]{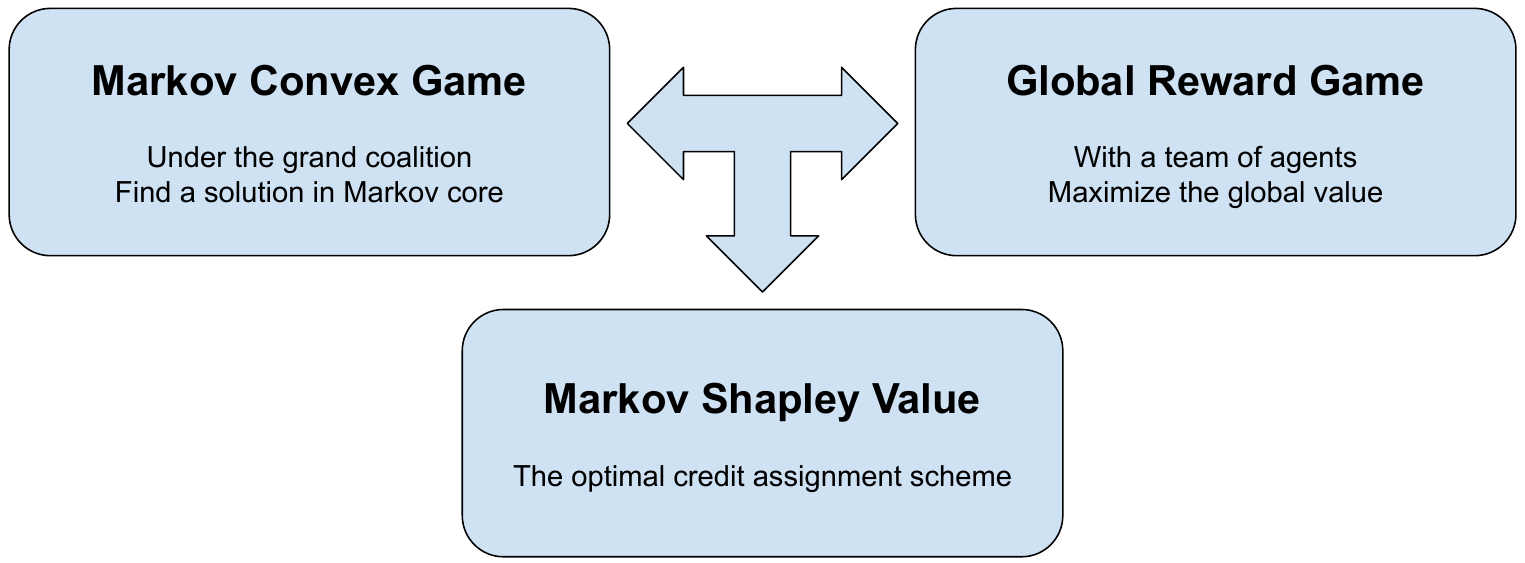}
        \caption{Solving a Markov convex game under the grand coalition by finding a solution in the Markov core, is equivalent to, solving a global reward game with a team of agents to maximize the global value. For this reason, Markov Shapley value belonging to the Markov core can be used as the optimal credit assignment scheme to solve a global reward game.}
    \label{fig:mcg_grg_msv}
    \end{figure}
        
\section{Shapley Value for Multi-Agent Reinforcement Learning}
\label{sec:shaley_value_for_multi-agent_reinforcement_learning}
    In this chapter, we introduce how Shapley value is incorporated into the framework of multi-agent reinforcement learning. Overall, we extend marginal contribution to Markov convex game. Next, we form Markov Shapley value in the Markov convex game by the new defined marginal contribution and prove that the properties of original Shapley value are inherited. Then, we derive a theoretical framework based on Bellman equation and Bellman operator, to regularize the optimality of Markov Shapley value and the approach to find it respectively. Furthermore, by approximating the components defined above, we derive Shapley value Q-learning, Shapley Q-value deep deterministic policy gradient, and Shapley model-free proximal policy gradient. Finally, we establish the relationship between the proposed framework and algorithms in this thesis, and prior works. To ease reading, the important definitions, conditions and theorems are wrapped with shading boxes which readers can follow to avoid losing themselves in the loads of mathematical languages.

    \subsection{Assumptions}
    \label{sec:assumptions_shaspley-q}
        \begin{assumption}
        \label{assm:basic_condition}
            We assume the following conditions hold: (1) The state space and the action space are finite and (2) The joint policy is stationary.
        \end{assumption}

        \begin{assumption}
        \label{assm:agent_policy_assumption}
            For the ease of analysis, in this thesis we assume that each agent's policy will not be affected by coalition formation. In other words, each agent's policy is regarded as its inherent feature, which is invariant throughout the interaction with other agents.
        \end{assumption}
        
        \begin{assumption}
        \label{assm:assumption_for_joint_policy_factorisation}
            Any coalition policy can be factorised into a permutation of decentralised (i.e., disjoint) policies, i.e., $\pi_{\scriptscriptstyle\mathcal{C}} = \mathlarger{\mathlarger{\times}}_{\scriptscriptstyle i \in \mathcal{C}} \pi_{i}$, where $\pi_{i}$ is agent $i$'s policy. Each $\pi_{\scriptscriptstyle\mathcal{C}}$ uniquely corresponds to a $V^{\pi_{\mathcal{C}}}(\mathbf{s})$ as a characteristic function (i.e., a set-valued function). Meanwhile, policies of the agents outside an coalition will not affect the coalition and therefore not change the coalition value (or coalition Q-value).
        \end{assumption}
        
        \begin{assumption}
        \label{assm:dummy_agent}
            If an agent $\mathit{i}$ is a dummy in an arbitrary state $\mathbf{s} \in \mathcal{S}$, it will not provide any contribution to any coalition $\mathcal{C}_{i} \ \mathlarger{\mathlarger{\subseteq}} \ \mathcal{N} \backslash \{i\}$ such that $V^{\pi_{\mathcal{C}}}(\mathbf{s}) = V^{\pi_{\mathcal{C} \cup \{i\}}}(\mathbf{s})$. Additionally, no members in the coalition $\mathcal{C}_{i}$ will react in another way after an agent $\mathit{i}$ joins.
        \end{assumption}
        
        \begin{assumption}
            If two arbitrary agent $i$ and agent $j$ are symmetric in an arbitrary state $\mathbf{s} \in \mathcal{S}$, $V^{\pi_{\mathcal{C} \cup \{i\}}}(\mathbf{s}) = V^{\pi_{\mathcal{C} \cup \{j\}}}(\mathbf{s})$ to any coalitions $\mathcal{C} \ \mathlarger{\mathlarger{\subseteq}} \ \mathcal{N} \backslash \{i, j\}$. Literally, the contributions of these two agents are equal to any coalition $\mathcal{C}$.
        \label{assm:symmetry}
        \end{assumption}
        
        \begin{assumption}
        \label{assm:max_shapley_value}
            For any agent $\mathit{i} \in \mathcal{N}$, its optimal Markov Shapley value in an arbitrary $\mathbf{s} \in \mathcal{S}$ denoted as $\max_{\pi_{i}} V_{i}^{\phi}(\mathbf{s})$ satisfies the following equation such that
            \begin{equation*}
                \max_{\pi_{i}} V_{i}^{\phi}(\mathbf{s}) = \sum_{\mathcal{C}_{i} \ \mathlarger{\mathlarger{\subseteq}} \ \mathcal{N} \backslash \{i\} } \frac{|\mathcal{C}_{i}|!(|\mathcal{N}|-|\mathcal{C}_{i}|-1)!}{|\mathcal{N}|!} \cdot \max_{\pi_{i}} \Phi_{i}(\mathbf{s} | \mathcal{C}_{i}),
            \end{equation*}
            where $\pi_{i}$ is agent $\mathit{i}$'s policy.
        \end{assumption}
        
        Assumption \ref{assm:basic_condition} shows the common conditions for the ease of analysis in the Markov decission process. Assumption \ref{assm:agent_policy_assumption} depicts a hypothesis that each agent's policy will not be affected by coalition formation, which supports the coalition policy factorisation shown in Assumption \ref{assm:assumption_for_joint_policy_factorisation}. Assumption \ref{assm:assumption_for_joint_policy_factorisation} is natural to hold given the chain rule in probability theory, the independence of each agent's policy and the definition of value function in reinforcement learning. Assumption \ref{assm:dummy_agent} and \ref{assm:symmetry} directly inherit the definitions from cooperative game theory \cite{chalkiadakis2011computational}. Assumption \ref{assm:max_shapley_value} inherits the definition from Shapley value \cite{shapley1953value} with extra consideration of agent $i$'s policy, an underlying condition of which is that the maximizer (i.e., $\pi_{i}$) of each $\Phi_{i}(\mathbf{s} | \mathcal{C}_{i}) \in \left\{ \Phi_{i}(\mathbf{s} | \mathcal{C}_{i}) | \mathcal{C}_{i} \ \mathlarger{\mathlarger{\subseteq}} \ \mathcal{N} \backslash \{i\} \right\}$ needs to be identical, in any $\mathbf{s} \in \mathcal{S}$. In other words, it implies that different permutations to from the grand coalition correspond to different long-term rewards that probably encode some unexpected events (i.e., each permutation is mapped to a marginal contribution of agent $i$), but with the same optimal policy as the solutions. Therefore, learning through Markov Shapley value is primarily for the fair credit assignment, with no changes to each agent's optimal policy. We argue for the existence of this condition by Example \ref{exp:max_shapley_value}.
        \begin{example}
        \label{exp:max_shapley_value}
            Suppose that there are two agents in total (i.e., $|\mathcal{N}|=2$), and we consider an arbitrary agent $i$ belonging to $\mathcal{N}$ whose action set is defined as $\mathcal{A}_{i} = \{ 0, 0.15, 0.25 \}$. Accordingly, there are only two intermediate coalitions for the agent $i$ to join and therefore two marginal contributions. To ease understanding, we only discuss a two-stage scenario and the result can be naturally extended to long-term scenarios. Agent $i$'s policy can be expressed as a sequence of actions such that $\pi_{i} = \left\langle a_{i}^{0}, a_{i}^{1} \right\rangle$. The set of marginal contributions of the agent $i$ is supposed to be a set $\mathcal{M}$ such that $$\mathcal{M} = \left\{ \Phi_{i}(\mathbf{s} | \{-i\}) := - (a_{i}^{0} + a_{i}^{1} - 0.5)^{2} + 1 + ||\mathbf{s}||_{2}, \Phi_{i}(\mathbf{s} | \emptyset) := \sin (a_{i}^{0} + a_{i}^{1}) + ||\mathbf{s}||_{2} \right\}.$$ Since $V_{i}^{\phi}(\mathbf{s}) = \frac{1}{2} \left( \Phi_{i}(\mathbf{s} | \{-i\}) + \Phi_{i}(\mathbf{s} | \emptyset) \right)$, it is easy to observe that Assumption \ref{assm:max_shapley_value} holds.
        \end{example}
        
    \subsection{Marginal Contribution}
    \label{subsec:preliminary_theoretical_results}
        By the view of cooperative game theory, the grand coalition is progressively formed by a permutation of agents. Accordingly, a marginal contribution measures the contribution of any agent $i$ to an arbitrary intermediate coalition, as shown in Definition \ref{def:marginal_contribution}.
        \mybox{
        \begin{definition}
        \label{def:marginal_contribution}
            In a Markov convex game, when a permutation of agents $\langle j_{1}, j_{2}, ..., j_{\scriptscriptstyle|\mathcal{N}|} \rangle, \forall j_{n} \in \mathcal{N}$ forms the grand coalition $\mathcal{N}$, where $n \in \{1, ..., |\mathcal{N}|\}, j_{a} \neq j_{b} \text{ if } a \neq b$, the marginal contribution of an agent $\mathit{i}$ is defined as the following equation such that
            \begin{equation}
                \Phi_{i}(\mathbf{s} | \mathcal{C}_{i}) = \max_{\pi_{\mathcal{C}_{i}}} V^{\pi_{\mathcal{C}_{i} \cup \{i\}}}(\mathbf{s}) - \max_{\pi_{\mathcal{C}_{i}}} V^{\pi_{\mathcal{C}_{i}}}(\mathbf{s}),
            \label{eq:marginal_contribution_v}
            \end{equation}
            where $\mathcal{C}_{i} = \{ j_{1}, ..., j_{n-1} \} \text{ for } j_{n}=i$ is an arbitrary intermediate coalition in which the agent $\mathit{i}$ would join during the process of forming the grand coalition.
        \end{definition}
        
        \begin{proposition}
            \label{prop:optimal_action_coalition_marginal_contribution}
            Agent $i$'s action marginal contribution can be derived as follows:
            \begin{equation}
                \begin{split}
                    \Upphi_{i}(\mathbf{s}, a_{i} | \mathcal{C}_{i}) 
                    = \max_{\mathbf{a}_{\scriptscriptstyle \mathcal{C}_{i}}} Q^{\pi_{\mathcal{C}_{i}}^{*}}(\mathbf{s}, \mathbf{a}_{\scriptscriptstyle\mathcal{C}_{i} \cup \{i\}})
                    - \max_{\mathbf{a}_{\mathcal{C}_{i}}} Q^{\pi_{\mathcal{C}_{i}}^*}(\mathbf{s}, \mathbf{a}_{\scriptscriptstyle\mathcal{C}_{i}}).
                \end{split}
            \label{eq:marginal_contribution_q}
            \end{equation}
        \end{proposition}
        \begin{proof}
            See the detailed proof in Appendix \ref{sec:proof_of_marginal_contribution}.
        \end{proof}
        }
        
        As Proposition \ref{prop:optimal_action_coalition_marginal_contribution} shows, an agent's action marginal contribution (analogous to Q-value) can be derived according to Eq.~\ref{eq:marginal_contribution_q}. It is usually more useful in solving MARL problems.
        
        We now introduce some preliminary results about the marginal contribution to support analysis of the Markov Shapley value. Although we show in Lemma \ref{lemm:condition_coalition_marginal_contribution} and Proposition \ref{prop:marginal_contribution_equal_value_factorisation} that the coalitional stability and efficiency of the marginal contribution are satisfied, the fairness cannot hold. This motivates us to further propose the Markov Shapley value to settle this problem.
        \begin{proposition}
        \label{prop:condition_coalition_marginal_contribution}
            $\forall \mathcal{C}_{i} \ \mathlarger{\mathlarger{\subseteq}} \ \mathcal{N}$ and $\forall \mathbf{s} \in \mathcal{S}$, Eq.~\ref{eq:mcg_assumption} is satisfied if and only if $\max_{\pi_{i}} \Phi_{i}(\mathbf{s}|\mathcal{C}_{i}) \geq 0$.
        \end{proposition}
        \begin{proof}
            See the detailed proof in Appendix \ref{sec:proof_of_marginal_contribution}.
        \end{proof}
        
        \begin{lemma}
        \label{lemm:condition_coalition_marginal_contribution}
            The optimal marginal contribution is a solution in the Markov core under a Markov convex game with the grand coalition.
        \end{lemma}
        \begin{proof}
            See the detailed proof in Appendix \ref{sec:proof_of_marginal_contribution}.
        \end{proof}
        
        \begin{proposition}
        \label{prop:marginal_contribution_equal_value_factorisation}
            In a Markov Convex Game with the grand coalition, the marginal contribution satisfies the property of efficiency: $\max_{\pi} V^{\pi}(\mathbf{s}) = \sum_{i \in \mathcal{N}} \max_{\pi_{i}} \Phi_{i}(\mathbf{s}|\mathcal{C}_{i})$.
        \end{proposition}
        \begin{proof}
            See the detailed proof in Appendix \ref{sec:proof_of_marginal_contribution}.
        \end{proof}
        
    \subsection{Markov Shapley Value}
    \label{subsec:generalised_shapley_value_for_mcg}
        It is apparent that a marginal contribution only considers one permutation to form the grand coalition. From the viewpoint of Shapley \cite{shapley1953value}, fairness is achieved through considering how much the an agent $i$ increases the optimal values (i.e. marginal contributions) of all possible coalitions when it joins in, i.e., $\max_{\pi_{\mathcal{C}_{i}}} V^{\pi_{\mathcal{C}_{i} \cup \{i\}}}(\mathbf{s}) - \max_{\pi_{\mathcal{C}_{i}}} V^{\pi_{\mathcal{C}_{i}}}(\mathbf{s}), \forall \mathcal{C}_{i} \ \mathlarger{\mathlarger{\subseteq}} \ \mathcal{N} \backslash \{i\}$. Following the same philosophy, we define the Shapley value under Markov convex game based on the novel marginal contributions and action marginal contributions defined in the Markov convex game (see Definition \ref{def:marginal_contribution} and Proposition \ref{prop:optimal_action_coalition_marginal_contribution}), as shown in Definition \ref{def:shapley_value}, named as \textit{Markov Shapley value} (MSV).
        \mybox{
        \begin{definition}
            Markov Shapley value is represented as 
            \begin{equation}
                V^{\phi}_{i}(\mathbf{s}) = \sum_{\mathcal{C}_{i} \ \mathlarger{\mathlarger{\subseteq}} \ \mathcal{N} \backslash \{i\} } \frac{|\mathcal{C}_{i}|!(|\mathcal{N}|-|\mathcal{C}_{i}|-1)!}{|\mathcal{N}|!} \cdot \Phi_{i}(\mathbf{s} | \mathcal{C}_{i}).
            \label{eq:shapley_value}
            \end{equation}
            With the deterministic policy, Markov Shapley value can be equivalently represented as
            \begin{equation}
                Q^{\phi}_{i}(\mathbf{s}, a_{i}) = \sum_{\mathcal{C}_{i} \ \mathlarger{\mathlarger{\subseteq}} \ \mathcal{N} \backslash \{i\} } \frac{|\mathcal{C}_{i}|!(|\mathcal{N}|-|\mathcal{C}_{i}|-1)!}{|\mathcal{N}|!} \cdot \Upphi_{i}(\mathbf{s}, a_{i} | \mathcal{C}_{i}).
            \label{eq:shapley_q_value}
            \end{equation}
            where $\Phi_{i}(\mathbf{s} | \mathcal{C}_{i})$ is defined in Eq.~\ref{eq:marginal_contribution_v} and $\Upphi_{i}(\mathbf{s}, a_{i} | \mathcal{C}_{i})$ is defined in Eq.~\ref{eq:marginal_contribution_q}.
        \label{def:shapley_value}
        \end{definition}
        }
        
        For convenience, we name Eq.~\ref{eq:shapley_q_value} as \textit{Markov Shapley Q-value} (MSQ). Briefly speaking, MSV calculates the weighted average of marginal contributions. Since a coalition may repeatedly appear among all permutations (i.e., $|\mathcal{N}|!$ permutations), the ratio between the occurrence frequency $|\mathcal{C}_{i}|!(|\mathcal{N}|-|\mathcal{C}_{i}|-1)!$ and the total frequency $|\mathcal{N}|!$ is used as a weight to describe the importance of the corresponding marginal contribution. Besides, the sum of all weights is equal to 1, so each weight can be interpreted as a probability measure. Consequently, MSV can be seen as the expectation of marginal contributions, denoted as $\mathbb{E}_{\mathcal{C}_{i} \sim Pr(\mathcal{C}_{i} | \mathcal{N} \backslash \{i\})}\left[ \Phi_{i}(\mathbf{s} | \mathcal{C}_{i}) \right]$. Note that $Pr(\mathcal{C}_{i} | \mathcal{N} \backslash \{i\})$ is a bell-shaped probability distribution. By the above relationship, Remark \ref{rmk:coalition_generation} is directly obtained.
        \begin{remark}
        \label{rmk:coalition_generation}
            Uniformly sampling different permutations is equivalent to sampling from $Pr(\mathcal{C}_{i} | \mathcal{N} \backslash \{i\})$, since the coalition generation is led by the permutations to form the grand coalition.
        \end{remark}
        
        Proposition \ref{prop:shapley_value_properties} shows three properties of MSV. The most important property is Property (ii) that aids the formulation of Shapley-Bellman optimality equation. Property (iii) provides a fundamental mechanism to quantitatively describe ``fairness'' among agents. Property (i) and (iii) play important roles in interpretation of credit assignment (or value factorisation). Property (iv) indicates that if two agents are symmetric, then their optimal MSVs should be equal, but the reverse does not necessarily hold. All these properties that define the fairness inherit the properties of the original Shapley value \cite{shapley1953value}.
        \begin{proposition}
        \label{prop:shapley_value_properties}
            Markov Shapley value possesses properties as follows: (i) identifiability of dummy agents: $V_{i}^{\phi}(\mathbf{s}) = 0$; (ii) efficiency: $\max_{\pi} V^{\pi}(\mathbf{s}) = \sum_{i \in \mathcal{N}} \max_{\pi_{i}} V_{i}^{\phi}(\mathbf{s})$; (iii) reflecting the contribution; and (iv) symmetry.
        \end{proposition}
        \begin{proof}
            See the detailed proof in Appendix \ref{sec:proof_of_markov_shapley_value}.
        \end{proof}
        
    \subsection{Shapley Value Based Multi-Agent Q-Learning}
    \label{subsec:shapley_q_learning}
        \mybox{
        \begin{condition}
            The conditions that enable Shapley-Bellman optimality equation to hold are shown as follows:
            \begin{itemize}
                \item[\textbf{C.1.}] Efficiency of MSV (i.e. the result from Proposition \ref{prop:shapley_value_properties});
                \item[\textbf{C.2.}] $Q^{\phi^{*}}_{i}(\mathbf{s}, a_{i}) = w_{i}(\mathbf{s}, a_{i}) \ Q^{\pi^{*}}(\mathbf{s}, \mathbf{a}) - b_{i}(\mathbf{s})$, where $w_{i}(\mathbf{s}, a_{i}) > 0$ and $b_{i}(\mathbf{s}) \geq 0$ are bounded and $\sum_{i \in \mathcal{N}} w_{i}(\mathbf{s}, a_{i})^{-1} b_{i}(\mathbf{s}) = 0$,
            \end{itemize}
        \label{cond:shapley_bellman_optimality_equation}
        \end{condition}
        }
        \paragraph{Shapley-Bellman Optimality Equation.} Based on Bellman optimality equation \cite{bellman1952theory} and Condition \ref{cond:shapley_bellman_optimality_equation}, we derive \textit{Shapley-Bellman optimality equation} (SBOE) for evaluating the optimal MSQ (i.e. an equivalent form to the optimal MSV) such that
        \begin{equation}
        \label{eq:shapley_q_optimality_equation}
            \mathbf{Q}^{\phi^{*}}(\mathbf{s}, \mathbf{a}) = \mathbf{w}(\mathbf{s}, \mathbf{a}) \sum_{\mathbf{s}' \in \mathcal{S}} Pr(\mathbf{s}' | \mathbf{s}, \mathbf{a}) \left[
            R
            + \ 
            \gamma \sum_{i \in \mathcal{N}} \max_{a_{i}} Q_{i}^{\phi^{*}}(\mathbf{s}', a_{i}) \right] - \mathbf{b}(\mathbf{s}),
        \end{equation}
        where $\mathbf{w}(\mathbf{s}, \mathbf{a}) = [w_{i}(\mathbf{s}, a_{i})]^{\top} \in \mathbb{R}^{\scriptscriptstyle|\mathcal{N}|}_{+}$; $\mathbf{b}(\mathbf{s}) = [b_{i}(\mathbf{s})]^{\top} \in \mathbb{R}^{\scriptscriptstyle|\mathcal{N}|}_{\geq 0}$; $\mathbf{Q}^{\phi^{*}}(\mathbf{s}, \mathbf{a}) = [Q^{\phi^{*}}_{i}(\mathbf{s}, a_{i})]^{\top} \in \mathbb{R}^{\scriptscriptstyle|\mathcal{N}|}_{\geq 0}$ and $Q^{\phi^{*}}_{i}(\mathbf{s}, a_{i})$ denotes the optimal MSQ. If Eq.~\ref{eq:shapley_q_optimality_equation} holds, the optimal MSQ is achieved. Moreover, as shown in Proposition \ref{prop:equiv_credit_assignment}, for any $\mathbf{s} \in \mathcal{S}$ and $\mathit{a}_{i}^{*} = \arg\max_{a_{i}} Q^{\phi^{*}}_{i}(\mathbf{s}, a_{i})$, we have a solution $w_{i}(\mathbf{s}, a_{i}^{*}) = 1 / |\mathcal{N}|$.\footnote{Note that this is only one solution of $w_{i}$ for the optimal action. In other words, there exist other solutions but perhaps less interpretable.} In other words, the assigned credits would be equal and each agent would receive $Q^{\pi^{*}}(\mathbf{s}, \mathbf{a}) / |\mathcal{N}|$ if performing the optimal actions, for which the efficiency still holds. This can be interpreted as an extremely fair credit assignment such that the credit to each agent is not discriminated if all of them perform optimally, regardless of their roles.\footnote{The equal credit assignment was also revealed by \cite{wang2020towards} from another perspective of analysis.} Nevertheless, $w_{i}(\mathbf{s}, a_{i})$ for $\mathit{a}_{i} \neq \arg\max_{a_{i}} Q^{\phi^{*}}_{i}(\mathbf{s}, a_{i})$ still needs to be learned.
        \begin{proposition}
        \label{prop:equiv_credit_assignment}
            For any $\mathbf{s} \in \mathcal{S}$ and $\mathit{a}_{i}^{*} = \arg\max_{a_{i}} Q^{\phi^{*}}_{i}(\mathbf{s}, a_{i})$, we have a solution $w_{i}(\mathbf{s}, a_{i}^{*}) = 1 / |\mathcal{N}|$.
        \end{proposition}
        \begin{proof}
            See the detailed proof in Appendix \ref{sec:proof_of_markov_shapley_value}.
        \end{proof}
        
        \paragraph{Shapley-Bellman Operator.} To find an optimal solution described in Eq.~\ref{eq:shapley_q_optimality_equation}, we now propose an operator called \textit{Shapley-Bellman operator} (SBO), i.e., $\mathlarger{\Upsilon}: \mathlarger{\mathlarger{\times}}_{i \in \mathcal{N}} Q_{i}^{\phi}(\mathbf{s}, a_{i}) \mapsto \mathlarger{\mathlarger{\times}}_{i \in \mathcal{N}} Q_{i}^{\phi}(\mathbf{s}, a_{i})$, which is formally defined as follows:
        \begin{equation}
            \mathlarger{\Upsilon} \left( \mathlarger{\mathlarger{\times}}_{i \in \mathcal{N}} Q_{i}^{\phi}(\mathbf{s}, a_{i}) \right) = \mathbf{w}(\mathbf{s}, \mathbf{a}) \sum_{\mathbf{s}' \in \mathcal{S}} Pr(\mathbf{s}' | \mathbf{s}, \mathbf{a}) \left[
            R + \ \gamma \sum_{i \in \mathcal{N}} \max_{a_{i}} Q_{i}^{\phi}(\mathbf{s}', a_{i}) \right] - \mathbf{b}(\mathbf{s}),
        \label{eq:shapley_q_operator}
        \end{equation}
        where $w_{i}(\mathbf{s}, a_{i}) = 1 / |\mathcal{N}|$ when $\mathit{a}_{i} = \arg\max_{a_{i}} Q^{\phi}_{i}(\mathbf{s}, a_{i})$. We prove that the optimal joint deterministic policy can be achieved by recursively running SBO in Theorem \ref{thm:shapley_q_optimal} that is proved based on the results from Lemma \ref{lemm:shapley_q_contraction_mapping} and Corollary \ref{coro:shapley_q_fixed_point}.
        \begin{lemma}
        \label{lemm:shapley_q_contraction_mapping}
            For all $\mathbf{s} \in \mathcal{S}$ and $\mathbf{a} \in \mathcal{A}$, Shapley-Bellman operator is a contraction mapping in a non-empty complete metric space when $\max_{\mathbf{s}} \big\{ \sum_{i \in \mathcal{N}} \max_{a_{i}} w_{i}(\mathbf{s}, a_{i}) \big\} < \frac{1}{\gamma}$.
        \end{lemma}
        \begin{proof}
            See the detailed proof in Appendix \ref{sec:proof_of_shapley-bellman_operator}.
        \end{proof}
        
        \begin{corollary}
        \label{coro:shapley_q_fixed_point}
            According to Banach fixed-point theorem \cite{banach1922operations}, Shapley-Bellman operator admits a unique fixed point. Moreover, starting by an arbitrary start point, the sequence recursively generated by Shapley-Bellman operator can finally converge to that fixed point.
        \end{corollary}
        \begin{proof}
            Since $\langle \mathbb{R}^{|\mathcal{N}|\times |\mathcal{S}| |\mathcal{A}|}, ||\cdot||_{1} \rangle$ is a non-empty complete metric space and Shapley-Bellman operator $\mathlarger{\Upsilon}$ is shown as a contraction mapping in Lemma \ref{lemm:shapley_q_contraction_mapping}, by Banach fixed-point theorem \cite{banach1922operations} we can directly conclude that Shapley-Bellman operator $\mathlarger{\Upsilon}$ admits the unique fixed point. Furthermore, starting by an arbitrary start point, the sequence recursively generated by Shapley-Bellman operator $\mathlarger{\Upsilon}$ can finally converge to that fixed point.
        \end{proof}
            
        \mybox{
        \begin{theorem}
        \label{thm:shapley_q_optimal}
            Shapley-Bellman operator converges to the optimal Markov Shapley Q-value and the corresponding optimal joint deterministic policy when $\max_{\mathbf{s}} \big\{ \sum_{i \in \mathcal{N}} \max_{a_{i}} w_{i}(\mathbf{s}, a_{i}) \big\} < \frac{1}{\gamma}$.
        \end{theorem}
        \begin{proof}
            By Corollary \ref{coro:shapley_q_fixed_point}, we get that Shapley-Bellman operator admits the unique fixed point. Since Shapley-Bellman optimality equation (i.e., Eq.~\ref{eq:shapley_q_optimality_equation}) is obviously a fixed point for Shapley-Bellman operator, it is not difficult to draw a conclusion that the optimal Markov Shapley Q-value is achieved. Since the sum of the optimal Markov Shapley Q-values is equal to the optimal global Q-value and the optimal global Q-value corresponds to the optimal joint deterministic policy, we show that the optimal joint deterministic policy is achieved. Furthermore, it is obvious that Shapley-Bellman optimality equation can be transformed back to the Bellman optimality equation with respect to the optimal global Q-value, given the property of efficiency of Markov Shapley value.
        \end{proof}
        }
                
        \paragraph{Stochastic Approximation of Shapley-Bellman Operator.} We now derive the stochastic approximation of Shapley-Bellman operator over the value space, i.e. a form of Q-learning derived from Shapley-Bellman operator. By sampling from $Pr(\mathbf{s}'|\mathbf{s}, \mathbf{a})$ via Monte Carlo method, the Q-learning algorithm can be expressed as follows:
        \begin{equation}
            \mathbf{Q}^{\phi}_{t+1}(\mathbf{s}, \mathbf{a}) \leftarrow \mathbf{Q}^{\phi}_{t}(\mathbf{s}, \mathbf{a}) + \alpha_{t}(\mathbf{s}, \mathbf{a}) \left[ \mathbf{w}(\mathbf{s}, \mathbf{a}) \left( R_{t} + \gamma \sum_{i \in \mathcal{N}} \max_{a_{i}} (Q_{i}^{\phi})_{t}(\mathbf{s}', a_{i}) \right) - \mathbf{b}(\mathbf{s}) - \mathbf{Q}^{\phi}_{t}(\mathbf{s}, \mathbf{a}) \right].
        \label{eq:shapley_q_learning_primal}
        \end{equation}
        
        \begin{lemma}[\cite{jaakkola1994convergence} ]
            The random process $\{\Delta_{t}\}$ taking values $\mathbb{R}^{n}$ defined as follows: $$\Delta_{t+1}(x) = (1 - \alpha_{t}(x)) \Delta_{t}(x) + \alpha_{t}(x) F_{t}(x),$$ which converges to 0 w.p.1 under the following assumptions:
            \begin{itemize}
                \item $0 \leq \alpha_{t} \leq 1$, $\sum_{t} \alpha_{t}(x) = \infty$ and $\sum_{t} \alpha_{t}^{2} \leq \infty$;
                \item $|| \mathbb{E}[F_{t}(x) | \mathcal{F}_{t}] ||_{W} \leq \delta || \Delta_{t} ||_{W}$, with $0 \leq \delta < 1$;
                \item $\textbf{var} [F_{t}(x) | \mathcal{F}_{t}] \leq C ( 1 + ||\Delta_{t}||_{W}^{2} )$, for $C > 0$.
            \end{itemize}
        \label{lemm:stochastic_process}
        \end{lemma}
    
        \mybox{
        \begin{theorem}
        \label{thm:proof_of_shapley_q_learning}
            For a Markov convex game, the Q-learning algorithm derived by Shapley-Bellman operator given by the update rule such that
            \begin{equation*}
                \mathbf{Q}^{\phi}_{t+1}(\mathbf{s}, \mathbf{a}) \leftarrow \mathbf{Q}^{\phi}_{t}(\mathbf{s}, \mathbf{a}) + \alpha_{t}(\mathbf{s}, \mathbf{a}) \left[ \mathbf{w}(\mathbf{s}, \mathbf{a}) \left( R_{t} + \gamma \sum_{i \in \mathcal{N}} \max_{a_{i}} (Q_{i}^{\phi})_{t}(\mathbf{s}', a_{i}) \right) - \mathbf{b}(\mathbf{s}) - \mathbf{Q}^{\phi}_{t}(\mathbf{s}, \mathbf{a}) \right],
            \end{equation*}
            converges w.p.1 to the optimal Markov Shapley Q-value if 
            \begin{equation}
                \sum_{t} \alpha_{t}(\mathbf{s}, \mathbf{a}) = \infty \ \ \ \ \ \ \ \ \sum_{t} \alpha^{2}_{t}(\mathbf{s}, \mathbf{a}) \leq \infty
            \label{eq:alpha_condition}
            \end{equation}
            for all $\mathbf{s} \in \mathcal{S}$ and $\mathbf{a} \in \mathcal{A}$ as well as $\max_{\mathbf{s}} \big\{ \sum_{i \in \mathcal{N}} \max_{a_{i}} w_{i}(\mathbf{s}, a_{i}) \big\} < \frac{1}{\gamma}$.
        \end{theorem}
        \begin{proof}
            See the detailed proof in Appendix \ref{sec:proof_of_shapley-bellman_operator}.
        \end{proof}
        }
        
        \paragraph{Shapley Q-Learning.} For the easy implementation, we conduct transformation for the stochastic approximation of SBO and derive \textit{Shapley Q-learning} (SHAQ) whose TD error is shown as follows:
        \begin{equation}
            \begin{split}
                \Delta(\mathbf{s}, \mathbf{a}, \mathbf{s}') = R \ + \ \gamma \sum_{i \in \mathcal{N}} \max_{a_{i}} Q_{i}^{\phi}(\mathbf{s}', a_{i})
                - \sum_{i \in \mathcal{N}} \delta_{i}(\mathbf{s}, a_{i}) \ Q^{\phi}_{i}(\mathbf{s}, a_{i}),
            \end{split}
        \label{eq:td_error_shapley_q_learning}
        \end{equation}
        where 
        \begin{equation}
            \delta_{i}(\mathbf{s}, a_{i}) = \begin{cases} 
                                                  1 & a_{i} = \arg\max_{a_{i}} Q^{\phi}_{i}(\mathbf{s}, a_{i}), \\
                                                  \alpha_{i}(\mathbf{s}, a_{i}) & a_{i} \neq \arg\max_{a_{i}} Q^{\phi}_{i}(\mathbf{s}, a_{i}).
                                             \end{cases}
        \label{eq:delta}
        \end{equation}
        Note that the closed-form expression of $\delta_{i}(\mathbf{s}, a_{i})$ is written as $|\mathcal{N}|^{-1} w_{i}(\mathbf{s}, a_{i})^{-1}$. If inserting the condition that $w_{i}(\mathbf{s}, a_{i}) = 1 / |\mathcal{N}|$ when $\mathit{a}_{i} = \arg\max_{a_{i}} Q^{\phi}_{i}(\mathbf{s}, a_{i})$ as well as defining $\delta_{i}(\mathbf{s}, a_{i})$ as $\alpha_{i}(\mathbf{s}, a_{i})$ when $\mathit{a}_{i} \neq \arg\max_{a_{i}} Q^{\phi}_{i}(\mathbf{s}, a_{i})$, Eq.~\ref{eq:delta} is obtained. The term $\mathbf{b}(\mathbf{s})$ is cancelled in Eq.~\ref{eq:td_error_shapley_q_learning} thanks to the condition that $\sum_{i \in \mathcal{N}} w_{i}(\mathbf{s}, a_{i})^{-1} b_{i}(\mathbf{s}) = 0$. Note that the condition to $w_{i}(\mathbf{s}, a_{i})$ in Theorem \ref{thm:shapley_q_optimal} should hold for the convergence of SHAQ in implementation. The details of the derivation of Shapley Q-Learning are shown in the following paragraphs. 
        
        By stochastic approximation in value space, i.e. sampling $\mathbf{s}'$ from $Pr(\mathbf{s}'|\mathbf{s}, \mathbf{a})$ via Monte Carlo method, Shapley-Bellman operator can be expressed as follows:
        \begin{equation}
            \mathbf{Q}^{\phi}(\mathbf{s}, \mathbf{a}) = \mathbf{w}(\mathbf{s}, \mathbf{a}) \left( R + \gamma \sum_{i \in \mathcal{N}} \max_{a_{i}} Q_{i}^{\phi}(\mathbf{s}', a_{i}) \right)  - \mathbf{b}(\mathbf{s}),
        \label{eq:stochstic_approximation_shapley_operator}
        \end{equation}
        where $\mathbf{w}(\mathbf{s}, \mathbf{a}) = [w_{i}(\mathbf{s}, a_{i})]^{\top} \in \mathbb{R}^{\scriptscriptstyle|\mathcal{N}|}_{+}$; $\mathbf{b}(\mathbf{s}) = [b_{i}(\mathbf{s})]^{\top} \in \mathbb{R}^{\scriptscriptstyle|\mathcal{N}|}_{+}$; and $\mathbf{Q}^{\phi}(\mathbf{s}, \mathbf{a}) = [Q_{i}^{\phi}(\mathbf{s}, a_{i})]^{\top} \in \mathbb{R}^{\scriptscriptstyle|\mathcal{N}|}_{+}$. 
        
        Since $\mathbf{w}(\mathbf{s}, \mathbf{a}) = diag \big( \mathbf{w}(\mathbf{s}, \mathbf{a}) \big) \ \mathbf{1}$, where $diag(\cdot)$ denotes the diagonalization of a vector\footnote{It is a square diagonal matrix with the elements of vector v on the main diagonal, and the other entries of the matrix are zeros.} and $\mathbf{1}$ denotes the vector of ones, Eq.~\ref{eq:stochstic_approximation_shapley_operator} can be equivalently represented as follows:
        \begin{equation}
             \mathbf{Q}^{\phi}(\mathbf{s}, \mathbf{a}) = diag \big( \mathbf{w}(\mathbf{s}, \mathbf{a}) \big) \ \mathbf{1} \ \left( R + \gamma \sum_{i \in \mathcal{N}} \max_{a_{i}} Q_{i}^{\phi}(\mathbf{s}', a_{i}) \right) - \mathbf{b}(\mathbf{s}).
        \label{eq:stochstic_approximation_shapley_operator_0}
        \end{equation}
        Since $w_{i}(\mathbf{s}, a_{i}) > 0, \forall i \in \mathcal{N}$, we can write the following equivalent form to Eq.~\ref{eq:stochstic_approximation_shapley_operator_0} such that
        \begin{equation} 
            diag \big( \mathbf{w}(\mathbf{s}, \mathbf{a}) \big)^{-1} \mathbf{Q}^{\phi}(\mathbf{s}, \mathbf{a}) = \mathbf{1} \ \left( R + \gamma \sum_{i \in \mathcal{N}} \max_{a_{i}} Q_{i}^{\phi}(\mathbf{s}', a_{i}) \right) - diag \big( \mathbf{w}(\mathbf{s}, \mathbf{a}) \big)^{-1} \mathbf{b}(\mathbf{s}).
        \label{eq:stochstic_approximation_shapley_operator_1}
        \end{equation}
        Next, we multiply $\mathbf{1}^{\top}$ on both sides and obtain the following equation such that
        \begin{equation}
            \sum_{i \in \mathcal{N}} \frac{1}{w_{i}(\mathbf{s}, a_{i})} \cdot Q_{i}^{\phi}(\textbf{s}, a_{i}) = |\mathcal{N}| \ \left( R + \gamma \sum_{i \in \mathcal{N}} \max_{a_{i}} Q_{i}^{\phi}(\mathbf{s}', a_{i}) \right) - \sum_{i \in \mathcal{N}} w_{i}(\mathbf{s}, a_{i})^{-1} b_{i}(\mathbf{s}).
        \label{eq:stochstic_approximation_shapley_operator_2}
        \end{equation}
        
        Since the condition that $\sum_{i \in \mathcal{N}} w_{i}(\mathbf{s}, a_{i})^{-1} b_{i}(\mathbf{s}) = 0$, by dividing $|\mathcal{N}|$ on both sides we get the following equation such that
        \begin{equation}
            \sum_{i \in \mathcal{N}} \frac{1}{|\mathcal{N}|w_{i}(\mathbf{s}, a_{i})} \cdot Q_{i}^{\phi}(s, a_{i}) = R + \gamma \sum_{i \in \mathcal{N}} \max_{a_{i}} Q_{i}^{\phi}(s, a_{i}).
        \label{eq:stochstic_approximation_shapley_operator_3}
        \end{equation}
        
        Since $w_{i}(\mathbf{s}, a_{i}) = 1 / |\mathcal{N}|$ when $\mathit{a}_{i} = \arg\max_{a_{i}} Q^{\phi}_{i}(\mathbf{s}, a_{i})$, by defining $\delta_{i}(\mathbf{s}, a_{i}) = \frac{1}{|\mathcal{N}| \ w_{i}(\mathbf{s}, a_{i})}$ we can get that 
        \begin{equation}
            \delta_{i}(\mathbf{s}, a_{i}) = \begin{cases} 
                                                  1 & a_{i} = \arg\max_{a_{i}} Q^{\phi}_{i}(\mathbf{s}, a_{i}), \\
                                                  \alpha_{i}(\mathbf{s}, a_{i}) & a_{i} \neq \arg\max_{a_{i}} Q^{\phi}_{i}(\mathbf{s}, a_{i}),
                                             \end{cases}
        \label{eq:delta_1}
        \end{equation}
        where $\alpha_{i}(\mathbf{s}, a_{i})$ is a variable that expresses $\frac{1}{|\mathcal{N}| \ w_{i}(\mathbf{s}, a_{i})}$ when $a_{i} \neq \arg\max_{a_{i}} Q^{\phi}_{i}(\mathbf{s}, a_{i})$ for the ease of implementation.
        
        Substituting Eq.~\ref{eq:delta_1} into Eq.~\ref{eq:stochstic_approximation_shapley_operator_3}, we can get the following equation such that
        \begin{equation}
            \sum_{i \in \mathcal{N}} \delta_{i}(\mathbf{s}, a_{i}) \ Q^{\phi}_{i}(\mathbf{s}, a_{i}) = R + \gamma \sum_{i \in \mathcal{N}} \max_{a_{i}} Q_{i}^{\phi}(\mathbf{s}', a_{i}).
        \label{eq:dual_shapley_q_operator}
        \end{equation}
        
        By rearranging Eq.~\ref{eq:dual_shapley_q_operator}, we obtain the TD error of Shapley Q-learning (SHAQ) such that
        \begin{equation}
            \Delta(\mathbf{s}, \mathbf{a}, \mathbf{s}') = R + \gamma \sum_{i \in \mathcal{N}} \max_{a_{i}} Q_{i}^{\phi}(\mathbf{s}', a_{i}) - \sum_{i \in \mathcal{N}} \delta_{i}(\mathbf{s}, a_{i}) \ Q^{\phi}_{i}(\mathbf{s}, a_{i}).
        \label{eq:td_error_shapley_q_learning_appendix}
        \end{equation}
        
        The solution resulting from the TD error of SHAQ is necessary for that resulting from the TD error of Eq.~\ref{eq:shapley_q_learning_primal} (i.e. the stochastic learning process that we proved to converge to the optimal Markov Shapley Q-value in Theorem \ref{thm:proof_of_shapley_q_learning}), and the condition $\max_{\mathbf{s}} \big\{ \sum_{i \in \mathcal{N}} \max_{a_{i}} w_{i}(\mathbf{s}, a_{i}) \big\} < \frac{1}{\gamma}$ is required to be satisfied so that the convergence to the optimality is possible to hold.
        
        Next, we give a proof to show that the optimal MSV is a solution in the Markov core, as Theorem \ref{thm:shapley_value_core} shows. As a result, \textit{solving SBOE is equivalent to solving the Markov core and SHAQ is actually a learning algorithm that reliably converges to the Markov core}. As per the definition of Markov core, we can say that SHAQ leads to the result that no agents have incentives to deviate from the grand coalition. Additionally, by Proposition \ref{prop:CFS_objective} we can conclude that reaching the Markov core is equivalent to reaching the optimal social welfare with respect to coalition structures. These two statements provide an interpretation of credit assignment for a global reward game. C.2 in Condition \ref{cond:shapley_bellman_optimality_equation} \textit{maintains the validity of the relationship between the optimal MSQ and the optimal global Q-value if there exist dummy agents} (see Remark \ref{rmk:dummy_agents}), so that the definition of SBOE is valid for MCG and MSQ in almost every case, which preserves the completeness of the theory.
        \begin{lemma}
        \label{lemm:markov_core_convex_set}
            Markov core is a convex set.
        \end{lemma}
        \begin{proof}
            See the detailed proof in Appendix \ref{sec:validity_and_interpretability}.
        \end{proof}
        
        \begin{theorem}
        \label{thm:shapley_value_core}
            The optimal Markov Shapley value is a solution in the Markov core for a Markov convex game under the grand coalition.
        \end{theorem}
        \begin{proof}
            The optimal Markov Shapley value is the affine combination of the optimal marginal contributions. We know that Markov core is a convex set by Lemma \ref{lemm:markov_core_convex_set} and the optimal marginal contribution is in the Markov core by Lemma \ref{lemm:condition_coalition_marginal_contribution}. Since the affine combination of the points in a convex set is still in this convex set, we get that the optimal Markov Shapley value is still in the Markov core.
        \end{proof}
        
        \begin{remark}
        \label{rmk:dummy_agents}
            For an arbitrary state $\mathbf{s} \in \mathcal{S}$, by C.2 in Condition \ref{cond:shapley_bellman_optimality_equation} it is not difficult to check that even if an arbitrary agent $i$ is a dummy (i.e., $Q^{\phi^{*}}_{i}(\mathbf{s}, a_{i}) = 0$ for some $i \in \mathcal{N}$), $Q^{\pi^{*}}(\mathbf{s}, \mathbf{a})$ and $Q^{\phi^{*}}_{j}(\mathbf{s}, a_{j}), \forall j \neq i$ would not be zero if $b_{i}(\mathbf{s}) \neq 0$. If the extreme case happens that for an arbitrary state $\mathbf{s} \in \mathcal{S}$ all agents are dummies, since $\sum_{i \in \mathcal{N}} w_{i}(\mathbf{s}, a_{i})^{-1} b_{i}(\mathbf{s}) = 0$ we are allowed to set $b_{i}(\mathbf{s}) = 0, \forall i \in \mathcal{N}$ so that $Q^{\pi^{*}}(\mathbf{s}, \mathbf{a}) = 0$ and the property of efficiency such that $\max_{\mathbf{a}} Q^{\pi^{*}}(\mathbf{s}, \mathbf{a}) = \sum_{i \in \mathcal{N}} \max_{a_{i}} Q_{i}^{\phi^{*}}(\mathbf{s}, a_{i})$ is still valid.
        \end{remark}
        
        \paragraph{Implementation of SHAQ.} We now describe a practical implementation of SHAQ for decentralized partially observable Markov decision process (Dec-POMDP) \cite{oliehoek2012decentralized} (i.e., a global reward game but with partial observations). First, the global state is replaced by the history of each agent to guarantee the existence of the optimal deterministic joint policy \cite{oliehoek2012decentralized}. Accordingly, Markov Shapley Q-value is denoted as $Q_{i}^{\phi}(\tau_{i}, a_{i})$, wherein $\tau_{i}$ is a history of partial observations of an arbitrary agent $i$. Since the paradigm of centralised training decentralised execution (CTDE) \cite{oliehoek2008optimal} is applied, the global state (i.e., $\mathbf{s} \in \mathcal{S}$) for $\hat{\alpha}_{i}(\mathbf{s}, a_{i})$ can be obtained during learning.
        \begin{proposition}
        \label{prop:shapley_value_approximate}
            Suppose that any action marginal contribution can be factorised into the form such that $\Upphi_{i}(\mathbf{s}, a_{i} | \mathcal{C}_{i}) = \sigma(\mathbf{s}, \mathbf{a}_{ \scriptscriptstyle\mathcal{C}_{i} \cup \{i\} }) \ \hat{Q}_{i}(\mathbf{s}, a_{i})$. With the condition that
            \begin{equation*}
                \mathbb{E}_{\mathcal{C}_{i} \sim Pr(\mathcal{C}_{i} | \mathcal{N} \backslash \{i\})} \left[ \sigma(\mathbf{s}, \mathbf{a}_{ \scriptscriptstyle\mathcal{C}_{i} \cup \{i\} }) \right] =
                \begin{cases} 
                     1 & \ \ a_{i} = \arg\max_{a_{i}} Q^{\phi}_{i}(\mathbf{s}, a_{i}), \\
                     K \in (0, 1) & \ \ a_{i} \neq \arg\max_{a_{i}} Q^{\phi}_{i}(\mathbf{s}, a_{i}),
                \end{cases}
            \end{equation*}
            we have
            \begin{equation}
                \begin{cases} 
                     Q_{i}^{\phi}(\mathbf{s}, a_{i}) = \hat{Q}_{i}(\mathbf{s}, a_{i}) & \ \ a_{i} = \arg\max_{a_{i}} \hat{Q}_{i}(\mathbf{s}, a_{i}), \\
                     \alpha_{i}(\mathbf{s}, a_{i}) \ Q^{\phi}_{i}(\mathbf{s}, a_{i}) = \hat{\alpha}_{i}(\mathbf{s}, a_{i}) \ \hat{Q}_{i}(\mathbf{s}, a_{i}) & \ \ a_{i} \neq \arg\max_{a_{i}} \hat{Q}_{i}(\mathbf{s}, a_{i}),
                \end{cases}
            \label{eq:shapley_q_approximate}
            \end{equation}
            where $\hat{\alpha}_{i}(\mathbf{s}, a_{i}) = \mathbb{E}_{\mathcal{C}_{i} \sim Pr(\mathcal{C}_{i} | \mathcal{N} \backslash \{i\})} \left[ \hat{\psi}_{i}(\mathbf{s}, a_{i}; \mathbf{a}_{ \scriptscriptstyle\mathcal{C}_{i} }) \right]$ and $\hat{\psi}_{i}(\mathbf{s}, a_{i}; \mathbf{a}_{ \scriptscriptstyle\mathcal{C}_{i} }) := \alpha_{i}(\mathbf{s}, a_{i}) \ \sigma(\mathbf{s}, \mathbf{a}_{ \scriptscriptstyle\mathcal{C}_{i} \cup \{i\} })$.
        \end{proposition}
        \begin{proof}
            See the detailed proof in Appendix \ref{sec:proof_of_the_implementation_of_shapley_q-learning}.    
        \end{proof}
        
        Compatible with the decentralised execution, we use only one parametric function $\hat{Q}_{i}(\tau_{i}, a_{i})$ to directly approximate $Q_{i}^{\phi}(\tau_{i}, a_{i})$. Under some conditions (see Proposition \ref{prop:shapley_value_approximate}), the information about coalition formation can be equivalently transferred to $\hat{\psi}_{i}(\mathbf{s}, a_{i}; \mathbf{a}_{ \scriptscriptstyle\mathcal{C}_{i} })$. As a result, $\delta_{i}(\mathbf{s}, a_{i})$ is equivalent to the form as follows:
        \begin{equation}
            \hat{\delta}_{i}(\mathbf{s}, a_{i}) = \begin{cases} 
                                                  1 & a_{i} = \arg\max_{a_{i}} \hat{Q}_{i}(\mathbf{s}, a_{i}), \\
                                                  \hat{\alpha}_{i}(\mathbf{s}, a_{i}) & a_{i} \neq \arg\max_{a_{i}} \hat{Q}_{i}(\mathbf{s}, a_{i}),
                                             \end{cases}
        \label{eq:delta_new}
        \end{equation}
        where $\hat{\alpha}_{i}(\mathbf{s}, a_{i}) = \mathbb{E}_{\mathcal{C}_{i} \sim Pr(\mathcal{C}_{i} | \mathcal{N} \backslash \{i\})} \left[ \hat{\psi}_{i}(\mathbf{s}, a_{i}; \mathbf{a}_{ \scriptscriptstyle\mathcal{C}_{i} }) \right]$. To solve partial observability, $\hat{Q}_{i}(\tau_{i}, a_{i})$ is empirically represented as recurrent neural network (RNN) with GRUs \cite{chung2014empirical}. $\hat{\psi}_{i}(\mathbf{s}, a_{i}; \mathbf{a}_{ \scriptscriptstyle\mathcal{C}_{i} })$ is directly approximated by a parametric function $\mathlarger{F}_{\mathbf{s}} + 1$ and thus $\hat{\alpha}_{i}(\mathbf{s}, a_{i})$ can be expressed as the following equation such that
        \begin{equation}
            \hat{\alpha}_{i}(\mathbf{s}, a_{i}) = \frac{1}{M}\sum_{k = 1}^{M} \mathlarger{F}_{\mathbf{s}} \Big( \hat{Q}_{\mathcal{C}_{i}^{k}}(\tau_{\mathcal{C}_{i}^{k}}, \mathbf{a}_{\mathcal{C}_{i}^{k}}), \ \hat{Q}_{i}(\tau_{i}, a_{i}) \Big) + 1,
        \label{eq:alpha_deep_representation}
        \end{equation}
        where $\hat{Q}_{\mathcal{C}_{i}^{k}}(\tau_{\mathcal{C}_{i}^{k}}, \mathbf{a}_{\mathcal{C}_{i}^{k}}) = \frac{1}{|\mathcal{C}_{i}^{k}|}\sum_{j \in \mathcal{C}_{i}^{k}} \hat{Q}_{j}(\tau_{j}, a_{j})$ and $\mathcal{C}_{i}^{k}$ is sampled $\mathit{M}$ times from $\mathit{Pr}(\mathcal{C}_{i} | \mathcal{N} \backslash \{i\})$ (implemented as Remark \ref{rmk:coalition_generation} suggests) to approximate $\mathbb{E}_{\mathcal{C}_{i} \sim \mathit{Pr}(\mathcal{C}_{i} | \mathcal{N} \backslash \{i\})}[ \hat{\psi}_{i}(\mathbf{s}, a_{i}; \mathbf{a}_{ \scriptscriptstyle\mathcal{C}_{i} }) ]$ using Monte Carlo approximation; and $\mathlarger{F}_{\mathbf{s}}$ is a monotonic function, with an absolute activation function on the output, whose weights are generated from hyper-networks with the global state as the input. We show that Eq.~\ref{eq:alpha_deep_representation} satisfies the condition that $\max_{\mathbf{s}} \big\{ \sum_{i \in \mathcal{N}} \max_{a_{i}} w_{i}(\mathbf{s}, a_{i}) \big\} < \frac{1}{\gamma}$ (see Proposition \ref{prop:hatalpha_satisfying_condition}) and therefore the implementation of SHAQ is reliable and trustworthy.
        \begin{proposition}
        \label{prop:hatalpha_satisfying_condition}
            $\hat{\alpha}_{i}(\mathbf{s}, a_{i})$ satisfies the condition $\max_{\mathbf{s}} \big\{ \sum_{i \in \mathcal{N}} \max_{a_{i}} w_{i}(\mathbf{s}, a_{i}) \big\} < \frac{1}{\gamma}$.
        \end{proposition}
        \begin{proof}
            See the detailed proof in Appendix \ref{sec:proof_of_the_implementation_of_shapley_q-learning}.
        \end{proof}
        
        Using the framework of fitted Q-learning \cite{ernst2005tree} to solve an extremely large number of states (which could be usually infinite) and plugging in the above designed modules, the practical least-square-error loss function derived from Eq.~\ref{eq:td_error_shapley_q_learning} can be expressed as follows:
        \begin{equation}
            \begin{split}
                \min_{\theta, \lambda} \mathbb{E}_{\mathbf{s}, \mathbf{\tau}, \mathbf{a}, R, \mathbf{\tau}'} \left[ \ \left( \ R \ + \ 
                \gamma \sum_{i \in \mathcal{N}} \max_{a_{i}} \hat{Q}_{i}(\tau_{i}', a_{i}; \theta^{-})
                - 
                \sum_{i \in \mathcal{N}} \hat{\delta}_{i}(\mathbf{s}, a_{i}; \lambda) \ \hat{Q}_{i}(\tau_{i}, a_{i}; \theta) \ \right)^{2} \ \right],
            \end{split}
        \label{eq:deep_shapley_q_learning_loss}
        \end{equation}
        where all agents share the parameters of $\hat{Q}_{i}(\mathbf{s}, a_{i}; \theta)$ and $\hat{\alpha}_{i}(\mathbf{s}, a_{i}; \lambda)$ respectively; and $\hat{Q}_{i}(\mathbf{s}', a_{i}; \theta^{-})$ works as the target Q-value function where $\theta^{-}$ is periodically updated. The learning procedure follows the paradigm of DQN \cite{mnih2015human}, with a replay buffer to store the online collection of agents' experiences. To provide an overview of Shapley Q-learning, we present the pseudo code in Algorithm \ref{alg:shapley_q}.
        \begin{algorithm}[ht!]
            \caption{Shapley Q-learning}
            \label{alg:shapley_q}
            \scalebox{0.95}{
            \begin{minipage}{1.0\linewidth}
            \begin{algorithmic}[1]
                \STATE Initialise a set of agents $\mathcal{N}$ and set $N = |\mathcal{N}|$
                \STATE Initialise $\hat{Q}_{i}(\tau_{i}, a_{i}; \theta)$ with the shared parameters among agents
                \STATE Initialise $\hat{\alpha}_{i}(\mathbf{s}, a_{i}; \lambda)$ with the shared parameters among agents
                \STATE Initialise $\hat{Q}_{i}(\tau_{i}, a_{i}; \theta^{-})$ by copying $\hat{Q}_{i}(\tau_{i}, a_{i}; \theta)$ with the shared parameters among agents
                \STATE Initialise a replay buffer $\mathcal{B}$
                \REPEAT
                    \STATE Initialise a container $\mathcal{E}$ for storing an episode
                    \STATE Observe an initial global state $\mathbf{s}^{1}$ and each agent's partial observation $\mathit{o}_{i}^{1}$ from an environment
                    \FOR{t=1:T}
                        \STATE Get $\tau_{i}^{t} = (o_{i}^{m})_{m=1:t}$ for each agent
                        \STATE For each agent $\mathit{i}$, select an action 
                        \begin{equation*}
                            \mathit{a}_{i}^{t} = \begin{cases}
                                                \text{a random action} & \text{with probability }\epsilon \\
                                                \arg\max_{a_{i}} \hat{Q}^{*}_{i}(\tau_{i}^{t}, a_{i}; \theta) & \text{otherwise}
                                             \end{cases}
                        \end{equation*}
                        \STATE Execute $\mathit{a}_{i}^{t}$ of each agent to get the global reward $\mathit{R}^{t}$, $\mathbf{s}^{t+1}$ and each agent's $\mathit{o}_{i}^{t+1}$
                        \STATE Store $\big\langle \mathbf{s}^{t}, (o_{i}^{t})_{i=1:N}, (a_{i}^{t})_{i=1:N}, R^{t}, \mathbf{s}^{t+1}, (o_{i}^{t+1})_{i=1:N} \big\rangle$ to $\mathcal{E}$
                    \ENDFOR
                    \STATE Store $\mathcal{E}$ to $\mathcal{B}$
                    \STATE Sample a batch of episodes with batch size B from $\mathcal{B}$
                        \FOR{each sampled episode}
                            \FOR{k=1:T}
                                \STATE Get each transition $\big\langle \mathbf{s}^{k}, (o_{i}^{k})_{i=1:N}, (a_{i}^{k})_{i=1:N}, R^{k}, \mathbf{s}^{k+1}, (o_{i}^{k+1})_{i=1:N} \big\rangle$
                                \STATE For each agent $\mathit{i}$, get $\tau_{i}^{k} = (o_{i}^{m})_{m=1:k}$
                                \STATE For each agent $\mathit{i}$, calculate $\hat{Q}_{i}(\tau_{i}^{k}, a_{i}^{k}; \theta)$
                                \STATE For each agent $\mathit{i}$, calculate $\alpha_{i}(\mathbf{s}^{k}, a_{i}^{k}; \lambda)$ by Algorithm \ref{alg:getting_alpha}
                                \STATE For each agent $\mathit{i}$, calculate $\delta_{i}(\mathbf{s}^{k}, a_{i}^{k}; \lambda)$ as follows:
                                    \begin{equation*}
                                        \hat{\delta}_{i}(\mathbf{s}^{k}, a_{i}^{k}; \lambda) = \begin{cases} 
                                                                              1 & a_{i}^{k} = \arg\max_{a_{i}} \hat{Q}_{i}(\mathbf{s}^{k}, a_{i}; \theta) \\
                                                                              \hat{\alpha}_{i}(\mathbf{s}^{k}, a_{i}^{k}; \lambda) & a_{i}^{k} \neq \arg\max_{a_{i}} \hat{Q}_{i}(\mathbf{s}^{k}, a_{i}; \theta) \ \ (\text{via Algorithm \ref{alg:getting_alpha}})
                                                                         \end{cases}
                                    \end{equation*}
                                \STATE For each agent $\mathit{i}$, get $\tau_{i}^{k+1} = (o_{i}^{m})_{m=1:k+1}$
                                \STATE For each agent $\mathit{i}$, get $\mathit{a}_{i}^{k+1}$ by $\arg\max_{a_{i}} \hat{Q}_{i}(\tau_{i}^{k+1}, a_{i}; \theta)$
                                \STATE For each agent $\mathit{i}$, calculate $\hat{Q}_{i}(\tau_{i}^{k+1}, a_{i}^{k+1}; \theta^{-})$
                            \ENDFOR
                        \ENDFOR
                        \STATE Construct a loss as follows:
                            \begin{equation*}
                                \begin{split}
                                    \min_{\theta, \lambda} \frac{1}{B} \sum_{k=1}^{B} \Big[ \ \big( \ R^{k} \ + \ 
                                    \gamma \sum_{i \in \mathcal{N}} \max_{a_{i}^{k}} \hat{Q}_{i}(\tau_{i}^{k+1}, a_{i}^{k+1}; \theta^{-}) - 
                                    \sum_{i \in \mathcal{N}} \hat{\delta}_{i}(\mathbf{s}^{k}, a_{i}^{k}; \lambda) \ \hat{Q}_{i}(\tau_{i}^{k}, a_{i}^{k}; \theta) \ \big)^{2} \ \Big]
                                \end{split}
                            \end{equation*}
                        \STATE Update $\theta$ and $\lambda$ through the above loss
                        \STATE Periodically update $\theta^{-}$ by copying $\theta$
               \UNTIL{$\hat{Q}_{i}(\tau_{i}, a_{i}; \theta)$ converges}
            \end{algorithmic}
            \end{minipage}
            }
        \end{algorithm}
        
        \begin{algorithm}[ht!]
            \caption{Calculating $\hat{\alpha}_{i}(\mathbf{s}, a_{i})$}
            \label{alg:getting_alpha}
            \scalebox{0.95}{
            \begin{minipage}{1.0\linewidth}
            \begin{algorithmic}[1]
               \STATE {\bfseries Input:} $\mathbf{s}$, $\big( \hat{Q}_{i}(\tau_{i}, a_{i}; \theta) \big)_{i=1:N}$, $\mathit{M}$
               \STATE {\bfseries Output: $\big( \hat{\alpha}_{i}(\mathbf{s}, a_{i}) \big)_{i=1:N}$}
               \FOR{each agent $\mathit{i}$}
                    \STATE Sample $\mathit{M}$ preceding coalitions $\mathcal{C}_{i}^{k} \sim \mathit{Pr}(\mathcal{C}_{i} | \mathcal{N} \backslash \{i\})$
                    \FOR{k=1:M}
                        \STATE Get $\hat{Q}_{\mathcal{C}_{i}^{k}}(\tau_{\mathcal{C}_{i}^{k}}, \mathbf{a}_{\mathcal{C}_{i}^{k}}) = \frac{1}{|\mathcal{C}_{i}^{k}|}\sum_{j \in \mathcal{C}_{i}^{k}} \hat{Q}_{j}(\tau_{j}, a_{j})$
                    \ENDFOR
                    \STATE Get $\hat{\alpha}_{i}(\mathbf{s}, a_{i}) = \mathlarger{\frac{1}{M}}\mathlarger{\sum}_{k = 1}^{M} \mathlarger{F}_{\mathbf{s}} \Big( \hat{Q}_{\mathcal{C}_{i}^{k}}(\tau_{\mathcal{C}_{i}^{k}}, \mathbf{a}_{\mathcal{C}_{i}^{k}}), \ \hat{Q}_{i}(\tau_{i}, a_{i}) \Big) + 1$
                \ENDFOR
            \end{algorithmic}
            \end{minipage}
            }
        \end{algorithm}
        
    \subsection{Shapley Value Based Multi-Agent Deterministic Policy Gradient}
    \label{subsec:shapley_q_value_for_multi-agent_policy_gradient}
        Since the Q-learning algorithm usually cannot deal with the problem with continuous actions, we further propose a deep deterministic policy gradient (DDPG) algorithm \cite{lillicrap2015continuous} based on Markov Shapley value, named as Shapley Q-value deep deterministic policy gradient (SQDDPG). First, SQDDPG belongs to actor-critic methods (i.e., an approximation of generalized policy iteration), while SHAQ belongs to value based methods (i.e., an approximation of value iteration). Since the value iteration is equivalent to the policy iteration with one-step policy evaluation, we can directly learn an explicit joint policy by optimizing the Markov Shapley Q-value for SQDDPG, as per the theory behind SHAQ. Specifically, the learning procedure of SQDDPG repeatedly performs the following two-step process:
        \begin{equation}
        \label{eq:sqddpg_v0}
            \begin{split}
                &\textbf{Step 1:} \quad \min_{\theta} \mathbb{E}_{\mathbf{s}, \mathbf{a}, R, \mathbf{s}'} \left[ \ \left( \ R \ + \ 
                \gamma \sum_{i \in \mathcal{N}} \hat{Q}_{i}^{\phi}(\mathbf{s}', a_{i}'; \theta^{-})
                - 
                \sum_{i \in \mathcal{N}} \hat{\delta}_{i}(\mathbf{s}, a_{i}; \lambda) \hat{Q}_{i}^{\phi}(\mathbf{s}, a_{i}; \theta) \ \right)^{2} \ \right]. \\
                &\textbf{Step 2:} \quad \pi_{i}(\mathbf{s}) \in \arg\max_{a_{i}} \hat{Q}_{i}^{\phi}(\mathbf{s}, a_{i}; \theta).
            \end{split}
        \end{equation}
        
        It is not difficult to observe that SQDDPG defined in Eq.~\ref{eq:sqddpg_v0} ideally converges to the same optimal Markov Shapley Q-values as SHAQ does (due to the equivalence between one-step policy iteration and value iteration) such that 
        \begin{equation}
        \label{eq:sqddpg_convergence}
            \mathbb{E}_{\mathbf{s}, \mathbf{s}'} \left[ \ \left( \ \max_{\mathbf{a}} R(\mathbf{s}, \mathbf{a}) \ + \ 
                \gamma \sum_{i \in \mathcal{N}} \max_{a_{i}'} \hat{Q}_{i}^{\phi^{*}}(\mathbf{s}', a_{i}')
                - 
                \sum_{i \in \mathcal{N}} \max_{a_{i}} \hat{Q}_{i}^{\phi^{*}}(\mathbf{s}, a_{i}) \ \right)^{2} \ \right] = 0.
        \end{equation}
        Nevertheless, when we derive SHAQ the important factor to decide $\hat{\delta}_{i}(\mathbf{s}, a_{i}; \lambda)$ is the solution of $w_{i}(\mathbf{s}, a_{i}^{*})$ that is defined as $1/|\mathcal{N}|$ in Proposition \ref{prop:equiv_credit_assignment}. Different from the discrete action space considered in SHAQ, for the continuous action space considered in SQDDPG, it is expensive to distinguish the optimal actions from the sub-optimal actions in practice. The main reason is that the continuous actions cannot be processed in parallel, while the discrete actions can. To address this problem, we propose an approximation. In more details, it regards the actions for exploration collected from the joint behaviour policy (i.e., an isotropic multivariable Gaussian distribution with a fixed variance) as the optimal actions. In other words, the Gaussian distribution takes the place of the original deterministic policy and becomes the optimal policy. Accordingly, the joint target policy also becomes a Gaussian distribution to maintain the consistency. As a result, Step 1 becomes the following objective function such that
        \begin{equation}
            \min_{\theta} \mathbb{E}_{\mathbf{s}, \mathbf{a}, R, \mathbf{s}', \mathbf{a}'} \left[ \ \left( \ R \ + \ 
                \gamma \sum_{i \in \mathcal{N}} \hat{Q}_{i}^{\phi}(\mathbf{s}', a_{i}'; \theta^{-})
                - 
                \sum_{i \in \mathcal{N}} \hat{Q}_{i}^{\phi}(\mathbf{s}, a_{i}; \theta) \ \right)^{2} \ \right].
        \end{equation}
        
        To guarantee that the joint policy is an $\epsilon$-soft policy (i.e., $\pi(\mathbf{a} | \mathbf{s}) > 0$), we additionally bound the random variable so that it becomes a truncated Gaussian distribution. Owing to the $\epsilon$-soft policy theorem \cite{sutton2018reinforcement}, the above approximation is guaranteed to converge to the nearly-optimal joint policy. Since the Gaussian distribution is formed by adding a white noise to the joint deterministic policy, the mode of $\pi(\mathbf{a} | \mathbf{s})$ (i.e., the mean of the Gaussian distribution) should be the optimal action. In other words, the original deterministic policy gradient is still valid under optimizing the mean of each agent's Gaussian distribution $\mu_{\theta_{i}}$.\footnote{Note that it is consistent with the Step 1 that the we tend to obtain a nearly-optimal policy by a Gaussian distribution centered with the optimal action for each agent.} Applying the deterministic policy gradient \cite{silver2014deterministic}, Step 2 becomes the following update operation such that
        \begin{equation}
        \label{eq:sqddpg_policy_improvement}
            \theta_{i} \leftarrow \theta_{i} + \alpha_{i} \nabla_{\theta_{i}}J(\theta_{i}),
        \end{equation}
        where
        \begin{equation}
            \nabla_{\theta_{i}}J(\theta_{i}) = \mathbb{E} \left[ \nabla_{\theta_{i}} \mu_{\theta_{i}}(\textbf{s}) \nabla_{a_{i}} \hat{Q}_{i}^{\phi}(\textbf{s}, a_{i}; \omega_{i}) \mathlarger{\mathlarger{|}}_{a_{i}=\mu_{\theta_{i}}(\textbf{s})} \right].
        \end{equation}
        
        About representation of the marginal contribution of an agent $i$, the input is the concatenation of two parts: information of the agent $i$ and information of the agents in the coalition $\mathcal{C}$. More specifically, the mathematical expression of the agent $i$'s marginal contribution is expressed as follows:
        \begin{equation}
        \label{eq:marginal_contribution}
            \hat{\Upphi}_{i}(\textbf{s}, \textbf{a}_{\scriptscriptstyle \mathcal{C} \cup \{i\}}): \mathcal{S} \times \mathcal{A}_{\scriptscriptstyle \mathcal{C}} \times \mathcal{A}_{i} \mapsto \mathbb{R},
        \end{equation}
        where $\mathcal{S}$ indicates a set of the global state of an environment; $\mathcal{A}_{\scriptscriptstyle \mathcal{C}}$ indicates actions of the agents in the coalition $\mathcal{C}$ that is invariant among different permutations of agents; and $\mathcal{A}_{i}$ indicates the agent $i$'s action. Followed by the probabilistic view that has been introduced in Section \ref{subsec:generalised_shapley_value_for_mcg}, Markov Shapley Q-value can be rewritten as follows:
        \begin{equation}
        \label{eq:asq_org}
            Q_{i}^{\phi}(\textbf{s}, a_{i}) = \mathbb{E}_{\mathcal{C} \sim Pr(\mathcal{C}|\mathcal{N} \backslash \{i\})} \left[ \Upphi_{i}(\mathbf{s}, a_{i} | \mathcal{C}_{i}) \right].
        \end{equation}
        
        To enable Eq.~\ref{eq:asq_org} to be tractable in realization, we use Monte Carlo estimation to approximate $Q_{i}^{\Phi}(\textbf{s}, a_{i})$ here. Also, substituting the marginal contribution for the approximate one in Eq.~\ref{eq:marginal_contribution}, we can approximate the Markov Shapley Q-value as the expression such that
        \begin{equation}
        \label{eq:shapley_approx}
            \hat{Q}_{i}^{\phi}(\textbf{s}, a_{i}) \approx \frac{1}{M} \sum_{k=1}^{M} \hat{\Upphi}_{i}(\textbf{s}, \textbf{a}_{\scriptscriptstyle \mathcal{C}_{k} \cup \{i\}}), \ \ \forall \mathcal{C}_{k} \sim Pr(\mathcal{C}|\mathcal{N} \backslash \{i\}).
        \end{equation}
        
        We show the pseudo code of SQDDPG in Algorithm \ref{alg:sqpg}.
        \begin{algorithm*}[ht!]
        \caption{Shapley Q-value deep deterministic policy gradient (SQDDPG)}
        \label{alg:sqpg}
            \begin{algorithmic}[1]
                \STATE Initialize actor parameters $\theta_{i}$, and critic (AMC) parameters $\omega_{i}$ for each agent $i \in \mathcal{N}$
                \STATE Initialize target actor parameters $\theta_{i}'$, and target critic parameters $\omega_{i}'$ for each agent $i \in \mathcal{N}$
                \STATE Initialize the sample size $M$ for approximating Markov Shapley Q-value
                \STATE Initialize the learning rate $\tau$ for updating target network
                \STATE Initialize the discount rate $\gamma$
                \FOR{episode = 1 to D}
                    \STATE Observe initial state $\textbf{s}_{1}$ from the environment
                    \FOR{t = 1 to T}
                        \STATE $u_{i} \gets \mu_{\theta_{i}}(\textbf{s}_{t}) + N_{t}$ for each agent $i$ 
                        \normalsize
                        \STATE Execute actions $\textbf{u}_{t}=\times_{i \in \mathcal{N}} u_{i}$ and observe the global reward $R_{t}$ and the next state $\textbf{s}_{t+1}$
                        \STATE Store $(\textbf{s}_{t}, \textbf{u}_{t}, R_{t}, \textbf{s}_{t+1})$ in the replay buffer $\mathcal{B}$
                        \STATE Sample a minibatch of G samples $(\textbf{s}_{k}, \textbf{u}_{k}, R_{k}, \textbf{s}_{k+1})$ from $\mathcal{B}$
                        \STATE Get $\textbf{a}_{k} = \times_{i \in \mathcal{N}} \mu_{\theta_{i}}(\textbf{s}_{k})$ for each sample $(\textbf{s}_{k}, \textbf{u}_{k}, R_{k}, \textbf{s}_{k+1})$
                        \STATE Get $\hat{\textbf{a}}_{k} = \times_{i \in \mathcal{N}} \left\{ \mu_{\theta_{i}'}(\textbf{s}_{k+1}) + N_{t} \right\}$ for each sample $(\textbf{s}_{k}, \textbf{u}_{k}, R_{k}, \textbf{s}_{k+1})$
                        \FOR{each agent $i$}
                            \normalsize
                            \STATE Sample $M$ ordered coalitions by $\mathcal{C} \sim Pr(\mathcal{C} | \mathcal{N} \backslash \{i\})$
                            \FOR{each sampled coalition $\mathcal{C}_{m}$}
                                \normalsize
                                \STATE Mask the irrelevant agents' actions for $\textbf{a}_{k}$, storing it to $\textbf{a}_{k}^{m}$
                                \STATE Mask the irrelevant agents' actions for $\hat{\textbf{a}}_{k}$, storing it to $\hat{\textbf{a}}_{k}^{m}$
                                \STATE Mask the irrelevant agents' actions for $\textbf{u}_{k}$ , storing it to $\textbf{u}_{k}^{m}$
                            \ENDFOR
                            \STATE Get $\hat{Q}_{i, k}^{\phi}(\textbf{s}_{k}, a_{i}; \omega_{i}) \gets \frac{1}{M} \sum_{m=1}^{M} \hat{\Upphi}_{i}(\textbf{s}_{k}, \textbf{a}_{k}^{m}; \omega_{i})$ for each sample $(\textbf{s}_{k}, \textbf{u}_{k}, R_{k}, \textbf{s}_{k+1})$
                            \STATE Get $\hat{Q}_{i, k}^{\phi}(\textbf{s}_{k}, u_{i}; \omega_{i}) \gets \frac{1}{M} \sum_{m=1}^{M} \hat{\Upphi}_{i}(\textbf{s}_{k}, \textbf{u}_{k}^{m}; \omega_{i})$ for each sample $(\textbf{s}_{k}, \textbf{u}_{k}, R_{k}, \textbf{s}_{k+1})$
                            \STATE Get $\hat{Q}_{i, k}^{\phi}(\textbf{s}_{k}, \hat{a}_{i}; \omega_{i}') \gets \frac{1}{M} \sum_{m=1}^{M} \hat{\Upphi}_{i}(\textbf{s}_{k}, \hat{\textbf{a}}_{k}^{m}; \omega_{i}')$ for each sample $(\textbf{s}_{k}, \textbf{u}_{k}, R_{k}, \textbf{s}_{k+1})$
                            \normalsize
                            \STATE Update $\theta_{i}$ by deterministic policy gradient:
                            \begin{equation*}
                                \nabla_{\theta_{i}}J(\theta_{i}) = \frac{1}{G} \sum_{k=1}^{G} \left\{ \nabla_{\theta_{i}} \mu_{\theta_{i}}(\textbf{s}_{k}) \nabla_{a_{i}} \hat{Q}_{i, k}^{\phi}(\textbf{s}_{k}, a_{i}; \omega_{i})|_{a_{i}=\mu_{\theta_{i}}(\textbf{s}_{k})} \right\}
                            \end{equation*}
                        \ENDFOR
                        \STATE Set $y_{k} = R_{k} + \gamma \sum_{i \in \mathcal{N}} \hat{Q}_{i, k}^{\phi}(\textbf{s}_{k}, \hat{a}_{i}; \omega_{i}')$ for each sample $(\textbf{s}_{k}, \textbf{u}_{k}, R_{k}, \textbf{s}_{k+1})$
                        \STATE Update $\omega_{i}$ for each agent $i$ by minimizing the optimization problem:
                        \begin{equation*}
                            \min_{\omega_{i}} \frac{1}{G} \sum_{k=1}^{G} \frac{1}{2}\left( \ y_{k} - \sum_{i \in \mathcal{N}} \hat{Q}_{i, k}^{\phi}(\textbf{s}_{k}, u_{i}; \omega_{i}) \right)^{2}
                        \end{equation*}
                        \STATE Update target network parameters for each agent $i$:
                            \begin{align*}
                                &\theta_{i}' \gets \tau \theta_{i} + (1 - \tau) \theta_{i}'\\
                                &\omega_{i}' \gets \tau \omega_{i} + (1 - \tau) \omega_{i}'
                            \end{align*}
                    \ENDFOR
                \ENDFOR
            \end{algorithmic}
        \end{algorithm*}
        
    \subsection{Limitation in Direct Approximation of Marginal Contribution}
    \label{subsec:problem_of_direct_approximation_of_cmc}
        As for the SQDDPG \cite{Wang_2020}, the Markov Shapley value is approximated by the direct approximation of marginal contributions. Although this approach is simple and easy-to-implement in practice, some properties such as the efficiency (i.e., the sum of maximum MSV is equal to the maximum grand coalition value) and the fairness may be violated. We now mathematically describe this phenomenon as follows. By the property of efficiency, we can get an expression such that
        \begin{equation}
            \max_{\pi} V^{\pi}(\mathbf{s}) = \sum_{i \in \mathcal{N}} \underbrace{\sum_{\mathcal{C}_{i} \subseteq \mathcal{N} \backslash \{i\}} \frac{|\mathcal{C}_{i}|!(|\mathcal{N}| - |\mathcal{C}_{i}| - 1)!}{|\mathcal{N}|!} \cdot \left[ \max_{\pi_{\mathcal{C}_{i} \cup \{i\}}} V^{\pi_{\mathcal{C}_{i} \cup \{i\}}}(\mathbf{s}) - \max_{\pi_{\mathcal{C}_{i}}} V^{\pi_{\mathcal{C}_{i}}}(\mathbf{s}) \right]}_{\text{each agent's Markov Shapley value}}.
        \label{eq:efficiency1}
        \end{equation}
        
        By rearranging Eq.~\ref{eq:efficiency1}, we can get the following equivalent relationship such that
        \begin{equation}
            \max_{\pi} V^{\pi}(\mathbf{s}) = \frac{1}{|\mathcal{N}|!} \sum_{m \in \Pi} \underbrace{\sum_{\mathcal{C}_{i} \in \mathcal{M}(m)} \left[\max_{\pi_{\mathcal{C}_{i} \cup \{i\}}} V^{\pi_{\mathcal{C}_{i} \cup \{i\}}}(\mathbf{s}) - \max_{\pi_{\mathcal{C}_{i}}} V^{\pi_{\mathcal{C}_{i}}}(\mathbf{s}) \right]}_{\text{the sum of intermediate coalition values generated from a permutation}},
        \label{eq:efficiency2}
        \end{equation}
        where $\Pi$ indicates the set of all permutations to form the grand coalition; $\mathcal{M}(m)$ indicates the set of all intermediate coalitions generated by a permutation $m$. It is not difficult to observe that if $V^{\pi_{\mathcal{C}_{i} \cup \{i\}}}(\mathbf{s})$ and $V^{\pi_{\mathcal{C}_{i}}}(\mathbf{s})$ belong to the same parametric function, the successive intermediate maximum coalition values will be cancelled. Therefore, only the maximum grand coalition value and the maximum empty coalition value are left. The coalition values fulfil the above operations are defined as the \textit{consistent coalition values}. Since the empty coalition value is defined as 0, it is obvious that the LHS is equal to the RHS in Eq.~\ref{eq:efficiency2}. If $\Phi_{i}(\cdot \ | \ \cdot;\ \theta_{i})$ is a direct parametric function, a combination of $\theta_{i}$ and the input of the information about some coalitions may be aligned with multiple possibilities of differences between successive maximum coalition values. This will highly probably lead to inconsistent coalition values. As a result, the property of efficiency does not hold (since Eq.~\ref{eq:efficiency2} becomes violated) and the fairness is corrupted.\footnote{Since the coalition values become inconsistent and each resultant marginal contribution probably implies the credit assignment from a different set of coalition values.} To address this problem, we suggest learning a parametric coalition value function and then use it to form the MSV. A fixed parametric function can be seen as a method to encode a set of coalition values. We call this method \textit{indirect approximation of marginal contributions}. Due to the complexity in learning, we conduct an analysis of the feasibility of learning in practice as Proposition \ref{prop:feasibility_msv_by_cv} shows (see Appendix \ref{subsec:fixing_the_problem_of_inconsistency} for details). Since the equivalence between the MSV and the MSQ, the above discussion on the MSV also adapts to the MSQ.
        \begin{proposition}
        \label{prop:feasibility_msv_by_cv}
            The approximate maximum Markov Shapley value generated by approximate maximum coalition values is feasible to be learned in practice.
        \end{proposition}
        \begin{proof}
            See the detailed proof in Appendix \ref{subsec:fixing_the_problem_of_inconsistency}.
        \end{proof}
        
        Owing to the need of precision and safety to the algorithms applied for real-world applications, we propose to incorporate the indirect approximation of marginal contributions into Shapley value based MARL algorithms. Since the framework of DDPG \cite{heinrich2016deep} is infeasible to fit the difference-type value functions without losing the advantage of marginal contributions,\footnote{Even if using Q-values, the subtrahend Q-value does not involve the agent's action and cannot be differentiable with respect to the parameters of policy.} we apply the framework of PPO \cite{schulman2017proximal} as the base algorithm and propose an algorithm named as \textit{Shapley model-free PPO} (SMFPPO). Different from SQDDPG, each agent's policy in SMFPPO is modelled as a parametric Gaussian distribution with a learnable mean and a learnable variance to deal with the continuous action problem.\footnote{The control variables of many real-world problems are belonging to the continuous actions, e.g., the active voltage control problem considered in this thesis where the reactive power is a continuous action.} The pseudo code of SMFPPO is shown in Algorithm \ref{alg:smfppo}.
        \begin{algorithm*}[ht!]
        \caption{Shapley Model-Free PPO (SMFPPO)}
        \label{alg:smfppo}
            \begin{algorithmic}[1]
                \STATE Initialize actor parameters $\theta_{i}$ for each agent $i \in \mathcal{N}$, and coalition Q-value parameters $\omega$ 
                \STATE Initialize the sample size $M$ for approximating Markov Shapley Q-value
                \STATE Initialize the discount rate $\gamma$ and the $\epsilon$
                \FOR{episode = 1 to D}
                    \STATE Observe initial state $\textbf{s}_{1}$ from the environment
                    \FOR{t = 1 to T}
                        \STATE Sample $a_{i}$ from $\pi_{\theta_{i}}(\textbf{s}_{t})$ for each agent $i$
                        \STATE Execute actions $\textbf{a} = \mathlarger{\mathlarger{\times}}_{\scriptscriptstyle i \in \mathcal{N}} a_{i}$ and observe the global reward $r_{t}$ and the next state $\textbf{s}_{t+1}$
                        \FOR{each agent $i$}
                            \STATE Sample $M$ coalitions by $\mathcal{C}_{m} \sim Pr(\mathcal{C} | \mathcal{N} \backslash \{i\})$
                            \FOR{each sampled coalition $\mathcal{C}_{m}$}
                                \STATE Mask the irrelevant agents' actions for $\textbf{a}$, storing it as $\textbf{a}_{\scriptscriptstyle \mathcal{C}_{m} \cup \{i\}}$ and $\textbf{a}_{\scriptscriptstyle \mathcal{C}_{m}}$
                                \STATE Get $\hat{\Upphi}_{i}^{m}(\textbf{s}_{t}, a_{i}; \omega) \gets \hat{Q}^{\pi_{\scriptscriptstyle \mathcal{C}_{m} \cup \{i\}}}(\textbf{s}_{t}, \textbf{a}_{\scriptscriptstyle \mathcal{C}_{m} \cup \{i\}}; \omega) - \hat{Q}^{\pi_{\scriptscriptstyle \mathcal{C}_{m}}}(\textbf{s}_{t}, \textbf{a}_{\scriptscriptstyle \mathcal{C}_{m}}; \omega)$
                            \ENDFOR
                            \STATE Get $\hat{Q}^{\Phi_{i}}(\textbf{s}_{t}, a_{i}; \omega) \gets \frac{1}{M} \sum_{m=1}^{M} \hat{\Upphi}_{i}^{m}(\textbf{s}_{t}, a_{i}; \omega)$
                            \STATE Update $\theta_{i}$ by maximizing the following loss:
                            \begin{equation*}
                                \max_{\theta_{i}} \mathbb{E} \left[ \min \left\{ \frac{\pi_{\theta_{i}}(a_{i} \vert \textbf{s})}{\pi_{\theta_{i, \text{old}}}(a_{i} | \mathbf{s})} \ \hat{Q}^{\Phi_{i}}(\textbf{s}_{t}, a_{i}; \omega), clip(\frac{\pi_{\theta_{i}}(a_{i} \vert \textbf{s})}{\pi_{\theta_{i, \text{old}}}(a_{i} | \mathbf{s})}, 1-\epsilon, 1+\epsilon) \ \hat{Q}^{\Phi_{i}}(\textbf{s}_{t}, a_{i}; \omega) \right\} \right]
                            \end{equation*}
                        \ENDFOR
                        \STATE Update $\omega$ by minimizing the optimization problem:
                        \begin{equation*}
                            \min_{\omega} \left( R_{t} + \hat{Q}^{\pi}(\textbf{s}_{t+1}, \textbf{a}; \omega) - \hat{Q}^{\pi}(\textbf{s}_{t}, \textbf{a}; \omega) \right)^{2}
                        \end{equation*}
                    \ENDFOR
                \ENDFOR
            \end{algorithmic}
        \end{algorithm*}
        
    \subsection{Relationship to Other Multi-Agent Reinforcement Learning Algorithms}
    \label{subsec:comparison_with_other_learning_algorithms}
        In this section, we compare SHAQ, SQDDPG and SMFPPO with other MARL algorithms and investigate the relationship to these algorithms such as VDN \cite{SunehagLGCZJLSL18}, COMA \cite{foerster2018counterfactual} and independent learning (IL) \cite{claus1998dynamics}. We wish our analysis can provide some insights into the further works on credit assignment for global reward game. The relationship among those algorithms is shown in Figure \ref{fig:taxonomy_marl}.
        \begin{figure}[ht!]
            \centering
            \includegraphics[width=.95\textwidth]{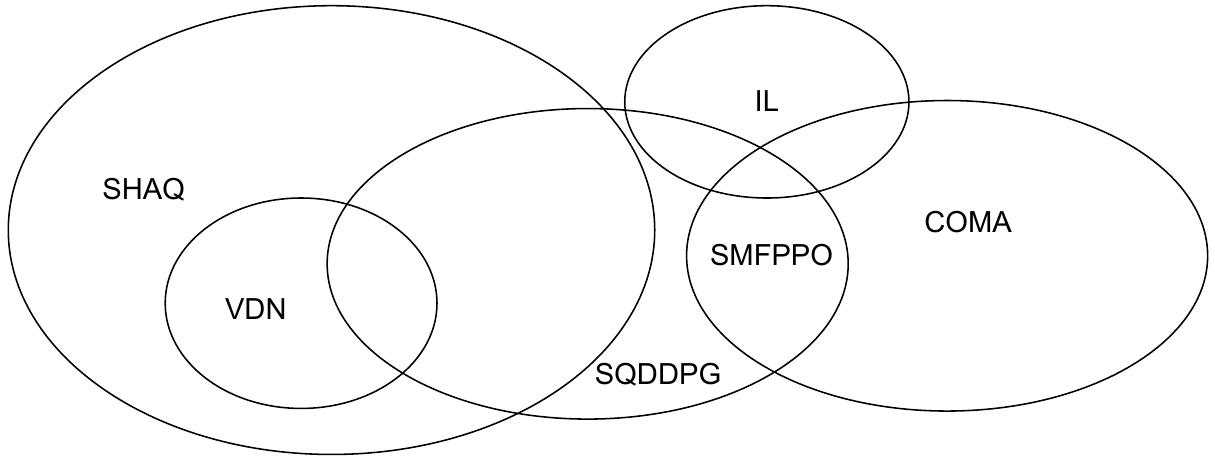}
            \caption{Taxonomy of the proposed Shapley value based and other MARL algorithms.}
            \label{fig:taxonomy_marl}
        \end{figure}

        \paragraph{Relationship to COMA.} Compared with COMA \cite{foerster2018counterfactual}, the credit assigned to each agent $i$ denoted by $\bar{Q}_{i}(\mathbf{s}, a_{i})$ is mathematically expressed as follows:
        \begin{equation*}
            \begin{split}
                \bar{Q}_{i}(\mathbf{s}, a_{i}) = \bar{Q}^{\pi}(\mathbf{s}, \mathbf{a}) - \bar{Q}^{\pi_{-i}}(\mathbf{s}, \mathbf{a}_{- i}), \\
                \bar{Q}^{\pi_{-i}}(\mathbf{s}, \mathbf{a}_{-i}) = \sum_{a_{i}} \pi_{i}(a_{i}|\mathbf{s}) \bar{Q}^{\pi}(\mathbf{s}, (\mathbf{a}_{-i}, a_{i})),
            \end{split}
        \end{equation*}
        where the subscript $-i$ indicates the set of all agents excluding $i$. $\bar{Q}_{i}(\mathbf{s}, a_{i})$ can be seen as the action marginal contribution between the grand coalition Q-value and the coalition Q-value excluding the agent $i$, under \textit{some permutation to form the grand coalition} wherein agent $i$ is located at the \textit{last position}. The efficiency is obviously violated (i.e., the sum of optimal action marginal contributions defined here is unlikely to be equal to the optimal grand coalition Q-value). In contrast to COMA, SMFPPO considers all permutations to form the grand coalition to preserve the efficiency and therefore learns a function that represents all coalition value functions.
        
        \paragraph{Relationship to VDN.} By setting $\delta_{i}(\mathbf{s}, a_{i}) = 1$ for all state-action pairs, SHAQ degrades to VDN \cite{SunehagLGCZJLSL18}. Although VDN tried to tackle the problem of dummy agents, \cite{SunehagLGCZJLSL18} did not give a theoretical guarantee on identifying it. The Markov Shapley value theory proposed in this thesis well addresses this issue from both theoretical and empirical aspects. These aspects show that VDN is a subclass of SHAQ, which has the same value loss function as SQDDPG. The theoretical framework proposed in this thesis answers the question of why VDN works well in most scenarios but poorly in some scenarios (i.e., $\delta_{i}(\mathbf{s}, a_{i})=1$ in Eq.~\ref{eq:td_error_shapley_q_learning} was incorrectly defined over the sub-optimal actions).
        
        \paragraph{Relationship to Independent Learning.} Independent learning (e.g., IQL \cite{claus1998dynamics}) can also be seen as a special credit assignment, however, the credit assigned to each agent is still with no intuitive interpretation. Mathematically, suppose that $\bar{Q}_{i}(\mathbf{s}, a_{i})$ is the independent Q-value of agent $i$, we can rewrite it in the form of the linear combination of action marginal contributions such that
        \begin{equation*}
            \bar{Q}_{i}(\mathbf{s}, a_{i}) = \mathbb{E}_{\mathcal{C}_{i} \sim Pr(\mathcal{C}_{i} | \mathcal{N} \backslash \{i\})} \left[ \bar{\Upphi}_{i}(\mathbf{s}, a_{i} | \mathcal{C}_{i}) \right].
        \end{equation*}
        
        It is intuitive to see that the independent Q-value is a direct approximation of MSQ, ignoring the detailed process of coalition formation, while SHAQ, SQDDPG and SMFPPO consider the coalition formation in representing credit assignment. This gives an explanation for why independent learning works well in some cooperative tasks \cite{PapoudakisC0A21}.
        
        \subsection{Relationship to Other Theoretical Frameworks}
        \label{subsec:relationship_with_other_theoretical_framework}
            Overall, Markov convex game has an intersection with Individual-Global-Max \cite{SonKKHY19}. In more details, Markov convex game can solve the scenarios more than the grand coalition, while the Individual-Global-Max solely considers the grand coalition. Under the grand coalition as the coalition structure, the Markov convex game is a subclass of Individual-Global-Max, by the view of credit assignment. More specifically, Individual-Global-Max describes a class of credit assignment as value distribution under a condition that the optimal policy of each agent's value leads to the optimal joint policy. In contrast, Markov convex game under the grand coalition presumes the convexity condition that leads to an analytic form of the distributed value, so it is a stronger condition and the theory is weaker. Distributed Q-learning \cite{lauer2000algorithm} gives another analytic form of each distributed value to satisfy the condition in Individual-Global-Max,\footnote{In distributed Q-learning, each distributed Q-value is constructed as $Q_{i}(\mathbf{s}, a_{i}) = \max_{\mathbf{a}_{-i}} Q(\mathbf{s}, \mathbf{a}_{-i}, a_{i})$, where $\mathbf{a}_{-i}$ indicates the joint actions of agents except for agent $i$.} so it is a weaker theory than Individual-Global-Max. Besides, distributed Q-learning is in parallel with the Markov convex game under the grand coalition, since both theories study the methods to fulfil the condition in Individual-Global-Max from different aspects. Although the theory of Markov convex game under the grand coalition is weaker than Individual-Global-Max, the condition of the Markov convex game gives more insights into the investigation of the full picture of Individual-Global-Max.
            \begin{figure}[ht!]
                \centering
                \includegraphics[width=.95\textwidth]{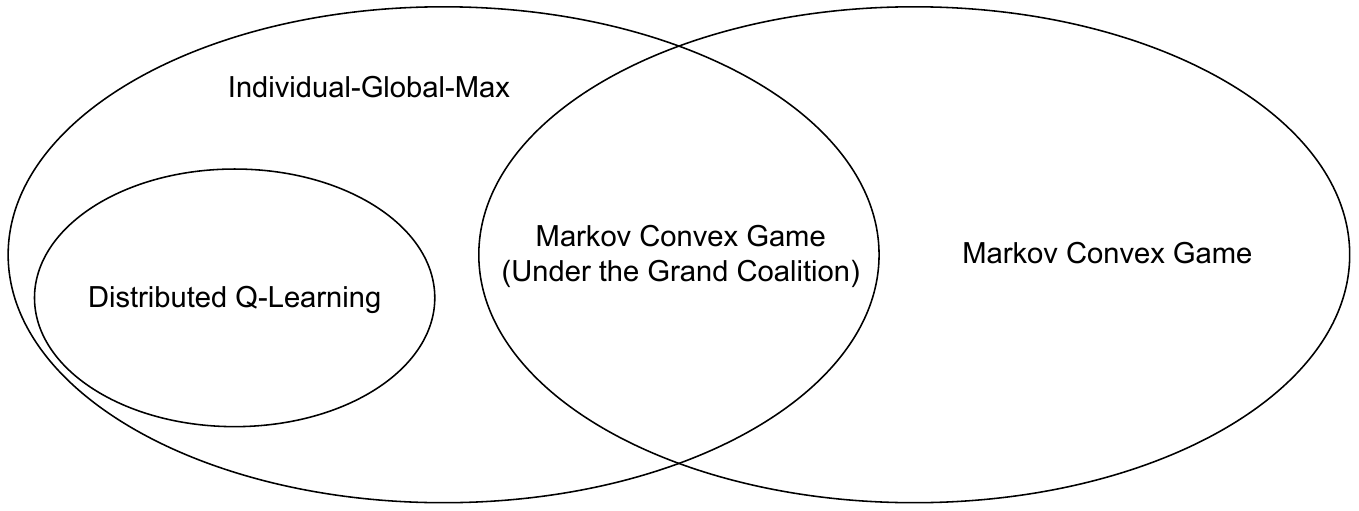}
                \caption{Taxonomy of the Markov convex game proposed in this thesis and the relevant theoretical frameworks for relating credit assignment to the optimal global value in multi-agent reinforcement learning. Markov convex game has an intersection with the theoretical framework of Individual-Global-Max, and can cover scenarios excluded from Individual-Global-Max such that the coalition structure is not the grand coalition. On the other hand, the Markov convex game under the grand coalition and distributed Q-learning are two theories that realize the condition of Individual-Global-Max and deepen the understanding of optimizing assigned credits to reach the optimal joint policy.}
                \label{fig:theoretical_framework_taxonomy}
            \end{figure}
        
\section{Tackling The Partial Observability Problem}
\label{sec:the_partial_observability_problem}
    In this section, we aim at presenting the variants of multi-agent reinforcement learning algorithms based on Shapley value to fit the partially observable environments. We first extend Markov convex game to the partially observable scenarios, named as partially observable Markov game (POMCG). Second, we propose two algorithms to solve the POMCG, named as partially observable Shapley policy iteration and partially observable Shapley value iteration, and prove their convergence to the optimal joint policy and the Markov core. Finally, we show that these two algorithms can be leveraged to guide the practical implementations of SQDDPG, SMFPPO and SHAQ to tackle the partially observable tasks.
    
    \subsection{Additional Assumptions about Partial Observability}
    \label{subsec:assumptions_partial_observability}
        In addition to the assumptions considered in Section \ref{sec:assumptions_shaspley-q}, we need the further assumptions on the analysis of partial observability as shown in Assumption \ref{assm:partial_observability} which is a common assumption applied to partially observable Markov decision process (POMDP).
        \begin{assumption}
        \label{assm:partial_observability}
            We assume that a value function is representable by a finite set of vectors, and therefore the value function is convex \cite{littman1995efficient}.
        \end{assumption}
        
    \subsection{Partially Observable Markov Convex Game}
    \label{subsec:problem_formulation}
        To adapt to the decentralised control in many real-world tasks with partial observation (e.g. the active voltage control in electric power systems), we extend the MCG to partial observability, named as partially observable Markov convex game (POMCG). The POMCG can be described as the following tuple $\left\langle \mathcal{N}, \mathcal{S}, \mathcal{A}, T, \Lambda, \pi, R_{t}, \gamma, \Omega, \mathcal{O}, Pr(\mathbf{s}_{0}) \right\rangle$. The two terms lacking in the MCG are a joint observation set $\mathcal{O} = \mathlarger{\mathlarger{\times}}_{i \in \mathcal{N}} \mathcal{O}_{i}$ where $\mathcal{O}_{i}$ is agent $i$'s observation set, and an observation probability function $\Omega(\mathbf{o}_{t+1} | \mathbf{s}_{t+1}, \mathbf{a}_{t}, \mathcal{CS})$ where $\mathbf{o} = \mathlarger{\mathlarger{\times}}_{\scriptscriptstyle i \in \mathcal{N}} o_{i} \in \mathcal{O}$, $\mathbf{s} \in \mathcal{S}$ and $\mathbf{a} \in \mathcal{A}$ by Definition \ref{def:pomcg}. Agent $i$'s \textit{action-observation history} (AOH) at timestep $t$ is defined as $h_{i, t} = \langle o_{i, 1}, a_{i, 1}, ..., a_{i, t-1}, o_{i, t}\rangle$ that can be recursively written as $h_{i, t} = \langle h_{i, t-1}, a_{i, t-1}, o_{i, t} \rangle$ \cite{oliehoek2016concise}. Likewise, the joint history is defined as $\mathbf{h}_{t} = \langle \mathbf{h}_{t-1}, \mathbf{a}_{t-1}, \mathbf{o}_{t} \rangle$. In the training phase of \textit{centralised training and decentralised execution} (CTDE), there exists a \textit{joint belief state} $b_{t}(\mathbf{s}_{t} | \mathcal{CS}) = Pr(\mathbf{s}_{t} | \mathbf{h}_{t}, \mathcal{CS}) \in \mathcal{B}(\mathcal{CS})$ to summarize the joint history and estimate the probability of states. Note that $\mathcal{B}(\mathcal{CS})$ is usually an infinite space that results in an infinite value space, but corresponding to some specific $\mathcal{CS}$. This is solved by representing the value function with finite number of vectors $\zeta \in \mathbb{R}^{|\mathcal{S}||\mathcal{CS}|}$. In more details, the resultant value function is convex \cite{sondik1971optimal} and can be expressed as $\max_{\zeta(\cdot | \mathcal{CS})} \sum_{\mathbf{s}} \zeta(\mathbf{s} | \mathcal{CS}) b(\mathbf{s} | \mathcal{CS})$. The minimal set of $\zeta(\cdot | \mathcal{CS})$ is unique \cite{littman1995efficient}. By simply applying \textit{Bayesian inference}, the successor joint belief state can be obtained as follows: $$b_{t+1}(\mathbf{s}_{t+1} | \mathcal{CS}) = \tau(\mathbf{o}_{t+1}, \mathbf{a}_{t}, b_{t} | \mathcal{CS}) = \frac{\Omega(\mathbf{o}_{t+1} | \mathbf{s}_{t+1}, \mathbf{a}_{t}, \mathcal{CS}) \left[ \sum_{\mathbf{s}_{t} \in \mathcal{S}} Pr(\mathbf{s}_{t+1} | \mathbf{s}_{t}, \mathbf{a}_{t}) b_{t}(\mathbf{s}_{t} | \mathcal{CS}) \right]}{Pr(\mathbf{o}_{t+1} | b_{t}, \mathbf{a}_{t}, \mathcal{CS})},$$
        where $Pr(\mathbf{o}_{t+1} | b_{t}, \mathbf{a}_{t}, \mathcal{CS}) = \sum_{\mathbf{s}_{t+1} \in \mathcal{S}} \left[ \Omega(\mathbf{o}_{t+1} | \mathbf{s}_{t+1}, \mathbf{a}_{t}, \mathcal{CS}) \sum_{\mathbf{s}_{t} \in \mathcal{S}} Pr(\mathbf{s}_{t+1} | \mathbf{s}_{t}, \mathbf{a}_{t}) b_{t}(\mathbf{s}_{t} | \mathcal{CS}) \right]$. The initial belief is defined as $b_{0}(\mathbf{s}_{0} | \mathcal{CS}) = Pr(\mathbf{s}_{0} | \mathcal{CS})$. The joint belief state constitutes a sufficient statistic of the joint history \cite{astrom1965optimal}. By defining a transition function and a reward function of belief state such that $T_{b}(b_{t+1}, b_{t}, \mathbf{a}_{t}, \mathcal{CS}) = Pr(b_{t+1} | b_{t}, \mathbf{a}_{t}, \mathcal{CS}) = \sum_{\mathbf{o}_{t+1} \in \mathcal{O}} \mathbb{I}_{\tau(b_{t}, \mathbf{a}_{t}, \mathbf{o}_{t+1})}(b_{t+1}) Pr(\mathbf{o}_{t+1} | b_{t}, \mathbf{a}_{t}, \mathcal{CS})$\footnote{$\mathbb{I}_{m}(x)$ is an indicator function. It equals to 1 when $x = m$, otherwise it equals to 0.} and $R_{b}(b_{t}, \mathbf{a}_{\scriptscriptstyle\mathcal{C} \scriptstyle, t}) = \sum_{\mathbf{s} \in \mathcal{S}} b_{t}(\mathbf{s} | \mathcal{CS}) R(\mathbf{s}, \mathbf{a}_{\scriptscriptstyle\mathcal{C} \scriptstyle, t})$, POMDP is transformed to an equivalent \textit{belief-MDP}. Note that each coalition $\mathcal{C}$ only corresponds to a subset of $\Lambda$ as shown in Definition \ref{def:extended_coalition_set}. It is known that solving a belief-MDP is equivalent to solving a corresponding POMDP \cite{aastrom1965optimal}. Since the training phase of POMCG is a specific form of POMDP, it is reasonable to solve POMCG as a belief-MDP.
        \mybox{
        \begin{definition}
        \label{def:pomcg}
            In a POMCG, the coalition structure $\mathcal{CS}$ affects the observability of an environment and therefore the belief state.
        \end{definition}
        }
        
        \begin{definition}
        \label{def:extended_coalition_set}
            For any $\mathcal{C} \ \mathlarger{\mathlarger{\subseteq}} \ \mathcal{N}$ belonging to a coalition structure included in a subset of $\Lambda$, the subset of $\Lambda$ is denoted as $\Psi(\mathcal{C}, \Lambda) \ \mathlarger{\mathlarger{\subseteq}} \ \Lambda$. In other words, only the coalition structure belonging to $\Psi(\mathcal{C}, \Lambda)$ contains $\mathcal{C}$, while $\Lambda \backslash \Psi(\mathcal{C}, \Lambda)$ does not.
        \end{definition}
        
    \subsection{Partially Observable Shapley Policy Iteration}
    \label{subsec:policy_iteration_for_pomcg}
        To enable the selection of control more easily during executions \cite{bertsekas2019reinforcement}, we propose an algorithmic framework in Algorithm \ref{alg:pi_pomcg} named as \textit{partially observable Shapley policy iteration} (POSPI) to learn the optimal joint policy. Although Line 10 to 19 is the procedure of sequentially updating the policies with \textit{no extra interactions with the environment}, it can be implemented by the \textit{simultaneous updates among agents} in practice. 
        \begin{algorithm}[ht!]
            \caption{Partially Observable Shapley Policy Iteration for POMCG.}
            \label{alg:pi_pomcg}
            \scalebox{1.0}{
            \begin{minipage}{1.0\linewidth}
            \begin{algorithmic}[1]
            
                \STATE Give an initialised $\pi^{k} = \mathlarger{\mathlarger{\times}}_{\scriptscriptstyle i \in \mathcal{N}} \pi_{i}^{k}$. 
                
                \REPEAT
                    \STATE \textbf{Policy Evaluation:} Compute $V^{\pi_{\mathcal{C}}^{k}}(b)$ for all $b$ to solve
                    
                    $$V^{\pi_{\mathcal{C}}^{k}}(b) = R_{b}(b, \mathbf{a}_{\scriptscriptstyle\mathcal{C}}^{k}) + \gamma \sum_{\mathbf{o}' \in \mathcal{O}} Pr(\mathbf{o}' | b, \mathbf{a}^{k}, \mathcal{CS}) V^{\pi_{\mathcal{C}}}(\tau(\mathbf{o}', \mathbf{a}^{k}, b | \mathcal{CS})),$$
                    by value iteration described in Lemma \ref{lemm:value_iteration}.
                    
                    \STATE \textbf{Policy Improvement:} Compute $\pi^{k+1}$ as follows:
                    \FOR{$\mathcal{CS} \in \Lambda$}
                        \FOR{$\mathcal{C} \in \mathcal{CS}$ and $b \in \mathcal{B}(\mathcal{CS})$}
                            \STATE $Q\left( b, \mathbf{a}_{\scriptscriptstyle\mathcal{C}} \right) \leftarrow R_{b} \left( b, \mathbf{a}_{\scriptscriptstyle\mathcal{C}} \right) + \gamma \sum_{\mathbf{o}' \in \mathcal{O}} Pr \left( \mathbf{o}' | b, \mathbf{a}, \mathcal{CS} \right) V^{\pi_{\mathcal{C}}^{k}} \left( \tau \left(\mathbf{o}', \mathbf{a}, b | \mathcal{CS} \right) \right)$.
                        \ENDFOR
                    \ENDFOR
                    
                    \FOR{$i \in \mathcal{N}$}
                        \STATE Set $w(\mathcal{C}_{i}) = Pr(\mathcal{C}_{i} | \mathcal{N} \backslash \{i\})$.
                        \FOR{$\mathcal{C}_{i} \ \mathlarger{\mathlarger{\subseteq}} \ \mathcal{N} \backslash \{i\}$, $\mathcal{CS} \in \Psi(\mathcal{C}_{i} \ \mathlarger{\mathlarger{\cup}} \ \{i\}, \Lambda)$ and $b \in \mathcal{B}(\mathcal{CS})$}
                            \STATE Get $Q(b, \mathbf{a}_{\scriptscriptstyle\mathcal{C}_{i} \cup \{i\}})$ and $Q(b, \mathbf{a}_{\scriptscriptstyle\mathcal{C}_{i}})$.
                        \ENDFOR
                        \STATE $\pi_{i}^{k+1} \leftarrow \arg\max_{\pi_{i}} \sum_{\mathcal{C}_{i} \in \mathcal{N} \backslash \{i\}} w(\mathcal{C}_{i}) \left[ Q(b, \mathbf{a}_{\scriptscriptstyle\mathcal{C}_{i} \cup \{i\}}) - Q(b, \mathbf{a}_{\scriptscriptstyle\mathcal{C}_{i}}) \right]$.
                        \FOR{$\mathcal{C}_{i} \ \mathlarger{\mathlarger{\subseteq}} \ \mathcal{N} \backslash \{i\}$, $\mathcal{CS} \in \Psi(\mathcal{C}_{i} \ \mathlarger{\mathlarger{\cup}} \ \{i\}, \Lambda)$ and $b \in \mathcal{B}(\mathcal{CS})$}
                            \STATE $Q\left( b, \mathbf{a}_{\scriptscriptstyle\mathcal{C}_{i} \cup \{i\}} \right) \leftarrow Q\left( b, \mathbf{a}_{\scriptscriptstyle\mathcal{C}_{i}}, \pi_{i}^{k+1}(b) \right)$.
                        \ENDFOR
                    \ENDFOR
                    
                    \STATE $\pi^{k} \leftarrow \pi^{k+1}$.
                    
                \UNTIL $V^{\pi_{\mathcal{C}}^{k+1}}(b) = V^{\pi_{\mathcal{C}}^{k}}\left( b \right)$, for all $\mathcal{C} \ \mathlarger{\mathlarger{\subseteq}} \ \mathcal{N}$, $\mathcal{CS} \in \Psi(\mathcal{C}, \Lambda)$ and $b \in \mathcal{B}(\mathcal{CS})$. 
            \end{algorithmic}
            \end{minipage}
            }
        \end{algorithm}

        \begin{lemma}
        \label{lemm:value_iteration}
            For all $b \in \mathcal{B}(\mathcal{CS})$ and $\mathcal{CS} \in \Psi(\mathcal{C}, \Lambda)$, the value iteration for $\pi_{\scriptscriptstyle\mathcal{C}}$ such that
            \begin{equation*}
                V_{m+1}(b) \leftarrow R_{b}(b, \mathbf{a}_{\scriptscriptstyle\mathcal{C}}) + \gamma \sum_{\mathbf{o}' \in \mathcal{O}} Pr(\mathbf{o}' | b, \mathbf{a}, \mathcal{CS}) V_{m}(\tau(\mathbf{o}', \mathbf{a}, b | \mathcal{CS})),
            \end{equation*}
            converges to $V^{\pi_{\mathcal{C}}}(b)$ as $m \rightarrow \infty$ under the infinity norm.
        \end{lemma}
        \begin{proof}
            See the detailed proof in Appendix \ref{sec:proof_of_shapley_policy_iteration_fo_pomcg}.
        \end{proof}
        
        \begin{proposition}
        \label{prop:policy_iteration}
            For all $\mathcal{C} \ \mathlarger{\mathlarger{\subseteq}} \ \mathcal{N}$, $b \in \mathcal{B}(\mathcal{CS})$ and $\mathcal{CS} \in \Psi(\mathcal{C}, \Lambda)$, partially observable Shapley policy iteration converges to the optimal coalition values and the optimal joint policy.
        \end{proposition}
        \begin{proof}
            See the detailed proof in Appendix \ref{sec:proof_of_shapley_policy_iteration_fo_pomcg}.
        \end{proof}
        
        \begin{corollary}
        \label{coro:shapley_policy_iteration_core}
            Partially observable Shapley policy iteration is guaranteed to converge to the Markov core.
        \end{corollary}
        \begin{proof}
            This result can be directly obtained by the result of Proposition \ref{prop:policy_iteration} that partially observable Shapley policy iteration converges to the optimal coalition values which can form the optimal Markov Shapley values, and the result of Theorem \ref{thm:shapley_value_core} that the optimal Markov Shapley value is a solution in the Markov core.
        \end{proof}
        
    \subsection{Partially Observable Shapley Value Iteration}
    \label{subsec:shapley_value_iteration_for_pomcg}
        \begin{proposition}
        \label{prop:value_iteration_pomcg}
            For all $\mathcal{C} \ \mathlarger{\mathlarger{\subseteq}} \ \mathcal{N}$, $b \in \mathcal{B}(\mathcal{CS})$ and $\mathcal{CS} \in \Psi(\mathcal{C}, \Lambda)$, the partially observable Shapley value iteration such that
            \begin{equation}
            \label{eq:shapley_value_iteration}
                Q_{m+1}(b, \mathbf{a}_{\scriptscriptstyle \mathcal{C}}) \leftarrow R_{b}(b, \mathbf{a}_{\scriptscriptstyle\mathcal{C}}) + \gamma \sum_{\mathbf{o}' \in \mathcal{O}} Pr(\mathbf{o}' | b, \mathbf{a}, \mathcal{CS}) \max_{\mathbf{a}_{\scriptscriptstyle \mathcal{C}}'} Q_{m}(\tau(\mathbf{o}', \mathbf{a}, b | \mathcal{CS}), \mathbf{a}_{\scriptscriptstyle \mathcal{C}}'),
            \end{equation}
            converges to the optimal coalition Q-values and the optimal joint policy.
        \end{proposition}
        \begin{proof}
            See the detailed proof in Appendix \ref{sec:proof_of_shapley_policy_iteration_fo_pomcg}.
        \end{proof}
        
        In this section, we propose \textit{partially observable Shapley value iteration} (POSVI) that is a value-based algorithm in contrast with the POSPI that is a policy-based algorithm. It can be observed that the POSVI is a special case of the POSPI, where the policy evaluation is performed only once. Proposition \ref{prop:value_iteration_pomcg} proves that the POSVI converges to the optimal coalition Q-values and the optimal joint policy.
        
        Since $R_{b}(b, \mathbf{a}_{\scriptscriptstyle\mathcal{C}}) = \sum_{\mathbf{s} \in \mathcal{S}} b(\mathbf{s} | \mathcal{CS}) R(\mathbf{s}, \mathbf{a}_{\scriptscriptstyle\mathcal{C}})$ and the identifiablility of $Pr(\mathbf{o}' | b, \mathbf{a}, \mathcal{CS})$ (i.e., the expression can be factorised into estimands), we can sample $b_{t}(\mathbf{s} | \mathcal{CS})$ and $Pr(\mathbf{o}' | b, \mathbf{a}, \mathcal{CS})$ to derive the following operation to update coalition Q-values such that
        \begin{equation}
        \label{eq:pomcg_q-learning}
            Q_{m+1}(b, \mathbf{a}_{\scriptscriptstyle \mathcal{C}}) \leftarrow R(\mathbf{s}, \mathbf{a}_{\scriptscriptstyle\mathcal{C}}) + \gamma \max_{\mathbf{a}_{\scriptscriptstyle \mathcal{C}}'} Q_{m}(\tau(\mathbf{o}', \mathbf{a}, b | \mathcal{CS}), \mathbf{a}_{\scriptscriptstyle \mathcal{C}}').
        \end{equation}
    
    \subsection{The Model-Free Implementation}
    \label{subsec:the_model-free_implementation}
        In this section, we discuss the model-free implementations of POSPI and POSVI to solve the estimation of coalition value functions under the partially observable scenarios for SQDDPG, SMFPPO and SHAQ. Since SQDDPG and SMFPPO belong to the category of multi-agent policy gradient methods (i.e. an instance of multi-agent policy iteration), we discuss the model-free implementation of POSPI for SQDDPG and SMFPPO. Similarly, we also discuss the model-free implementation of POSVI for SHAQ.

        \paragraph{Implementation of POSPI for SQDDPG and SMFPPO.} The input of each agent's policy in the setting of POMCG is a belief state rather than an exact state, which is able to be inferred by its observation and the belief state at the preceding timestep \cite{oliehoek2016concise}. Motivated by this result, each agent's policy as an aggregation model of the belief inference model and the policy function is modelled as a recurrent neural network (RNN) for SQDDPG and SMFPPO. Similarly, the coalition value function can be modelled as an RNN also, but with the concatenation of all agents' observations as the input. Nevertheless, in practice the concatenation of agents' observations is empirically shown to be sufficient to represent the state in many scenarios \cite{LyuXDA21}. In light of this finding, SQDDPG and SMFPPO directly use multi-layer perceptrons (MLPs) to model the coalition value function with the input as the concatenation of agents' observations.
        
        \paragraph{Implementation of POSVI for SHAQ.} The coalition Q-value $Q(\tau(\mathbf{o}', \mathbf{a}, b | \mathcal{CS}), \mathbf{a}_{\scriptscriptstyle \mathcal{C}}')$ can be directly modelled as an RNN, where the belief state $b$ is represented as the hidden state. This gives an evidence of why the practical implementation of SHAQ uses RNN to handle partially observable scenarios (before further approximation to handling decentralised execution as shown in Proposition \ref{prop:shapley_value_approximate}). In more details, the Markov Shapley Q-value is directly modelled as an RNN in SHAQ, which can be seen as an aggregation model that linearly combines multiple RNNs of coalition Q-values. The linearity is resulting from a fact that Markov Shapley Q-value is equal to the convex combination of differences of coalition Q-values.


\chapter{Evaluation on Benchmark Tasks}
\label{chap:evaluation_on_benchmark_tasks}
    In this chapter, we show the evaluation of SQDDPG and SHAQ on several conventional benchmarks in the community of machine learning research. In the evaluation, we not only show the performance improvement, but also provide some visualization and analysis to demonstrate the interpretability of Markov Shapley value, which may show the potential application to the real-world tasks. For example, Markov Shapley value can be regarded as an index to evaluate the decision of agents. Moreover, if an agent is under cyber-attack, the changes of Markov Shapley value can capture the anomaly.
    
    \section{Evaluation of SQDDPG}
    \label{sec:evaluation_of_sqddpg}
        We evaluate SQDDPG on Cooperative Navigation, Prey-and-Predator \cite{lowe2017multi} and Traffic Junction \cite{NIPS2016_6398}. In the experiments, we compare SQDDPG with two independent algorithms (with decentralised critics), such as independent DDPG (IDDPG) \cite{lillicrap2015continuous} and independent A2C (IA2C) \cite{sutton2018reinforcement}, and two state-of-the-art methods with centralised critics, such as MADDPG \cite{lowe2017multi} and COMA \cite{foerster2018counterfactual}. To keep the fairness of comparison, the policy and critic networks for all MARL algorithms are parameterized by MLPs. All models are trained by the Adam optimizer \cite{kingma2014adam}. The details of experimental setups are given in Appendix \ref{sec:experimental_settings_of_benchmark_for_sqddpg}.
    
        \subsection{Cooperative Navigation}
        \label{subsec:cooperative_navigation}
            \begin{figure*}[ht!]
                \centering
                \begin{subfigure}[b]{0.35\textwidth}
                    \centering
                    \includegraphics*[width=\textwidth]{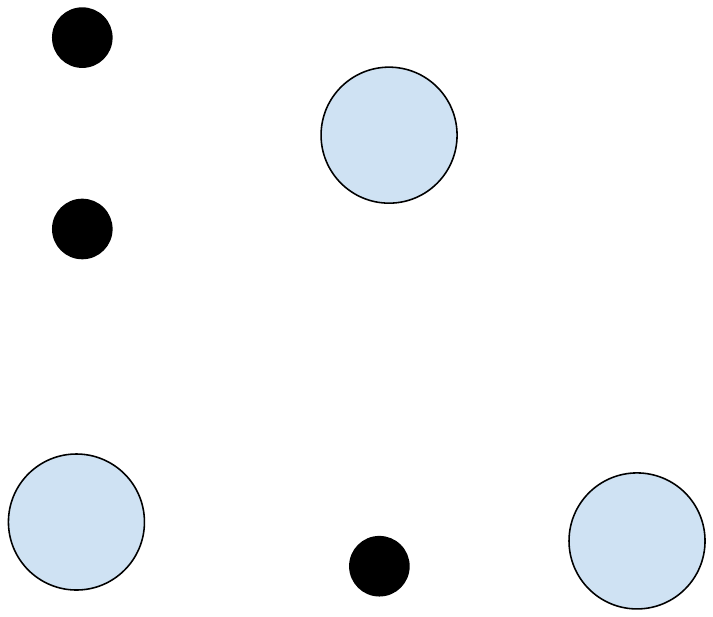}
                    \caption{Cooperative Navigation.}
                    \label{fig:cooperative_navigation_demo}
                \end{subfigure}%
                ~
                \begin{subfigure}[b]{0.35\textwidth}
                    \centering
                    \includegraphics*[width=\textwidth]{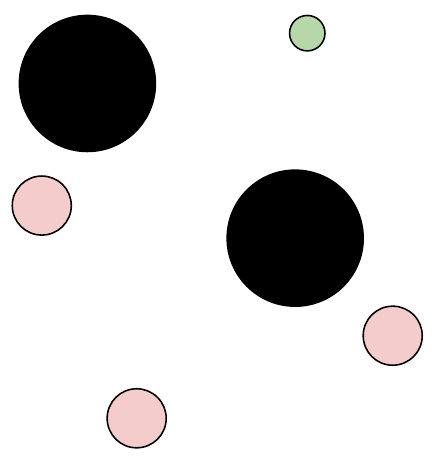}
                    \caption{Predator-Prey.}
                    \label{fig:predator_prey_demo}
                \end{subfigure}%
                \caption{Demonstration of the environments of Cooperative Navigation and Predator-Prey. (a) The circle in blue represents an agent and the circle in black represents a target. (b) The circle in red is a predator, the circle in green is a prey and the circle in black is an obstacle that cannot be crossed through.}
            \end{figure*}
            
            \paragraph{Environment Settings.} In this environment, there are 3 agents that are controllable and 3 targets, as shown in Figure \ref{fig:cooperative_navigation_demo}. The objective of each agent is moving towards a target, with no prior allocations of the targets to the agents. The observation of each agent in this environment involves the current position and velocity, the displacement to three targets, and the displacement to other agents. The action space of each agent includes \texttt{move\_up}, \texttt{move\_down}, \texttt{move\_right}, \texttt{move\_left} and \texttt{stay}. The global reward of this environment is defined as the negative sum of the distance between each target and the nearest agent to it. If a collision happens, the global reward will be reduced by 1.
            
            \paragraph{Performance Analysis.} As seen from Figure \ref{fig:3_agents_reward_spread}, the SQDDPGs with variant sample sizes (i.e., M in Eq.~\ref{eq:shapley_approx}) outperform the baselines on the convergence rate. We believe that if more training episodes are permitted, the algorithms except for IA2C can achieve the similar performance as SQDDPG. Therefore, our result supports the argument that the credit assignment method converges faster than learning with the shared reward approach \cite{balch1997learning,balch1999reward}. As the sample size grows, the approximate Shapley Q-value estimation will be more accurate and easier to converge to the optimal value. This explains the reason why the convergence rate of SQDDPG becomes faster when the sample size increases. Since we show that SQDDPG with the sample size of 1 can finally reach almost the same performance as other variants, we just run the SQDDPG with the sample size of 1 in the rest of experiments to reduce the computational complexity.
            \begin{figure}[ht!]
                \centering
                \includegraphics[scale=0.55]{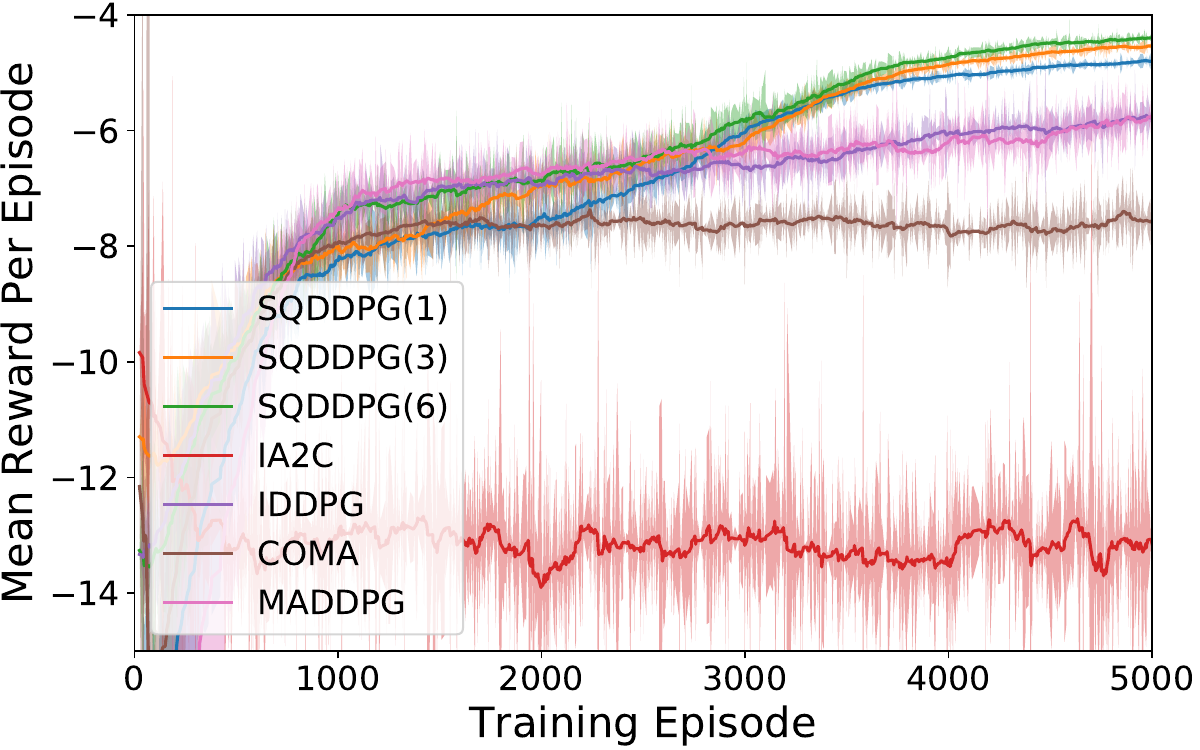}
                \caption{Mean reward per episode during training in Cooperative Navigation. SQDDPG(n) indicates SQDDPG with the sample size (i.e., M in Eq.~\ref{eq:shapley_approx}) of n. In the rest of experiments, since only SQDDPG with the sample size of 1 is run, we just use SQDDPG to represent SQDDPG(1).}
                \label{fig:3_agents_reward_spread}
            \end{figure}
            
        \subsection{Predator-Prey}
        \label{subsec:prey_predator}
            \paragraph{Environment Settings.} In this environment, the agents that can be controlled are 3 predators, while the prey is a random agent. The specific demonstration is shown in Figure \ref{fig:predator_prey_demo}. The aim of each predator is coordinating to capture the prey with the turns (timesteps) as less as possible. The observation of each predator involves the current position and velocity, the respective displacement to the prey and other predators, and velocity of the prey. The action space is the same as that defined in Cooperative Navigation. If the positions of any predator and the prey are overlapped, it indicates that the prey is captured. The global reward is defined as the negative minimal distance between the predators and the prey. If the prey is captured by any predator, the global reward will be 10 and the game terminates. 
            \begin{figure}[ht!]
                \centering
                \includegraphics[scale=0.55]{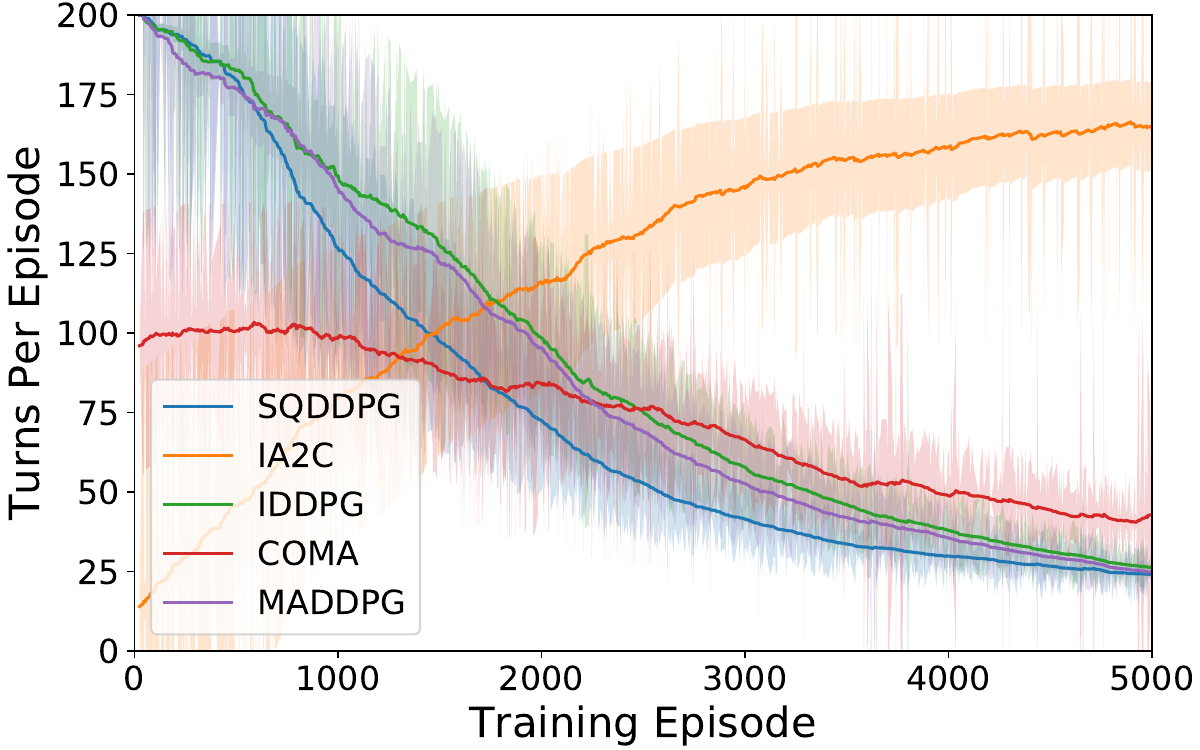}
                \caption{Turns to capture the prey per episode during training in Prey-and-Predator.}
                \label{fig:3_agents_turns_tag}
            \end{figure}
            
            \paragraph{Performance Analysis.} As Figure \ref{fig:3_agents_turns_tag} shows, SQDDPG converges fastest with around 25 turns to capture the prey, followed by MADDPG, IDDPG and COMA. This is because the fair credit assignment to each agent induced by the mechanism of Shapley value can well address the dummy agent problem we mentioned in Example \ref{exa:didactic_example} in Section \ref{subsec:credit_assignment}. IA2C is terribly the worst among all these algorithms which could suffer from the same issue as analysed for Cooperative Navigation.
            
        \subsection{Traffic Junction}
        \label{subsec:traffic_junction}
            \begin{figure*}[ht!]
            \centering
                \begin{subfigure}[b]{0.32\textwidth}
                    \includegraphics[width=\textwidth]{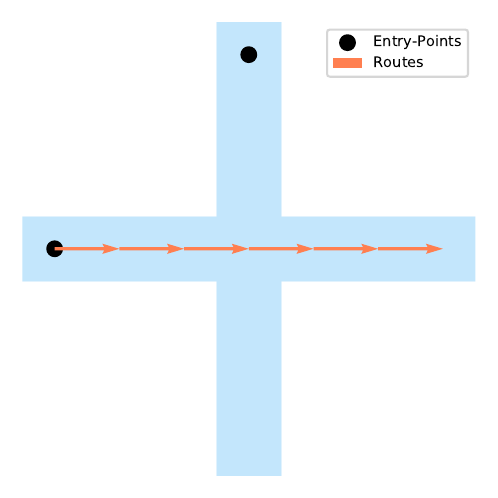}
                    \caption{Easy.}
                    \label{fig:easy_env}
                \end{subfigure}
                \begin{subfigure}[b]{0.32\textwidth}
                    \includegraphics[width=\textwidth]{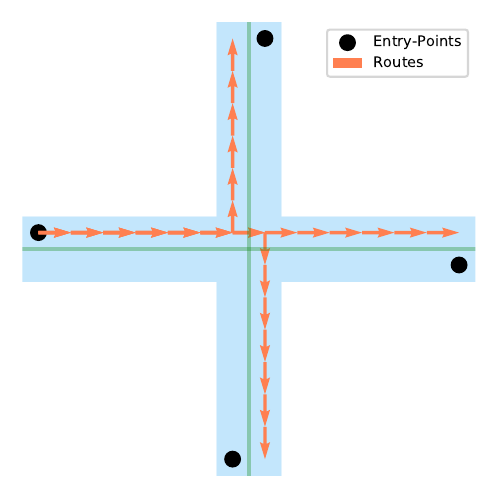} 
                    \caption{Medium.}
                    \label{fig:medium_env}
                \end{subfigure}
                \begin{subfigure}[b]{0.32\textwidth}
                    \includegraphics[width=\textwidth]{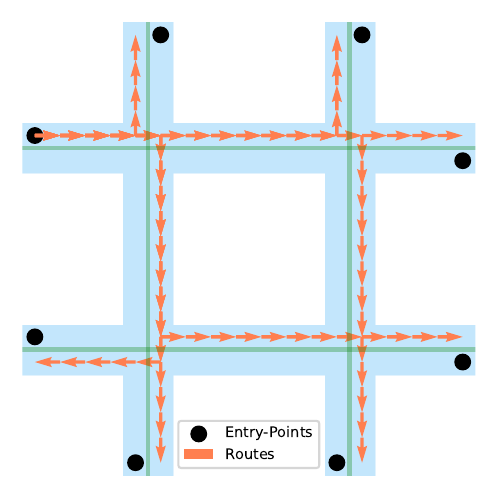}
                    \caption{Hard.}
                    \label{fig:hard_env}
                \end{subfigure}
                \caption{Visualizations of traffic junction environment. The black points represent the available entry points. The orange arrows represent the available routes at each entry point. The green lines separate the two-way roads.}
                \label{fig:tf_roads}
            \end{figure*}
            
            \paragraph{Environment Settings.} In this environment, cars move along the predefined routes which intersect on one or more traffic junctions. At each timestep, new cars enter into the environment with the probability $p_{\text{arrive}}$, and the total number of cars is restricted to $N_{\max}$. After a car finishes its mission, it will be removed from the environment and possibly sampled back onto a new route. Each car has a limited vision of 1, which means that it can only observe the circumstance within the 3x3 region surrounding it. No communication between cars is permitted in our experiment, in contrast to other experiments on the same task \cite{NIPS2016_6398,das2018tarmac}. The action space of each car includes \texttt{gas} and \texttt{brake}. The global reward function is $\sum_{i=1}^{N} \text{- }0.01t_{i}$, where $t_{i}$ is the timesteps that car $i$ is continuously active on the road in one mission and $N$ is the total number of cars. Additionally, if any collision happens, the global reward will be reduced by 10. We evaluate performance by the success rate, i.e., the proportion of episodes in which no collisions happen.
            
            \paragraph{Performance Analysis.} We compare SQDDPG with the baselines on the easy, medium and hard versions of Traffic Junction. The easy version is constituted of one traffic junction of two one-way roads on a $7\times7$ grid with $N_{\max}\text{ = }5$ and $p_{\text{arrive}}\text{ = }0.3$. The medium version is constituted of one traffic junction of two-way roads on a $14\times14$ grid with $N_{\max}\text{ = }10$ and $p_{\text{arrive}}\text{ = }0.2$. The hard version is constituted of four connected traffic junctions of two-way roads on a $18\times18$ grid with $N_{\max}\text{ = }20$ and $p_{\text{arrive}}\text{ = }0.05$. The demonstrations of the environments are shown in Figure \ref{fig:tf_roads}. From Table \ref{tab:traffic_junction}, we can see that on the easy version, except for IA2C, other algorithms can achieve the success rates over $93\%$, since this scenario is too easy. On the medium and hard versions, SQDDPG outperforms the baselines with the success rate of $88.98\%$ on the medium version and $87.04\%$ on the hard version, which demonstrates that SQDDPG is capable of solving the large-scale problems. Furthermore, the performance of SQDDPG significantly exceeds no-communication algorithms' performance reported as $84.9\%$ and $74.1\%$ in \cite{das2018tarmac}.
            \begin{table*}[ht!]
            \centering
            \caption{Success rate on Traffic Junction, tested with 20, 40, and 60 steps per episode in easy, medium and hard versions respectively. The results are obtained by running each algorithm after training for 1000 episodes.}
            \begin{tabular}{cccccc}
                \toprule
                    \textbf{Difficulty} & \textbf{IA2C} & \textbf{IDDPG} & \textbf{COMA} & \textbf{MADDPG} & \textbf{SQDDPG}  \\ 
                \midrule 
                    Easy & 65.01\%   & 93.08\%  & 93.01\%   & \textbf{93.72\%}   & 93.26\%  \\
                    Medium  & 67.51\%  & 84.16\% & 82.48\%  & 87.92\%         & \textbf{88.98\%} \\
                    Hard  & 60.89\% & 64.99\%  & 85.33\%   & 84.21\%         & \textbf{87.04\%} \\ 
            \bottomrule
            \end{tabular}
            \label{tab:traffic_junction}
            \end{table*}
            
        \subsection{Understanding Markov Shapley Value}
        \label{subsec:interpretability_of_sqddpg}
            To study and interpret the credit assignment, we visualize the Q-values of each MARL algorithm for one randomly selected trajectory of states and actions from an expert policy on Predator-Prey.\footnote{Note that the selected expert policy could be highly likely sub-optimal compared with the optimal policy forming the learned optimal Q-values.} For visualizing conveniently, we normalize the Q-values by min-max normalization \cite{patro2015normalization} for each MARL algorithm. We can see from Figure \ref{fig:credit_vis} that the credit assignment of SQDDPG is more interpretable than the baselines. Specifically, it is intuitive that the credit assigned to each agent by SQDDPG is inversely proportional to its distance to the prey. On the contrary, other MARL algorithms do not explicitly show this property. To verify the hypothesis, we also evaluate it quantitatively by Pearson correlation coefficient \cite{pearson1895vii} with 1000 randomly selected transition samples, to summarize the correlation between the credit assignment and the reciprocal of each predator's distance to the prey. The value of Pearson correlation coefficient is greater, the stronger the inverse proportion. As Table \ref{tab:correlation} shows, SQDDPG expresses the inverse proportion significantly, with the Pearson correlation coefficient as 0.3210. If a predator is closer to the prey, it is more likely to capture it and that predator's contribution should be more significant. Consequently, we demonstrate that Markov Shapley Q-value as a credit assignment scheme reflects the contribution to the team.
            \begin{figure*}[ht!]
                \centering
                \begin{subfigure}[b]{\textwidth}
                    \centering
                    \includegraphics*[width=\textwidth]{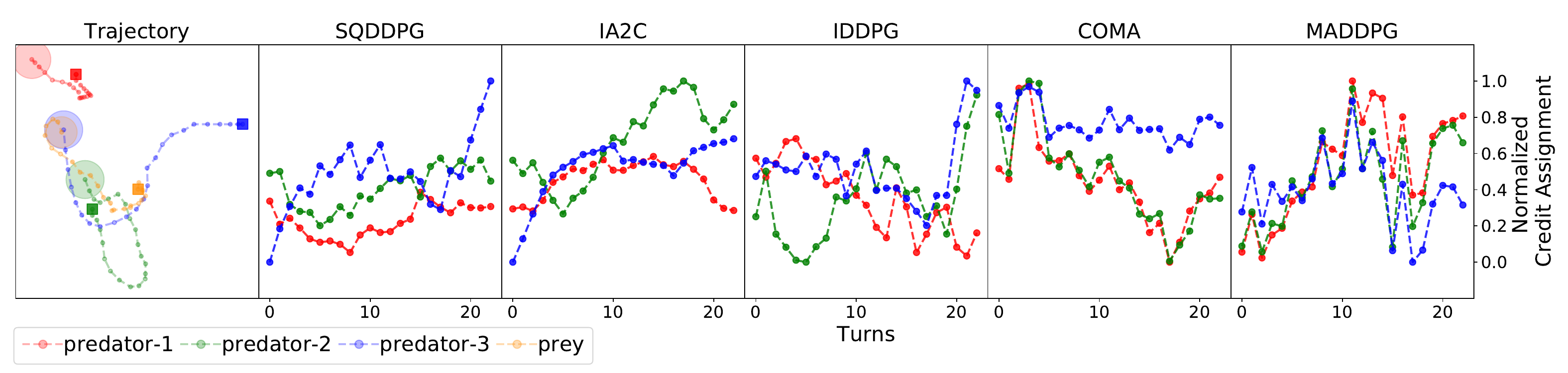}
                \end{subfigure}%
                \\
                \begin{subfigure}[b]{\textwidth}
                    \centering
                    \includegraphics*[width=\textwidth]{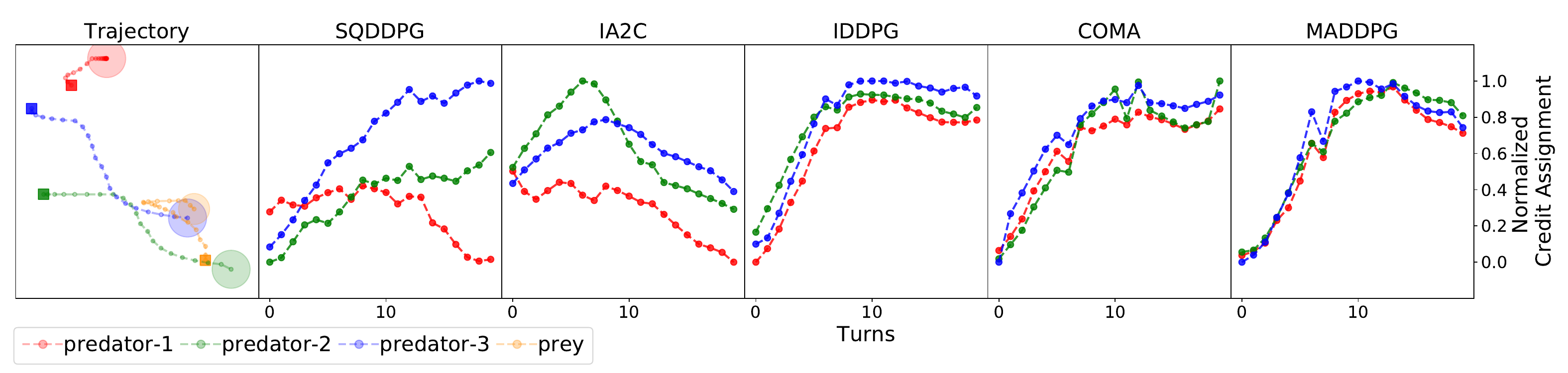}
                \end{subfigure}%
                \\
                \begin{subfigure}[b]{\textwidth}
                    \centering
                    \includegraphics*[width=\textwidth]{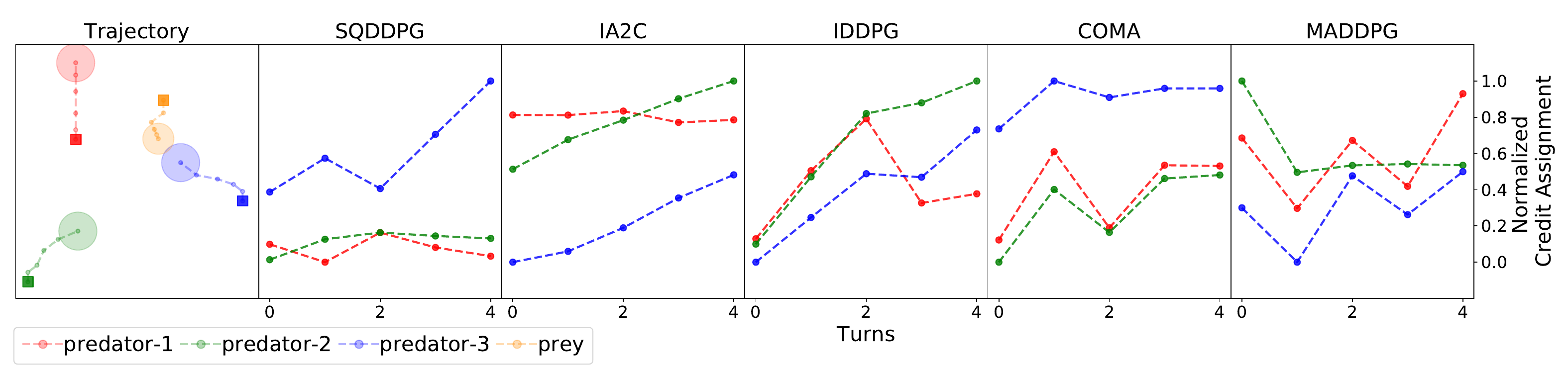}
                \end{subfigure}
                \\
                \begin{subfigure}[b]{\textwidth}
                    \centering
                    \includegraphics*[width=\textwidth]{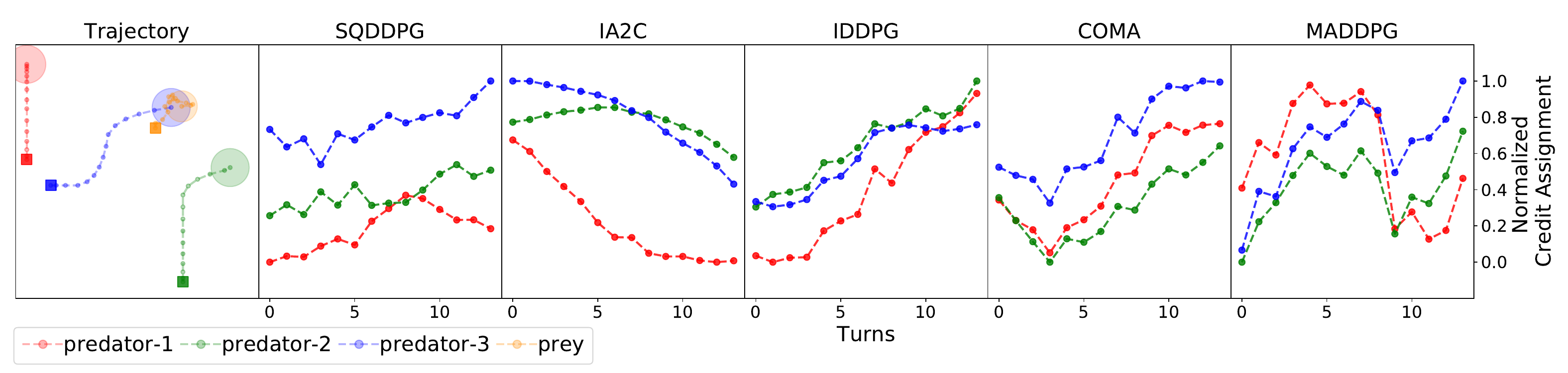}
                \end{subfigure}
                \\
                \begin{subfigure}[b]{\textwidth}
                    \centering
                    \includegraphics*[width=\textwidth]{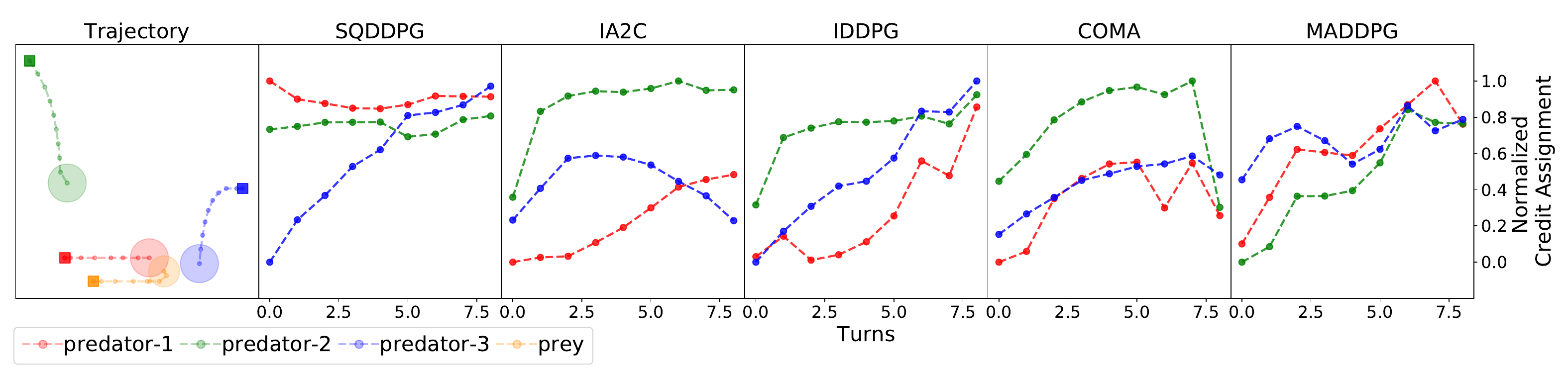}
                \end{subfigure}
                \caption{Credit assignment to each predator for a fixed trajectory. The leftmost figure records a trajectory sampled by an expert policy. The square represents the initial position, whereas the circle indicates the final position of each agent. The dots on the trajectory indicates each agent's temporary positions. The other figures show the normalized credit assignments generated by different MARL algorithms according to this trajectory.}
                \label{fig:credit_vis}
            \end{figure*}
            
            \begin{table*}[ht!]
                \centering
                \caption{Pearson correlation coefficient between the credit assignment to each predator and the reciprocal of its distance to the prey. This test is conducted by 1000 randomly selected episode samples.}
                \begin{tabular}{cccccc}
                    \toprule
                        \textbf{} & \textbf{IA2C} & \textbf{IDDPG} & \textbf{COMA} & \textbf{MADDPG} & \textbf{SQDDPG}  \\ 
                    \midrule 
                        Pearson correlation coefficient        &  0.0508  & 0.0061 & 0.1274 & 0.0094 & \textbf{0.3210} \\
                        two-tailed p-value &  1.6419e-1  & 8.6659e-1 & 4.6896e-4 & 7.9623e-1 & \textbf{1.9542e-19} \\
                \bottomrule
                \end{tabular}
                \label{tab:correlation}
            \end{table*}
        
    \section{Evaluation of SHAQ}
    \label{sec:evaluation_of_SHAQ}
        In this section, we show the evaluation of SHAQ on Predator-Prey \cite{bohmer2020deep} and various tasks in StarCraft Multi-Agent Challenge (SMAC).\footnote{The version that we use in this paper is SC2.4.6.2.69232 rather than the newer SC2.4.10. As reported from \cite{rashid2020weighted}, the performance is not comparable across versions.} The baselines that we select for comparison are COMA \cite{foerster2018counterfactual}, VDN \cite{SunehagLGCZJLSL18}, QMIX \cite{RashidSWFFW18}, MASAC \cite{iqbal2019actor}, QTRAN \cite{SonKKHY19}, QPLEX \cite{wang2020qplex} and W-QMIX (including CW-QMIX and OW-QMIX) \cite{rashid2020weighted}. The implementation details of SHAQ are shown in Appendix \ref{subsec:implementation_details_shapley_q_learning}, whereas the implementation of baselines are from \cite{rashid2020weighted}.\footnote{The source code of baseline implementation is from \url{https://github.com/oxwhirl/wqmix}.} For all experiments, we use the $\epsilon$-greedy exploration strategy, where $\epsilon$ is annealed from 1 to 0.05. The annealing timesteps vary among different experiments. For Predator-Prey, we apply 1 million timesteps for annealing, following the setup from \cite{wang2020qplex}. For the easy and hard maps in SMAC, we apply 50k time steps for annealing, the same as that leveraged in \cite{samvelyan2019starcraft}; while for the super-hard maps in SMAC, we apply 1 million timesteps for annealing to acquire more explorations so that more state-action pairs can be visited. About the replay buffer size, we set it as 5000 for all algorithms that is the same as \cite{rashid2020weighted}. To fairly evaluate all algorithms, we run each experiment with 5 random seeds. All figures showing experimental results are plotted with the median and 25\%-75\% quartile shading. The ablation study of SHAQ is shown in Section \ref{subsec:ablation_study}.
        
        \subsection{Predator-Prey}
        \label{subsec:predator-prey}
            \begin{figure}[ht!]
                \centering
                \includegraphics[scale=0.55]{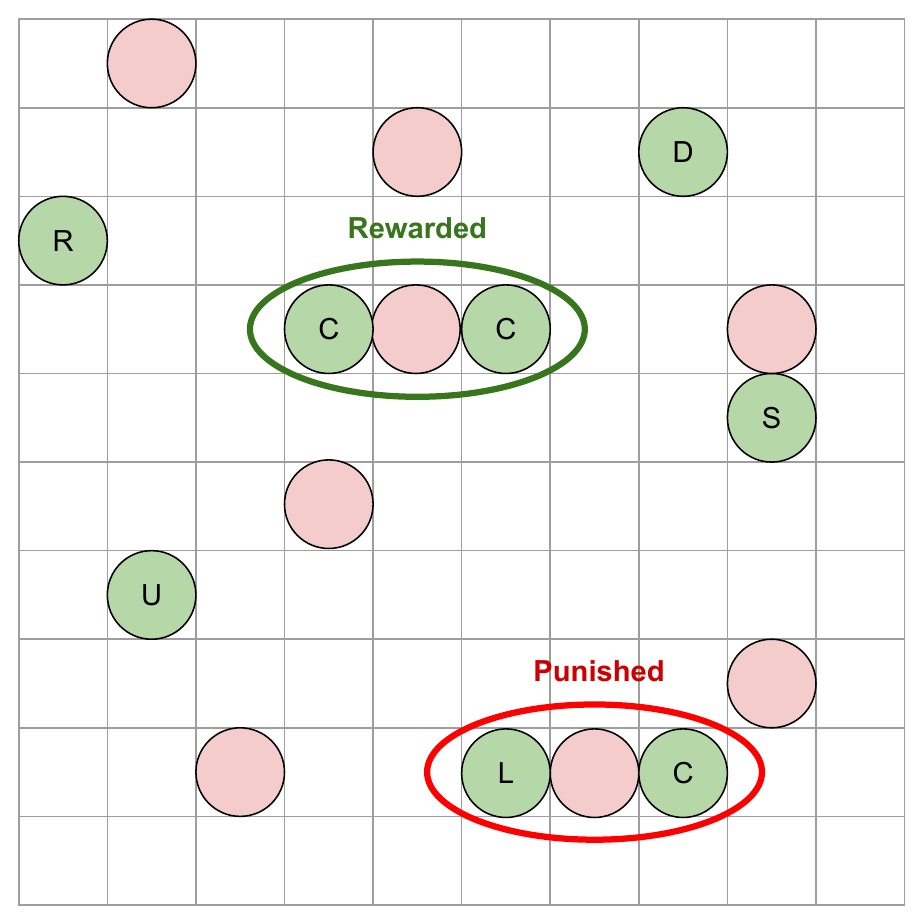}
                \caption{Environment of the grid-world version of Predator-Prey. The circle in red indicates a prey, while the circle in green indicates a predator. The alphabet insides a green circle indicates an action such that ``R'' means moving right, ``L'' means moving left, ``C'' means capturing, ``U'' means moving up, ``D'' means moving down, and ``S'' means staying (i.e., doing nothing).}
            \label{fig:predator_prey_demo_2}
            \end{figure}
            
            \begin{figure*}[ht!]
                \centering
                \begin{subfigure}[b]{0.48\textwidth}
                    \centering
                    \includegraphics[width=\textwidth]{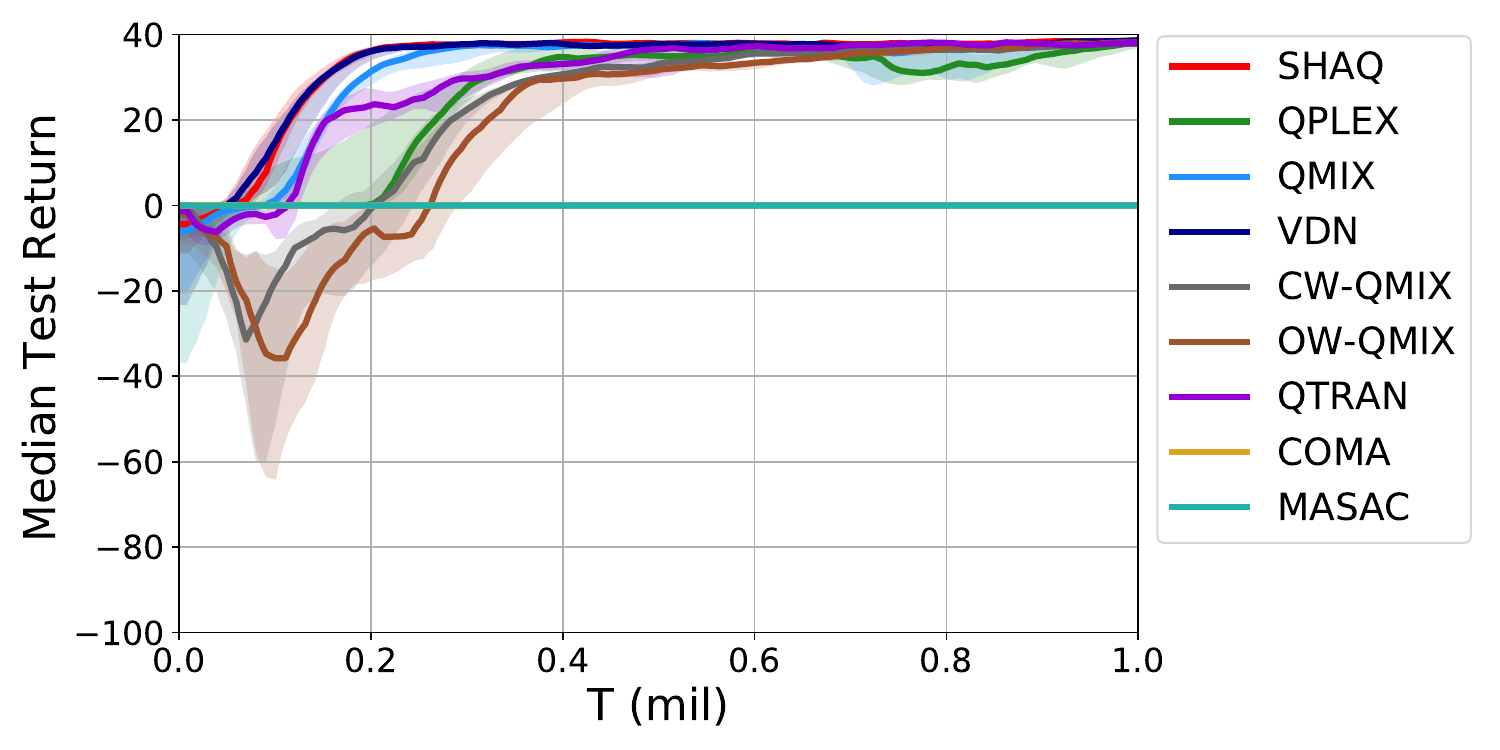}
                    \caption{p=-0.5.}
                \end{subfigure}
                ~
                \begin{subfigure}[b]{0.48\textwidth}
                    \centering
                    \includegraphics[width=\textwidth]{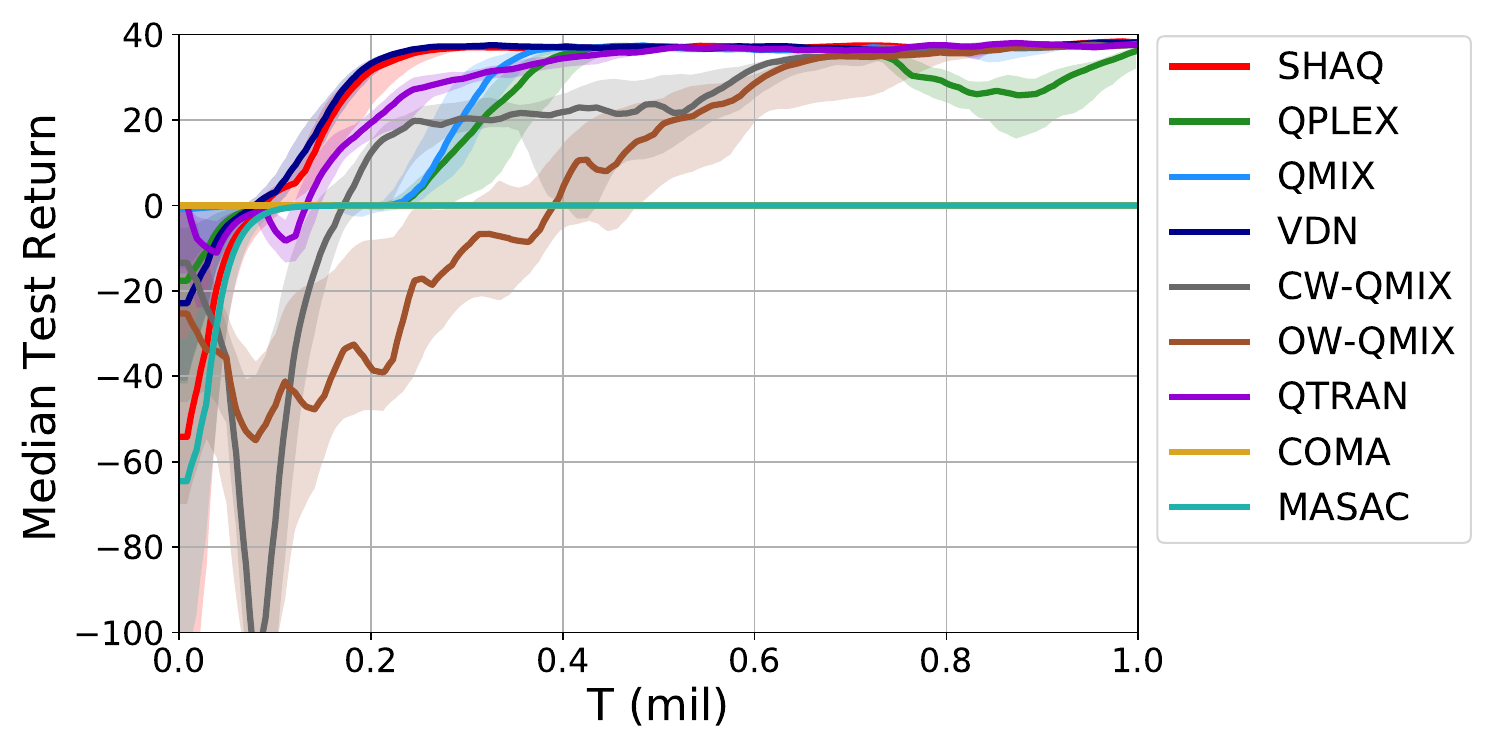}
                    \caption{p=-1.}
                \end{subfigure}
                ~
                \begin{subfigure}[b]{0.48\textwidth}
                    \centering
                    \includegraphics[width=\textwidth]{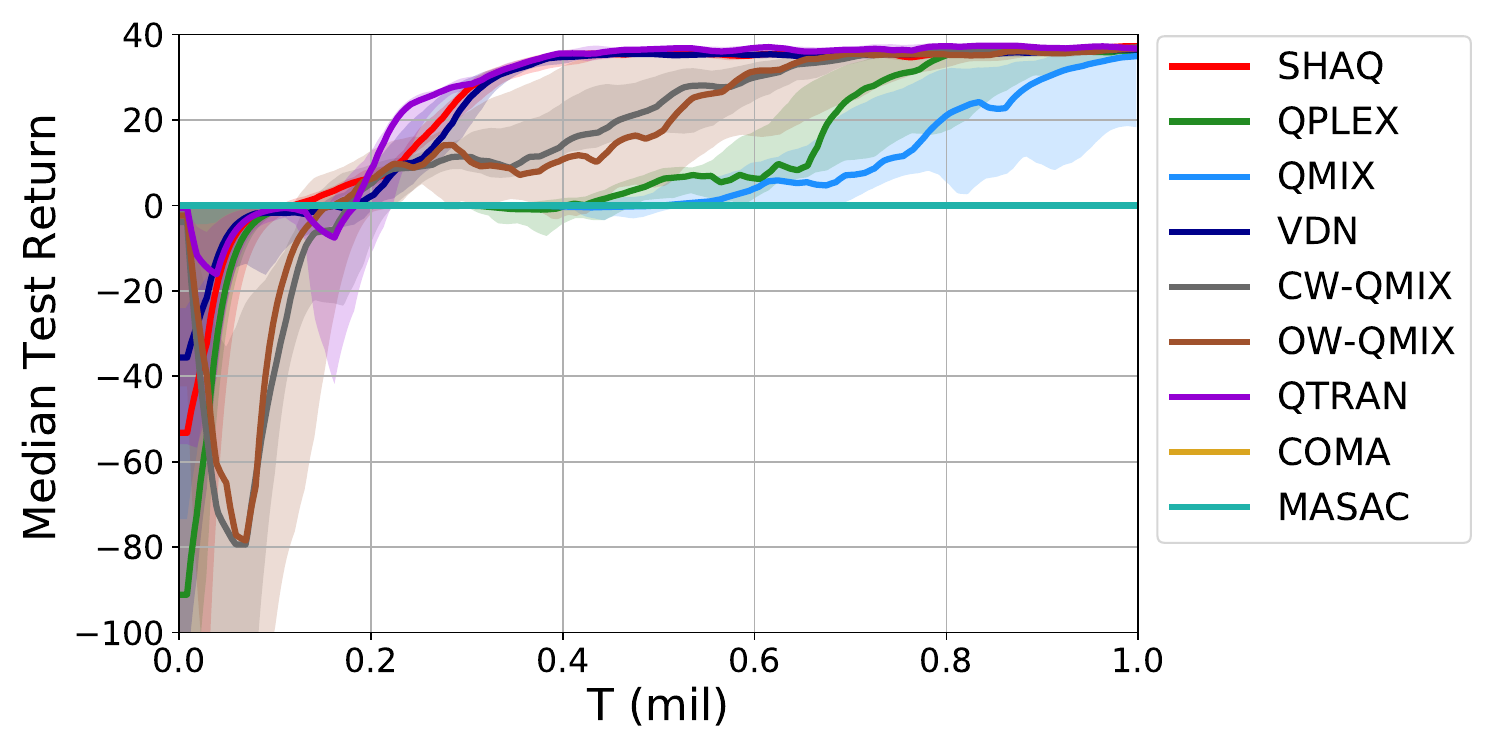}
                    \caption{p=-2.}
                \end{subfigure}
                \caption{Median test return for Predator-Prey with different values of p.}
            \label{fig:predator_prey}
            \end{figure*}
            
            \paragraph{Environment Settings.} In this environment, the world is formed as a 10x10 grid, where 8 predators are controllable, aiming at capturing 8 preys that randomly move \cite{bohmer2020deep}. Each predator's observation is a 5x5 sub-grid centering around it. If a prey is captured by the coordination between 2 predators, these 2 predators will be rewarded by 10. On the other hand, each unsuccessful attempt by only 1 predator will be punished by a negative reward p. In this experiment, we study the performance of each algorithm under different values of p (that describes different levels of coordination). The illustration of this environment is shown in Figure \ref{fig:predator_prey_demo_2}.
            
            \paragraph{Performance Analysis.} As \cite{rashid2020weighted} reported, only QTRAN and W-QMIX can solve this task, while \cite{wang2020qplex} found that the failure was primarily due to the lack of explorations. As a result, we apply the identical epsilon annealing schedule (i.e. 1 million time steps) employed in \cite{wang2020qplex}. As Figure \ref{fig:predator_prey} shows, SHAQ can always solve the tasks with different values of p. With the epsilon annealing strategy from \cite{wang2020qplex}, W-QMIX does not perform as well as reported in \cite{rashid2020weighted}. The reason could be its poor robustness to the increased explorations \cite{rashid2020weighted} for this environment (see the evidential experimental results in Figure \ref{fig:wqmix_pp} in Appendix \ref{subsec:extra_weighted_qmix_variant_hyperparameters}). The good performance of VDN validates our analysis in Section \ref{subsec:comparison_with_other_learning_algorithms}, whereas the performance of QTRAN is surprisingly almost invariant to the value of p. Furthermore, the performance of QPLEX and QMIX becomes apparently worse when p=-2. The failure of MASAC and COMA could be due to that relative overgeneralisation\footnote{Relative overgeneralisation is a common game theoretic pathology that the sub-optimal actions are preferred when matched with arbitrary actions from the collaborating agents \cite{wei2016lenient}.} prevents policy gradient methods from better coordination \cite{wei2018multiagent}.
        
        \subsection{StarCraft Multi-Agent Challenge}
        \label{subsec:smac}
            \paragraph{Environment Settings.} We now evaluate SHAQ on the more challenging StarCraft Multi-Agent Challenge (SMAC), the environmental settings of which are the same as that in \cite{samvelyan2019starcraft}. StarCraft II is a real-time strategy game that simulates a battle between two armies of units. A group of agents are controlled by the learned MARL algorithms, while the other group are controlled by the built-in game AI. SMAC creates some scenarios (maps) based on the game engine of StarCraft II. The goal of SMAC is maximizing the winning rate, i.e., the ratio of games won to the games played. Each agent's observation space is constructed based the following features within the sight range: $\texttt{distance}$, $\texttt{relative x}$, $\texttt{relative y}$, $\texttt{health}$, $\texttt{shield}$, and $\texttt{unit\_type}$. The shield is a protector that protects the agents from attacks which needs to be destroyed prior to reducing the health. The action space is constituted of the following discrete actions: $\texttt{move[direction]}$, $\texttt{attack[enemy\_id]}$, $\texttt{stop}$, and $\texttt{no-op}$. $\texttt{no-op}$ indicates conducting no operation and the dead agent can only execute this action. The number of combinations of above actions ranges from 7 to 70, depending on the scenarios. The reward function is shaped based on the hit-point damage dealt and enemy units killed (i.e., rewarding 10), along with the the bonus for the winning of a battle (i.e., rewarding 200). The rewards are scaled to improve the training stability, so that the maximum cumulative rewards that can be achieved in each scenario is around 20. To broadly compare the performance of SHAQ with other baselines, we select 4 easy maps: 8m, 3s5z, 1c3s5z and 10m\_vs\_11m; 3 hard maps: 5m\_vs\_6m, 3s\_vs\_5z and 2c\_vs\_64zg; and 4 super-hard maps: 3s5z\_vs\_3s6z, Corridor, MMM2 and 6h\_vs\_8z. All training is through online data collection. To maintain the conciseness, we only show partial results in the main part of this thesis and leave the rest results in Appendix \ref{subsec:experimental_results_on_extra_smac_maps}. 
            
            \paragraph{Performance Analysis.} It shows in Figure \ref{fig:easy_hard_smac} that SHAQ outperforms all baselines on all maps, except for 6h\_vs\_8z. On 6h\_vs\_8z, SHAQ can beat all baselines except for CW-QMIX. VDN performs well on 4 maps, but bad on the other 2 maps, which still validates our analysis in Section \ref{subsec:comparison_with_other_learning_algorithms}. QMIX and QPLEX perform well on the most of maps, except for 3s\_vs\_5z, 2c\_vs\_64zg and 6h\_vs\_8z. As for COMA, MADDPG and MASAC, their poor performances could be due to the weak adaptability to challenging tasks. Although QTRAN can theoretically represent the complete class of the global Q-value \cite{SonKKHY19}, its complicated learning paradigm could impede the convergence to the value function for challenging tasks and therefore result in the poor performance. Although W-QMIX performs well on some maps, owing to lacking a law on hyperparameter tuning \cite{rashid2020weighted} it is difficult to be adapted to all scenarios (see Appendix \ref{subsec:extra_weighted_qmix_variant_hyperparameters}).
            \begin{figure*}[ht!]
                \centering
                \begin{subfigure}[b]{0.48\linewidth}
                    \centering
                    \includegraphics[width=\textwidth]{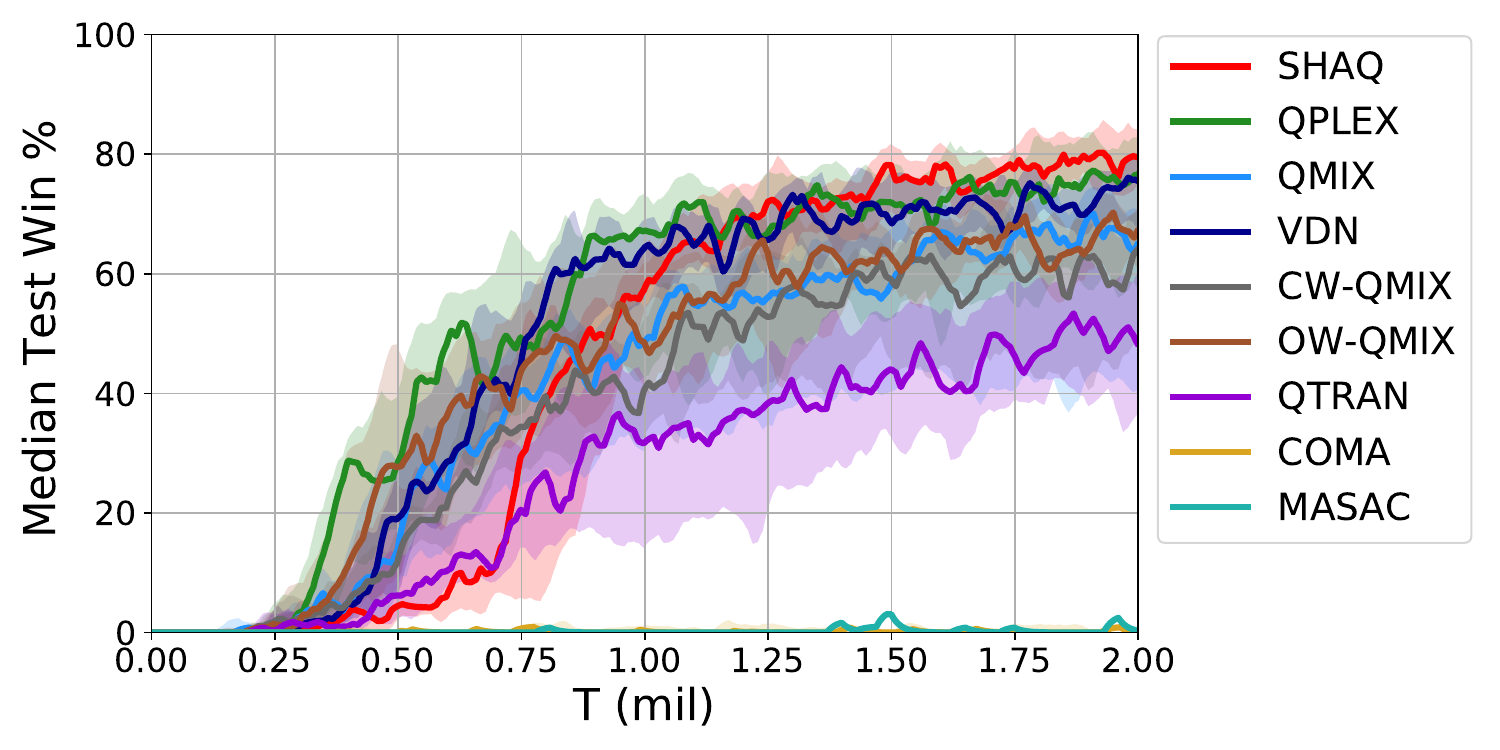}
                    \caption{5m\_vs\_6m.}
                \end{subfigure}
                ~
                \begin{subfigure}[b]{0.48\linewidth}
                    \centering
                    \includegraphics[width=\textwidth]{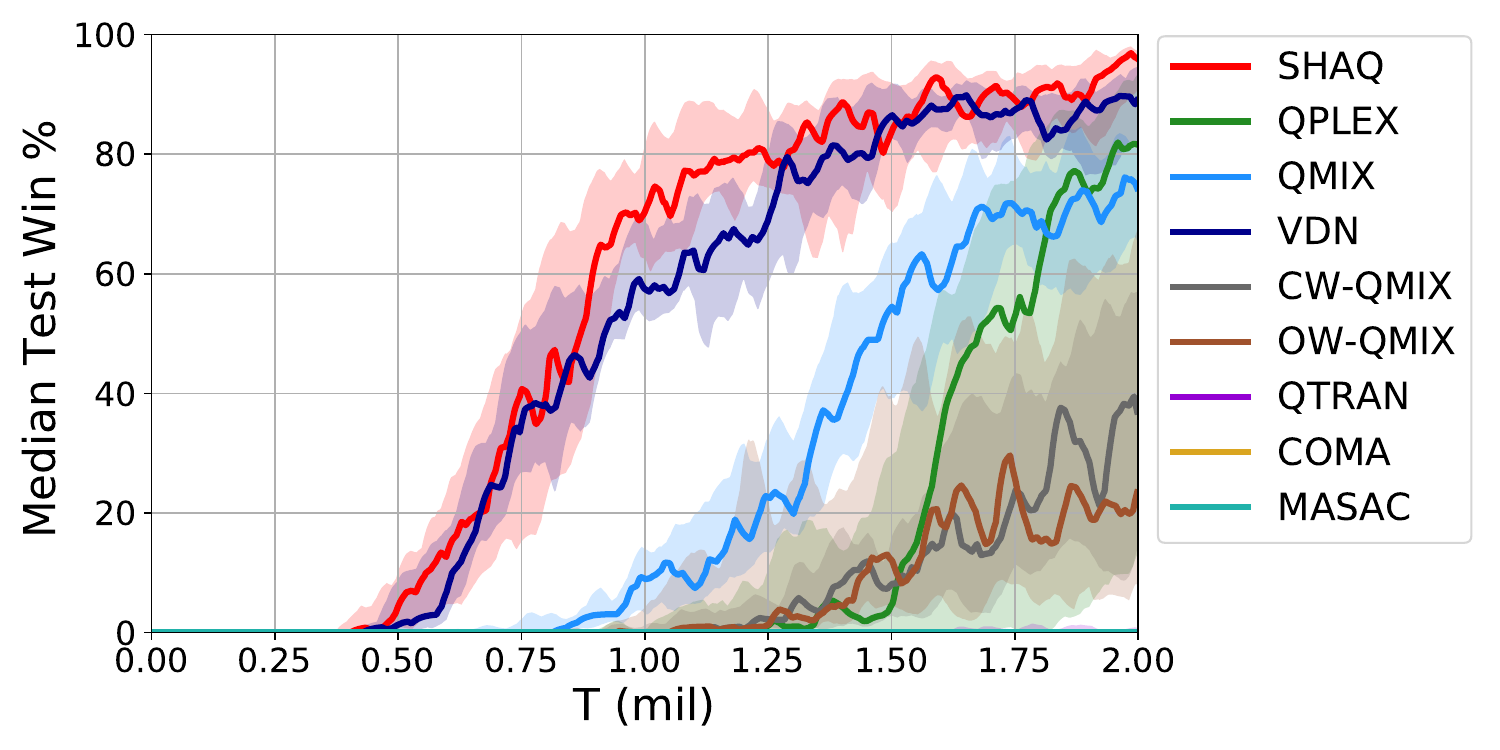}
                    \caption{3s\_vs\_5z.}
                \end{subfigure}
                ~
                \begin{subfigure}[b]{0.48\linewidth}
                    \centering              
                    \includegraphics[width=\textwidth]{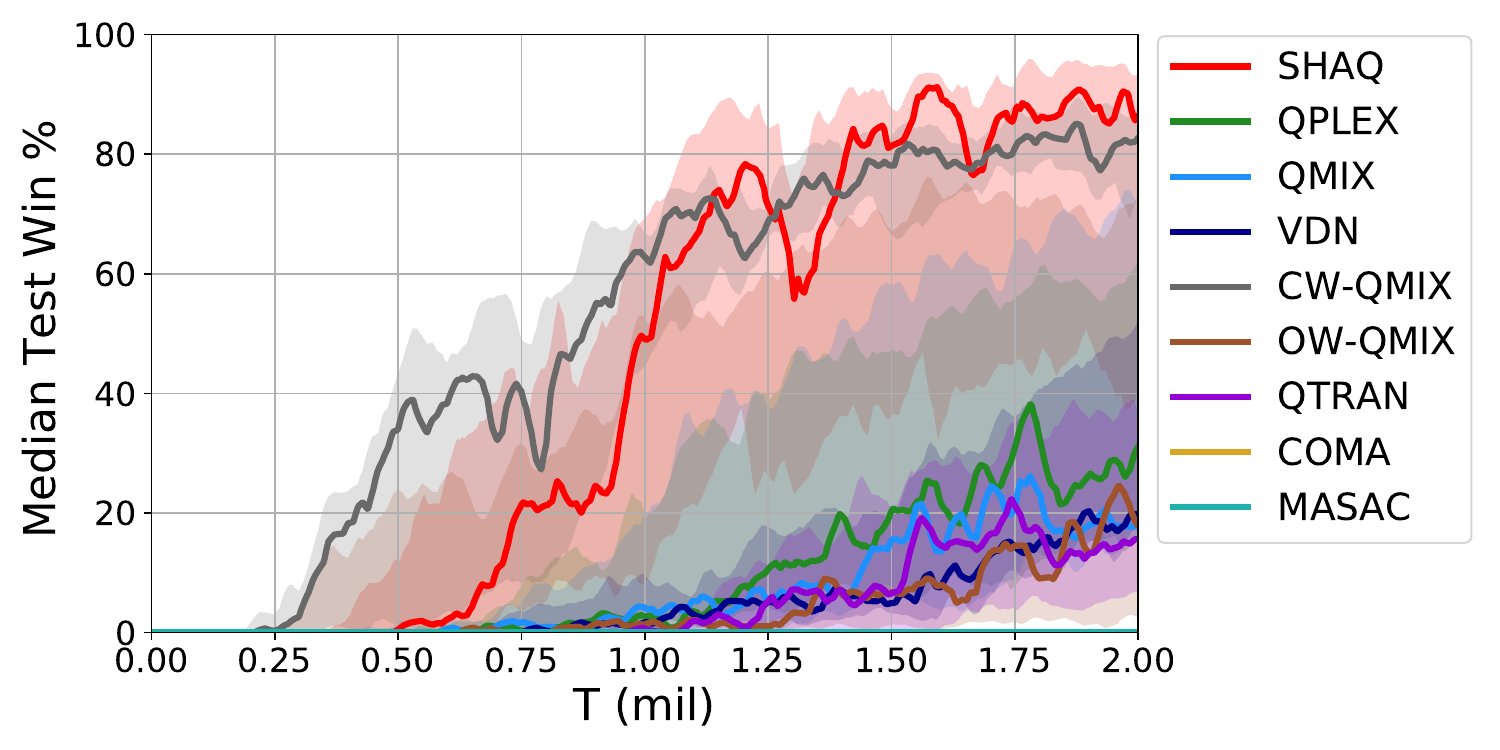}
                    \caption{2c\_vs\_64zg.}
                \end{subfigure}
                ~
                \begin{subfigure}[b]{0.48\linewidth}
                    \centering
                    \includegraphics[width=\textwidth]{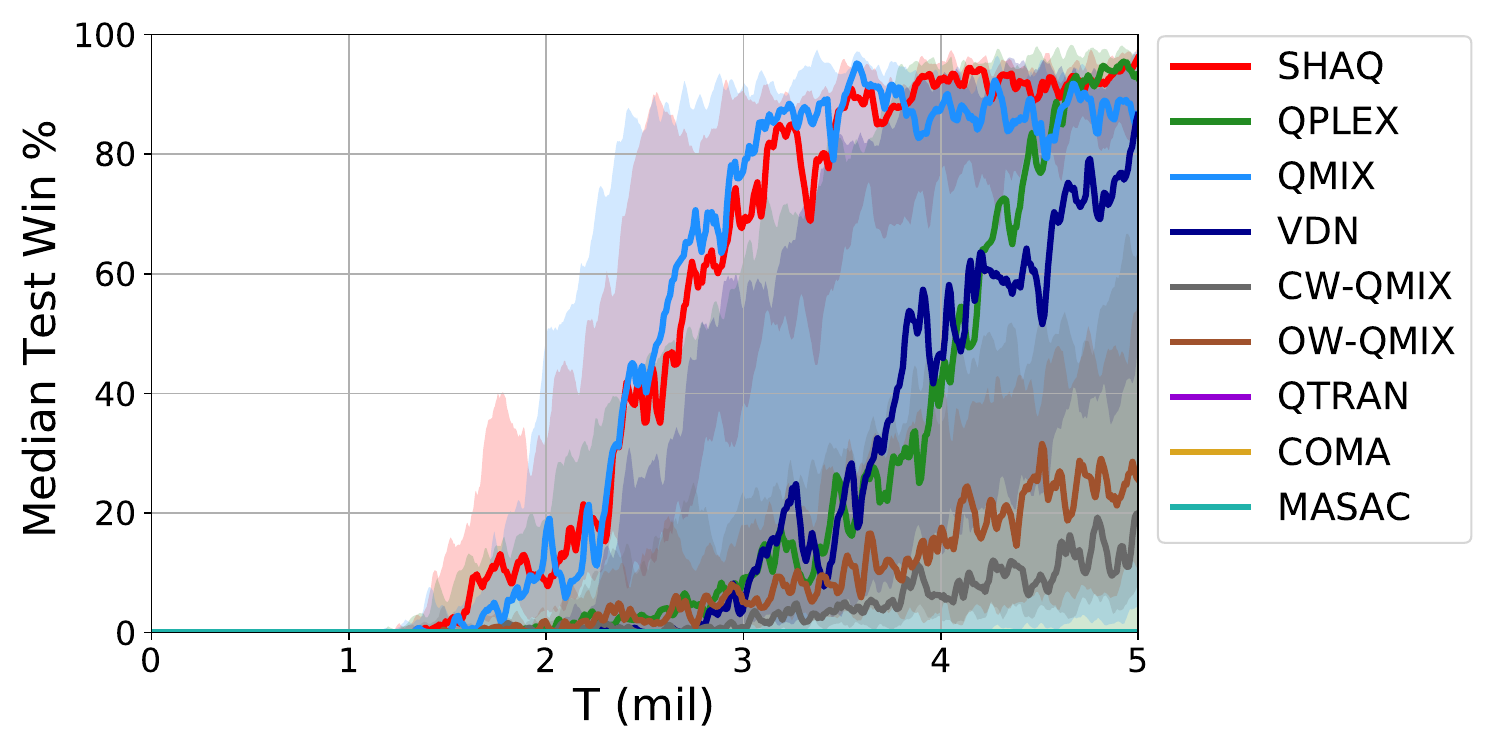}
                    \caption{3s5z\_vs\_3s6z.}
                \end{subfigure}
                ~
                \begin{subfigure}[b]{0.48\linewidth}
                    \centering
                    \includegraphics[width=\textwidth]{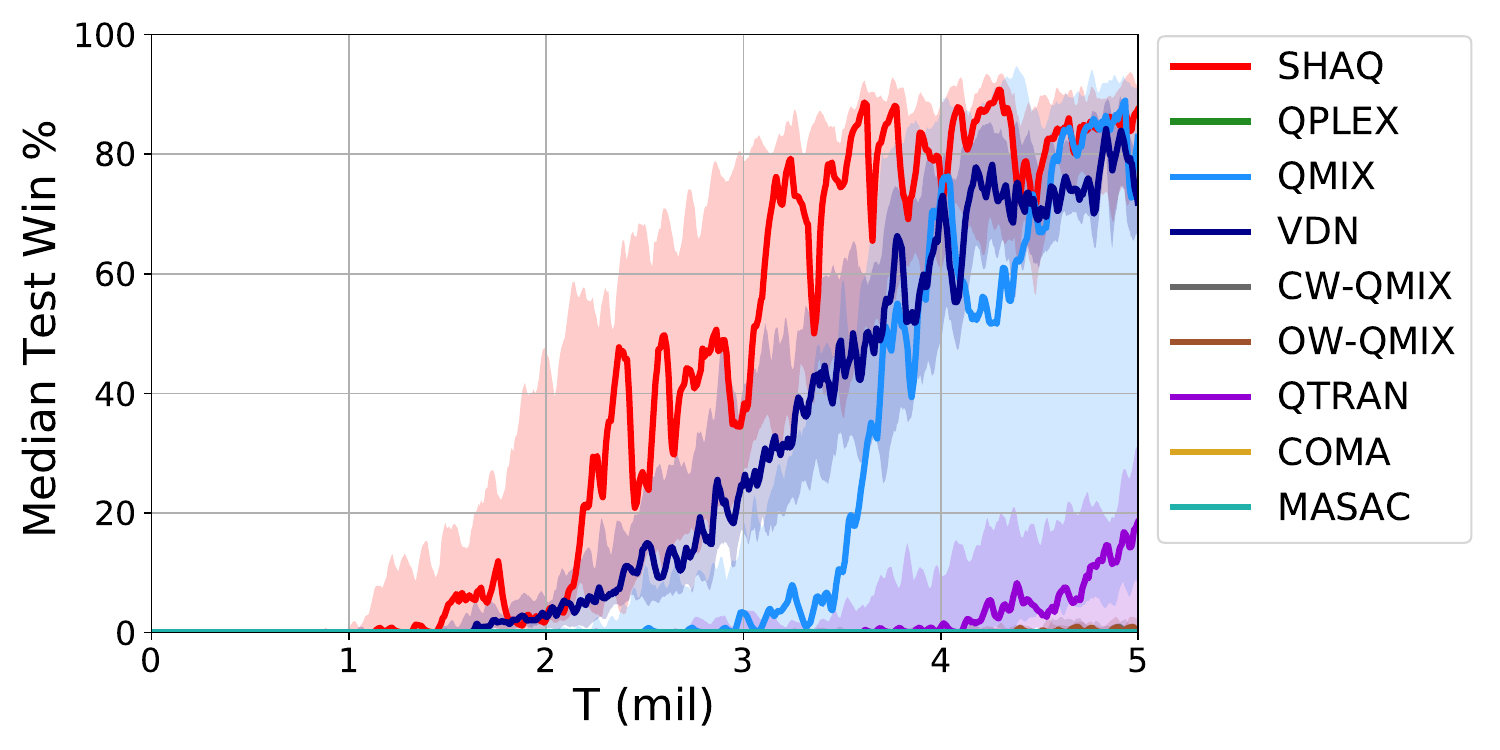}
                    \caption{Corridor.}
                \end{subfigure}
                ~
                \begin{subfigure}[b]{0.48\linewidth}
                    \centering
                    \includegraphics[width=\textwidth]{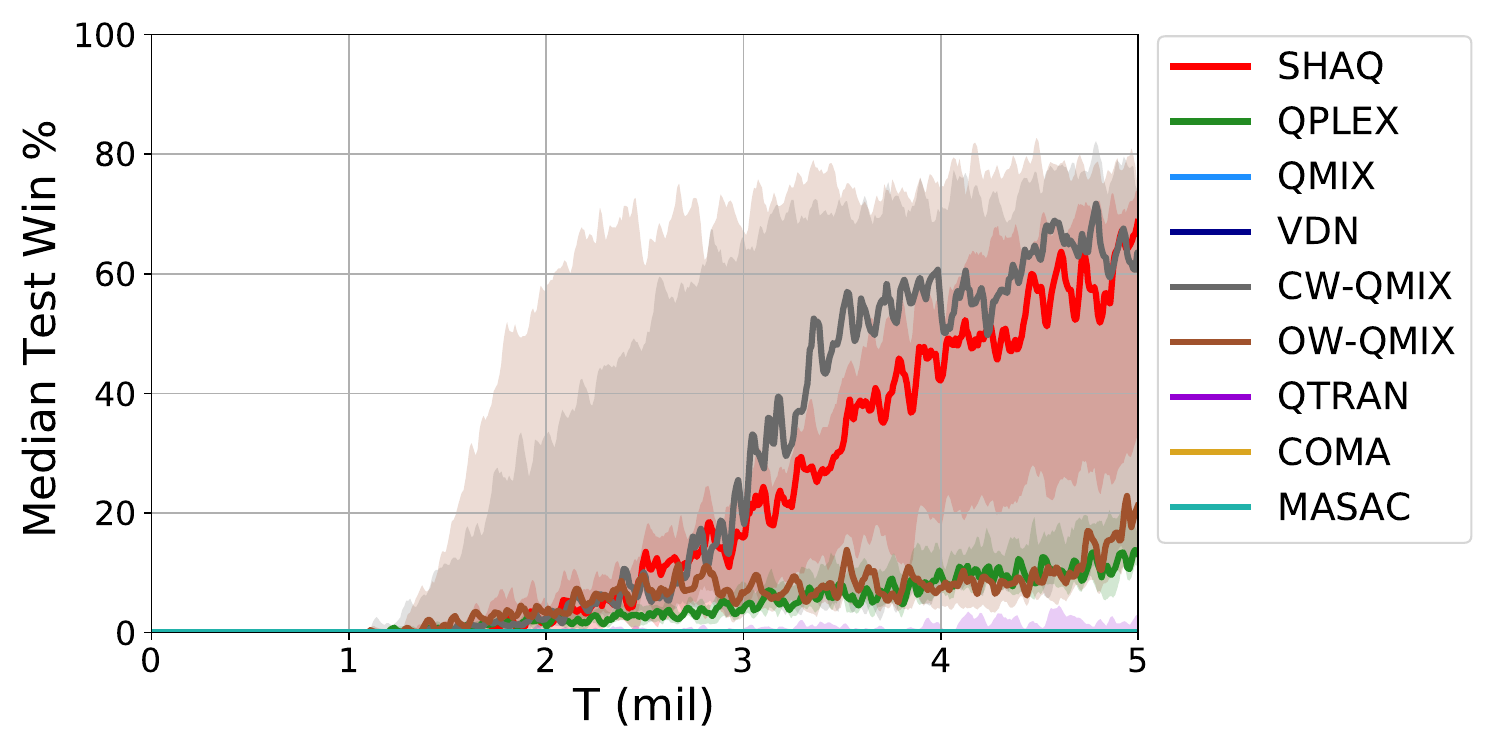}
                    \caption{6h\_vs\_8z.}
                \end{subfigure}
                \caption{Median test win \% for hard (a-c), and super-hard (d-f) maps of SMAC.}
            \label{fig:easy_hard_smac}
            \end{figure*}
            
        \subsection{Ablation Study}
        \label{subsec:ablation_study}
            \begin{figure*}[ht!]
            \centering
                \begin{subfigure}[b]{0.48\textwidth}
                    \centering                \includegraphics[width=\textwidth]{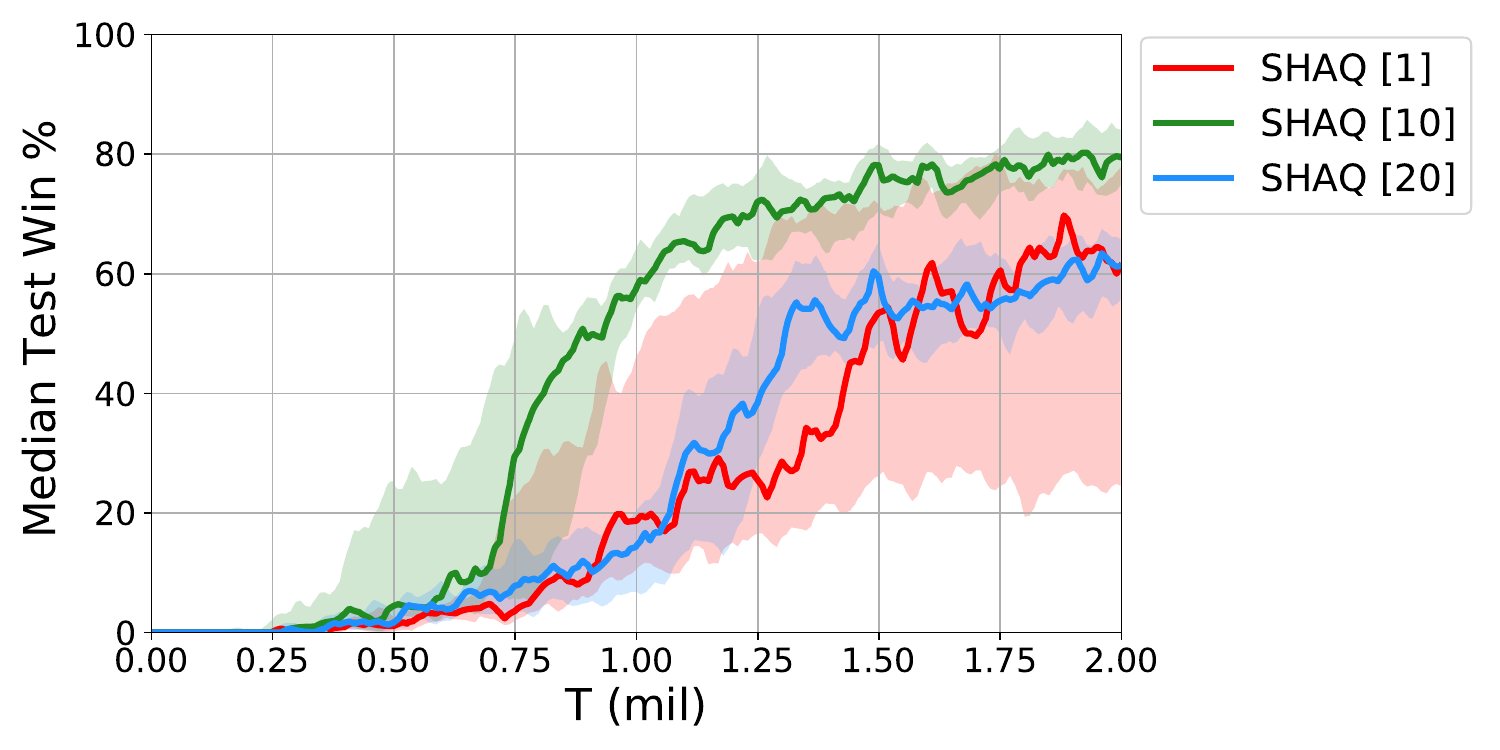}
                    \caption{Comparison among different values of M on 5m\_vs\_6m. The $[\cdot]$ indicates the value of M.}
                \label{fig:ablation_5m_vs_6m}
                \end{subfigure}
                \quad
                \begin{subfigure}[b]{0.48\textwidth}
                    \centering                 \includegraphics[width=\textwidth]{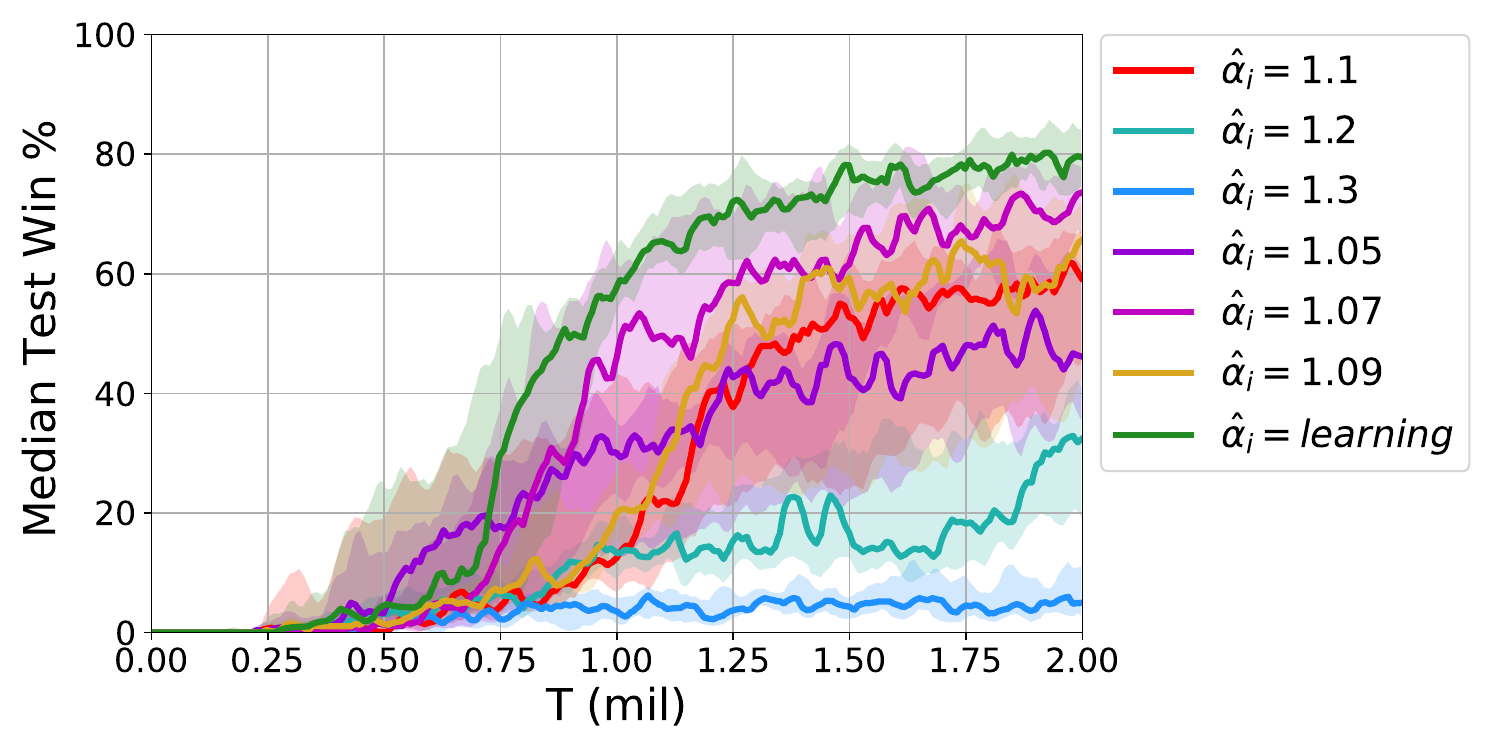}
                    \caption{Comparison between the manually preset and the learned $\hat{\alpha}_{i}(s, a_{i})$ on 5m\_vs\_6m.}
                \label{fig:manual_approximate_alpha}
                \end{subfigure}
                \quad
                \begin{subfigure}[b]{0.48\textwidth}
                    \centering                \includegraphics[width=\textwidth]{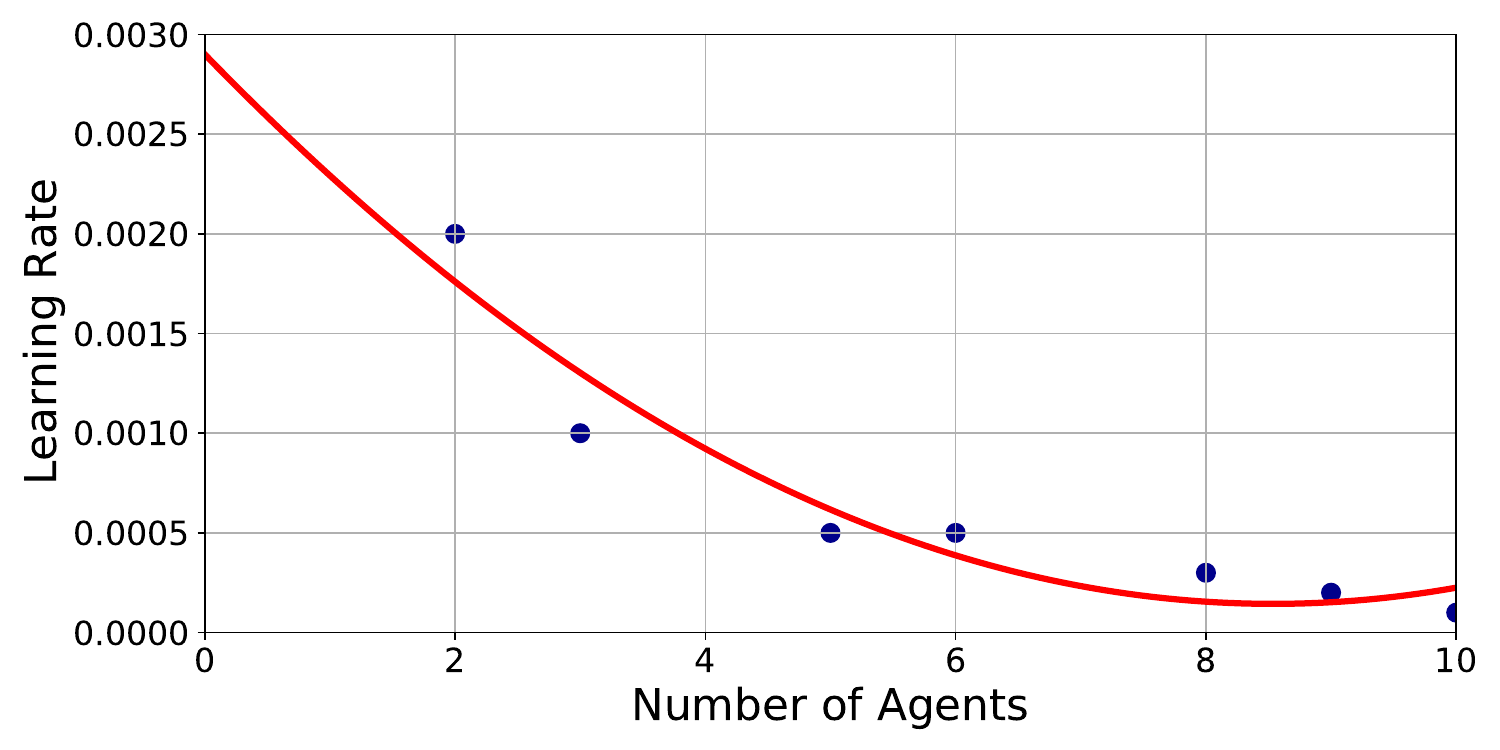}
                    \caption{Relationship between learning rate of $\hat{\alpha}_{i}(s, a_{i})$ and the number of agents (in red curve).}
                \label{fig:correlation_num_agents_to_mixer_lr}
                \end{subfigure}
                \caption{Three ablation studies of SHAQ on SMAC.}
            \label{fig:ablation_study}
            \end{figure*}
            
            We also conduct the ablation study of SHAQ, such as the sample size M for approximating $\hat{\alpha}_{i}(\mathbf{s}, a_{i})$, the empirical selection law on the learning rate of $\hat{\alpha}_{i}(\mathbf{s}, a_{i})$, and the necessity of learning $\hat{\alpha}_{i}(\mathbf{s}, a_{i})$ rather than manual setting. These results show that SHAQ is an easy-to-use algorithm which is potential to be applied to a new scenario with the minimum effort on tuning hyperparameters.
            
            \paragraph{Sample Size M for Approximating $\hat{\alpha}(\mathbf{s}, a_{i})$.} To study the impact of the sample size M on the performance of SHAQ, we conduct an ablation study as Figure \ref{fig:ablation_5m_vs_6m} shows. We observe that the small M is able to achieve fast convergence rate but with high variance, while the large M is with low variance but comparatively slow convergence rate. This observation is consistent with the conclusions from stochastic optimisation \cite{Byrd2012,NIPS2015_effc299a}. As a result, we select the sample size M as 10 in practice, to trade off between convergence rate and variance. 
            
            \paragraph{An Empirical Law on Selecting the Learning Rate of $\hat{\alpha}_{i}(s, a_{i})$.} To provide an empirical law on selecting the learning rate of $\hat{\alpha}_{i}(s, a_{i})$, we statistically fit a curve of the learning rate with respect to the number of controllable agents by the experimental results on SMAC that is shown in Figure \ref{fig:correlation_num_agents_to_mixer_lr}. It is seen that the learning rate of $\hat{\alpha}_{i}(s, a_{i})$ is generally negatively related to the number of agents. In other words, as the number of agents grows the learning rate of $\hat{\alpha}_{i}(s, a_{i})$ is recommended to be smaller. For example, if the number of agents is more than 10, the learning rate of $\hat{\alpha}_{i}(s, a_{i})$ is recommended to be 0.0001 as the guidance from Figure \ref{fig:correlation_num_agents_to_mixer_lr}.
            
            \paragraph{The Necessity of Learning $\hat{\alpha}_{i}(\mathbf{s}, a_{i})$.} Some readers may be concerned about the necessity of learning $\hat{\alpha}_{i}(\mathbf{s}, a_{i})$. To address this concern, we study the necessity of learning $\hat{\alpha}_{i}(\mathbf{s}, a_{i})$ on 5m\_vs\_6m. Since the learned $\hat{\alpha}_{i}(\mathbf{s}, a_{i})$ finally converges to $1.1029$, we grid search the fixed values of $\hat{\alpha}_{i}(\mathbf{s}, a_{i})$ around this number. As Figure \ref{fig:manual_approximate_alpha} shows, $\hat{\alpha}_{i}(\mathbf{s}, a_{i})$ with manually preset fixed value cannot work as well as the learned $\hat{\alpha}_{i}(\mathbf{s}, a_{i})$. Therefore, we validate the necessity of learning $\hat{\alpha}_{i}(\mathbf{s}, a_{i})$ here.
            
        \subsection{Understanding Markov Shapley Value}
        \label{subsec:interpretability_of_shaq}
            \begin{figure*}[ht!]
            \centering
                \begin{subfigure}[b]{0.42\textwidth}
                	\centering
            	    \includegraphics[width=\textwidth]{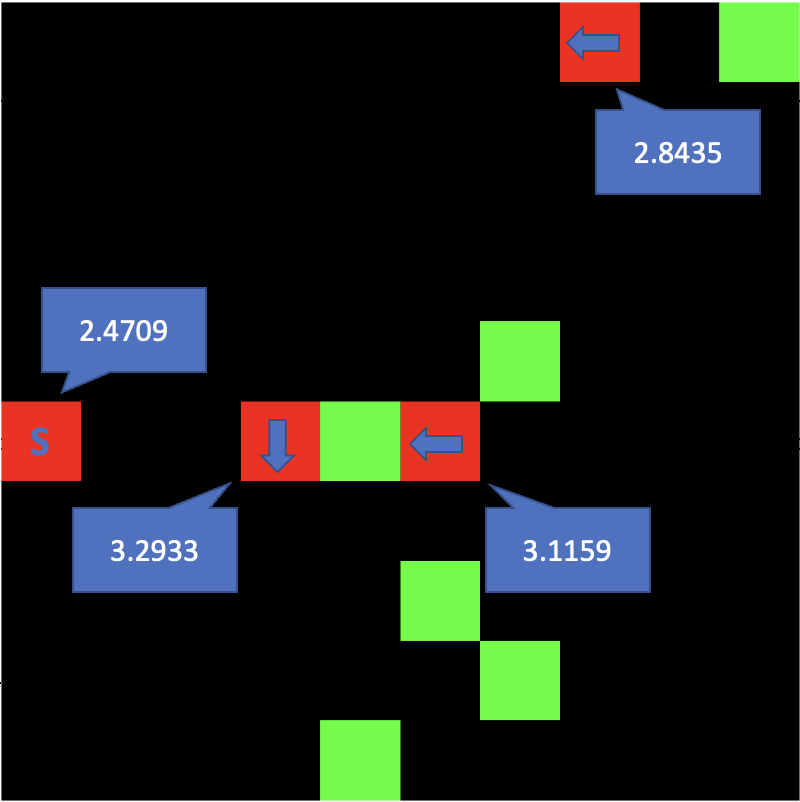}
            	    \caption{SHAQ: $\epsilon$-greedy.}
            	\label{fig:shaq_epsilon_greedy_pp}
                \end{subfigure}
                ~
                \begin{subfigure}[b]{0.42\textwidth}
                    \centering                \includegraphics[width=\textwidth]{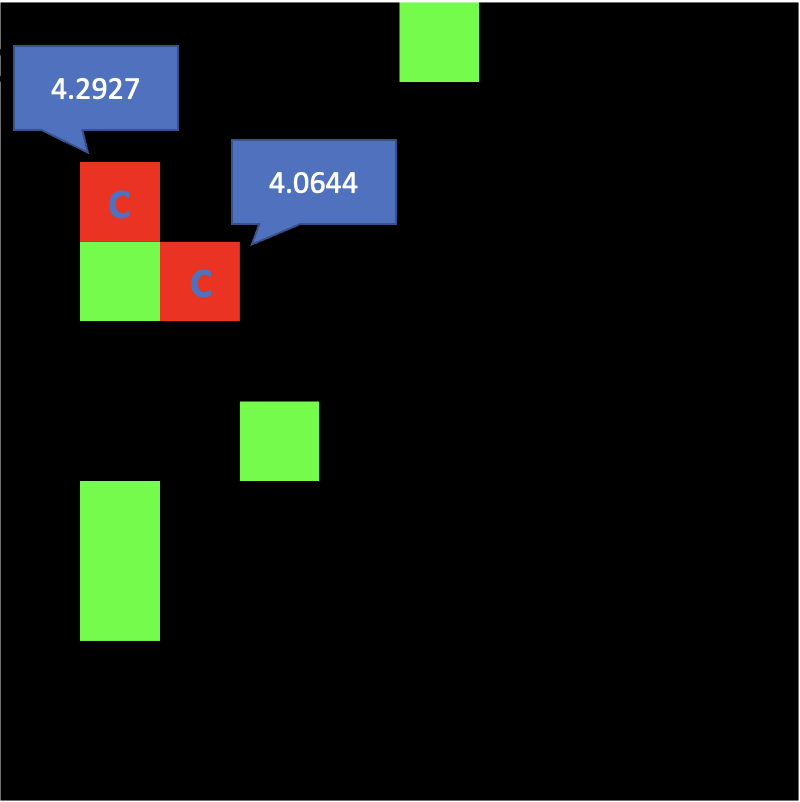}
                    \caption{SHAQ: greedy.}
                \label{fig:shaq_greedy_pp}
                \end{subfigure}
                \caption{Visualisation of the evaluation of SHAQ on Predator-Prey: each red square is a controllable predator, whereas each green square indicates a prey. Each predator's factorised Q-value is reported in a bubble in blue and the symbols within the squares indicate the action of each predator (i.e., arrows imply the movement direction, ``S'' implies staying and ``C'' implies capturing a prey that is valid only when the agent is around a prey). The epsilon in Figure (a) is chosen as 0.8, so it is highly likely that a random action is executed (i.e. $\epsilon$-greedy policy). While the actions performed in Figure (b) are optimal (i.e. greedy policy).}
            \label{fig:study_on_factorised_q_values_pp}
            \end{figure*}
            
            \paragraph{Scenario 1.} To verify that SHAQ possesses the interpretability, we show its credit assignment on Predator-Prey. As \cite{bohmer2020deep} illustrated, if both predators are around and capture a prey, they will be rewarded as 10. Otherwise, if any single predator attempts to capture a prey, they will be punished by p (that is set to -1 in this demonstration). As we see from Figure \ref{fig:shaq_greedy_pp}, all predators are around and capture a prey, so both of them perform the optimal actions and deserve almost the equal optimal credit assignment as $4.2927$ and $4.0644$, which verifies our theoretical claim. From Figure \ref{fig:shaq_epsilon_greedy_pp}, it can be seen that two predators are far away from preys, so they receive low credits as $2.4709$ and $2.8435$. On the other hand, the other two predators are around a prey, but they do not perform the optimal action ``capture'', so they receive less credits than the two predators in Figure \ref{fig:shaq_greedy_pp}. Nevertheless, they are around a prey, so they perform better than those predators that are far away from preys and receive comparatively greater credits as $3.2933$ and $3.1159$. The coherent credit assignment in both Figure \ref{fig:shaq_epsilon_greedy_pp} and \ref{fig:shaq_greedy_pp} implies that the assigned credits reflect the contributions (i.e., each agent receives the credit that is consistent with its decision).
            \begin{figure*}[ht!]
            \centering
                \begin{subfigure}[b]{0.48\textwidth}
                	\centering
            	    \includegraphics[width=\textwidth]{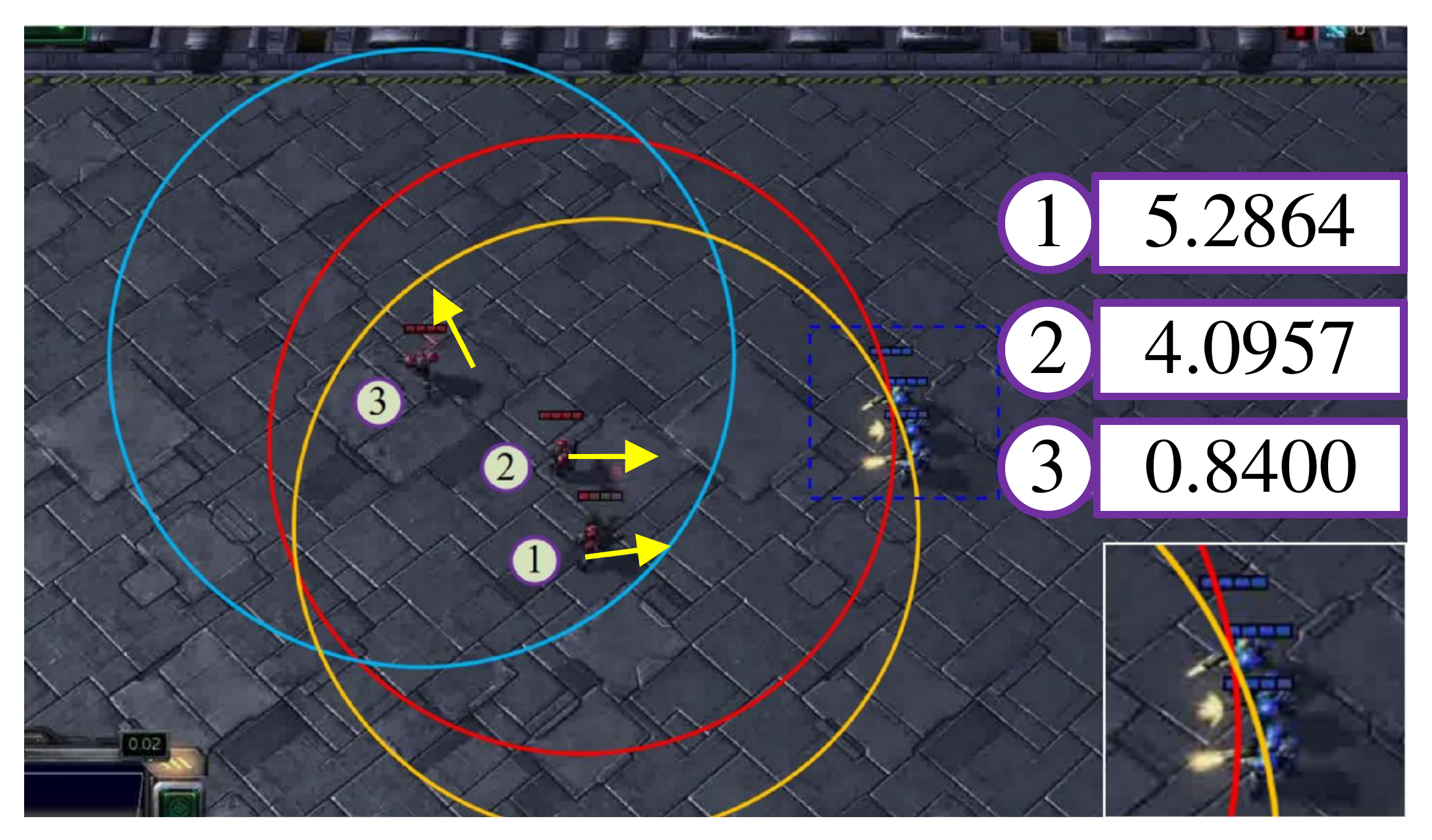}
            	    \caption{SHAQ: $\epsilon$-greedy.}
            	\label{fig:shaq_epsilon_greedy}
                \end{subfigure}
                ~
                \begin{subfigure}[b]{0.48\textwidth}
                    \centering                \includegraphics[width=\textwidth]{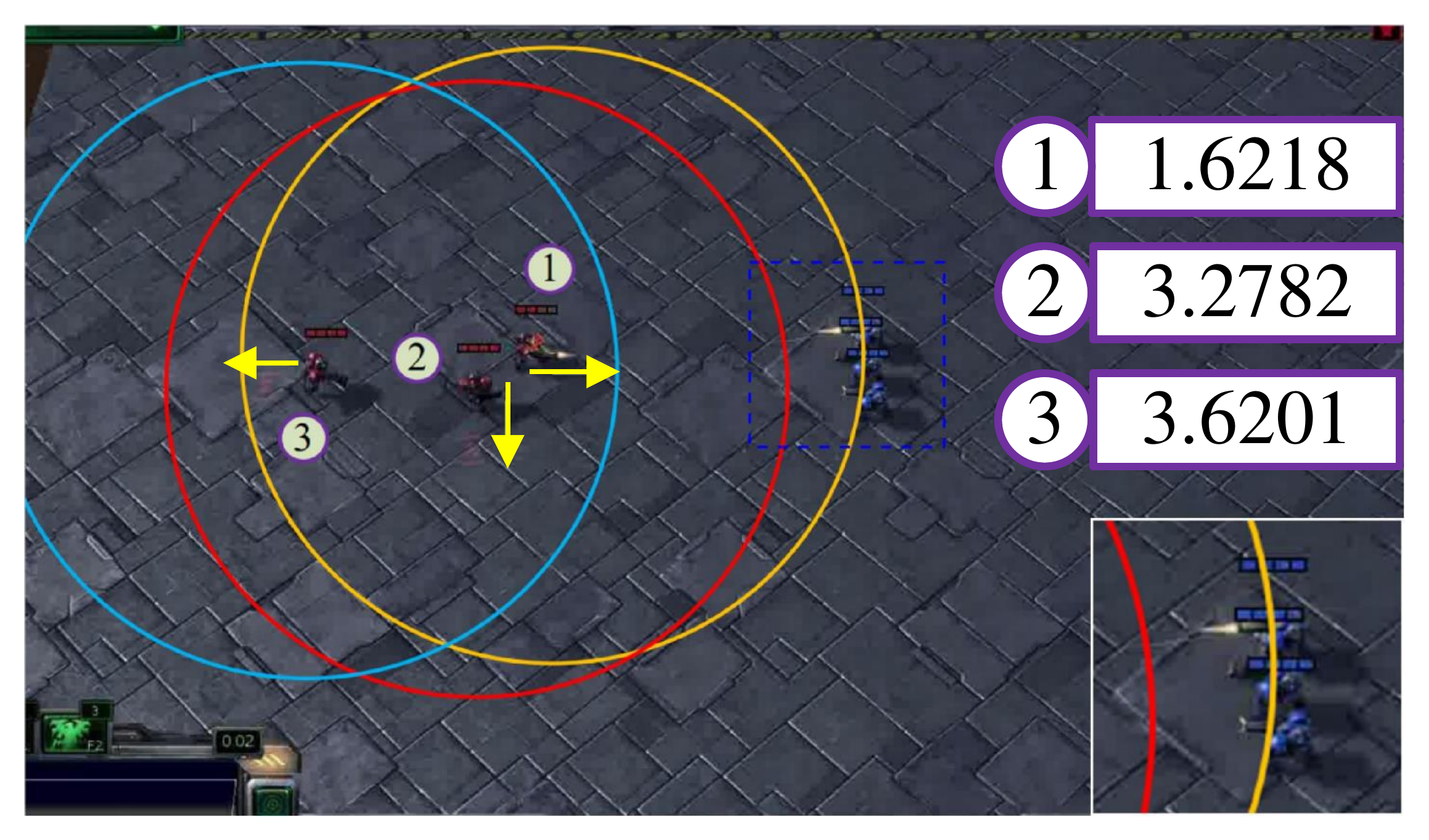}
                    \caption{VDN: $\epsilon$-greedy.}
                \label{fig:vdn_epsilon_greedy}
                \end{subfigure}
                ~
                \begin{subfigure}[b]{0.48\textwidth}
                    \centering                \includegraphics[width=\textwidth]{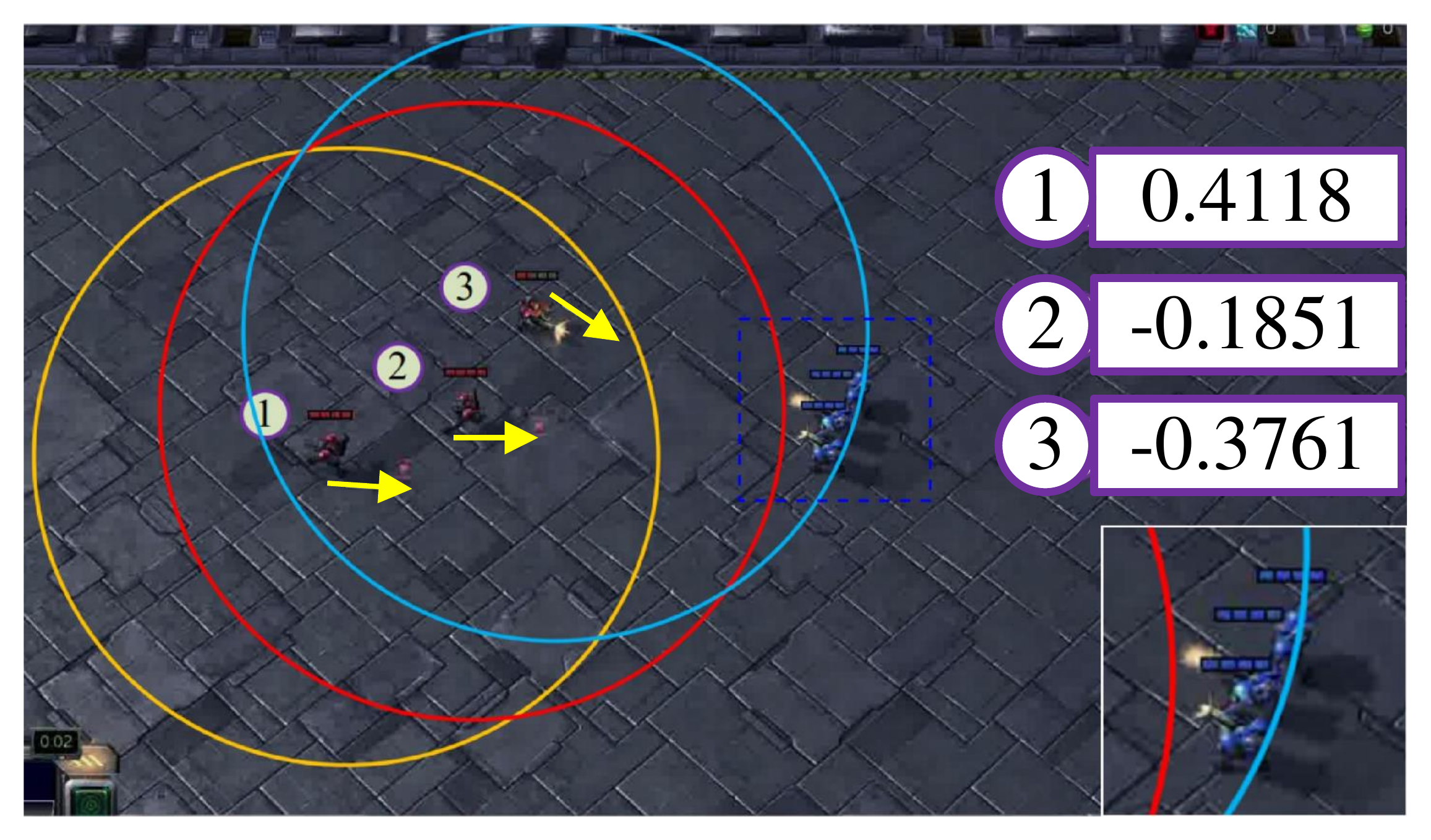}
                    \caption{QMIX: $\epsilon$-greedy.}
                \label{fig:qmix_epsilon_greedy}
                \end{subfigure}
                ~
                \begin{subfigure}[b]{0.48\textwidth}
            	    \centering        	    \includegraphics[width=\textwidth]{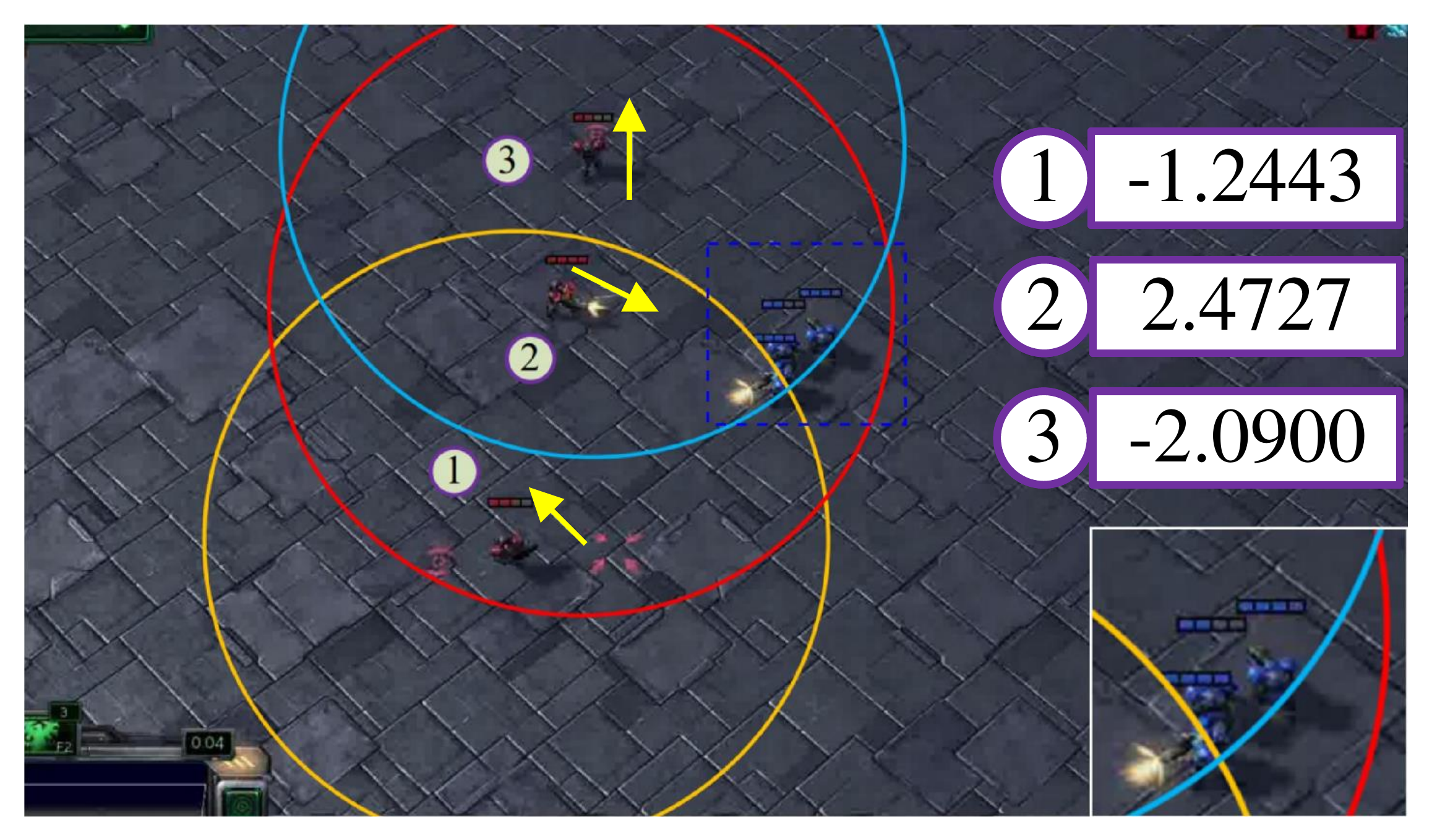}
            	    \caption{QPLEX: $\epsilon$-greedy.}
            	\label{fig:qplex_epsilon_greedy}
                \end{subfigure}
                
                \begin{subfigure}[b]{0.48\textwidth}
            	    \centering        	    \includegraphics[width=\textwidth]{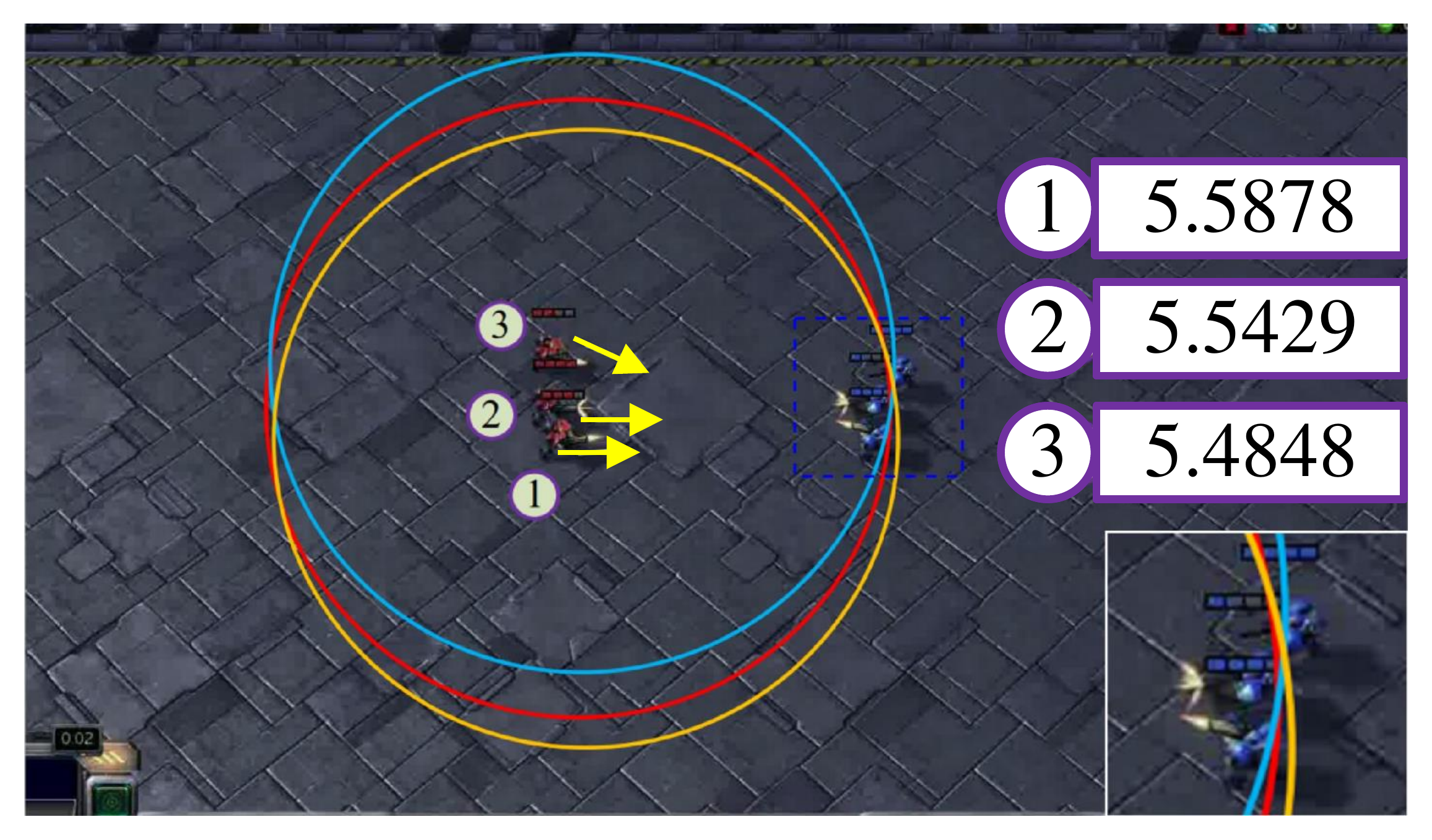}
            	    \caption{SHAQ: greedy.}
            	\label{fig:shaq_greedy}
                \end{subfigure}
                ~
                \begin{subfigure}[b]{0.48\textwidth}
                	\centering        	    \includegraphics[width=\textwidth]{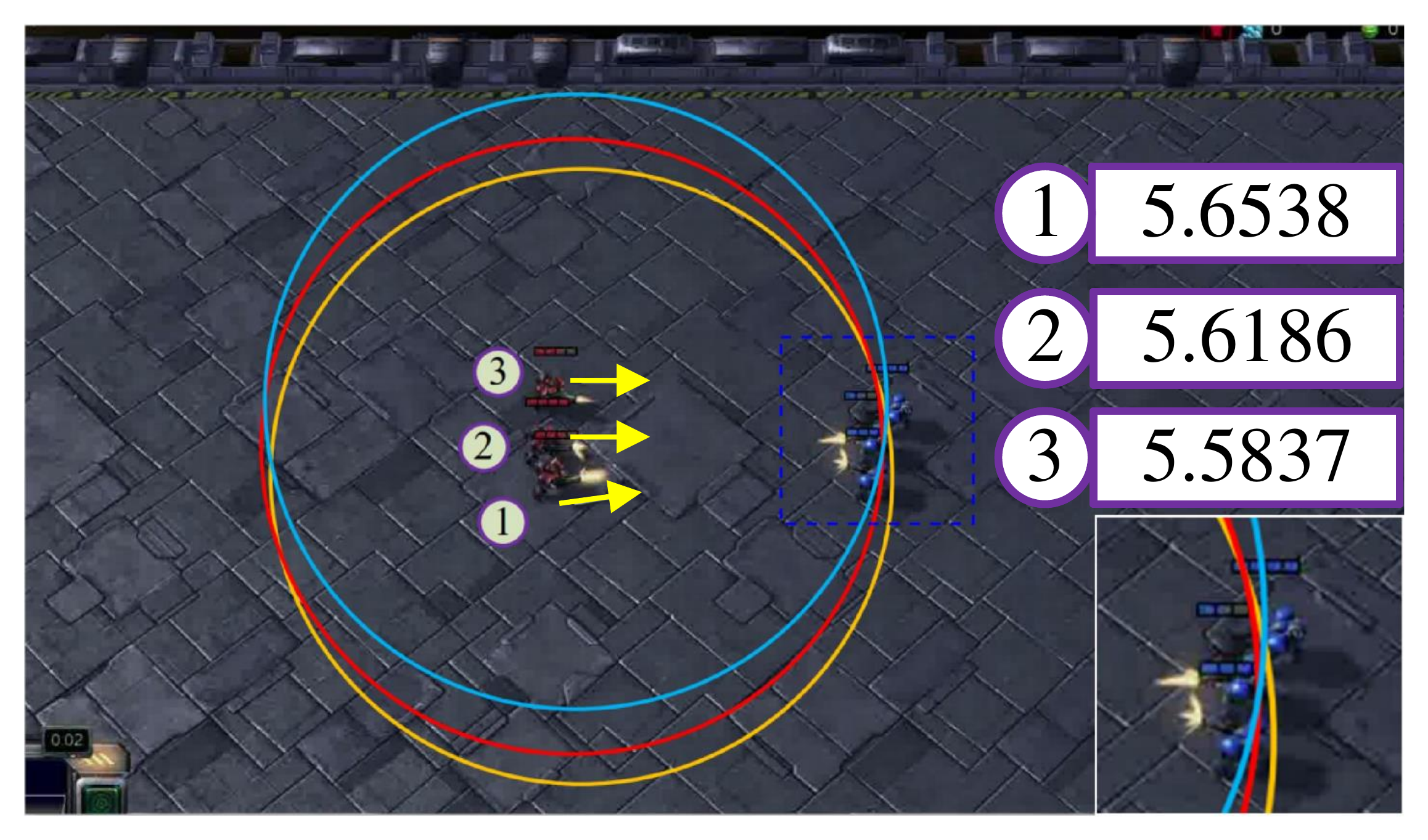}
            	    \caption{VDN: greedy.}
            	\label{fig:vdn_greedy}
                \end{subfigure}
                ~
                \begin{subfigure}[b]{0.48\textwidth}
            	    \centering        	    \includegraphics[width=\textwidth]{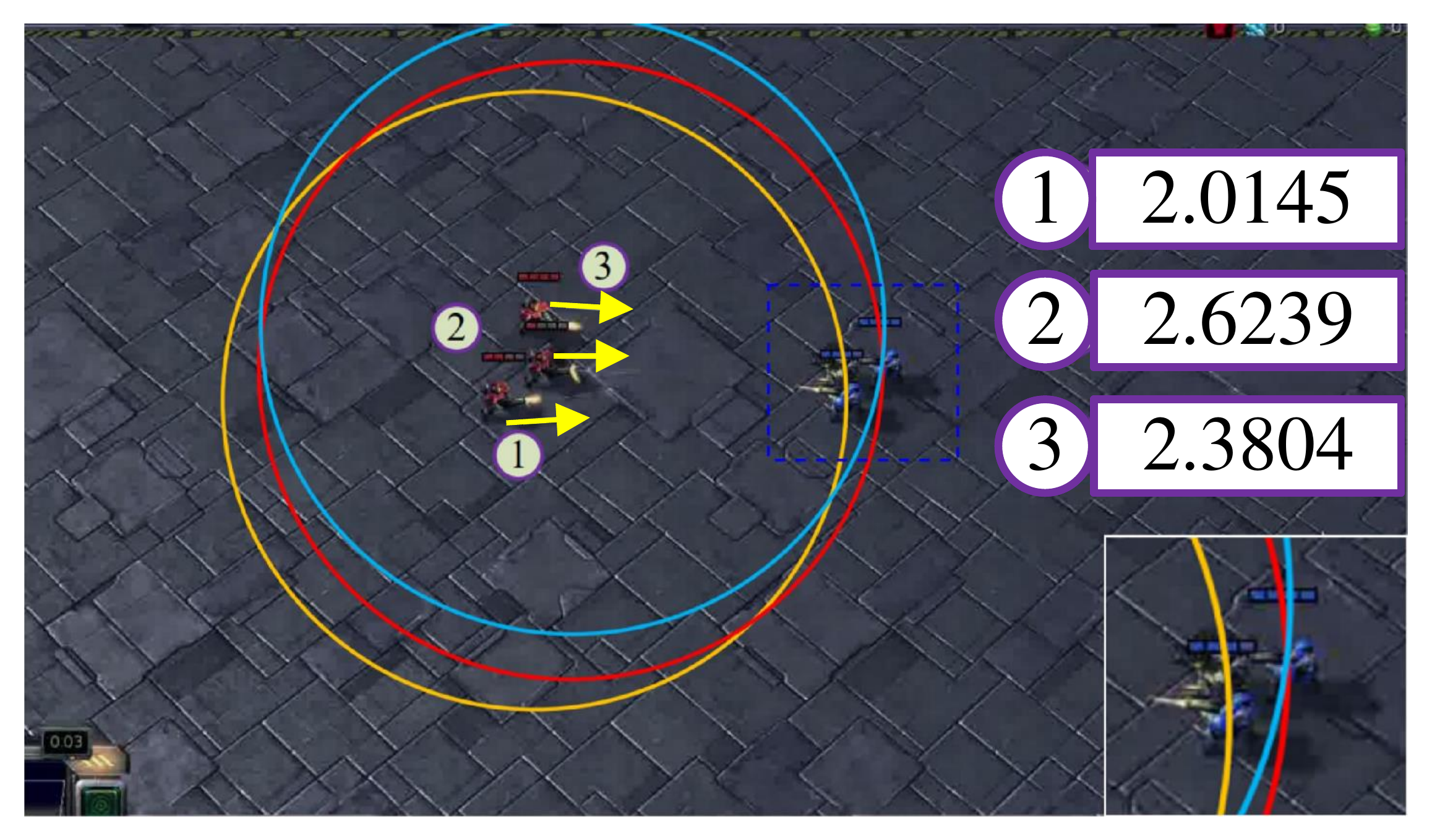}
            	    \caption{QMIX: greedy.}
            	\label{fig:qmix_greedy}
                \end{subfigure}
                ~
                \begin{subfigure}[b]{0.48\textwidth}
            	    \centering        	    \includegraphics[width=\textwidth]{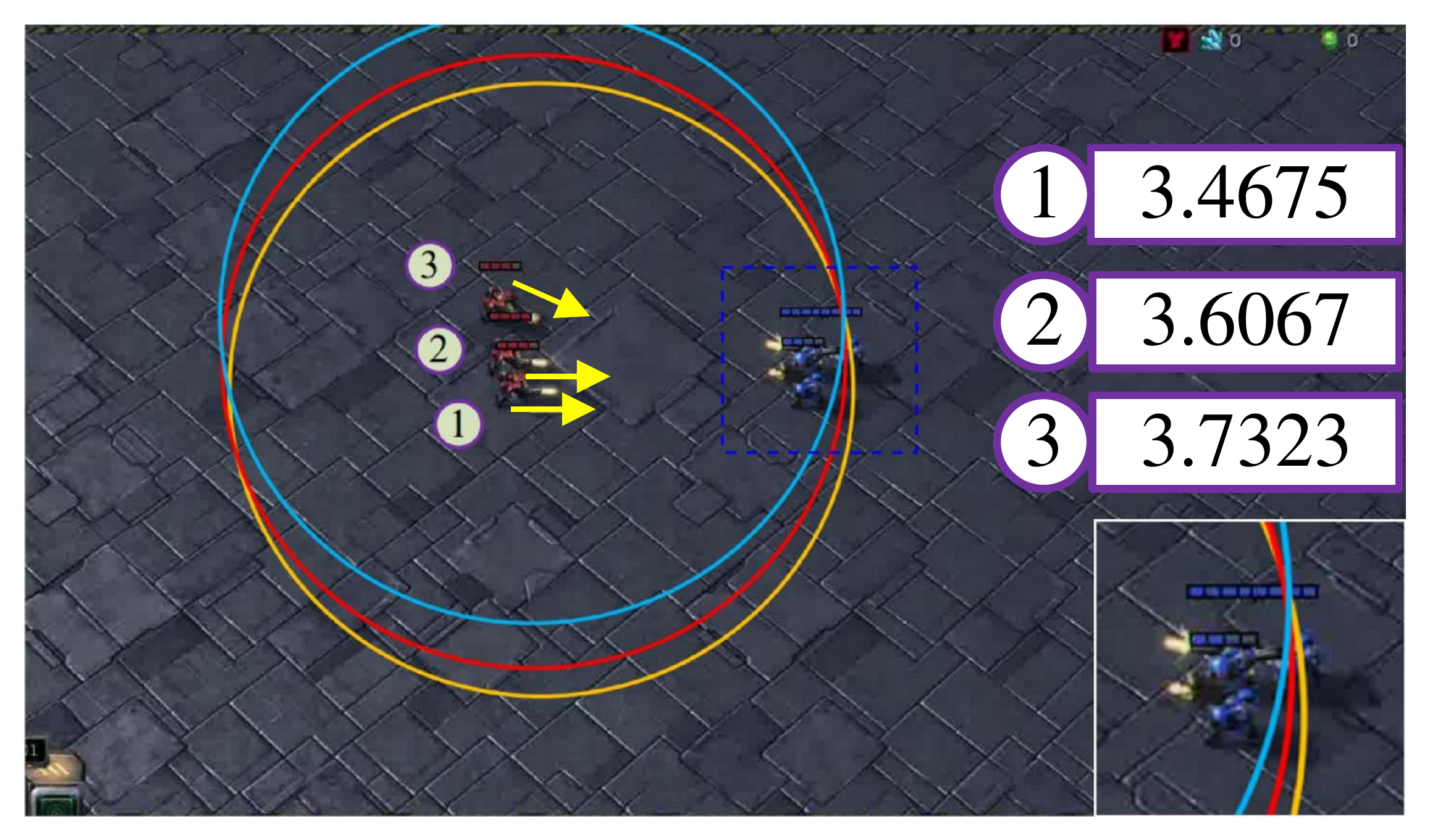}
            	    \caption{QPLEX: greedy.}
            	\label{fig:qplex_greedy}
                \end{subfigure}
                \caption{Visualisation of the test for SHAQ and baselines on 3m in SMAC: each colored circle is the centered attacking range of a controllable agent (in red), and each agent's factorised Q-value is reported on the right. We mark the direction that each agent face by an arrow for clarity.}
            \label{fig:study_on_factorised_q_values}
            \end{figure*}
            
            \paragraph{Scenario 2.} To further show the interpretability of SHAQ, we conduct a test on 3m (i.e., a simple task in SMAC), demonstrating the learned MSQs of both $\epsilon$-greedy policy (for obtaining the mixed optimal and sub-optimal actions) and greedy policy (for obtaining the pure optimal actions). As seen from Figure \ref{fig:shaq_epsilon_greedy}, agent 3 faces the direction opposite to enemies, meanwhile, the enemies are out of its attacking range. It can be understood as that agent 3 does not contribute to the team and thus it is almost a dummy agent. Its MSQ is therefore 0.84 (around 0), which correctly describes the manner of a dummy agent (verifying (i) in Proposition \ref{prop:shapley_value_properties}). In contrast, agent 1 and agent 2 are attacking enemies, while agent 1 suffers from more attacks (with lower health) than agent 2. As a result, agent 1 contributes more than agent 2 and therefore its MSQ is greater, which verifies the property of reflecting the contribution (verifying (iii) in Proposition \ref{prop:shapley_value_properties}). On the other hand, we can see from Figure \ref{fig:shaq_greedy} that with the optimal policies all agents receive almost identical MSQs (verifying the theoretical results in Section \ref{subsec:shapley_q_learning}). 
            
            The above results well verify the theoretical analysis that we deliver in Chapter \ref{cha:theory_and_method}. To verify that the MSQs learned by SHAQ are non-trivial, we also show the resulting Q-values of VDN, QMIX and QPLEX. It is surprising that the Q-values of these baselines are also almost identical among agents for the optimal actions. Since VDN is a subclass of SHAQ and possesses the same form of loss function for the optimal actions, it is reasonable that it obtains the similar results to SHAQ. The exploration of the results of QMIX and QPLEX deserves to be conducted in the future work. As for the sub-optimal actions, VDN does not possess an explicit interpretation as SHAQ due to its incorrect definition of $\delta_{i}(\mathbf{s}, a_{i})=1$ over the sub-optimal actions (verifying the statement in Section \ref{subsec:comparison_with_other_learning_algorithms}). Similarly, QMIX and QPLEX cannot show explicit interpretation of sub-optimal actions either.
            \begin{figure*}[ht!]
            \centering
                \begin{subfigure}[b]{0.48\textwidth}
                	\centering
            	    \includegraphics[width=\textwidth]{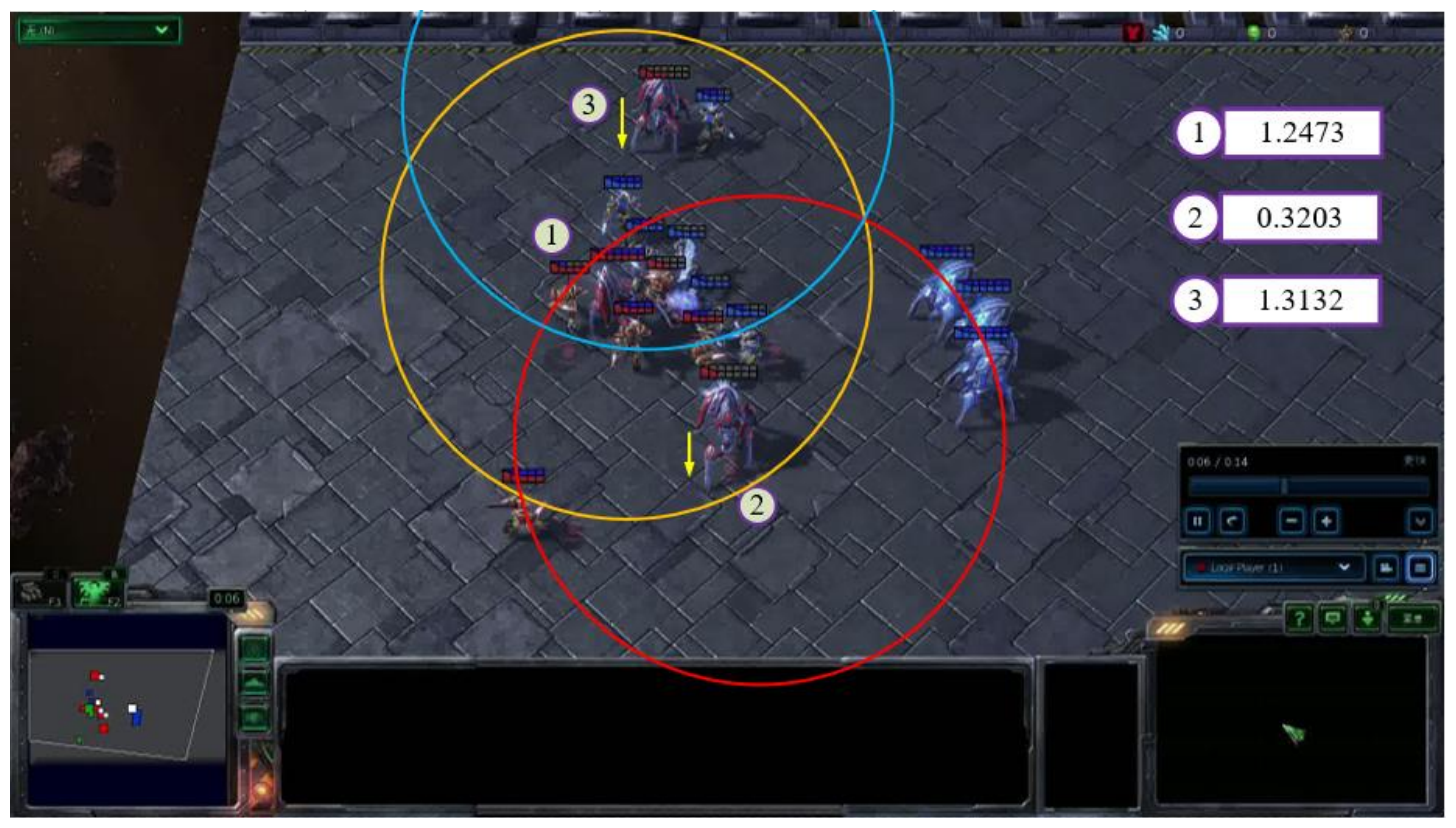}
            	    \caption{SHAQ: $\epsilon$-greedy.}
            	\label{fig:shaq_epsilon_greedy_appendix}
                \end{subfigure}
                ~
                \begin{subfigure}[b]{0.48\textwidth}
                    \centering                \includegraphics[width=\textwidth]{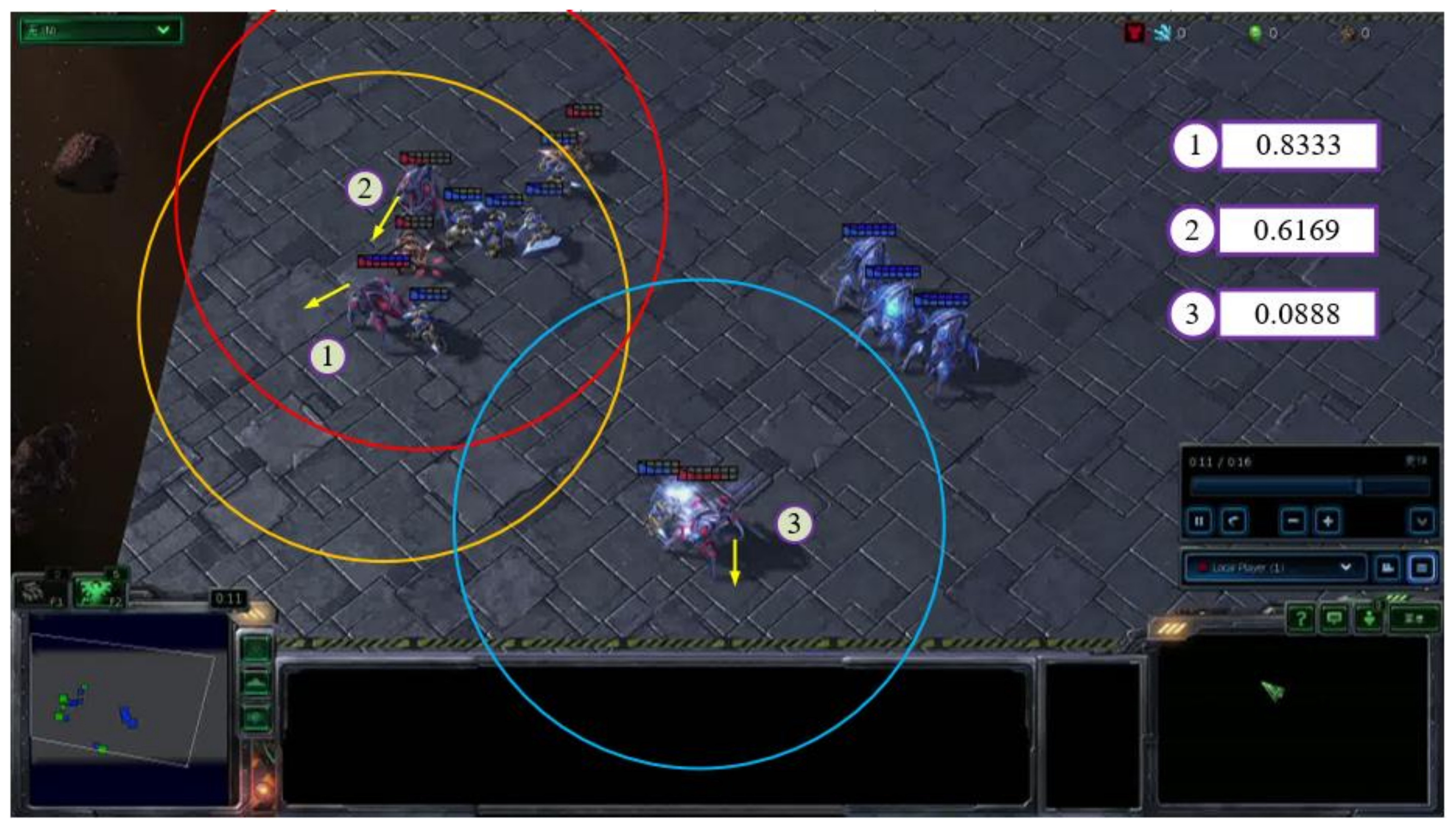}
                    \caption{VDN: $\epsilon$-greedy.}
                \label{fig:vdn_epsilon_greedy_appendix}
                \end{subfigure}
                ~
                \begin{subfigure}[b]{0.48\textwidth}
                    \centering                \includegraphics[width=\textwidth]{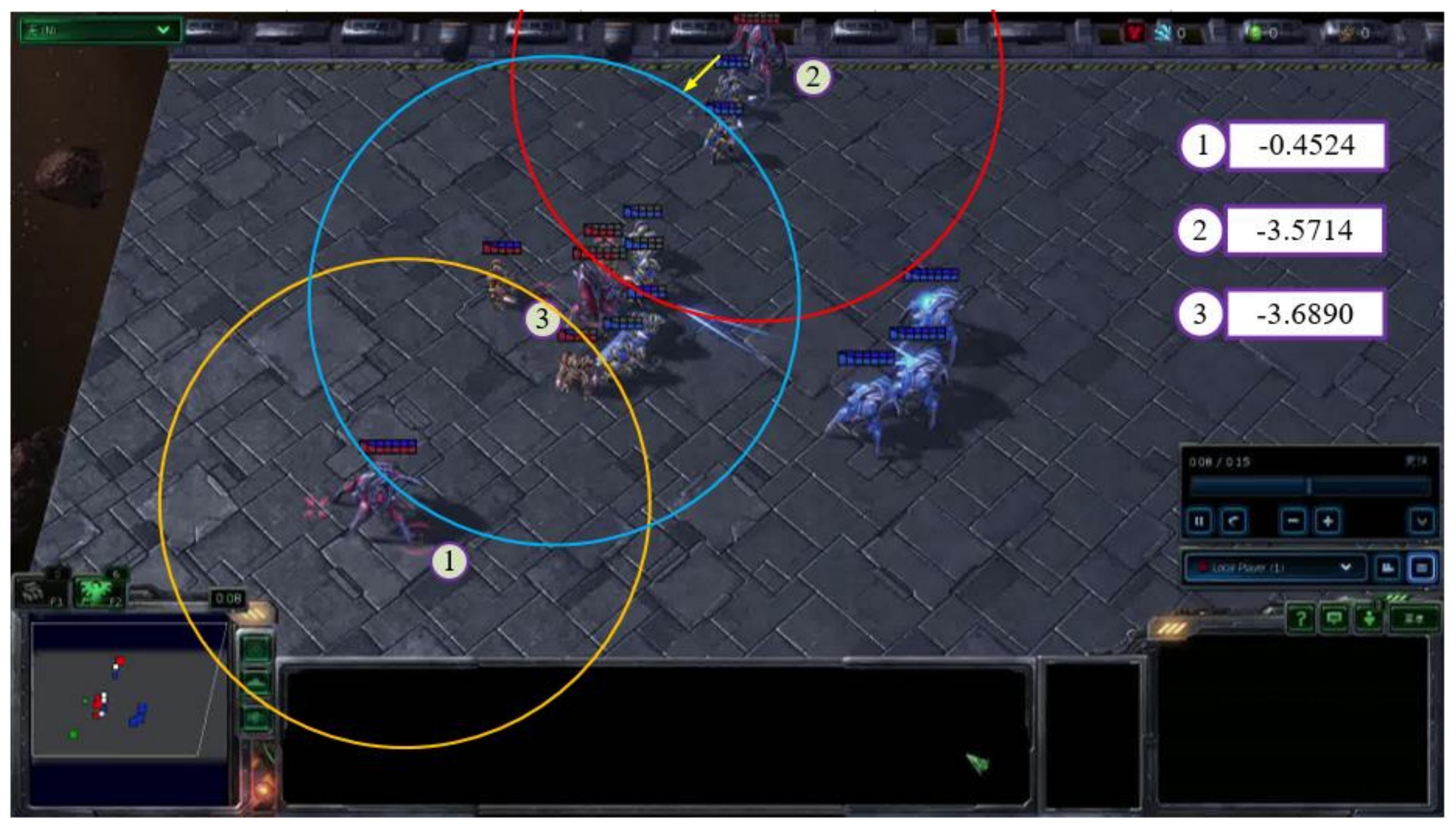}
                    \caption{QMIX: $\epsilon$-greedy.}
                \label{fig:qmix_epsilon_greedy_appendix}
                \end{subfigure}
                ~
                \begin{subfigure}[b]{0.48\textwidth}
            	    \centering        	    \includegraphics[width=\textwidth]{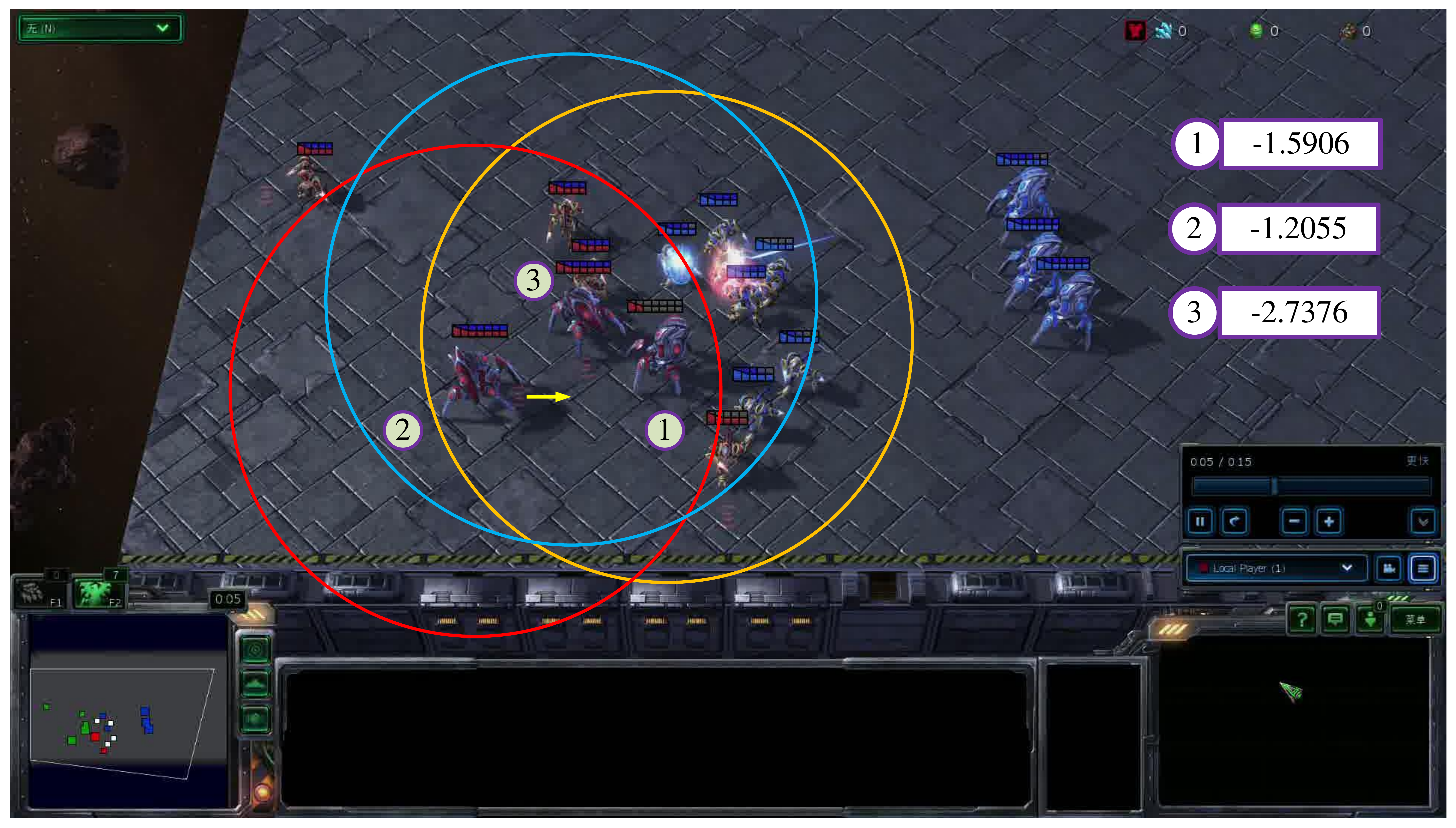}
            	    \caption{QPLEX: $\epsilon$-greedy.}
            	\label{fig:qplex_epsilon_greedy_appendix}
                \end{subfigure}
                
                \begin{subfigure}[b]{0.48\textwidth}
            	    \centering        	    \includegraphics[width=\textwidth]{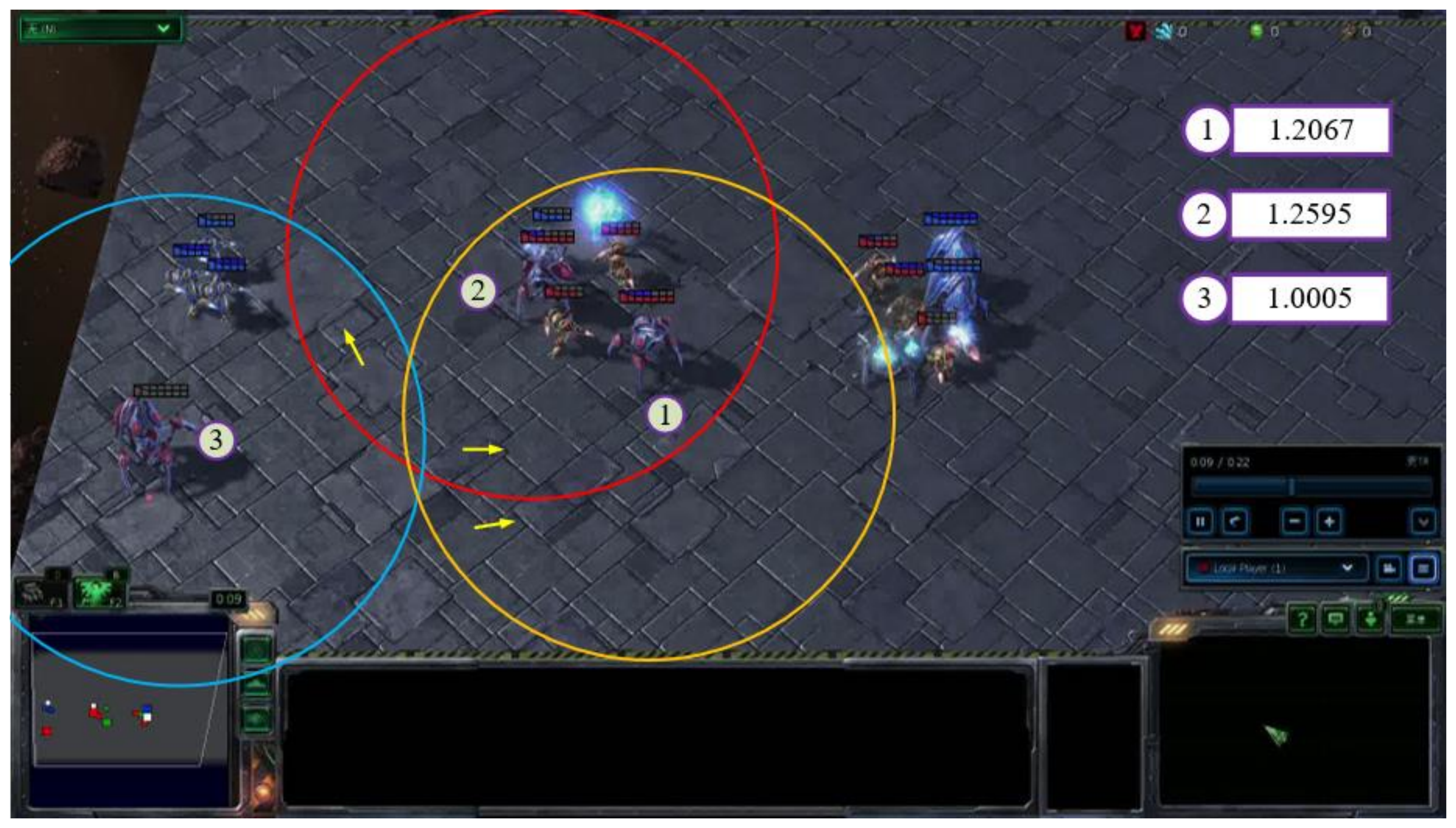}
            	    \caption{SHAQ: greedy.}
            	\label{fig:shaq_greedy_appendix}
                \end{subfigure}
                ~
                \begin{subfigure}[b]{0.48\textwidth}
                	\centering        	    \includegraphics[width=\textwidth]{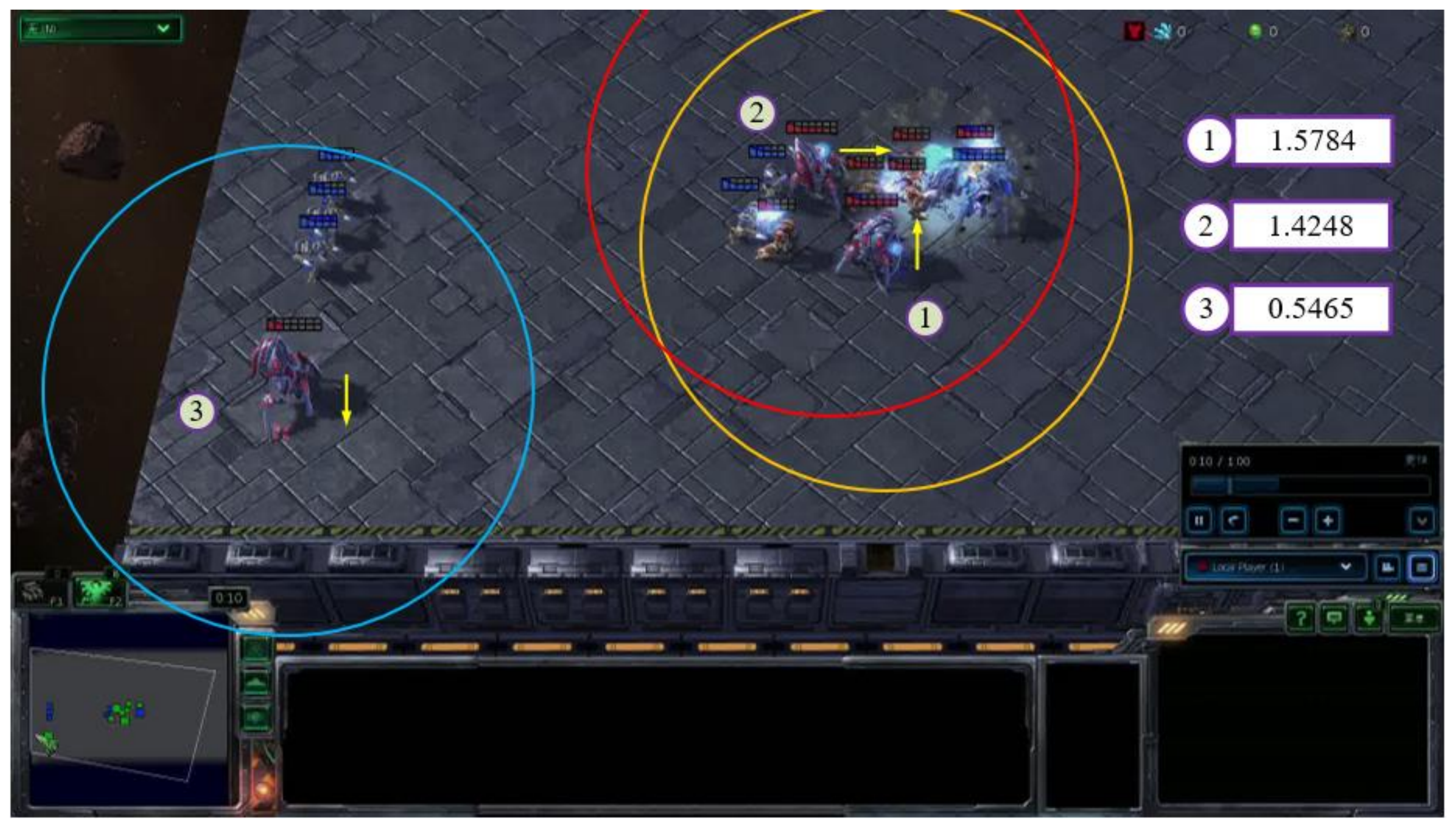}
            	    \caption{VDN: greedy.}
            	\label{fig:vdn_greedy_appendix}
                \end{subfigure}
                ~
                \begin{subfigure}[b]{0.48\textwidth}
            	    \centering        	    \includegraphics[width=\textwidth]{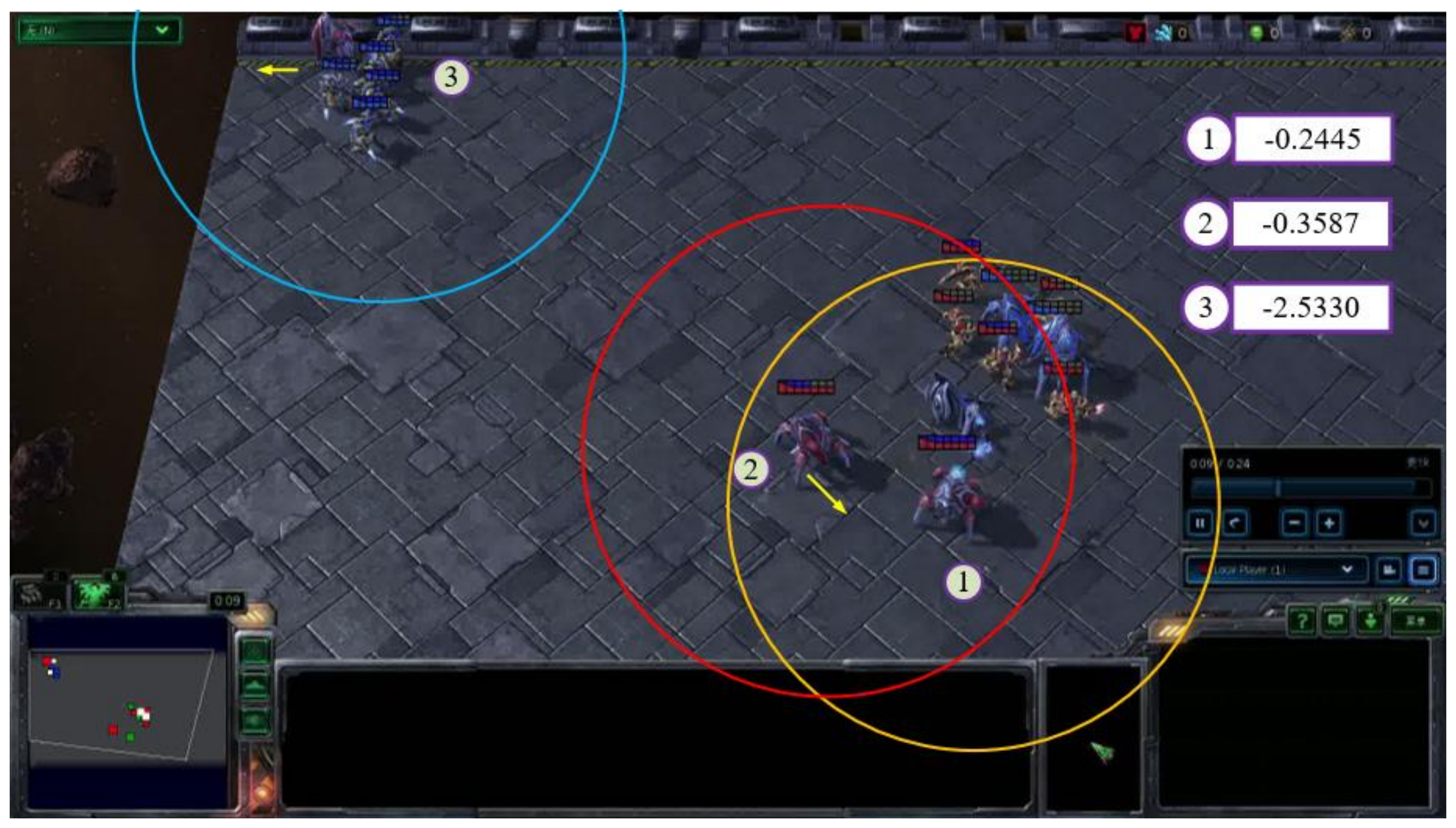}
            	    \caption{QMIX: greedy.}
            	\label{fig:qmix_greedy_appendix}
                \end{subfigure}
                ~
                \begin{subfigure}[b]{0.48\textwidth}
            	    \centering        	    \includegraphics[width=\textwidth]{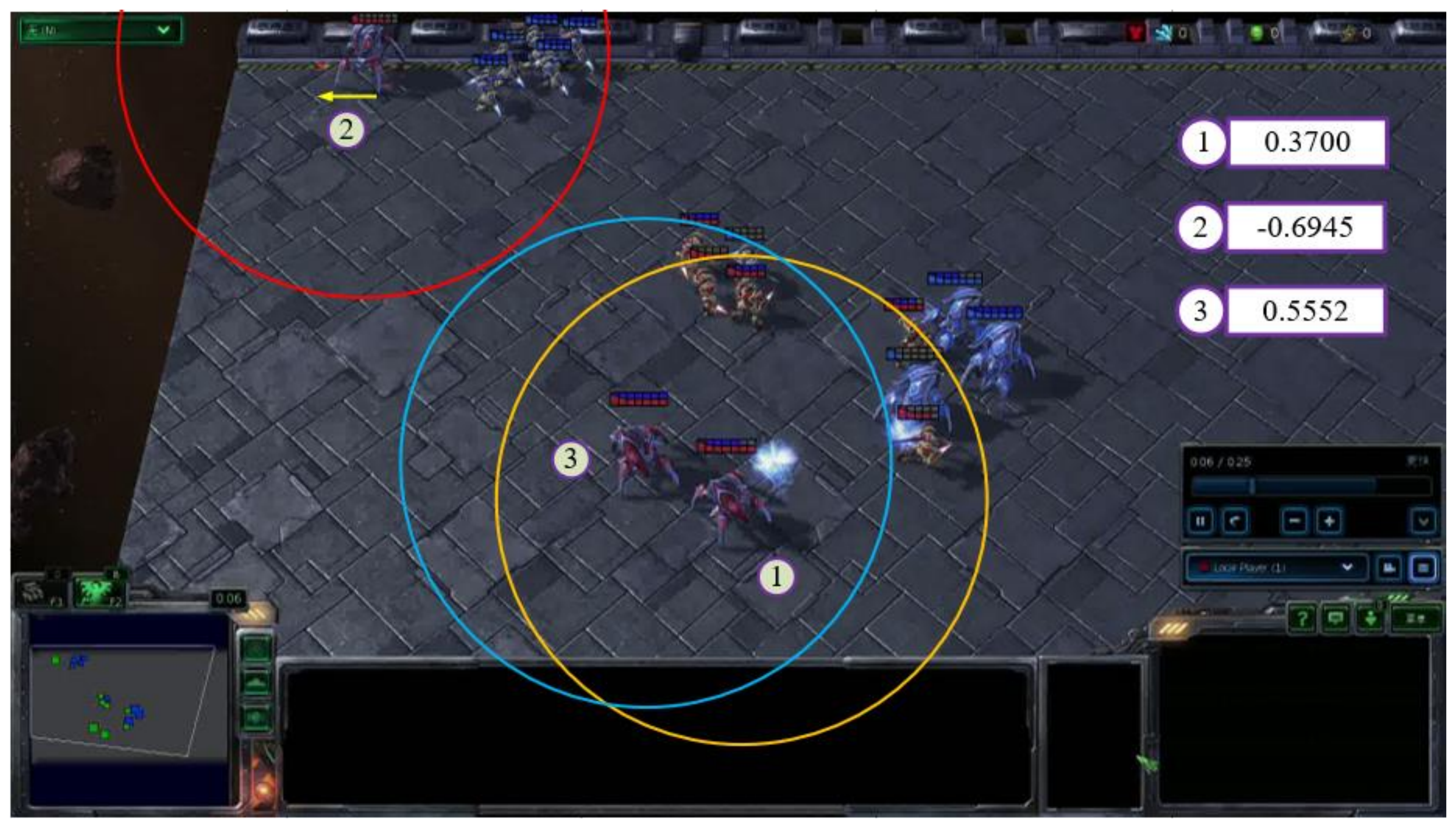}
            	    \caption{QPLEX: greedy.}
            	\label{fig:qplex_greedy_appendix}
                \end{subfigure}
                \caption{Visualisation of the evaluation for SHAQ and other baselines on 3s5z\_vs\_3s6z in SMAC: each colored circle is the centered attacking range of a controllable agent (in red), and each agent's factorised Q-value is reported on the right. We mark the direction that each moving agent face by an arrow.}
            \label{fig:study_on_factorised_q_values_appendix}
            \end{figure*}
            
            \paragraph{Scenario 3.}  To verify our theoretical results more firmly, we show the Q-values on a more complicated scenario in SMAC, i.e. 3s5z\_vs\_3s6z during test in Figure \ref{fig:study_on_factorised_q_values_appendix}. First, we take a look into the optimal actions. SHAQ can still demonstrate the equal credit assignment as we claimed before. Unfortunately, VDN does not explicitly show equal credit assignment. The possible reason is that part of parameters of Q-value are shared between optimal actions and sub-optimal actions. Therefore, the parametric effect of the mistakes committed on the sub-optimal actions to the optimal actions by VDN during learning may be exaggerated when the number of agents increases. About QMIX and QPLEX, the Q-values of optimal actions are difficult to be interpreted in this complicated scenario. For both of the two algorithms, the agent who is responsible for kiting\footnote{\url{https://en.wikipedia.org/wiki/Glossary_of_video_game_terms}.} (i.e., agent 3 for QMIX and agent 2 for QPLEX) receives the lowest credit, however, it is an important role to the team in a combat tactic. 
            
            Next, we focus on the sub-optimal actions. As for SHAQ, agent 1 and agent 3 are participating into the battle, so deserving almost the equal credit assignment. However, agent 2 drops teammates and escapes from the center of the battle field, so it contributes almost nothing to the team. As a result, it can be seen as a dummy agent and thus obtains the credit almost equal to 0. This is again consistent with our theoretical analysis. About VDN, it coincidentally assigns the credit almost equal to 0 to the dummy agent (i.e., agent 3) in this scenario. Nevertheless, the low credit assignments to the other 2 agents who participate in the battle is difficult to be interpreted. About QMIX, agent 2 and agent 3 who participate in the battle receive the lowest credits, while agent 1 who escapes from the battle field receives the highest credit. For QPLEX, the agents' behaviours are totally difficult to be interpreted.


\chapter{Application to Active Voltage Control in Power Distribution Networks}
\label{chap:application_to_distributed_voltage_control}
    In this chapter, we first model the active voltage control problem in power distribution networks as a Dec-POMDP to fit it into multi-agent reinforcement learning. Then, we specifically discuss how to design a global reward function that encodes the goal of this problem (i.e., confining the voltage within a safety range with the minimum power loss). Next, we evaluate the performance of SQDDPG and SMFPPO on three scenarios with different network topologies, load profiles and power profiles. Moreover, we compare the performance of SQDDPG and SMFPPO with the traditional control methods, and show its potential to solve this real-world challenge. Finally, we provide an illustrative example to demonstrate the possible physical implication in power distribution networks for the Markov Shapley value.
    
\section{Problem Formulation}
\label{sec:problem_formulation}
    For ease of operations, a large-scale power network is divided into multiple regions and there are several PVs installed into each region that is managed by the responsible distribution network operator. Each PV is with an inverter that generates active power, part of which is consumed by the local loads while the surplus can be injected into the grid. It is well known that PV power injection may induce voltage rise for the grid which can be solved by the reactive power injection from the PV inverters themselves. The reactive power injection creates voltage difference on the inductive distribution network and thus help restore the network voltage towards the nominal voltage, denoted as $v_{\scriptscriptstyle \text{ref}}$. The reactive power and voltage regulation of distributed PVs has a local-global characteristic. On one hand, each PV inverter can measure its local voltage and regulate its reactive power accordingly (e.g. Q-V droop control). On the other hand, the reactive injection at one node may affect the voltage on other nodes, meaning that the voltage regulation should be globally optimised. This local-global characteristic makes the voltage regulation problem a good candidate for multi-agent control \cite{burger2019restructuring1,burger2019restructuring2}. In our problem, each PV inverter is controlled in a distributed manner, i.e., each PV is considered an agent. All agents within a region share the observation of this region.\footnote{Sharing observation in this problem is reasonable, since only the sensor measurements (e.g. voltage, active power, etc.) are shared, which are not directly related to the commercial profits \cite{burger2019restructuring1,burger2019restructuring2}. The observation of each PV is collected by the distribution network owner and then the full information within the region is sent to each agent.} Since each agent can only observe partial information of the whole grid and maintaining the safety of the power network is a common goal among agents, it is reasonable to model the problem as a Dec-POMDP \cite{oliehoek2016concise} that can be mathematically described as a 10-tuple such that $\langle \mathcal{I}, \mathcal{S}, \mathcal{A}, \mathcal{R}, \mathcal{O}, T, R, \Omega, \rho, \gamma \rangle$, where $\rho$ is the probability distribution for drawing the initial state and $\gamma$ is the discount factor. 
    
    \paragraph{Agent Set.} There is a set of agents controlling a set of PV inverters denoted as $\mathcal{I}$. Each agent is located at some node in $\mathcal{G}$ (i.e. a graph representing the power network defined as before). We define a function $\mathit{g}: \mathcal{I} \rightarrow V$ to indicate the node where an agent is located.
    
    \paragraph{Region Set.} The whole power network is separated into $\mathit{M}$ regions, whose union is denoted as $\mathcal{R} = \{ \mathcal{R}_{k} \subset V \ | \ k < M, k \in \mathbb{N} \}$, where $\bigcup_{\mathcal{R}_{k} \in \mathcal{R}} \mathcal{R}_{k} \subseteq V$ and $\mathcal{R}_{k_{1}} \cap \mathcal{R}_{k_{2}} = \emptyset$ if $k_{1} \neq k_{2}$. We define a function $\mathit{f}: V \rightarrow \mathcal{R}$ that maps a node to the region where it is involved.
    
    \paragraph{State and Observation Set.} The state set is defined as $\mathcal{S} = \mathcal{L} \times \mathcal{P} \times \mathcal{Q} \times \mathcal{V}$, where $\mathcal{L} = \{ (\textbf{p}^{\scriptscriptstyle L}, \textbf{q}^{\scriptscriptstyle L}) : \textbf{p}^{\scriptscriptstyle L}, \textbf{q}^{\scriptscriptstyle L} \in (0, \infty)^{\scriptscriptstyle |V|} \}$ is a set of (active and reactive) powers of loads; $\mathcal{P} = \{ \textbf{p}^{\scriptscriptstyle PV} : \textbf{p}^{\scriptscriptstyle PV} \in (0, \infty)^{\scriptscriptstyle |\mathcal{I}|} \}$ is a set of active powers generated by PVs; $\mathcal{Q} = \{ \textbf{q}^{\scriptscriptstyle PV} : \textbf{q}^{\scriptscriptstyle PV} \in (0, \infty)^{\scriptscriptstyle |\mathcal{I}|} \}$ is a set of reactive powers generated by PV inverters at the preceding step; $\mathcal{V} = \{ (\textbf{v}, \mathbf{\theta}) : \textbf{v} \in (0, \infty)^{\scriptscriptstyle |V|}, \mathbf{\theta} \in [-\pi, \pi]^{\scriptscriptstyle |V|} \}$ is a set of voltage wherein $\textbf{v}$ is a vector of voltage magnitudes and $\mathbf{\theta}$ is a vector of voltage phases measured in radius. $v_{i}$, $p_{i}^{\scriptscriptstyle L}$, $q_{i}^{\scriptscriptstyle L}$, $p_{i}^{\scriptscriptstyle PV}$ and $q_{i}^{\scriptscriptstyle PV}$ are denoted as the components of the vectors $\textbf{v}$, $\textbf{p}^{\scriptscriptstyle L}$, $\textbf{q}^{\scriptscriptstyle L}$, $\textbf{p}^{\scriptscriptstyle PV}$ and $\textbf{q}^{\scriptscriptstyle PV}$ respectively. We define a function $\mathit{h}: \mathbb{P}(V) \rightarrow \mathbb{P}(\mathcal{S})$ that maps a subset of $V$ to its correlated measures, where $\mathbb{P}(\mathcal{X})$ denotes the power set of an arbitrary set $\mathcal{X}$. The observation set is defined as $\mathcal{O} = \mathlarger{\mathlarger{\times}}_{\scriptscriptstyle i \in \mathcal{I}} \mathcal{O}_{i}$, where $\mathcal{O}_{i} = (h \circ  f \circ g) (\mathit{i})$ indicates the measures within the region where agent $\mathit{i}$ is located.
    
    \paragraph{Action Set.} Each agent $\mathit{i} \in \mathcal{I}$ is equipped with a continuous action set $\mathcal{A}_{i} = \{ \mathit{a}_{i} \ : \ - c \leq \mathit{a}_{i} \leq c, c > 0 \}$. The continuous action represents the ratio of maximum reactive power it generates. In more details, the reactive power generated from the $k$th PV inverter is $q_{k}^{\scriptscriptstyle PV} = a_{k} \ \sqrt{(s_{k}^{\scriptscriptstyle \max})^{2} - (p_{k}^{\scriptscriptstyle PV})^{2}}$, where $s_{k}^{\scriptscriptstyle \max}$ is the maximum apparent power of the $\mathit{k}$th node that is dependent on the physical capacity of the PV inverter.\footnote{Note that the reactive power range actually dynamically changes at each time step.}\footnote{Yielding $(q_{k}^{\scriptscriptstyle PV})_{t}$ at each time step $t$ is equivalent to yielding $\Delta_{t} (q_{k}^{\scriptscriptstyle PV})$ (i.e. the change of reactive power generation at each time step), since $(q_{k}^{\scriptscriptstyle PV})_{t} = (q_{k}^{\scriptscriptstyle PV})_{t-1} + \Delta_{t} (q_{k}^{\scriptscriptstyle PV})$. For easily satisfying the safety condition, we directly yield $q_{k}^{\scriptscriptstyle PV}$ at each time step in this work.} If $a_{k} > 0$, it means penetrating reactive powers to the distribution network. If $a_{k} < 0$, it means absorbing reactive powers from the distribution network. The value of $c$ is usually selected as per the loading capacity of a distribution network, which is for the safety of operations. The joint action set is denoted as $\mathcal{A} = \mathlarger{\mathlarger{\times}}_{\scriptscriptstyle i \in \mathcal{I}} \mathcal{A}_{i}$.
    
    \paragraph{State Transition Probability Function.} Since the state includes the last action and the change of loads is random (that theoretically can be modelled as any probabilistic distribution), we can naturally define the state transition probability function as $T: \mathcal{S} \times \mathcal{A} \times \mathcal{S} \rightarrow [0, 1]$ that follows Markov decision process. Specifically, $T(\mathbf{s}_{t+1}, \mathbf{s}_{t}, \mathbf{a}_{t}) = Pr(\mathbf{s}_{t+1} | \delta(\mathbf{s}_{t}, \mathbf{a}_{t}))$, where $\mathbf{a}_{t} \in \mathcal{A}$ and $\mathbf{s}_{t}, \mathbf{s}_{t+1} \in \mathcal{S}$. $\delta(\mathbf{s}_{t}, \mathbf{a}) \mapsto \mathbf{s}_{t+\tau}$ denotes the solution of the power flow, whereas $Pr(\mathbf{s}_{t+1} | \mathbf{s}_{t+\tau})$ describes the change of loads (i.e. highly correlated to the user behaviours). $\tau \ll \Delta t$ is an extremely short interval much less than the time interval between two controls (i.e. a time step) and $\Delta t = 1$ in this thesis.
    
    \paragraph{Observation Probability Function.} We now define the observation probability function. In the context of electric power network, it describes the measurement errors that may occur in sensors. Mathematically, we can define it as $\Omega: \mathcal{S} \times \mathcal{A} \times \mathcal{O} \rightarrow [0, 1]$. Specifically, $\Omega(\mathbf{o}_{t+1} | \mathbf{s}_{t+1}, \mathbf{a}_{t}) = \mathbf{s}_{t+1} + \mathcal{N}(\mathbf{0}, \Sigma)$, where $\mathcal{N}(\mathbf{0}, \Sigma)$ is an isotropic multi-variable Gaussian distribution and $\Sigma$ is dependent on the physical properties of sensors (e.g. smart meters). 

    \paragraph{Reward Function.} The reward function is defined as follows:
    \begin{equation}
        R = - \frac{1}{|V|} \sum_{i \in V} l_{v}(v_{i}) - \alpha \cdot l_{q}(\mathbf{q}^{\scriptscriptstyle PV}),
    \label{eq:reward_function}
    \end{equation}
    where $l_{v}(\cdot)$ is a voltage barrier function and $l_{q}(\mathbf{q}^{\scriptscriptstyle PV}) = \frac{1}{|\mathcal{I}|}||\mathbf{q}^{\scriptscriptstyle PV}||_{1}$ is the reactive power generation loss (i.e. a type of power loss approximation easy for computation). The objective is to control the voltage within a safety range around $v_{\text{\tiny{ref}}}$, while the reactive power generation is as less as possible, i.e., $l_{q}(\mathbf{q}^{\scriptscriptstyle PV}) < \epsilon \text{ and } \epsilon > 0$. Similar to the mathematical tricks used in $\beta$-VAE \cite{higgins2016beta}, by KKT conditions we can transform a constrained reward to a unconstrained reward by a Lagrangian multiplier $\alpha \in (0, 1)$ shown in Eq.~\ref{eq:reward_function}. Since $l_{v}(\cdot)$ is not easy to define in practice (i.e., it affects $l_{q}(\mathbf{q}^{\scriptscriptstyle PV})$), we aim at studying for a good choice in this paper. Although the action range has been restricted to avoid the violence of the loading capacity of power distribution networks, in simulation there still exist possibilities that this accident could happen. To address this problem, if the violence of the loading capacity appears, the system will backtrack to the last state and terminate the simulation, meanwhile, a penalty of $-200$ will become the reward instead of the one calculated in Eq.~\ref{eq:reward_function}.
    
    \paragraph{Objective Function.} The objective function of this problem is $\max_{\pi} \mathbb{E}_{\pi}[\sum_{t=0}^{\infty} \gamma^{t} R_{t}]$, where $\pi = \mathlarger{\mathlarger{\times}}_{\scriptscriptstyle i \in \mathcal{I}} \pi_{i}$; $\pi_{i}: \bar{\mathcal{O}}_{i} \times \mathcal{A}_{i} \rightarrow [0, 1]$ and $\bar{\mathcal{O}}_{i}=(\mathcal{O}_{i}^{\tau})_{\tau=1}^{h}$ is a history of observations with the length as $\mathit{h}$. Literally, we need to find the optimal joint policy $\pi$ to maximize the discounted cumulative rewards.

\section{Voltage Barrier Function}
\label{sec:voltage_loss}
    \begin{figure*}[ht!]
        \centering
        \begin{subfigure}[b]{0.45\textwidth}
        	\centering
    	    \includegraphics[width=\textwidth]{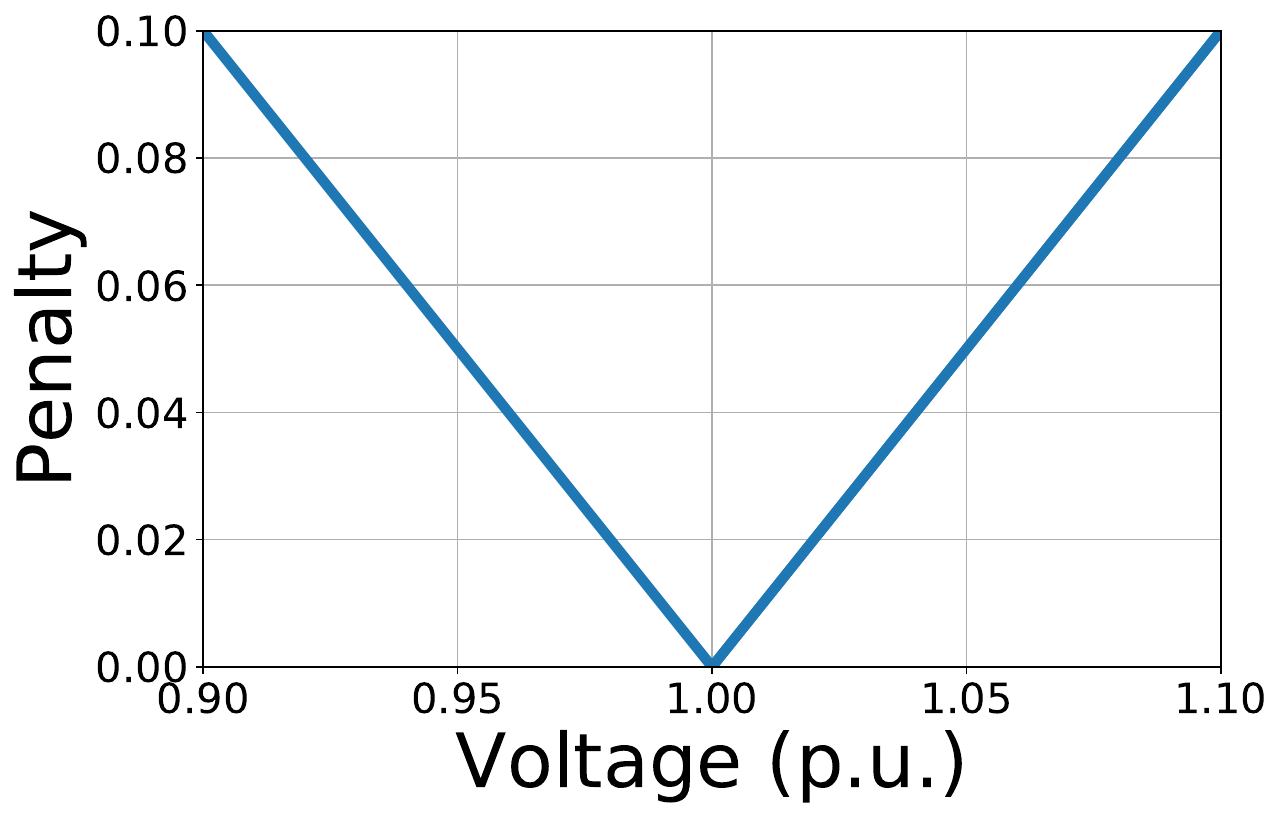}
    	    \caption{L1-shape.}
    	\label{fig:l1_loss}
        \end{subfigure}
        ~
        \begin{subfigure}[b]{0.45\textwidth}
            \centering                \includegraphics[width=\textwidth]{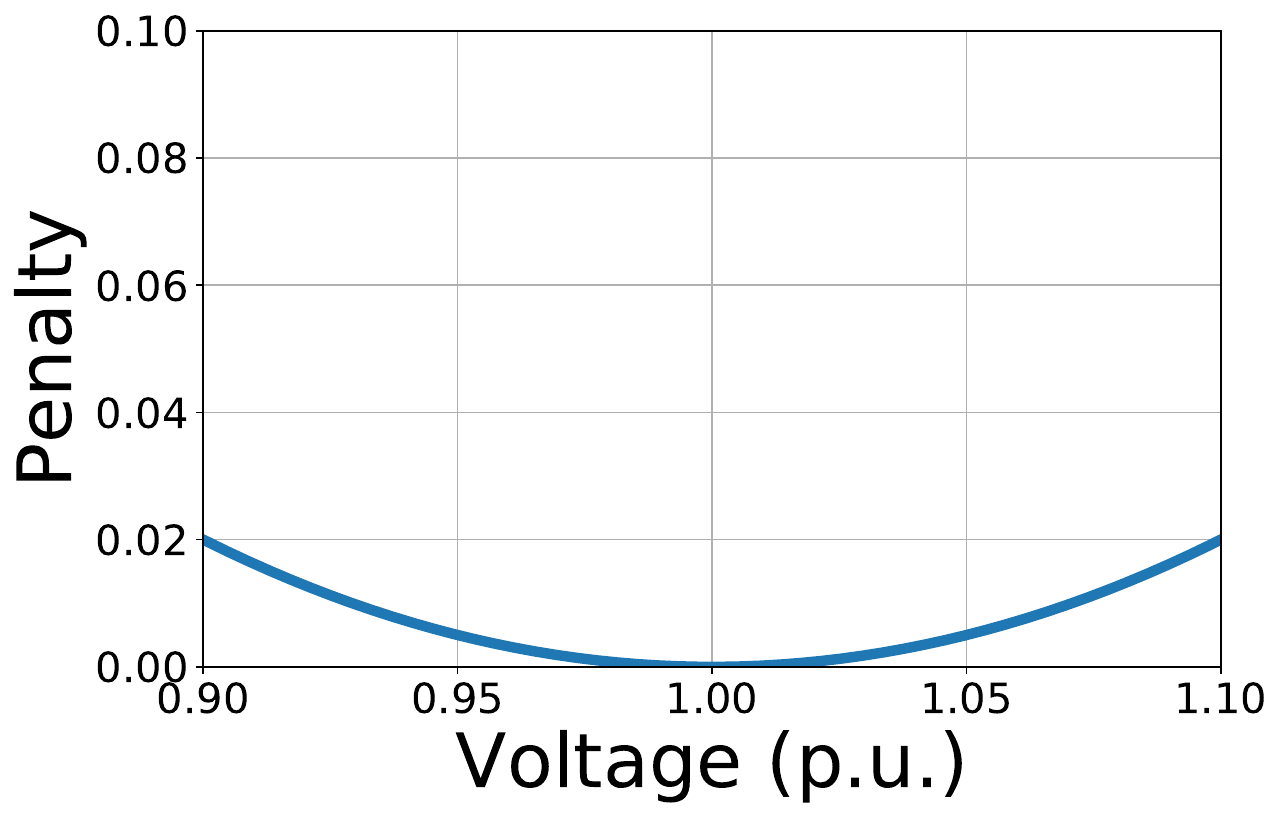}
            \caption{L2-shape.}
        \label{fig:l2_loss}
        \end{subfigure}
        ~
        \begin{subfigure}[b]{0.45\textwidth}
            \centering                \includegraphics[width=\textwidth]{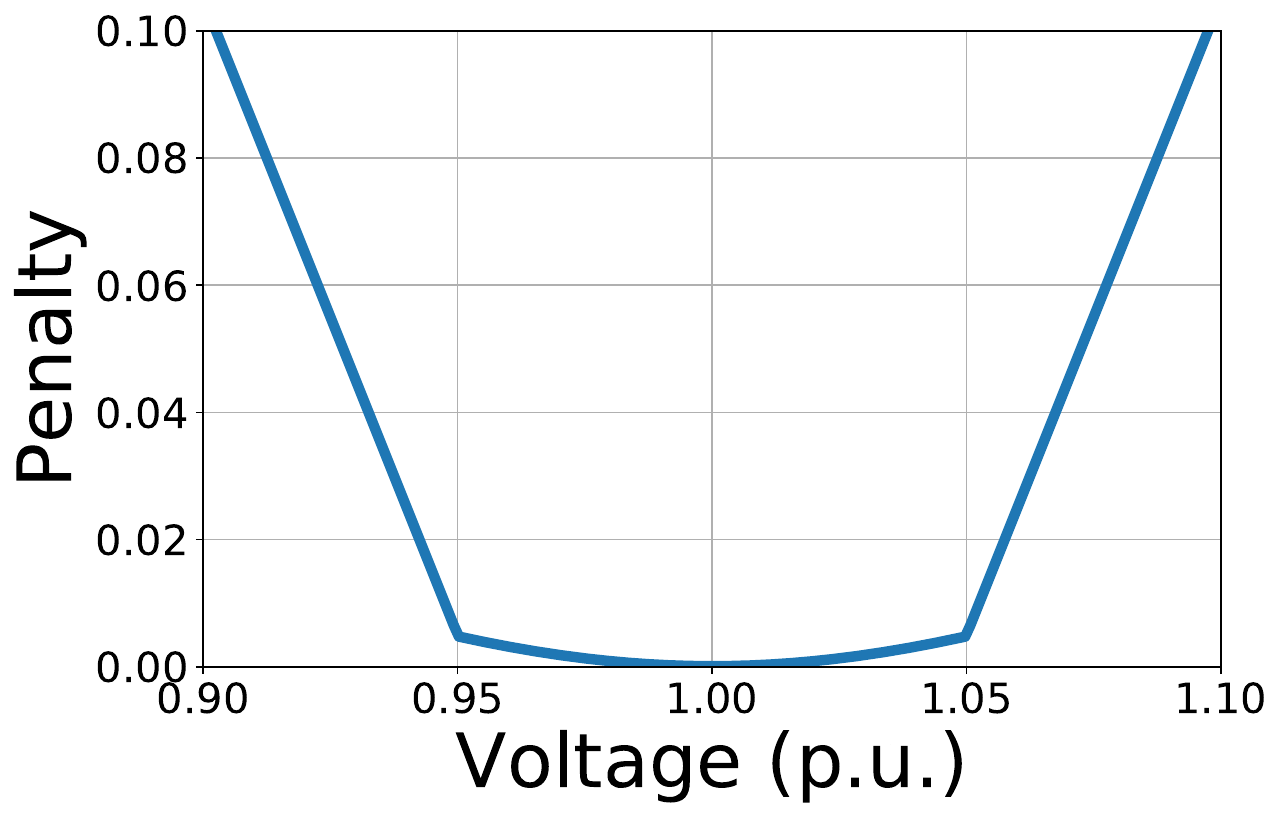}
            \caption{Bowl-shape.}
        \label{fig:bowl_loss}
        \end{subfigure}
        \caption{Three voltage barrier functions, where L1-shape and L2-shape are 2 baselines while Bowl-shape is proposed in this thesis.}
        \label{fig:voltage_loss}
    \end{figure*}
    
    We define $v_{\text{\tiny{ref}}} = 1 \ p.u.$ in this thesis, and the voltage needs to be controlled within the safety range from $0.95 \ p.u.$ to $1.05 \ p.u.$, which sets the constraint of control. The voltage constraint is difficult to be handled in MARL, so we use a barrier function to represent the constraint. L1-shape (see Figure \ref{fig:l1_loss}) was most frequently used in the previous work \cite{cao2020distributed,cao2020multi,cao2021data}, however, this may lead to wasteful reactive power generations since $\frac{|\Delta l_{v}|}{\alpha|\Delta l_{q}|} \gg 1$ within the safety range of voltage. Although L2-shape (see Figure \ref{fig:l2_loss}) may alleviate this problem, it may be slow to guide the policy outside the safety range. To address these problems, we propose a barrier function called Bowl-shape that combines the advantages of L1-shape and L2-shape. It gives a steep gradient outside the safety range, while it provides a slighter gradient as voltage tends to the $v_{\text{\tiny{ref}}}$ that enables $\frac{|\Delta l_{v}|}{\alpha|\Delta l_{q}|} \rightarrow 0$ as $v \rightarrow v_{\text{\tiny{ref}}}$.
    
    We now show and discuss the analytic forms of all voltage barrier functions mentioned above. The L1-shape can be expressed as follows:
    \begin{equation}
        l_{v}(v_{k}) = | v_{k} - v_{\text{\tiny{ref}}} |, \ \ \forall k \in V.
    \end{equation}
    The L2-shape can be expressed as follows:
    \begin{equation}
        l_{v}(v_{k}) = ( v_{k} - v_{\text{\tiny{ref}}} )^{2}, \ \ \forall k \in V.
    \end{equation}
    The Bowl-shape can be expressed as follows:
    \begin{equation}
        l_{v}(v_{k}) = 
        \begin{cases}
            a \cdot | v_{k} - v_{\text{\tiny{ref}}} | - b & \text{If } | v_{k} - v_{\text{\tiny{ref}}} | > 0.05, \\
            - c \cdot \mathcal{N}(v_{k} \ | \ v_{\text{\tiny{ref}}}, 0.1) + d & \text{Otherwise},
        \end{cases}
    \end{equation}
    where $a, b, c, d$ are 4 hyperparameters to adjust the shape and smoothness of function that are set to $2, 0.095, 0.01, 0.04$ respectively in this thesis; $\mathcal{N}(v_{k} \ | \ v_{\text{\tiny{ref}}}, 0.1)$ is a density function of the Gaussian distribution with the mean as $v_{\text{\tiny{ref}}}$ and the standard deviation as 0.1. In addition to the significance of satisfying the objective of active voltage control, this construction can also be interpreted as a sort of statistical implication. $v_{k}$ is assumed to follow the Laplace distribution outside the safety range, while it is assumed to follow the Gaussian distribution inside the safety range. Thereby, the active voltage control problem can be transformed to the maximum likelihood estimation (MLE) over a mixture distribution over voltage with a constraint on the reactive power generation.
    
\section{Simulation Settings}
\label{sec:experimental_setups}
    \paragraph{Power Network Topology.} Two MV networks, IEEE 33-bus \cite{baran1989network} and 141-bus \cite{khodr2008maximum} are modified as systems under test.\footnote{The original topologies and parameters can be found in MATPOWER \cite{zimmerman2010matpower} description file on \url{https://github.com/MATPOWER/matpower/tree/master/data}.} To show the flexibility on the network with multi-voltage levels, we construct a 110kV-20kV-0.4kV (high-medium-low voltage) 322-bus network using benchmark topology from SimBench \cite{meinecke20simbench}. For each network, a main branch is firstly determined and the control regions are partitioned by the shortest path between the terminal bus and the coupling point on the main branch. Each region consists of 1-4 PVs dependent on various regional sizes. The specific network description and partition are shown in Appendix \ref{subsec:network_topology}. To give an overall picture of the task, we demonstrate the 33-bus network in Figure \ref{fig:33_bus_illustration}.
    \begin{figure*}[ht!]
        \centering
        \includegraphics[width=0.95\linewidth]{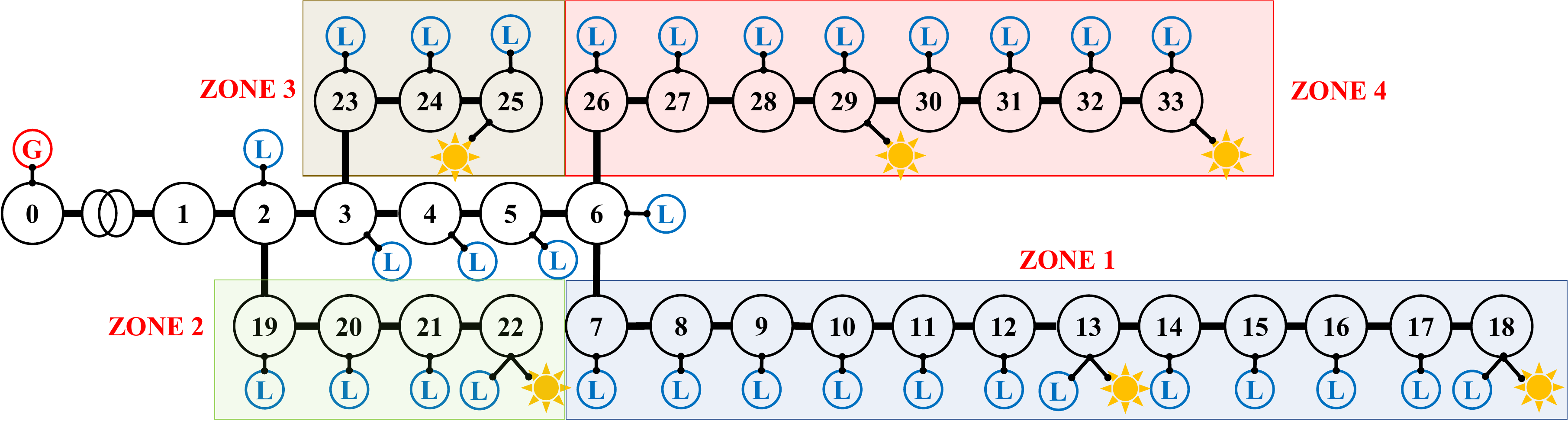}
        \caption{Illustration on 33-bus network. Each bus is indexed by a circle with a number. 4 control regions are partitioned by the smallest path from the terminal to the main branch (bus 1-6). We control the voltages on bus 2-33 whereas bus 0-1 represent the substation or main grid with the constant voltage and infinite active and reactive power capacity. \textbf{G} represents an external generator; small \textbf{L}s represent loads; and the sun emoji represents the location where a PV is installed.}
    \label{fig:33_bus_illustration}
    \end{figure*}
    
    \paragraph{Data Description.} The load profile of each network is modified based on the real-time Portuguese electricity consumption accounting for 232 consumers of 3 years.\footnote{\url{https://archive.ics.uci.edu/ml/datasets/ElectricityLoadDiagrams20112014}.} To highlight the difference between residential and industrial users, we randomly perturb $\pm 5\%$ on the default power factors defined in the case files and accordingly generate real-time reactive power consumption. The solar data is collected from Elia group,\footnote{\url{https://www.elia.be/en/grid-data/power-generation/} \\ \url{solar-pv-power-generation-data}.} i.e. a Belgium’s power network operator. The load and PV data are then interpolated with 3-min resolution that is consistent with the real-time control period in the grid. To distinguish among different solar radiation levels in various regions, the 3-year PV generations from 10 cites/regions are collected and PVs in the same control region possess the same generation profiles. We define the PV penetration rate ($PR$) as the ratio between rated PV generation and rated load consumption. In this thesis, we set $PR\in\{2.5,4,2.5\}$ as the default $PR$ for different topologies. We oversize each PV inverter by 20\% of its maximum active power generation to satisfy the IEEE grid code \cite{ieee2018ieee}. Besides, each PV inverter is considered to be able to generate reactive power in the STATCOM mode during night \cite{varma2018pv}. The median and 25\%-75\% quantile shadings of PV generation, and the mean and minima-maxima shading of the loads are illustrated in Figure \ref{fig:pv_load_profile}. The details are described in Appendix \ref{subsec:data_descriptions}.
    \begin{figure*}[ht!]
        \centering
        \begin{subfigure}[b]{0.32\textwidth}
        	\centering
    	    \includegraphics[width=\textwidth]{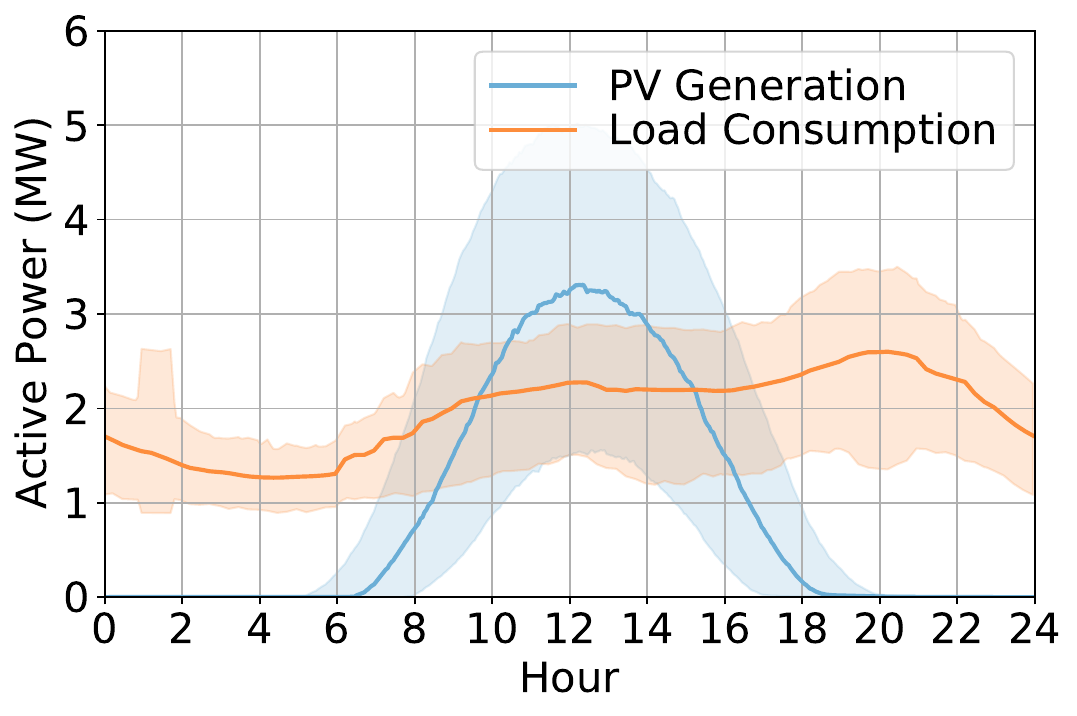}
    	    \caption{33-bus network.}
        \end{subfigure}
        ~
        \begin{subfigure}[b]{0.32\textwidth}
            \centering                
            \includegraphics[width=\textwidth]{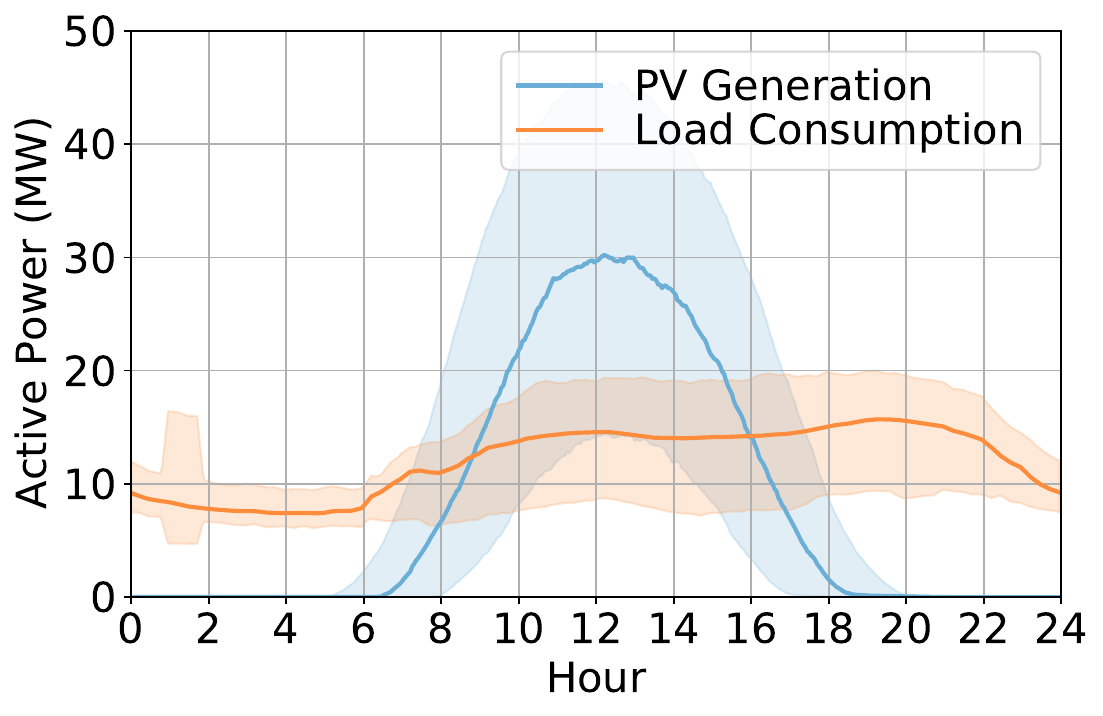}
            \caption{141-bus network.}
        \end{subfigure}
        ~
        \begin{subfigure}[b]{0.32\textwidth}
            \centering                
            \includegraphics[width=\textwidth]{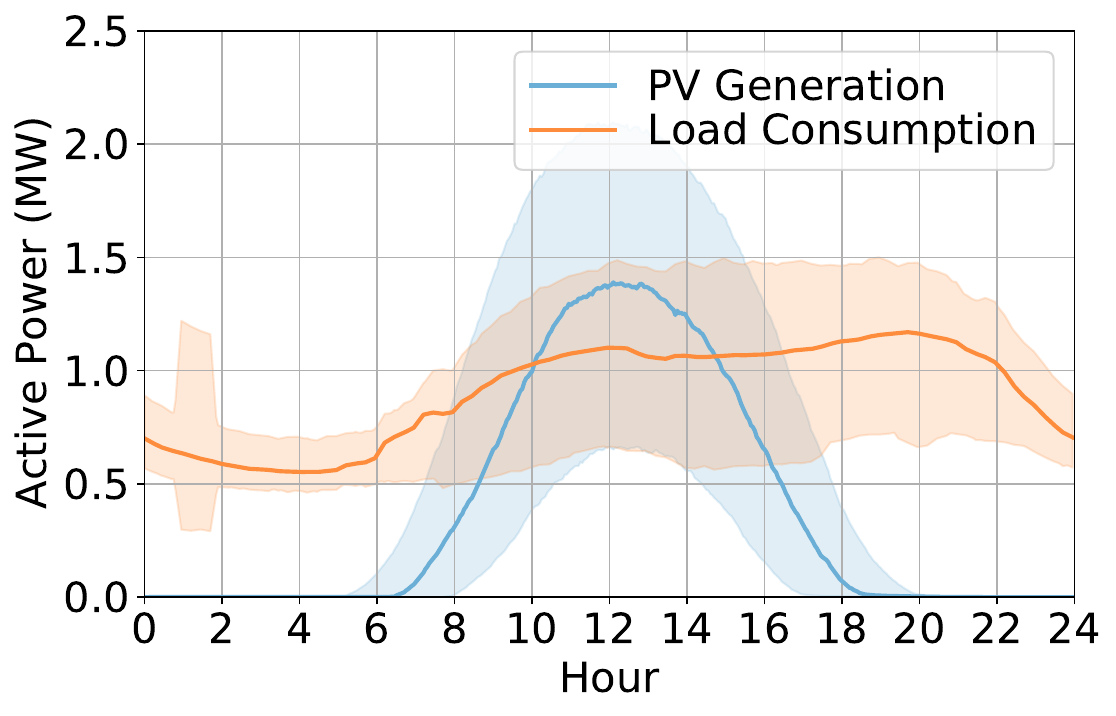}
            \caption{322-bus network.}
        \end{subfigure}
        \caption{Active PV generations and load consumption.}
        \label{fig:pv_load_profile}
    \end{figure*}
    
    \paragraph{MARL Simulation Setting.} We now describe the simulation setting by the view of MARL. In the 33-bus network, there are 4 regions with 6 agents. In the 141-bus network, there are 9 regions with 22 agents. In the 322-bus network, there are 22 regions with 38 agents. The discount factor $\gamma$ is set to $0.99$. $\alpha$ in Eq.~\ref{eq:reward_function} is set to $0.1$. To guarantee the safety of distribution networks, we manually set the range of actions for each scenario, with $[-0.8, 0.8]$ for 33-bus network, $[-0.6, 0.6]$ for 141-bus network, and $[-0.7, 0.7]$ for 322-bus network. During training, we randomly sample the initial state for an episode and each episode lasts for 240 time steps (i.e. a half day). Every simulation is run with 5 random seeds and the test results during training are given by the median and the 25\%-75\% quartile shading. Each test is conducted every 20 episodes with 10 randomly selected episodes for evaluation.
    
    \paragraph{Evaluation Metrics.} In simulation, we use the following two metrics to evaluate the performance of algorithms. We aim to find algorithms and reward functions with high \textit{control rate} (CR) and low \textit{power loss} (PL). The details of the metrics are shown as follows:
    \begin{itemize}[leftmargin=*]
    \item \textit{Control rate}: It calculates the ratio of the timesteps where all buses' voltages are under control to the total time steps during each episode.
    \item \textit{Power loss}: It calculates the average power loss (i.e., the total power loss of the power network divided by the number of buses) per timestep during each episode.
    \end{itemize}

    \paragraph{MARL Algorithm Settings.} We evaluate the performance of SQDDPG and SMFPPO, compared with state-of-the-art MARL algorithms, e.g. IDDPG \cite{Wang_2020}, MADDPG \cite{LoweWTHAM17}, COMA \cite{foerster2018counterfactual}, IPPO \cite{de2020independent}, MAPPO \cite{yu2021surprising}, and MATD3 \cite{ackermann2019reducing} on this real-world problem with continuous actions. Since COMA can only work for discrete actions, we conduct some modification to make it work for continuous actions (see Appendix \ref{subsec:continuous_coma} for more details). The details of algorithmic settings are shown in Appendix \ref{subsec:algo_settings_and_training_details}.
    
\section{Experimental Results}
\label{sec:experimental_results}
    \subsection{Main Results}
    \label{subsec:main_results}
        \paragraph{Diverse Algorithm Performance under Distinct Reward Functions.} To clearly show the relationship among all baseline algorithms and reward functions, we plot 3D surfaces over CR and PL with respect to algorithm types and reward types (i.e. distinct voltage barrier functions) in Figure \ref{fig:3d_reward_comparison}. It is apparent that the performance of algorithms are highly correlated with the reward types. In other words, the same algorithm could perform diversely even trained by different reward functions with the same objective but different shapes. This motivates us to find a comparatively good shape of reward function (i.e. mainly dependent on voltage barrier functions) to each scenario (averaging the performance of algorithms) to evaluate SQDDPG and SMFPPO.
        \begin{figure*}[ht!]
            \centering
            \begin{subfigure}[b]{0.325\linewidth}
                \centering
                \includegraphics[width=\textwidth]{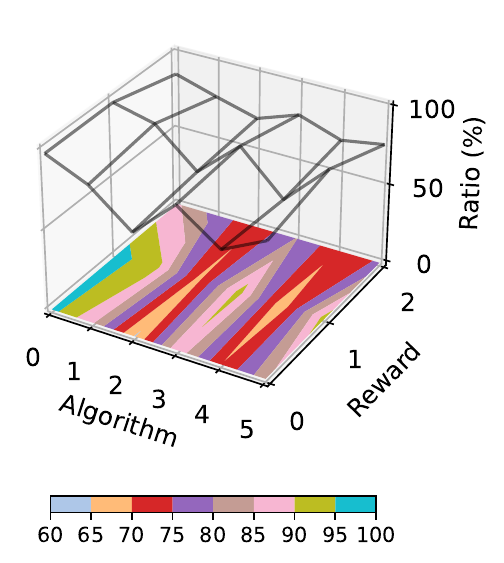}
                \caption{CR-33.}
            \end{subfigure}
            \hfill
            \begin{subfigure}[b]{0.325\linewidth}
                \centering
                \includegraphics[width=\textwidth]{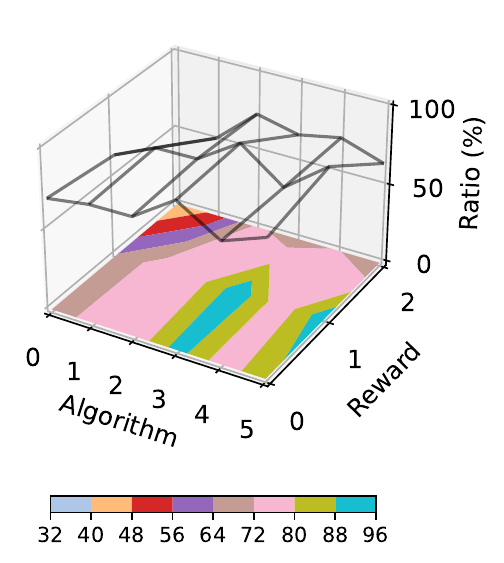}
                \caption{CR-141.}
            \end{subfigure}
            \hfill
            \begin{subfigure}[b]{0.325\linewidth}
                \centering
                \includegraphics[width=\textwidth]{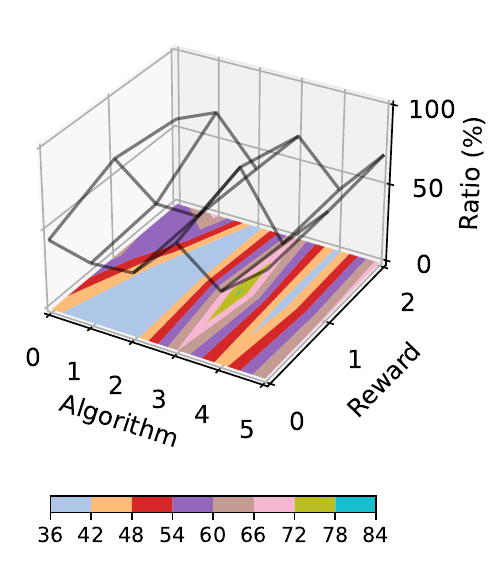}
                \caption{CR-322.}
            \end{subfigure}
            \hfill
            \begin{subfigure}[b]{0.325\linewidth}
                \centering                \includegraphics[width=\textwidth]{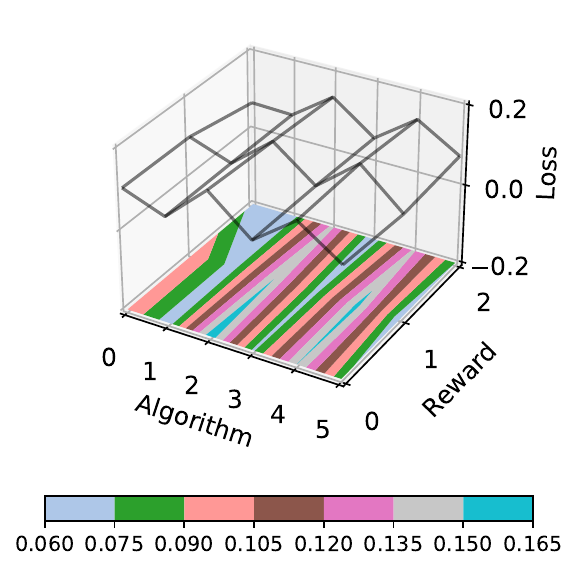}
                \caption{PL-33.}
            \end{subfigure}
            \hfill
            \begin{subfigure}[b]{0.325\linewidth}
                \centering
                \includegraphics[width=\textwidth]{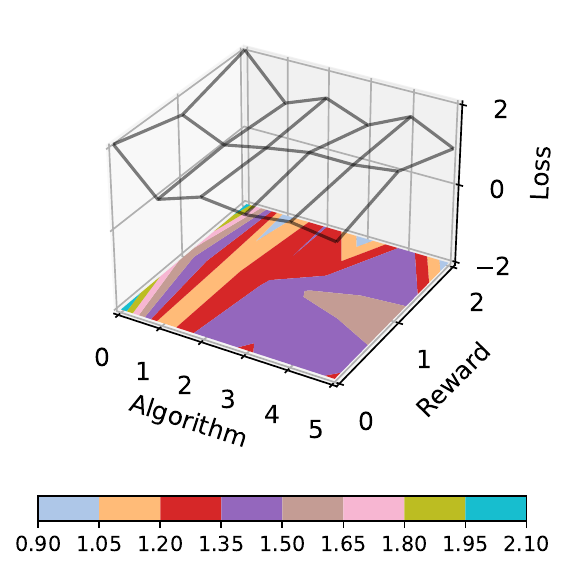}
                \caption{PL-141.}
            \end{subfigure}
            \hfill
            \begin{subfigure}[b]{0.325\linewidth}
                \centering
                \includegraphics[width=\textwidth]{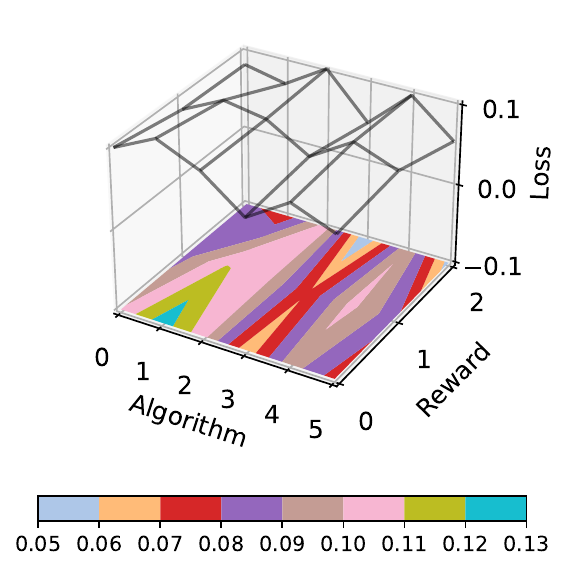}
                \caption{PL-322.}
            \end{subfigure}
            \caption{Median performances of overall algorithms trained with various rewards consist of distinct voltage barrier functions shown in 3D surfaces. The sub-caption indicates [metric]-[scenario].}
        \label{fig:3d_reward_comparison}
        \end{figure*}
        
        \paragraph{Voltage Barrier Function Comparisons.} To investigate the effect of different voltage barrier functions and select the best choice for each scenario, we show the median performance of overall six baseline MARL algorithms in Figure \ref{fig:reward_comparison}. It can be observed that the Bowl-shape can preserve the high CR, while maintain the low PL on the 33-bus and the 141-bus networks. Although L1-shape can achieve the best CR on the 33-bus and the 141-bus networks, its PL on the 141-bus network is the highest. L2-shape performs the worst on the 33-bus and the 141-bus networks, but performs the best on the 322-bus network with the highest CR and the lowest PL. The reason could be that its slighter gradients is more suitable for adapting to many agents. In summary, the above results show that the L1-shape, the Bowl-shape and the L2-shape are the best choices for the 33-bus network, the 141-bus network and the 322-bus network, respectively. As a result, we use the above setting of voltage barrier functions to evaluate SQDDPG and SMFPPO.
        \begin{figure*}[ht!]
            \centering
            \begin{subfigure}[b]{0.325\linewidth}
                \centering
                \includegraphics[width=\textwidth]{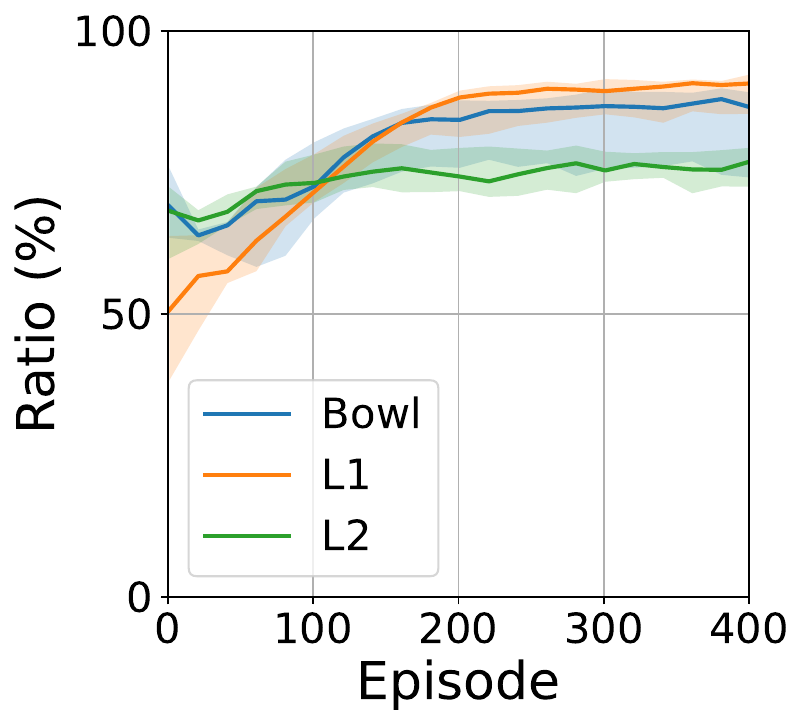}
                \caption{CR-33.}
            \end{subfigure}
            \hfill
            \begin{subfigure}[b]{0.325\linewidth}
                \centering
                \includegraphics[width=\textwidth]{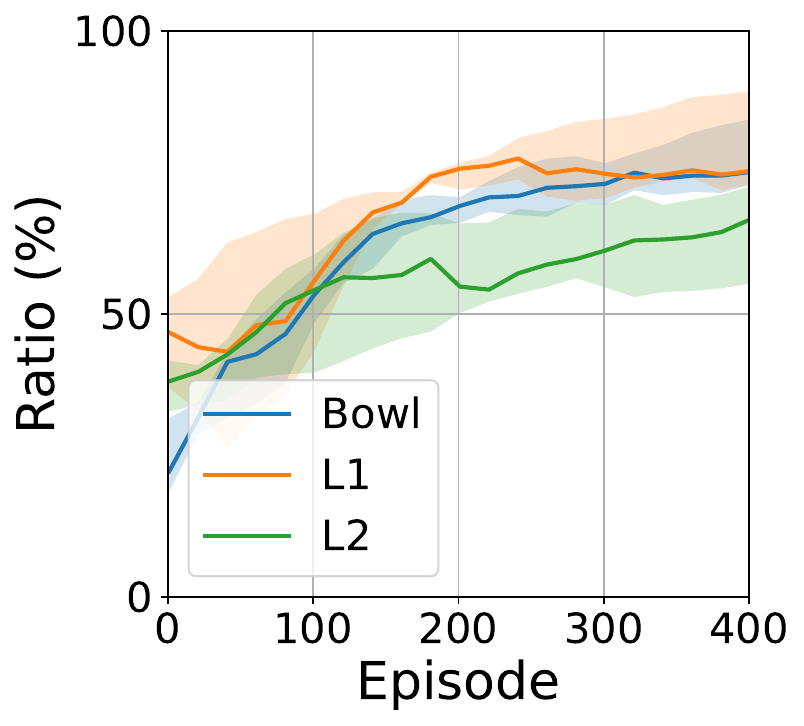}
                \caption{CR-141.}
            \end{subfigure}
            \hfill
            \begin{subfigure}[b]{0.325\linewidth}
                \centering
                \includegraphics[width=\textwidth]{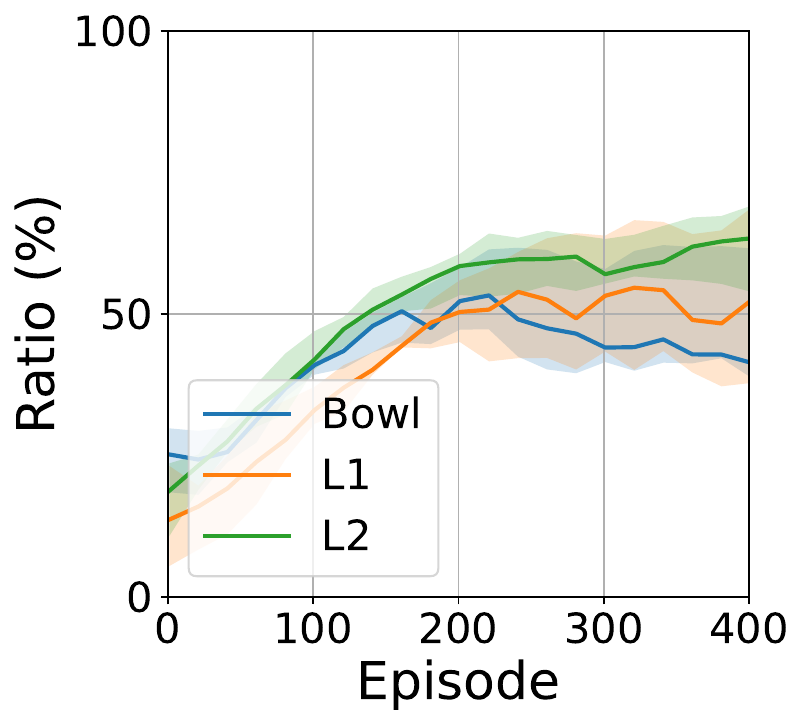}
                \caption{CR-322.}
            \end{subfigure}
            \hfill
            \begin{subfigure}[b]{0.325\linewidth}
                \centering                \includegraphics[width=\textwidth]{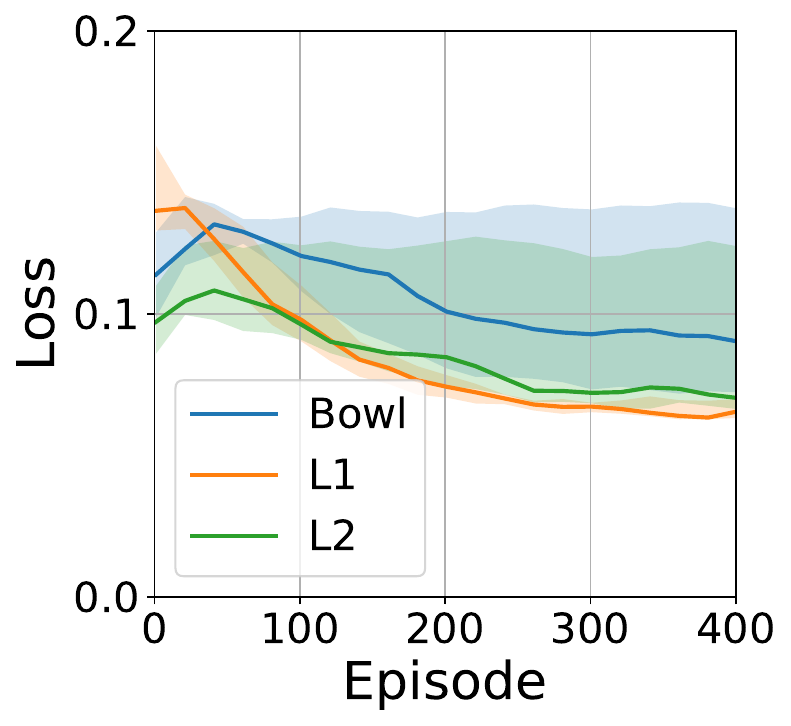}
                \caption{PL-33.}
            \end{subfigure}
            \hfill
            \begin{subfigure}[b]{0.325\linewidth}
                \centering
                \includegraphics[width=\textwidth]{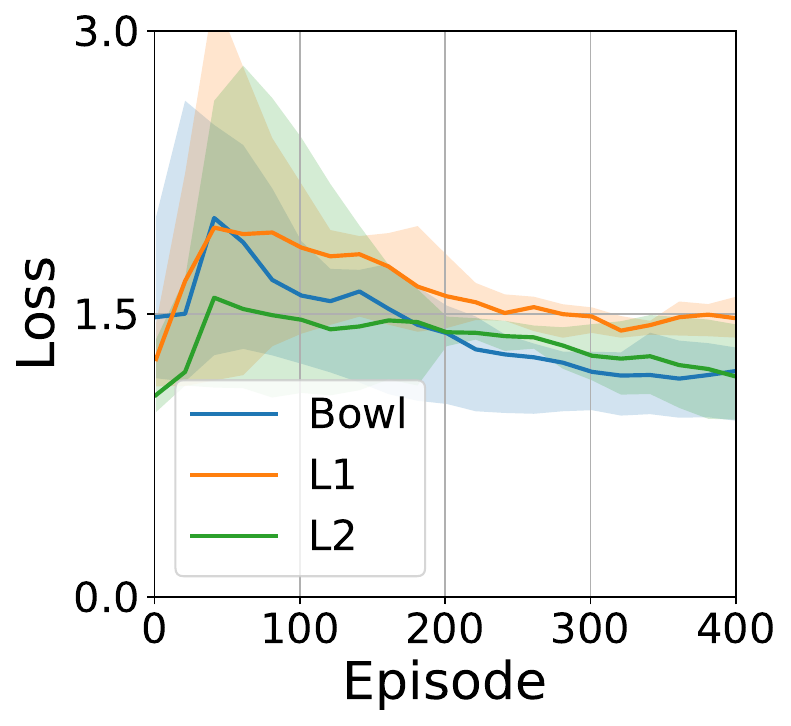}
                \caption{PL-141.}
            \end{subfigure}
            \hfill
            \begin{subfigure}[b]{0.325\linewidth}
                \centering
                \includegraphics[width=\textwidth]{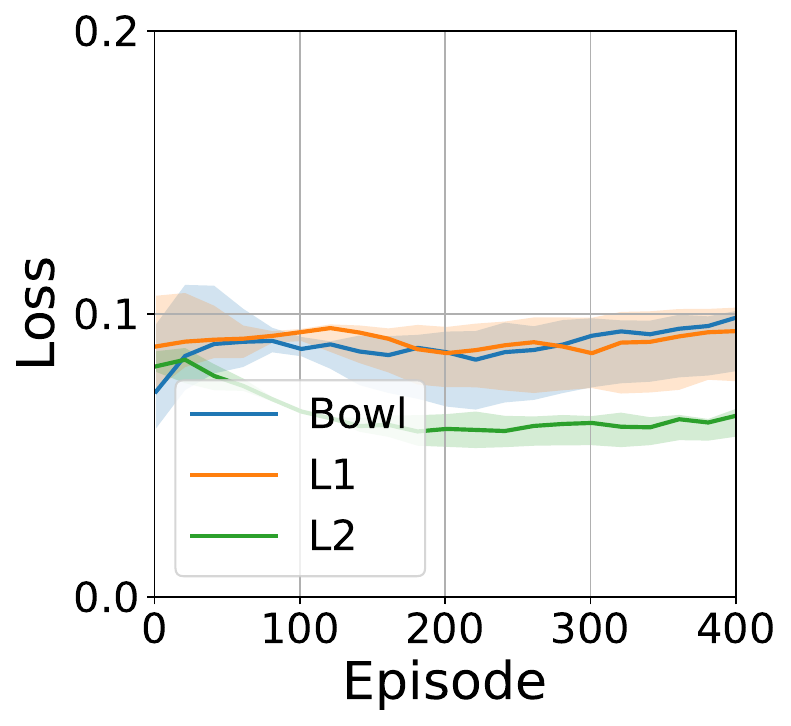}
                \caption{PL-322.}
            \end{subfigure}
            \caption{Median performance of overall algorithms with different voltage barrier functions. The sub-caption indicates [metric]-[scenario].}
        \label{fig:reward_comparison}
        \end{figure*}
        
        \paragraph{Algorithm Performance.} Based on the selection of the voltage barrier function for each scenario, we now show the main results of all algorithms on all scenarios in Figure \ref{fig:alg_comparison}. MADDPG and MATD3 generally perform well on all scenarios. COMA performs well over CR on 33-bus networks and the performance falls on the large scale scenarios, but its PL is high. This reveals the limitation of COMA on the scaling to many agents. Although MAPPO and IPPO perform well in games \cite{de2020independent,yu2021surprising}, their performance on the real-world power network problems are poor. Compared with IPPO and MAPPO, the main difference of SMFPPO is that the evaluation of return is based on the Shapley value mechanism. The superior performance of SMFPPO verifies the effectiveness of the main contribution of this thesis. IDDPG generally performs at the middle place, which may be due to non-stationary dynamics led by multi-agents \cite{LoweWTHAM17}. SQDDPG and SMFPPO generally perform well.
        \begin{figure*}[ht!]
            \centering
            \begin{subfigure}[b]{0.325\linewidth}
                \centering                \includegraphics[width=\textwidth]{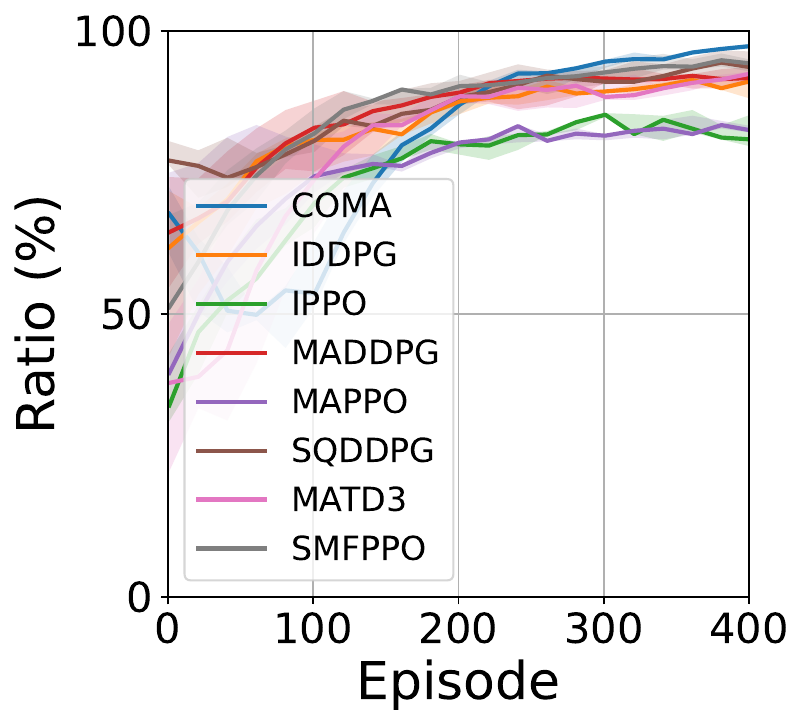}
                \caption{CR-L1-33.}
            \end{subfigure}
            \hfill
            \begin{subfigure}[b]{0.325\linewidth}
                \centering
                \includegraphics[width=\textwidth]{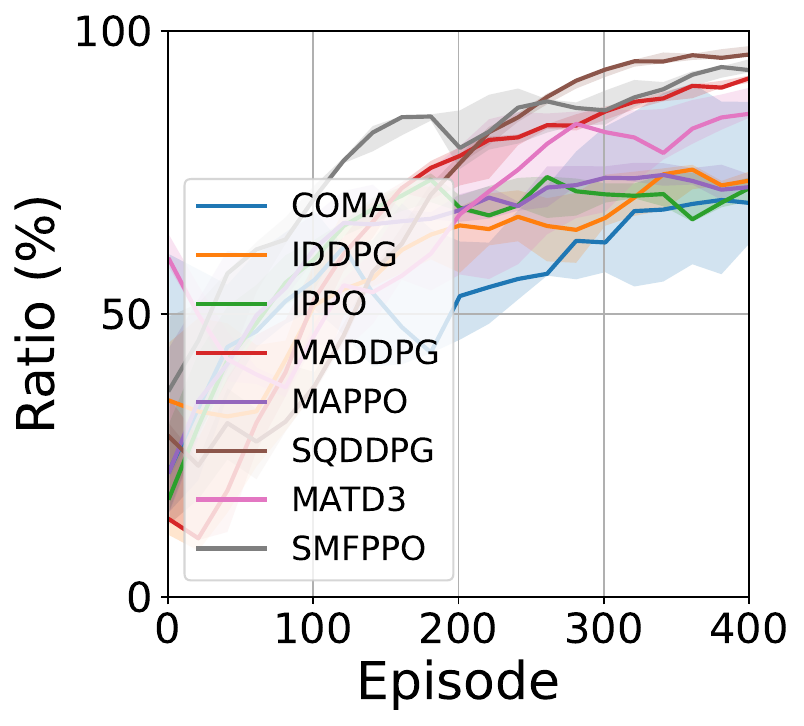}
                \caption{CR-Bowl-141.}
            \end{subfigure}
            \hfill
            \begin{subfigure}[b]{0.325\linewidth}
                \centering
                \includegraphics[width=\textwidth]{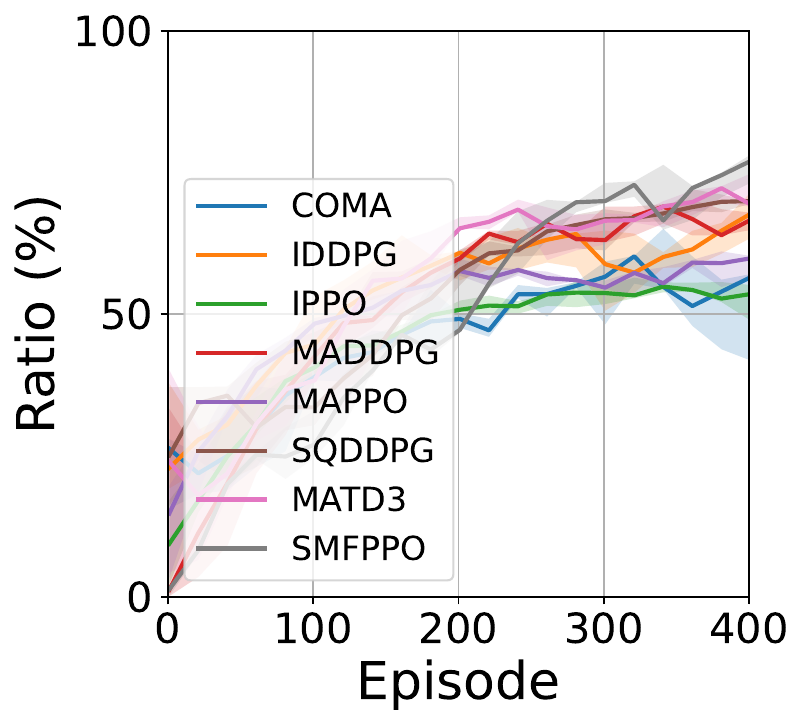}
                \caption{CR-L2-322.}
            \end{subfigure}
            \hfill
            \begin{subfigure}[b]{0.325\linewidth}
                \centering
                \includegraphics[width=\textwidth]{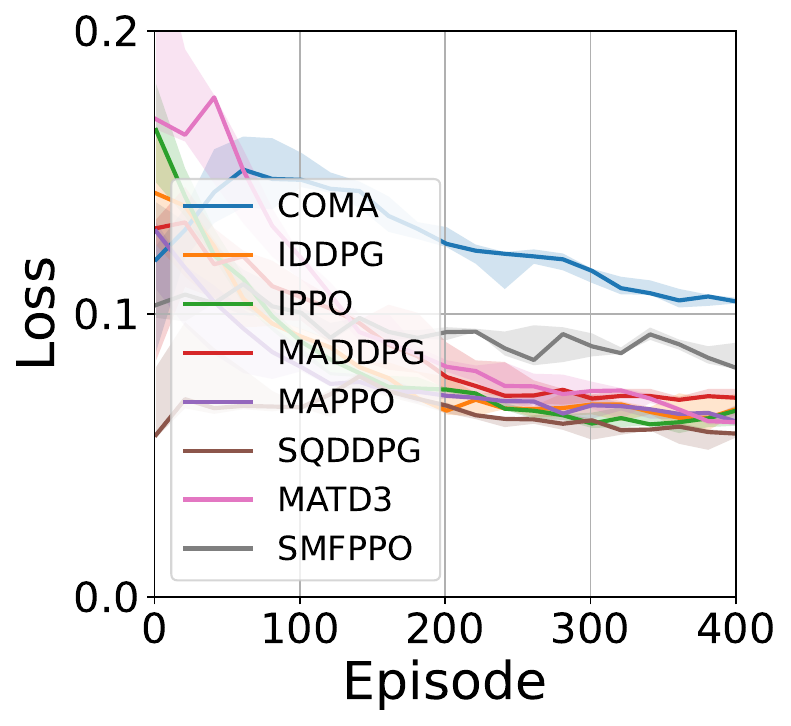}
                \caption{PL-L1-33.}
            \end{subfigure}
            \hfill
            \begin{subfigure}[b]{0.325\linewidth}
                \centering
                \includegraphics[width=\textwidth]{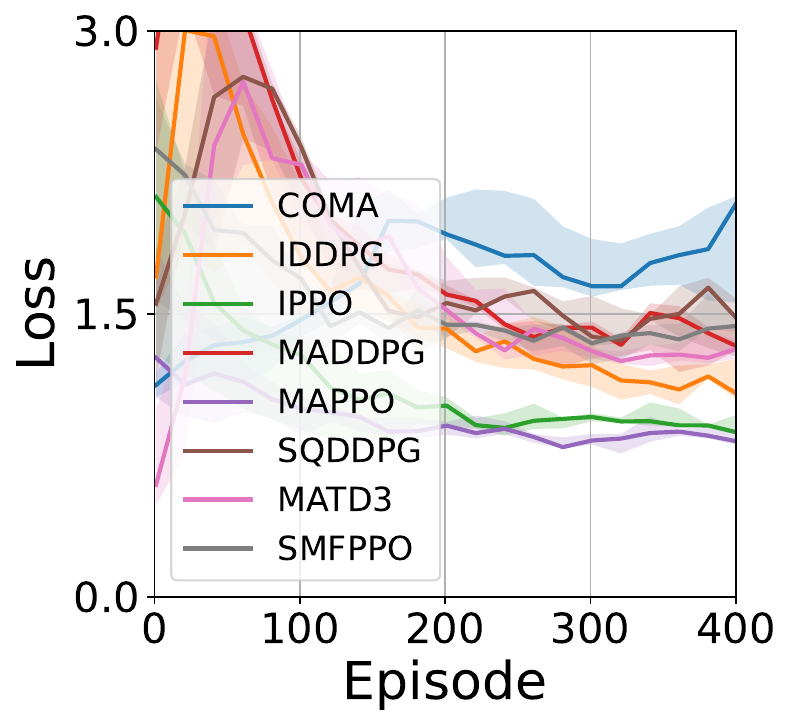}
                \caption{PL-Bowl-141.}
            \end{subfigure}
            \hfill
            \begin{subfigure}[b]{0.325\linewidth}
                \centering
                \includegraphics[width=\textwidth]{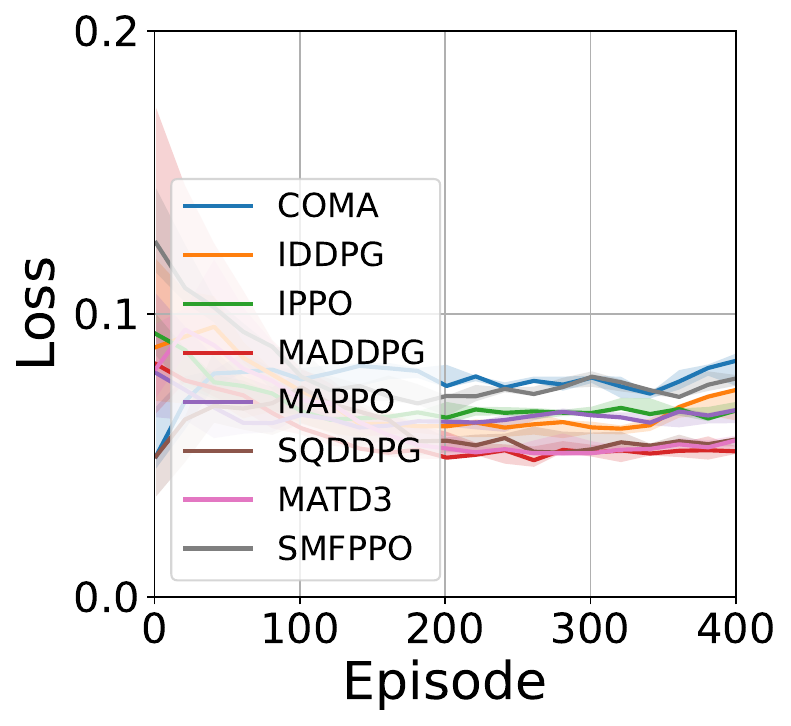}
                \caption{PL-L2-322.}
            \end{subfigure}
            \caption{Median CR and PL of algorithms with different voltage barrier functions. The sub-caption indicates [metric]-[barrier]-[scenario].}
        \label{fig:alg_comparison}
        \end{figure*}
        
    \subsection{Comparison between MARL and Traditional Control Methods}
    \label{subsec:comparison_between_marl_traditional_control_methods}
        To compare SMFPPO and SQDDPG with the traditional control methods, we conduct a series of tests on various network topologies (i.e. the 33-bus, the 141-bus, and the 322-bus networks). The traditional control method candidates are OPF \cite{gan2013optimal} and droop control \cite{jahangiri2013distributed}. For conciseness, we only demonstrate the voltage and the power of a typical bus with a PV installed (i.e. one of the most difficult buses to control) during a day (i.e. 480 consecutive timesteps) in summer and winter respectively. Given the current observation, either SMFPPO or SQDDPG performs an action once. In contrast, both traditional methods perform with an optimization procedure to asymptotically reach the stable and safe voltages. Besides, the droop gain of the droop control needs to be tuned and the OPF needs the global observation and network topology. Therefore, SMFPPO and SQDDPG are \textit{more economic and cost effective on running algorithms during execution}. This is also one of the motivations why we aim at investigating the potential of applying the technique of MARL to the real-world physical systems like power systems.
        
        \paragraph{One Bus in the 33-Bus Network.} From Figure \ref{fig:case_study_33_smfppo} and \ref{fig:case_study_33_sqddpg}, it can be seen that all methods control the voltage within the safety range in both summer and winter. Both SMFPPO and SQDDPG execute less power loss than the droop control but higher than the OPF. This phenomenon is possibly due to the fact that droop control is a fully distributed algorithm which cannot explicitly reduce the power loss and OPF is a centralised algorithm with the known system model, while SMFPPO and SQDDPG lie between these 2 types of algorithms. It is worth noting that the actions executed from SMFPPO and SQDDPG are similar to that of droop control. In comparison with the smoothness of the actions executed from SQDDPG, the actions executed from SMFPPO perform in a zigzag phenomenon.
        \begin{figure*}[ht!]
            \centering
            \begin{subfigure}[b]{0.30\textwidth}
            	\centering
        	    \includegraphics[width=\textwidth]{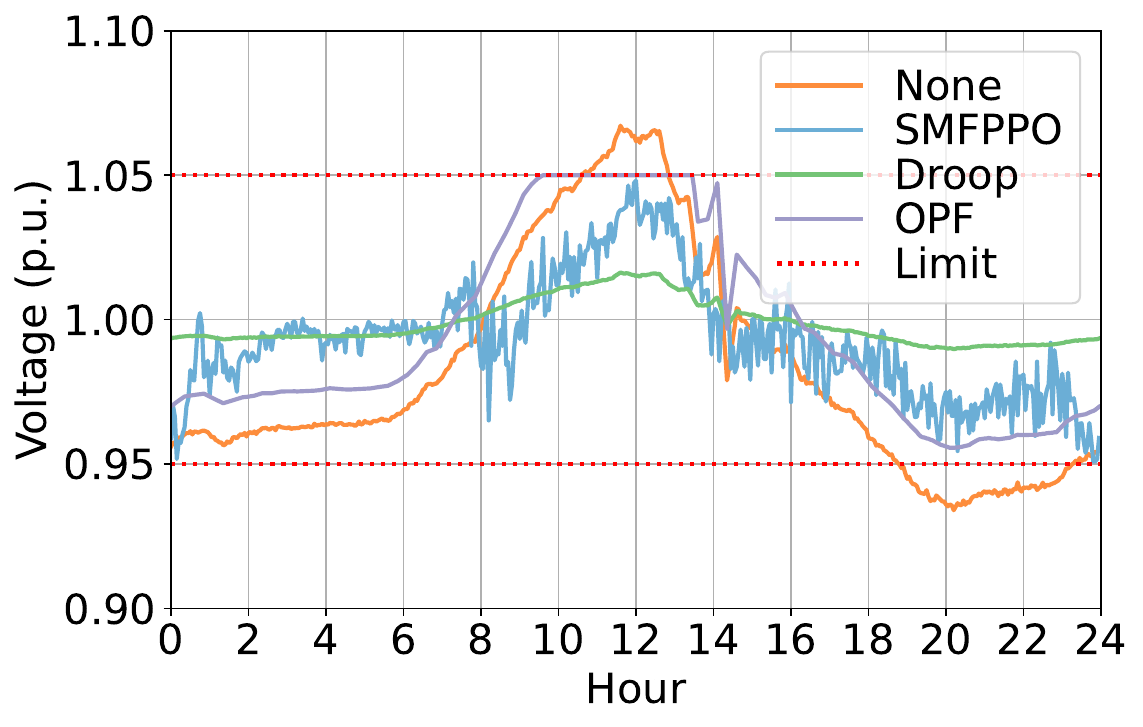}
            \end{subfigure}
            \quad
            \begin{subfigure}[b]{0.30\textwidth}
                \centering                
                \includegraphics[width=\textwidth]{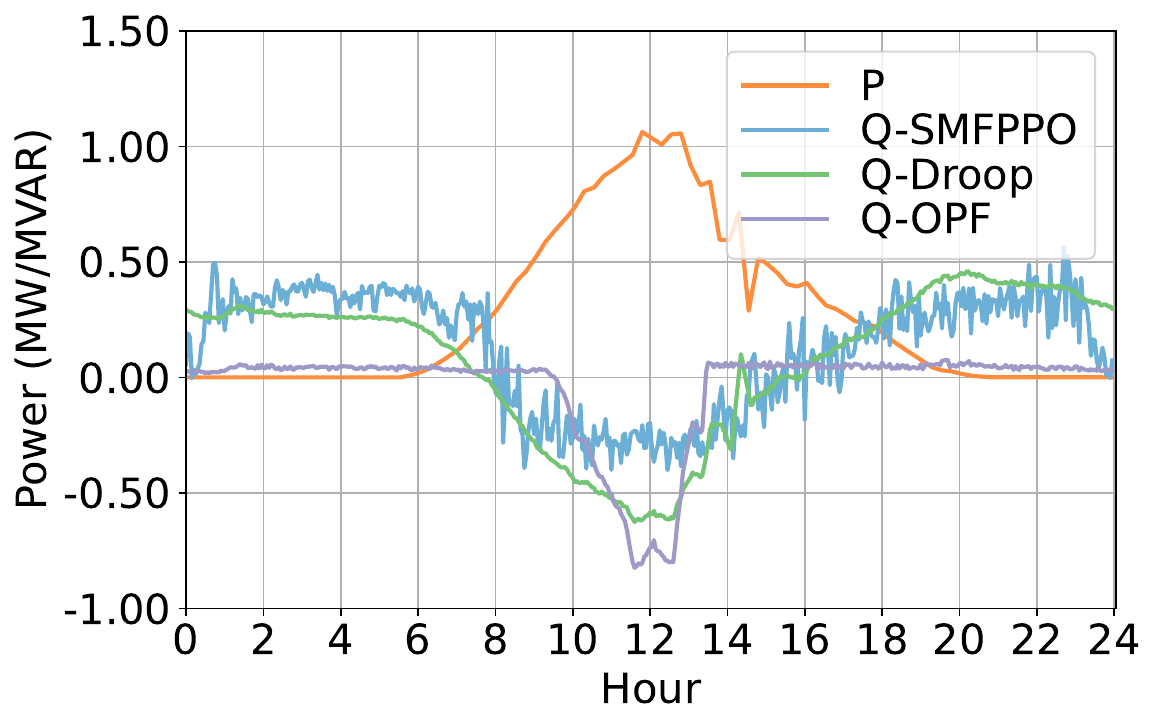}
            \end{subfigure}
            \quad
            \begin{subfigure}[b]{0.30\textwidth}
                \centering                
                \includegraphics[width=\textwidth]{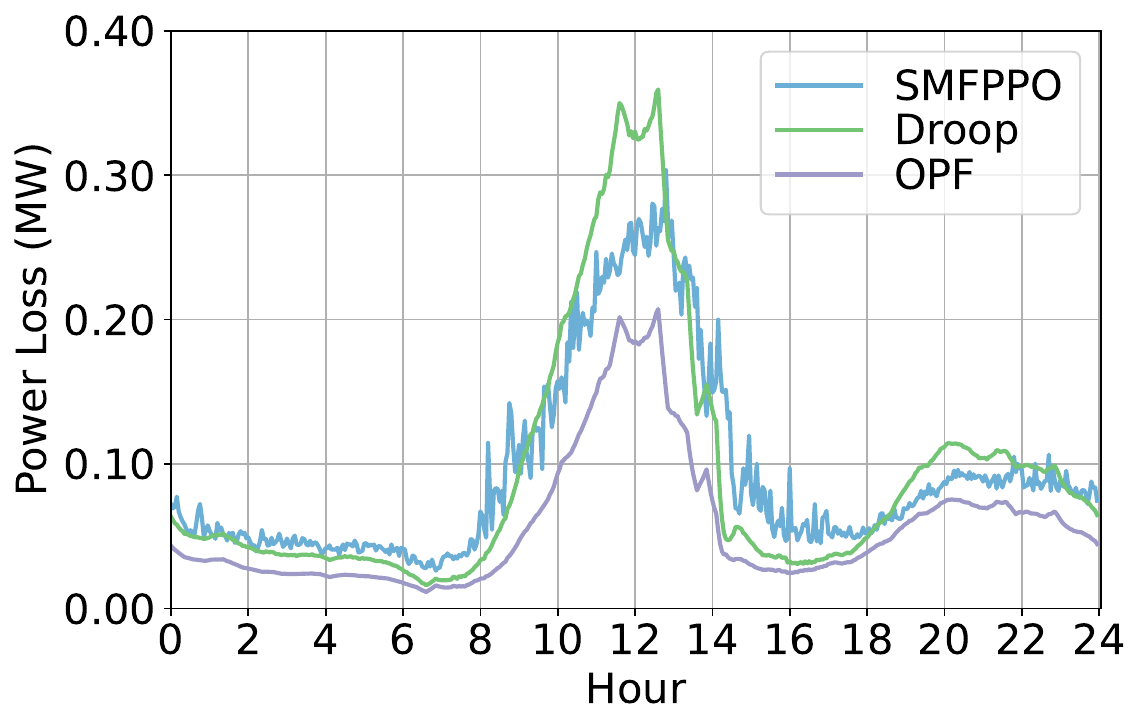}
            \end{subfigure}
            \quad
            \begin{subfigure}[b]{0.30\textwidth}
            	\centering
        	    \includegraphics[width=\textwidth]{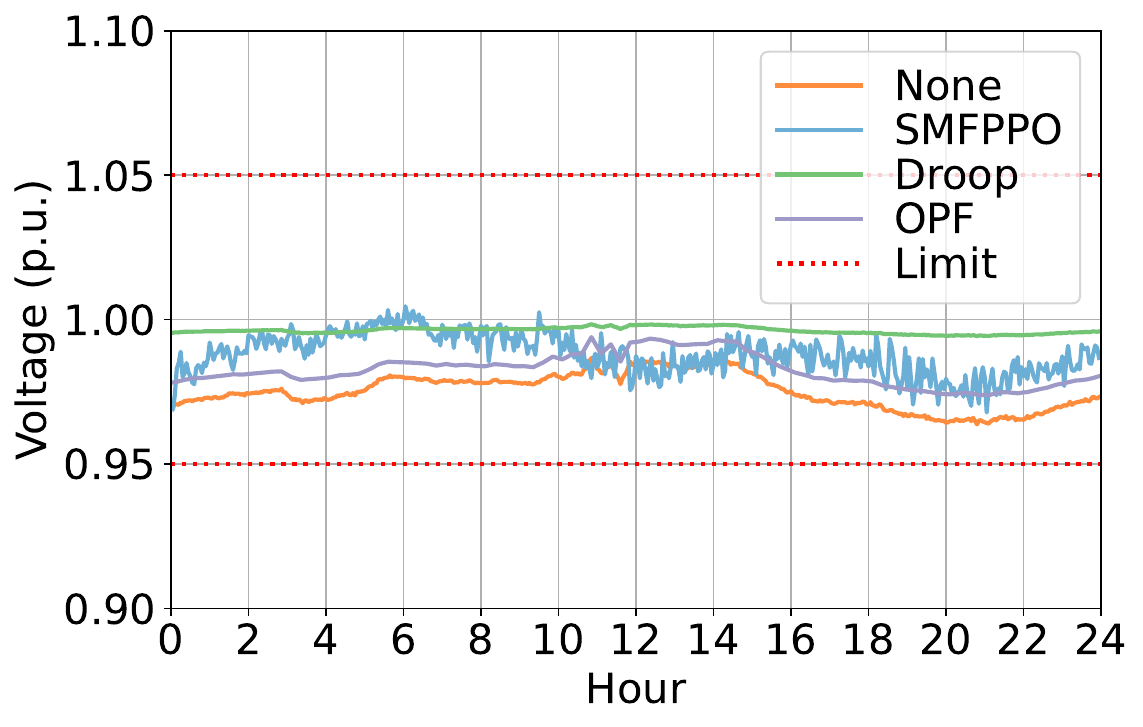}
        	    \caption{Voltage.}
            \end{subfigure}
            \quad
            \begin{subfigure}[b]{0.30\textwidth}
                \centering                
                \includegraphics[width=\textwidth]{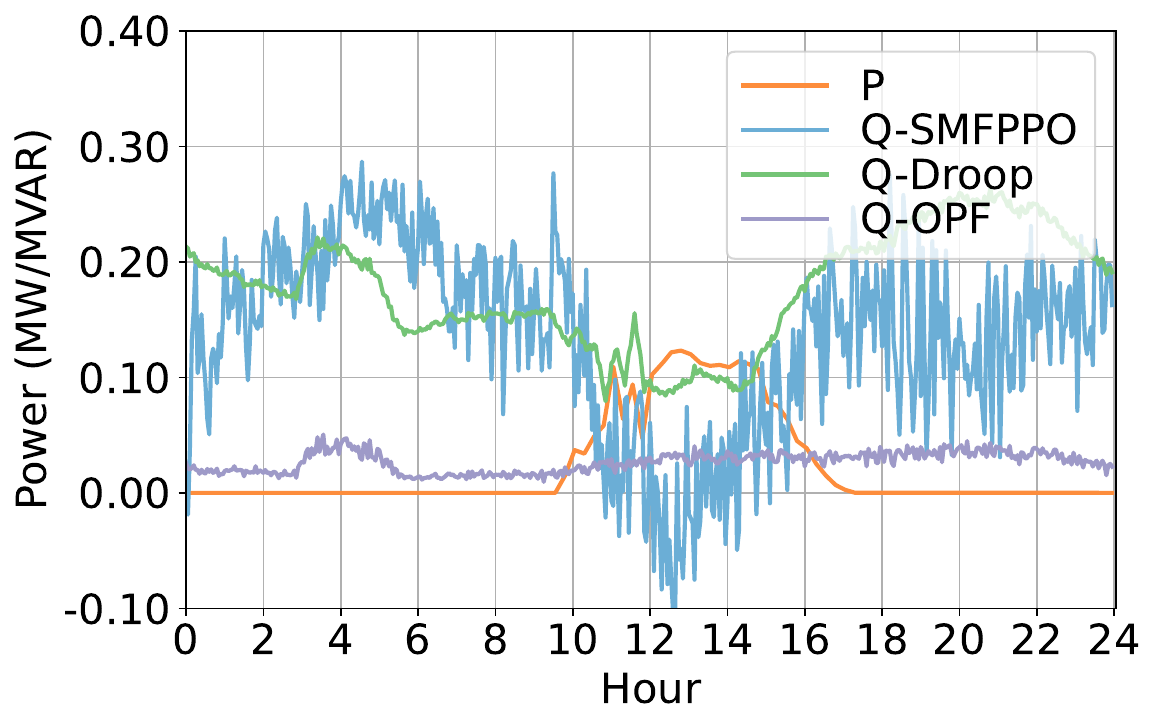}
                \caption{Power.}
            \end{subfigure}
            \quad
            \begin{subfigure}[b]{0.30\textwidth}
                \centering                
                \includegraphics[width=\textwidth]{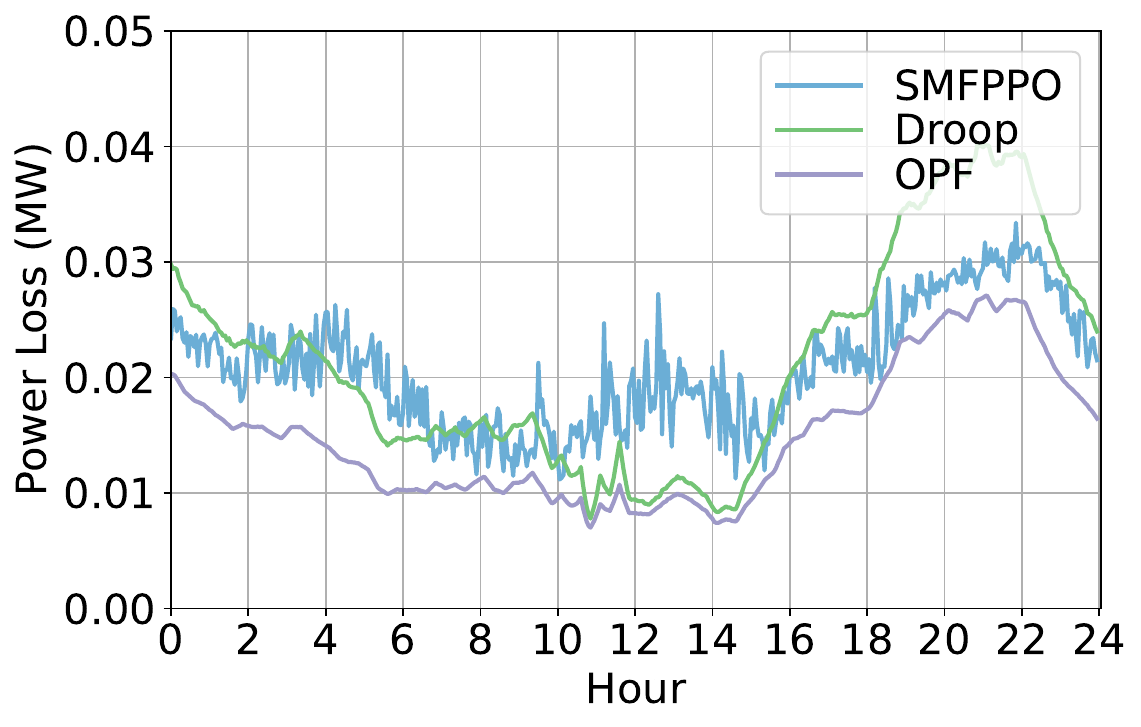}
                \caption{Power loss.}
            \end{subfigure}
            \caption{Comparing SMFPPO with traditional control methods on bus 18 during a day in the 33-bus network. 1st row: results of a summer day. 2nd row: results of a winter day. None and limit in (a) represent the voltage with no control and the safety voltage range respectively. P and Q in (b) indicate the PV active power and the reactive power by various methods.}
        \label{fig:case_study_33_smfppo}
        \end{figure*}
        
        \begin{figure*}[ht!]
            \centering
            \begin{subfigure}[b]{0.30\textwidth}
            	\centering
        	    \includegraphics[width=\textwidth]{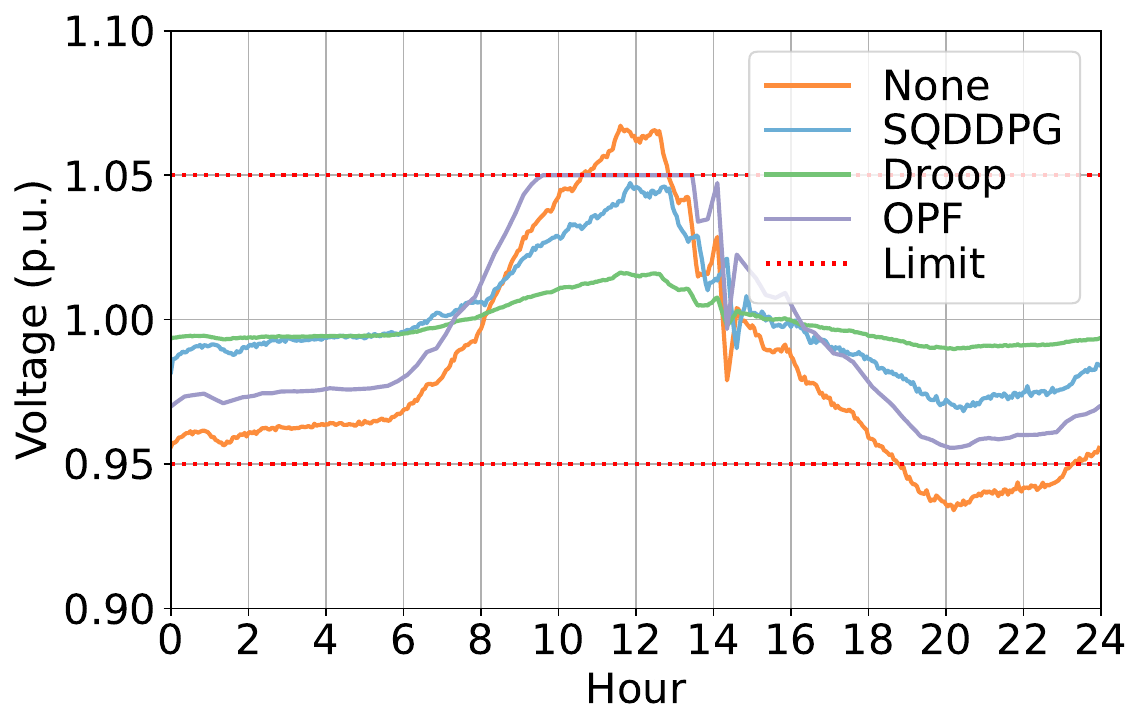}
            \end{subfigure}
            \quad
            \begin{subfigure}[b]{0.30\textwidth}
                \centering                
                \includegraphics[width=\textwidth]{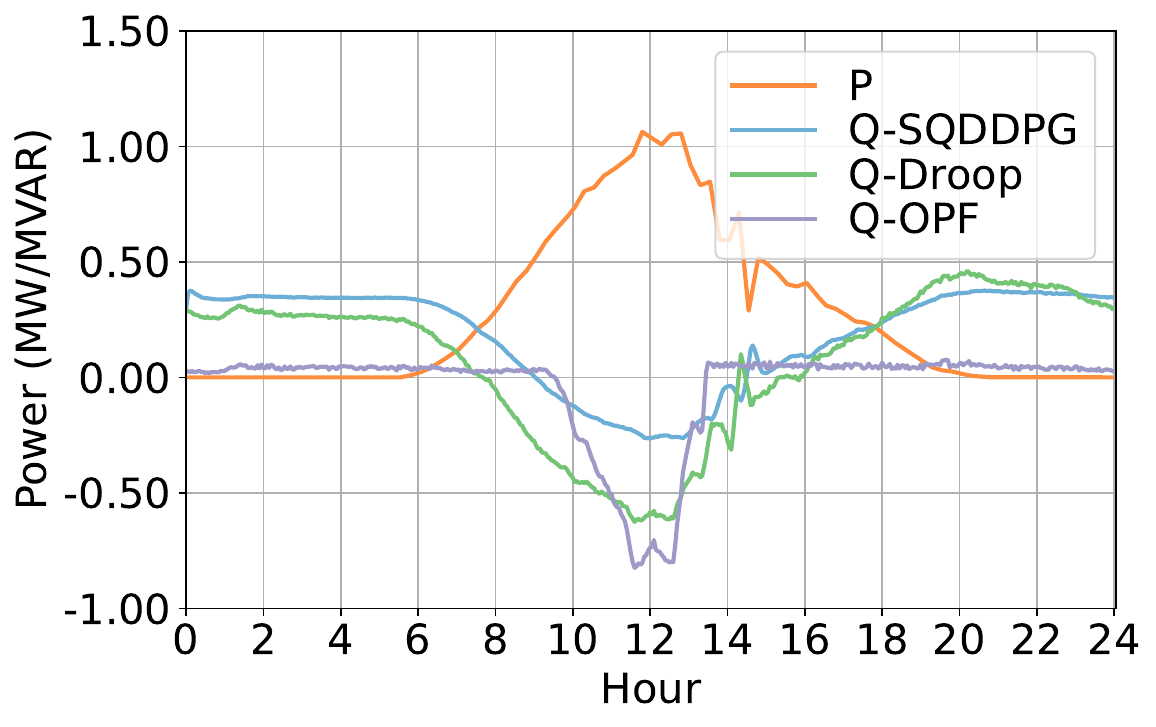}
            \end{subfigure}
            \quad
            \begin{subfigure}[b]{0.30\textwidth}
                \centering                
                \includegraphics[width=\textwidth]{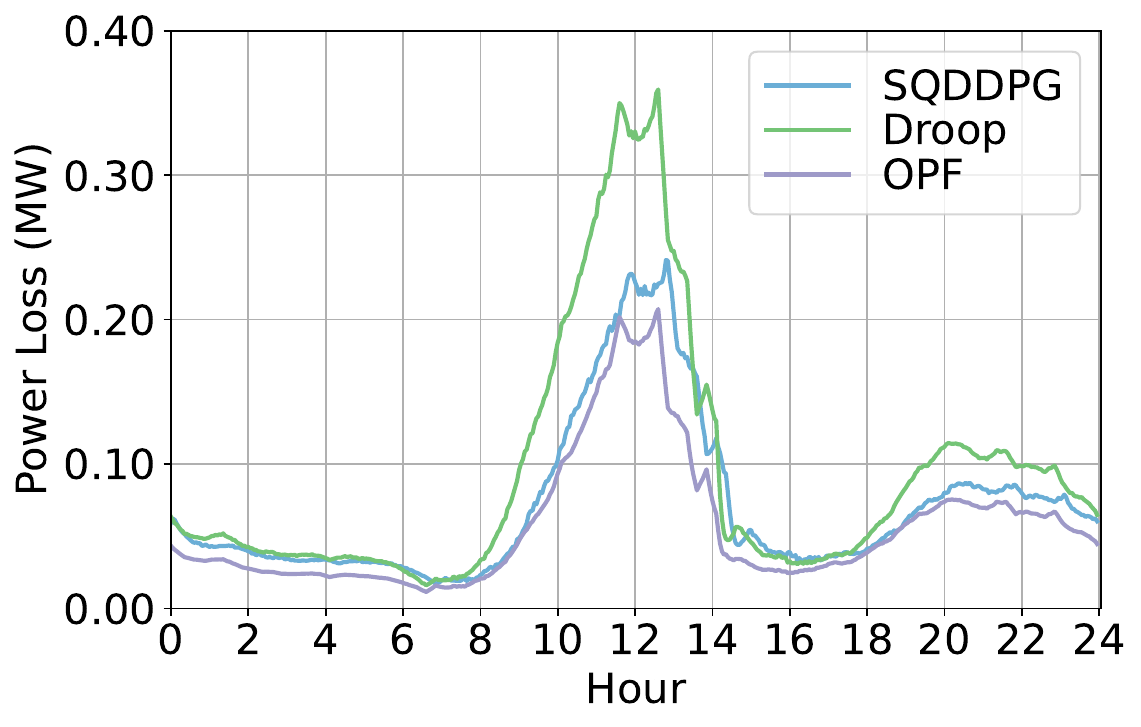}
            \end{subfigure}
            \quad
            \begin{subfigure}[b]{0.30\textwidth}
            	\centering
        	    \includegraphics[width=\textwidth]{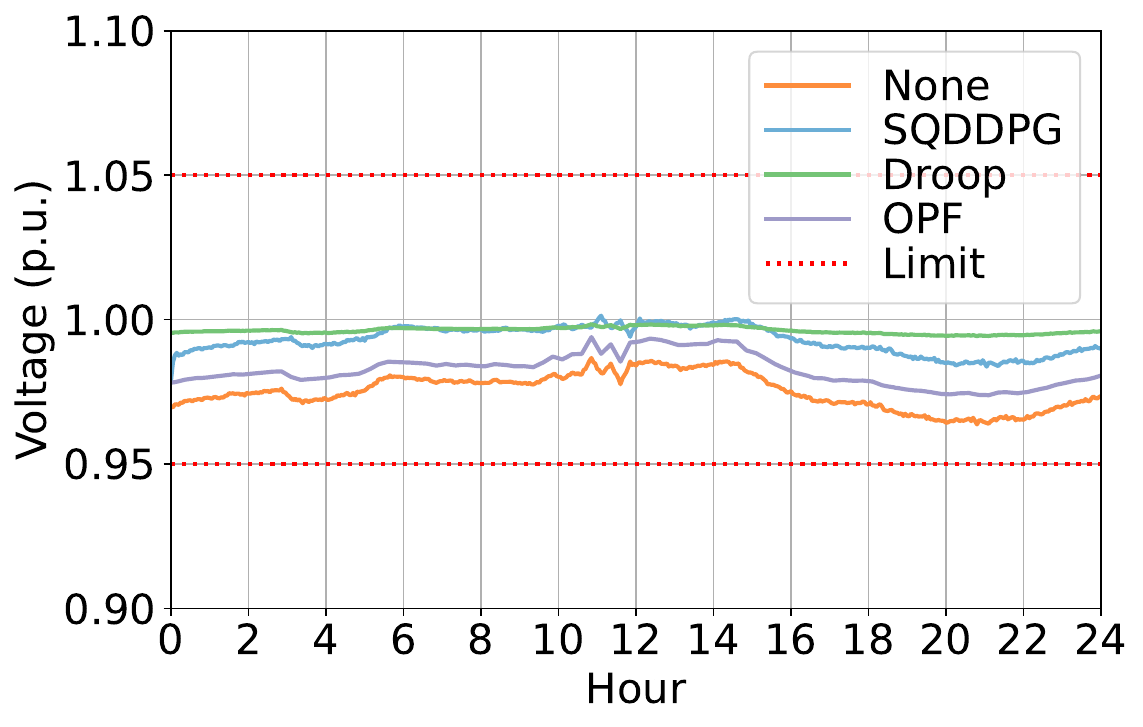}
        	    \caption{Voltage.}
            \end{subfigure}
            \quad
            \begin{subfigure}[b]{0.30\textwidth}
                \centering                
                \includegraphics[width=\textwidth]{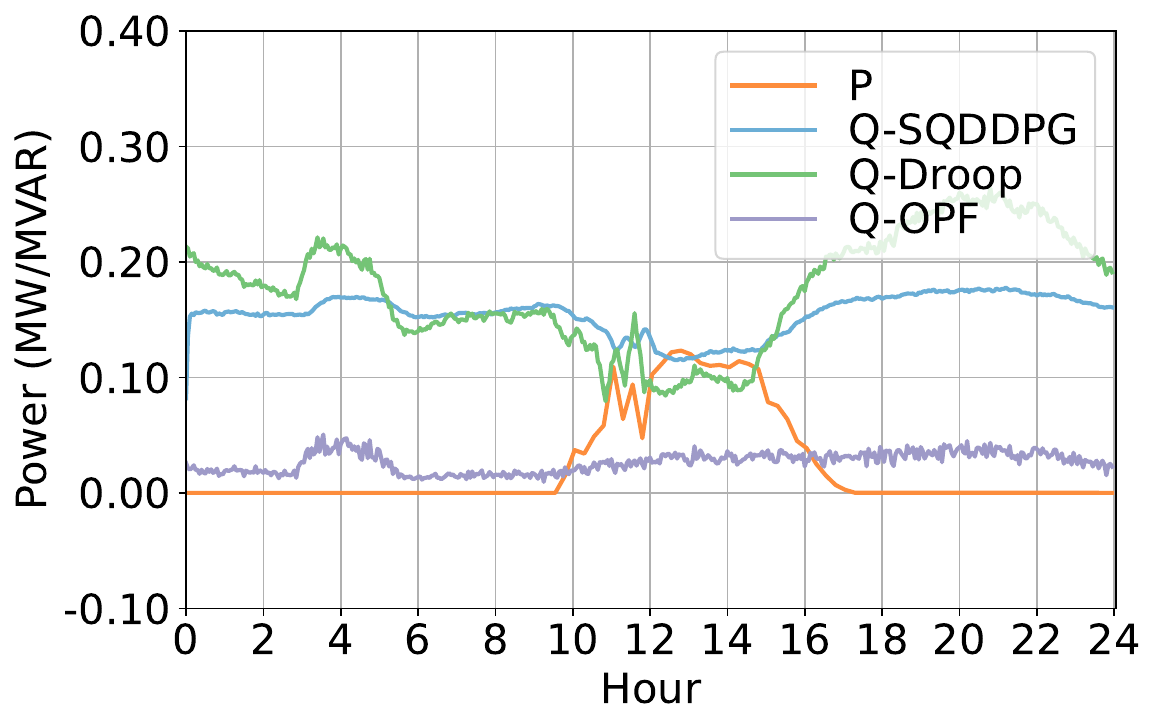}
                \caption{Power.}
            \end{subfigure}
            \quad
            \begin{subfigure}[b]{0.30\textwidth}
                \centering                
                \includegraphics[width=\textwidth]{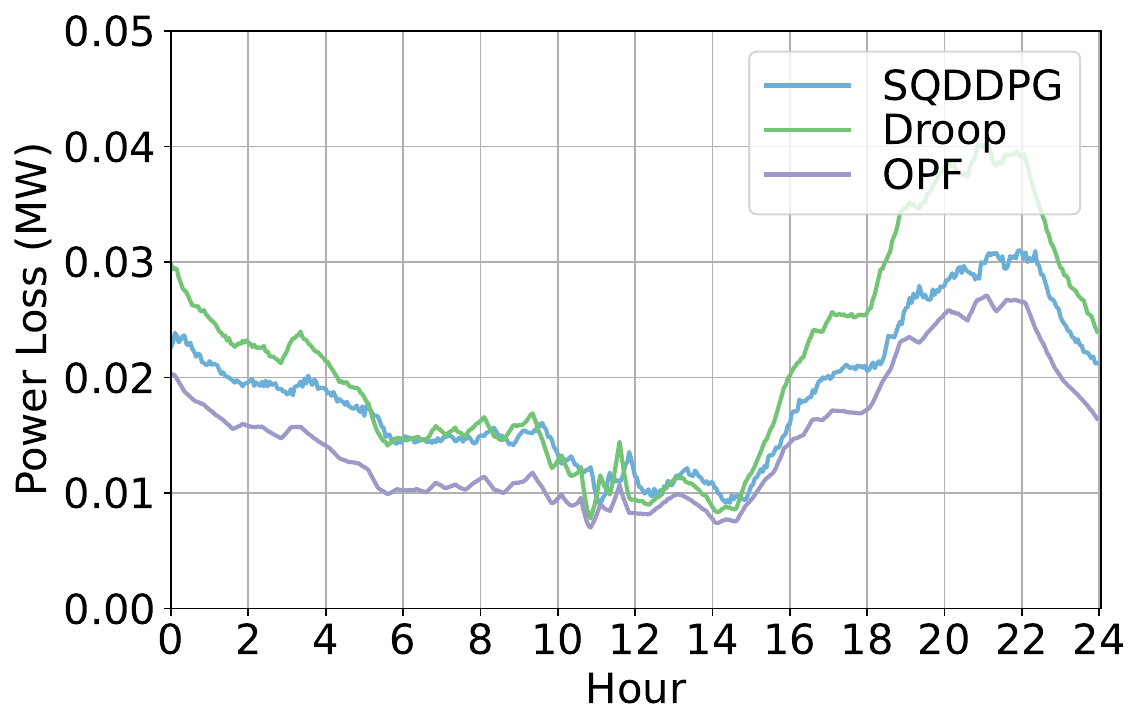}
                \caption{Power loss.}
            \end{subfigure}
            \caption{Comparing SQDDPG with traditional control methods on bus 18 during a day in the 33-bus network. 1st row: the results of a summer day. 2nd row: the results of a winter day. None and limit in (a) represent the voltage with no control and the safety voltage range respectively. P and Q in (b) indicate the PV active power and the reactive power by various methods.}
        \label{fig:case_study_33_sqddpg}
        \end{figure*}
        
        \paragraph{One Bus in the 141-Bus Network.} Figure \ref{fig:case_study_141_smfppo} and \ref{fig:case_study_141_sqddpg} show the results of a typical bus in the 141-bus network. In summer, all methods can control the voltage within the safety range. Nonetheless, the power loss of SMFPPO and SQDDPG is far lower than the droop control. In winter, all methods can still control the voltage within the safety range, however, SMFPPO and SQDDPG behave differently compared with the traditional control methods in generating the reactive power. For example, SMFPPO absorbs the reactive power in winter, which is different from the strategies of the droop control and the OPF, still yielding a low power loss.
        \begin{figure*}[ht!]
            \centering
            \begin{subfigure}[b]{0.30\textwidth}
            	\centering
        	    \includegraphics[width=\textwidth]{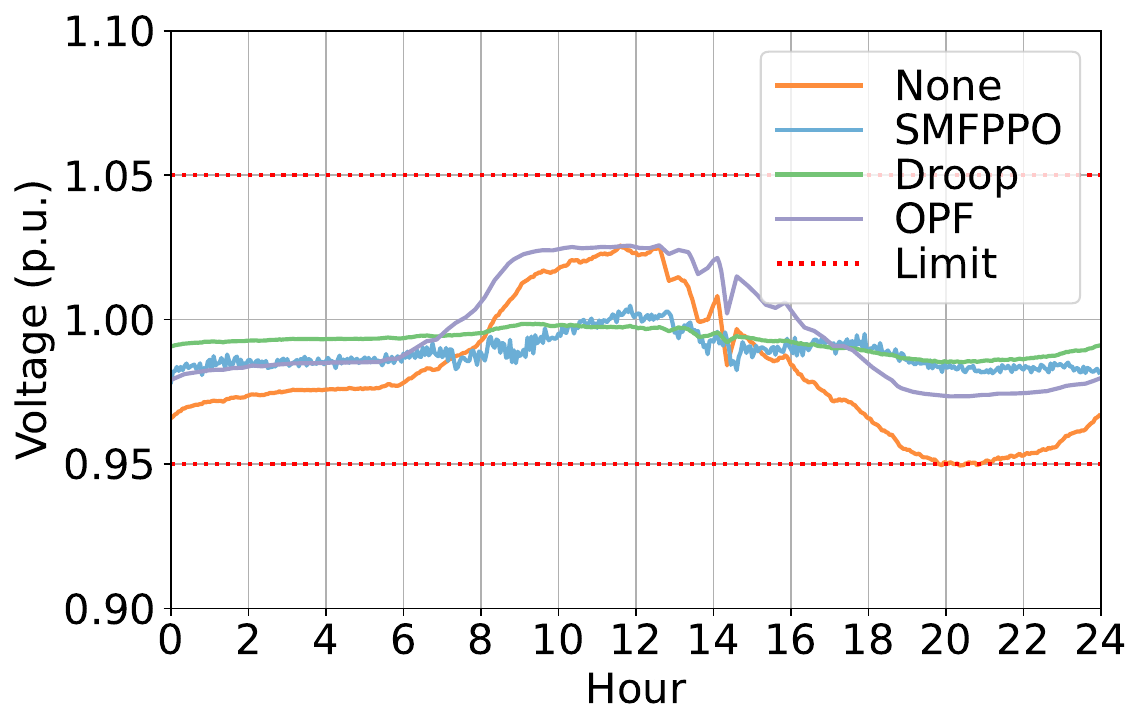}
            \end{subfigure}
            \quad
            \begin{subfigure}[b]{0.30\textwidth}
                \centering                
                \includegraphics[width=\textwidth]{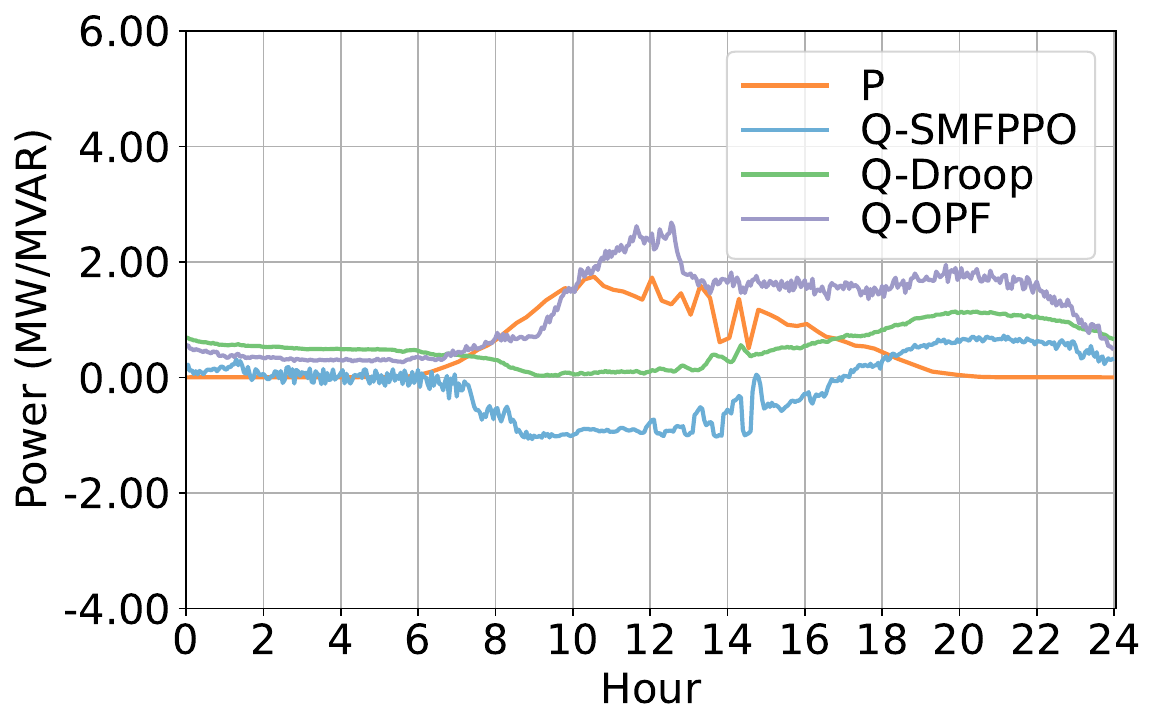}
            \end{subfigure}
            \quad
            \begin{subfigure}[b]{0.30\textwidth}
                \centering                
                \includegraphics[width=\textwidth]{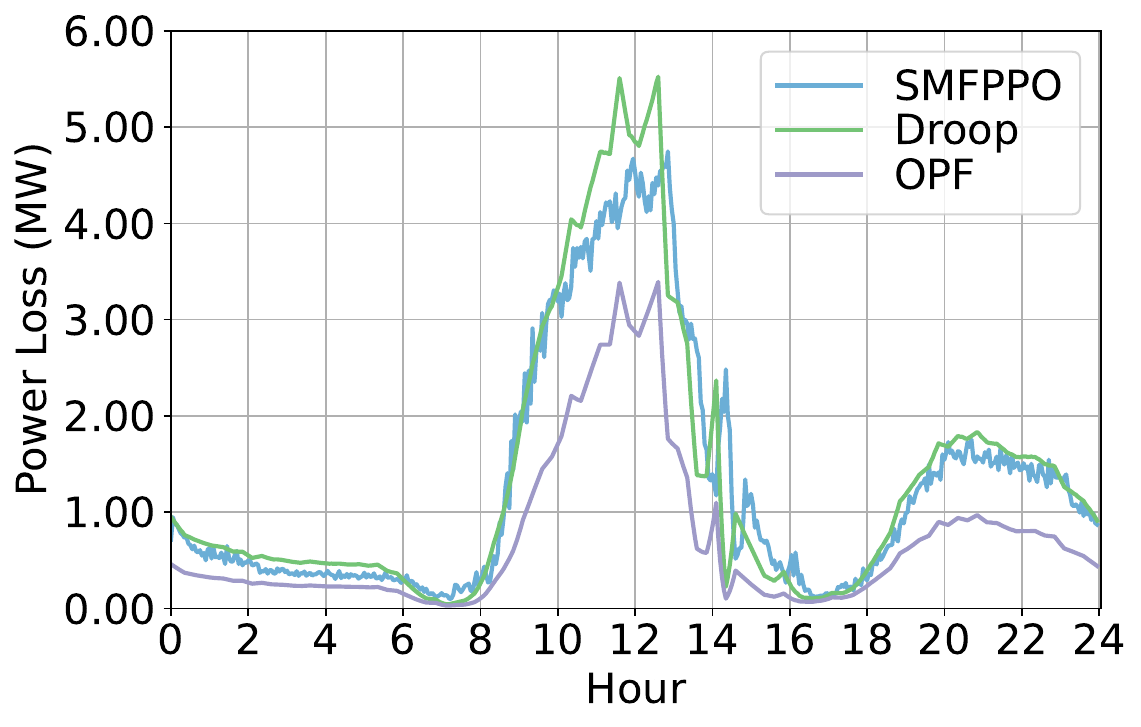}
            \end{subfigure}
            \quad
            \begin{subfigure}[b]{0.30\textwidth}
            	\centering
        	    \includegraphics[width=\textwidth]{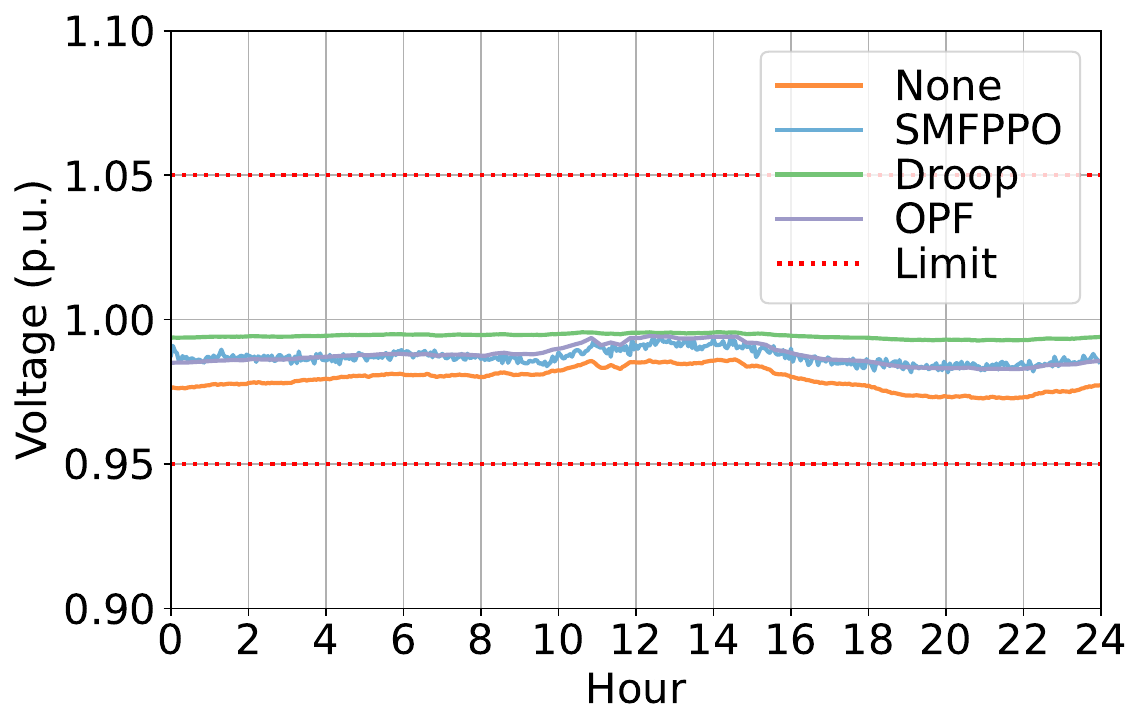}
        	    \caption{Voltage.}
            \end{subfigure}
            \quad
            \begin{subfigure}[b]{0.30\textwidth}
                \centering                
                \includegraphics[width=\textwidth]{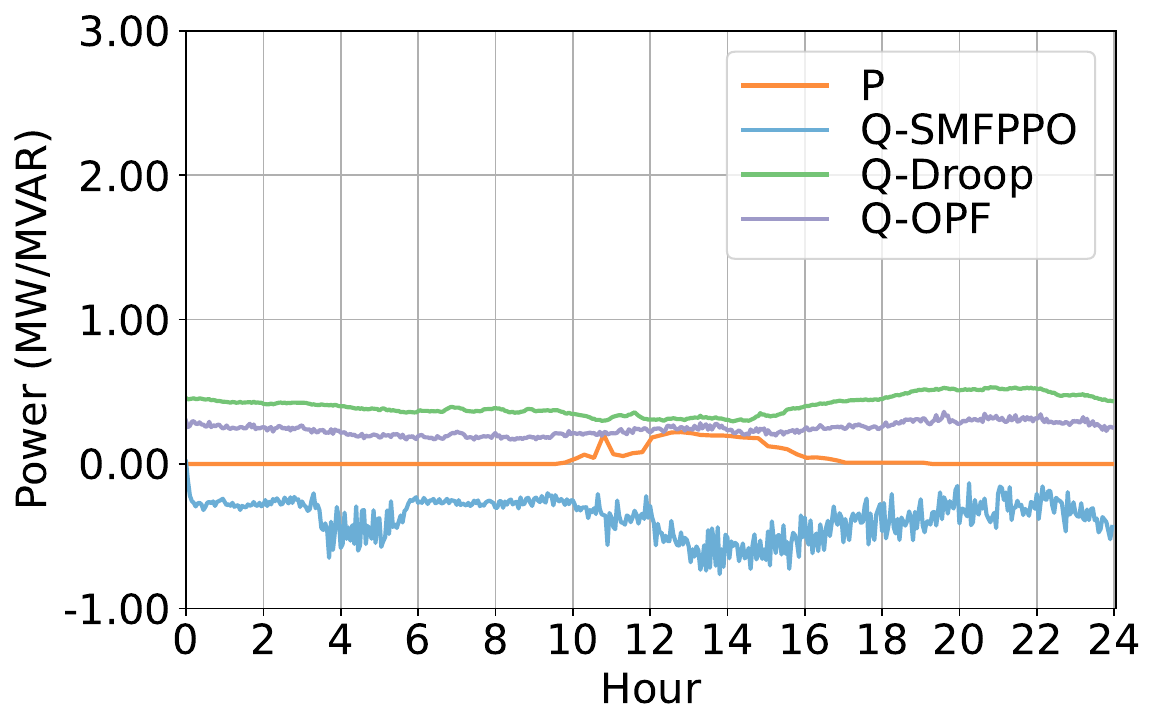}
                \caption{Power.}
            \end{subfigure}
            \quad
            \begin{subfigure}[b]{0.30\textwidth}
                \centering                
                \includegraphics[width=\textwidth]{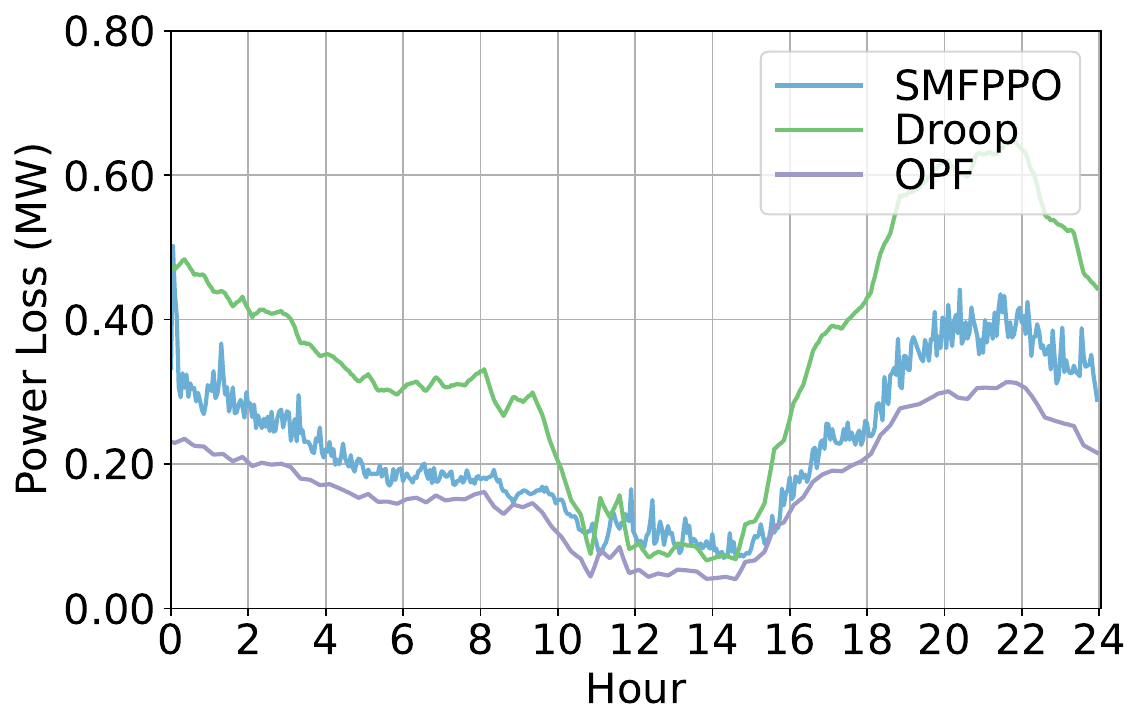}
                \caption{Power Loss.}
            \end{subfigure}
            \caption{Comparing SMFPPO with traditional control methods on a typical bus during a day in the 141-bus network. 1st row: the results of a summer day. 2nd row: the results of a winter day. None and limit in (a) represent the voltage with no control and the safety range respectively. P and Q in (b) indicate the PV active power and the reactive power by various methods.}
        \label{fig:case_study_141_smfppo}
        \end{figure*}
        
        \begin{figure*}[ht!]
            \centering
            \begin{subfigure}[b]{0.30\textwidth}
            	\centering
        	    \includegraphics[width=\textwidth]{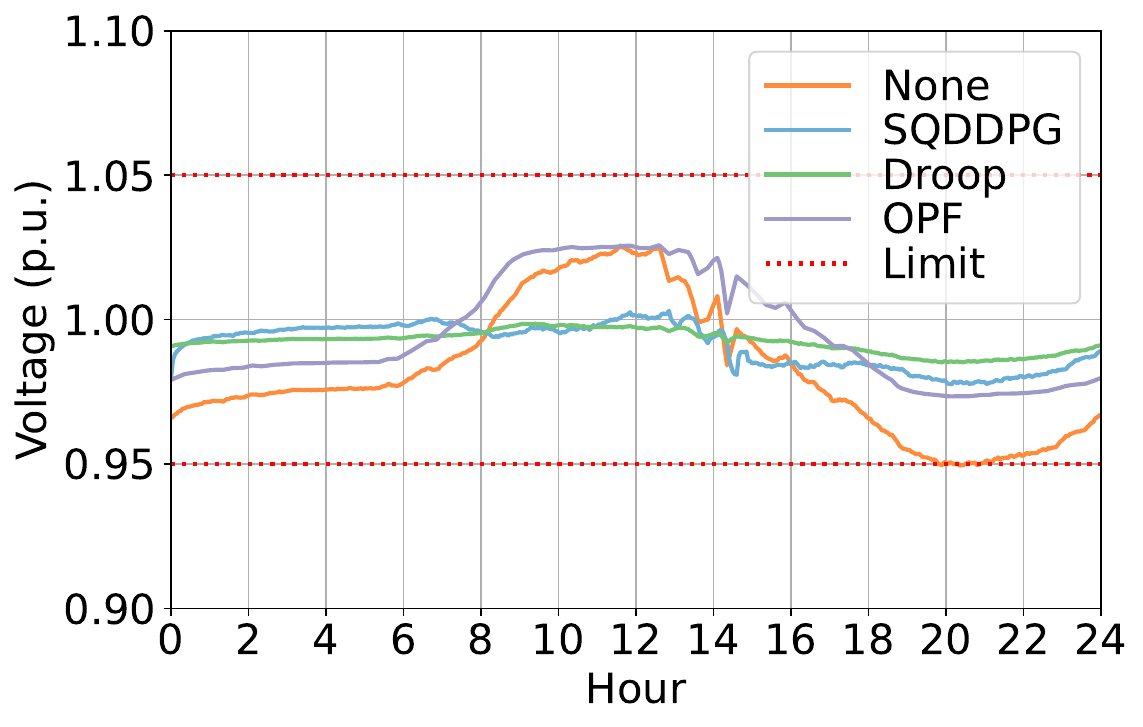}
            \end{subfigure}
            \quad
            \begin{subfigure}[b]{0.30\textwidth}
                \centering                
                \includegraphics[width=\textwidth]{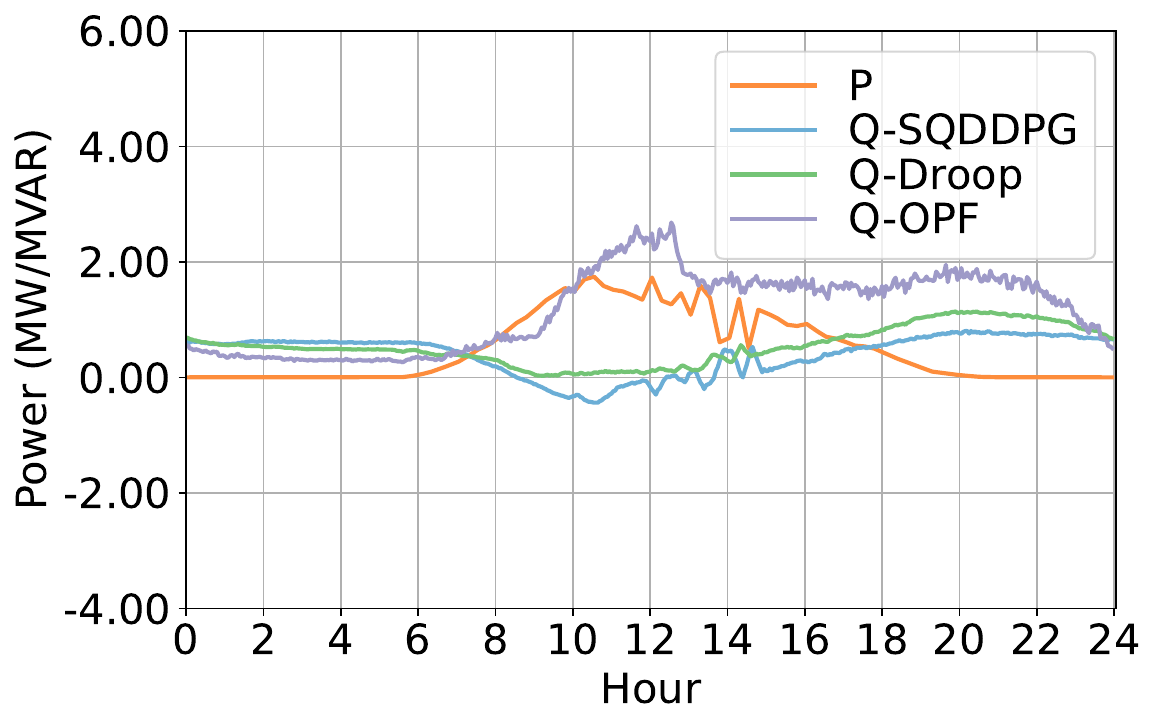}
            \end{subfigure}
            \quad
            \begin{subfigure}[b]{0.30\textwidth}
                \centering                
                \includegraphics[width=\textwidth]{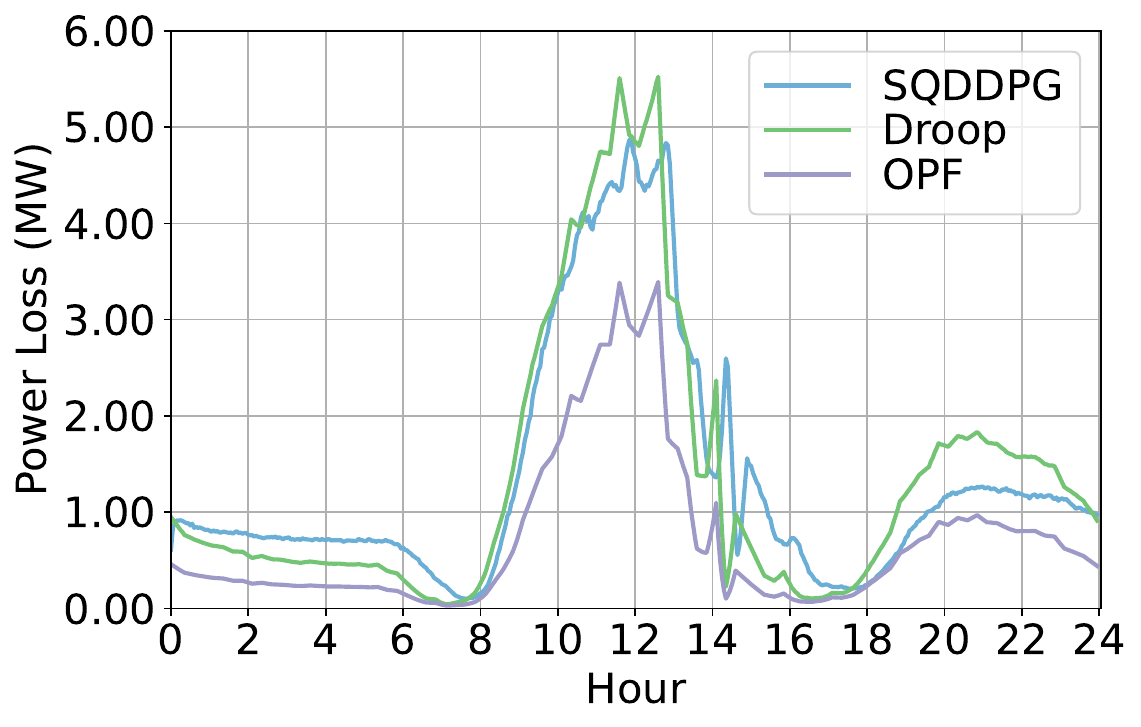}
            \end{subfigure}
            \quad
            \begin{subfigure}[b]{0.30\textwidth}
            	\centering
        	    \includegraphics[width=\textwidth]{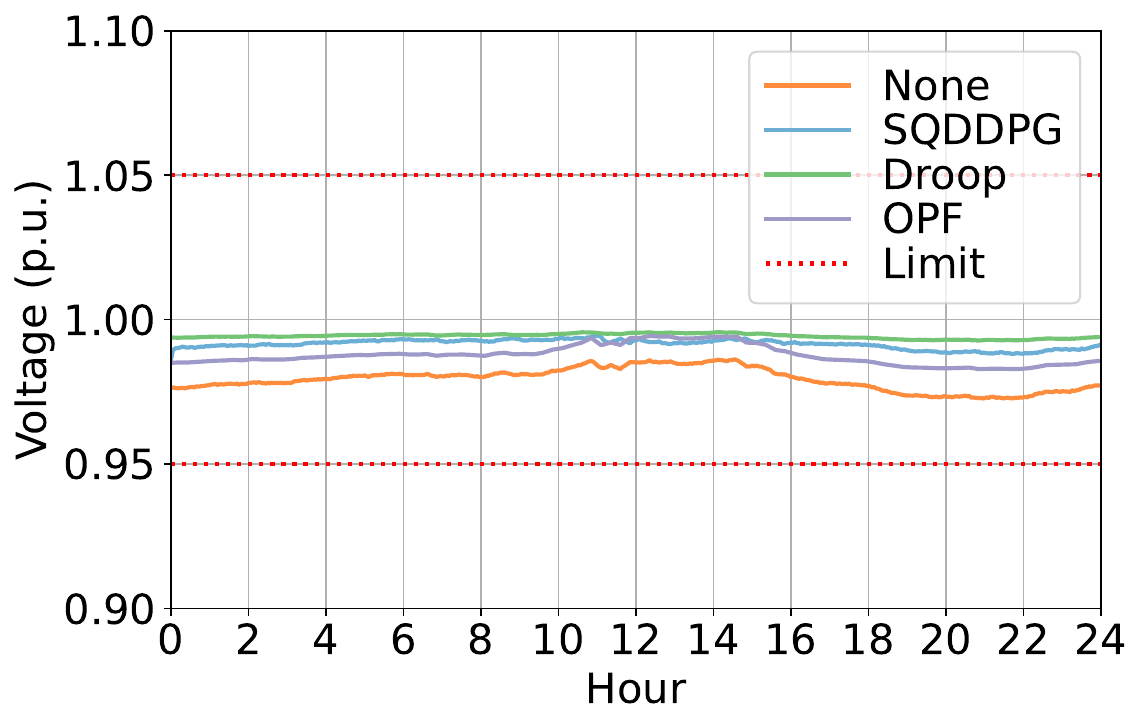}
        	    \caption{Voltage.}
            \end{subfigure}
            \quad
            \begin{subfigure}[b]{0.30\textwidth}
                \centering                
                \includegraphics[width=\textwidth]{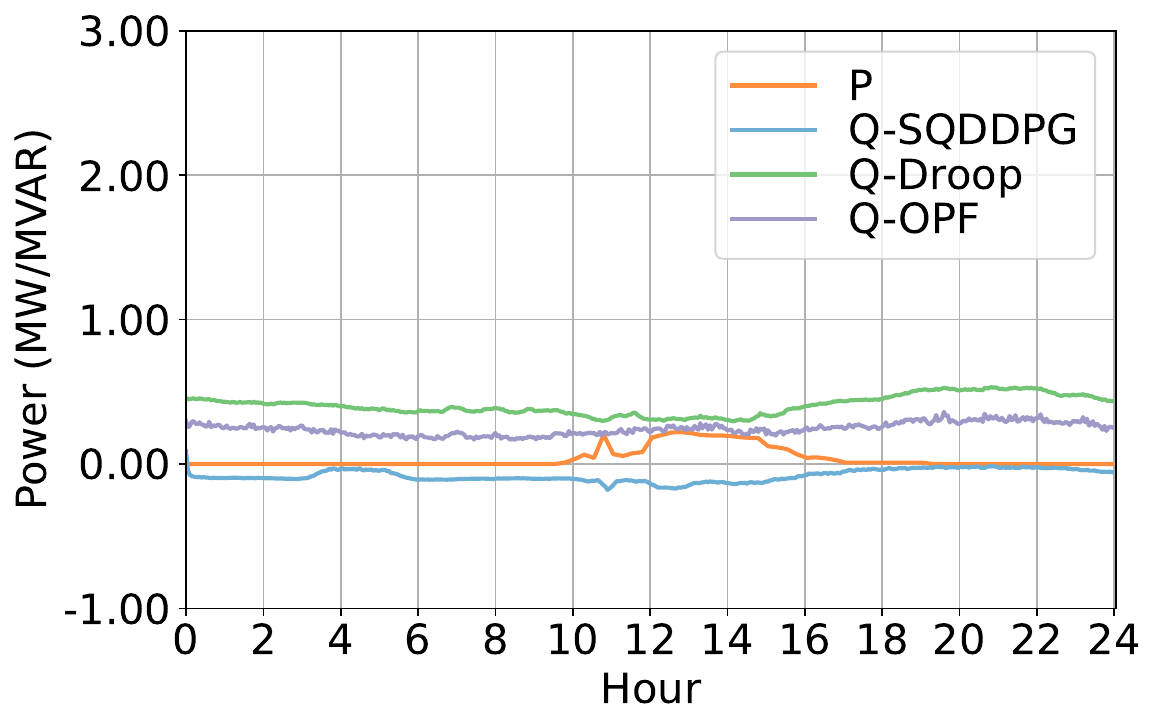}
                \caption{Power.}
            \end{subfigure}
            \quad
            \begin{subfigure}[b]{0.30\textwidth}
                \centering                
                \includegraphics[width=\textwidth]{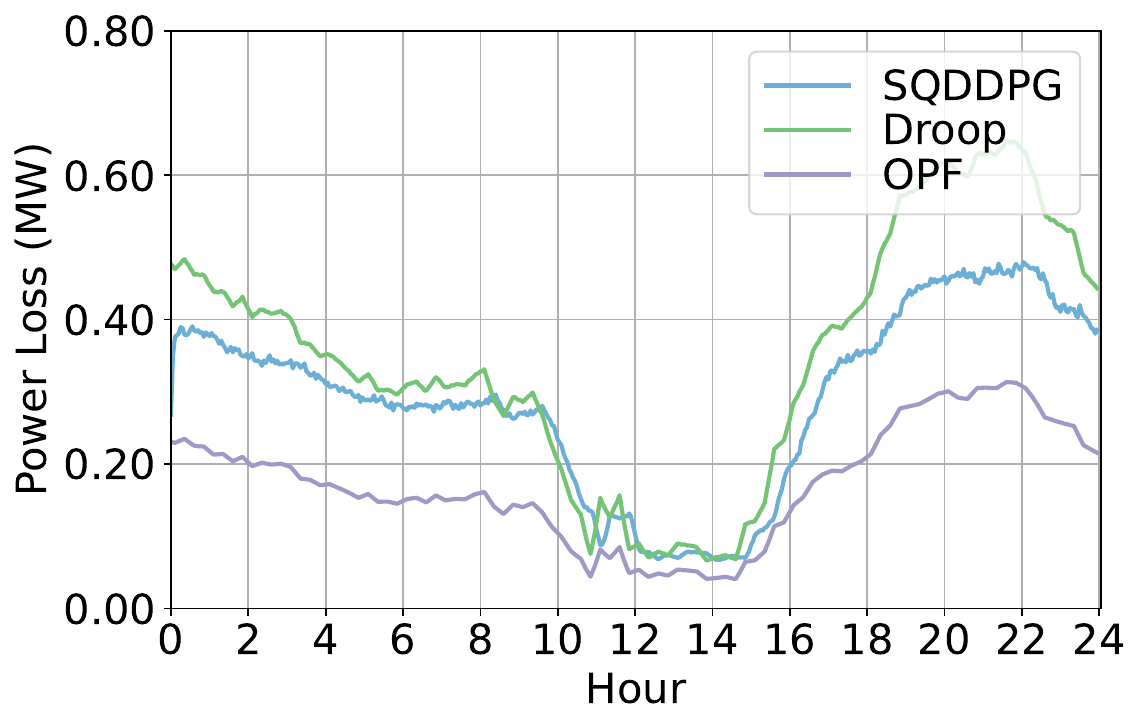}
                \caption{Power Loss.}
            \end{subfigure}
            \caption{Comparing SQDDPG with traditional control methods on a typical bus during a day in the 141-bus network. 1st row: the results of a summer day. 2nd row: the results of a winter day. None and limit in (a) represent the voltage with no control and the safety range respectively. P and Q in (b) indicate the PV active power and the reactive power by various methods.}
        \label{fig:case_study_141_sqddpg}
        \end{figure*}
        
        \paragraph{One Bus in the 322-Bus Network.} Figure \ref{fig:case_study_322_smfppo} and \ref{fig:case_study_322_sqddpg} show the results of a typical bus in the 322-bus network. In summer, it can be seen that SMFPPO can well control the voltage within the safety range while SQDDPG cannot. Additionally, the droop control can still control the voltage within the safety range, whereas the OPF slightly exceeds the lower limit. The inferior performance of OPF is perhaps due to the reason that the 322-bus network is so large and complicated that it may suffer the computational catastrophe with respect to the inverse of the topological matrix. In winter, all methods can control the voltage within the safety range, though the voltage at this time is originally within the safety range without no control. It can be observed that SQDDPG cannot generate the accurate power loss (i.e. the low power loss with the insufficient reactive power or the high power loss with the excessive reactive power). This is probably due to the incorrect credit assignment led by the direct approximation of marginal contribution in the implementation of SQDDPG (see Section \ref{subsec:problem_of_direct_approximation_of_cmc}) and therefore it leads to the difficulty of converging to the optimal joint policy. In contrast, SMFPPO generates comparatively more accurate reactive power and therefore the more appropriate power loss. It is not difficult to see that the droop control is a competitive baseline, however, its drawback is that the extra inner-loop optimization at each timestep is needed.
        \begin{figure*}[ht!]
            \centering
            \begin{subfigure}[b]{0.30\textwidth}
            	\centering
        	    \includegraphics[width=\textwidth]{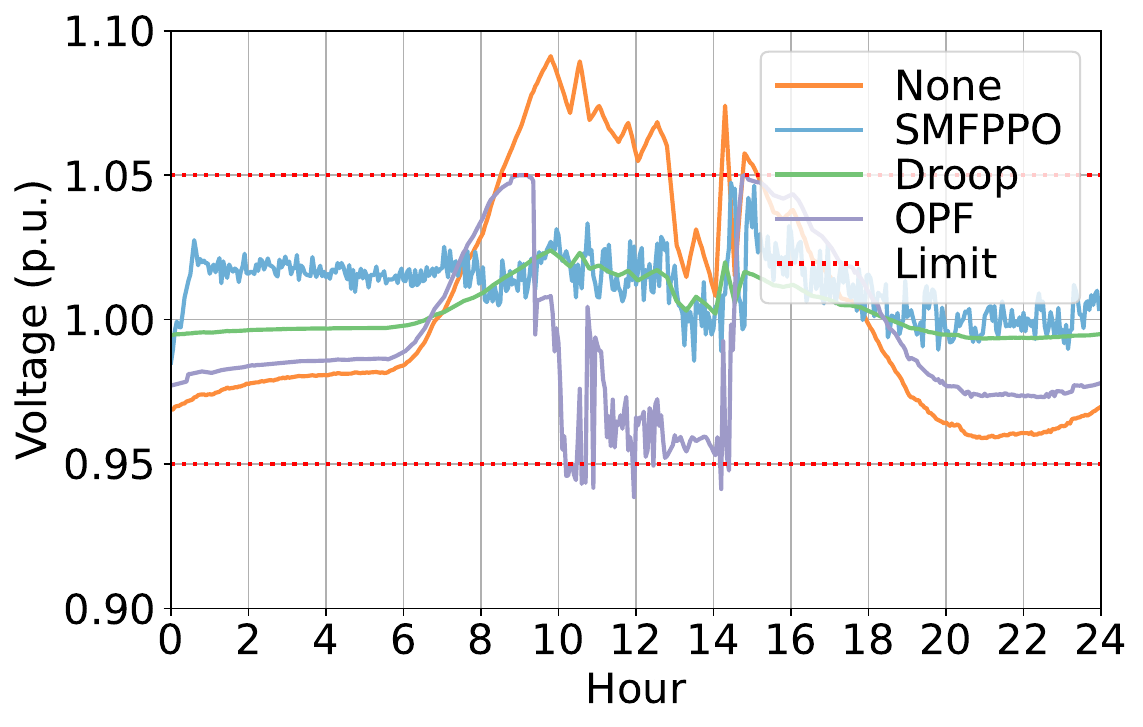}
            \end{subfigure}
            \quad
            \begin{subfigure}[b]{0.30\textwidth}
                \centering                
                \includegraphics[width=\textwidth]{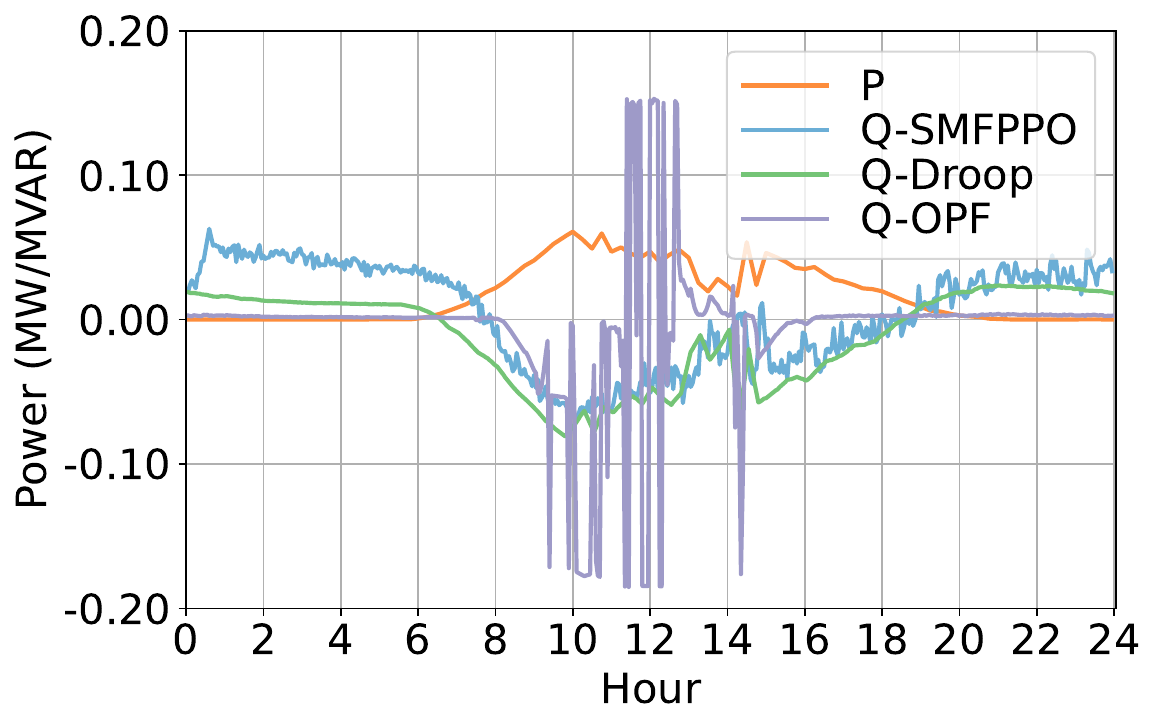}
            \end{subfigure}
            \quad
            \begin{subfigure}[b]{0.30\textwidth}
                \centering                
                \includegraphics[width=\textwidth]{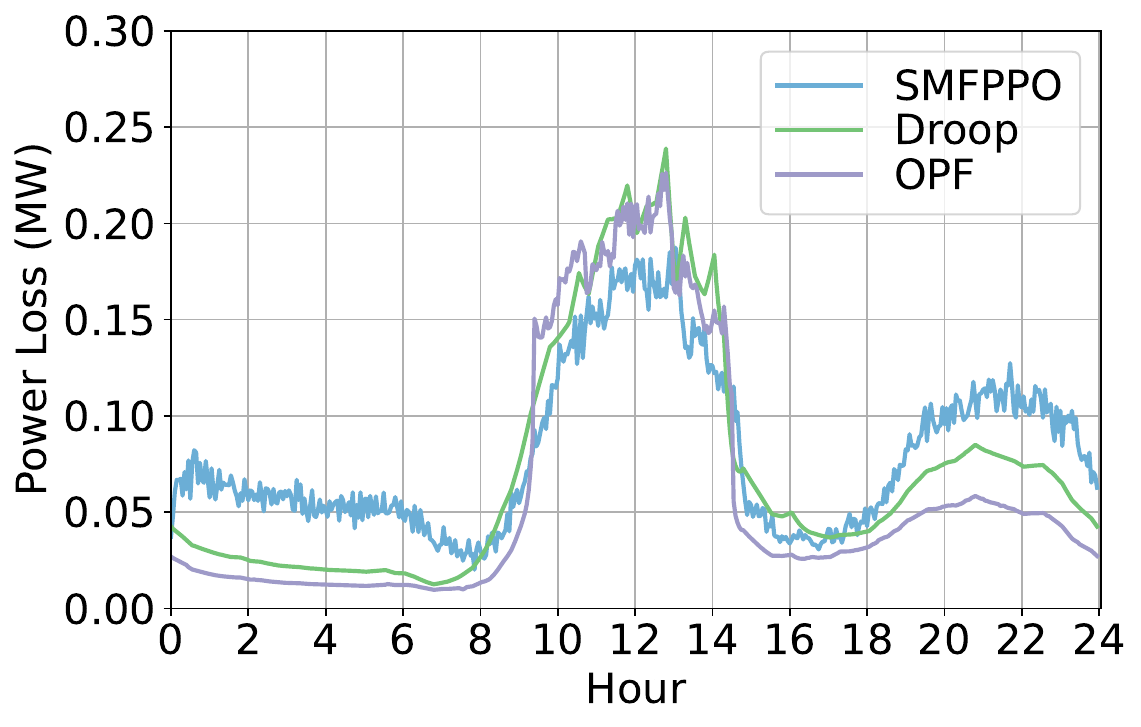}
            \end{subfigure}
            \quad
            \begin{subfigure}[b]{0.30\textwidth}
            	\centering
        	    \includegraphics[width=\textwidth]{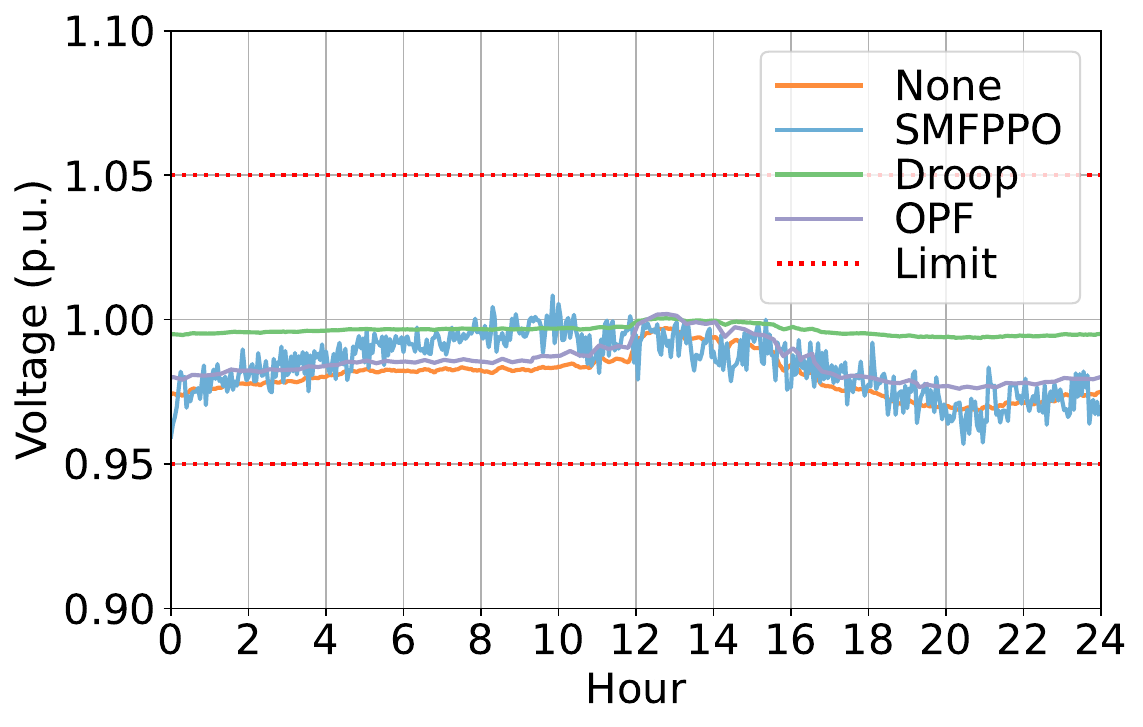}
        	    \caption{Voltage.}
            \end{subfigure}
            \quad
            \begin{subfigure}[b]{0.30\textwidth}
                \centering                
                \includegraphics[width=\textwidth]{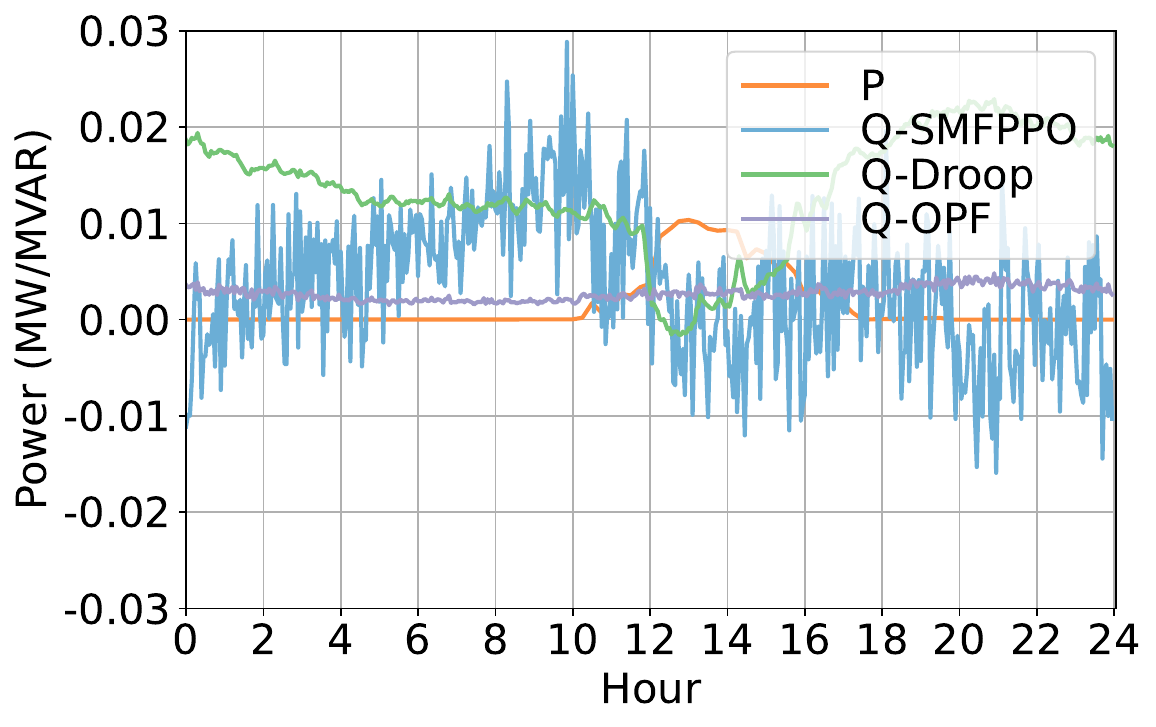}
                \caption{Power.}
            \end{subfigure}
            \quad
            \begin{subfigure}[b]{0.30\textwidth}
                \centering                
                \includegraphics[width=\textwidth]{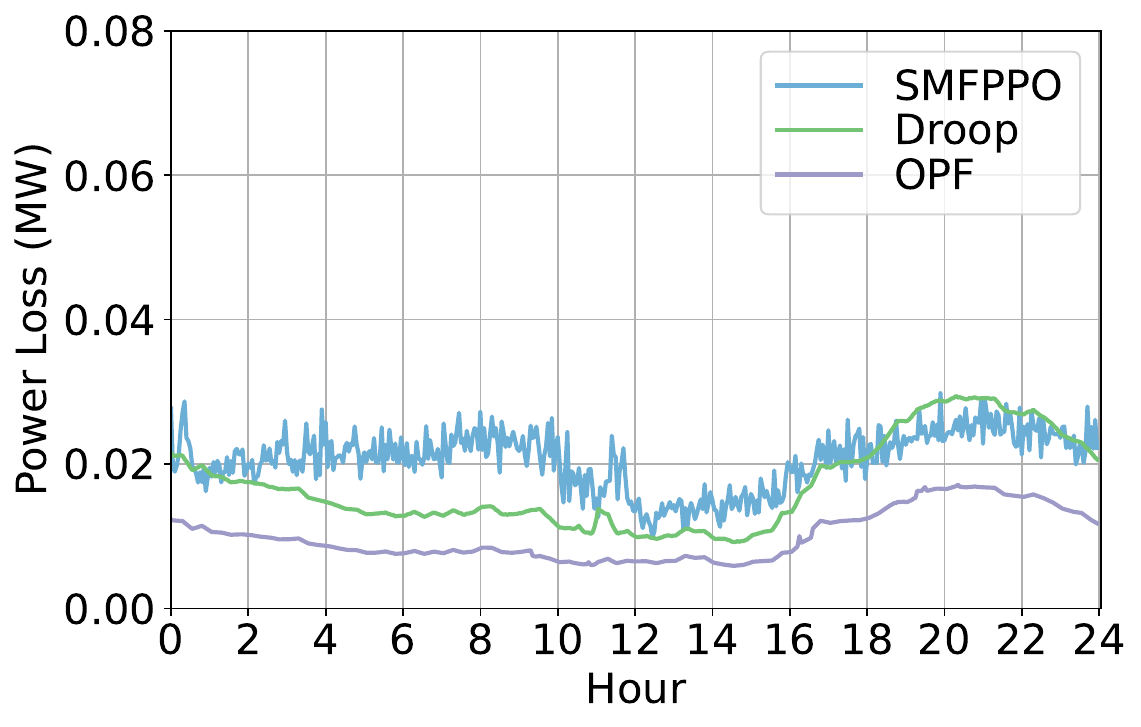}
                \caption{Power Loss.}
            \end{subfigure}
            \caption{Comparing SMFPPO with traditional control methods on a typical bus during a day in the 322-bus network. 1st row: the results of a summer day. 2nd row: the results of a winter day. None and limit in (a) represent the voltage with no control and the safety range respectively. P and Q in (b) indicate the PV active power and the reactive power by various methods.}
        \label{fig:case_study_322_smfppo}
        \end{figure*}
        
        \begin{figure*}[ht!]
            \centering
            \begin{subfigure}[b]{0.30\textwidth}
            	\centering
        	    \includegraphics[width=\textwidth]{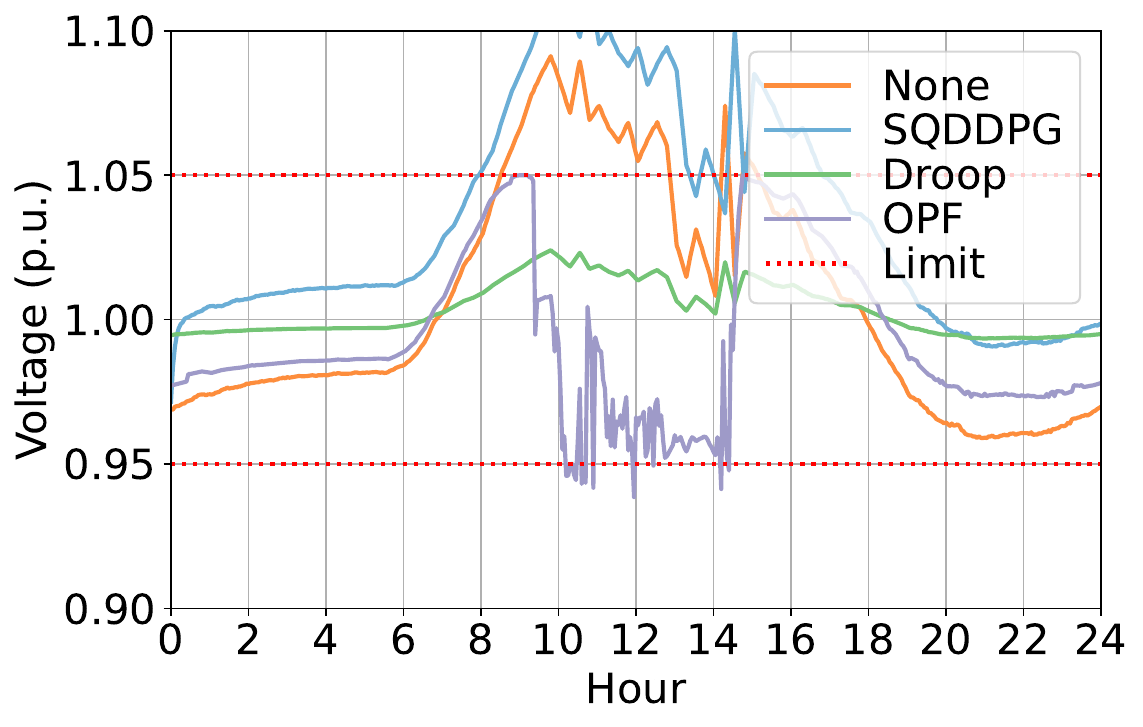}
            \end{subfigure}
            \quad
            \begin{subfigure}[b]{0.30\textwidth}
                \centering                
                \includegraphics[width=\textwidth]{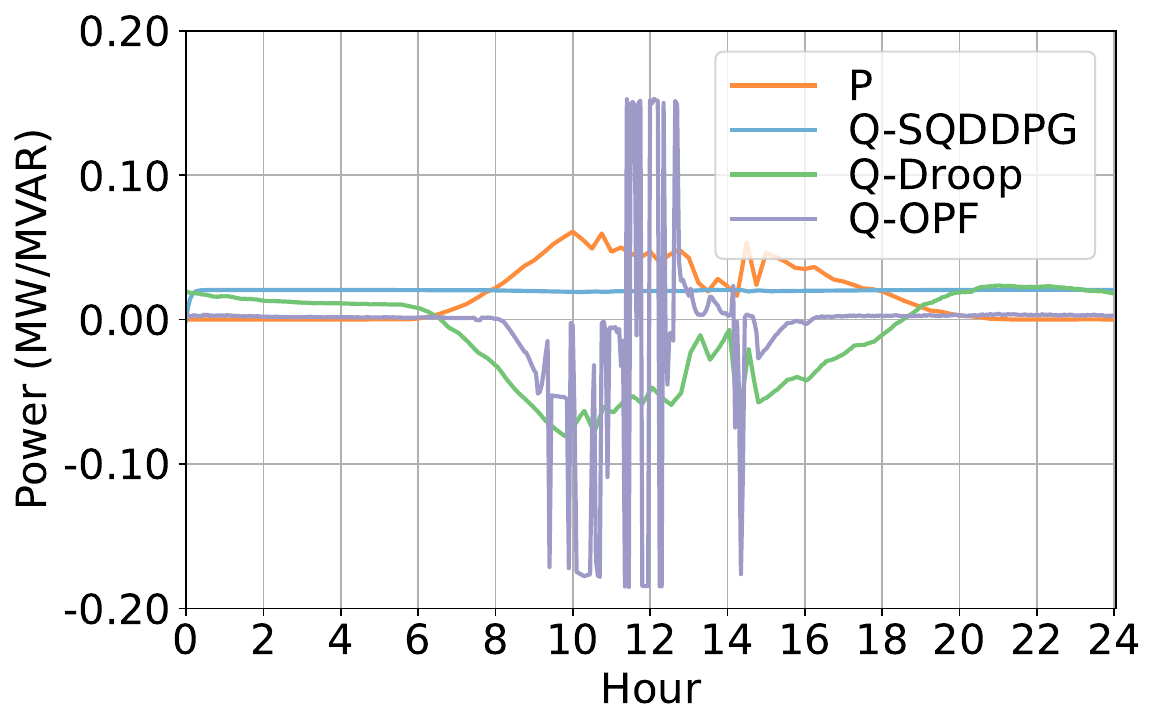}
            \end{subfigure}
            \quad
            \begin{subfigure}[b]{0.30\textwidth}
                \centering                
                \includegraphics[width=\textwidth]{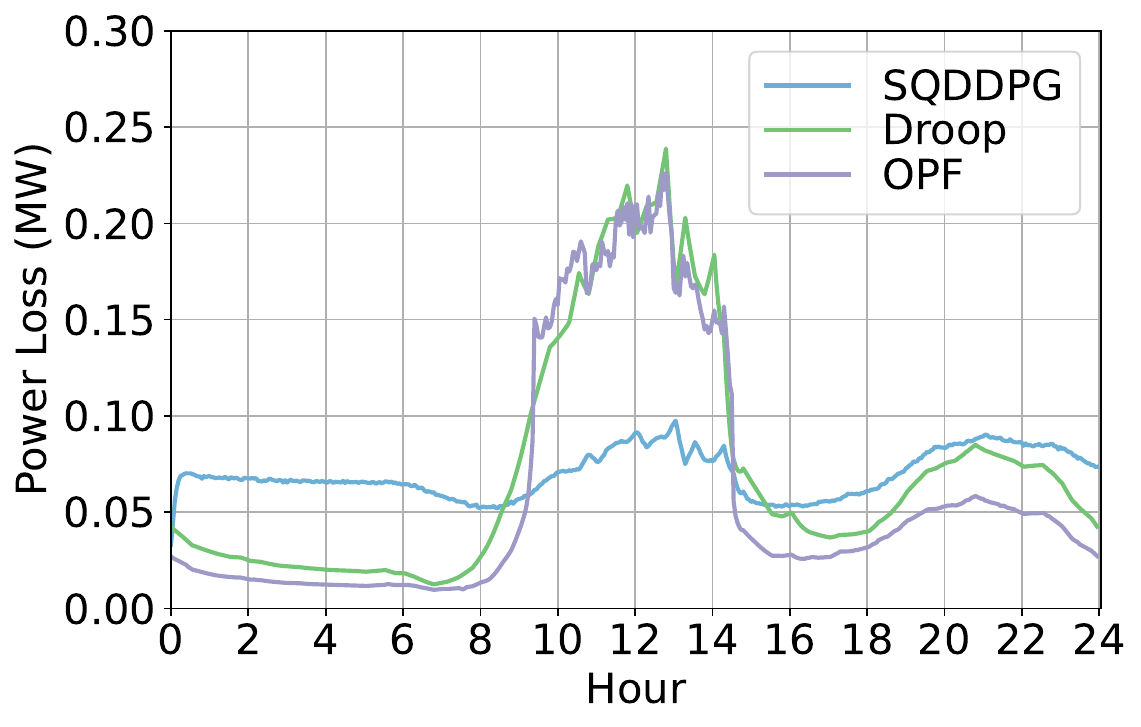}
            \end{subfigure}
            \quad
            \begin{subfigure}[b]{0.30\textwidth}
            	\centering
        	    \includegraphics[width=\textwidth]{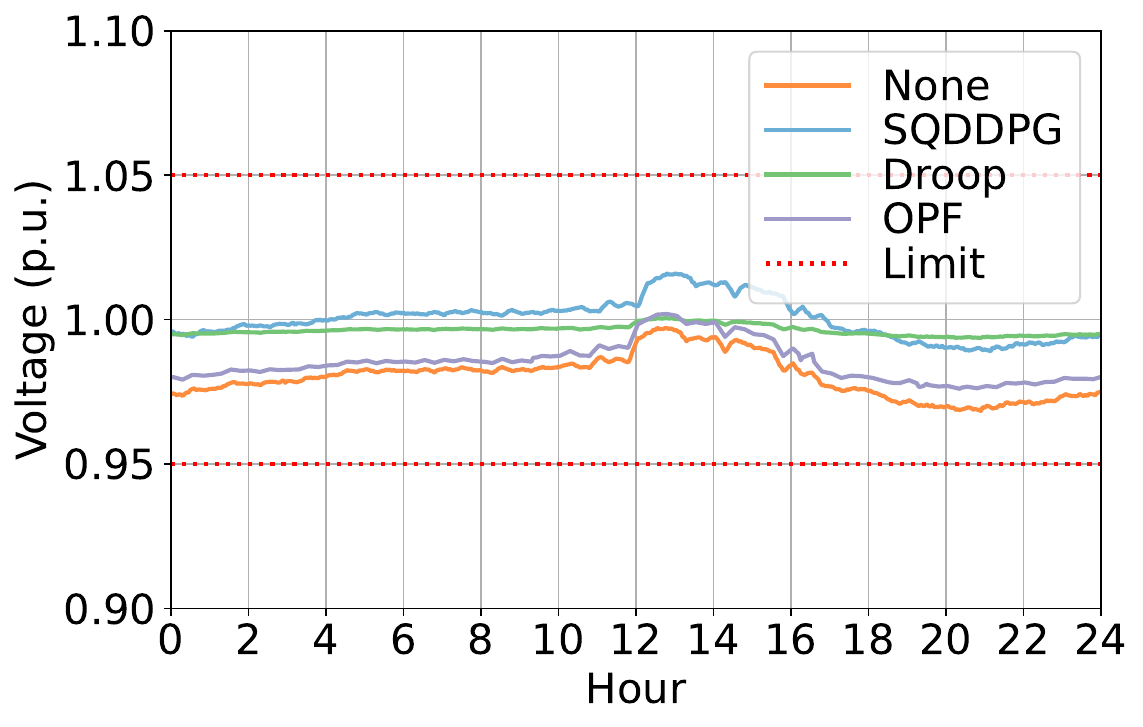}
        	    \caption{Voltage.}
            \end{subfigure}
            \quad
            \begin{subfigure}[b]{0.30\textwidth}
                \centering                
                \includegraphics[width=\textwidth]{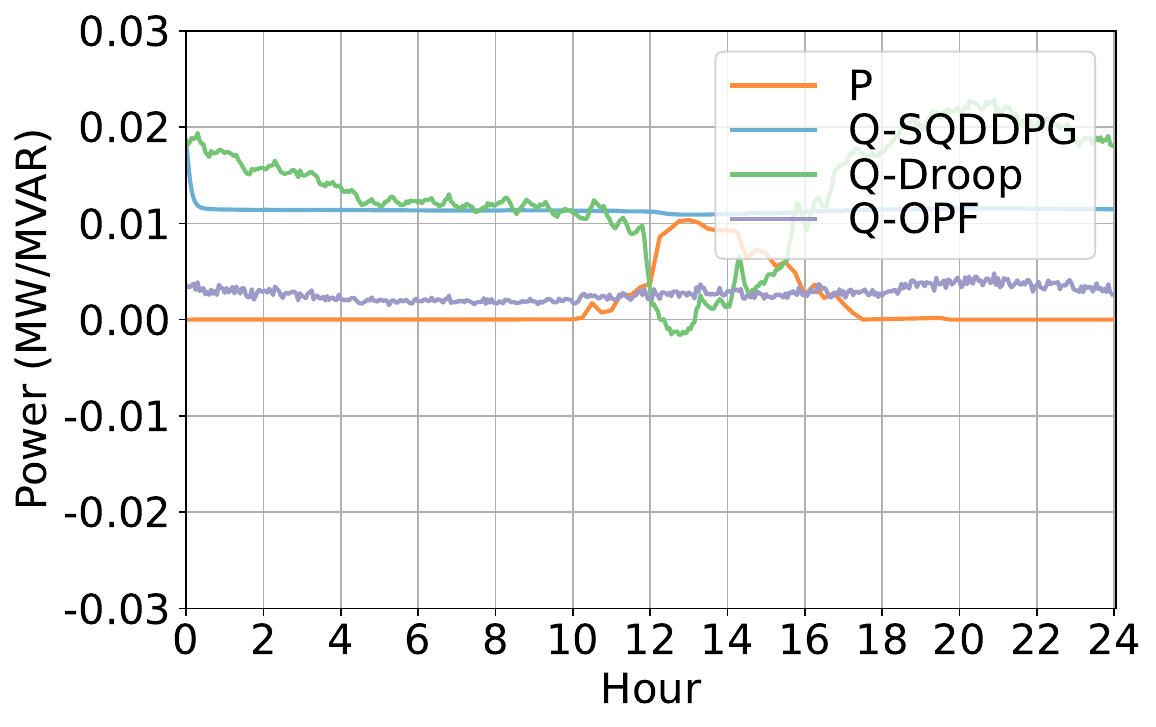}
                \caption{Power.}
            \end{subfigure}
            \quad
            \begin{subfigure}[b]{0.30\textwidth}
                \centering                
                \includegraphics[width=\textwidth]{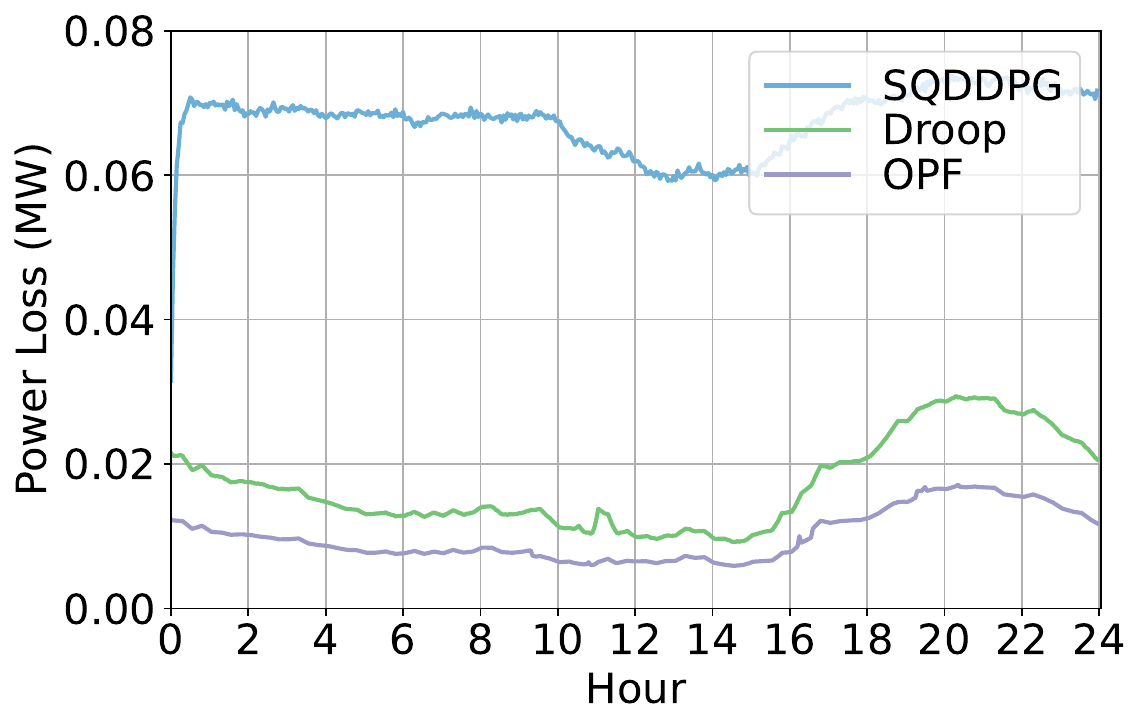}
                \caption{Power Loss.}
            \end{subfigure}
            \caption{Comparing SQDDPG with traditional control methods on a typical bus during a day in the 322-bus network. 1st row: the results of a summer day. 2nd row: the results of a winter day. None and limit in (a) represent the voltage with no control and the safety range respectively. P and Q in (b) indicate the PV active power and the reactive power by various methods.}
        \label{fig:case_study_322_sqddpg}
        \end{figure*}
        
        \begin{figure*}[ht!]
            \centering
            \begin{subfigure}[b]{0.30\textwidth}
            	\centering
        	    \includegraphics[width=\textwidth]{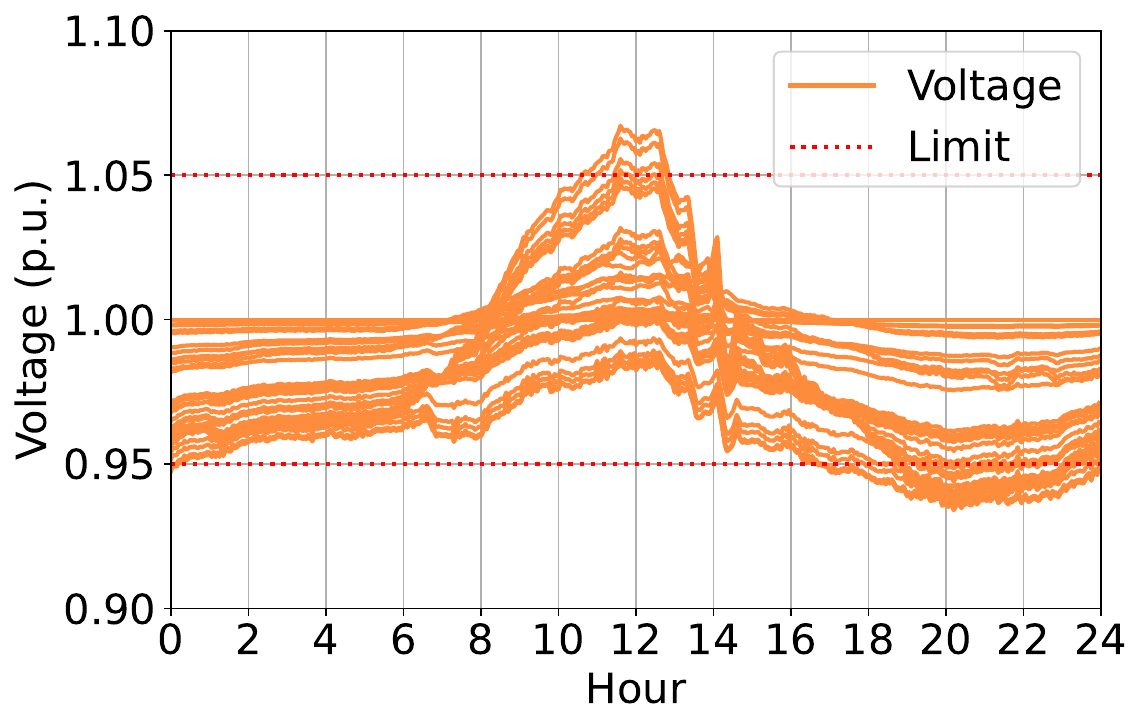}
        	    \caption{33-summer.}
            \end{subfigure}
            \quad
            \begin{subfigure}[b]{0.30\textwidth}
                \centering                
                \includegraphics[width=\textwidth]{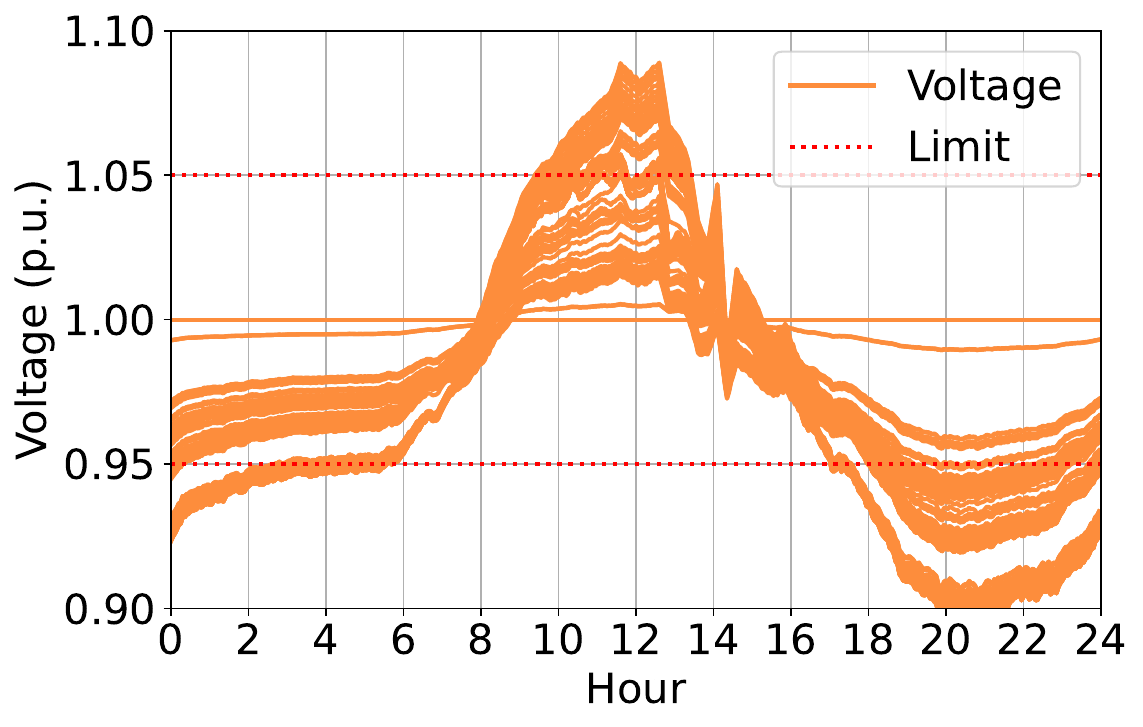}
                \caption{141-summer.}
            \end{subfigure}
            \quad
            \begin{subfigure}[b]{0.30\textwidth}
                \centering                
                \includegraphics[width=\textwidth]{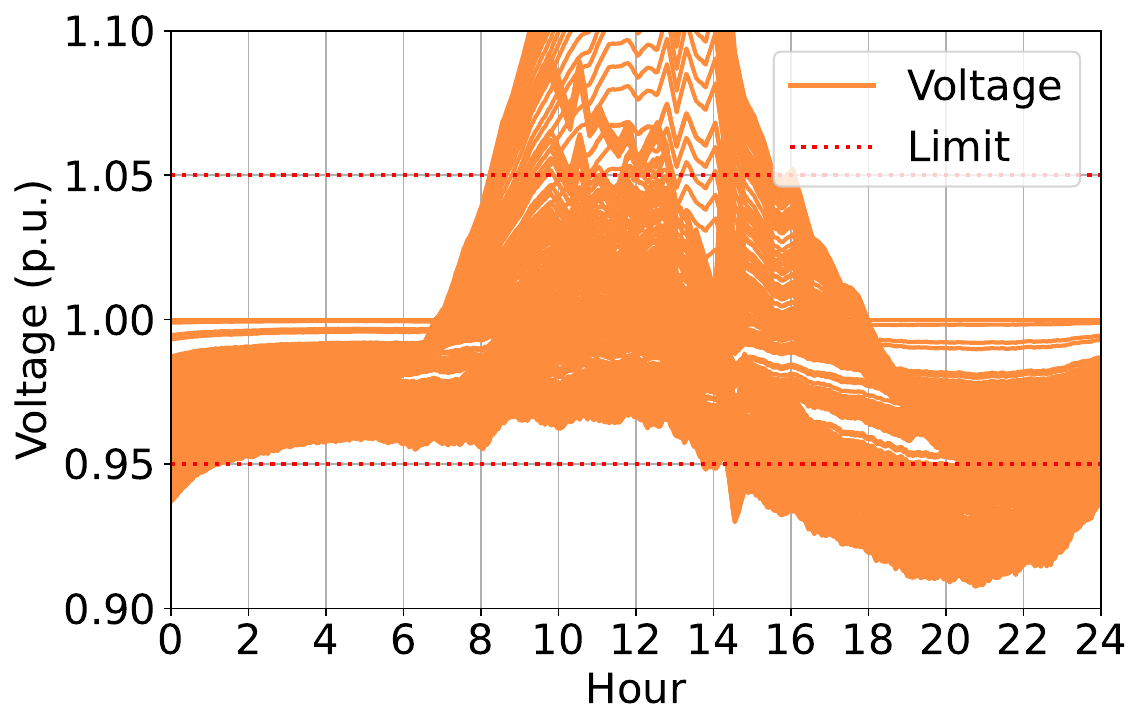}
                \caption{322-summer.}
            \end{subfigure}
            \quad
            \begin{subfigure}[b]{0.30\textwidth}
            	\centering
        	    \includegraphics[width=\textwidth]{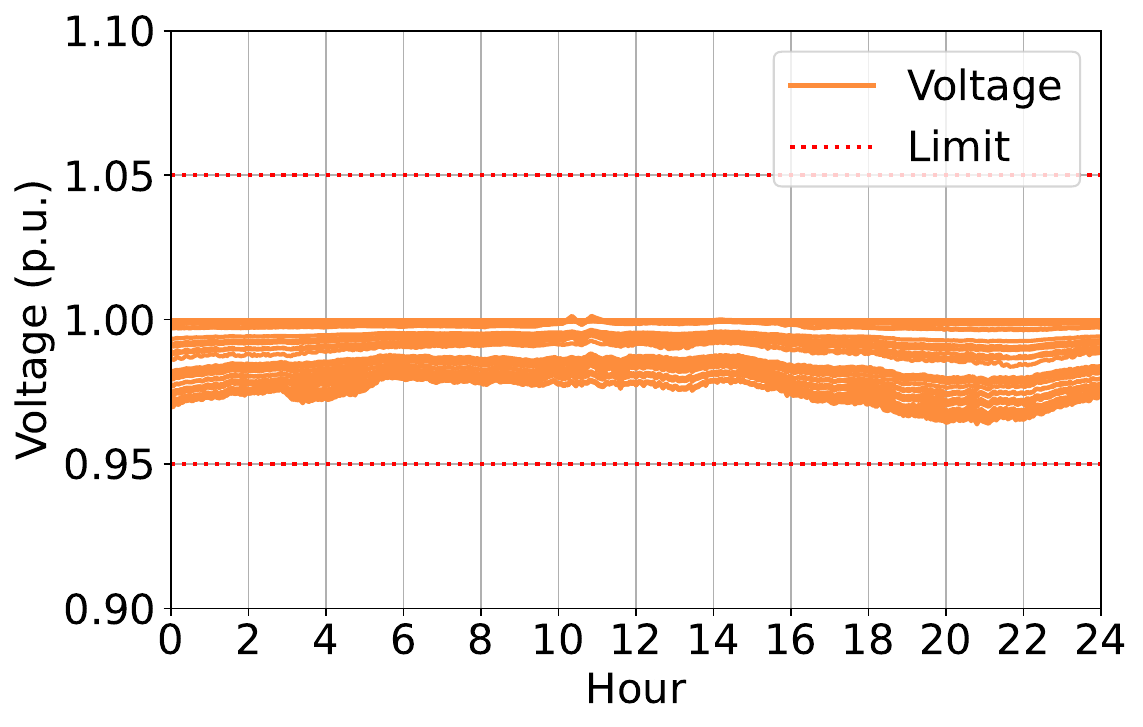}
        	    \caption{33-winter.}
            \end{subfigure}
            \quad
            \begin{subfigure}[b]{0.30\textwidth}
                \centering                
                \includegraphics[width=\textwidth]{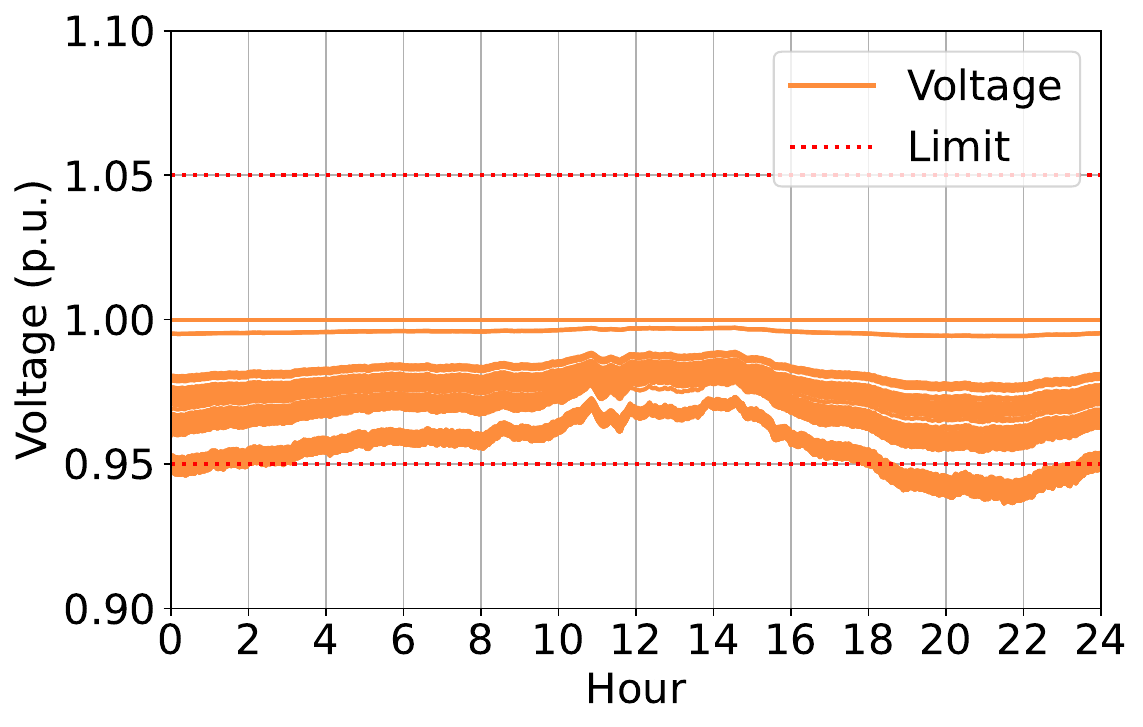}
                \caption{141-winter.}
            \end{subfigure}
            \quad
            \begin{subfigure}[b]{0.30\textwidth}
                \centering                
                \includegraphics[width=\textwidth]{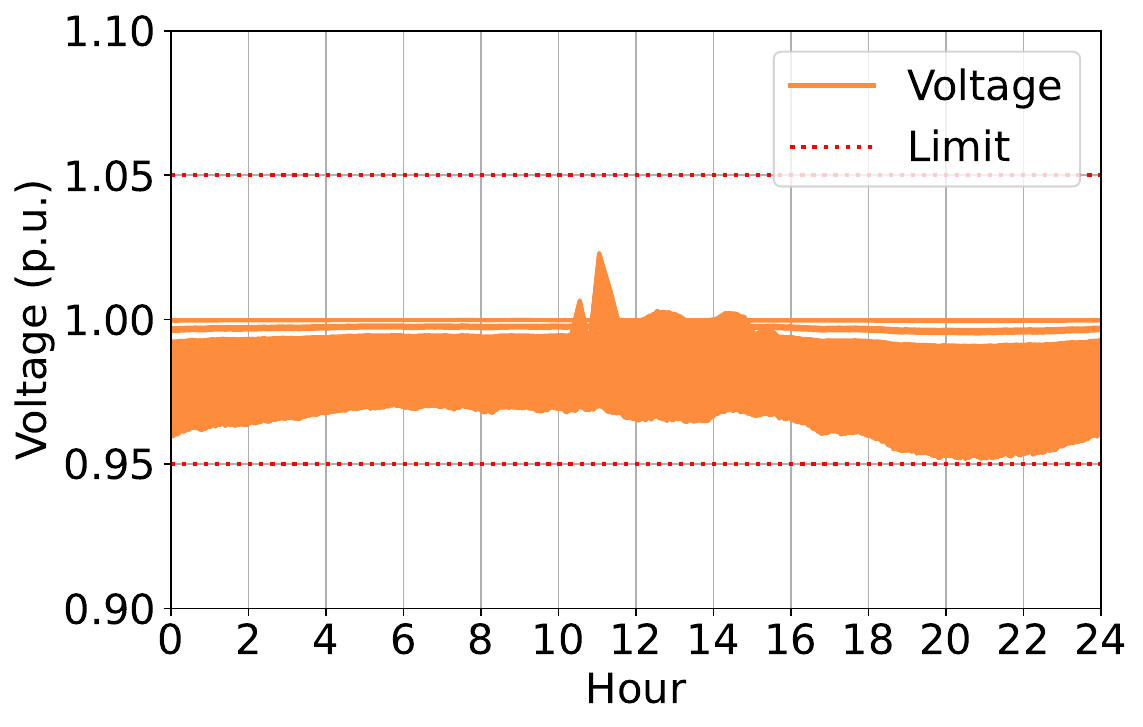}
                \caption{322-winter.}
            \end{subfigure}
            \caption{Status of all buses' voltage (in orange lines) in a day on the 33-bus, the 141-bus and the 322-bus networks in summer and winter. The red dashed lines are the safety boundaries. Each caption above indicates [network]-[season].}
        \label{fig:case_study_detail_pf}
        \end{figure*}
        
        \paragraph{Analysis for All Buses.} To give the whole picture of active voltage control for the days we select for demonstration above, we show the status of all buses with no control for all scenarios in Figure \ref{fig:case_study_detail_pf}; as well as the status of all buses under control methods in Figure \ref{fig:case_study_33_detail} about the 33-bus network, Figure \ref{fig:case_study_141_detail} about the 141-bus network and Figure \ref{fig:case_study_322_detail} about the 322-bus network. In winter, almost all methods can control the voltage of all buses within the safety range in all scenarios. For this reason, we only focus on the results of summer in the following discussion. 
        
        In the 33-bus network and the 141-bus network, it is obvious that all methods can control the voltage within the safety range. In the 322-bus network, the performance of droop control is far better than OPF, SMFPPO and SQDDPG, which is the only method controlling all buses' voltages within the safety range. The reason for the failure of the OPF is probably due to the computational burden as we discussed before. It is worth noting that the success in the droop control highly relies on a high-bandwidth inner loop in the inverter controller (i.e., analogous to the optimization procedure to solve a static game in multi-agent learning), so the effective control rate is much higher than the sample rate \cite{MajzoubiZCK0S20}. The failure of SMFPPO is probably due to the fact that the increasing number of agents leads to the difficulty of estimating the accurate coalition value function and therefore the Markov Shapley value. It is seen that the patterns of the resulting voltage yielded by SQDDPG are so close to the droop control. In contrast, the patterns of the resulting voltage yielded by SMFPPO lie between the OPF and the droop control. This is an interesting phenomenon which deserves to be investigated in the future work.
        \begin{figure*}[ht!]
            \centering
            \begin{subfigure}[b]{0.45\textwidth}
            	\centering
        	    \includegraphics[width=\textwidth]{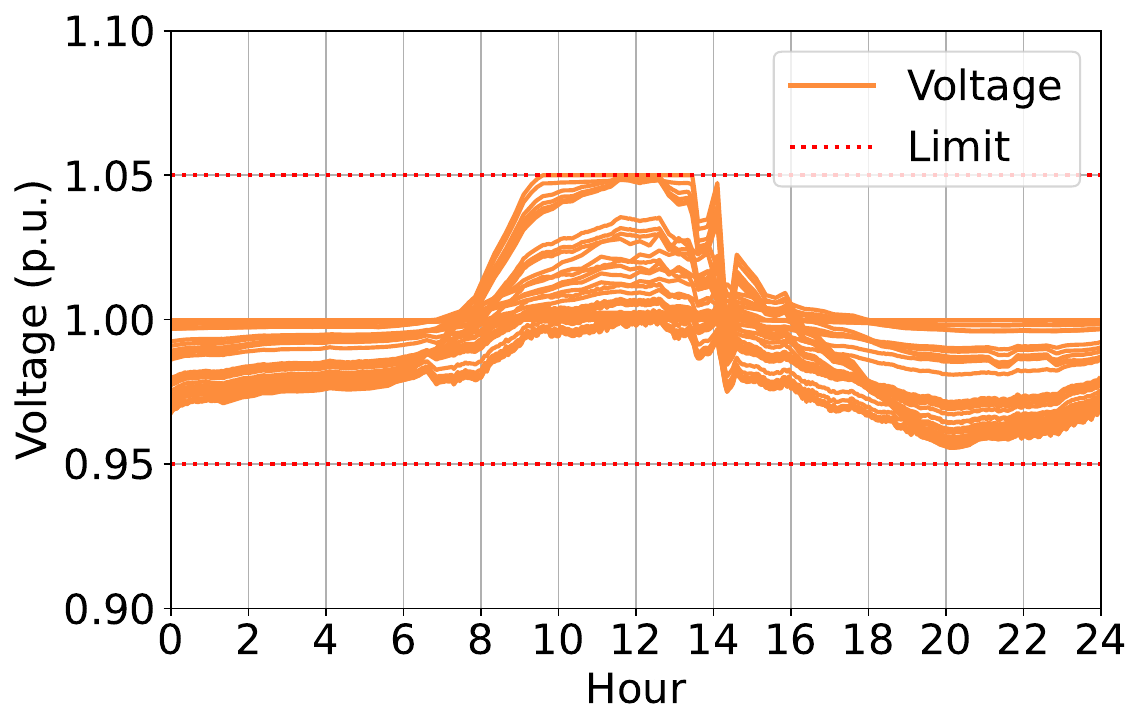}
        	    \caption{OPF-summer.}
            \end{subfigure}
            \quad
            \begin{subfigure}[b]{0.45\textwidth}
                \centering                
                \includegraphics[width=\textwidth]{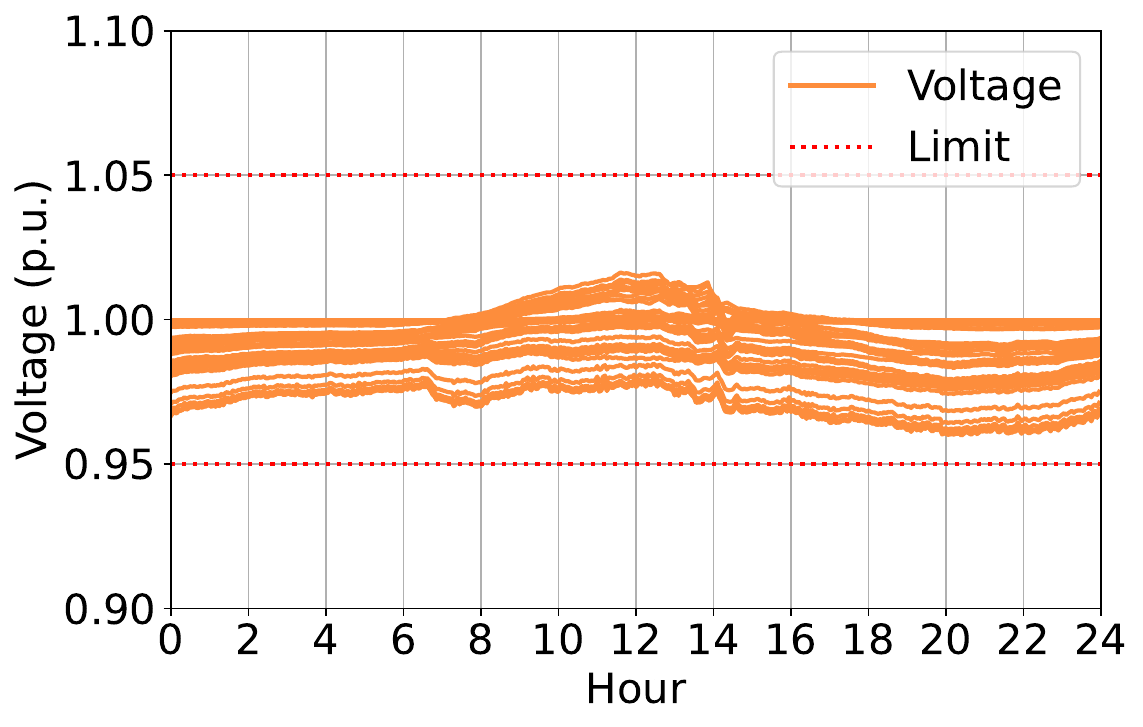}
                \caption{Droop-summer.}
            \end{subfigure}
            \quad
            \begin{subfigure}[b]{0.45\textwidth}
                \centering                
                \includegraphics[width=\textwidth]{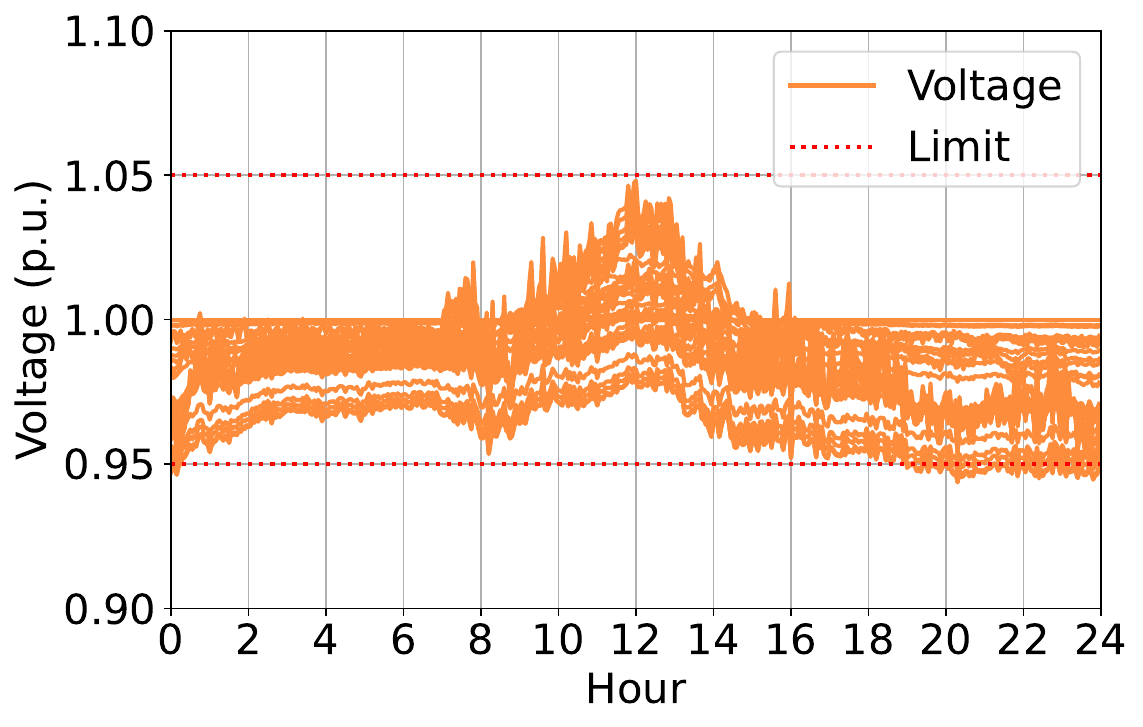}
                \caption{SMFPPO-summer.}
            \end{subfigure}
            \quad
            \begin{subfigure}[b]{0.45\textwidth}
                \centering                
                \includegraphics[width=\textwidth]{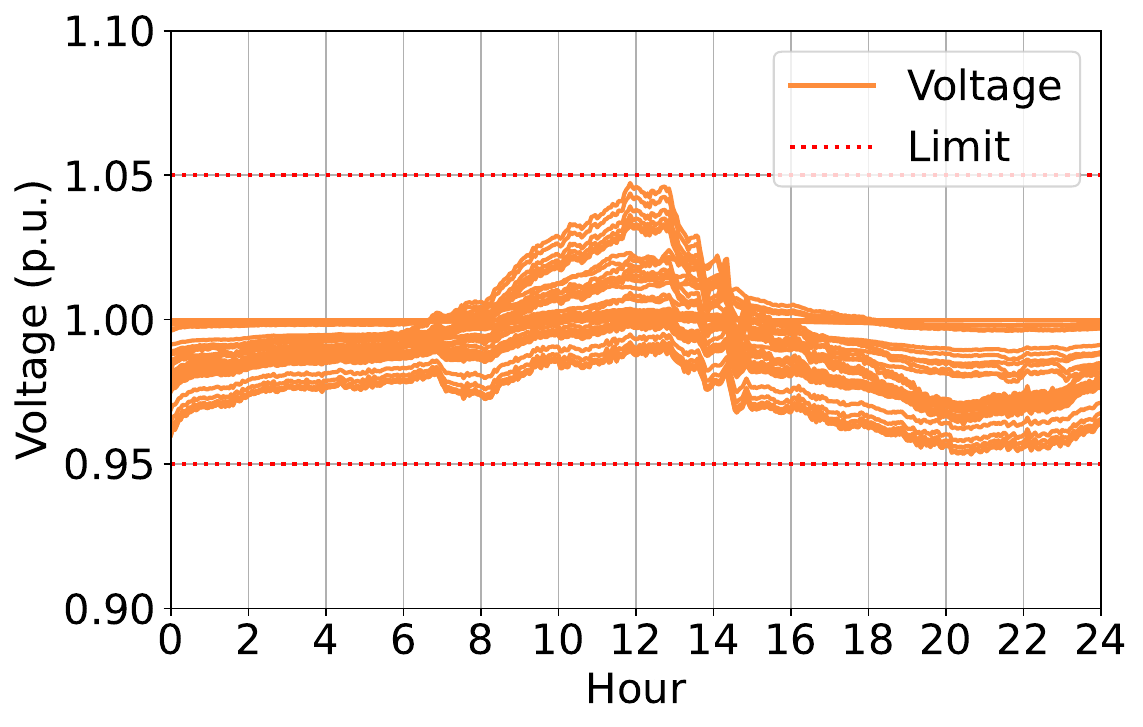}
                \caption{SQDDPG-summer.}
            \end{subfigure}
            \quad
            \begin{subfigure}[b]{0.45\textwidth}
            	\centering
        	    \includegraphics[width=\textwidth]{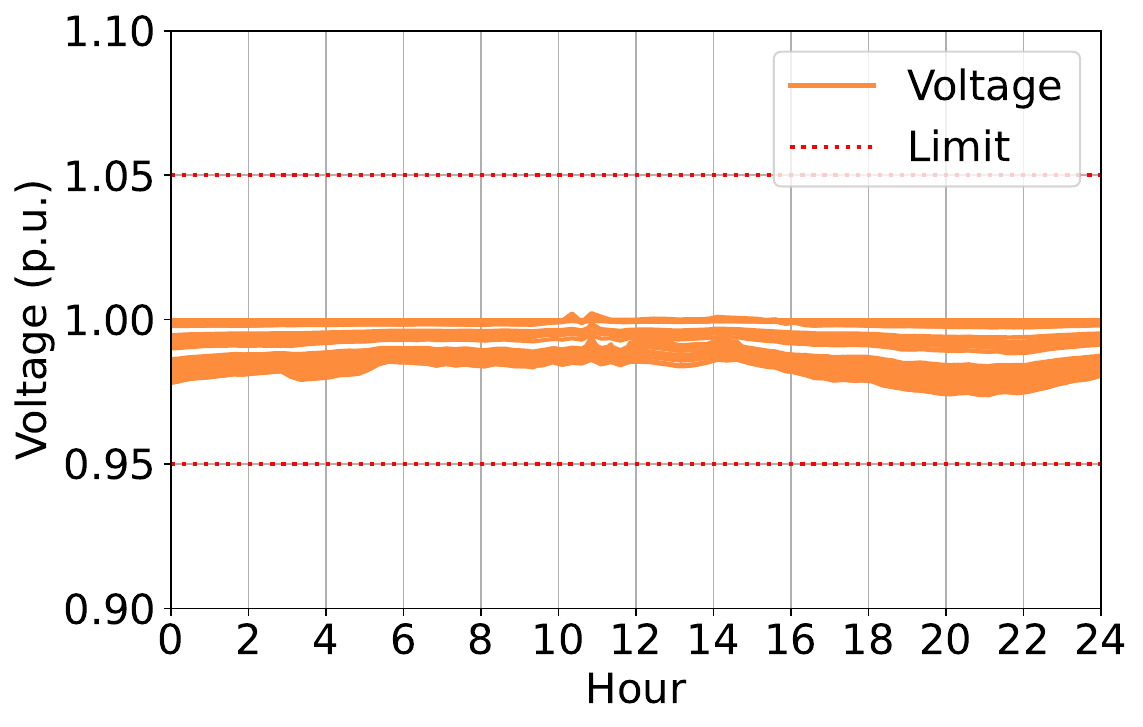}
        	    \caption{OPF-winter.}
            \end{subfigure}
            \quad
            \begin{subfigure}[b]{0.45\textwidth}
                \centering                
                \includegraphics[width=\textwidth]{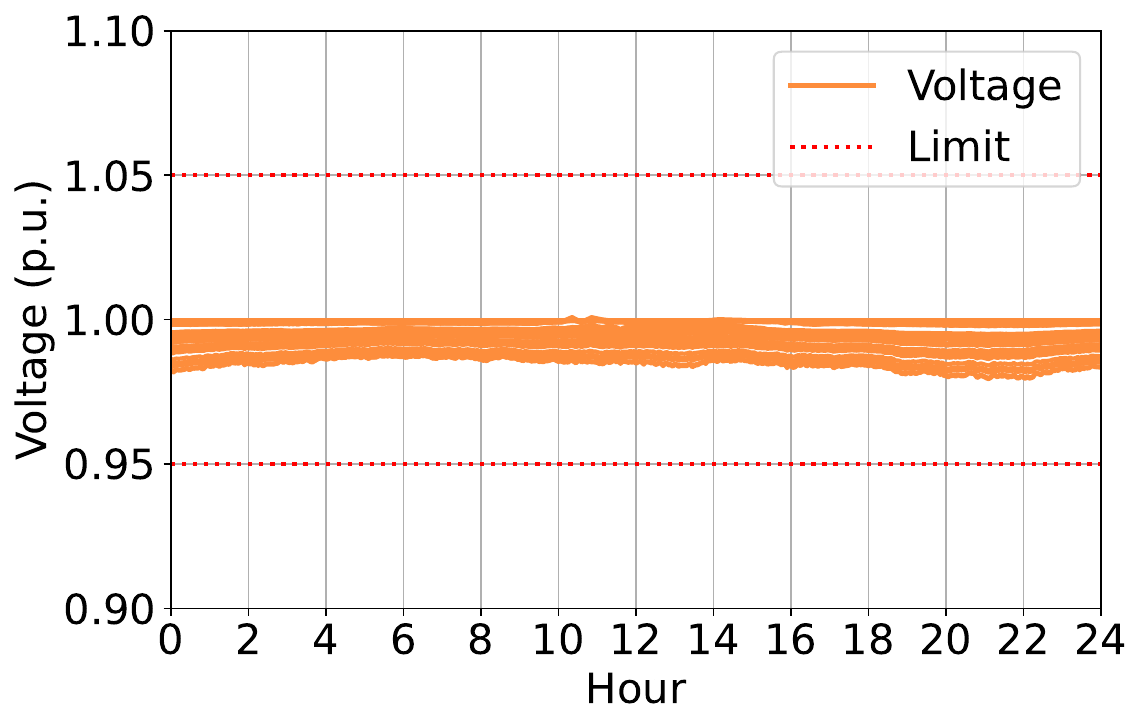}
                \caption{Droop-winter.}
            \end{subfigure}
            \quad
            \begin{subfigure}[b]{0.45\textwidth}
                \centering                
                \includegraphics[width=\textwidth]{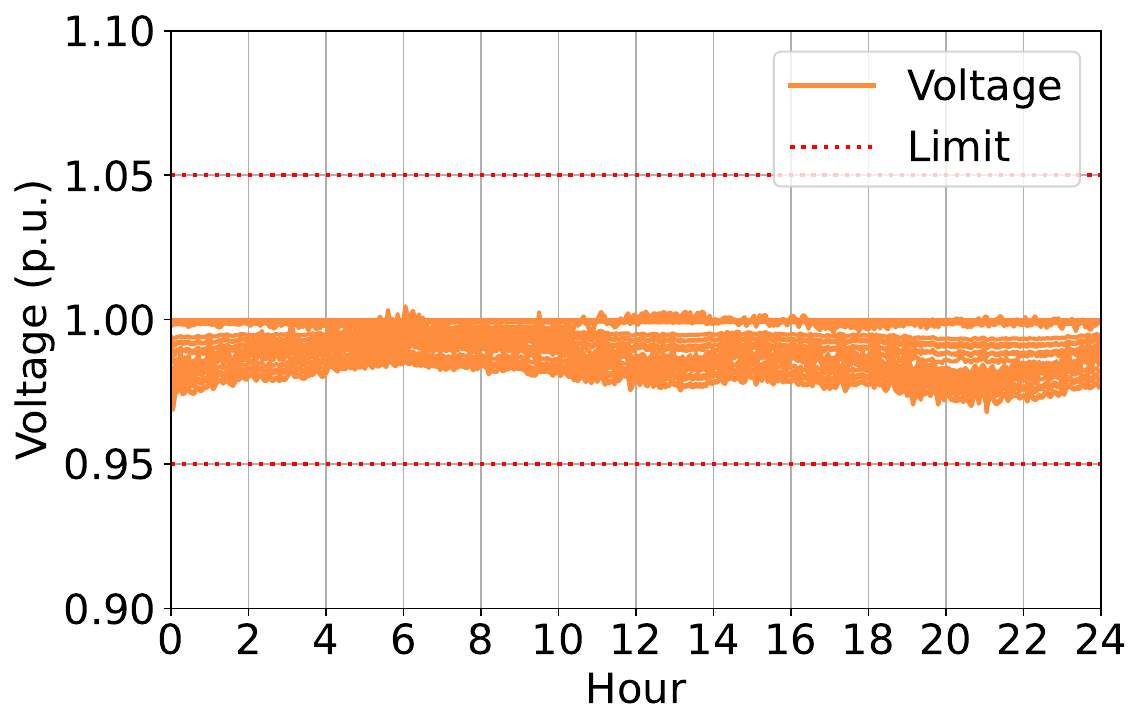}
                \caption{SMFPPO-winter.}
            \end{subfigure}
            \quad
            \begin{subfigure}[b]{0.45\textwidth}
                \centering                
                \includegraphics[width=\textwidth]{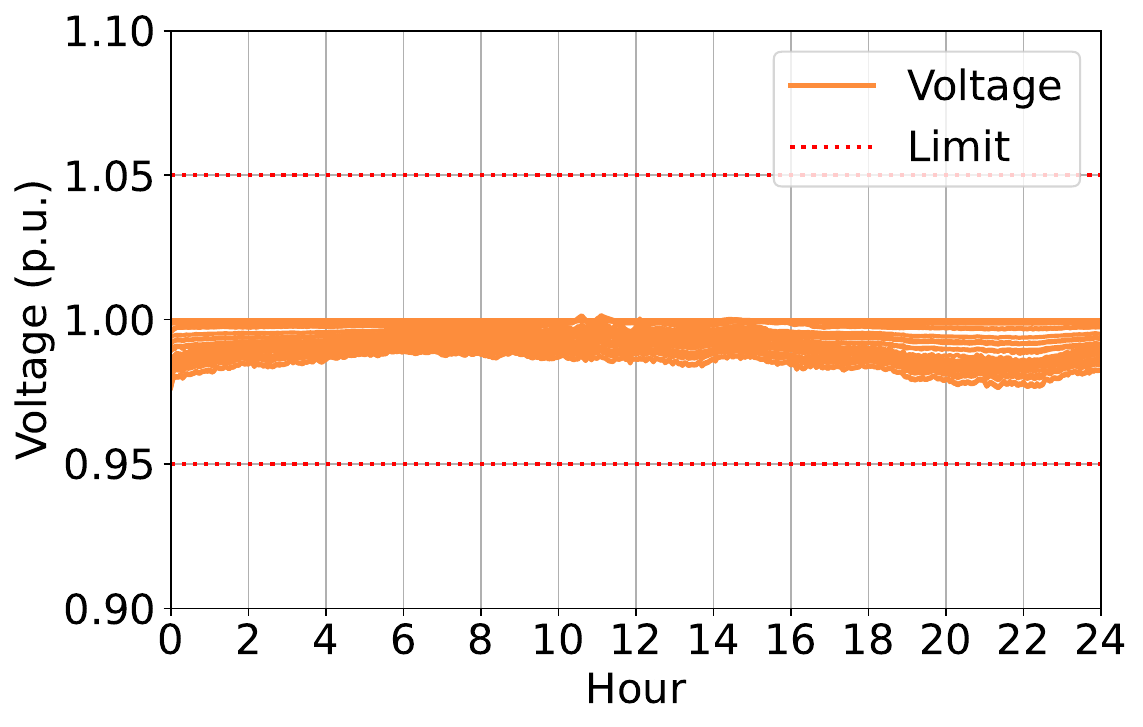}
                \caption{SQDDPG-winter.}
            \end{subfigure}
            \caption{Status of all buses in a day in the 33-bus network. The orange lines are the variation of buses' voltage and red dashed lines are the safety boundaries. Each caption above indicates [method]-[season].}
        \label{fig:case_study_33_detail}
        \end{figure*}
        
        \begin{figure*}[ht!]
            \centering
            \begin{subfigure}[b]{0.45\textwidth}
            	\centering
        	    \includegraphics[width=\textwidth]{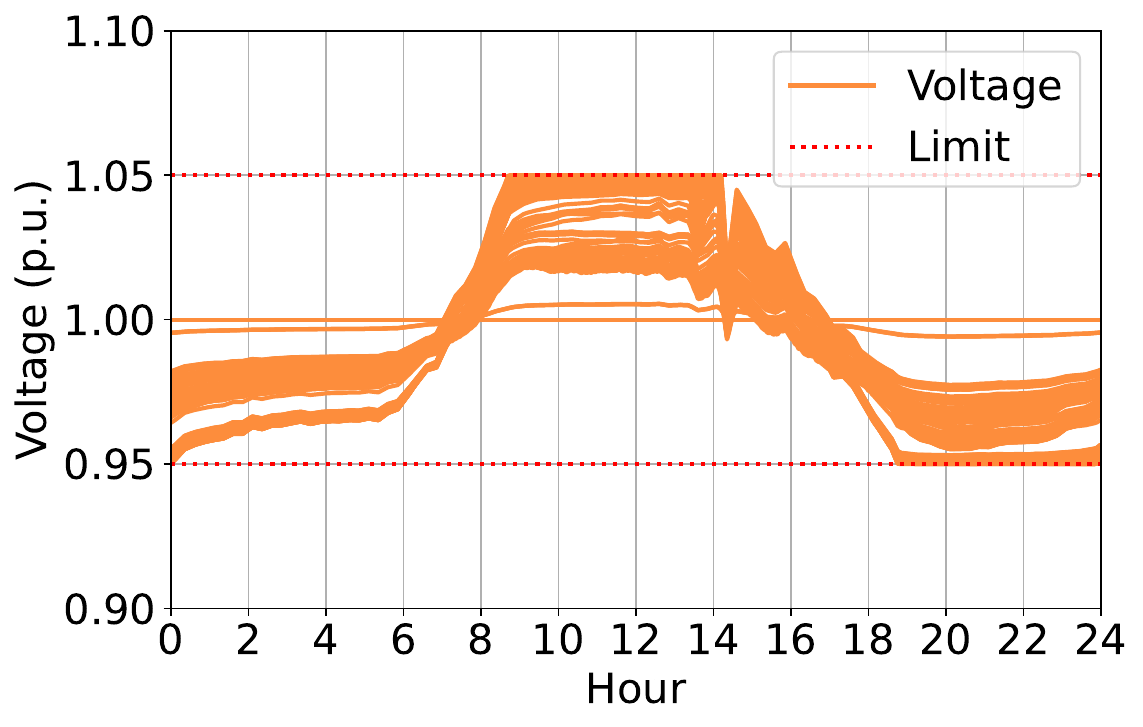}
        	    \caption{OPF-summer.}
            \end{subfigure}
            \quad
            \begin{subfigure}[b]{0.45\textwidth}
                \centering                
                \includegraphics[width=\textwidth]{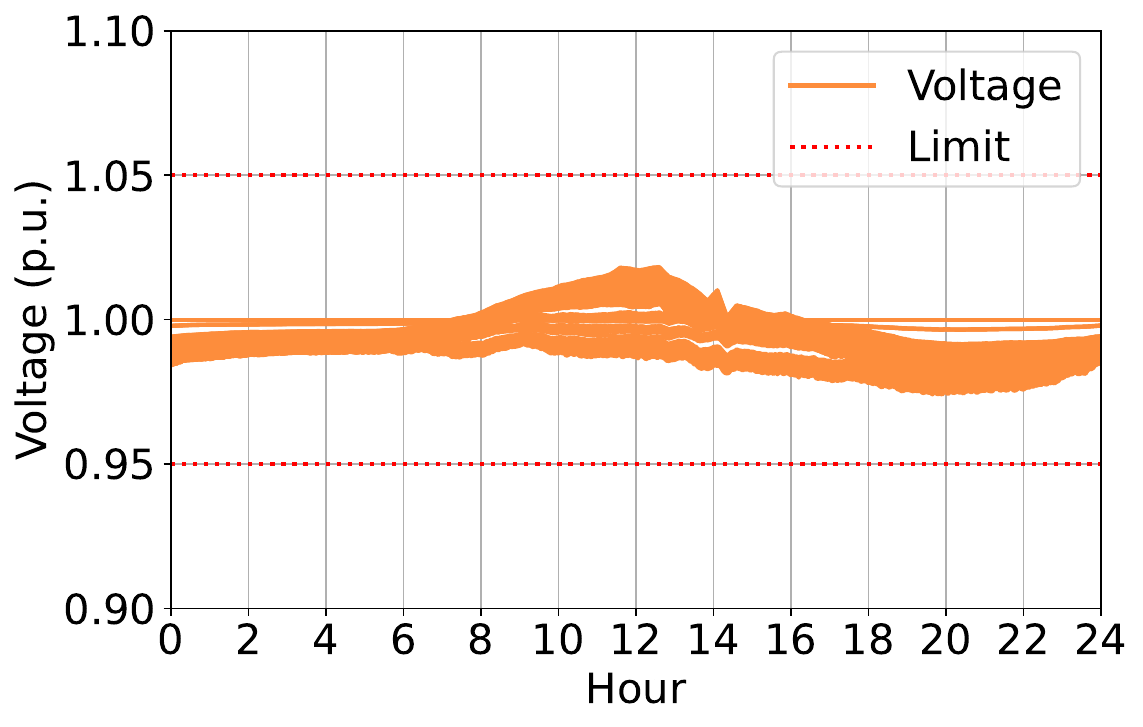}
                \caption{Droop-summer.}
            \end{subfigure}
            \quad
            \begin{subfigure}[b]{0.45\textwidth}
                \centering                
                \includegraphics[width=\textwidth]{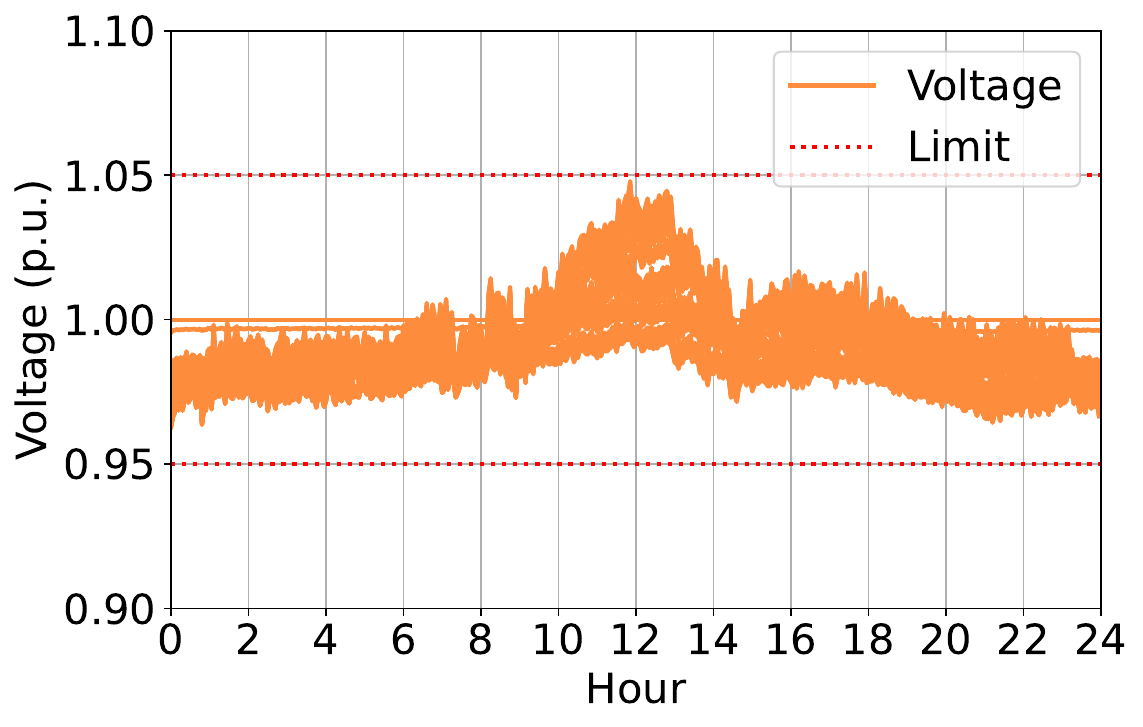}
                \caption{SMFPPO-summer.}
            \end{subfigure}
            \quad
            \begin{subfigure}[b]{0.45\textwidth}
                \centering                
                \includegraphics[width=\textwidth]{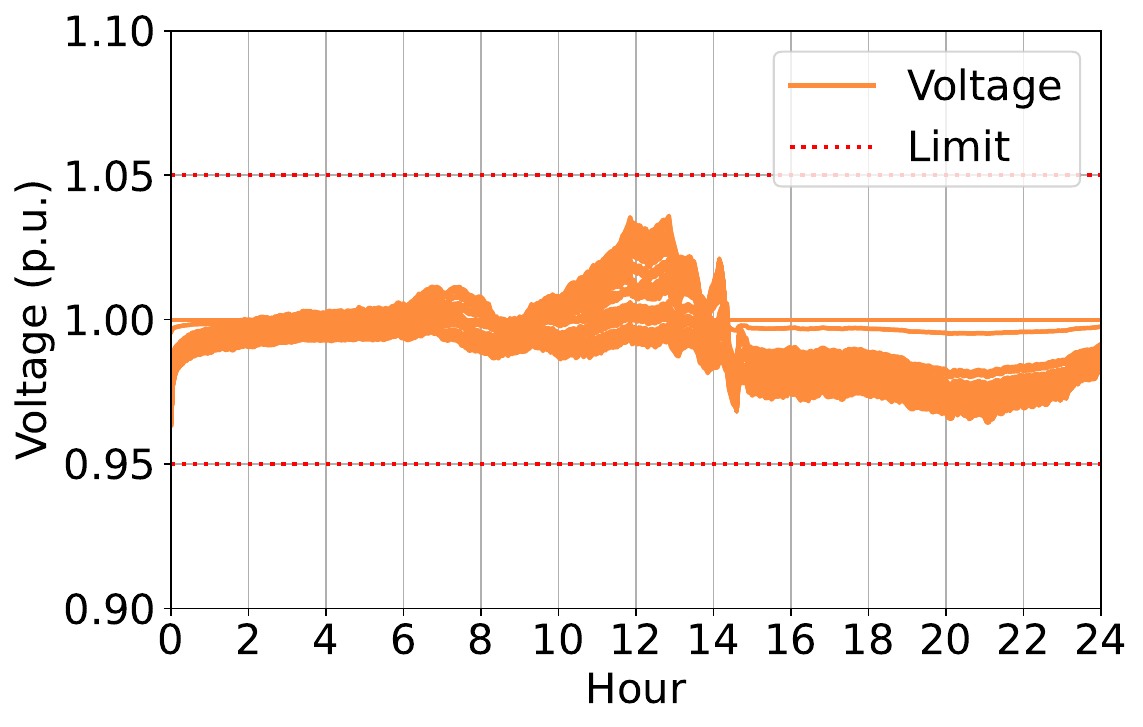}
                \caption{SQDDPG-summer.}
            \end{subfigure}
            \quad
            \begin{subfigure}[b]{0.45\textwidth}
            	\centering
        	    \includegraphics[width=\textwidth]{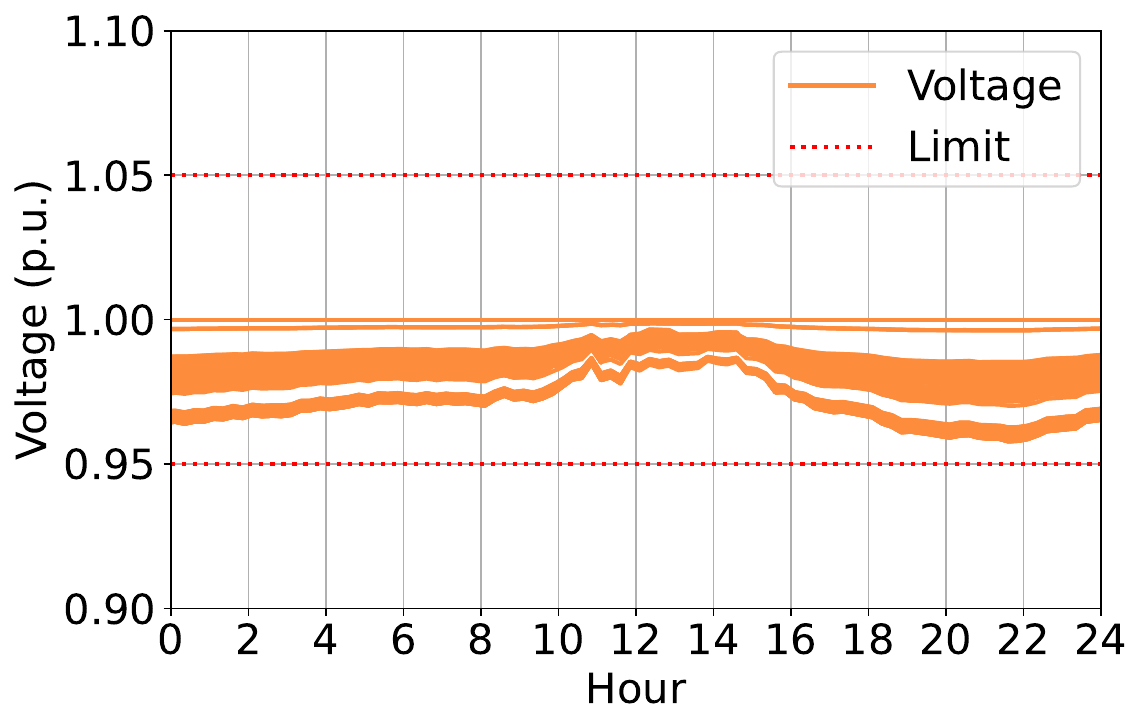}
        	    \caption{OPF-winter.}
            \end{subfigure}
            \quad
            \begin{subfigure}[b]{0.45\textwidth}
                \centering                
                \includegraphics[width=\textwidth]{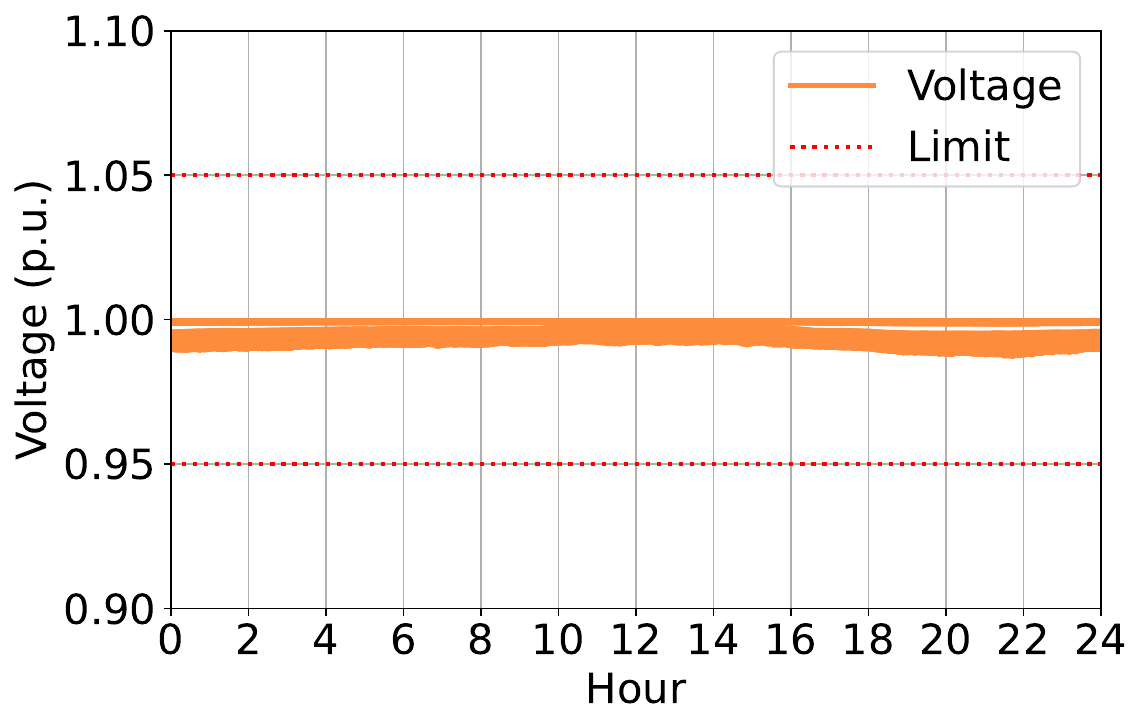}
                \caption{Droop-winter.}
            \end{subfigure}
            \quad
            \begin{subfigure}[b]{0.45\textwidth}
                \centering                
                \includegraphics[width=\textwidth]{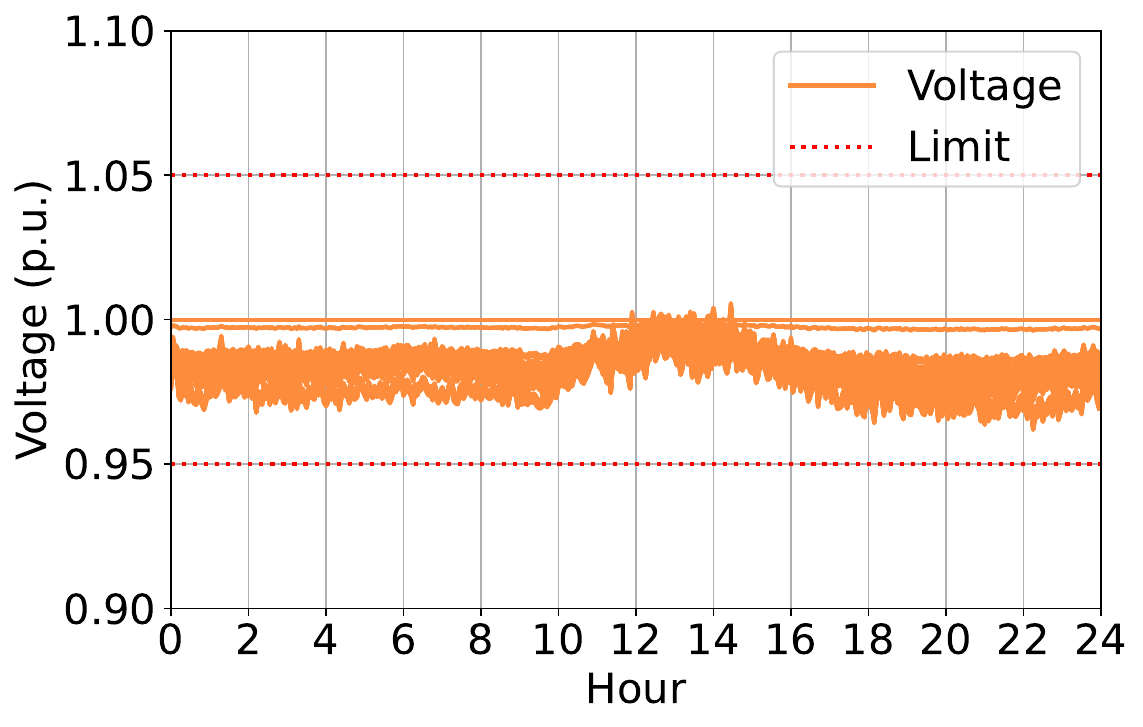}
                \caption{SMFPPO-winter.}
            \end{subfigure}
            \quad
            \begin{subfigure}[b]{0.45\textwidth}
                \centering                
                \includegraphics[width=\textwidth]{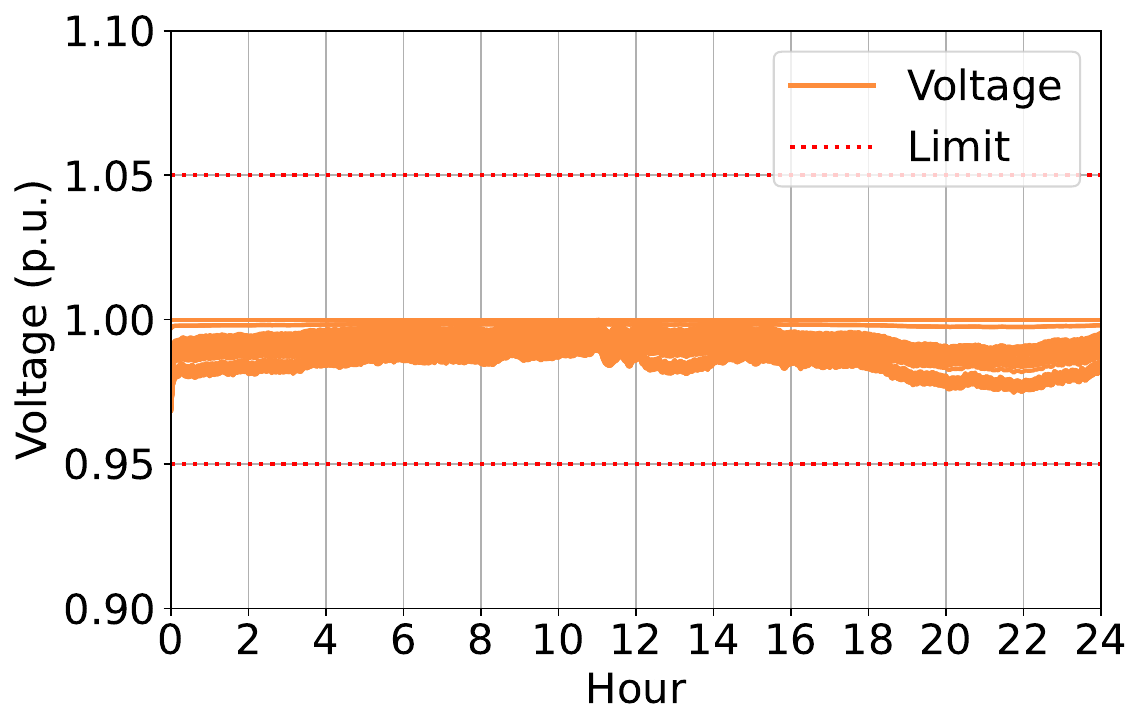}
                \caption{SQDDPG-winter.}
            \end{subfigure}
            \caption{Status of all buses in a day in the 141-bus network. The orange lines are the variation of buses' voltage and red dashed lines are the safety boundaries. Each caption above indicates [method]-[season].}
        \label{fig:case_study_141_detail}
        \end{figure*}
        
        \begin{figure*}[ht!]
            \centering
            \begin{subfigure}[b]{0.45\textwidth}
            	\centering
        	    \includegraphics[width=\textwidth]{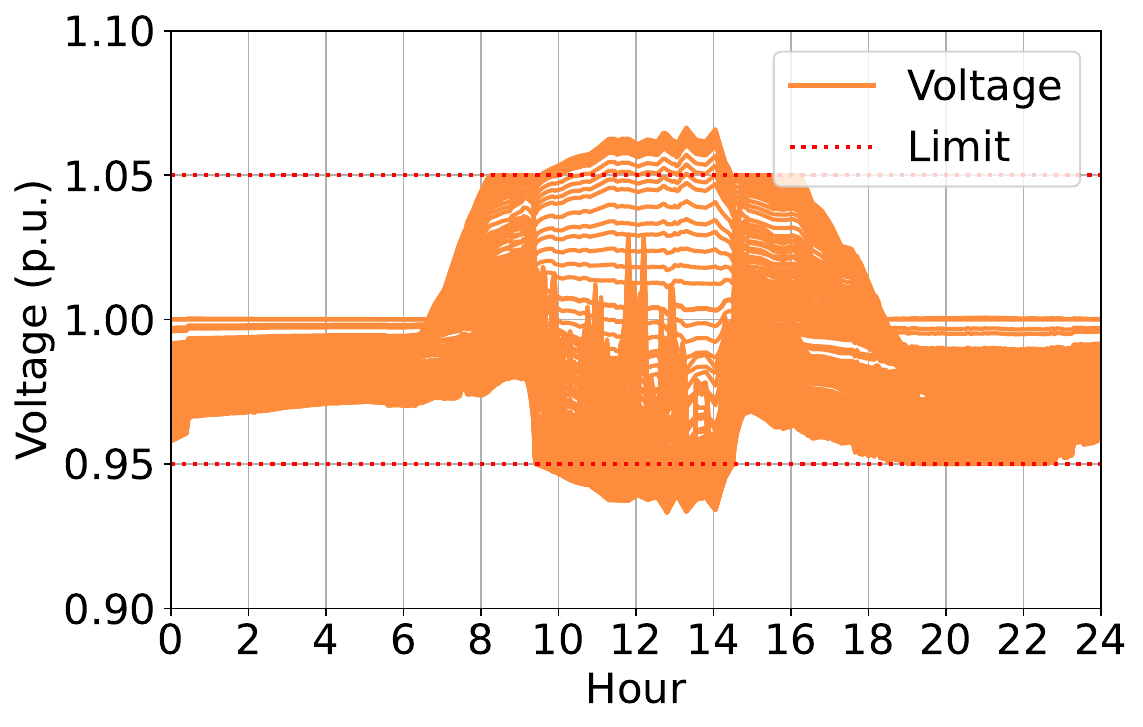}
        	    \caption{OPF-summer.}
            \end{subfigure}
            \quad
            \begin{subfigure}[b]{0.45\textwidth}
                \centering                
                \includegraphics[width=\textwidth]{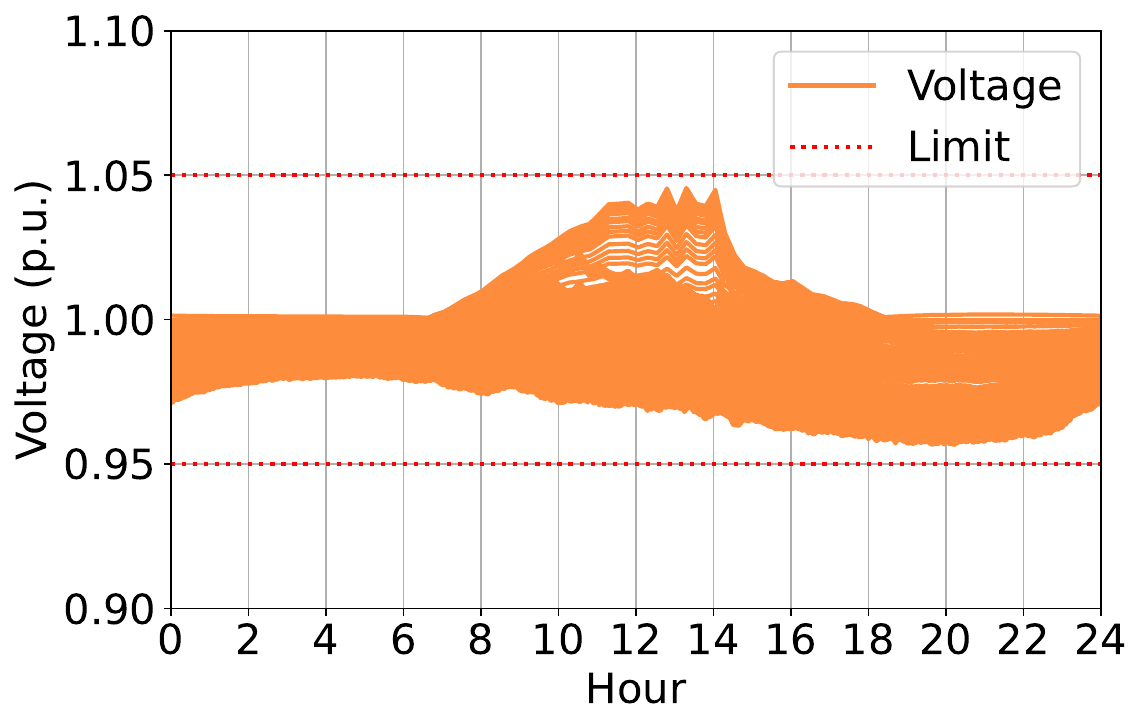}
                \caption{Droop-summer.}
            \end{subfigure}
            \quad
            \begin{subfigure}[b]{0.45\textwidth}
                \centering                
                \includegraphics[width=\textwidth]{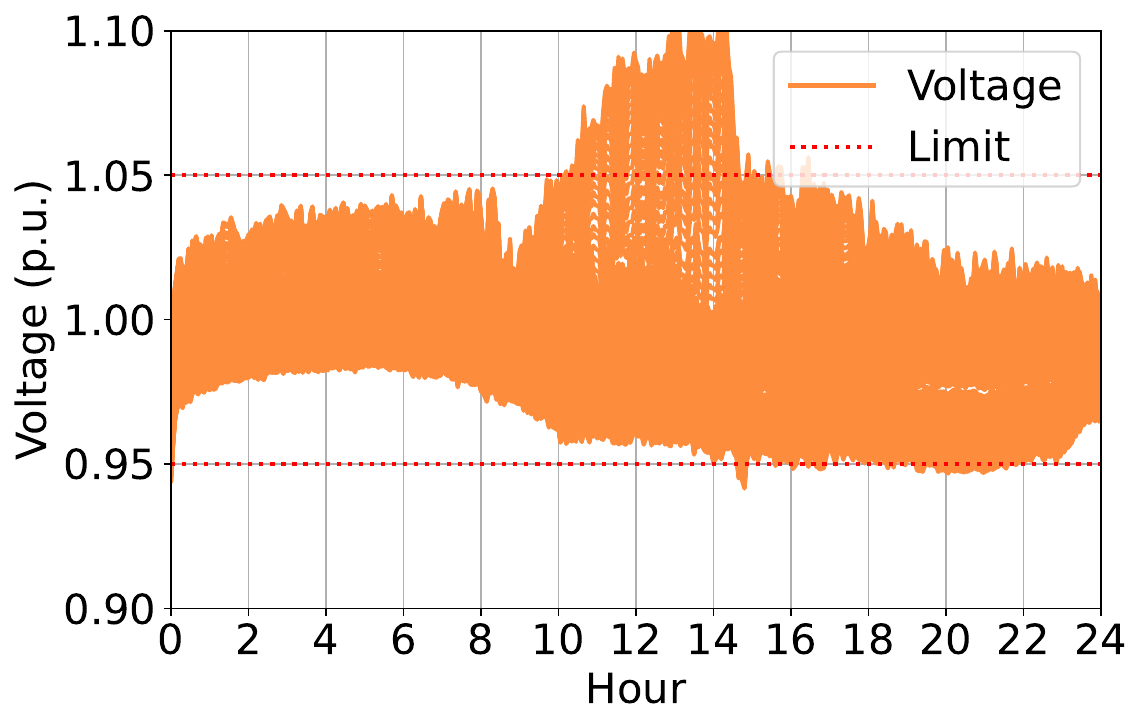}
                \caption{SMFPPO-summer.}
            \end{subfigure}
            \quad
            \begin{subfigure}[b]{0.45\textwidth}
                \centering                
                \includegraphics[width=\textwidth]{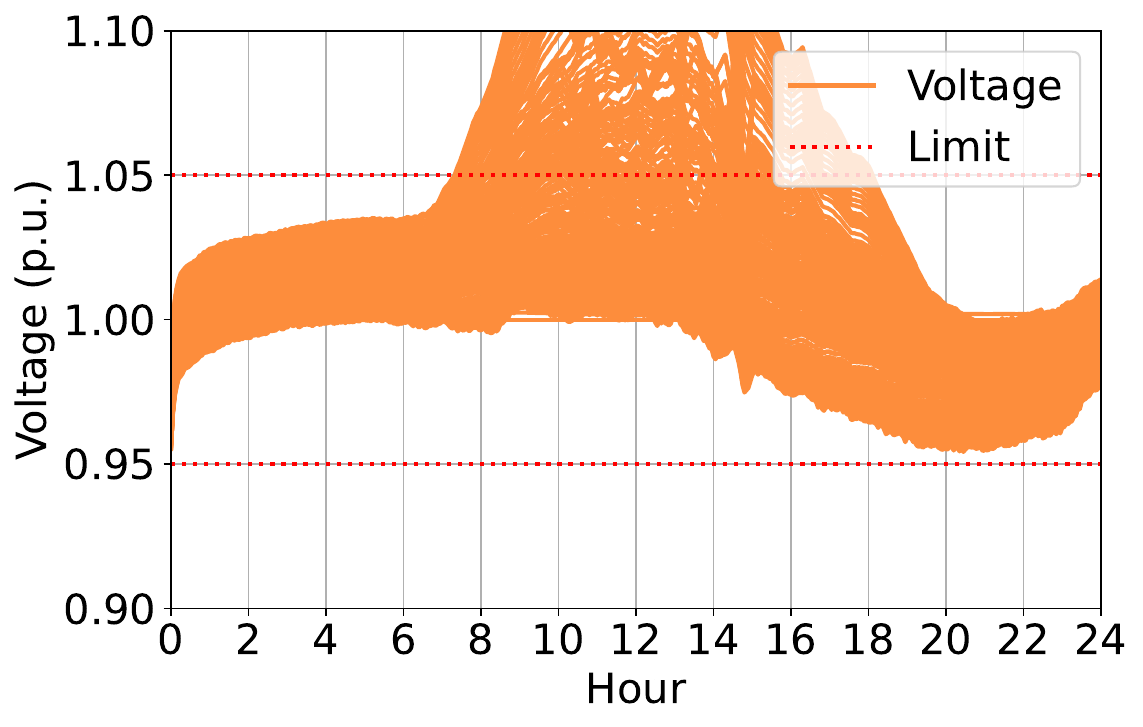}
                \caption{SQDDPG-summer.}
            \end{subfigure}
            \quad
            \begin{subfigure}[b]{0.45\textwidth}
            	\centering
        	    \includegraphics[width=\textwidth]{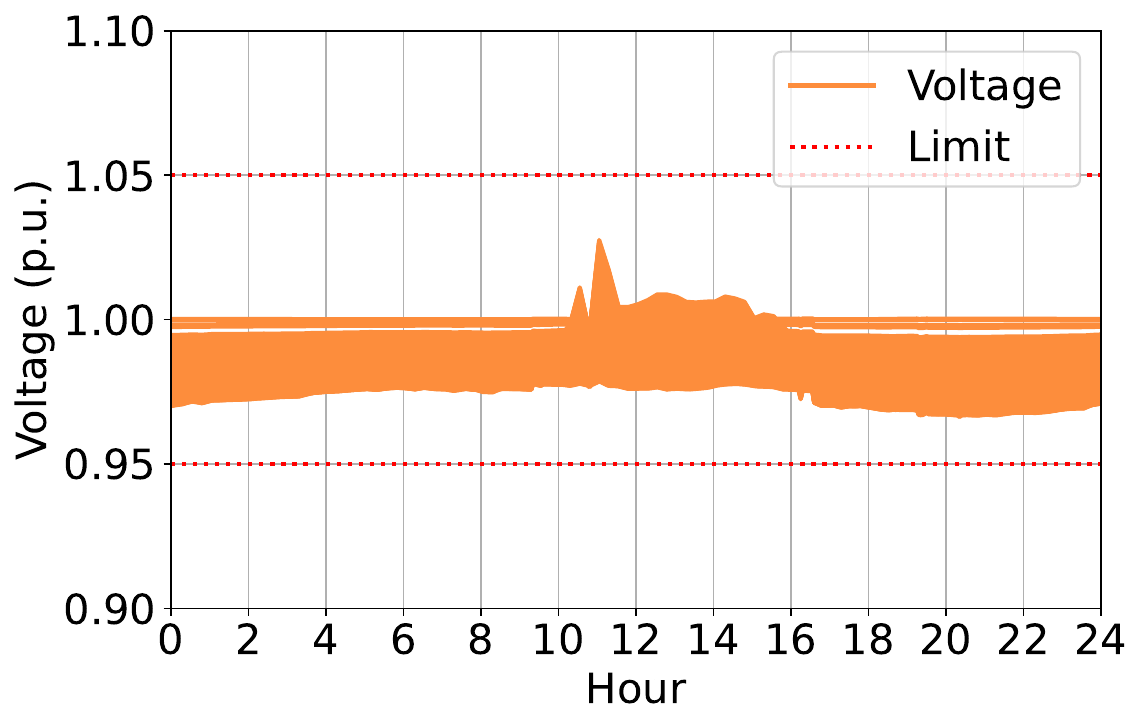}
        	    \caption{OPF-winter.}
            \end{subfigure}
            \quad
            \begin{subfigure}[b]{0.45\textwidth}
                \centering                
                \includegraphics[width=\textwidth]{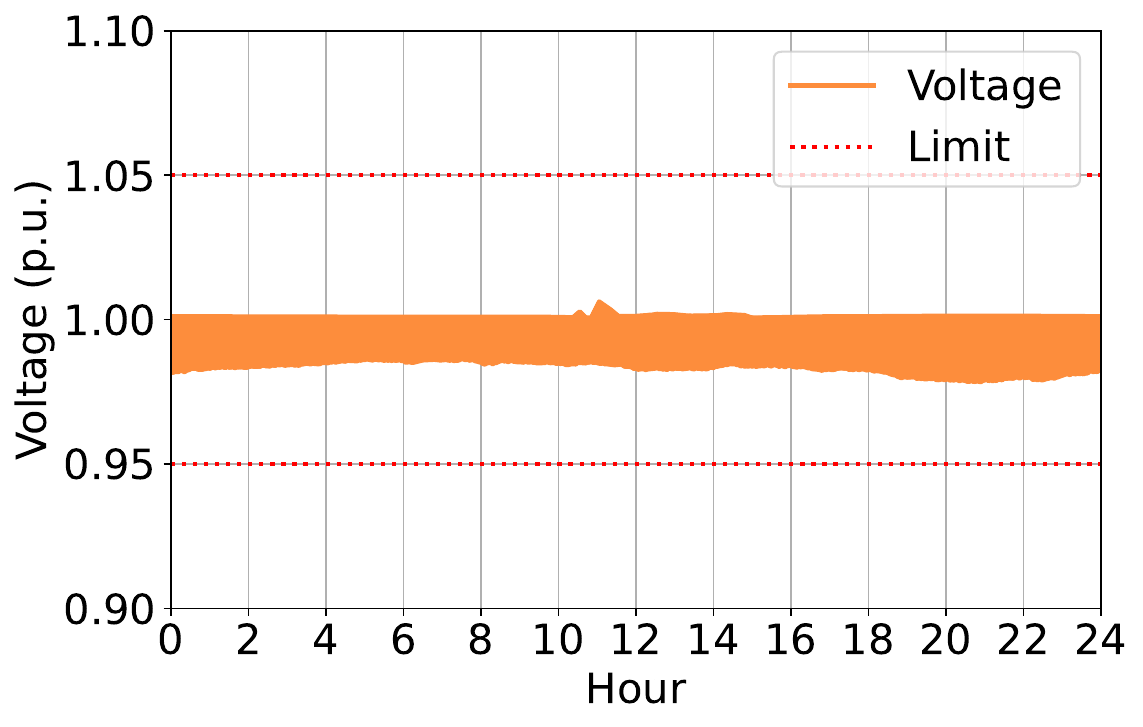}
                \caption{Droop-winter.}
            \end{subfigure}
            \quad
            \begin{subfigure}[b]{0.45\textwidth}
                \centering                
                \includegraphics[width=\textwidth]{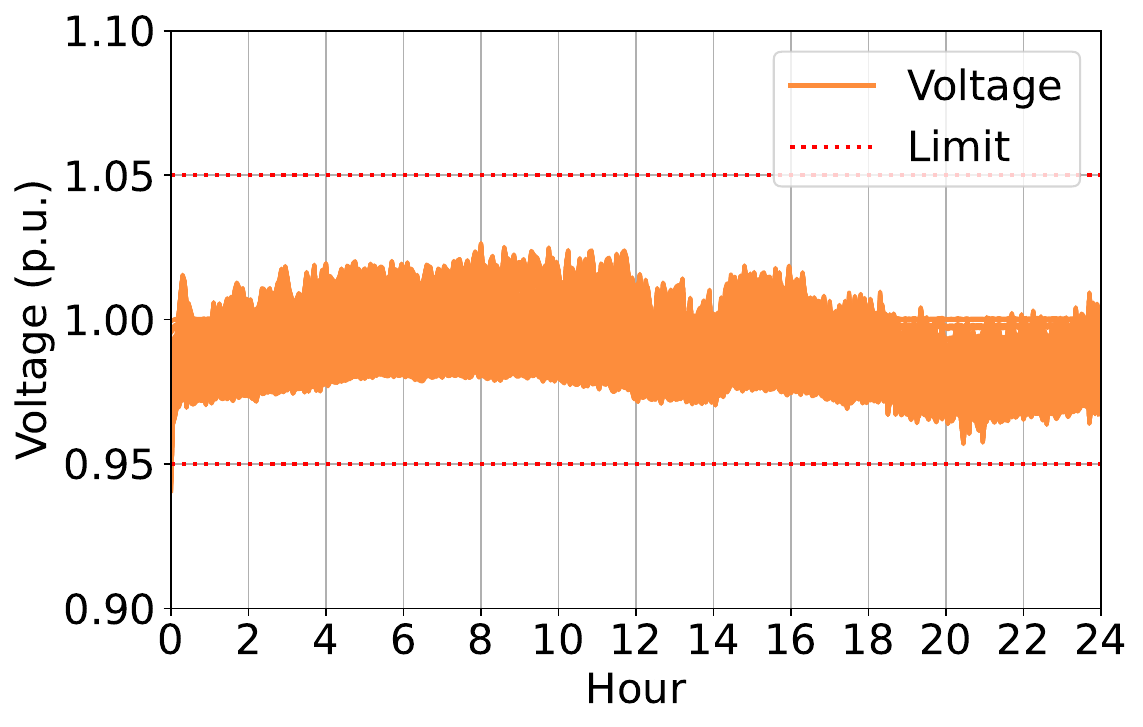}
                \caption{SMFPPO-winter.}
            \end{subfigure}
            \quad
            \begin{subfigure}[b]{0.45\textwidth}
                \centering                
                \includegraphics[width=\textwidth]{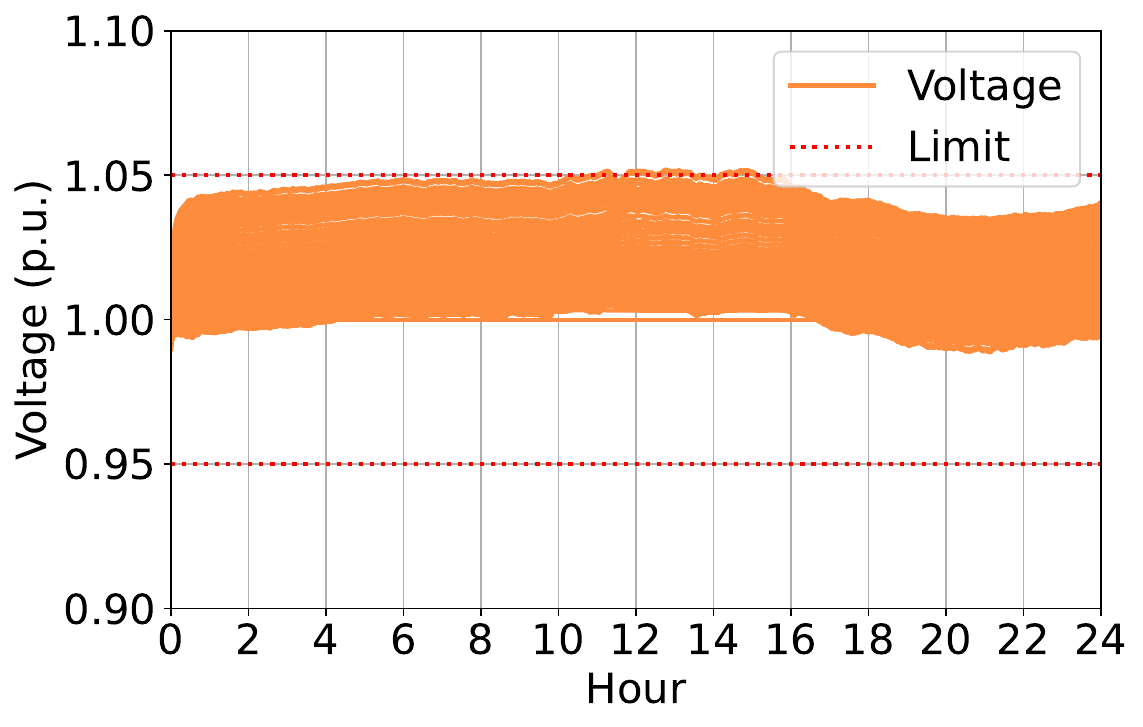}
                \caption{SQDDPG-winter.}
            \end{subfigure}
            \caption{Status of all buses in a day in the 322-bus network. The orange lines are the variation of buses' voltage and red dashed lines are the safety boundaries. Each caption above indicates [method]-[season].}
        \label{fig:case_study_322_detail}
        \end{figure*}
        
        \paragraph{Discussion.} We now discuss the phenomenons that we observe from the simulation results. 
        \begin{itemize}[leftmargin=*]
            \item It is obvious that the voltage barrier function may impact the performance of an algorithm (even with tiny change to the common goal). This may be of general importance as many real-world problems may contain constraints that are not indirect in the objective function. An overall methodology is needed for designing barrier functions in state-constraint MARL.
            \item MARL such as SMFPPO and SQDDPG may scale well to the number of agents and the complexity of networks, and only requires a very low control rate for the active voltage control.
            \item The combination of learning algorithms with domain knowledge is a potential roadmap towards the interpretable MARL. For the active voltage control problem, the domain knowledge may be presented as the network topology, the inner control loop (say droop control), and the load pattern. The exploitation of such domain knowledge reduces the dimensions of MARL exploration space and may offer a lower bound of performance as a guarantee. Encoding domain knowledge such as the network topology as a priori for model-based MARL is also a potential direction.
        \end{itemize}
        
    \subsection{Understanding Markov Shapley Value}
    \label{subsubsec:understanding_markov_shapley_value}
        In this section, we attempt to explore the interpretation of Markov Shapley value for the active voltage control in power distribution networks. Our method to investigate this novel task is firstly giving a hypothesis for the possible implication of Markov Shapley value to the power network and then verifying it with the sample data collected from testing cases. We hypothesize that each agent $i$'s Markov Shapley value is correlated to $p^{PV}_{i} + 0.01 \times q^{PV}_{i}$ which is the possible physical implication that may influence the active voltage control. Recall that $p^{PV}_{i}$ and $q^{PV}_{i}$ are the active power and the reactive power related to the bus where agent $i$ is located. To avoid the Simpson's paradox \cite{pearl2009causality},\footnote{Simpson's paradox is a phenomenon in statistics that a trend appears in several groups of data but disappears or reverses when the groups are combined. This phenomenon is usually caused by confounding variables, so one direct way to solve it is observing the confounding variables.} we separately calculate the Pearson correlation coefficient for each case with respect to the network topology and the season (both of which are possible confounding variables).
        
        \paragraph{Results and Analysis.} The result are shown in Table \ref{tab:correlation_power_network}. It can be seen from the table that SQDDPG is more correlated to the formula $p^{PV}_{i} + 0.01 \times q^{PV}_{i}$ than SMFPPO in the 33-bus network and the 141-bus network. Although the Markov Shapley value of SMFPPO should be theoretically more accurate than SQDDPG, \textit{the excessive approximation error during learning} of SMFPPO may lead to the difficulty of fitting the correct coalition value functions. Owing to the small number of agents and therefore the limited possibilities of functional forms of marginal contribution, SQDDPG can successfully learn the accurate marginal contributions. On the other hand, when the number of agents increase shown in the cases of 322-bus networks, \textit{the inaccurate representation of marginal contributions} in SQDDPG is extremely exaggerated and even cannot estimate the correlation in the positive way. Despite the burden of fitting, SMFPPO can still show the positive correlation in the cases of 322-bus networks. The result raises a dilemma on \textit{the trade-off between the ease of fitting and the accuracy of representation}.
        \begin{table}[ht!]
        \caption{Pearson correlation coefficient between the Markov Shapley value and the physical implication in the power distribution networks such that $p^{PV}_{i} + 0.01 \times q^{PV}_{i}$ in variant test cases. The results are calculated by 480 samples generated by trained multi-agent models. Note that the p-values of the results are within the accepted significance level (i.e. within 5\%). To enable the table to be neat, we choose not to report it in the table.}
        \centering
        \scalebox{1.0}{
        \begin{tabular}{@{}ccccccc@{}}
        \toprule
               & 33-summer & 33-winter & 141-summer & 141-winter & 322-summer & 322-winter \\ \midrule
        SQDDPG & 0.812        & 0.579        & 0.591         & 0.585         & -0.757        & -0.375        \\
        SMFPPO & 0.513        & 0.247        & 0.325         & 0.203         & 0.175         & 0.072         \\
        \bottomrule
        \end{tabular}
        }
        \label{tab:correlation_power_network}
        \end{table}
        
        \paragraph{Discussion.} The above results show the correlation between the Markov Shapley value and the active power adding small portion of reactive power of the bus where it is located. This finding can verify the effectiveness and representative capability (encoding the physical implication) of Markov Shapley value, given a known physical fact in the power distribution networks (i.e., the droop control sets the local active power as the input). From the reverse direction, it is an evidence to let us believe that Markov Shapley value can potentially ``discover'' more physical facts to an environment that is unknown to human beings. The remaining issue is how the linkage between the Markov Shapley value and the ``language'' that is familiar to human beings can be established, which might be a future direction to enhance the interpretability of Markov Shapley value.

\chapter{Conclusion and Future Work} 
\label{chap:conclusion}
    This thesis aims at solving cooperative multi-agent reinforcement learning using credit assignment, which is a longstanding problem in a cooperative game called global reward game. In a global reward game, agents only receive a global reward when performing decentralised policies to interact with the environment. The global reward encodes a shared goal to which all agents need to collaborate. Technically, the agents should jointly optimize the cumulative discounted global rewards (a.k.a. the global value). During training, if only with the global value as a signal to improve (or learn) agents' policies, the resulting joint policy could be inefficient (e.g., only one agent learns a beneficial policy to solve the task, while the other agents perform like dummies). This motivates us to introduce a concept in cooperative game theory called Shapley value as a credit assignment scheme to fairly assign the credit to each agent.
    
    Since the Shapley value is a concept for a cooperative game model in the cooperative game theory called convex game, it cannot be directly applied to the global reward game. To bridge this gap, we firstly extend the convex game and Shapley value to Markov decision process to fit the setting of global reward game, named as Markov convex game and Markov Shapley value, respectively. We prove that Markov Shapley value is a solution to the Markov convex game, agreeing on a solution concept that we extend from the core (i.e. a solution concept for the convex game) named as Markov core. Then, we show that a global reward game with credit assignment can be represented as a Markov convex game (under the grand coalition) with a payoff distribution scheme. This implies that the Markov Shapley value as a payoff distribution scheme solving a Markov convex game can be applied as a solution to the credit assignment in a global reward game. We further prove that Markov Shapley value inherits the fairness of the original Shapley value, which is one of the motivations of employing Markov Shaley value to solve the global reward game. In more details, the fairness property can assist the interpretation of agents' behaviours.
    
    To achieve the optimal global value, each agent should reach the optimal Markov Shapley value. For the ease of incorporating Markov Shapley value into multi-agent reinforcement learning (MARL), we drive its equivalent form called Markov Shapley Q-value. We propose an optimality criterion called Bellman-Shapley optimality equation and an operator called Shapley-Bellman operator, extended from the well known Bellman optimality equation and Bellman operator in Markov decision process and reinforcement learning. The Shapley-Bellman optimality equation describes an evaluation of the optimal Markov Shapley Q-value and the optimal joint policy. We prove that recursively running the Shapley-Bellman operator leads to the Shapley-Bellman optimality equation. 
    
    Furthermore, we derive the stochastic approximation of Bellman-Shapley operator and prove its convergence to the optimal Markov Shapley Q-value under some technical assumptions. Based on the additional function approximation of Markov Shapley Q-value via deep learning, we propose two practical MARL algorithms such as Shapley Q-learning (SHAQ) and Shapley Q-value deep deterministic policy gradient (SQDDPG). SHAQ is a value-based algorithm which aims at solving the problems with discrete actions, while SQDDPG is a policy-based algorithm which aims at solving the problems with continuous actions. SHAQ and SQDDPG are testified in the benchmark tasks in the community of machine learning research, such as Predator-Prey, Traffic Junction, and StarCraft Multi-Agent Challenges (SMAC). Both algorithms demonstrate superior performance to the baseline algorithms and the ability of interpreting agents' behaviours. To address the issue of the direct approximation of marginal contribution in SQDDPG, we further propose Shapley model-free proximal policy gradient (SMFPPO). In addition, we extend Markov convex game to partial observability, named as partially observable Markov convex game. We also propose Shapley value iteration and Shapley policy gradient, which support the Shapley value based MARL algorithms in implementation to tackle the partially observable problems.
    
    Finally, we apply the SQDDPG and SMFPPO to a real-world problem in energy networks called active voltage control in power distribution networks. In more details, the main objective is controlling the voltage of all buses in distributed manner in a power distribution network within the safety range. To enable the problem to be compatible with MARL, we are the first one to formally formulate this problem as a Dec-POMDP (i.e. matching each component in power distribution networks to a concept in the Dec-POMDP). In simulation, SMFPPO and SQDDPG generally perform better than the baseline algorithms. Moreover, we compare these two algorithms with the traditional control methods: droop control and optimal power flow (OPF). Due to the extra inner-loop control (i.e. time consuming) in the droop control and the additional global observation and system specifications (i.e., unavailable in rapid-changing scenarios caused by the injection of renewable energy) in the OPF, SMFPPO and SQDDPG cannot beat them in the large-scale scenario. In the small-scale scenarios, both MARL algorithms can reach the same control performance, and yield less power loss than the droop control. Moreover, we investigate the connection between the Markov Shapley Q-value and the physical implication of power distribution networks. In our finding, the Markov Shapley Q-value is correlated to the active power adding the reactive power of the bus where it is located. This shows a motivation of studying a general methodology to connect the Markov Shapley Q-value with the physical implications in the future work.
    

\section{Summary of Achievements}
\label{sec:summary_of_achievements}
    The thesis has proposed several innovations which push forward the advances of cooperative multi-agent reinforcement learning with credit assignment. The progress is summarized as follows.
    \begin{itemize}
        \item[1.] \textbf{Establishing the connection between cooperative game theory and global reward game.} Convex game in the cooperative game theory is extended to Markov decision process named Markov convex game, and connected with the global reward game under theoretical guarantees. The background and theory behind the convex game are rich in payoff distribution, which makes credit assignment in the global reward game valid with reasonable and solid foundations.
        
        \item[2.] \textbf{Introducing Shapley value into the global reward game.} For the sake of the connection between the Markov convex game and the global reward game, a payoff distribution scheme in the cooperative game theory called Shapley value is valid to be introduced into the global reward game as a credit assignment scheme. The main advantage of Shapley value is that it can assign credit fairly. To make Shapley value fit the global reward game, we extend it to Markov convex game, named as Markov Shapley value. Besides, we derive an equivalent form of Markov Shapley value called Markov Shapley Q-value, which is a more useful form to multi-agent reinforcement learning.
        
        \item[3.] \textbf{Constructing a theoretical framework to describe the Markov Shapley Q-value.} We extend Bellman optimality equation to describe the optimal Markov Shapley Q-value called Shapley-Bellman optimality equation, whereby the reach of the optimal joint policy is guaranteed. Moreover, we propose Shapley-Bellman operator by extending the well known Bellman operator. Recursively running Shapley-Bellman operator is proved to reach the optimal Markov Shapley Q-value. 
        
        \item[4.] \textbf{Proposing three MARL algorithms based on Markov Shapley Q-value.} Based on the theoretical framework of Markov Shapley Q-value, we propose three multi-agent reinforcement learning algorithms named as Shapley Q-learning (SHAQ), Shapley Q-value deep deterministic policy gradient (SQDDPG) and Shapley model-free proximal policy gradient (SMFPPO). SHAQ belongs to the value-based methods, while SQDDPG and SMFPPO belong to the policy-based methods.
        
        \item[5.] \textbf{Generalising Markov convex game to partial observablility to enable Shapley value based MARL algorithms to solve partially observable tasks.} Markov convex game is generalised to partial observability named as partially observable Markov convex game (POMCG), in which we also propose partially observable Shapley value iteration and partially observable Shapley policy gradient. Thanks to these results, the Shapley value based MARL algorithms can be implemented in some tricks to deal with the partially observable tasks.
        
        \item[6.] \textbf{Formulating the active voltage control problem as a Dec-POMDP.} To solve the active voltage control in power distribution networks by MARL, we formulate the problem as a Dec-POMDP which is a partially observable version of the global reward game. For the sake of the theoretical results of the POMCG, it is feasible to leverage the Shapley value based MARL algorithms to solve the problem. In simulation, it shows that the Shapley value based MARL algorithms can partially solve the active voltage control problem, still facing the challenge to the large-scale network. 
        
        \item[7.] \textbf{Releasing the open-source simulator for the active voltage control in power distribution networks.} To bridge the gap between the machine learning community and the power society, we release an open-source environment that is friendly to MARL for simulating the process of the active voltage control in power distribution networks. The machine learning researchers can therefore easily attempt any state-of-the-art MARL algorithm to solve this real-world problem.
    \end{itemize}












\section{Future Work}
\label{sec:future_work}
    Although the thesis has achieved convincing breakthroughs in multi-agent reinforcement learning with credit assignment and its application in the active voltage control problem, there still exist multiple issues to be tackled in the future work.
    
    \paragraph{Alignment between Markov Shapley value and Physical Implication.} Although we have shown some implication related to Markov Shapley value, these are based on the hypotheses from human beings. In other words, one should give a hypothesis and then use some statistical metric (e.g. Pearson correlation coefficient) to verify the correlation. However, the successes of special cases cannot guarantee the success in general situations. To address this issue, it is possible to incorporate the prior knowledge about the physical systems or other environments into the construction of marginal contributions. Thereby, it is potential to align the existing theoretical results in physical systems or other environments with the concepts of Shapley value. As a result, Markov Shapley value could become understandable to human beings with physical implication.
    
    \paragraph{Fully Decentralised Training under Credit Assignment.} The training paradigm in this thesis mainly concentrates on the centralised training, which means that agents should share the collected data in a central hub during training to gather coalition information for estimating the related term (e.g. the marginal contribution). Nevertheless, in many cases the data gathering is not allowed, e.g., due to the privacy problems and communication difficulties. Therefore, it is necessary to derive a methodology that can well estimate the coalition information based on the history of local information to each agent. To this end, the Markov Shapley value is able to be estimated in the fully decentralised manner. The main challenge here is how a predictor for an agent should be designed, which may be influenced by the uncertainty of other agents' behaviours. Technically, this is highly related to a cutting-edge research problem called generalizable multi-agent reinforcement learning.
    
    \paragraph{Improving Voltage Barrier Function.} About the application aspect of this thesis, one future work that could be improving the design of voltage barrier functions. As our simulation results show, the reward shape such as voltage barrier function seriously affects the result of MARL algorithms. On the other hand, the design of voltage barrier functions could influence the trade-off between the voltage safety and the power loss, which forms a dilemma. For this reason, it is necessary to design a voltage barrier function that solves both the dilemma in the voltage control problem itself and the generalizability to different MARL algorithms.

\section{Publications}
\label{sec:publications}
    The works introduced in this thesis are constituted of the following publications \cite{Wang_2020,wang2021multi,wang2022shaq}.
    \begin{enumerate}
        \item \textbf{Wang, Jianhong, Yuan Zhang, Tae-Kyun Kim, and Yunjie Gu. "Shapley Q-value: A local reward approach to solve global reward games." In Proceedings of the AAAI Conference on Artificial Intelligence, vol. 34, no. 05, pp. 7285-7292. 2020.} $\rightarrow$ This paper extended the convex game to Markov decision process for the first time, meanwhile, it proved that a global reward game can be represented by a Markov convex game under the grand coalition. Moreover, this paper proposed a novel algorithm named SQDDPG, which firstly incorporated Shapley value into multi-agent reinforcement learning heuristically. The contents are covered by Chapter 3 and 4 in this thesis.
        
        \item \textbf{Wang, Jianhong, Wangkun Xu, Yunjie Gu, Wenbin Song, and Tim C. Green. "Multi-agent reinforcement learning for active voltage control on power distribution networks." Advances in Neural Information Processing Systems 34 (2021): 3271-3284.} $\rightarrow$ This paper firstly defined the active voltage control problem in power distribution networks as a decentralised partially observable Markov decision process (Dec-POMDP) \cite{oliehoek2012decentralized} in a rigorous manner. Additionally, we released an open-source environment (simulator) for this problem and evaluated the state-of-the-art multi-agent algorithms on it. The contents are covered by Chapter 5 in this thesis.

        \item \textbf{Wang, Jianhong, Yuan Zhang, Yunjie Gu, and Tae-Kyun Kim. "Shaq: Incorporating shapley value theory into multi-agent q-learning." Advances in Neural Information Processing Systems 35 (2022): 5941-5954.} $\rightarrow$ This paper formally defined and analysed the generalisation of Shapley value to Markov decision process called Markov Shapley value, which filled in the gap left in \cite{Wang_2020}. Furthermore, the Markov Shapley value was incorporated into multi-agent Q-learning called SHAQ. Its optimality and reliability (the convergence property) were rigorously analysed and proved. The contents are covered by Chapter 3 and 4.
    \end{enumerate}

\appendix

\chapter{Mathematical Proofs}

\section{Proof of Marginal Contribution}
\label{sec:proof_of_marginal_contribution}
    \begingroup
    \def\theproposition{\ref{prop:optimal_action_coalition_marginal_contribution}}
    \begin{proposition}
        Agent $i$'s action marginal contribution can be derived as follows:
        \begin{equation}
            \begin{split}
                \Upphi_{i}(\mathbf{s}, a_{i} | \mathcal{C}_{i}) 
                = \max_{\mathbf{a}_{\scriptscriptstyle \mathcal{C}_{i}}} Q^{\pi_{\mathcal{C}_{i}}^{*}}(\mathbf{s}, \mathbf{a}_{\scriptscriptstyle\mathcal{C}_{i} \cup \{i\}})
                - \max_{\mathbf{a}_{\mathcal{C}_{i}}} Q^{\pi_{\mathcal{C}_{i}}^*}(\mathbf{s}, \mathbf{a}_{\scriptscriptstyle\mathcal{C}_{i}}).
            \end{split}
        \end{equation}
    \end{proposition}
    \begin{proof}
        We now rewrite $\max_{\pi_{\mathcal{C}_{i}}} V^{\pi_{\mathcal{C}_{i} \cup \{i\}}}(\mathbf{s})$ as follows:
        \begin{align}
            \max_{\pi_{\mathcal{C}_{i}}} V^{\pi_{\mathcal{C}_{i} \cup \{i\}}}(\mathbf{s}) &= \max_{\pi_{\mathcal{C}_{i}}} \sum_{\mathbf{a}_{\scriptscriptstyle\mathcal{C}_{i} \cup \{i\}}} \pi_{\scriptscriptstyle\mathcal{C}_{i} \cup \{i\}}(\mathbf{a}_{\scriptscriptstyle\mathcal{C}_{i} \cup \{i\}} | \mathbf{s}) \ Q^{\pi_{\mathcal{C}_{i} \cup \{i\}}}(\mathbf{s}, \mathbf{a}_{\scriptscriptstyle\mathcal{C}_{i} \cup \{i\}}) \nonumber \\
            &= \max_{\mathbf{a}_{\mathcal{C}_{i}}} \max_{\pi_{\mathcal{C}_{i}}} Q^{\pi_{\mathcal{C}_{i} \cup \{i\}}}(\mathbf{s}, \mathbf{a}_{\scriptscriptstyle\mathcal{C}_{i} \cup \{i\}}) \nonumber \\
            &\triangleq \max_{\mathbf{a}_{\mathcal{C}_{i}}} Q^{\pi_{\mathcal{C}_{i}}^{*}}(\mathbf{s}, \mathbf{a}_{\scriptscriptstyle\mathcal{C}_{i} \cup \{i\}}).
            \label{eq:V_rewirte_1}
        \end{align}
        
        Similarly, we rewrite $\max_{\pi_{\mathcal{C}_{i}}} V^{\pi_{\mathcal{C}_{i}}}(\mathbf{s})$ as follows:
        \begin{align}
            \max_{\pi_{\mathcal{C}_{i}}} V^{\pi_{\mathcal{C}_{i}}}(\mathbf{s}) = \max_{\mathbf{a}_{\mathcal{C}_{i}}} \max_{\pi_{\mathcal{C}_{i}}} Q^{\pi_{\mathcal{C}_{i}}}(\mathbf{s}, \mathbf{a}_{\scriptscriptstyle\mathcal{C}_{i}}) = \max_{\mathbf{a}_{\mathcal{C}_{i}}} Q^{\pi_{\mathcal{C}_{i}}^{*}}(\mathbf{s}, \mathbf{a}_{\scriptscriptstyle\mathcal{C}_{i}}).
            \label{eq:V_rewrite_2}
        \end{align}
        
        Since $\max_{\pi_{\mathcal{C}_{i}}} V^{\pi_{\mathcal{C}_{i}}}(\mathbf{s})$ is irrelevant to $\mathit{a}_{i}$, by Eq.~\ref{eq:V_rewirte_1} and \ref{eq:V_rewrite_2} we can get that
        \begin{equation}
            \Upphi_{i}(\mathbf{s}, a_{i}|\mathcal{C}_{i}) = \max_{\mathbf{a}_{\mathcal{C}_{i}}} Q^{\pi_{\mathcal{C}_{i}}^{*}}(\mathbf{s}, \mathbf{a}_{\scriptscriptstyle\mathcal{C}_{i} \cup \{i\}})
            - \max_{\mathbf{a}_{\mathcal{C}_{i}}} Q^{\pi_{\mathcal{C}_{i}}^{*}}(\mathbf{s}, \mathbf{a}_{\scriptscriptstyle\mathcal{C}_{i}}).
        \label{eq:coalition_marginal_contribution_action}
        \end{equation}
        
        By Eq.~\ref{eq:coalition_marginal_contribution_action}, we can get the following result such that
        \begin{align}
            \Upphi_{i}^{*}(\mathbf{s}, a_{i}|\mathcal{C}_{i}) &= \max_{\pi_{i}} \Upphi_{i}(\mathbf{s}, a_{i}|\mathcal{C}_{i}) \nonumber \\
            &= \max_{\pi_{i}} \bigg\{ \max_{\mathbf{a}_{\mathcal{C}_{i}}} Q^{\pi_{\mathcal{C}_{i}}^{*}}(\mathbf{s}, \mathbf{a}_{\scriptscriptstyle\mathcal{C}_{i} \cup \{i\}})
            - \max_{\mathbf{a}_{\mathcal{C}_{i}}} Q^{\pi_{\mathcal{C}_{i}}^{*}}(\mathbf{s}, \mathbf{a}_{\scriptscriptstyle\mathcal{C}_{i}}) \bigg\} \nonumber \\
            &= \max_{\pi_{i}} \bigg\{ \max_{\mathbf{a}_{\mathcal{C}_{i}}} \max_{\pi_{\mathcal{C}_{i}}} Q^{\pi_{\mathcal{C}_{i} \cup \{i\}}}(\mathbf{s}, \mathbf{a}_{\scriptscriptstyle\mathcal{C}_{i} \cup \{i\}}) -  \max_{\mathbf{a}_{\mathcal{C}_{i}}} \max_{\pi_{\mathcal{C}_{i}}} Q^{\pi_{\mathcal{C}_{i}}}(\mathbf{s}, \mathbf{a}_{\scriptscriptstyle\mathcal{C}_{i}}) \bigg\} \nonumber \\
            &= \max_{\pi_{i}} \max_{\mathbf{a}_{\mathcal{C}_{i}}} \max_{\pi_{\mathcal{C}_{i}}} Q^{\pi_{\mathcal{C}_{i} \cup \{i\}}}(\mathbf{s}, \mathbf{a}_{\scriptscriptstyle\mathcal{C}_{i} \cup \{i\}}) - \max_{\mathbf{a}_{\mathcal{C}_{i}}} \max_{\pi_{\mathcal{C}_{i}}} Q^{\pi_{\mathcal{C}_{i}}}(\mathbf{s}, \mathbf{a}_{\scriptscriptstyle\mathcal{C}_{i}}) \nonumber \\
            &= \max_{\mathbf{a}_{\mathcal{C}_{i}}} \max_{\pi_{\mathcal{C}_{i} \cup \{i\}}} Q^{\pi_{\mathcal{C}_{i} \cup \{i\}}}(\mathbf{s}, \mathbf{a}_{\scriptscriptstyle\mathcal{C}_{i} \cup \{i\}}) - \max_{\mathbf{a}_{\mathcal{C}_{i}}} \max_{\pi_{\mathcal{C}_{i}}} Q^{\pi_{\mathcal{C}_{i}}}(\mathbf{s}, \mathbf{a}_{\scriptscriptstyle\mathcal{C}_{i}}) \nonumber \\
            &= \max_{\mathbf{a}_{\mathcal{C}_{i}}} Q^{\pi_{\mathcal{C}_{i} \cup \{i\}}^{*}}(\mathbf{s}, \mathbf{a}_{\scriptscriptstyle\mathcal{C}_{i} \cup \{i\}})
            - \max_{\mathbf{a}_{\mathcal{C}_{i}}} Q^{\pi_{\mathcal{C}_{i}}^{{*}}}(\mathbf{s}, \mathbf{a}_{\scriptscriptstyle\mathcal{C}_{i}}).
        \label{eq:optimal_marginal_q}
        \end{align}
        The proof is completed.
    \end{proof}
    
    \begingroup
    \def\theproposition{\ref{prop:condition_coalition_marginal_contribution}}
    \begin{proposition}
        $\forall \mathcal{C}_{i} \ \mathlarger{\mathlarger{\subseteq}} \ \mathcal{N}$ and $\forall \mathbf{s} \in \mathcal{S}$, Eq.~\ref{eq:mcg_assumption} is satisfied if and only if $\max_{\pi_{i}} \Phi_{i}(\mathbf{s}|\mathcal{C}_{i}) \geq 0$.
    \end{proposition}
    \endgroup
    \begin{proof}
        $\forall \mathcal{C}_{i} \ \mathlarger{\mathlarger{\subseteq}} \ \mathcal{N}$ and $\forall \mathbf{s} \in \mathcal{S}$, given that Eq.~\ref{eq:mcg_assumption} is satisfied, with the fact that $\mathcal{C}_{i} \ \mathlarger{\mathlarger{\cap}} \ \{i\} = \emptyset$ we can get the equation such that
        \begin{equation}
            \begin{split}
                \max_{\pi_{\mathcal{C}_{i} \cup \{i\}}} V^{\pi_{\mathcal{C}_{i} \cup \{i\}}}(\mathbf{s}) \geq \max_{\pi_{\mathcal{C}_{i}}} V^{\pi_{\mathcal{C}_{i}}}(\mathbf{s})
                + \max_{\pi_{i}} V^{\pi_{i}}(\mathbf{s}).
            \end{split}
        \label{eq:condition_1}
        \end{equation}
        Since $\max_{\pi_{i}} V^{\pi_{i}}(\mathbf{s}) \geq 0$ by the definition in Markov convex game, we can easily get the equation such that
        \begin{equation}
            \begin{split}
                \max_{\pi_{\mathcal{C}_{i} \cup \{i\}}} V^{\pi_{\mathcal{C}_{i} \cup \{i\}}}(\mathbf{s}) - \max_{\pi_{\mathcal{C}_{i}}} V^{\pi_{\mathcal{C}_{i}}}(\mathbf{s}) \geq 0.
            \end{split}
        \label{eq:condition_2}
        \end{equation}
        Therefore, we can get the equation such that
        \begin{equation}
            \begin{split}
                \max_{\pi_{i}} \Phi_{i}(\mathbf{s}|\mathcal{C}_{i}) \geq 0.
            \end{split}
        \label{eq:condition_3}
        \end{equation}
        With the same conditions, the reverse direction of proof apparently holds by going through from Eq.~\ref{eq:condition_3} to \ref{eq:condition_1}. By Definition \ref{def:shapley_value}, Eq.~\ref{eq:condition_3} determines the range of Markov Shapley value, which is consistent with the range of the coalition value in definition.
    \end{proof}

    \begingroup
    \def\thelemma{\ref{lemm:condition_coalition_marginal_contribution}}
    \begin{lemma}
        The optimal marginal contribution is a solution in the Markov core under a Markov convex game with the grand coalition.
    \end{lemma}
    \endgroup
    \begin{proof}
        The complete proof is as follows.
        
        Firstly, if we would like to prove that the optimal marginal contribution is a payoff distribution scheme in the Markov core (with the grand coalition), we just need to prove that for any intermediate coalition $\mathcal{C} \ \mathlarger{\mathlarger{\subseteq}} \ \mathcal{N}$, the following condition is satisfied such that
        \begin{equation}
            \max_{\pi_{\mathcal{C}}} \Phi(\mathbf{s}|\mathcal{C}) \geq \max_{\pi_{\mathcal{C}}} V^{\pi_{\mathcal{C}}}(\mathbf{s}), \ \forall \mathbf{s} \in \mathcal{S},
        \label{eq:appendix_core}
        \end{equation}
        where $\max_{\pi_{\mathcal{C}}} \Phi(\mathbf{s}|\mathcal{C}) = \sum_{i \in \mathcal{C}} \max_{\pi_{i}} \Phi_{i}(\mathbf{s}|\mathcal{C}_{i})$. 
        
        Suppose for the sake of contradiction that we have $\max_{\pi_{\mathcal{C}}} \Phi(\mathbf{s}|\mathcal{C}) < \max_{\pi_{\mathcal{C}}} V^{\pi_{\mathcal{C}}}(\mathbf{s})$ for some $\mathbf{s} \in \mathcal{S}$ and some coalition $\mathcal{C} = \{ j_{1}, j_{2}, ..., j_{|\mathcal{C}|} \} \ \mathlarger{\mathlarger{\subseteq}} \ \mathcal{N}$, where $\mathit{j}_{n} \in \mathcal{C}$ and $n \in \{1, 2, ..., |\mathcal{C}|\}$. 
        We can assume without the loss of generality that the coalition $\mathcal{C}$ is generated by the permutation $\langle j_{1}, j_{2}, ..., j_{|\mathcal{C}|} \rangle$, i.e., the agents joins in $\mathcal{C}$ following the order $j_{1}, j_{2}, ..., j_{|\mathcal{C}|}$. Now, for each $n \in \{1, 2, ..., |\mathcal{C}|\}$, we have $\{ j_{1}, j_{2}, ..., j_{n-1} \} \ \mathlarger{\mathlarger{\subseteq}} \ \{ 1, 2, ..., j_{n}-1 \}$. Following Eq.~\ref{eq:mcg_assumption}, we can write out the inequality as follows:
        \begin{equation}
            \begin{split}
                \max_{\pi_{\mathcal{C}_{\cup}^{n}}} V^{\pi_{\mathcal{C}_{\cup}^{n}}}(\mathbf{s}) +
                \max_{\pi_{\mathcal{C}_{\cap}^{n}}} V^{\pi_{\mathcal{C}_{\cap}^{n}}}(\mathbf{s}) \geq
                \max_{\pi_{\mathcal{C}_{m}^{n}}} V^{\pi_{\mathcal{C}_{m}^{n}}}(\mathbf{s}) + \max_{\pi_{\mathcal{C}_{k}^{n}}} V^{\pi_{\mathcal{C}_{k}^{n}}}(\mathbf{s}),\\
                \mathcal{C}_{k}^{n} = \{ 1, 2, ..., j_{n}-1 \}, \quad
                \mathcal{C}_{m}^{n} = \{ j_{1}, j_{2}, ..., j_{n} \}, \\
                \mathcal{C}_{\cap}^{n} = \mathcal{C}_{m}^{n} \ \mathlarger{\mathlarger{\cap}} \ \mathcal{C}_{k}^{n} = \{ j_{1}, j_{2}, ..., j_{n-1} \}, \quad
                \mathcal{C}_{\cup}^{n} = \mathcal{C}_{m}^{n} \ \mathlarger{\mathlarger{\cup}} \ \mathcal{C}_{k}^{n} = \{ 1, 2, ..., j_{n} \}. \\
            \end{split}
        \label{eq:appendix_ecg}
        \end{equation}
        
        Next, we rearrange Eq.~\ref{eq:appendix_ecg} and the following inequality is obtained such that
        \begin{equation}
            \begin{split}
                \max_{\pi_{\mathcal{C}_{\cup}^{n}}} V^{\pi_{\mathcal{C}_{\cup}^{n}}}(\mathbf{s}) - \max_{\pi_{\mathcal{C}_{k}^{n}}} V^{\pi_{\mathcal{C}_{k}^{n}}}(\mathbf{s}) \geq
                \max_{\pi_{\mathcal{C}_{m}^{n}}} V^{\pi_{\mathcal{C}_{m}^{n}}}(\mathbf{s}) - \max_{\pi_{\mathcal{C}_{\cap}^{n}}} V^{\pi_{\mathcal{C}_{\cap}^{n}}}(\mathbf{s}),\\
            \end{split}
        \label{eq:appendix_rearrange_ecg}
        \end{equation}
        
        Since we can express $\max_{\pi_{\mathcal{C}}} V^{\pi_{\mathcal{C}}}(\mathbf{s})$ as follows:
        \begin{align}
            \max_{\pi_{\mathcal{C}}} V^{\pi_{\mathcal{C}}}(\mathbf{s}) &= \max_{\pi_{j_{1}}} V^{\pi_{j_{1}}}(\mathbf{s}) - \max_{\pi_{\emptyset}} V^{\pi_{\emptyset}}(\mathbf{s}) \nonumber \\
            &+ \max_{\pi_{\{j_{1}, j_{2}\}}} V^{\pi_{\{j_{1}, j_{2}\}}}(\mathbf{s}) - \max_{\pi_{j_{1}}} V^{\pi_{j_{1}}}(\mathbf{s}) \nonumber \\
            &+ \qquad \qquad \qquad \qquad \vdots \nonumber \\
            &+ \max_{\pi_{\mathcal{C}}} V^{\pi_{\mathcal{C}}}(\mathbf{s}) - \max_{\pi_{\mathcal{C} \backslash \{j_{n}\}}} V^{\pi_{\mathcal{C} \backslash \{j_{n}\}}}(\mathbf{s}).
        \label{eq:appendix_prop1_contradiction_-1}
        \end{align}
        By Definition \ref{def:marginal_contribution} we can obviously get the following equations such that
        \begin{align}
            \Phi_{i}(\mathbf{s}|\mathcal{C}_{i}) = \Phi_{i}(\mathbf{s}|\mathcal{C}_{k}^{n}) &= \max_{\pi_{\mathcal{C}_{k}^{n}}} V^{\pi_{\mathcal{C}_{\cup}^{n}}}(\mathbf{s}) - \max_{\pi_{\mathcal{C}_{k}^{n}}} V^{\pi_{\mathcal{C}_{k}^{n}}}(\mathbf{s}).
        \label{eq:appendix_prop1_contradiction_0}
        \end{align}
        By taking the maximum operator over $\pi_{i}$ to Eq.~\ref{eq:appendix_prop1_contradiction_0}, we can get that
        \begin{align}
            \max_{\pi_{i}} \Phi_{i}(\mathbf{s}|\mathcal{C}_{i}) = \max_{\pi_{i}} \Phi_{i}(\mathbf{s}|\mathcal{C}_{k}^{n}) = \max_{\pi_{\mathcal{C}_{\cup}^{n}}} V^{\pi_{\mathcal{C}_{\cup}^{n}}}(\mathbf{s}) - \max_{\pi_{\mathcal{C}_{k}^{n}}} V^{\pi_{\mathcal{C}_{k}^{n}}}(\mathbf{s}).
        \label{eq:appendix_prop1_contradiction_1}
        \end{align}
        By adding up these inequalities in Eq.~\ref{eq:appendix_rearrange_ecg} for all $\mathcal{C} \ \mathlarger{\mathlarger{\subseteq}} \ \mathcal{N}$ and inserting the results from Eq.~\ref{eq:appendix_prop1_contradiction_-1} and \ref{eq:appendix_prop1_contradiction_1}, we can directly obtain a new inequality such that
        \begin{equation}
            \sum_{i \in \mathcal{C}} \max_{\pi_{i}} \Phi_{i}(\mathbf{s}|\mathcal{C}_{i}) = \max_{\pi_{\mathcal{C}}} \Phi(\mathbf{s}|\mathcal{C}) \geq \max_{\pi_{\mathcal{C}}} V^{\pi_{\mathcal{C}}}(\mathbf{s}).
        \label{eq:appendix_prop1_contradiction}
        \end{equation}
        It is obvious that Eq.~\ref{eq:appendix_prop1_contradiction} contradicts the suppose, so we have showed that Eq.~\ref{eq:appendix_core} always holds for any coalition $\mathcal{C} \ \mathlarger{\mathlarger{\subseteq}} \ \mathcal{N}$. For this reason, we can get the conclusion that marginal contribution is a solution in the Markov core under the Markov convex game with the grand coalition.
    \end{proof}
    
    \begingroup
    \def\theproposition{\ref{prop:marginal_contribution_equal_value_factorisation}}
    \begin{proposition}
        In a Markov convex game with the grand coalition, the marginal contribution satisfies the property of efficiency: $\max_{\pi} V^{\pi}(\mathbf{s}) = \sum_{i \in \mathcal{N}} \max_{\pi_{i}} \Phi_{i}(\mathbf{s}|\mathcal{C}_{i})$.
    \end{proposition}
    \endgroup
    \begin{proof}
        For any $\mathcal{C}_{i} \ \mathlarger{\mathlarger{\subseteq}} \ \mathcal{N} \backslash \{i\}$ and $\mathit{i} \in \mathcal{N}$, according to Eq.~\ref{eq:marginal_contribution_v} we can get the equation such that
        \begin{equation}
            \begin{split}
                \max_{\pi_{i}} \Phi_{i}(\mathbf{s} | \mathcal{C}_{i}) = \max_{\pi_{\mathcal{C}_{i} \cup \{i\}}} V^{\pi_{\mathcal{C}_{i} \cup \{i\}}}(\mathbf{s})
                - \max_{\pi_{\mathcal{C}_{i}}} V^{\pi_{\mathcal{C}_{i}}}(\mathbf{s}),
            \end{split}
        \label{eq:coalition_marginal_contribution_max}
        \end{equation}
        where $\max_{\pi_{\mathcal{C}_{i} \cup \{i\}}} V^{\pi_{\mathcal{C}_{i}}}(\mathbf{s})=\max_{\pi_{\mathcal{C}_{i}}} V^{\pi_{\mathcal{C}_{i}}}(\mathbf{s})$, since the decision of agent $\mathit{i}$ will not affect the value of $\mathcal{C}_{i}$ (i.e., the coalition excluding agent $\mathit{i}$).
        Given the definition that $V^{\pi_{\emptyset}}(\mathbf{s})=0$ and the result from Eq.~\ref{eq:coalition_marginal_contribution_max}, by Assumption \ref{assm:assumption_for_joint_policy_factorisation} we can get the equations such that
        \begin{align}
            &\quad \ \max_{\pi} V^{\pi}(\mathbf{s}) \nonumber \\ 
            &= \max_{\pi_{\{j_{1}\}}} V^{\pi_{\{j_{1}\}}}(\mathbf{s}) - \max_{\pi_{\emptyset}} V^{\pi_{\emptyset}}(\mathbf{s}) \nonumber \\
            &+ \max_{\pi_{\{j_{1}, j_{2}\}}} V^{\pi_{\{j_{1}\}}}(\mathbf{s}) - \max_{\pi_{\{j_{1}\}}} V^{\pi_{\{j_{1}\}}}(\mathbf{s}) \nonumber \\
            &+ \qquad \qquad \qquad \qquad \vdots \nonumber \\
            &+ \max_{\pi} V^{\pi}(\mathbf{s}) - \max_{\pi_{\mathcal{N} \backslash \{j_{n}\}}} V^{\pi_{\mathcal{N} \backslash \{j_{n}\}}}(\mathbf{s})
            = \sum_{i \in \mathcal{N}} \max_{\pi_{i}} \Phi_{i}(\mathbf{s}|\mathcal{C}_{i}).
        \end{align}
    \end{proof}

\section{Proof of Markov Shapley Value}
\label{sec:proof_of_markov_shapley_value}
    \begingroup
    \def\theproposition{\ref{prop:shapley_value_properties}}
    \begin{proposition}
        Markov Shapley value possesses properties as follows: (i) identifiability of dummy agents: $V_{i}^{\phi}(\mathbf{s}) = 0$; (ii) efficiency: $\max_{\pi} V^{\pi}(\mathbf{s}) = \sum_{i \in \mathcal{N}} \max_{\pi_{i}} V_{i}^{\phi}(\mathbf{s})$; (iii) reflecting the contribution; and (iv) symmetry.
    \end{proposition}
    \endgroup
    \begin{proof}
        The complete proof is as follows. Since the marginal contribution is an implementation to fulfil (iii) and Markov Shapley value is actually a convex combination of marginal contributions, (iii) is still preserved. We will next prove the (i), followed by (ii) and (iv). For any agent $\mathit{i} \in \mathcal{N}$ and any state $\mathbf{s} \in \mathcal{S}$, its Markov Shapley value denoted as $V_{i}^{\phi}(\mathbf{s})$.
        
        \textbf{Proof of (i):} Let us define $\Pi(\mathcal{N})$ as the set of all permutations of agents. Suppose that an arbitrary agent $i$ is a dummy agent for an arbitrary state $\mathbf{s} \in \mathcal{S}$. For any permutation $\mathit{m} \in \Pi(\mathcal{N})$ of agents to form the grand coalition, by Assumption \ref{assm:dummy_agent} we have $\max_{\pi_{\mathcal{C}_{i}^{m}}} V^{\pi_{\mathcal{C}_{i}^{m}}}(\mathbf{s})=\max_{\pi_{\mathcal{C}_{i}^{m}}} V^{\pi_{\mathcal{C}_{i}^{m} \cup \{i\}}}(\mathbf{s})$, thereby $\Phi_{i}(\mathbf{s}|\mathcal{C}_{i}^{m})=0$, where $\mathcal{C}_{i}^{m}$ denotes the intermediate coalition generated from permutation $m$ that agent $i$ would join. Also, the above analysis is valid for all permutations of agents to form the grand coalition. By Definition \ref{def:shapley_value}, it is not difficult to see that the dummy agent's Markov Shapley value will be 0 such that $V_{i}^{\phi}(\mathbf{s}) = 0$. The proof of (i) completes.
        
        \textbf{Proof of (ii):} The objective is proving that Markov Shapley value satisfies the following equation such that
        \begin{equation*}
            \max_{\pi} V^{\pi}(\mathbf{s}) = \sum_{i \in \mathcal{N}} \max_{\pi_{i}} V_{i}^{\phi}(\mathbf{s}), \quad \forall \mathbf{s} \in \mathcal{S}.
        \end{equation*}
        
        By the result from Proposition \ref{prop:marginal_contribution_equal_value_factorisation} and Assumption \ref{assm:assumption_for_joint_policy_factorisation}, for an arbitrary permutation $\mathit{m} \in \Pi(\mathcal{N})$ we can get the equation such that
        \begin{equation*}
            \max_{\pi} V^{\pi}(\mathbf{s}) = \sum_{i \in \mathcal{N}} \max_{\pi_{i}} \Phi_{i}(\mathbf{s}|\mathcal{C}_{i}^{m}), \quad \forall \mathbf{s} \in \mathcal{S},
        \end{equation*}
        
        where $\mathcal{C}_{i}^{m}$ denotes the intermediate coalition generated from permutation $m$ that agent $i$ would join and $\Phi_{i}(\mathbf{s}|\mathcal{C}_{i}^{m})$ is the corresponding marginal contribution. If we consider all possible permutations of agents to form the grand coalition and add all these inequalities, we can get the following equation such that
        \begin{equation*}
            \sum_{m \in \Pi(\mathcal{N})} \max_{\pi} V^{\pi}(\mathbf{s}) = \sum_{m \in \Pi(\mathcal{N})} \sum_{i \in \mathcal{N}} \max_{\pi_{i}} \Phi_{i}(\mathbf{s}|\mathcal{C}_{i}^{m}), \quad \forall \mathbf{s} \in \mathcal{S}.
        \end{equation*}
        
        By dividing $|\mathcal{N}|!$ on the both sides, we can get that
        \begin{equation}
            \frac{1}{|\mathcal{N}|!} \sum_{m \in \Pi(\mathcal{N})} \max_{\pi} V^{\pi}(\mathbf{s}) = \frac{1}{|\mathcal{N}|!} \sum_{i \in \mathcal{N}} \sum_{m \in \Pi(\mathcal{N})} \max_{\pi_{i}} \Phi_{i}(\mathbf{s}|\mathcal{C}_{i}^{m}), \quad \forall \mathbf{s} \in \mathcal{S}.
        \label{eq:shapley_property_3}
        \end{equation}
        
        Next, to ease life we start from the LHS of Eq.~\ref{eq:shapley_property_3}. We directly get the following equation such that
        \begin{equation}
            \frac{1}{|\mathcal{N}|!} \sum_{m \in \Pi(\mathcal{N})} \max_{\pi} V^{\pi}(\mathbf{s}) = \frac{1}{|\mathcal{N}|!} \cdot |\mathcal{N}|! \cdot \max_{\pi} V^{\pi}(\mathbf{s}) = \max_{\pi} V^{\pi}(\mathbf{s}).
        \label{eq:shapley_property_3_left}
        \end{equation}
        
        Now, we start processing the RHS of Eq.~\ref{eq:shapley_property_3}. By rearranging it, we can get the equations such that
        \begin{align}
            \frac{1}{|\mathcal{N}|!} \sum_{i \in \mathcal{N}} \sum_{m \in \Pi(\mathcal{N})} \max_{\pi_{i}} \Phi_{i}(\mathbf{s}|\mathcal{C}_{i}^{m}) &= \sum_{i \in \mathcal{N}} \frac{1}{|\mathcal{N}|!} \sum_{m \in \Pi(\mathcal{N})} \max_{\pi_{i}} \Phi_{i}(\mathbf{s}|\mathcal{C}_{i}^{m}) \nonumber \\
            &\quad (\text{The identical $\mathcal{C}_{i}^{m}$ in different permutations is written as $\mathcal{C}_{i}$} \nonumber \\
            &\quad \ \ \text{and we can rearrange the equation as follows.}) \nonumber \\
            &= \sum_{i \in \mathcal{C}} \frac{1}{|\mathcal{N}|!} \sum_{\mathcal{C}_{i} \subseteq \mathcal{N} \backslash \{i\}} |\mathcal{C}_{i}|!(|\mathcal{N}|-|\mathcal{C}_{i}|-1)! \cdot \max_{\pi_{i}} \Phi_{i}(\mathbf{s}|\mathcal{C}_{i}) \nonumber \\
            &= \sum_{i \in \mathcal{N}} \sum_{\mathcal{C}_{i} \subseteq \mathcal{N} \backslash \{i\}} \frac{|\mathcal{C}_{i}|!(|\mathcal{N}|-|\mathcal{C}_{i}|-1)!}{|\mathcal{N}|!} \cdot \max_{\pi_{i}} \Phi_{i}(\mathbf{s}|\mathcal{C}_{i}).
        \label{eq:shapley_property_3_right_1}
        \end{align}
        
        By Assumption \ref{assm:max_shapley_value}, we can get the following equations such that
        \begin{align}
            \sum_{i \in \mathcal{N}} \sum_{\mathcal{C}_{i} \subseteq \mathcal{N} \backslash \{i\}} \frac{|\mathcal{C}_{i}|!(|\mathcal{N}|-|\mathcal{C}_{i}|-1)!}{|\mathcal{N}|!} \cdot \max_{\pi_{i}} \Phi_{i}(\mathbf{s}|\mathcal{C}_{i})
            = \sum_{i \in \mathcal{N}} \max_{\pi_{i}} V_{i}^{\phi}(\mathbf{s})
        \label{eq:shapley_property_3_right_2}
        \end{align}
        
        Inserting the results from Eq.~\ref{eq:shapley_property_3_left} and \ref{eq:shapley_property_3_right_2} to Eq.~\ref{eq:shapley_property_3}, we can get the equation such that
        \begin{equation*}
            \max_{\pi} V^{\pi}(\mathbf{s}) = \sum_{i \in \mathcal{N}} \max_{\pi_{i}} V_{i}^{\phi}(\mathbf{s}), \quad \forall \mathbf{s} \in \mathcal{S}.
        \end{equation*}
        Therefore, the proof for (ii) completes.
        
        \textbf{Proof of (iv):} We would like to prove that if two agents are symmetric for an arbitrary state $\mathbf{s} \in \mathcal{S}$, then their optimal Markov Shapley values should be equal. As Assumption \ref{assm:symmetry} illustrates, suppose that agents $i$ and $j$ are symmetric for an arbitrary state $\mathbf{s} \in \mathcal{S}$, $V^{\pi_{\mathcal{C} \cup \{i\}}}(\mathbf{s}) = V^{\pi_{\mathcal{C} \cup \{j\}}}(\mathbf{s})$ for any coalitions $\mathcal{C} \ \mathlarger{\mathlarger{\subseteq}} \ \mathcal{N} \backslash \{i, j\}$. Given an arbitrary permutation $m \in \Pi(\mathcal{N})$, let $m'$ denote the permutation obtained by exchanging $i$ and $j$ such that $\mathcal{C}_{i}^{m} = \mathcal{C}_{j}^{m'}$, $\mathcal{C}_{i}^{m'} = \mathcal{C}_{j}^{m}$ and $\mathcal{C}_{l}^{m'} = \mathcal{C}_{l}^{m}, \forall l \neq i, j$. Next, we aim to prove that $\max_{\pi_{i}} \Phi_{i}(\mathbf{s} | \mathcal{C}_{i}^{m}) = \max_{\pi_{j}} \Phi_{j}(\mathbf{s} | \mathcal{C}_{j}^{m'})$, for the state $\mathbf{s}$.
        
        We first suppose that $i$ precedes $j$ in $m$. Then we have $\mathcal{C}_{i}^{m} = \mathcal{C}_{j}^{m'}$. Setting $\mathcal{C} = \mathcal{C}_{i}^{m} = \mathcal{C}_{j}^{m'}$, for the state $\mathbf{s}$ we can obtain that
        \begin{equation*}
            \begin{split}
                \max_{\pi_{i}} \Phi_{i}(\mathbf{s} | \mathcal{C}_{i}^{m}) = \max_{\pi_{\mathcal{C} \cup \{i\}}} V^{\pi_{\mathcal{C} \cup \{i\}}}(\mathbf{s}) - \max_{\pi_{\mathcal{C}}} V^{\pi_{\mathcal{C}}}(\mathbf{s}), \\
                \max_{\pi_{j}} \Phi_{j}(\mathbf{s} | \mathcal{C}_{j}^{m'}) = \max_{\pi_{\mathcal{C} \cup \{j\}}} V^{\pi_{\mathcal{C} \cup \{j\}}}(\mathbf{s}) - \max_{\pi_{\mathcal{C}}} V^{\pi_{\mathcal{C}}}(\mathbf{s}).
            \end{split}
        \end{equation*}
        By symmetry, we have $V^{\pi_{\mathcal{C} \cup \{i\}}}(\mathbf{s}) = V^{\pi_{\mathcal{C} \cup \{j\}}}(\mathbf{s})$, which directly implies that $\max_{\pi_{i}} \Phi_{i}(\mathbf{s} | \mathcal{C}_{i}^{m}) = \max_{\pi_{j}} \Phi_{j}(\mathbf{s} | \mathcal{C}_{j}^{m'})$.
        
        Second, we suppose that $j$ precedes $i$ in $m$. Setting $\mathcal{C} = \mathcal{C}_{i}^{m} \backslash \{j\}$, for the state $\mathbf{s}$ we have 
        \begin{equation*}
            \begin{split}
                \max_{\pi_{i}} \Phi_{i}(\mathbf{s} | \mathcal{C}_{i}^{m}) = \max_{\pi_{\mathcal{C} \cup \{j\} \cup \{i\}}} V^{\pi_{\mathcal{C} \cup \{j\} \cup \{i\}}}(\mathbf{s}) - \max_{\pi_{\mathcal{C} \cup \{j\}}} V^{\pi_{\mathcal{C} \cup \{j\}}}(\mathbf{s}), \\
                \max_{\pi_{j}} \Phi_{j}(\mathbf{s} | \mathcal{C}_{j}^{m'}) = \max_{\pi_{\mathcal{C} \cup \{j\} \cup \{i\}}} V^{\pi_{\mathcal{C} \cup \{j\} \cup \{i\}}}(\mathbf{s}) - \max_{\pi_{\mathcal{C} \cup \{i\}}} V^{\pi_{\mathcal{C} \cup \{i\}}}(\mathbf{s}).
            \end{split}
        \end{equation*}
        
        Since $\mathcal{C} \ \mathlarger{\mathlarger{\subseteq}} \ \mathcal{N} \backslash \{i, j\}$, by symmetry we have $V^{\pi_{\mathcal{C} \cup \{j\}}}(\mathbf{s}) = V^{\pi_{\mathcal{C} \cup \{i\}}}(\mathbf{s})$ and thus $\max_{\pi_{i}} \Phi_{i}(\mathbf{s} | \mathcal{C}_{i}^{m}) = \max_{\pi_{j}} \Phi_{j}(\mathbf{s} | \mathcal{C}_{j}^{m'})$. Therefore, we have proved that $\max_{\pi_{i}} \Phi_{i}(\mathbf{s} | \mathcal{C}_{i}^{m}) = \max_{\pi_{j}} \Phi_{j}(\mathbf{s} | \mathcal{C}_{j}^{m'})$ for any $m \in \Pi(\mathcal{N})$. It is not difficult to observe that $m \mapsto m'$ is a one-to-one mapping, so $\Pi(\mathcal{N}) = \left\{ m' | m \in \Pi(\mathcal{N}) \right\}$. 
        
        By Assumption \ref{assm:max_shapley_value}, for an arbitrary state $\mathbf{s} \in \mathcal{S}$ wherein agents are symmetric, we can directly have 
        \begin{equation*}
            \begin{split}
                \max_{\pi_{i}} V^{\phi}_{i}(\mathbf{s}) &= \sum_{\mathcal{C}_{i} \ \mathlarger{\mathlarger{\subseteq}} \ \mathcal{N} \backslash \{i\} } \frac{|\mathcal{C}_{i}|!(|\mathcal{N}|-|\mathcal{C}_{i}|-1)!}{|\mathcal{N}|!} \cdot \max_{\pi_{i}} \Phi_{i}(\mathbf{s} | \mathcal{C}_{i}) \\
                &= \frac{1}{|\mathcal{N}|!} \sum_{m \in \Pi(\mathcal{N})} \max_{\pi_{i}} \Phi_{i}(\mathbf{s} | \mathcal{C}_{i}^{m}) \\
                &= \frac{1}{|\mathcal{N}|!} \sum_{m' \in \Pi(\mathcal{N})} \max_{\pi_{j}} \Phi_{j}(\mathbf{s} | \mathcal{C}_{j}^{m'}) \\
                &= \sum_{\mathcal{C}_{j} \ \mathlarger{\mathlarger{\subseteq}} \ \mathcal{N} \backslash \{j\} } \frac{|\mathcal{C}_{j}|!(|\mathcal{N}|-|\mathcal{C}_{j}|-1)!}{|\mathcal{N}|!} \cdot \max_{\pi_{j}} \Phi_{j}(\mathbf{s} | \mathcal{C}_{j}) \\
                &= \max_{\pi_{j}} V^{\phi}_{j}(\mathbf{s}).
            \end{split}
        \end{equation*}
        The proof of (iv) completes.
    \end{proof}
    
    \begingroup
    \def\theproposition{\ref{prop:equiv_credit_assignment}}
    \begin{proposition}
        For any $\mathbf{s} \in \mathcal{S}$ and $\mathit{a}_{i}^{*} = \arg\max_{a_{i}} Q^{\phi^{*}}_{i}(\mathbf{s}, a_{i})$, we have a solution $w_{i}(\mathbf{s}, a_{i}^{*}) = 1 / |\mathcal{N}|$.
    \end{proposition}
    \endgroup
    \begin{proof}
        First, according to the Bellman's principle of optimality \cite{bellman1952theory,sutton2018reinforcement}, we can write out the Bellman optimality equation for the optimal global Q-value such that
        \begin{equation}
            Q^{\pi^{*}}(\mathbf{s}, \mathbf{a}) = \sum_{\mathbf{s}'} Pr(\mathbf{s}' | \mathbf{s}, \mathbf{a}) \big[ R + \gamma \max_{\mathbf{a}} Q^{\pi^{*}}(\mathbf{s}', \mathbf{a}) \big].
        \label{eq:joint_bellman_optimality_equation}
        \end{equation}
        
        For convenience, we only consider the finite state space and action space here. By the property of efficiency (i.e., (2) in Proposition \ref{prop:shapley_value_properties}), we can get the approximation of the optimal global Q-value w.r.t. optimal actions such that
        \begin{equation}
            \max_{\mathbf{a}} Q^{\pi^{*}}(\mathbf{s}', \mathbf{a}) = \sum_{i \in \mathcal{N}} \max_{a_{i}} Q_{i}^{\phi^{*}}(\mathbf{s}', a_{i}).
            \label{eq:optimal_shapley_q_condition}
        \end{equation}
        
        Suppose that for all $\mathbf{s} \in \mathcal{S}$ and $a_{i} \in \mathcal{A}_{i}$, for each agent $i$ there exists bounded $\mathit{w}_{i}(\mathbf{s}, a_{i}) > 0$ and $b_{i}(\mathbf{s}) \geq 0$ that can project $Q^{\pi^{*}}(\mathbf{s}, \mathbf{a})$ onto the space of $Q^{\phi^{*}}_{i}(\mathbf{s}, a_{i})$ such that
        \begin{equation}
            Q^{\phi^{*}}_{i}(\mathbf{s}, a_{i}) = w_{i}(\mathbf{s}, a_{i}) \ Q^{\pi^{*}}(\mathbf{s}, \mathbf{a}) - b_{i}(\mathbf{s}).
        \label{eq:global_q_projection}
        \end{equation}
        
        If we denote $\mathbf{w}(\mathbf{s}, \mathbf{a}) = [w_{i}(\mathbf{s}, a_{i})]^{\top} \in \mathbb{R}^{\scriptscriptstyle|\mathcal{N}|}_{> 0}$, $\mathbf{b}(\mathbf{s}) = [b_{i}(\mathbf{s})]^{\top} \in \mathbb{R}^{\scriptscriptstyle|\mathcal{N}|}_{\geq 0}$ and $\mathbf{Q}^{\phi^{*}}(\mathbf{s}, \mathbf{a}) = [Q^{\phi^{*}}_{i}(\mathbf{s}, a_{i})]^{\top} \in \mathbb{R}^{\scriptscriptstyle|\mathcal{N}|}_{\geq 0}$, given Eq.~\ref{eq:global_q_projection} we can write that
        \begin{equation}
            \mathbf{Q}^{\phi^{*}}(\mathbf{s}, \mathbf{a}) = \mathbf{w}(\mathbf{s}, \mathbf{a}) \ Q^{\pi^{*}}(\mathbf{s}, \mathbf{a}) - \mathbf{b}(\mathbf{s}).
        \label{eq:global_q_projection_joint}
        \end{equation}
        
        Besides, we suppose that $\sum_{i \in \mathcal{N}} w_{i}(\mathbf{s}, a_{i})^{-1} b_{i}(\mathbf{s}) = 0$.
        
        Combined with Eq.~\ref{eq:optimal_shapley_q_condition} and \ref{eq:global_q_projection_joint}, we can rewrite Eq.~\ref{eq:joint_bellman_optimality_equation} to the equation as follows:
        \begin{equation}
            \begin{split}
                \mathbf{Q}^{\phi^{*}}(\mathbf{s}, \mathbf{a}) = \mathbf{w}(\mathbf{s}, \mathbf{a}) \sum_{\mathbf{s}'} Pr(\mathbf{s}' | \mathbf{s}, \mathbf{a}) \big[
                R \ + \ 
                \gamma \sum_{i \in \mathcal{N}} \max_{a_{i}} Q_{i}^{\phi^{*}}(\mathbf{s}', a_{i}) \big] - \mathbf{b}(\mathbf{s}).
            \end{split}
        \end{equation}
        
        From Eq.~\ref{eq:global_q_projection}, we know that $w_{i}(\mathbf{s}, a_{i}) > 0$. Therefore, we can rewrite Eq.~\ref{eq:global_q_projection} to the following equation such that
        \begin{equation}
            w_{i}(\mathbf{s}, a_{i})^{-1} \ \big( Q^{\phi^{*}}_{i}(\mathbf{s}, a_{i}) + b_{i}(\mathbf{s}) \big) = Q^{\pi^{*}}(\mathbf{s}, \mathbf{a}).
        \label{eq:projection_global_q_optimality_trans}
        \end{equation}
        
        If we sum up Eq.~\ref{eq:projection_global_q_optimality_trans} for all agents, we can obtain that
        \begin{equation}
            \sum_{i \in \mathcal{N}} w_{i}(\mathbf{s}, a_{i})^{-1} \ \big( Q^{\phi^{*}}_{i}(\mathbf{s}, a_{i}) + b_{i}(\mathbf{s}) \big) = |\mathcal{N}| \ Q^{\pi^{*}}(\mathbf{s}, \mathbf{a}).
        \end{equation}
        
        Since $\sum_{i \in \mathcal{N}} w_{i}(\mathbf{s}, a_{i})^{-1} b_{i}(\mathbf{s}) = 0$, we can get the following equation such that
        \begin{equation}
            \sum_{i \in \mathcal{N}} \frac{1}{|\mathcal{N}| \ w_{i}(\mathbf{s}, a_{i})} \cdot Q^{\phi^{*}}_{i}(\mathbf{s}, a_{i}) = \ Q^{\pi^{*}}(\mathbf{s}, \mathbf{a}).
        \label{eq:global_q_optimality_trans_1}
        \end{equation}
        
        Inserting Eq.~\ref{eq:optimal_shapley_q_condition} into Eq.~\ref{eq:global_q_optimality_trans_1}, we can get the following equation such that
        \begin{equation}
            \max_{\mathbf{a}} \sum_{i \in \mathcal{N}} \frac{1}{|\mathcal{N}| \ w_{i}(\mathbf{s}, a_{i})} \cdot Q^{\phi^{*}}_{i}(\mathbf{s}, a_{i}) = \sum_{i \in \mathcal{N}} \max_{a_{i}} Q_{i}^{\phi^{*}}(\mathbf{s}, a_{i}).
        \end{equation}
        
        Since $\mathbf{a} = \mathlarger{\mathlarger{\times}}_{\scriptscriptstyle i \in \mathcal{N}} a_{i}$, we can get that
        \begin{equation}
            \sum_{i \in \mathcal{N}} \max_{a_{i}} \frac{1}{|\mathcal{N}| \ w_{i}(\mathbf{s}, a_{i})} \cdot Q^{\phi^{*}}_{i}(\mathbf{s}, a_{i}) = \sum_{i \in \mathcal{N}} \max_{a_{i}} Q_{i}^{\phi^{*}}(\mathbf{s}, a_{i}).
        \end{equation}
        
        It is apparent that $\forall \mathbf{s} \in \mathcal{S}$ and $\mathit{a}_{i}^{*} = \arg\max_{a_{i}} Q^{\phi^{*}}_{i}(\mathbf{s}, a_{i})$, we have a solution $w_{i}(\mathbf{s}, a_{i}^{*}) = 1 / |\mathcal{N}|$, when we consider the terms of each agent on the LHS and RHS are adequately equal such that 
        \begin{equation*}
            \max_{a_{i}} \frac{1}{|\mathcal{N}| \ w_{i}(\mathbf{s}, a_{i})} \cdot Q^{\phi^{*}}_{i}(\mathbf{s}, a_{i}) = \max_{a_{i}} Q_{i}^{\phi^{*}}(\mathbf{s}, a_{i}), \quad \forall i.
        \end{equation*}
        
        On the other hand, there exist other solutions such that 
        \begin{equation*}
            \max_{a_{i}} \frac{1}{|\mathcal{N}| \ w_{i}(\mathbf{s}, a_{i})} \cdot Q^{\phi^{*}}_{i}(\mathbf{s}, a_{i}) = \max_{a_{j}} Q_{j}^{\phi^{*}}(\mathbf{s}, a_{j}), \quad \exists i \neq j.
        \end{equation*}
        However, this class of solutions is unable to be solved in the analytic form when the further information about the task in addition to the Markov Shapley value or the rules of credit assignment is unknown.
    \end{proof}
    
\section{Proof of Shapley-Bellman Operator}
\label{sec:proof_of_shapley-bellman_operator}
    \setcounter{lemma}{6}
    \begin{lemma}[\cite{dales2003introduction}]
    \label{lem:banach_algebra}
        A set of real matrices $\mathcal{M}$ with a sub-multiplicative norm is a Banach algebra and a non-empty complete metric space where the metric is induced by the sub-multiplicative norm. A sub-multiplicative norm $|| \cdot ||$ is a norm satisfying the following inequality such that
        \begin{equation*}
            \forall \mathbf{A}, \mathbf{B} \in \mathcal{M}: ||\mathbf{A} \mathbf{B}|| \leq ||\mathbf{A}|| \ ||\mathbf{B}||.
        \end{equation*}
    \end{lemma}
    
    \begin{lemma}
        For a set of real matrices $\mathcal{M}$, given an arbitrary matrix $\mathbf{A} = [a_{ij}] \in \mathbb{R}^{m \times n}$, $||\mathbf{A}||_{1} = \max_{1 \leq j \leq n} \sum_{1 \leq i \leq m} |a_{ij}|$ is a sub-multiplicative norm.
    \label{lem:matrices_norm_1_metric_space}
    \end{lemma}
    \begin{proof}
        First, we select two arbitrary matrices belonging to $\mathcal{M}$, i.e. $\mathbf{A}=[a_{ik}] \in \mathbb{R}^{m\times r}$ and $\mathbf{B}=[b_{kj}] \in \mathbb{R}^{r\times n}$. Then, we start proving that $||\cdot||_{1}$ is a sub-multiplicative norm as follows:
        \begin{align*}
            ||\mathbf{A} \mathbf{B}||_{1} &= \bigg|\bigg| \bigg[ \sum_{1 \leq k \leq r} a_{ik} b_{kj} \bigg] \bigg| \bigg|_{1} \\
            &= \max_{1 \leq j \leq n} \sum_{1 \leq i \leq m} \bigg| \sum_{1 \leq k \leq r} a_{ik} b_{kj} \bigg| \\
            &\quad (\text{By triangle inequality, we can obtain the following inequality.}) \\
            &\leq \max_{1 \leq j \leq n} \sum_{1 \leq i \leq m} \sum_{1 \leq k \leq r} \big| a_{ik} b_{kj} \big| \\
            &= \max_{1 \leq j \leq n} \sum_{1 \leq i \leq m} \sum_{1 \leq k \leq r} \big| a_{ik} \big| \ \big| b_{kj} \big| \\
            &= \max_{1 \leq j \leq n} \sum_{1 \leq k \leq r} \sum_{1 \leq i \leq m} \big| a_{ik} \big| \ \big| b_{kj} \big| \\
            &= \max_{1 \leq j \leq n} \sum_{1 \leq k \leq r} \big| b_{kj} \big| \sum_{1 \leq i \leq m} \big| a_{ik} \big| \\
            &\leq \big|\big| \mathbf{B} \big| \big|_{1} \max_{1 \leq k \leq r} \sum_{1 \leq i \leq m} \big| a_{ik} \big| \\
            &= \big|\big| \mathbf{B} \big| \big|_{1} \big|\big| \mathbf{A} \big| \big|_{1}
            = \big|\big| \mathbf{A} \big| \big|_{1} \big|\big| \mathbf{B} \big| \big|_{1}.
        \end{align*}
        Therefore, we prove that given an arbitrary matrix $\mathbf{A} = [a_{ij}] \in \mathbb{R}^{m \times n}$, $||\mathbf{A}||_{1} = \max_{1 \leq j \leq n} \sum_{1 \leq i \leq m} |a_{ij}|$ is a sub-multiplicative norm.
    \end{proof}
    
    \begingroup
    \def\thelemma{\ref{lemm:shapley_q_contraction_mapping}}
    \begin{lemma}
        For all $\mathbf{s} \in \mathcal{S}$ and $\mathbf{a} \in \mathcal{A}$, Shapley-Bellman operator is a contraction mapping in a non-empty complete metric space when $\max_{\mathbf{s}} \big\{ \sum_{i \in \mathcal{N}} \max_{a_{i}} w_{i}(\mathbf{s}, a_{i}) \big\} < \frac{1}{\gamma}$.
    \end{lemma}
    \endgroup
    \begin{proof}
        To ease life, we firstly define some variables that will be used for proof such that 
        \begin{align*}
            &\mathbf{Q}^{\phi} = \times_{i \in \mathcal{N}} Q_{i}^{\phi} \in \mathbb{R}^{|\mathcal{N}|\times|\mathcal{S}||\mathcal{A}|}, \\
            &\mathbf{w} \in \mathbb{R}^{|\mathcal{N}|\times|\mathcal{S}||\mathcal{A}|}, \\
            &Pr \in \mathbb{R}^{|\mathcal{S}||\mathcal{A}|\times|\mathcal{S}|}, \\
            &\mathbf{1} = [1,1,...,1]^{\top},
        \end{align*}
        
        where $\mathcal{A}=\mathlarger{\mathlarger{\times}}_{i \in \mathcal{N}} \mathcal{A}_{i}$. Then, for an arbitrary matrix $\mathbf{A} \in \mathbb{R}^{m \times n}$, we define the $||\cdot||_{1}$ for the induced matrix norm such that
        \begin{equation*}
            ||\mathbf{A}||_{1} = \max_{1 \leq j \leq n} \sum_{1 \leq i \leq m} |a_{ij}|,
        \end{equation*}
        
        where $\mathit{a}_{ij}$ is an arbitrary element in $\mathbf{A}$. By Lemma \ref{lem:matrices_norm_1_metric_space}, $||\cdot||_{1}$ defined here is a sub-multiplicative norm. By Lemma \ref{lem:banach_algebra}, the set of real matrices $\mathbb{R}^{|\mathcal{N}|\times |\mathcal{S}| |\mathcal{A}|}$ with the norm $||\cdot||_{1}$ is a Banach algebra and a non-empty complete metric space with the metric induced by $||\cdot||_{1}$.
        
        To show that the operator $\mathlarger{\Upsilon}$ is a contraction mapping in the supremum norm, we just need to show that for any $\mathbf{Q}^{\phi}_{1} = \times_{i \in \mathcal{N}} \big( Q_{i}^{\phi} \big)_{1} \in \mathbb{R}^{|\mathcal{N}|\times |\mathcal{S}| |\mathcal{A}|}$ and $\mathbf{Q}^{\phi}_{2} = \times_{i \in \mathcal{N}} \big( Q_{i}^{\phi} \big)_{2} \in \mathbb{R}^{|\mathcal{N}|\times |\mathcal{S}| |\mathcal{A}|}$, we have $|| \mathlarger{\Upsilon} \mathbf{Q}^{\phi}_{1}  - \mathlarger{\Upsilon} \mathbf{Q}^{\phi}_{2} ||_{1} \leq \delta ||\mathbf{Q}^{\phi}_{1} - \mathbf{Q}^{\phi}_{2}||_{1}$, where $\delta \in (0, 1)$.
        \begin{align*}
            &\quad || \mathlarger{\Upsilon} \mathbf{Q}^{\phi}_{1}  - \mathlarger{\Upsilon} \mathbf{Q}^{\phi}_{2} ||_{1} \\
            &= \max_{\mathbf{s}, \mathbf{a}} \mathbf{1}^{\top} \bigg| \mathbf{w}(\mathbf{s}, \mathbf{a}) \ \sum_{\mathbf{s}' \in \mathcal{S}} Pr(\mathbf{s}'|\mathbf{s}, \mathbf{a}) \big[ R(\mathbf{s}, \mathbf{a}) + \gamma \sum_{i \in \mathcal{N}} \max_{a_{i}} \big( Q_{i}^{\phi} \big)_{1}(\mathbf{s}', a_{i}) \big] - \mathbf{b}(\mathbf{s}) \\
            &- \mathbf{w}(\mathbf{s}, \mathbf{a}) \ \sum_{\mathbf{s}' \in \mathcal{S}} Pr(\mathbf{s}'|\mathbf{s}, \mathbf{a}) \big[ R(\mathbf{s}, \mathbf{a}) + \gamma \sum_{i \in \mathcal{N}} \max_{a_{i}} \big( Q_{i}^{\phi} \big)_{2}(\mathbf{s}', a_{i}) \big] + \mathbf{b}(\mathbf{s}) \bigg| \\
            &= \gamma \max_{\mathbf{s}, \mathbf{a}} \mathbf{1}^{\top} \bigg| \mathbf{w}(\mathbf{s}, \mathbf{a}) \ \sum_{\mathbf{s}' \in \mathcal{S}} Pr(\mathbf{s}'|\mathbf{s}, \mathbf{a}) \big[ \sum_{i \in \mathcal{N}} \max_{a_{i}} \big( Q_{i}^{\phi} \big)_{1}(\mathbf{s}', a_{i}) - \sum_{i \in \mathcal{N}} \max_{a_{i}} \big( Q_{i}^{\phi} \big)_{2}(\mathbf{s}', a_{i}) \big] \bigg| \\
        \end{align*}
        \begin{align*}
            &\quad \gamma \max_{\mathbf{s}, \mathbf{a}} \mathbf{1}^{\top} \bigg| \mathbf{w}(\mathbf{s}, \mathbf{a}) \ \sum_{\mathbf{s}' \in \mathcal{S}} Pr(\mathbf{s}'|\mathbf{s}, \mathbf{a}) \big[ \sum_{i \in \mathcal{N}} \max_{a_{i}} \big( Q_{i}^{\phi} \big)_{1}(\mathbf{s}', a_{i}) - \sum_{i \in \mathcal{N}} \max_{a_{i}} \big( Q_{i}^{\phi} \big)_{2}(\mathbf{s}', a_{i}) \big] \bigg| \\
            &\leq \gamma \max_{\mathbf{s}, \mathbf{a}} \mathbf{1}^{\top} \bigg| \mathbf{w}(\mathbf{s}, \mathbf{a}) \bigg| \max_{\mathbf{s}, \mathbf{a}} \bigg| \sum_{\mathbf{s}' \in \mathcal{S}} Pr(\mathbf{s}'|\mathbf{s}, \mathbf{a}) \big[ \sum_{i \in \mathcal{N}} \max_{a_{i}} \big( Q_{i}^{\phi} \big)_{1}(\mathbf{s}', a_{i}) - \sum_{i \in \mathcal{N}} \max_{a_{i}} \big( Q_{i}^{\phi} \big)_{2}(\mathbf{s}', a_{i}) \big] \bigg| \\
            &\quad(\text{If we write $\delta = \gamma \max_{\mathbf{s}, \mathbf{a}} \mathbf{1}^{\top} \big| \mathbf{w}(\mathbf{s}, \mathbf{a}) \big|$, we can have the following equation.}) \\
            &= \delta \max_{\mathbf{s}, \mathbf{a}} \bigg| \sum_{\mathbf{s}' \in \mathcal{S}} Pr(\mathbf{s}'|\mathbf{s}, \mathbf{a}) \big[ \sum_{i \in \mathcal{N}} \max_{a_{i}} \big( Q_{i}^{\phi} \big)_{1}(\mathbf{s}', a_{i}) - \sum_{i \in \mathcal{N}} \max_{a_{i}} \big( Q_{i}^{\phi} \big)_{2}(\mathbf{s}', a_{i}) \big] \bigg| \\
            &\leq \delta \max_{\mathbf{s}, \mathbf{a}} \sum_{\mathbf{s}' \in \mathcal{S}} Pr(\mathbf{s}'|\mathbf{s}, \mathbf{a}) \bigg| \sum_{i \in \mathcal{N}} \max_{a_{i}} \big( Q_{i}^{\phi} \big)_{1}(\mathbf{s}', a_{i}) - \sum_{i \in \mathcal{N}} \max_{a_{i}} \big( Q_{i}^{\phi} \big)_{2}(\mathbf{s}', a_{i}) \bigg| \\
            &= \delta \bigg| \sum_{i \in \mathcal{N}} \big[ \max_{a_{i}} \big( Q_{i}^{\phi} \big)_{1}(\mathbf{s}', a_{i}) - \max_{a_{i}} \big( Q_{i}^{\phi} \big)_{2}(\mathbf{s}', a_{i}) \big] \bigg| \\
            &\quad (\text{By triangle inequality, we can obtain the following inequality.}) \\
            &\leq \delta \sum_{i \in \mathcal{N}} \bigg| \max_{a_{i}} \big( Q_{i}^{\phi} \big)_{1}(\mathbf{s}', a_{i}) - \max_{a_{i}} \big( Q_{i}^{\phi} \big)_{2}(\mathbf{s}', a_{i}) \bigg| \\
            &\leq \delta \sum_{i \in \mathcal{N}} \max_{a_{i}} \bigg| \big( Q_{i}^{\phi} \big)_{1}(\mathbf{s}', a_{i}) - \big( Q_{i}^{\phi} \big)_{2}(\mathbf{s}', a_{i}) \bigg| \\
            &\quad (\text{Since $\mathbf{a} = \mathlarger{\mathlarger{\times}}_{\scriptscriptstyle{i \in \mathcal{N}}} a_{i}$, we have the following equation.}) \\
            &= \delta \max_{\mathbf{a}} \sum_{i \in \mathcal{N}} \bigg| \big( Q_{i}^{\phi} \big)_{1}(\mathbf{s}', a_{i}) - \big( Q_{i}^{\phi} \big)_{2}(\mathbf{s}', a_{i}) \bigg| \\
            &\leq \delta \max_{\mathbf{z}, \mathbf{a}} \sum_{i \in \mathcal{N}} \bigg| \big( Q_{i}^{\phi} \big)_{1}(\mathbf{z}, a_{i}) - \big( Q_{i}^{\phi} \big)_{2}(\mathbf{z}, a_{i}) \bigg| = \delta || \mathbf{Q}^{\phi}_{1}  - \mathbf{Q}^{\phi}_{2} ||_{1}.
        \end{align*}
        
        Now, we need to discuss the condition to $\delta \in (0, 1)$. Apparently, that $\delta > 0$ holds, so we just need to discuss the condition to guarantee that $\delta < 1$. We now have the following discussion such that
        \begin{align*}
            &\quad \ \ \delta = \gamma \max_{\mathbf{s}, \mathbf{a}} \mathbf{1}^{\top} \big| \mathbf{w}(\mathbf{s}, \mathbf{a}) \big| < 1 \ (\text{Since $w_{i}(\mathbf{s}, a_{i}) > 0$.})\\
            &\Rightarrow \gamma \max_{\mathbf{s}, \mathbf{a}} \sum_{i \in \mathcal{N}} w_{i}(\mathbf{s}, a_{i}) < 1 \ (\text{When $\gamma \neq 0$, we can have the following inequality.})\\
            &\Rightarrow \max_{\mathbf{s}, \mathbf{a}} \sum_{i \in \mathcal{N}} w_{i}(\mathbf{s}, a_{i}) < \frac{1}{\gamma} \ (\text{Since $\mathbf{a} = \mathlarger{\mathlarger{\times}}_{\scriptscriptstyle i \in \mathcal{N}} a_{i}$, we have the following equation.}) \\
            &\Rightarrow \max_{\mathbf{s}} \left\{ \sum_{i \in \mathcal{N}} \max_{a_{i}} w_{i}(\mathbf{s}, a_{i}) \right\} < \frac{1}{\gamma}.
        \end{align*}
        
        Therefore, we show that Shapley-Bellman operator $\mathlarger{\Upsilon}$ is a contraction mapping in the non-empty complete metric space generated by $\mathbb{R}^{|\mathcal{N}| \times |\mathcal{S}| |\mathcal{A}|}$ with the metric induced by $||\cdot||_{1}$, when $$\max_{\mathbf{s}} \left\{ \sum_{i \in \mathcal{N}} \max_{a_{i}} w_{i}(\mathbf{s}, a_{i}) \right\} < \frac{1}{\gamma}.$$ Finally, it is apparent that $w_{i}(\mathbf{s}, a_{i}) = 1 / |\mathcal{N}|$ when $\mathit{a}_{i} = \arg\max_{a_{i}} Q^{\phi}_{i}(\mathbf{s}, a_{i})$ satisfies the above condition.
    \end{proof}
    
    \begingroup
    \def\thetheorem{\ref{thm:proof_of_shapley_q_learning}}
    \begin{theorem}
        For a Markov convex game, the Q-learning algorithm derived by Shapley-Bellman operator given by the update rule such that
        \begin{equation*}
            \mathbf{Q}^{\phi}_{t+1}(\mathbf{s}, \mathbf{a}) \leftarrow \mathbf{Q}^{\phi}_{t}(\mathbf{s}, \mathbf{a}) + \alpha_{t}(\mathbf{s}, \mathbf{a}) \left[ \mathbf{w}(\mathbf{s}, \mathbf{a}) \left( R_{t} + \gamma \sum_{i \in \mathcal{N}} \max_{a_{i}} (Q_{i}^{\phi})_{t}(\mathbf{s}', a_{i}) \right) - \mathbf{b}(\mathbf{s}) - \mathbf{Q}^{\phi}_{t}(\mathbf{s}, \mathbf{a}) \right],
        \end{equation*}
        
        converges w.p.1 to the optimal Markov Shapley Q-value if 
        \begin{equation}
            \sum_{t} \alpha_{t}(\mathbf{s}, \mathbf{a}) = \infty \ \ \ \ \ \ \ \ \sum_{t} \alpha^{2}_{t}(\mathbf{s}, \mathbf{a}) \leq \infty
        \end{equation}
        
        for all $\mathbf{s} \in \mathcal{S}$ and $\mathbf{a} \in \mathcal{A}$ as well as $\max_{\mathbf{s}} \left\{ \sum_{i \in \mathcal{N}} \max_{a_{i}} w_{i}(\mathbf{s}, a_{i}) \right\} < \frac{1}{\gamma}$.
    \end{theorem}
    \endgroup
    \begin{proof}
        The proof follows the sketch of proving the convergence of Q-learning given by \cite{melo2001convergence}. First, we rewrite Eq.~\ref{eq:shapley_q_learning_primal} in the form such that
        $$\mathbf{Q}^{\phi}_{t}(\mathbf{s}, \mathbf{a}) = \left(1 - \alpha_{t}(\mathbf{s}, \mathbf{a})\right) \mathbf{Q}^{\phi}_{t}(\mathbf{s}, \mathbf{a}) + \alpha_{t}(\mathbf{s}, \mathbf{a}) \left[ \mathbf{w}(\mathbf{s}, \mathbf{a}) \left( R_{t} + \gamma \sum_{i \in \mathcal{N}} \max_{a_{i}} (Q_{i}^{\phi})_{t}(\mathbf{s}', a_{i}) \right) - \mathbf{b}(\mathbf{s}) \right].$$
        By subtracting $\mathbf{Q}^{\phi^{*}}(\mathbf{s}, \mathbf{a})$ and letting 
        $$\Delta_{t}(\mathbf{s}, \mathbf{a}) = \mathbf{Q}^{\phi}_{t}(\mathbf{s}, \mathbf{a}) - \mathbf{Q}^{\phi^{*}}(\mathbf{s}, \mathbf{a}),$$ we can transform Eq.~\ref{eq:shapley_q_learning_primal} to
        $$\Delta_{t+1}(\mathbf{s}, \mathbf{a}) = (1 - \alpha_{t}(\mathbf{s}, \mathbf{a})) \Delta_{t}(\mathbf{s}, \mathbf{a}) + \alpha_{t}(\mathbf{s}, \mathbf{a}) F_{t}(\mathbf{s}, \mathbf{a}),$$
        
        where $$F_{t}(\mathbf{s}, \mathbf{a}) = \mathbf{w}(\mathbf{s}, \mathbf{a}) \left( R_{t} + \gamma \sum_{i \in \mathcal{N}} \max_{a_{i}} (Q_{i}^{\phi})_{t}(\mathbf{s}', a_{i}) \right) - \mathbf{b}(\mathbf{s}) - \mathbf{Q}^{\phi^{*}}(\mathbf{s}, \mathbf{a}).$$
        
        Since $\mathbf{s}' \in \mathcal{S}$ is a random sample from Markov chain, so we can get that
        \begin{align*}
            \mathbb{E}[ F_{t}(\mathbf{s}, \mathbf{a}) | \mathcal{F}_{t} ] &= \sum_{\mathbf{s}' \in \mathcal{S}} Pr(\mathbf{s}'|\mathbf{s}, \mathbf{a}) \left[ \mathbf{w}(\mathbf{s}, \mathbf{a}) \left( R_{t} + \gamma \sum_{i \in \mathcal{N}} \max_{a_{i}} (Q_{i}^{\phi})_{t}(\mathbf{s}', a_{i}) \right) - \mathbf{b}(\mathbf{s}) - \mathbf{Q}^{\phi^{*}}(\mathbf{s}, \mathbf{a}) \right] \\
            &= \mathbf{w}(\mathbf{s}, \mathbf{a}) \sum_{\mathbf{s}' \in \mathcal{S}} Pr(\mathbf{s}'|\mathbf{s}, \mathbf{a}) \left( R_{t} + \gamma \sum_{i \in \mathcal{N}} \max_{a_{i}} (Q_{i}^{\phi})_{t}(\mathbf{s}', a_{i}) \right) - \mathbf{b}(\mathbf{s}) - \mathbf{Q}^{\phi^{*}}(\mathbf{s}, \mathbf{a}) \\
            & \quad \left(\text{Since $\max_{\mathbf{s}} \left\{ \sum_{i \in \mathcal{N}} \max_{a_{i}} w_{i}(\mathbf{s}, a_{i}) \right\} < \frac{1}{\gamma}$.}\right) \\
            &= \mathlarger{\Upsilon} \mathbf{Q}^{\phi}_{t}(\mathbf{s}, \mathbf{a}) - \mathlarger{\Upsilon} \mathbf{Q}^{\phi^{*}}(\mathbf{s}, \mathbf{a}).
        \end{align*}
        
        By the results from Theorem \ref{lemm:shapley_q_contraction_mapping}, we can get that
        \begin{equation*}
            ||\mathbb{E}[ F_{t}(\mathbf{s}, \mathbf{a}) | \mathcal{F}_{t} ]||_{1} \leq \delta ||\mathbf{Q}^{\phi}_{t}(\mathbf{s}, \mathbf{a}) - \mathbf{Q}^{\phi^{*}}(\mathbf{s}, \mathbf{a})||_{1} = \delta ||\Delta_{t}(\mathbf{s}, \mathbf{a})||_{1},
        \end{equation*}
        where $\delta \in (0, 1)$.
        
        Next, we get that
        \begin{align*}
            \textbf{var}[F_{t}(\mathbf{s}, \mathbf{a})| \mathcal{F}_{t}] &= \mathlarger{\mathbb{E}} \Bigg[ \bigg( \mathbf{w}(\mathbf{s}, \mathbf{a}) \big( R_{t} + \gamma \sum_{i \in \mathcal{N}} \max_{a_{i}} (Q_{i}^{\phi})_{t}(\mathbf{s}', a_{i}) \big) - \mathbf{b}(\mathbf{s}) - \mathbf{Q}^{\phi^{*}}(\mathbf{s}, \mathbf{a}) \\
            &- \mathlarger{\Upsilon} \mathbf{Q}^{\phi}_{t}(\mathbf{s}, \mathbf{a}) + \mathbf{Q}^{\phi^{*}}(\mathbf{s}, \mathbf{a}) \bigg)^{2} \Bigg] \\
            &= \mathlarger{\mathbb{E}} \left[ \left( \mathbf{w}(\mathbf{s}, \mathbf{a}) \big( R_{t} + \gamma \sum_{i \in \mathcal{N}} \max_{a_{i}} (Q_{i}^{\phi})_{t}(\mathbf{s}', a_{i}) \big) - \mathbf{b}(\mathbf{s}) - \mathlarger{\Upsilon} \mathbf{Q}^{\phi}_{t}(\mathbf{s}, \mathbf{a}) \right)^{2} \right] \\
            &= \textbf{var} \left[ \mathbf{w}(\mathbf{s}, \mathbf{a}) \big( R_{t} + \gamma \sum_{i \in \mathcal{N}} \max_{a_{i}} (Q_{i}^{\phi})_{t}(\mathbf{s}', a_{i}) \big) - \mathbf{b}(\mathbf{s}) \bigg| \mathcal{F}_{t} \right].
        \end{align*}
        
        Since $R_{t}$, $\mathbf{w}(\mathbf{s}, \mathbf{a})$ and $\mathbf{b}(\mathbf{s})$ are bounded, it clearly verifies that 
        \begin{equation*}
            \textbf{var}[F_{t}(\mathbf{s}, \mathbf{a})| \mathcal{F}_{t}] \leq C (1 + ||\Delta_{t}(\mathbf{s}, \mathbf{a})||_{1}^{2})
        \end{equation*}
        for some constant $C$.
        
        Finally, by Lemma \ref{lemm:stochastic_process} it is easy to see that $\Delta_{t}$ converges to 0 w.p.1, i.e., $ \mathbf{Q}^{\phi}_{t}(\mathbf{s}, \mathbf{a})$ converges to $ \mathbf{Q}^{\phi^{*}}(\mathbf{s}, \mathbf{a})$ w.p.1, given the condition in Eq.~\ref{eq:alpha_condition}.
    \end{proof}

\section{Proof of Validity and Interpretability}
\label{sec:validity_and_interpretability}
    \begingroup
    \def\thelemma{\ref{lemm:markov_core_convex_set}}
    \begin{lemma}
        Markov core is a convex set.
    \end{lemma}
    \endgroup
    \begin{proof}
        Let $\big( \max_{\pi_{i}} x_{i}(\mathbf{s}) \big)_{i \in \mathcal{N}}$ and $\big( \max_{\pi_{i}} y_{i}(\mathbf{s}) \big)_{i \in \mathcal{N}}$ be two vectors in the Markov core and $\alpha \in [0, 1)$ be an arbitrary scalar. To ease the derivation, for any $i \in \mathcal{N}$ we let $\max_{\pi_{i}} z_{i}(\mathbf{s}) = \alpha \max_{\pi_{i}} x_{i}(\mathbf{s}) + (1 - \alpha) \max_{\pi_{i}} y_{i}(\mathbf{s})$. By definition, for any coalition $\mathcal{C} \subseteq \mathcal{N}$ we have
        \begin{align*}
            \max_{\pi_{\mathcal{C}}} z(\mathbf{s} | \mathcal{C}) &= \sum_{i \in \mathcal{C}} \max_{\pi_{i}} z_{i}(\mathbf{s}) \\
            &= \sum_{i \in \mathcal{C}} \alpha \max_{\pi_{i}} x_{i}(\mathbf{s}) + (1 - \alpha) \max_{\pi_{i}} y_{i}(\mathbf{s})\\
            &= \alpha \sum_{i \in \mathcal{C}} \max_{\pi_{i}} x_{i}(\mathbf{s}) + (1 - \alpha) \sum_{i \in \mathcal{C}} \max_{\pi_{i}} y_{i}(\mathbf{s})\\
            &\geq \alpha \max_{\pi_{\mathcal{C}}} V^{\pi_{\mathcal{C}}}(\mathbf{s}) + (1 - \alpha) \max_{\pi_{\mathcal{C}}} V^{\pi_{\mathcal{C}}}(\mathbf{s})
            = \max_{\pi_{\mathcal{C}}} V^{\pi_{\mathcal{C}}}(\mathbf{s}).
        \end{align*}
        Therefore, we have proved that Markov core is a convex set.
    \end{proof}

\section{Proof of The Implementation of Shapley Q-Learning}
\label{sec:proof_of_the_implementation_of_shapley_q-learning}
    \begingroup
    \def\theproposition{\ref{prop:shapley_value_approximate}}
    \begin{proposition}
        Suppose that any action marginal contribution can be factorised into the form such that $\Upphi_{i}(\mathbf{s}, a_{i} | \mathcal{C}_{i}) = \sigma(\mathbf{s}, \mathbf{a}_{ \scriptscriptstyle\mathcal{C}_{i} \cup \{i\} }) \ \hat{Q}_{i}(\mathbf{s}, a_{i})$. With the condition that
        \begin{equation*}
            \mathbb{E}_{\mathcal{C}_{i} \sim Pr(\mathcal{C}_{i} | \mathcal{N} \backslash \{i\})}[\sigma(\mathbf{s}, \mathbf{a}_{ \scriptscriptstyle\mathcal{C}_{i} \cup \{i\} })] =
            \begin{cases} 
                 1 & \ \ a_{i} = \arg\max_{a_{i}} Q^{\phi}_{i}(\mathbf{s}, a_{i}), \\
                 K \in (0, 1) & \ \ a_{i} \neq \arg\max_{a_{i}} Q^{\phi}_{i}(\mathbf{s}, a_{i}),
            \end{cases}
        \end{equation*}
        
        we have
        \begin{equation}
            \begin{cases} 
                 Q_{i}^{\phi}(\mathbf{s}, a_{i}) = \hat{Q}_{i}(\mathbf{s}, a_{i}) & \ \ a_{i} = \arg\max_{a_{i}} \hat{Q}_{i}(\mathbf{s}, a_{i}), \\
                 \alpha_{i}(\mathbf{s}, a_{i}) \ Q^{\phi}_{i}(\mathbf{s}, a_{i}) = \hat{\alpha}_{i}(\mathbf{s}, a_{i}) \ \hat{Q}_{i}(\mathbf{s}, a_{i}) & \ \ a_{i} \neq \arg\max_{a_{i}} \hat{Q}_{i}(\mathbf{s}, a_{i}),
            \end{cases}
        \end{equation}
        
        where $\hat{\alpha}_{i}(\mathbf{s}, a_{i}) = \mathbb{E}_{\mathcal{C}_{i} \sim Pr(\mathcal{C}_{i} | \mathcal{N} \backslash \{i\})}[ \hat{\psi}_{i}(\mathbf{s}, a_{i}; \mathbf{a}_{ \scriptscriptstyle\mathcal{C}_{i} }) ]$ and $\hat{\psi}_{i}(\mathbf{s}, a_{i}; \mathbf{a}_{ \scriptscriptstyle\mathcal{C}_{i} }) := \alpha_{i}(\mathbf{s}, a_{i}) \ \sigma(\mathbf{s}, \mathbf{a}_{ \scriptscriptstyle\mathcal{C}_{i} \cup \{i\} })$.
    \end{proposition}
    \endgroup
    \begin{proof}
        We suppose for any $\mathbf{s} \in \mathcal{S}$ and $\mathbf{a} \in \mathcal{A}$, we have $\Upphi_{i}(\mathbf{s}, a_{i} | \mathcal{C}_{i}) = \sigma(\mathbf{s}, \mathbf{a}_{ \scriptscriptstyle\mathcal{C}_{i} \cup \{i\} }) \ \hat{Q}_{i}(\mathbf{s}, a_{i})$ and $\mathbb{E}_{\mathcal{C}_{i}}[\sigma(\mathbf{s}, \mathbf{a}_{ \scriptscriptstyle\mathcal{C}_{i} \cup \{i\} })] = 1$ when $\mathit{a}_{i} = \arg\max_{a_{i}} Q^{\phi}_{i}(\mathbf{s}, a_{i})$. By the definition of the Markov Shapley Q-value, it is not difficult to obtain
        \begin{align*}
            Q^{\phi}_{i}(\mathbf{s}, a_{i}) &= \mathbb{E}_{\scriptscriptstyle\mathcal{C}_{i}}[ \Upphi_{i}(\mathbf{s}, a_{i} | \mathcal{C}_{i}) ] \\
            &= \mathbb{E}_{\scriptscriptstyle\mathcal{C}_{i}}[ \sigma(\mathbf{s}, \mathbf{a}_{ \scriptscriptstyle\mathcal{C}_{i} \cup \{i\} }) \ \hat{Q}_{i}(\mathbf{s}, a_{i}) ] \\
            &= \mathbb{E}_{\scriptscriptstyle\mathcal{C}_{i}}[ \sigma(\mathbf{s}, \mathbf{a}_{ \scriptscriptstyle\mathcal{C}_{i} \cup \{i\} }) ] \ \hat{Q}_{i} (\mathbf{s}, a_{i}).
        \end{align*}
        
        Recall that $\delta_{i}(\mathbf{s}, a_{i})$ is defined as follows:
        \begin{equation}
            \delta_{i}(\mathbf{s}, a_{i}) = \begin{cases} 
                                                  1 & a_{i} = \arg\max_{a_{i}} Q^{\phi}_{i}(\mathbf{s}, a_{i}), \\
                                                  \alpha_{i}(\mathbf{s}, a_{i}) & a_{i} \neq \arg\max_{a_{i}} Q^{\phi}_{i}(\mathbf{s}, a_{i}).
                                             \end{cases}
        \label{eq:delta_copy}
        \end{equation}
        
        If $a_{i} = \arg\max_{a_{i}} Q^{\phi}_{i}(\mathbf{s}, a_{i})$, it is not difficult to get that $Q_{i}^{\phi}(\mathbf{s}, a_{i}) = \hat{Q}_{i}(\mathbf{s}, a_{i})$.
        
        If $a_{i} \neq \arg\max_{a_{i}} Q^{\phi}_{i}(\mathbf{s}, a_{i})$, we can have the following equation such that
        \begin{align*}
            \alpha_{i}(\mathbf{s}, a_{i}) \ Q^{\phi}_{i}(\mathbf{s}, a_{i}) &= \alpha_{i}(\mathbf{s}, a_{i}) \ \mathbb{E}_{\scriptscriptstyle\mathcal{C}_{i}}[ \sigma(\mathbf{s}, \mathbf{a}_{ \scriptscriptstyle\mathcal{C}_{i} \cup \{i\} }) \ \hat{Q}_{i}(\mathbf{s}, a_{i}) ] \\
            &= \mathbb{E}_{\scriptscriptstyle\mathcal{C}_{i}}[ \alpha_{i}(\mathbf{s}, a_{i}) \sigma(\mathbf{s}, \mathbf{a}_{ \scriptscriptstyle\mathcal{C}_{i} \cup \{i\} }) ] \ \hat{Q}_{i}(\mathbf{s}, a_{i}) \\
            &:= \mathbb{E}_{\scriptscriptstyle\mathcal{C}_{i}}[ \hat{\psi}_{i}(\mathbf{s}, a_{i}; \mathbf{a}_{ \scriptscriptstyle\mathcal{C}_{i} }) ] \ \hat{Q}_{i}(\mathbf{s}, a_{i}),
        \end{align*}
        
        where $\alpha_{i}(\mathbf{s}, a_{i}) \sigma(\mathbf{s}, \mathbf{a}_{ \scriptscriptstyle\mathcal{C}_{i} \cup \{i\} })$ is defined as $\hat{\psi}_{i}(\mathbf{s}, a_{i}; \mathbf{a}_{ \scriptscriptstyle\mathcal{C}_{i} })$. Since under this situation $\hat{Q}_{i}(\mathbf{s}, a_{i})$ is always a scaled $Q^{\phi}_{i}(\mathbf{s}, a_{i})$ with the scale of $1/K$, the decisions are consistent to the original decision.
    \end{proof}
    
    \begingroup
    \def\theproposition{\ref{prop:hatalpha_satisfying_condition}}
    \begin{proposition}
        $\hat{\alpha}_{i}(\mathbf{s}, a_{i})$ satisfies the condition $\max_{\mathbf{s}} \big\{ \sum_{i \in \mathcal{N}} \max_{a_{i}} w_{i}(\mathbf{s}, a_{i}) \big\} < \frac{1}{\gamma}$.
    \end{proposition}
    \endgroup
    \begin{proof}
        As introduced in the main part of paper, when $a_{i} \neq \arg\max_{a_{i}} \hat{Q}_{i}(\mathbf{s}, a_{i})$, $\hat{\alpha}_{i}(\mathbf{s}, a_{i})$ is implemented as follows:
        \begin{equation*}
            \hat{\alpha}_{i}(\mathbf{s}, a_{i}) = \frac{1}{M}\sum_{k = 1}^{M} \mathlarger{F}_{\mathbf{s}} \Big( \hat{Q}_{\mathcal{C}_{i}^{k}}(\tau_{\mathcal{C}_{i}^{k}}, \mathbf{a}_{\mathcal{C}_{i}^{k}}), \ \hat{Q}_{i}(\tau_{i}, a_{i}) \Big) + 1,
        \label{eq:delta_deep_representation}
        \end{equation*}
        
        where
        \begin{equation*}
            \hat{Q}_{\mathcal{C}_{i}^{k}}(\tau_{\mathcal{C}_{i}^{k}}, \mathbf{a}_{\mathcal{C}_{i}^{k}}) = \frac{1}{|\mathcal{C}_{i}^{k}|}\sum_{j \in \mathcal{C}_{i}^{k}} \hat{Q}_{j}(\tau_{j}, a_{j})
        \end{equation*}
        
        and $\mathcal{C}_{i}^{k} \sim \mathit{Pr}(\mathcal{C}_{i} | \mathcal{N} \backslash \{i\})$ that follows the distribution with respect to the occurrence frequency of $\mathcal{C}_{i}$; and $\mathlarger{F}_{\mathbf{s}}(\cdot, \cdot)$ is a monotonic function with an absolute activation function on the output whose weights are generated from hypernetworks with the global state as the input, similar to the architecture of QMIX \cite{RashidSWFFW18}. Since $\mathlarger{F}_{\mathbf{s}}(\cdot, \cdot) \geq 0$ always holds, it is not difficult to obtain that $\hat{\alpha}_{i}(\mathbf{s}, a_{i}) \geq 1$ always holds. As Eq.~\ref{eq:shapley_q_approximate} shows, it is not difficult to get that $\alpha_{i}(\mathbf{s}, a_{i}) = K^{-1} \ \hat{\alpha}_{i}(\mathbf{s}, a_{i})$. Since $K \in (0, 1)$, we get that $\alpha_{i}(\mathbf{s}, a_{i}) > 1$.
        
        As introduced in the main part of this thesis, the following equation is satisfied such that
        \begin{equation*}
            \delta_{i}(\mathbf{s}, a_{i}) = \frac{1}{|\mathcal{N}| \ w_{i}(\mathbf{s}, a_{i})}.
        \end{equation*}
        
        For all $\mathbf{s} \in \mathcal{S}$ and $a_{i} \neq \arg\max_{a_{i}} \hat{Q}_{i}(\mathbf{s}, a_{i})$, $\delta_{i}(\mathbf{s}, a_{i}) = \alpha_{i}(\mathbf{s}, a_{i}) > 1$. So, we can derive that
        \begin{equation*}
            \begin{split}
                &\quad \ \ w_{i}(\mathbf{s}, a_{i}) = \frac{1}{|\mathcal{N}| \ \alpha_{i}(\mathbf{s}, a_{i})} \\
                &\Rightarrow \max_{a_{i}} w_{i}(\mathbf{s}, a_{i}) = \max_{a_{i}} \frac{1}{|\mathcal{N}| \ \alpha_{i}(\mathbf{s}, a_{i})} = \frac{1}{|\mathcal{N}| \ \min_{a_{i}} \alpha_{i}(\mathbf{s}, a_{i})} < \frac{1}{|\mathcal{N}|} \\
                &\Rightarrow 0 < \sum_{i \in \mathcal{N}} \max_{a_{i}} w_{i}(\mathbf{s}, a_{i}) < 1.
            \end{split}
        \end{equation*}
        
        For all $\mathbf{s} \in \mathcal{S}$ and $a_{i} = \arg\max_{a_{i}} \hat{Q}_{i}(\mathbf{s}, a_{i})$, $\delta_{i}(\mathbf{s}, a_{i}) = \hat{\delta}_{i}(\mathbf{s}, a_{i}) = 1$. So, we can derive that
        \begin{equation*}
            \begin{split}
                &\quad \ \ w_{i}(\mathbf{s}, a_{i}) = \frac{1}{|\mathcal{N}|} \\
                &\Rightarrow \sum_{i \in \mathcal{N}} \max_{a_{i}} w_{i}(\mathbf{s}, a_{i}) = 1.
            \end{split}
        \end{equation*}
        
        Therefore, we can directly obtain that for all $\mathbf{s} \in \mathcal{S}$ and $\mathbf{a} \in \mathcal{A}$, $$0 < \max_{\mathbf{s}} \Big\{ \sum_{i \in \mathcal{N}} \max_{a_{i}} w_{i}(\mathbf{s}, a_{i}) \Big\} \leq 1.$$
        
        Since $\gamma \in (0, 1)$, we can get that $\frac{1}{\gamma} > 1$. As a result, we show that for all $\mathbf{s} \in \mathcal{S}$ and $\mathbf{a} \in \mathcal{A}$, $$0 < \max_{\mathbf{s}} \Big\{ \sum_{i \in \mathcal{N}} \max_{a_{i}} w_{i}(\mathbf{s}, a_{i}) \Big\} < \frac{1}{\gamma}.$$
        
        We conclude that our implementation of $\hat{\alpha}_{i}(\mathbf{s}, a_{i})$ satisfies the condition in Theorem \ref{thm:shapley_q_optimal}.
    \end{proof}

\section{Proof of Fixing The Inconsistency Problem}
\label{subsec:fixing_the_problem_of_inconsistency}
    To fix the problem of inconsistent coalition values, a possible solution is learning only one parametric coalition value function $\max_{\pi_{\mathcal{C}}} V^{\pi_{\mathcal{C}}}(\mathbf{s}; \ \theta)$ to represent a set of maximum coalition values. More specifically, the maximum coalition values with the identical coalition within two successive coalition marginal contributions can be cancelled so that the properties of efficiency and fairness are preserved. Next, we prove the feasibility of implementing this method in practice, i.e., \textit{whether the approximation error of this method is under control}.
    
    \begin{assumption}
        The ability of learning a function is invariant to model classes. That is, $\forall \hat{f}(\cdot), \hat{g}(\cdot) \in L^{p}$, $\left|\hat{f}(\cdot) - f(\cdot)\right| \leq \epsilon$ and $\left|\hat{g}(\cdot) - g(\cdot)\right| \leq \epsilon$ hold.
    \label{assm:model_learning}
    \end{assumption}
    
    \setcounter{lemma}{8}
    \begin{lemma}
        Assuming that $\left| \max_{\pi_{\mathcal{C}}} \hat{V}^{\pi_{\mathcal{C}}}(\mathbf{s}; \ \theta) - \max_{\pi_{\mathcal{C}}} V^{\pi_{\mathcal{C}}}(\mathbf{s}; \theta)\right| \leq \epsilon$, the maximum coalition marginal contribution $\max_{\pi_{i}} \hat{\Phi}_{i}(\mathbf{s}|\mathcal{C}_{i}; \theta)$ generated by $\hat{V}^{\pi_{\mathcal{C}}}(\mathbf{s}; \theta)$ is with the approximation error bound $2\epsilon$.
    \label{lemm:approximation_error_bound_cmc}
    \end{lemma}
    \begin{proof}
        The maximum coalition marginal contribution generated by $\max_{\pi_{\mathcal{C}}} \hat{V}^{\pi_{\mathcal{C}}}(\mathbf{s}; \theta)$ is expressed as $\max_{\pi_{i}} \hat{\Phi}_{i}(\mathbf{s}|\mathcal{C}_{i}; \theta) = \max_{\pi_{\mathcal{C}_{i} \cup \{i\}}} \hat{V}^{\pi_{\mathcal{C}_{i} \cup \{i\}}}(\mathbf{s};\theta) - \max_{\pi_{\mathcal{C}_{i}}} \hat{V}^{\pi_{\mathcal{C}_{i}}}(\mathbf{s};\theta)$. Using the assumption, we can get the following approximation error bound for the coalition marginal contribution such that
        \begin{align*}
            &\quad \left| \max_{\pi_{i}} \hat{\Phi}_{i}(\mathbf{s}|\mathcal{C}_{i}; \theta) - \max_{\pi_{i}} \Phi_{i}(\mathbf{s}|\mathcal{C}_{i}) \right| \\
            &= \left| \max_{\pi_{\mathcal{C}_{i} \cup \{i\}}} \hat{V}^{\pi_{\mathcal{C}_{i} \cup \{i\}}}(\mathbf{s};\theta) - \max_{\pi_{\mathcal{C}_{i}}} \hat{V}^{\pi_{\mathcal{C}_{i}}}(\mathbf{s};\theta) - \left[ \max_{\pi_{\mathcal{C}_{i} \cup \{i\}}} V^{\pi_{\mathcal{C}_{i} \cup \{i\}}}(\mathbf{s}) - \max_{\pi_{\mathcal{C}_{i}}} V^{\pi_{\mathcal{C}_{i}}}(\mathbf{s}) \right] \right| \\
            &\leq \left| \max_{\pi_{\mathcal{C}_{i} \cup \{i\}}} \hat{V}^{\pi_{\mathcal{C}_{i} \cup \{i\}}}(\mathbf{s};\theta) - \max_{\pi_{\mathcal{C}_{i}} \cup \{i\}} V^{\pi_{\mathcal{C}_{i} \cup \{i\}}}(\mathbf{s}) \right| + \left| \max_{\pi_{\mathcal{C}_{i}}} \hat{V}^{\pi_{\mathcal{C}_{i}}}(\mathbf{s};\theta) - \max_{\pi_{\mathcal{C}_{i}}} V^{\pi_{\mathcal{C}_{i}}}(\mathbf{s}) \right| \\
            &= 2\epsilon.
        \end{align*}
    \end{proof}
    
    \begin{lemma}
        Suppose that the approximation error bound of a coalition value function is $\epsilon$, i.e., $\left| \max_{\pi_{i}} \hat{\Phi}_{i}(\mathbf{s}|\mathcal{C}_{i}; \theta) - \max_{\pi_{i}} \Phi_{i}(\mathbf{s}|\mathcal{C}_{i}) \right| \leq \epsilon$. The corresponding approximate Markov Shapley value is also with the approximation error bound $\epsilon$.
    \label{lemm:approximation_error_bound_msv}
    \end{lemma}
    \begin{proof}
        Given Assumption \ref{assm:max_shapley_value}, the approximate maximum MSV can be expressed as the following equation:
        \begin{equation}
            \max_{\pi_{i}} \hat{V}_{i}^{\phi}(\mathbf{s}; \theta) = \sum_{\mathcal{C}_{i} \subseteq \mathcal{N}} \frac{|\mathcal{C}_{i}|!(|\mathcal{N}| - |\mathcal{C}_{i}| - 1)!}{|\mathcal{N}|!} \cdot \max_{\pi_{i}} \hat{\Phi}_{i}(\mathbf{s}|\mathcal{C}_{i}; \theta).
        \label{eq:msv}
        \end{equation}
        
        By the result from Eq.~\ref{eq:msv}, we can get the approximation error bound of MSV such that 
        \begin{align*}
            \left| \max_{\pi_{i}} \hat{V}_{i}^{\phi}(\mathbf{s}; \theta) - \max_{\pi_{i}} V_{i}^{\phi}(\mathbf{s}) \right| &= \bigg| \sum_{\mathcal{C}_{i} \subseteq \mathcal{N}} \frac{|\mathcal{C}_{i}|!(|\mathcal{N}| - |\mathcal{C}_{i}| - 1)!}{|\mathcal{N}|!} \cdot \max_{\pi_{i}} \hat{\Phi}_{i}(\mathbf{s}|\mathcal{C}_{i}; \theta) \\
            &- \sum_{\mathcal{C}_{i} \subseteq \mathcal{N}} \frac{|\mathcal{C}_{i}|!(|\mathcal{N}| - |\mathcal{C}_{i}| - 1)!}{|\mathcal{N}|!} \cdot \max_{\pi_{i}} \Phi_{i}(\mathbf{s}|\mathcal{C}_{i}) \bigg| \\
            &= \left| \sum_{\mathcal{C}_{i} \subseteq \mathcal{N}} \frac{|\mathcal{C}_{i}|!(|\mathcal{N}| - |\mathcal{C}_{i}| - 1)!}{|\mathcal{N}|!} \cdot \left( \max_{\pi_{i}} \hat{\Phi}_{i}(\mathbf{s}|\mathcal{C}_{i}; \theta) - \max_{\pi_{i}} \Phi_{i}(\mathbf{s}|\mathcal{C}_{i}) \right) \right| \\
            &\leq \sum_{\mathcal{C}_{i} \subseteq \mathcal{N}} \frac{|\mathcal{C}_{i}|!(|\mathcal{N}| - |\mathcal{C}_{i}| - 1)!}{|\mathcal{N}|!} \cdot \left| \max_{\pi_{i}} \hat{\Phi}_{i}(\mathbf{s}|\mathcal{C}_{i}; \theta) - \max_{\pi_{i}} \Phi_{i}(\mathbf{s}|\mathcal{C}_{i}) \right| \\
            &\leq \sum_{\mathcal{C}_{i} \subseteq \mathcal{N}} \frac{|\mathcal{C}_{i}|!(|\mathcal{N}| - |\mathcal{C}_{i}| - 1)!}{|\mathcal{N}|!} \cdot \epsilon = \epsilon.
        \end{align*}
    \end{proof}
    
    \begingroup
    \def\theproposition{\ref{prop:feasibility_msv_by_cv}}
        \begin{proposition}
            The approximate maximum Markov Shapley value generated by approximate maximum coalition values is feasible to be learned in practice.
        \end{proposition}
    \endgroup
    \begin{proof}
        First, we suppose that the approximation error bound of the directly approximate maximum coalition marginal contribution is given by $\left| \max_{\pi_{i}} \hat{\Phi}_{i}(\mathbf{s}|\mathcal{C}_{i}; \beta) - \max_{\pi_{i}} \Phi_{i}(\mathbf{s}|\mathcal{C}_{i}) \right| \leq \epsilon$. By Assumption \ref{assm:model_learning}, we can suppose that $\left| \max_{\pi_{\mathcal{C}}} \hat{V}^{\pi_{\mathcal{C}}}(\mathbf{s}; \ \theta) - \max_{\pi_{\mathcal{C}}} V^{\pi_{\mathcal{C}}}(\mathbf{s})\right| \leq \epsilon$ and by Lemma \ref{lemm:approximation_error_bound_cmc} the approximation error bound of the corresponding approximate maximum coalition marginal contribution is given by $\left| \max_{\pi_{i}} \hat{\Phi}_{i}(\mathbf{s}|\mathcal{C}_{i}; \theta) - \max_{\pi_{i}} \Phi_{i}(\mathbf{s}|\mathcal{C}_{i}) \right| \leq 2\epsilon$. By Lemma \ref{lemm:approximation_error_bound_msv}, the approximation error bounds of the maximum MSVs generated from these two methods are equal to that of the corresponding approximate coalition marginal contributions. Therefore, the approximation error bound of the approximate maximum MSV generated from the approximate maximum coalition values is only two times as the direct approximation of the coalition marginal contribution. As a result, we can conclude that the approximate maximum MSV generated by approximate maximum coalition values is feasible to be learned in practice.
    \end{proof}
        
\section{Proof of MARL Algorithms for POMCG}
\label{sec:proof_of_shapley_policy_iteration_fo_pomcg}
    \begin{theorem}[\cite{smallwood1973optimal}]
        If the value function is convex, for all $b \in \mathcal{B}$, the value iteration for belief-MDP such that
        \begin{equation}
            V_{m+1}(b) \leftarrow \max_{a} R(b, a) + \gamma \sum_{o' \in \mathcal{O}} Pr(o' | b, a) V_{m}\left(\tau(o', a, b)\right),
        \label{eq:belief_mdp_vi}
        \end{equation}
        converges to the optimal value function as $m \rightarrow \infty$ under the infinity norm, satisfying Bellman optimality equation such that
        \begin{equation*}
            V^{\pi^{*}}(b)  = \max_{a} R(b, a) + \gamma \sum_{o' \in \mathcal{O}} Pr(o' | b, a) V^{\pi^{*}}(\tau(o', a, b)),
        \end{equation*}
        where $\pi^{*}$ denotes the optimal policy.
    \label{thm:belief_mdp_vi}
    \end{theorem}
        
    \begingroup
    \def\thelemma{\ref{lemm:value_iteration}}
    \begin{lemma}
        For all $b \in \mathcal{B}(\mathcal{CS})$ and $\mathcal{CS} \in \Psi(\mathcal{C}, \Lambda(\mathcal{CS}))$, the value iteration for $\pi_{\scriptscriptstyle\mathcal{C}}$ such that
        \begin{equation*}
            V_{m+1}(b) \leftarrow R_{b}(b, \mathbf{a}_{\scriptscriptstyle\mathcal{C}}) + \gamma \sum_{\mathbf{o}' \in \mathcal{O}} Pr(\mathbf{o}' | b, \mathbf{a}, \mathcal{CS}) V_{m}(\tau(\mathbf{o}', \mathbf{a}, b | \mathcal{CS})),
        \end{equation*}
        
        converges to $V^{\pi_{\mathcal{C}}}(b)$ as $m \rightarrow \infty$ under the infinity norm.
    \end{lemma}
    \endgroup
    \begin{proof}
        The training phase of POMCG is a form of belief-MDP. If we rewrite $Pr(\mathbf{o}' | b, \mathbf{a}, \mathcal{CS})$ as $T(\mathbf{o}' | b, \mathbf{a})$ and $\tau(\mathbf{o}', \mathbf{a}, b | \mathcal{CS})$ as $\zeta(\mathbf{o}', \mathbf{a}, b)$ for any $\mathcal{CS} \in \Psi(\mathcal{C}, \Lambda(\mathcal{CS}))$, the value iteration for $\pi_{\scriptscriptstyle\mathcal{C}}$ is equivalently expressed as:
        \begin{equation}
            V_{m+1}(b) \leftarrow R_{b}(b, \mathbf{a}_{\scriptscriptstyle\mathcal{C}}) + \gamma \sum_{\mathbf{o}' \in \mathcal{O}} T(\mathbf{o}' | b, \mathbf{a}) V_{m}(\zeta(\mathbf{o}', \mathbf{a}, b)).
        \label{eq:pomcg_vi_new}
        \end{equation}
        
        The corresponding Bellman optimality equation is expressed as:
        \begin{equation}
            V^{\pi_{\mathcal{C}}^{*}}(b) = \max_{\mathbf{a}_{\mathcal{C}}} R_{b}(b, \mathbf{a}_{\scriptscriptstyle\mathcal{C}}) + \gamma \sum_{\mathbf{o}' \in \mathcal{O}} T(\mathbf{o}' | b, \mathbf{a}) V^{\pi_{\mathcal{C}}^{*}}(\zeta(\mathbf{o}', \mathbf{a}, b)).
        \label{eq:pomcg_boe}
        \end{equation}
        
        If we consider a stationary joint policy for an arbitrary coalition $\mathcal{C} \ \mathlarger{\mathlarger{\subseteq}} \ \mathcal{N}$, then $\mathbf{a}_{\scriptscriptstyle\mathcal{C}} \in \left\{ \pi_{\scriptscriptstyle\mathcal{C}}(b) \right\}$. From Theorem \ref{thm:belief_mdp_vi}, $\pi_{\scriptscriptstyle\mathcal{C}}$ solves Eq.~\ref{eq:pomcg_boe} such that
        \begin{equation}
            V_{m+1}(b) = R_{b}(b, \mathbf{a}_{\scriptscriptstyle\mathcal{C}}) + \gamma \sum_{\mathbf{o}' \in \mathcal{O}} T(\mathbf{o}' | b, \mathbf{a}) V_{m}(\zeta(\mathbf{o}', \mathbf{a}, b)).
        \label{eq:pomcg_boe_new}
        \end{equation}
        
        From Theorem \ref{thm:belief_mdp_vi}, Eq.~\ref{eq:pomcg_vi_new} converges to $V^{\pi_{\mathcal{C}}}(b)$ for all $b \in \mathcal{B}(\mathcal{CS})$ and $\mathcal{CS} \in \Psi(\mathcal{C}, \Lambda(\mathcal{CS}))$, as $m \rightarrow \infty$ under the infinity norm.
    \end{proof}
    
    \setcounter{lemma}{10}
    \begin{lemma}
    \label{lemm:policy_improvemnent_msv}
        Policy improvement of each agent by improving the Markov Shapley value improves coalition values.
    \end{lemma}
    \begin{proof}
        For any agent $\mathit{i} \in \mathcal{N}$ and any $b \in \mathcal{B}(\mathcal{CS})$ and $\mathcal{CS} \in \Psi(\mathcal{C}, \Lambda(\mathcal{CS}))$, if we conduct the policy improvement we have the following inequalities such that 
        \begin{align*}
            V_{i}^{\phi^{k}}(b)
            &= \sum_{\mathcal{C}_{i} \ \mathlarger{\mathlarger{\subseteq}} \ \mathcal{N} \backslash \{i\}} \frac{|\mathcal{C}_{i}|!(|\mathcal{N}| - |\mathcal{C}_{i}| - 1)!}{|\mathcal{N}|!} \cdot \left[ V^{\pi_{\mathcal{C}_{i} \cup \{i\}}^{k}}(b) - V^{\pi_{\mathcal{C}_{i}}^{k}}(b) \right] \\
            &= \sum_{\mathcal{C}_{i} \ \mathlarger{\mathlarger{\subseteq}} \ \mathcal{N} \backslash \{i\} } \frac{|\mathcal{C}_{i}|!(|\mathcal{N}|-|\mathcal{C}_{i}|-1)!}{|\mathcal{N}|!} \cdot \left[ Q^{\pi_{\mathcal{C}_{i} \cup \{i\}}^{k}}(b, \mathbf{a}_{\scriptscriptstyle\mathcal{C}_{i} \cup \{i\}}^{k}) - Q^{\pi_{\mathcal{C}_{i}}^{k}}(b, \mathbf{a}_{\scriptscriptstyle\mathcal{C}_{i}}^{k}) \right] \\
            &= Q_{i}^{\phi^{k}}(b, \pi_{i}^{k}(b))
            \leq Q_{i}^{\phi^{k}}(b, \pi_{i}^{k+1}(b)) \\
            &= \sum_{\mathcal{C}_{i} \ \mathlarger{\mathlarger{\subseteq}} \ \mathcal{N} \backslash \{i\} } \frac{|\mathcal{C}_{i}|!(|\mathcal{N}|-|\mathcal{C}_{i}|-1)!}{|\mathcal{N}|!} \cdot \left[ Q^{\pi_{\mathcal{C}_{i} \cup \{i\}}^{k}}(b, \mathbf{a}_{\scriptscriptstyle\mathcal{C}_{i}}^{k}, \pi_{i}^{k+1}(b)) - Q^{\pi_{\mathcal{C}_{i}}^{k}}(b, \mathbf{a}_{\scriptscriptstyle\mathcal{C}_{i}}^{k}) \right],
        \end{align*}
        
        where we can get that $V^{\pi_{\mathcal{C}_{i} \cup \{i\}}^{k}}(b) \leq Q^{\pi_{\mathcal{C}_{i} \cup \{i\}}^{k}}(b, \mathbf{a}_{\scriptscriptstyle\mathcal{C}_{i}}^{k}, \pi_{i}^{k+1}(b))$.
        
        If we conduct the above policy improvement for every agent, it is not difficult to induce that for all $\mathcal{C} \ \mathlarger{\mathlarger{\subseteq}} \ \mathcal{N}$, $b \in \mathcal{B}(\mathcal{CS})$ and $\mathcal{CS} \in \Psi(\mathcal{C}, \Lambda(\mathcal{CS}))$, we have
        \begin{equation*}
            V^{\pi_{\mathcal{C}}^{k}}(b) = Q^{\pi_{\mathcal{C}}^{k}}(b, \mathbf{a}_{\scriptscriptstyle\mathcal{C}}^{k}) \leq Q^{\pi_{\mathcal{C}}^{k+1}}(b, \pi_{\scriptscriptstyle\mathcal{C}}^{k+1}(b)) = V^{\pi_{\mathcal{C}}^{k+1}}(b).
        \end{equation*}
    \end{proof}
    
    \begingroup
    \def\theproposition{\ref{prop:policy_iteration}}
    \begin{proposition}
        For all $\mathcal{C} \ \mathlarger{\mathlarger{\subseteq}} \ \mathcal{N}$, $b \in \mathcal{B}(\mathcal{CS})$ and $\mathcal{CS} \in \Psi(\mathcal{C}, \Lambda(\mathcal{CS}))$, partially observable Shapley policy iteration converges to the optimal coalition values and the optimal joint policy.
    \end{proposition}
    \endgroup
    \begin{proof}
        The proof sketch follows Proposition 4.6.1 in \cite{bertsekas2019reinforcement}. Here we rewrite $Pr(\mathbf{o}' | b, \mathbf{a}, \mathcal{CS})$ as $T(\mathbf{o}' | b, \mathbf{a})$ and $\tau(\mathbf{o}', \mathbf{a}, b | \mathcal{CS})$ as $\zeta(\mathbf{o}', \mathbf{a}, b)$ for any $\mathcal{CS} \in \Psi(\mathcal{C}, \Lambda(\mathcal{CS}))$. For any $k$, we first consider policy evaluation for policy $\pi^{k+1}$:
        \begin{equation*}
            V_{m+1}^{\mathcal{C}}(b) \leftarrow R_{b}(b, \mathbf{a}_{\scriptscriptstyle\mathcal{C}}^{k+1}) + \gamma \sum_{\mathbf{o}' \in \mathcal{O}} T(\mathbf{o}' | b, \mathbf{a}^{k+1}) V_{m}^{\mathcal{C}}(\zeta(\mathbf{o}', \mathbf{a}, b)),
        \end{equation*}
        
        for all $\mathcal{C} \ \mathlarger{\mathlarger{\subseteq}} \ \mathcal{N}$, $b \in \mathcal{B}(\mathcal{CS})$ and $\mathcal{CS} \in \Psi(\mathcal{C}, \Lambda(\mathcal{CS}))$, $m = 0, 1, ...$, and 
        \begin{equation*}
            V_{0}^{\mathcal{C}}(b) = V^{\pi_{\mathcal{C}}^{k}}(b).
        \end{equation*}
        
        By the result from Lemma \ref{lemm:policy_improvemnent_msv}, \textit{after the procedure of policy improvement}, for all $\mathcal{C} \ \mathlarger{\mathlarger{\subseteq}} \ \mathcal{N}$, $b \in \mathcal{B}(\mathcal{CS})$ and $\mathcal{CS} \in \Psi(\mathcal{C}, \Lambda(\mathcal{CS}))$, we have 
        \begin{equation*}
            V_{0}^{\mathcal{C}}(b)
            \leq R_{b}(b, \mathbf{a}_{\scriptscriptstyle\mathcal{C}}^{k+1}) + \gamma \sum_{\mathbf{o}' \in \mathcal{O}} T(\mathbf{o}' | b, \mathbf{a}^{k+1}) V_{0}^{\mathcal{C}}(\zeta(\mathbf{o}', \mathbf{a}^{k+1}, b))
            = V_{1}^{\mathcal{C}}(b).
        \end{equation*}
        
        Using the above inequality, we have
        \begin{align*}
            V_{1}^{\mathcal{C}}(b) &= R_{b}(b, \mathbf{a}_{\scriptscriptstyle\mathcal{C}}^{k+1}) + \gamma \sum_{\mathbf{o}' \in \mathcal{O}} T(\mathbf{o}' | b, \mathbf{a}^{k+1}) V_{0}^{\mathcal{C}}(\zeta(\mathbf{o}', \mathbf{a}^{k+1}, b)) \\
            &\leq R_{b}(b, \mathbf{a}_{\scriptscriptstyle\mathcal{C}}^{k+1}) + \gamma \sum_{\mathbf{o}' \in \mathcal{O}} T(\mathbf{o}' | b, \mathbf{a}^{k+1}) V_{1}^{\mathcal{C}}(\zeta(\mathbf{o}', \mathbf{a}^{k+1}, b)) = V_{2}^{\mathcal{C}}(b).
        \end{align*}
        
        To continue similarly, we have the following inequalities such that
        \begin{equation*}
            V_{0}^{\mathcal{C}}(b) \leq V_{1}^{\mathcal{C}}(b) \leq V_{2}^{\mathcal{C}}(b) \leq ... \leq V_{m}^{\mathcal{C}}(b) \leq V_{m+1}^{\mathcal{C}}(b) \leq ...,
        \end{equation*}
        
        for all $\mathcal{C} \ \mathlarger{\mathlarger{\subseteq}} \ \mathcal{N}$, $b \in \mathcal{B}(\mathcal{CS})$ and $\mathcal{CS} \in \Psi(\mathcal{C}, \Lambda(\mathcal{CS}))$.
        
        By Lemma \ref{lemm:value_iteration}, we know that $V_{m}^{\mathcal{C}}(b) \rightarrow V^{\pi_{\mathcal{C}}^{k+1}}(b)$, as $m \rightarrow \infty$. As a result, we obtain that $$V^{\pi_{\mathcal{C}}^{k}}(b) = V_{0}^{\mathcal{C}}(b) \leq V^{\pi_{\mathcal{C}}^{k+1}}(b),$$ for all $\mathcal{C} \ \mathlarger{\mathlarger{\subseteq}} \ \mathcal{N}$, $b \in \mathcal{B}(\mathcal{CS})$ and $\mathcal{CS} \in \Psi(\mathcal{C}, \Lambda(\mathcal{CS}))$, and $k = 0, 1, ...$.
        
        It is easy to see that the sequence of generated coalition values is improving, and since the number of stationary policies is assumed to be finite, after a finite number of iterations, we can obtain that $V^{\pi_{\mathcal{C}}^{k}}(b) = V^{\pi_{\mathcal{C}}^{k+1}}(b)$, for all $\mathcal{C} \ \mathlarger{\mathlarger{\subseteq}} \ \mathcal{N}$, $b \in \mathcal{B}(\mathcal{CS})$ and $\mathcal{CS} \in \Psi(\mathcal{C}, \Lambda(\mathcal{CS}))$.
        
        Therefore, for all $\mathcal{C} \ \mathlarger{\mathlarger{\subseteq}} \ \mathcal{N}$, $b \in \mathcal{B}(\mathcal{CS})$ and $\mathcal{CS} \in \Psi(\mathcal{C}, \Lambda(\mathcal{CS}))$, we can obtain the equality
        \begin{equation*}
            V^{\pi_{\mathcal{C}}^{k}}(b) = \max_{\mathbf{a}_{\mathcal{C}}} R_{b}(b, \mathbf{a}_{\scriptscriptstyle\mathcal{C}}) + \gamma \sum_{\mathbf{o}' \in \mathcal{O}} T(\mathbf{o}' | b, \mathbf{a}) V^{\pi_{\mathcal{C}}^{k}}(\zeta(\mathbf{o}', \mathbf{a}, b)).
        \end{equation*}
        
        As a result, $V^{\pi_{\mathcal{C}}^{k}}(b)$ solves the Bellman equation and we can write it as $V^{\pi_{\mathcal{C}}^{k}}(b) = V^{\pi_{\mathcal{C}}^{*}}(b)$. It is not difficult to see that if setting $\mathcal{CS} = \{ \mathcal{N} \}$ and $\mathcal{C} = \mathcal{N}$, we can directly get that $V^{\pi^{k}}(b) = V^{\pi^{*}}(b)$ and conclude that $\pi^{k}$ converges to the optimal joint policy.
    \end{proof}
    
    \begingroup
    \def\theproposition{\ref{prop:value_iteration_pomcg}}
    \begin{proposition}
        For all $\mathcal{C} \ \mathlarger{\mathlarger{\subseteq}} \ \mathcal{N}$, $b \in \mathcal{B}(\mathcal{CS})$ and $\mathcal{CS} \in \Psi(\mathcal{C}, \Lambda(\mathcal{CS}))$, the partially observable Shapley value iteration such that
        \begin{equation}
            Q_{m+1}(b, \mathbf{a}_{\scriptscriptstyle \mathcal{C}}) \leftarrow R_{b}(b, \mathbf{a}_{\scriptscriptstyle\mathcal{C}}) + \gamma \sum_{\mathbf{o}' \in \mathcal{O}} Pr(\mathbf{o}' | b, \mathbf{a}, \mathcal{CS}) \max_{\mathbf{a}_{\scriptscriptstyle \mathcal{C}}'} Q_{m}(\tau(\mathbf{o}', \mathbf{a}, b | \mathcal{CS}), \mathbf{a}_{\scriptscriptstyle \mathcal{C}}'),
        \end{equation}
        converges to the optimal coalition Q-values and the optimal joint policy.
    \end{proposition}
    \endgroup
    \begin{proof}
        If we rewrite $Pr(\mathbf{o}' | b, \mathbf{a}, \mathcal{CS})$ as $T(\mathbf{o}' | b, \mathbf{a})$ and $\tau(\mathbf{o}', \mathbf{a}, b | \mathcal{CS})$ as $\zeta(\mathbf{o}', \mathbf{a}, b)$ for any $\mathcal{CS} \in \Psi(\mathcal{C}, \Lambda(\mathcal{CS}))$, the value iteration for $\pi_{\scriptscriptstyle\mathcal{C}}$ is equivalently expressed as:
        \begin{equation}
        \label{eq:shapley_value_iteration_new}
            Q_{m+1}(b, \mathbf{a}_{\scriptscriptstyle \mathcal{C}}) \leftarrow R_{b}(b, \mathbf{a}_{\scriptscriptstyle\mathcal{C}}) + \gamma \sum_{\mathbf{o}' \in \mathcal{O}} T(\mathbf{o}' | b, \mathbf{a}) \max_{\mathbf{a}_{\scriptscriptstyle \mathcal{C}}'} Q_{m}(\zeta(\mathbf{o}', \mathbf{a}, b), \mathbf{a}_{\scriptscriptstyle \mathcal{C}}').
        \end{equation}
        
        If we represent $\max_{\mathbf{a}_{\scriptscriptstyle \mathcal{C}}} Q(b, \mathbf{a}_{\scriptscriptstyle \mathcal{C}})$ as $\hat{V}(b)$, the above equation can be transformed as follows:
        \begin{equation}
        \label{eq:shapley_value_iteration_new_1}
            \hat{V}_{m+1}(b) \leftarrow \max_{\mathbf{a}_{\scriptscriptstyle\mathcal{C}}} R_{b}(b, \mathbf{a}_{\scriptscriptstyle\mathcal{C}}) + \gamma \sum_{\mathbf{o}' \in \mathcal{O}} T(\mathbf{o}' | b, \mathbf{a}) \hat{V}_{m}(\zeta(\mathbf{o}', \mathbf{a}, b)).
        \end{equation}
        
        By Theorem \ref{thm:belief_mdp_vi}, it is not difficult to see that Eq.~\ref{eq:shapley_value_iteration_new_1} converges to the following Bellman optimality equation of the transformed coalition value such that
        \begin{equation}
        \label{eq:shapley_value_iteration_new_2}
            \hat{V}^{\pi^{*}_{\mathcal{C}}}(b) = \max_{\mathbf{a}_{\scriptscriptstyle\mathcal{C}}} R_{b}(b, \mathbf{a}_{\scriptscriptstyle\mathcal{C}}) + \gamma \sum_{\mathbf{o}' \in \mathcal{O}} T(\mathbf{o}' | b, \mathbf{a}) \hat{V}^{\pi^{*}_{\mathcal{C}}}(\zeta(\mathbf{o}', \mathbf{a}, b)),
        \end{equation}
        
        where $\pi^{*}_{\scriptscriptstyle\mathcal{C}}$ is the optimal policy of coalition $\mathcal{C}$. If setting $\mathcal{CS} = \{ \mathcal{N} \}$ and $\mathcal{C} = \mathcal{N}$, we can directly get that the optimal joint policy is achieved. Since $\hat{V}^{\pi^{*}_{\mathcal{C}}}(b) = \max_{\mathbf{a}_{\scriptscriptstyle \mathcal{C}}} Q^{\pi^{*}_{\mathcal{C}}}(b, \mathbf{a}_{\scriptscriptstyle \mathcal{C}})$, we can get that the optimal coalition Q-values is obtained.
    \end{proof}
\chapter{Experimental Details}

\section{Experimental Details of Benchmarks for SQDDPG}
\label{sec:experimental_settings_of_benchmark_for_sqddpg}
    As for the setups of experiments, because different environments may involve variant complexity and dynamics, we give different hyperparameters for variant tasks.\footnote{The open-source code is released on \url{https://github.com/hsvgbkhgbv/SQDDPG}.} All of algorithms use MLPs as hidden layers for the policy networks. All of policy networks only use one hidden layer. About the critic networks, every algorithm uses MLPs with one hidden layer. For each experiment, we maintain the learning rate, entropy regularization coefficient, update frequency, batch size and the number of hidden units identical on each algorithm, except for the algorithms with the natural gradients (e.g., COMA and A2C). These algorithms need the special learning rates to maintain the stability of training. In experiments, each agent has its own observation in execution for the policy. During training, the agents with the centralised critics share the observations while those with the decentralised critics only observe its own observation. The rest of details in experimental setups are introduced as below. All of models are trained by the Adam optimizer \cite{kingma2014adam} with the default hyperparameters (except for the learning rate).
    
    \subsection{Additional Details of Cooperative Navigation} 
        The specific hyperparameters of the algorithms solving Cooperative Navigation are shown in Table \ref{tab:hyperparameters_cn}. 
        \begin{table*}[ht!]
        \caption{Hyperparameters for Cooperative Navigation.}
        \centering
        \scalebox{0.85}{
            \begin{tabular}{lcl}
            \toprule
            \textbf{Hyperparameters}               & \textbf{\#} & \textbf{Description}                                                \\ 
            \midrule
            hidden units                           & 32            & The number of hidden units for both policy and critic network \\
            training episodes                      & 5000          & The number of training episodes                               \\
            episode length                         & 200           & Maximum time steps per episode                                \\
            discount factor                        & 0.9           & The importance of future rewards                              \\
            update frequency for behaviour network & 100           & Behaviour network updates every \# steps                      \\
            learning rate for policy network       & 1e-4          & Policy network learning rate                                  \\
            learning rate for policy network(COMA)  & 1e-2          & Policy network learning rate for COMA                                 \\
            learning rate for policy network(IA2C)  & 1e-6          & Policy network learning rate for IA2C                                 \\                
            learning rate for critic network       & 1e-3          & Critic network learning rate                                  \\
            learning rate for critic network(COMA)  & 1e-4          & Critic network learning rate for COMA                                 \\
            learning rate for critic network(IA2C)  & 1e-5          & Critic network learning rate for IA2C                                 \\
            update frequency for target network    & 200           & Target network updates every \# steps                         \\
            target update rate                     & 0.1           & Target network update rate                                    \\
            entropy regularization coefficient     & 1e-2          & Weight or regularization for exploration                      \\
            batch size                             & 32            & The number of  transitions for each update                    \\ 
            \bottomrule
            \end{tabular}}
            \setlength{\abovecaptionskip}{10pt}
        \label{tab:hyperparameters_cn}
        \end{table*}
    
    \subsection{Additional Details of Prey-and-Predator} 
        The specific hyperparameters of each algorithm solving Predator-Prey are shown in Table \ref{tab:hyperparameters_pp}.
        \begin{table*}[ht!]
        \centering
        \caption{Hyperparameters for Predator-Prey.}
        \scalebox{0.8}{
            \begin{tabular}{lcl}
            \toprule
            \textbf{Hyperparameters}               & \textbf{\#} & \textbf{Description}                                                \\ 
            \midrule
            hidden units                           & 128           & The number of hidden units for both policy and critic network \\
            training episodes                      & 5000          & The number of training episodes                               \\
            episode length                         & 200           & Maximum time steps per episode                                \\
            discount factor                        & 0.99          & The importance of future rewards                              \\
            update frequency for behaviour network & 100           & Behaviour network updates every \# steps                      \\
            learning rate for policy network       & 1e-4          & Policy network learning rate                                  \\
            learning rate for policy network(COMA/IA2C)  & 1e-3    & Policy network learning rate for COMA and IA2C                                 \\
            learning rate for critic network       & 5e-4          & Critic network learning rate                                  \\
            learning rate for critic network(COMA/IA2C)  & 1e-4    & Critic network learning rate for COMA and IA2C                                 \\
            update frequency for target network    & 200           & Target network updates every \# steps                         \\
            target update rate                     & 0.1           & Target network update rate                                    \\
            entropy regularization coefficient     & 1e-3          & Weight or regularization for exploration                      \\
            batch size                             & 128           & The number of  transitions for each update                    \\ 
            \bottomrule
            \end{tabular}}
            \setlength{\abovecaptionskip}{10pt} 
        \label{tab:hyperparameters_pp}
        \end{table*}

    \subsection{Additional Details of Traffic Junction} 
        The specific hyperparameters of the algorithms solving Traffic Junction are shown in Table \ref{tab:hyperparameters_tj}. To exhibit the training procedure in more details, we also show the figures of mean rewards (see Figure \ref{fig:traffic_easy_reward}$\sim$\ref{fig:traffic_hard_reward}) and the figures of success rate (see Figure \ref{fig:traffic_easy_success}$\sim$\ref{fig:traffic_hard_success}).
        \begin{figure*}[ht!]
            \centering
            \begin{subfigure}[b]{0.48\textwidth}
                \includegraphics*[width=\textwidth]{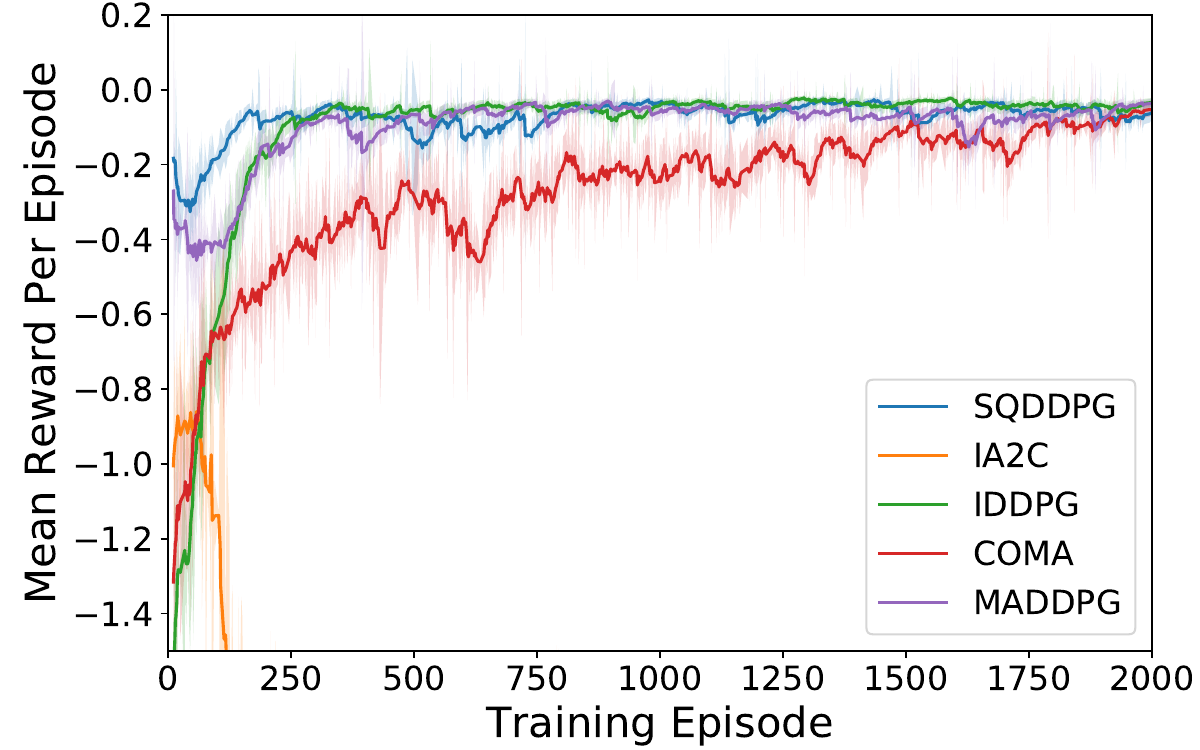}
                \caption{Mean reward per episode during training in Traffic Junction on easy version.}
                \label{fig:traffic_easy_reward}
            \end{subfigure}
            \begin{subfigure}[b]{0.48\textwidth}
                \includegraphics*[width=\textwidth]{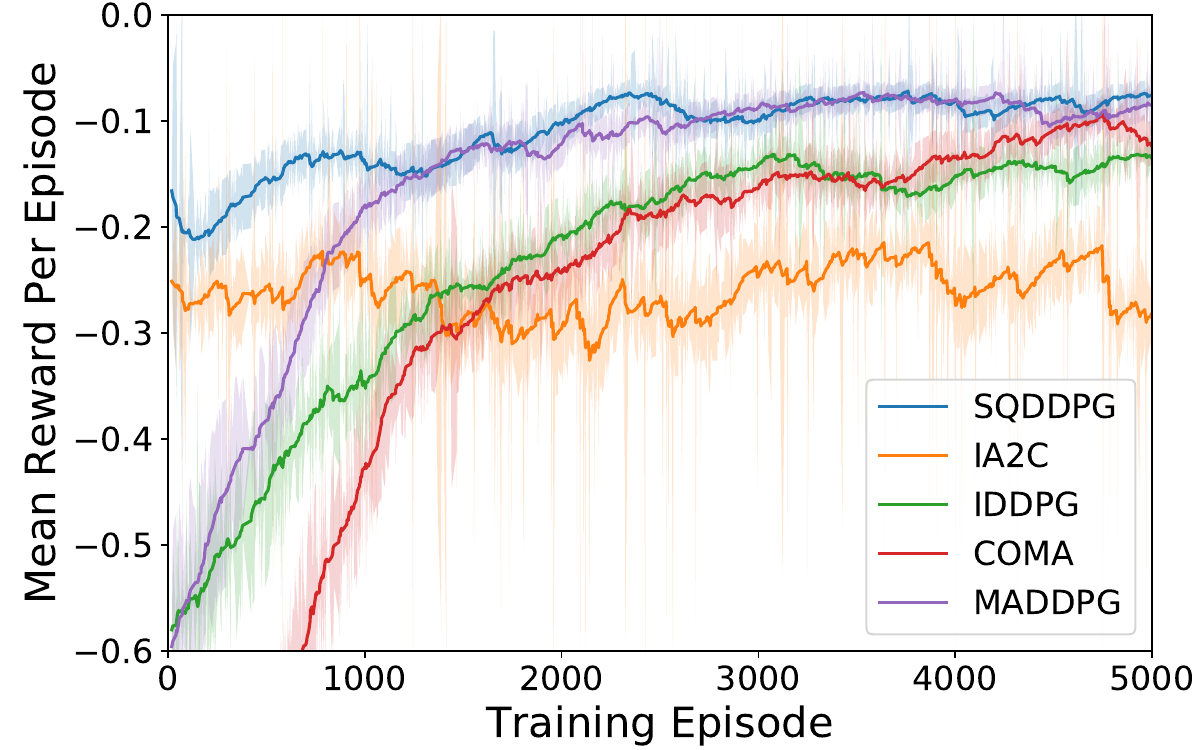}
                \caption{Mean reward per episode during training in Traffic Junction on medium version.}
                \label{fig:traffic_medium_reward}
            \end{subfigure}
            \begin{subfigure}[b]{0.48\textwidth}
                \includegraphics*[width=\textwidth]{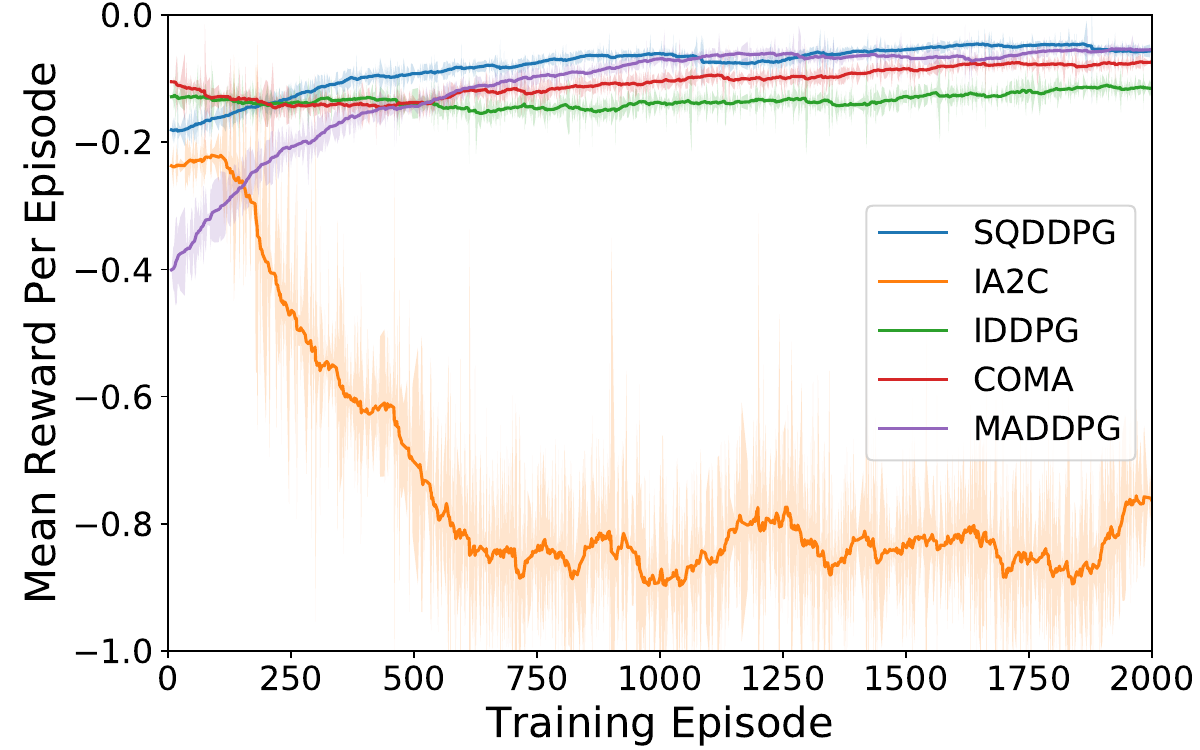}
                \caption{Mean reward per episode during training in Traffic Junction on high version.}
                \label{fig:traffic_hard_reward}
            \end{subfigure}
            \quad
            \begin{subfigure}[b]{0.48\textwidth}
                \includegraphics*[width=\textwidth]{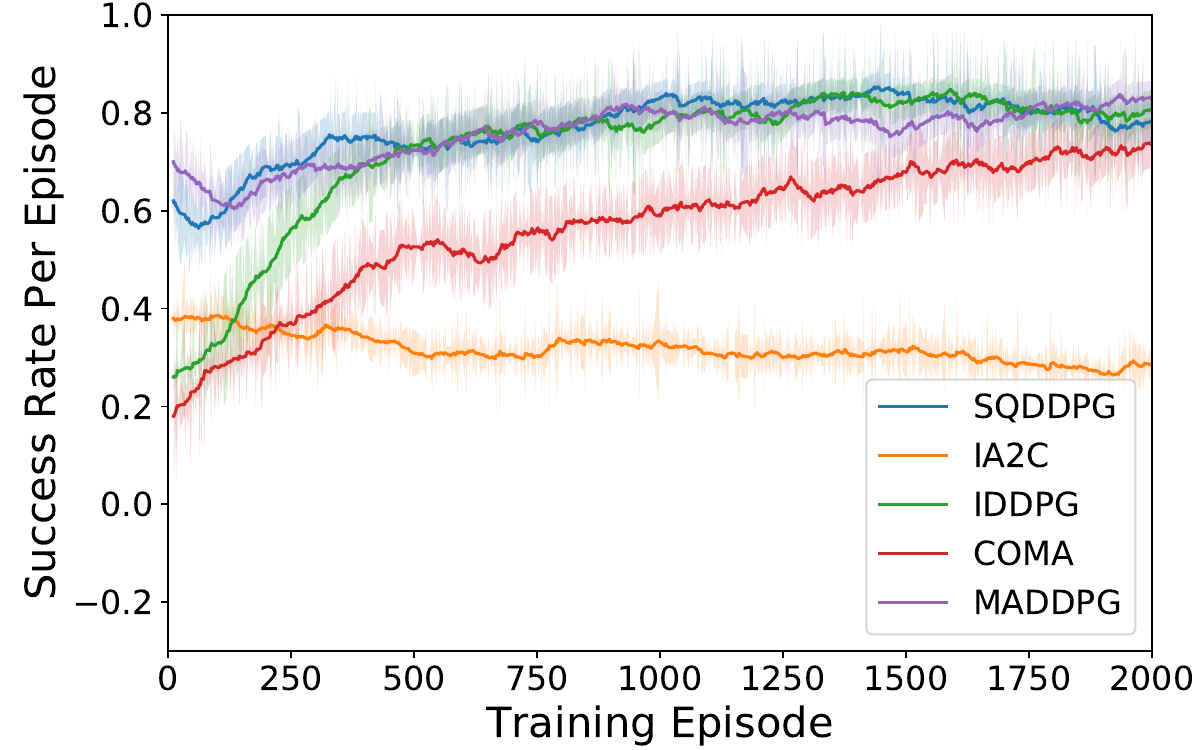}
                \caption{Success rate per episode during training in Traffic Junction on easy version.}
                \label{fig:traffic_easy_success}
            \end{subfigure}
            \begin{subfigure}[b]{0.48\textwidth}
                \includegraphics*[width=\textwidth]{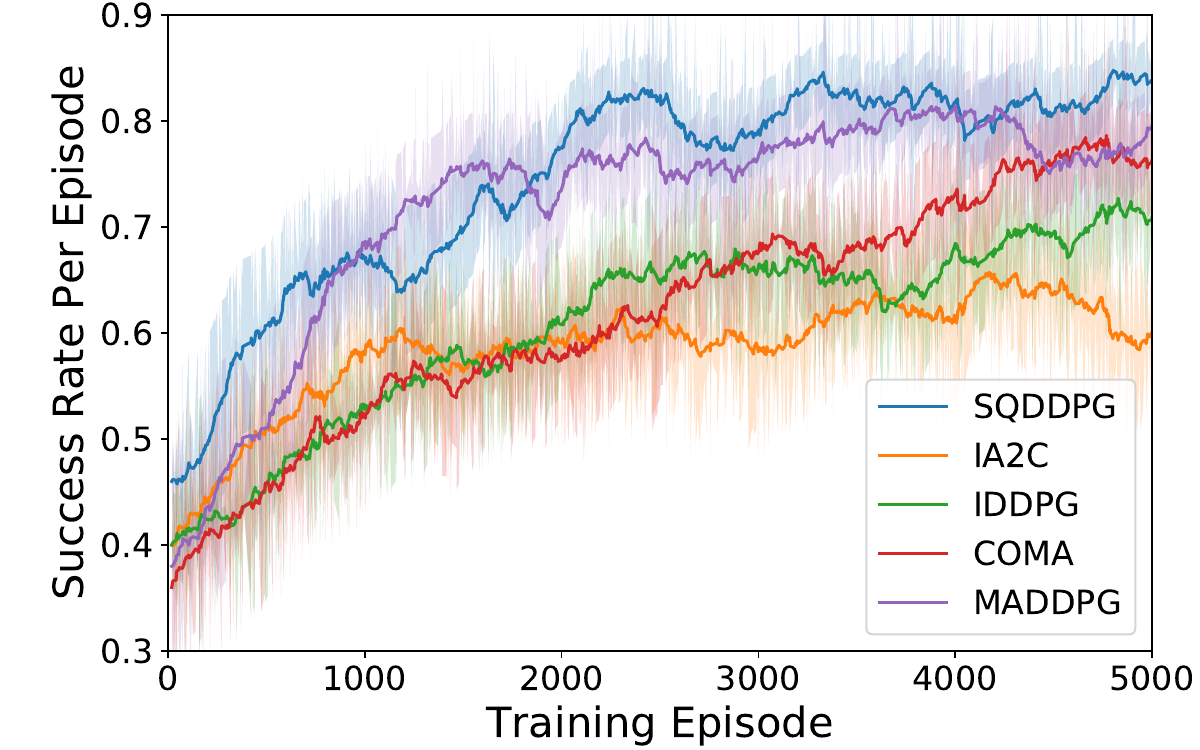}
                \caption{Success rate per episode during training in Traffic Junction on medium version.}
                \label{fig:traffic_medium_success}
            \end{subfigure}
            \begin{subfigure}[b]{0.48\textwidth}
                \includegraphics*[width=\textwidth]{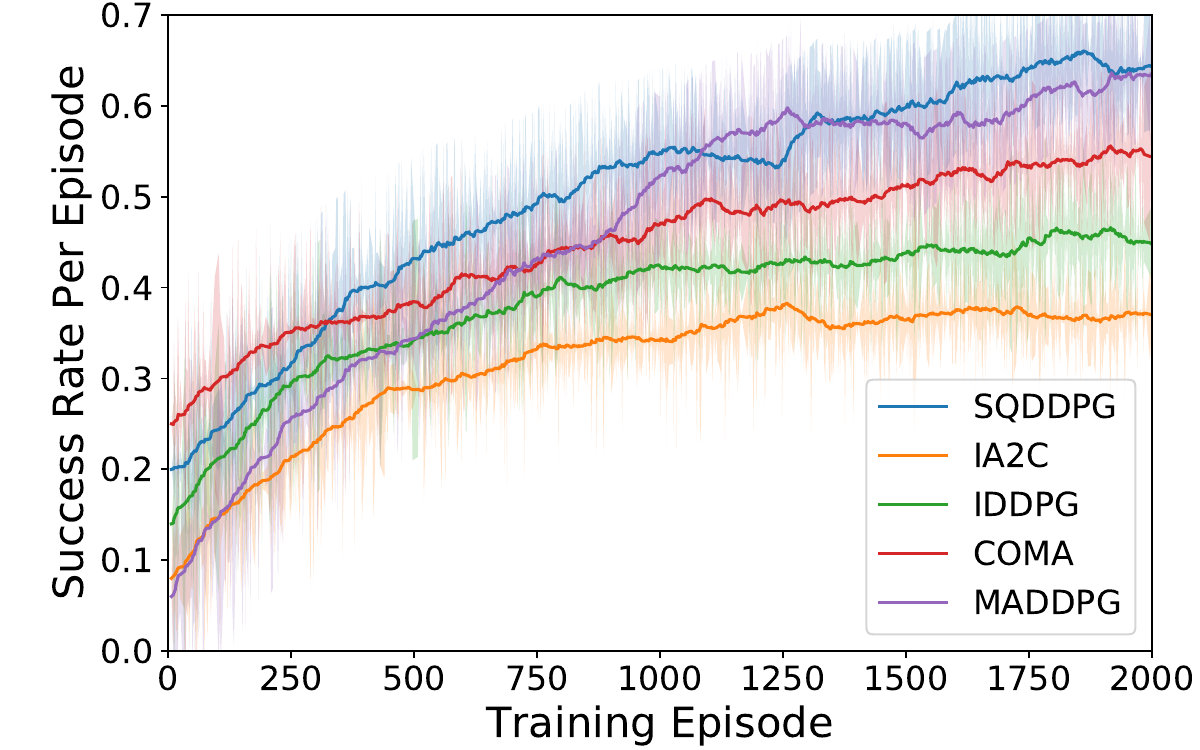}
                \caption{Success rate per episode during training in Traffic Junction on hard version.}
                \label{fig:traffic_hard_success}
            \end{subfigure}
            \caption{Mean reward and success rate per episode during training in the Traffic Junction environment on all difficulty levels.}
            \label{fig:traffic_junction_curve}
        \end{figure*}
        
        \begin{table*}[ht!]
        \caption{Setting of Traffic Junction for different difficulty levels. \textbf{$p_{\text{arrive}}$} means the probability to add an available car into the environment. \textbf{$N_{\max}$} means the existing number of the cars. \textbf{Entry-Points \#} means the number of possible entry points for each car. \textbf{Routes \#} means the number of possible routes starting from every entry point.}
        \centering
        \scalebox{1.0}{
            \begin{tabular}{cccccccc}
                \toprule
                \textbf{Difficulty}  & \textbf{$p_{\text{arrive}}$} & \textbf{$N_{\max}$} & \textbf{Entry-Points \#} & \textbf{Routes \#} &  \textbf{Two-way} &\textbf{Junctions \#} & \textbf{Dimension} \\ 
                \midrule
                Easy  & 0.3  &  5  & 2  & 1 & F & 1 & 7x7 \\
                Medium & 0.2 &  10  & 4  & 3 & T & 1 & 14x14 \\
                Hard & 0.05  &  20  & 8  & 7 & T & 4 & 18x18 \\
                \bottomrule
            \end{tabular}}
        \label{tab:tj_setting}
        \end{table*}
    
        \begin{table*}[ht!]
        \caption{Hyperparameters for Traffic Junction.}
        \centering
        \scalebox{0.73}{
            \begin{tabular}{lcccl}
            \toprule
            \textbf{Hyperparameters}               & \textbf{Easy} & \textbf{Meidum} & \textbf{Hard} & \textbf{Description}                                                \\ 
            \midrule
            hidden units                           & 128           & 128             & 128           & The number of hidden units for both policy and critic network \\
            training episodes                      & 2000          & 5000            & 2000          & The number of training episodes                               \\
            episode length                         & 50            & 50              & 100           & Maximum time steps per episode                                \\
            discount factor                        & 0.99          & 0.99            & 0.99          & The importance of future rewards                              \\
            update frequency for behaviour network & 25            & 25              & 25            & Behaviour network updates every \# steps                      \\
            learning rate for policy network       & 1e-4          & 1e-4            & 1e-4          & Policy network learning rate                                  \\
            learning rate for critic network       & 1e-3          & 1e-3            & 1e-3          & Critic network learning rate                                  \\
            update frequency for target network    & 50            & 50              & 50            & Target network updates every \# steps                         \\
            target update rate                     & 0.1           & 0.1             & 0.1           & Target network update rate                                    \\
            entropy regularization coefficient     & 1e-4          & 1e-4            & 1e-4          & Weight or regularization for exploration                      \\
            batch size                             & 64            & 32              & 32            & The number of  transitions for each update                    \\ 
            \bottomrule
            \end{tabular}}
            \setlength{\abovecaptionskip}{10pt}
        \label{tab:hyperparameters_tj}
        \end{table*}

\section{Experimental Details of Benchmarks for SHAQ}
\label{sec:experimental_settings_of_benchmark_for_shaq}
    \subsection{Implementation Details of Shapley Q-learning}
    \label{subsec:implementation_details_shapley_q_learning}
        We now provide the additional implementation details that are omitted from the main part of this thesis.\footnote{The open-source code of the implementation of SHAQ is released on \url{https://github.com/hsvgbkhgbv/shapley-q-learning}.} First, $\mathit{F}_{s}(\cdot, \cdot)$ is a 3-layer network (consecutively with two affine transformation and an activation of absolute), where the hidden-layer dimension is 32. The parameters of each affine transformation are generated by hyper-networks \cite{ha2017hyper} with the global state as the input, whose details are shown in Table \ref{tab:SHAQ_hypernet}. The architecture of each agent's Q-value is a RNN with GRUs cell \cite{chung2014empirical}, whose hidden-layer dimension is 64. The input dimension is state dimension and the output dimension is action dimension.
        \begin{table}[ht!]
        	\caption{The specifications for $\mathit{F}_{s}(\cdot, \cdot)$.}
        	\begin{center} 
        		\begin{small}
        			\begin{sc}
        			  \scalebox{0.95}{
        				\begin{tabular}{ll}
        					\toprule
        					\textbf{Network} & \textbf{Structure} \\
        					\midrule
        					1st weight matrix & [ linear(state\_dim, 64), ReLU, linear(64, 32*2), absolute ]\\
        					1st bias & [ linear(state\_dim, 64) ]\\
        				    2nd weight matrix & [ linear(state\_dim, 64), ReLU, linear(64, 32), absolute ]\\
        				    2nd bias & [ linear(state\_dim, 32), ReLU, linear(32, 1) ] \\
        					\bottomrule 
        				\end{tabular}
        				}
        			\end{sc}
        		\end{small}
        	\end{center}
        \label{tab:SHAQ_hypernet}
        \end{table}
        
        Taking the lesson of training two coupling modules from GANs \cite{goodfellow2014generative}, we provide two separate learning rates for $\hat{\alpha}_{i}(\mathbf{s}, a_{i})$ and $\hat{Q}_{i}(\mathbf{s}, a_{i})$. The learning rate for $\hat{Q}_{i}(\mathbf{s}, a_{i})$ is fixed at 0.0005 for all tasks. Nevertheless, the learning rate for $\hat{\alpha}_{i}(\mathbf{s}, a_{i})$ is dependent on the number of controllable agents. We use the RMSProp optimizer \cite{tieleman2012lecture} for training in all tasks. All models are implemented in PyTorch 1.4.0 and each experiment is run on Nvidia GeForce RTX 2080Ti for 4 to 26 hours with a single process of environment.
    
    \subsection{Hyperparameters of Baselines}
    \label{subsec:hyperparameters_of_baselines}
        The hyperparameters of all baselines except for SQDDPG \cite{Wang_2020} are consistent with \cite{rashid2020weighted} and \cite{wang2020qplex}. The hyperparamers of SQDDPG are shown as follows: (1) The policy network is consistent with the other baselines, while the critic network is with 3 hidden layers and each layer whose dimension is 64. (2) The policy network is updated every 2 timesteps, while the critic network is updated each timestep. (3) The multiplier of the entropy of policy is 0.005. The rest of settings are identical with other baselines.
        
    \subsection{Predator-Prey}
    \label{subsec:predator_prey}
        We give the experimental setups of Predator-Prey \cite{bohmer2020deep} in Table \ref{tab:prey_and_predator_hyperparameters}.
        \begin{table}[ht!]
        	\caption{The experimental setups of Predator-Prey.}
        	\begin{center}
        		\begin{small}
        			\begin{sc}
        			  \scalebox{0.85}{
        				\begin{tabular}{lcll}
        					\toprule
        					\textbf{Hyperparameters} & \textbf{Value} & \textbf{Description}  \\
        					\midrule
        					batch size & 32 & The number of episodes for each update\\
        					discount factor $\gamma$ & 0.99 & The importance of future rewards  \\
        					replay buffer size & 5,000 & The maximum number of episodes to store in memory\\
        					episode length & 200 & Maximum time steps per episode \\
        					test episode & 16 & The number of episodes for evaluating the performance   \\
        					test interval  & 10,000 & The time step frequency for evaluating the performance \\
        					epsilon start & 1.0 & The start epsilon $\epsilon$ value for exploration \\
        					epsilon finish & 0.05 & The final epsilon $\epsilon$ value for exploration \\
        					exploration step & 1,000,000 & The number of steps for linearly annealing $\epsilon$  \\
        					max training step & 1,000,000 & The number of training steps \\
        				    target update interval & 200  &  The update frequency for target network \\
        					learning rate  & 0.0001 & The learning rate for $\delta_{i}(\mathbf{s}, a_{i})$  \\
        					$\alpha$ for W-QMIX variants & 0.1 & The weight for CW-QMIX and OW-QMIX  \\
        					sample size & 10 & The sample size for coalition sampling   \\
        					\bottomrule 
        				\end{tabular}
        				}
        			\end{sc}
        		\end{small}
        	\end{center}
        	\label{tab:prey_and_predator_hyperparameters}
        \end{table}
        
    \subsection{StarCraft Multi-Agent Challenge}
    \label{subsec:starcraft_multiagent_benchmark_tasks}
        In this thesis, we evaluate SHAQ on 11 typical combat scenarios in SMAC that can be classified into three categories: easy (8m, 3s5z, 1c3s5z and 10m\_vs\_11m), hard (5m\_vs\_6m, 3s\_vs\_5z and 2c\_vs\_64zg), and super-hard (3s5z\_vs\_3s6z, Corridor, MMM2 and 6h\_vs\_8z). More details of these tasks are provided in Table \ref{tab:smac_benchmarks}. The specific experimental setups for SMAC are shown in Table \ref{tab:smac_hyperparameters} and \ref{tab:relation_learning_rate_and_agents}.
        \begin{table}[t]
        	\caption{Introduction of maps and characters in SMAC.}
        	\begin{center}
        		\begin{small}
        			\begin{sc}
        			 \scalebox{0.8}{
        				\begin{tabular}{cllc}
        					\toprule
        					\textbf{Map Name} & \textbf{Ally Units} & \textbf{Enemy Units} & \textbf{Categories}  \\
        					\midrule
        					3s5z    &   3 Stalkers $\&$ 5 Zealots  &  3 Stalkers $\&$ 5 Zealots    &   easy \\
        					1c3s5z  &   1 Colossi $\&$ 3 Stalkers $\&$ 5 Zealots &    1 Colossi $\&$ 3 Stalkers $\&$ 5 Zealots   &   easy  \\
        					8m               &   8 Marines    & 8 Marines    & easy \\
        					10m\_vs\_11m     &   10 Marines   & 11 Marines   &  easy \\
        					5m\_vs\_6m       &   5 Marines  &  6 Marines   &   hard \\
        					3s\_vs\_5z       &   3 Stalkers   & 5 Zealots    &  hard\\
        					2c\_vs\_64zg     &   2 Colossi   &  64 Zerglings  & hard\\
        					3s5z\_vs\_3s6z   &  3 Stalkers $\&$ 5 Zealots  &  3 Stalkers $\&$ 6 Zealots   & super-hard \\
        					MMM2 & 1 Medivac, 2 Marauders $\&$  7 Marines     & 1 Medivac, 3 Marauders $\&$  8 Marines    &  super-hard  \\
        					6h\_vs\_8z &  6 Hydralisks    &  8 Zerglings   & super-hard  \\
        					Corridor & 6 Zealots & 24 Zerglings & super-hard \\
        					\bottomrule
        				\end{tabular}
        				}
        			\end{sc}
        		\end{small}
        	\end{center}
        	\label{tab:smac_benchmarks}
        \end{table}
    
        \begin{table}[ht!]
        	\caption{The experimental setups for SMAC.}
        	\begin{center}
        		\begin{small}
        			\begin{sc}
        			    \scalebox{0.7}{
        				\begin{tabular}{lcccl}
        					\toprule
        					\textbf{Hyperparameters} & \textbf{Easy} &\textbf{Hard} &\textbf{Super Hard} & \textbf{Description}  \\
        					\midrule
        					    batch size &    32   &   32   &   32   &   The number of episodes for each update \\
        						discount factor $\gamma$ &  0.99    &    0.99    &  0.99    &   The importance of future rewards  \\
        						replay buffer size & 5,000 & 5,000 & 5,000 & The maximum number of episodes to store in memory\\
        					    max training step  & 2,000,000    & 2,000,000  &  5,000,000    & The number of training steps\\
        					    test episode &   32   &   32   &   32  &  The number of episodes for evaluation  \\
        					    test interval  & 10,000 & 10,000 & 10,000  & The time step frequency for evaluating the performance \\
        					    epsilon start & 1.0  &  1.0  &  1.0  &  The start epsilon $\epsilon$ value for exploration \\
        					    epsilon finish  &  0.05  &   0.05   &    0.05   &   The final epsilon $\epsilon$  value for exploration \\
        					    exploration step &  50,000   &  50,000  &   1,000,000   & The number of steps for linearly annealing $\epsilon$ \\
        					    target update interval &    200   &   200   &   200   &  The update frequency for target network \\ 
        					  	$\alpha$ for OW-QMIX  &    0.5   &   0.5    &   0.5   &   The weight for OW-QMIX\\
        					  	$\alpha$ for CW-QMIX  &    0.75  &   0.75   &   0.75  &   The weight for CW-QMIX\\
        					  	sample size           &    10    &    10    &   10    &   The sample size for coalition sampling   \\
        					\bottomrule
        				\end{tabular}
        				}
        			\end{sc}
        		\end{small}
        	\end{center}
        	\label{tab:smac_hyperparameters}
        \end{table}
        
        \begin{table}[ht]
        	\caption{The learning rate for training $\hat{\alpha}_{i}(\mathbf{s}, a_{i})$ of SHAQ for various maps in SMAC.}
        	\begin{center}
        		\begin{small}
        			\begin{sc}
        				\scalebox{0.95}{
        					\begin{tabular}{ccc}
        						\toprule
        						\textbf{Map Name} & \textbf{Number of Agents} & \textbf{Learning Rate for $\hat{\alpha}_{i}(\mathbf{s}, a_{i})$} \\
        						\midrule
        						    2c\_vs\_64zg	&  2   & 0.002  \\
        						    3s\_vs\_5z 	    &  3   & 0.001  \\
        							5m\_vs\_6m		&  5   & 0.0005 \\
        							6h\_vs\_8z		&  6   & 0.0005 \\
        							Corridor        &  6   & 0.0005 \\
        							8m              &  8   & 0.0003 \\
        							3s5z			&  8   & 0.0003 \\
        							3s5z\_vs\_3s6z	&  8   & 0.0003 \\
        						    1c3s5z			&  9   & 0.0002 \\
        							10m\_vs\_11m    & 10   & 0.0001 \\
        						    MMM2            & 10   & 0.0001 \\
        						\bottomrule
        					\end{tabular}
        				}
        			\end{sc}
        		\end{small}
        	\end{center}
        	\label{tab:relation_learning_rate_and_agents}
        \end{table}
    
    \subsection{Experimental Results on Extra SMAC Maps}
    \label{subsec:experimental_results_on_extra_smac_maps}
        To thoroughly compare the performance of SHAQ with baselines, we also run experiments on 5 extra maps in SMAC as Figure \ref{fig:extra_results_smac} shows. 8m, 3s5z, 1c3s5z and 10m\_vs\_11m are an easy maps and MMM2 is a super-hard map. The strategy of epsilon annealing is consistent with the previous experiments for SMAC. It is obvious that SHAQ also performs generally well on these 5 maps.
        \begin{figure*}[ht!]
            \centering
            \begin{subfigure}[b]{0.48\linewidth}
                \centering
                \includegraphics[width=\textwidth]{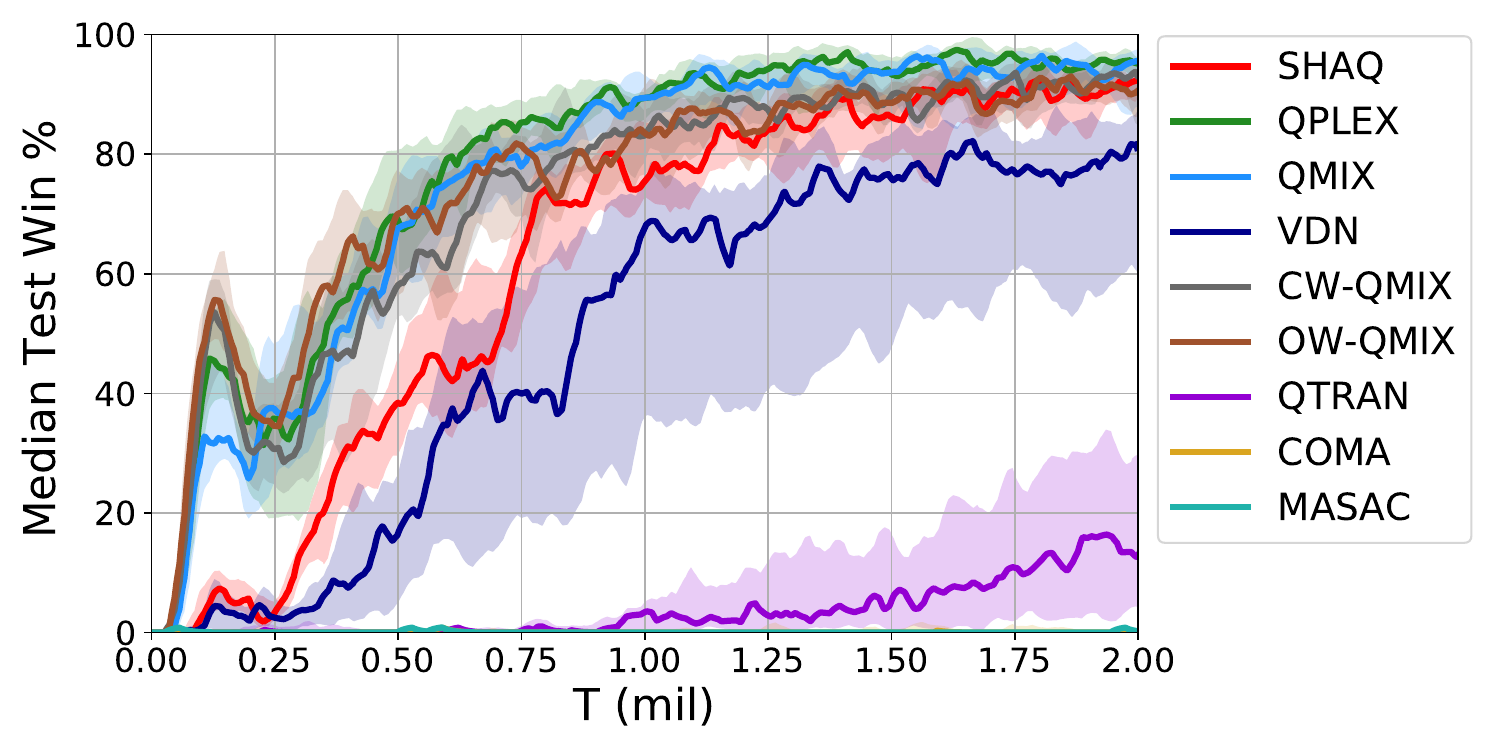}
                \caption{3s5z.}
            \end{subfigure}
            ~
            \begin{subfigure}[b]{0.48\linewidth}
                \centering
                \includegraphics[width=\textwidth]{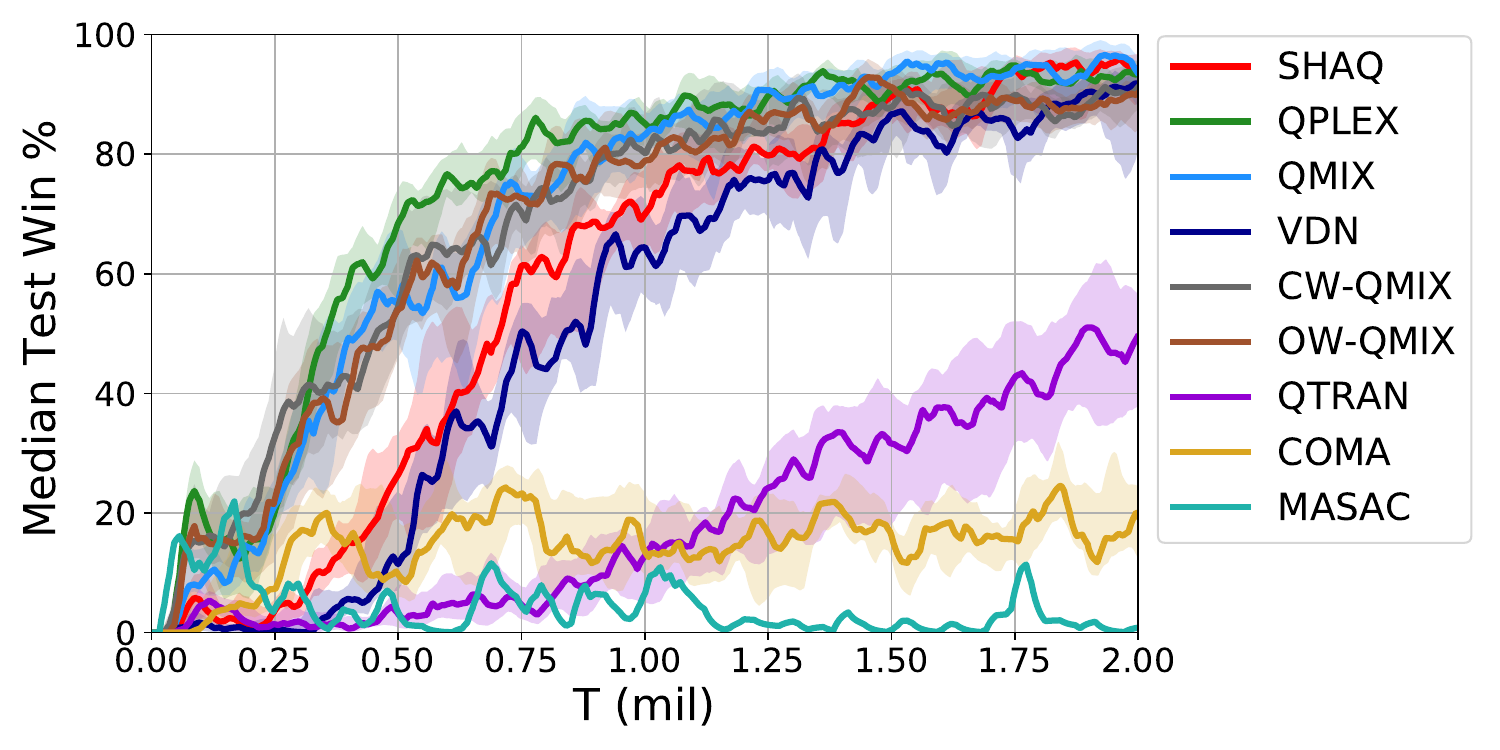}
                \caption{1c3s5z.}
            \end{subfigure}
            ~
            \begin{subfigure}[b]{0.48\linewidth}
                \centering
                \includegraphics[width=\textwidth]{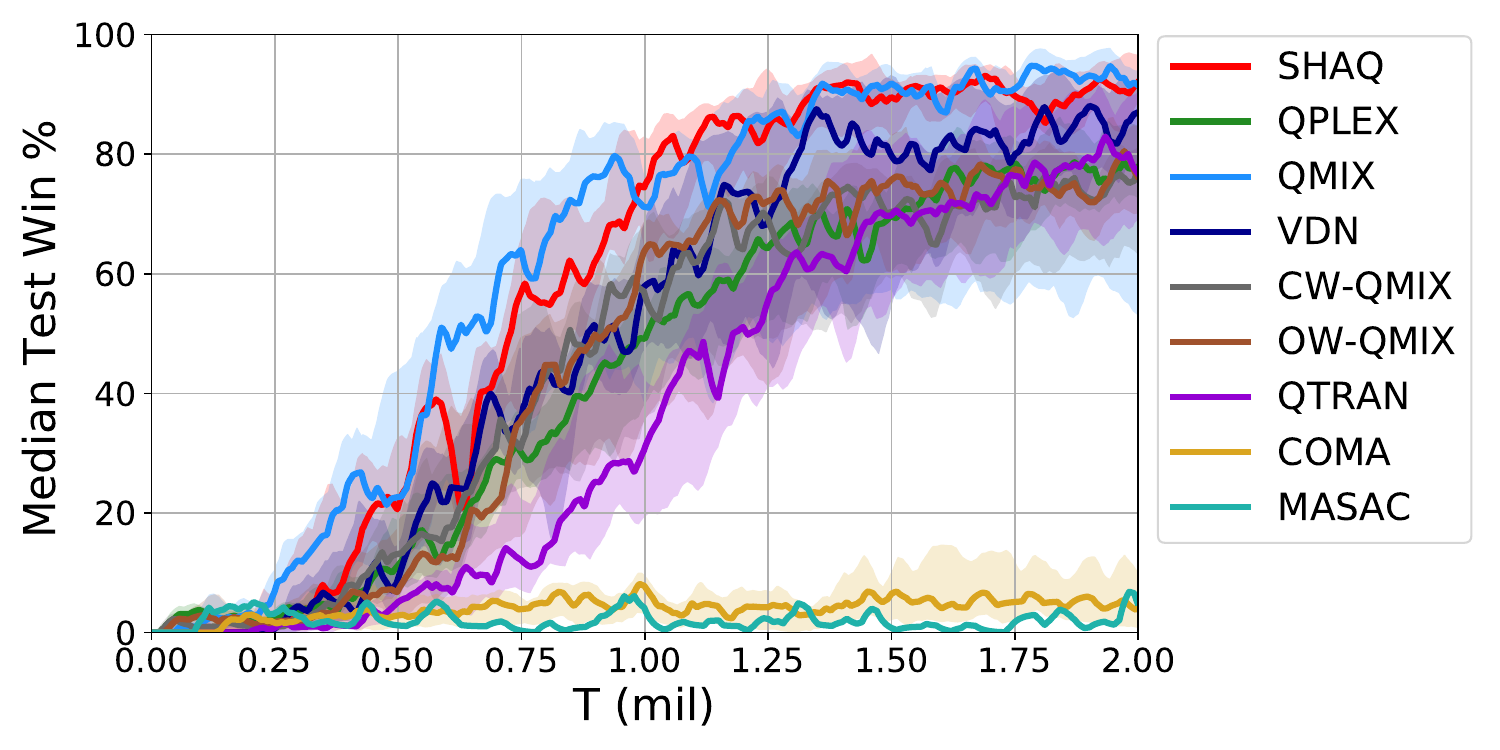}
                \caption{10m\_vs\_11m.}
            \end{subfigure}
            ~
            \begin{subfigure}[b]{0.48\textwidth}
                \centering                \includegraphics[width=\textwidth]{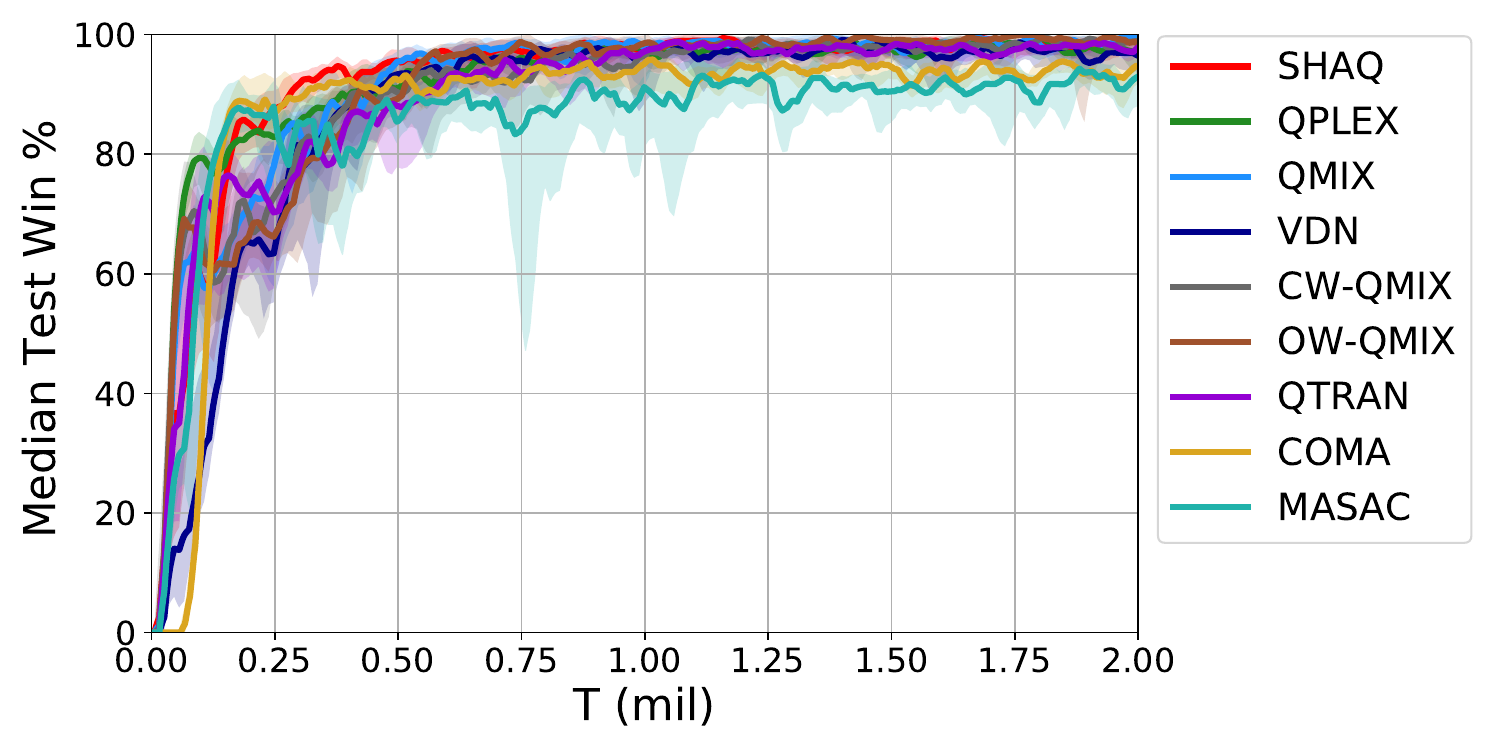}
                \caption{8m.}
            \end{subfigure}
            ~
            \begin{subfigure}[b]{0.48\textwidth}
                \includegraphics[width=\textwidth]{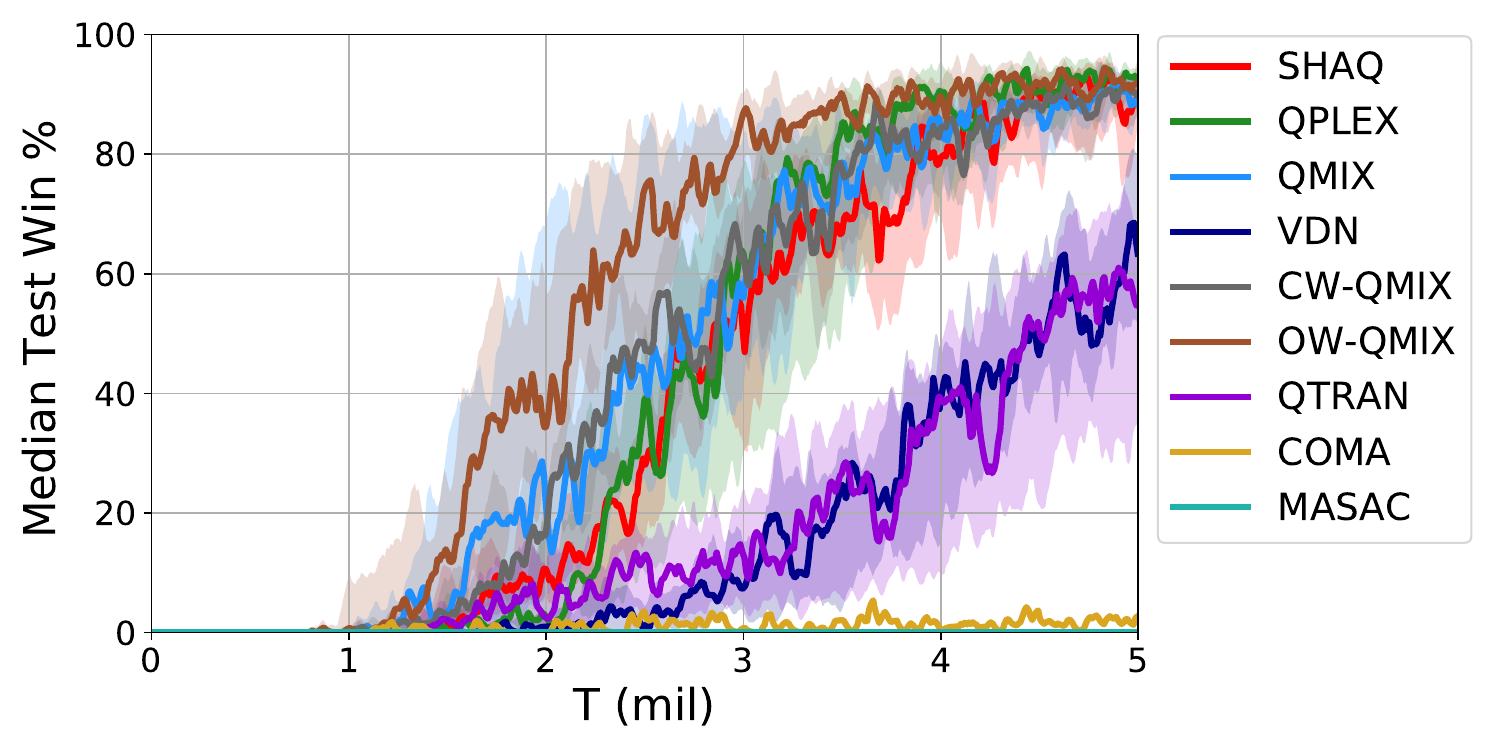}
                \caption{MMM2.}
            \end{subfigure}
            \caption{Median test win \% for 5 extra maps in SMAC.}
        \label{fig:extra_results_smac}
        \end{figure*}
        
    \subsection{Extra Animations for SMAC}
    \label{subsec:extra_animations_for_smac}
        We show the intermediate animations generated from the test of SHAQ on all maps in SMAC in Figure \ref{fig:3s5z_visualization} - \ref{fig:6h_vs_8z_visualization}, so that readers can have an intuitive picture for experiments. In these figures, the controllable agents are in red while the enemies are in blue.
        \begin{figure}[ht!]
    		\centering
    		\includegraphics[width=0.95\textwidth]{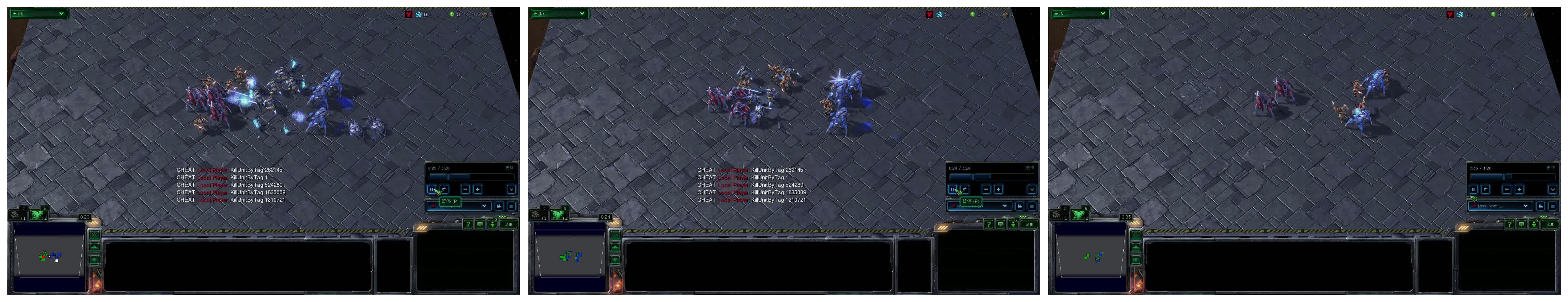}
    		\caption{Intermediate animations for 3s5z.}
    		\label{fig:3s5z_visualization}
    	\end{figure}
    	
    	\begin{figure}[ht!]
    		\centering
    		\includegraphics[width=0.95\textwidth]{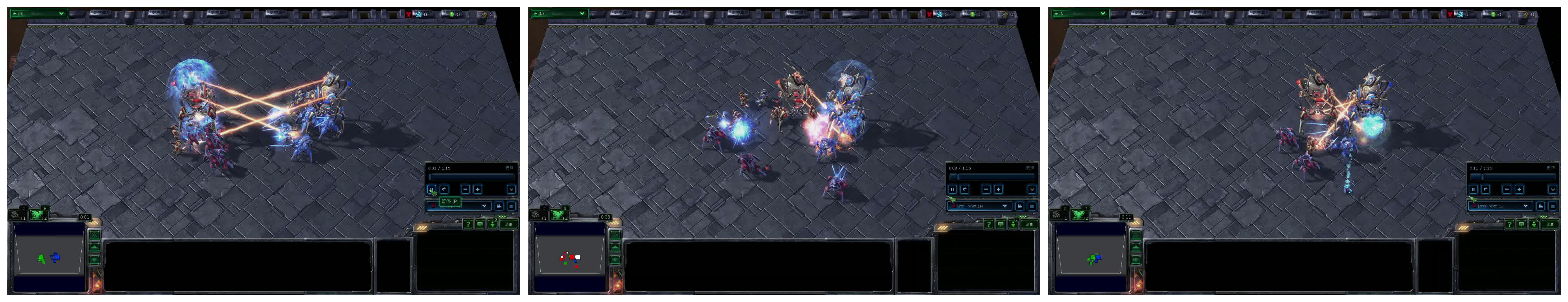}
    		\caption{Intermediate animations for 1c3s5z.}
    		\label{fig:1c3s5z_visualization}
    	\end{figure}
    	
    	\begin{figure}[ht!]
    		\centering
    		\includegraphics[width=0.95\textwidth]{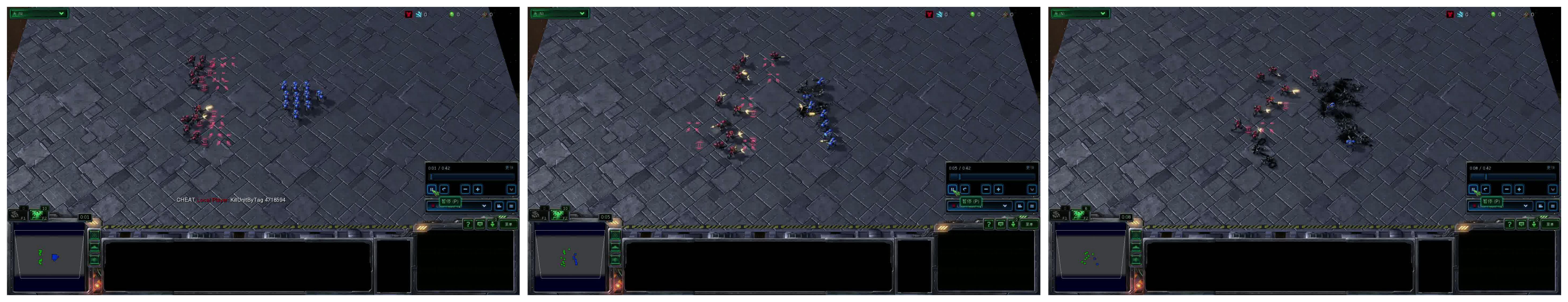}
    		\caption{Intermediate animations for 10m\_vs\_11m.}
    		\label{fig:10m_vs_11m_visualization}
    	\end{figure}
    	
    	\begin{figure}[ht!]
    		\centering
    		\includegraphics[width=0.95\textwidth]{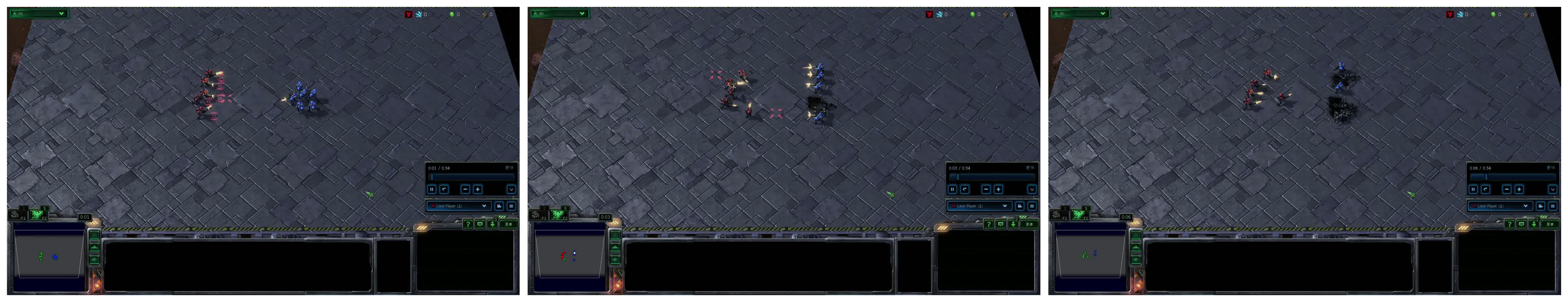}
    		\caption{Intermediate animations for 5m\_vs\_6m.}
    		\label{fig:5m_vs_6m_visualization}
    	\end{figure}
    	
    	\begin{figure}[ht!]
    		\centering
    		\includegraphics[width=0.95\textwidth]{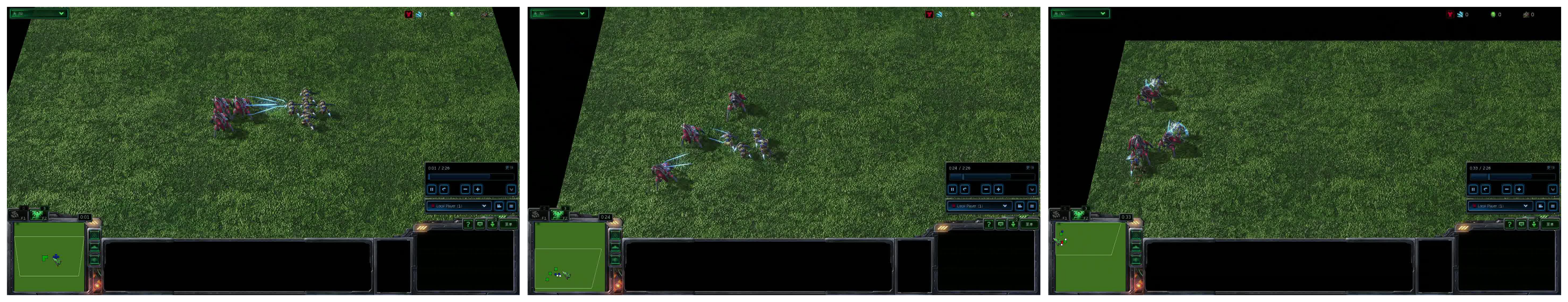}
    		\caption{Intermediate animations for 3s\_vs\_5z.}
    		\label{fig:3s_vs_5z_visualization}
    	\end{figure}
    
    	\begin{figure}[ht!]
    		\centering
    		\includegraphics[width=0.95\textwidth]{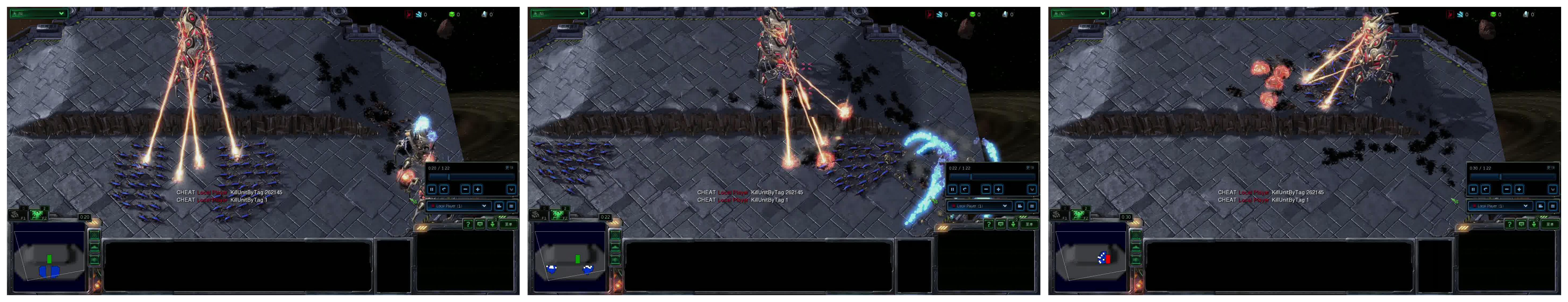}
    		\caption{Intermediate animations for 2c\_vs\_64zg.}
    		\label{fig:2c_vs_64zg_visualization}
    	\end{figure}
    	
    	\begin{figure}[ht!]
    		\centering
    		\includegraphics[width=0.95\textwidth]{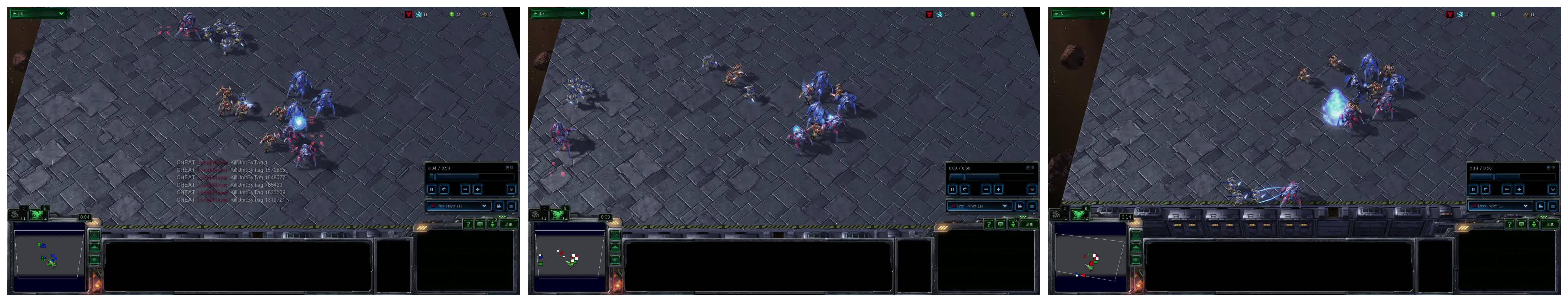}
    		\caption{Intermediate animations for 3s5z\_vs\_3s6z.}
    		\label{fig:3s5z_vs_3s6zg_visualization}
    	\end{figure}
    	
    	\begin{figure}[ht!]
    		\centering
    		\includegraphics[width=0.95\textwidth]{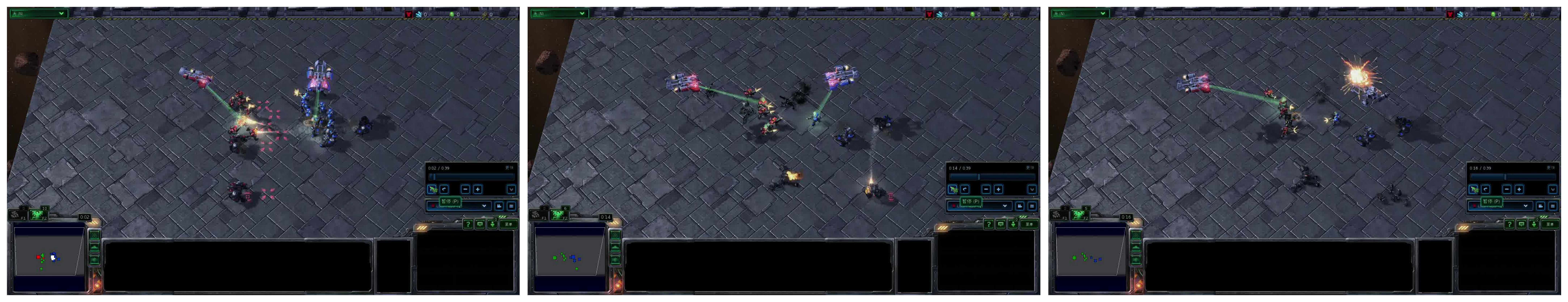}
    		\caption{Intermediate animations for MMM2.}
    		\label{fig:MMM2_visualization}
    	\end{figure}
    	
    	\begin{figure}[ht!]
    		\centering
    		\includegraphics[width=0.95\textwidth]{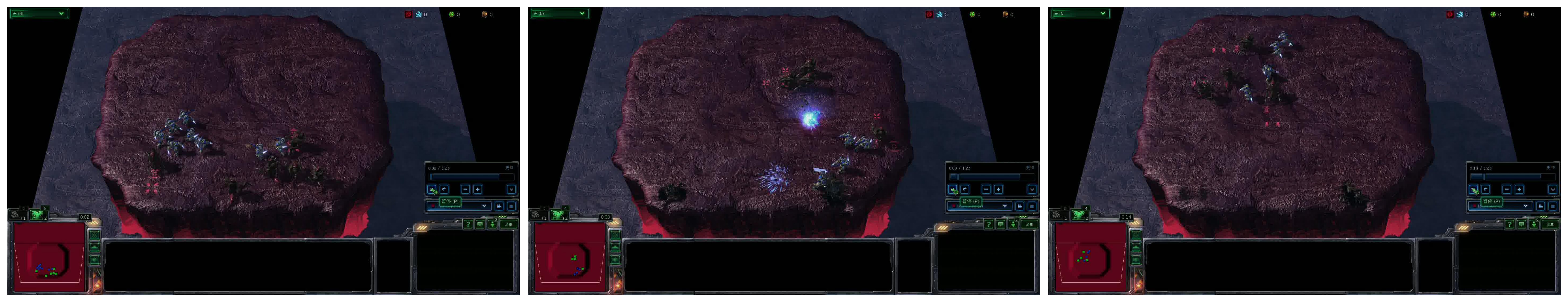}
    		\caption{Intermediate animations for 6h\_vs\_8z.}
    		\label{fig:6h_vs_8z_visualization}
    	\end{figure}
    
    \subsection{Extra Experimental Results on W-QMIX}
    \label{subsec:extra_weighted_qmix_variant_hyperparameters}
        We show the results of W-QMIX in Figure \ref{fig:wqmix_pp} with the annealing steps as 50k to support that the poor performance of W-QMIX on Predator-Prey is due to its poor robustness to the increased explorations.
        \begin{figure*}[ht!]
            \centering
            \begin{subfigure}[b]{0.55\linewidth}
                \centering
                \includegraphics[width=\textwidth]{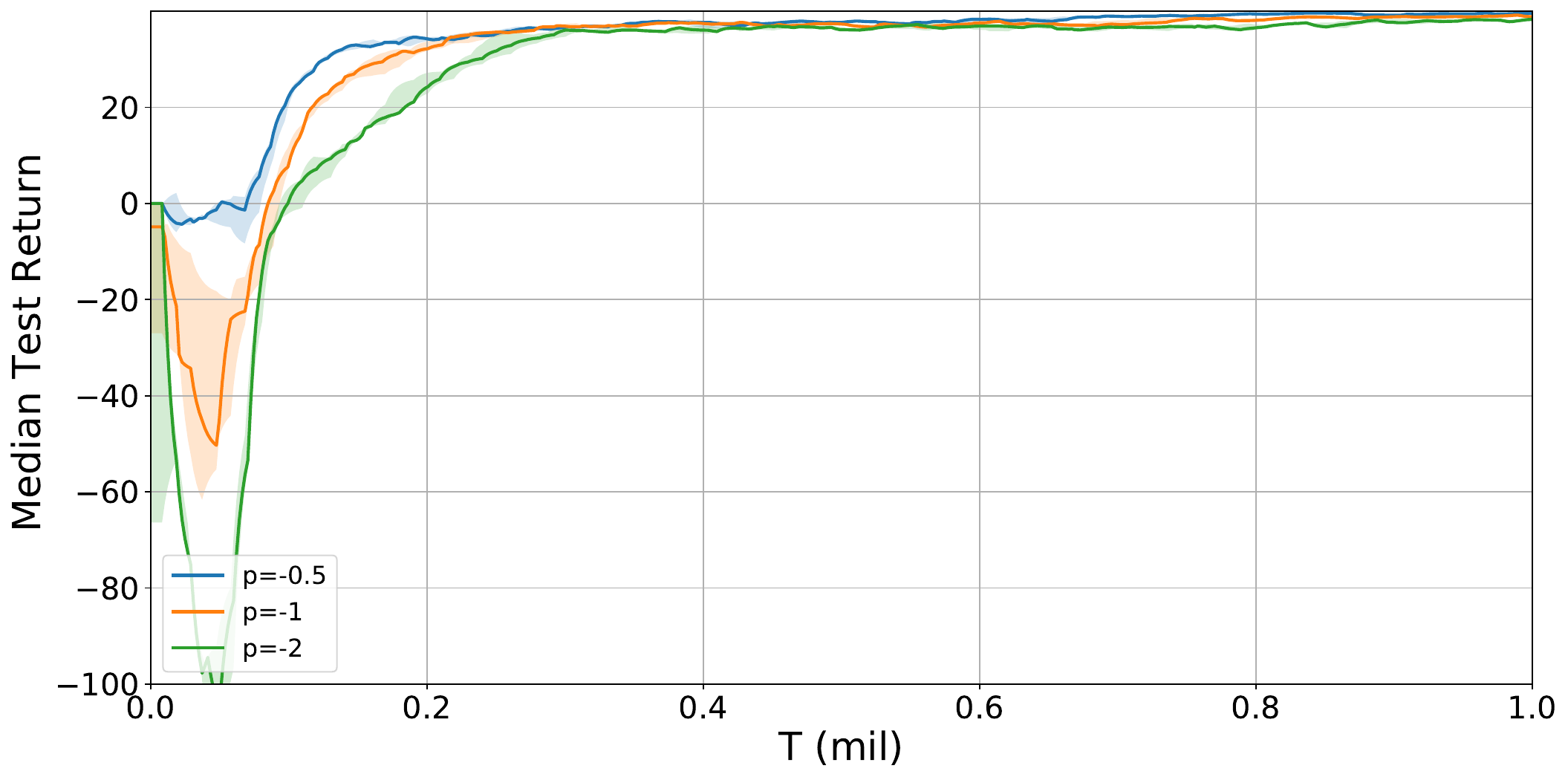}
                \caption{CW-QMIX.}
            \label{fig:wqmix_pp_b}
            \end{subfigure}
            ~
            \begin{subfigure}[b]{0.55\linewidth}
                \centering
                \includegraphics[width=\textwidth]{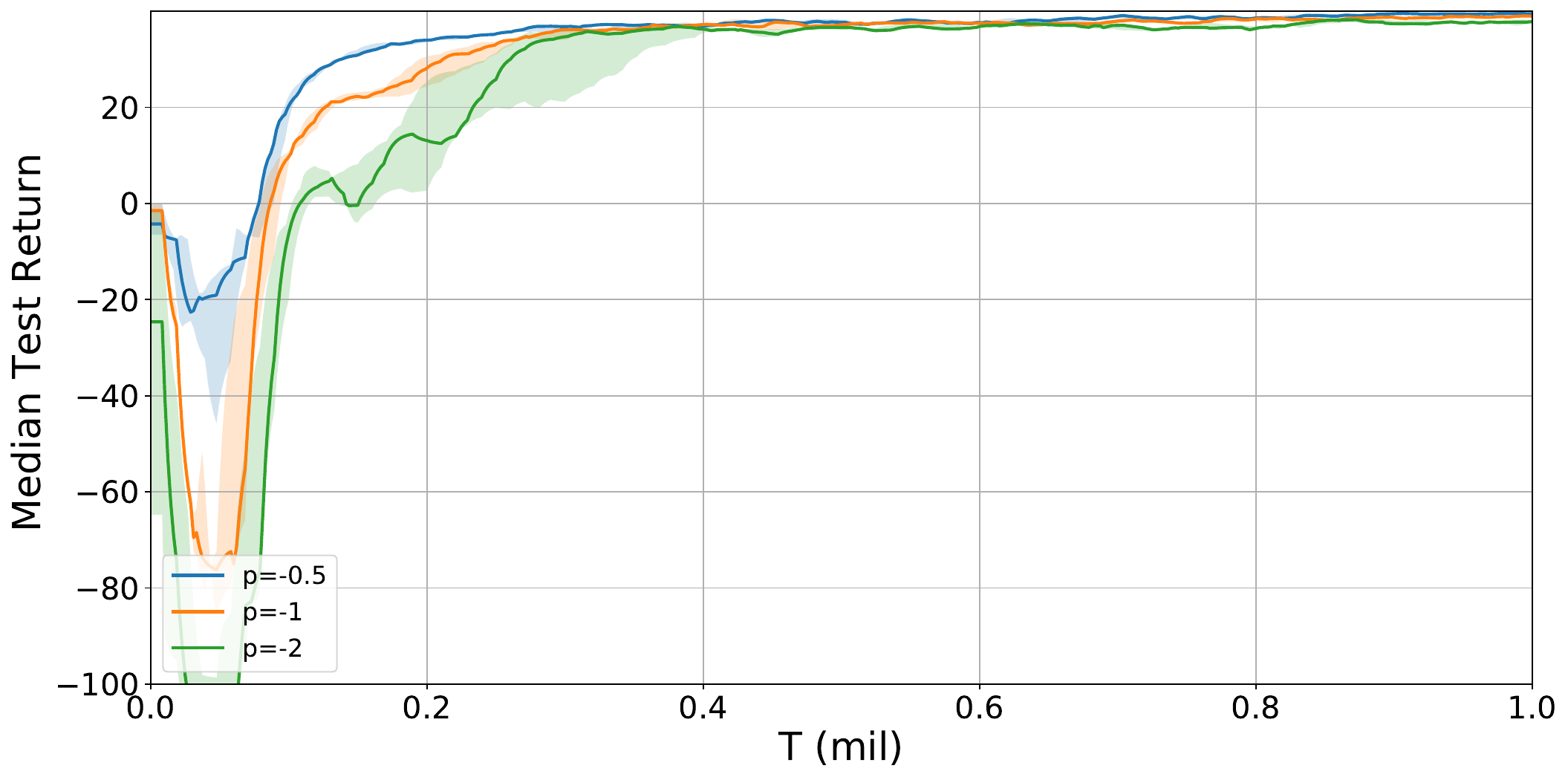}
                \caption{OW-QMIX.}
            \label{fig:wqmix_pp_c}
            \end{subfigure}
            \caption{Median test return for W-QMIX (including OW-QMIX and CW-QMIX) on Predator-Prey.}
        \label{fig:wqmix_pp}
        \end{figure*}
        
        To show the significance of tuning $\mathbf{\alpha}$ for W-QMIX, we also run W-QMIX with $\alpha=0.1$ in addition to the best $\alpha$ reported in \cite{rashid2020weighted}. We can observe from Figure \ref{fig:easy_hard_smac_wqmix} that the performance of W-QMIX is not comparatively identical under each choice of $\alpha$. As a result, W-QMIX suffers from the separate tuning of $\alpha$ for each scenario. Unfortunately, \cite{rashid2020weighted} did not provide an empirical law for selecting $\alpha$, while SHAQ enjoys an empirical law to select $\hat{\alpha}_{i}(\mathbf{s}, a_{i})$ as Figure \ref{fig:manual_approximate_alpha} shows.
        \begin{figure*}[ht!]
            \centering
            \begin{subfigure}[b]{0.48\textwidth}
                \centering
                \includegraphics[width=\textwidth]{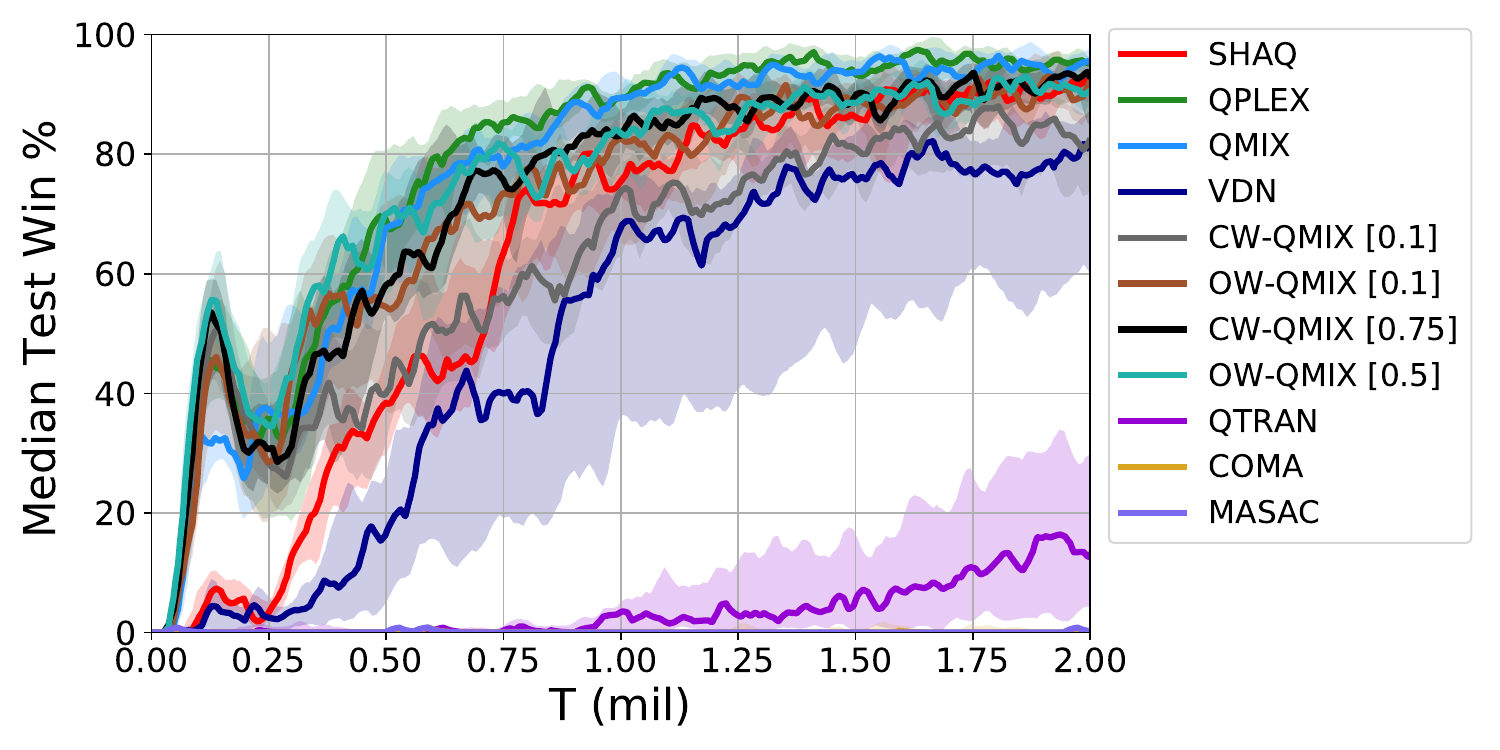}
                \caption{3s5z.}
            \end{subfigure}
            ~
            \begin{subfigure}[b]{0.48\textwidth}
                \centering
                \includegraphics[width=\textwidth]{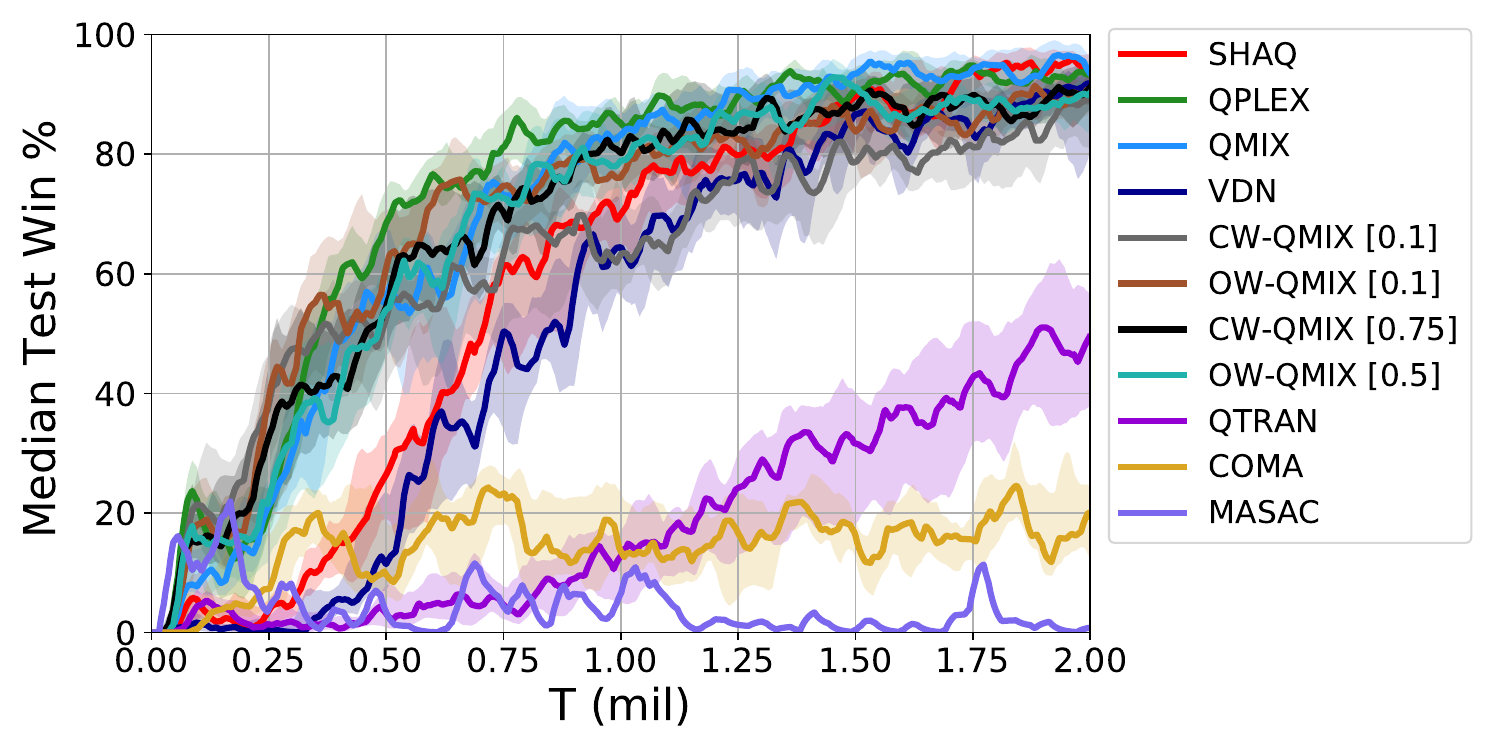}
                \caption{1c3s5z.}
            \end{subfigure}
            ~
            \begin{subfigure}[b]{0.48\textwidth}
                \centering
                \includegraphics[width=\textwidth]{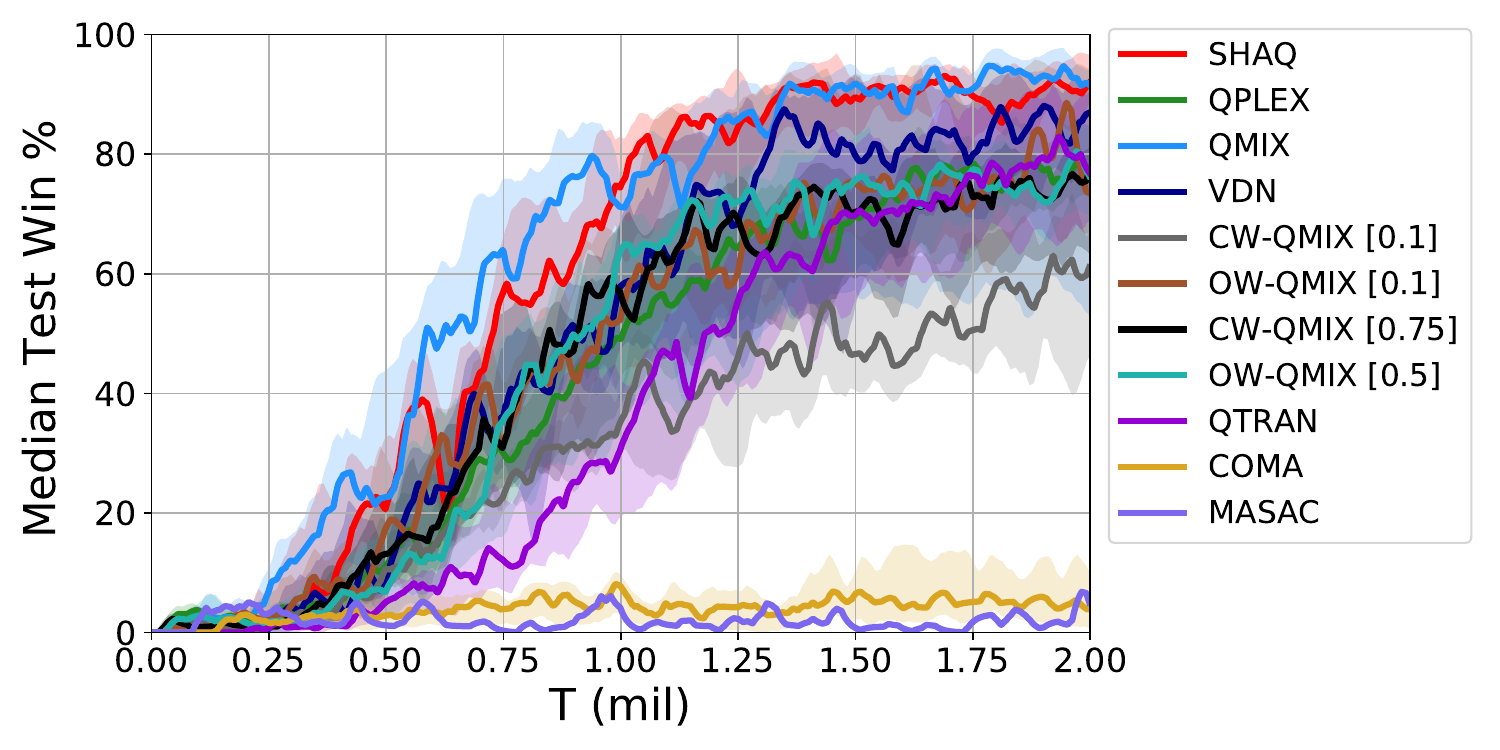}
                \caption{10m\_vs\_11m.}
            \end{subfigure}
            ~
            \begin{subfigure}[b]{0.48\textwidth}
                \centering
                \includegraphics[width=\textwidth]{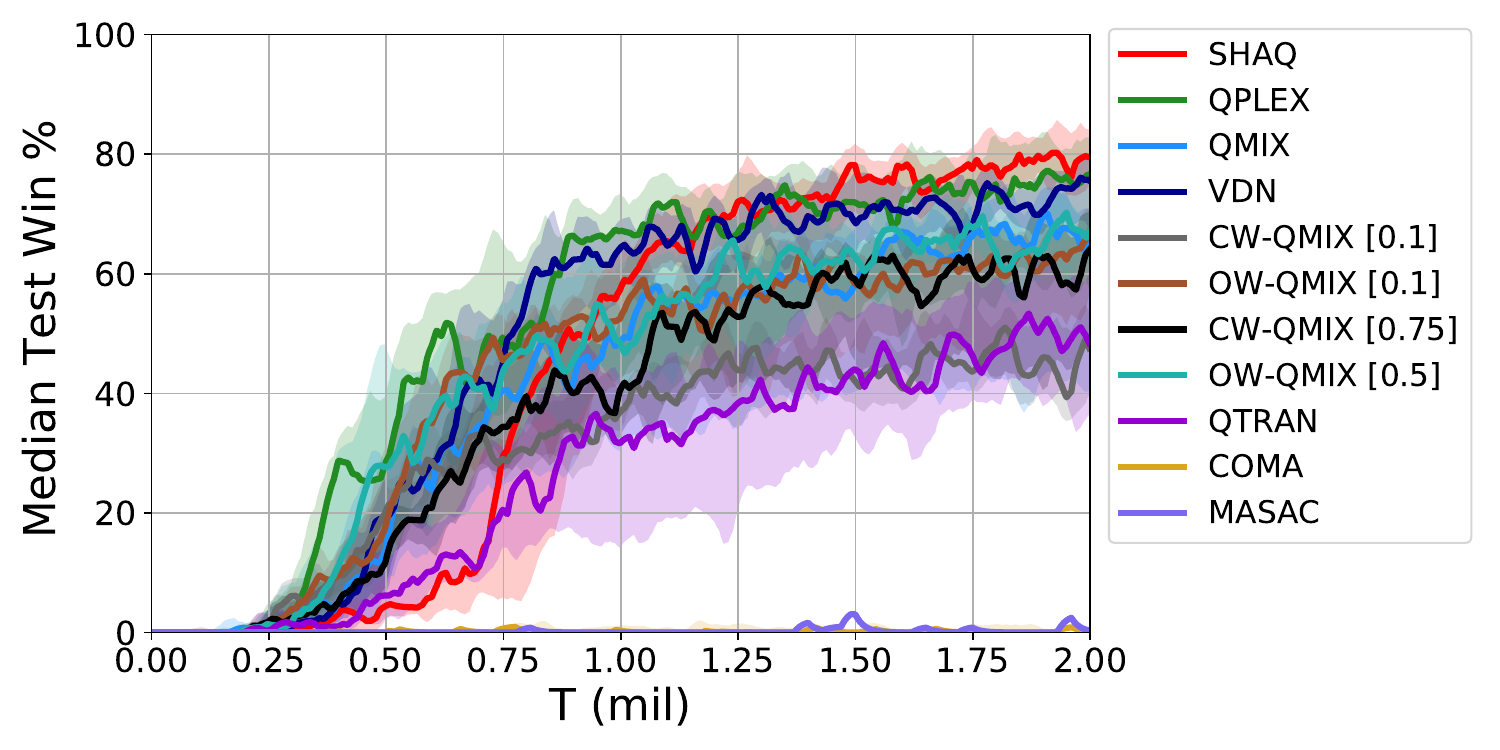}
                \caption{5m\_vs\_6m.}
            \end{subfigure}
            ~
            \begin{subfigure}[b]{0.48\textwidth}
                \centering
                \includegraphics[width=\textwidth]{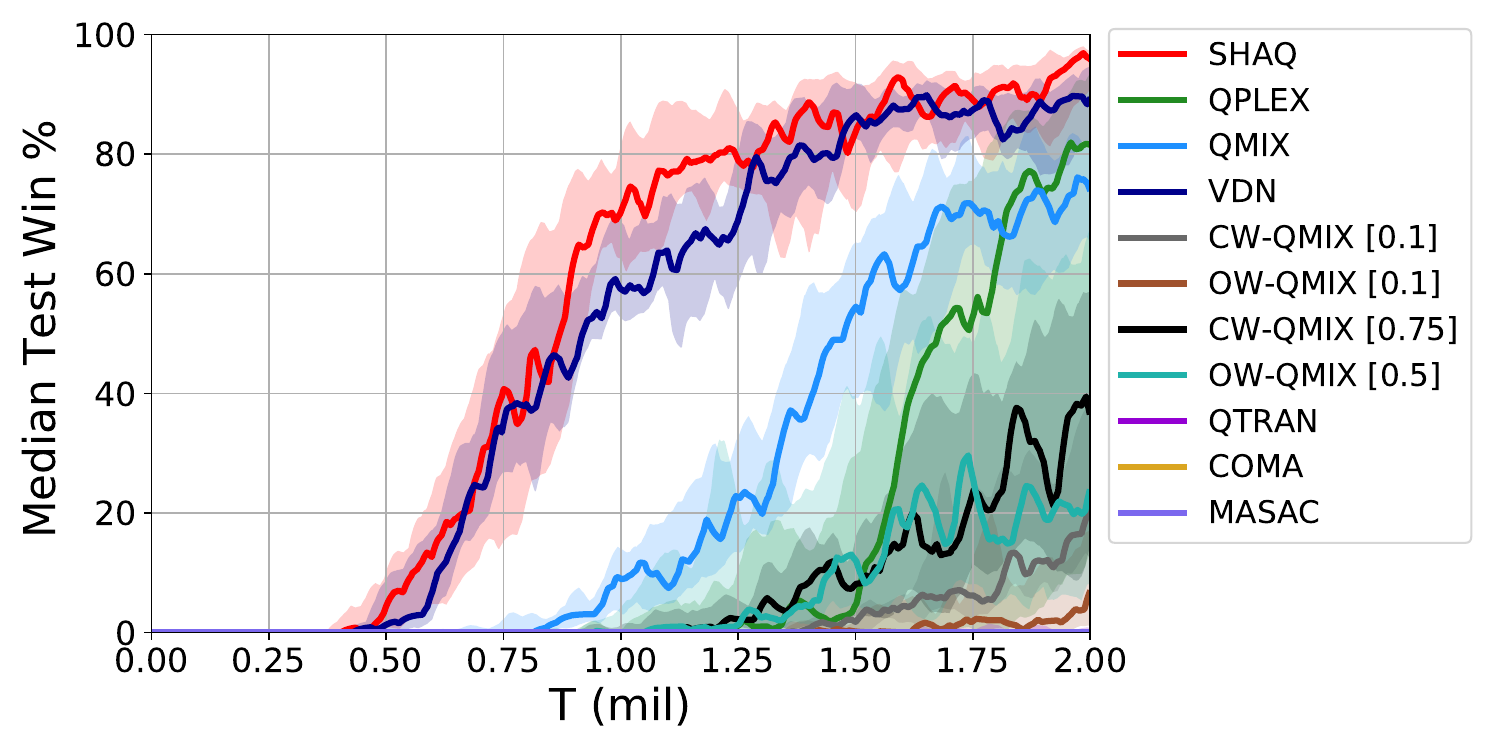}
                \caption{3s\_vs\_5z.}
            \end{subfigure}
            ~
            \begin{subfigure}[b]{0.48\textwidth}
                \centering
                \includegraphics[width=\textwidth]{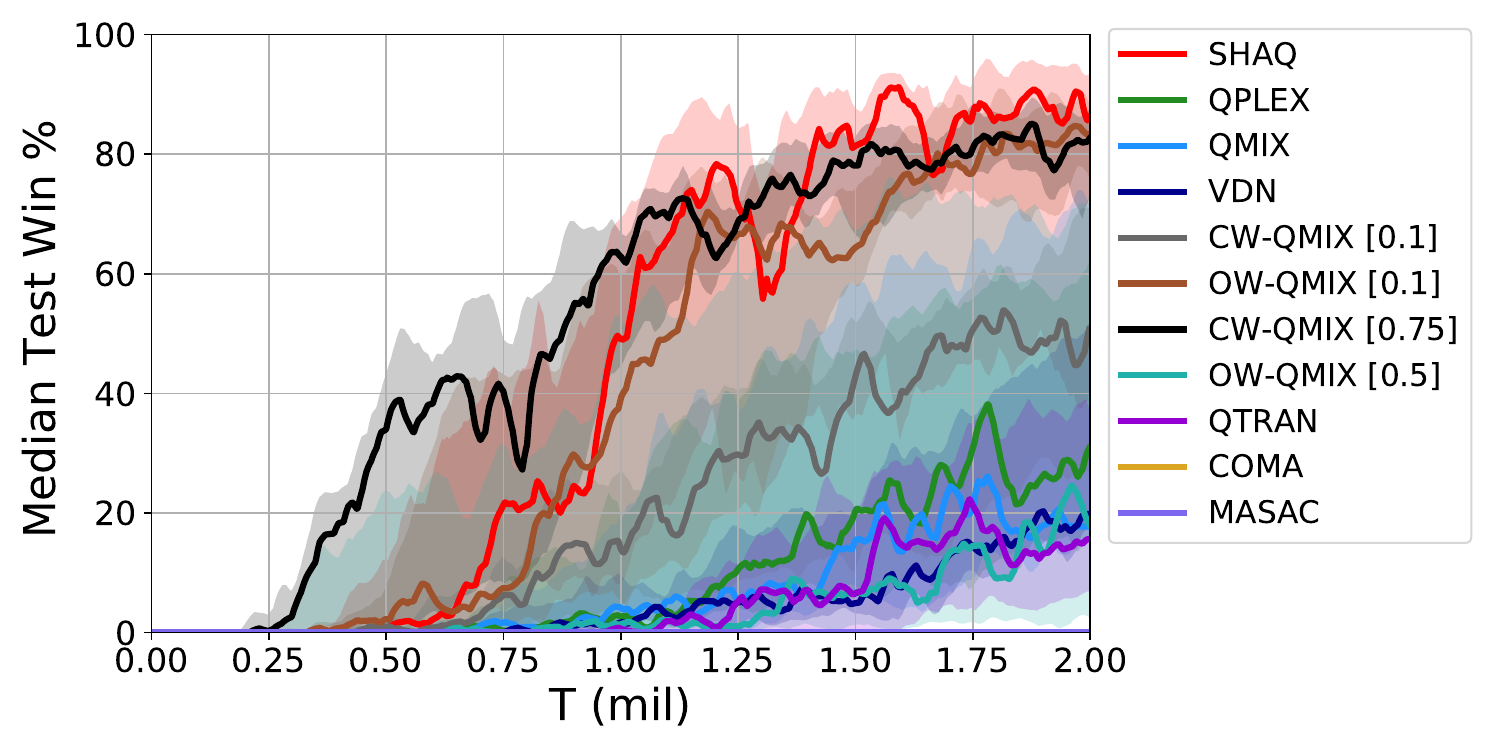}
                \caption{2c\_vs\_64zg.}
            \end{subfigure}
            \caption{Median test win \% for easy (1st row) and hard (2nd row) maps of SMAC for W-QMIX with different $\alpha$.}
            \label{fig:easy_hard_smac_wqmix}
        \end{figure*}

\section{Simulation Details of Active Voltage Control}
\label{sec:experimental_settings_of_distributed_active_voltage_control}
    \subsection{COMA with Continuous Actions}
    \label{subsec:continuous_coma}
        COMA \cite{foerster2018counterfactual} is an MARL algorithm with credit assignment via the mechanism of counterfactual regret, however, it can only serve for the discrete action space. In this thesis, to enable COMA to be eligible for the continuous action space, we conduct some tiny adjustments on the construction of Q-value for each agent. The original version of calculating each agent's Q-value assignment with respect to the discrete action is shown as follows:
        \begin{equation}
            Q_{i}(\mathbf{s}, \mathbf{a}) = Q(\mathbf{s}, \mathbf{a}) - \sum_{a_{i}' \in \mathcal{A}_{i}} \pi_{i}(a_{i}' | \tau_{i}) Q(\mathbf{s}, \mathbf{a}_{-i}, a_{i}'),
        \label{eq:coma_discrete}
        \end{equation}
        
        where $\tau_{i}$ is a history of agent $i$; $\mathbf{a}_{-i} = \times_{j \neq i} a_{j}$. To fit the continuous action, we simply change Eq.~\ref{eq:coma_discrete} to the form such that
        \begin{equation}
            Q_{i}(\mathbf{s}, \mathbf{a}) = Q(\mathbf{s}, \mathbf{a}) - \int_{a_{i}' \in \mathcal{A}_{i}} Q(\mathbf{s}, \mathbf{a}_{-i}, a_{i}') \ d\pi_{i}(a_{i}' | \tau_{i}),
        \label{eq:continuous_coma}
        \end{equation}
        
        where $\pi_{i}(a_{i}' | \tau_{i})$ is a Gaussian distribution over $a_{i}'$. In practice, $\int_{a_{i}' \in \mathcal{A}_{i}} Q(\mathbf{s}, \mathbf{a}_{-i}, a_{i}') \ d\pi_{i}(a_{i}' | \tau_{i})$ is approximated via Monte Carlo sampling, so it can be rewritten as follows:
        \begin{equation}
            Q_{i}(\mathbf{s}, \mathbf{a}) = Q(\mathbf{s}, \mathbf{a}) - \frac{1}{M} \sum_{k=1}^{M} Q(\mathbf{s}, \mathbf{a}_{-i}, (a_{i}')_{k} ), \ \ \ (a_{i}')_{k} \sim \pi_{i}(a_{i}' | \tau_{i}).
        \label{eq:continuous_coma_approximate}
        \end{equation}
    
    \subsection{Algorithm Settings and Training Details}
    \label{subsec:algo_settings_and_training_details}
        Since IDDPG and MADDPG do not possess any extra hyperparameters (other than common settings), we only introduce the special hyperparameters of COMA, MATD3, SQDDPG, SMFPPO, IPPO, and MAPPO. All hyperparameters reported here are tuned by the grid search and the best ones are selected as the final choice.
        
        \paragraph{Common Settings.} All algorithms are trained with online learning (i.e., for the on-policy algorithm like COMA, SMFPPO, MAPPO and IPPO, the behaviour policies/values are updated every 60 timesteps; and for the off-policy algorithms like SQDDPG, IDDPG and MADDPG, the behaviour policies/values are updated every 60 timesteps, where all data used for training are collected online) and the target policy/critic networks are updated every interval that is twice as the update interval of behaviour policy/critic introduced above. Taking the lesson from \cite{Wang_2020,fujimoto2018addressing}, for all algorithms except for MAPPO, IPPO and SMFPPO, the critic networks are updated with 10 epochs while the policy networks are updated with 1 epoch. For MAPPO, IPPO and SMFPPO, 10 epochs of training are conducted for both policy network and critic network. All algorithms are trained under the normalised reward and the action bound enforcement trick \cite{haarnoja2018soft} (i.e. working better than the hard clipping in our trials). The target update learning rate is set to $0.1$ (for off-policy algorithms). The gradient is clipped with L1 norm and the clip bound is set to $1$. The batch size of training data is set to 32 and the replay buffer size for off-policy algorithms is set to $5,000$. 
        
        Agent ID is concatenated with the observation and the layer normalisation \cite{ba2016layer} is applied to the first layer after the observation input. The parameters are shared among agents in this simulation. As for the policy network, RNN with GRUs \cite{chung2014empirical} is applied as a filter to solve the partial observation problems. The critic network is constructed with pure MLPs. The general setting of the policy and critic networks are shown in Table \ref{tab:network_spec}. During training, a fixed standard deviation as $1.0$ is applied to conduct the exploration. For the policy loss with entropy, the entropy penalty is set to 1e-3. The parameter initialisation is implemented by sampling from the Gaussian distribution $\mathcal{N}(0, 0.1)$. RMSProp \cite{tieleman2012lecture} is used as the optimizer, with the learning rate of 1e-4 for updating both policies and critics. The activation function of hidden layers is hyperbolic tangent function (Tanh) for PPO based algorithms \cite{engstrom2020implementation}, while it is ReLU for the rest algorithms.
        
        \paragraph{COMA.} The sample size M of COMA for the continuous action proposed in this thesis is set to 10 in simulation. 
        
        \paragraph{MATD3.} The clip boundary c for clipping the exploration noise is set to 1 in simulation. 
        
        \paragraph{SQDDPG.} The sample size M of SQDDPG is set to 10 in simulation. 
        
        \paragraph{IPPO, MAPPO and SMFPPO.} We apply generalised advantage estimation (GAE) \cite{chulmanMLJA15} to evaluate the return with $\lambda=0.95$ for IPPO and MAPPO to reach the best performance. The return of SMFPPO is evaluated by the Shapley value mechanism with the sample size M of 10. The value loss coefficient is set to 1. The $\epsilon$ for clipping the objective function is set to 0.1 for 33-bus networks and 0.3 for the rest scenarios. Since the policies of PPO based methods are modelled as learnable Gaussian distributions, the exploration range of log std is critical to the performance. It is from -1.0 to 0.5 for all scenarios for IPPO and MAPPO, while from -1.0 to 0.5 for the 33-bus network; from 0.0 to 0.5 for the rest scenarios for SMFPPO. The parameter initialisation is implemented by orthogonal initialisation \cite{saxe2013exact}. All the tricks of PPO are from the suggestions in \cite{engstrom2020implementation}.
        \begin{table}[ht!]
        	\caption{General specifications of policy and value networks.}
        	\vskip 0.15in
        	\begin{center}
        		\begin{small}
        			\begin{sc}
        			  \scalebox{1.0}{
        				\begin{tabular}{ll}
        					\toprule
        					\textbf{Network} & \textbf{Structure} \\
        					\midrule
        					Policy & GRU(state\_dim, 64) $\rightarrow$ Layernorm() $\rightarrow$ ReLU/Tanh() \\
        					& $\rightarrow$ Linear(64, 64) $\rightarrow$ Linear(64, action\_dim) \\
        				    Critic & Linear(input\_dim, 64) $\rightarrow$ Layernorm() $\rightarrow$ ReLU/Tanh() \\ 
        				    & $\rightarrow$ Linear(state\_dim, 64) $\rightarrow$ ReLU/Tanh() $\rightarrow$ Linear(64, output\_dim) \\
        					\bottomrule 
        				\end{tabular}
        				}
        			\end{sc}
        		\end{small}
        	\end{center}
        	\label{tab:network_spec}
        	\vskip -0.1in
        \end{table}
    
    \subsection{Process of Simulation}
    \label{subsec:process_of_simulations}
        We plot the flow chart in Figure \ref{fig:env_diagram} to illustrate the process of the simulation for the active voltage control on power distribution networks.\footnote{The open-source code for the simulator is placed on \url{https://github.com/Future-Power-Networks/MAPDN}.} At the beginning of each episode, a series of consecutive PV and load profiles for 480 timesteps (i.e. 1 day) is in the buffer. At each timestep, the relevant PV and load profile are extracted, combined with the voltage status computed by Pandapower \cite{thurner2018pandapower} (i.e., computing the power flow) to establish the next state. Additionally, the reward is also calculated according to the result computed by Pandapower. Before fed to agents, the received state will be split into a batch of observations as per the region where each agent is located. Each agent only receives a local observation and the global reward, then it makes next decision. The above procedure is repeated until the end of an episode.
        \begin{figure}[ht!]
            \centering
            \includegraphics[width=0.65\textwidth]{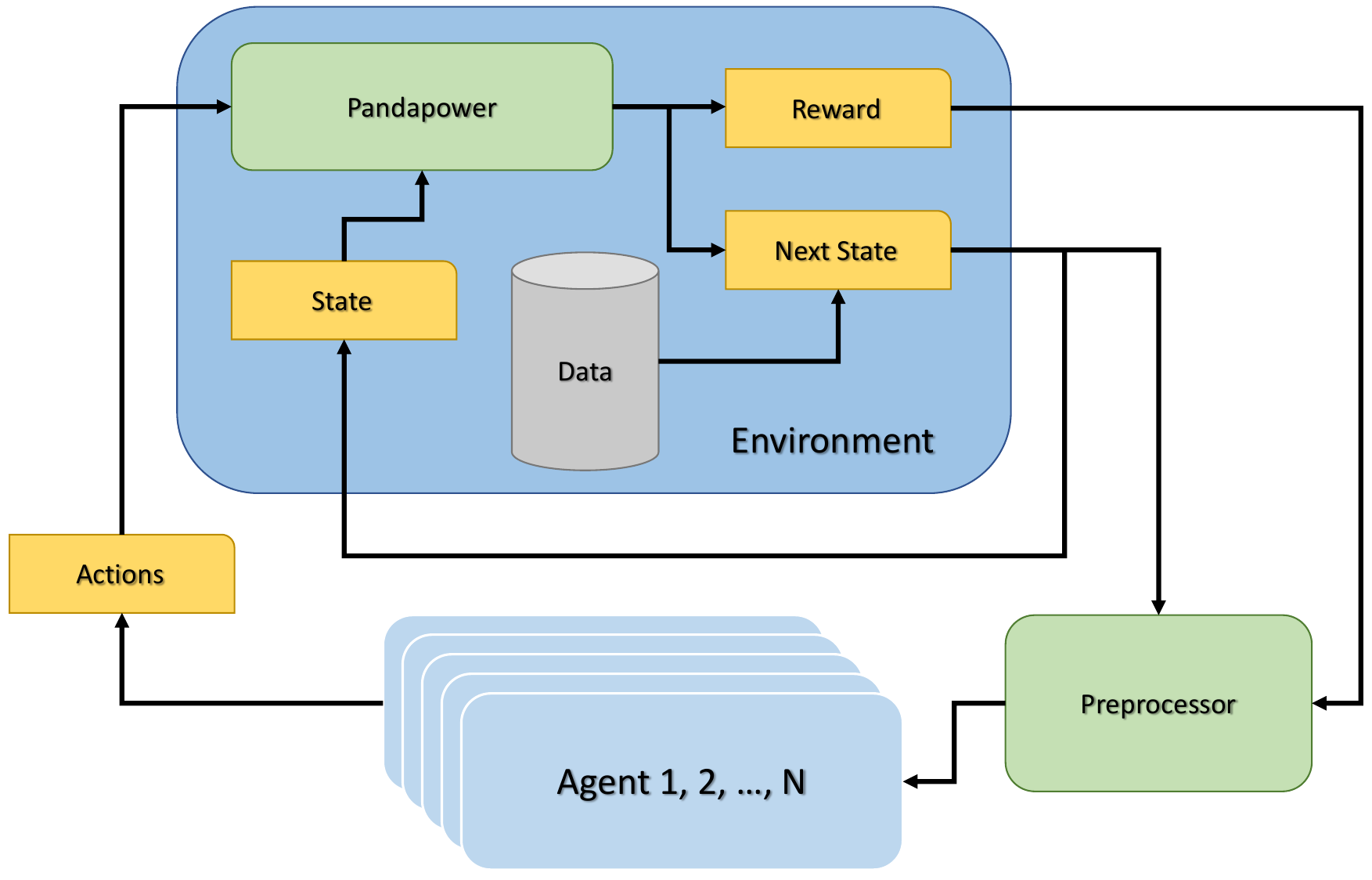}
            \caption{Flow chart of the implementation of environment for active voltage control on power distribution networks.}
        \label{fig:env_diagram}
        \end{figure}
        
    \subsection{Network Topology}
    \label{subsec:network_topology}
        We present 3 MV/LV distribution network models, each of which is composed of distinct topology and parameters, a load profile (including both active and reactive powers) describing different user behaviours, and a PV profile describing the active power generation from PVs. Although it is possible to partition the control regions by the voltage sensitivity of each bus \cite{chai2018network}, they are commonly determined by different distribution network owners in practice. Consequently, the control regions in this thesis are partitioned by the shortest path between the coupling bus and the terminal bus. Besides, each region consists of 1-4 PVs depending on the zonal sizes. A summary of the 3 networks is recorded in Table \ref{tab:network_summary} and the specific topologies are demonstrated in Figure \ref{fig:network_topology}.
        \begin{figure}[ht!]
            \centering
            \begin{subfigure}[h]{0.32\textwidth}
                  \centering
                  \includegraphics[width=\textwidth]{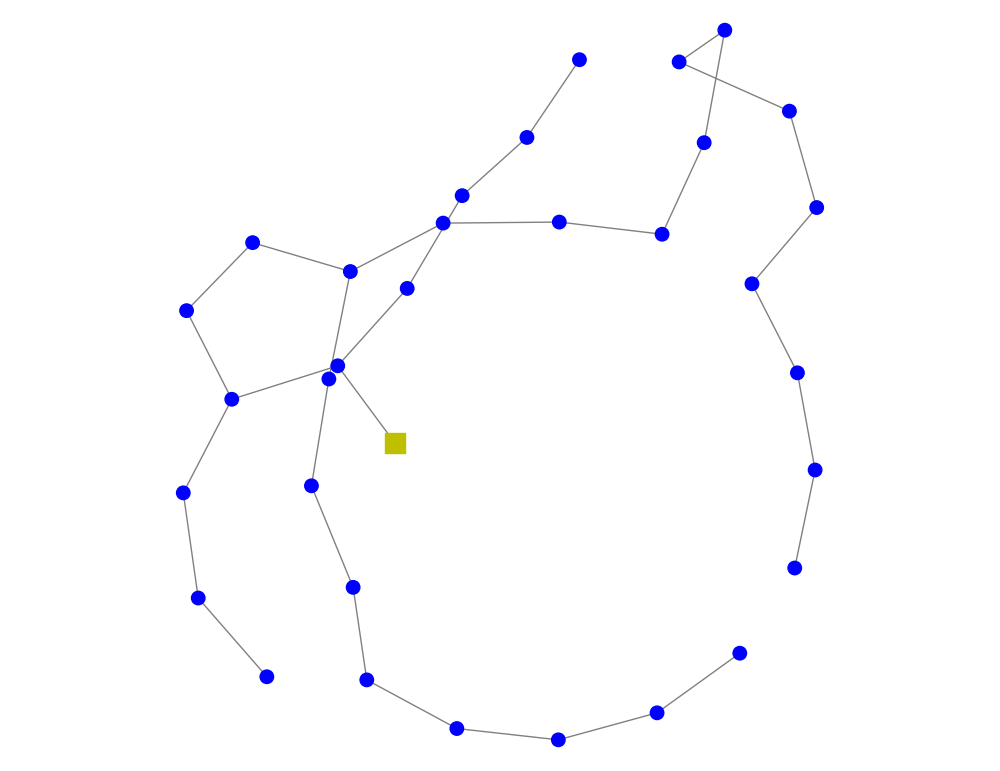}
                  \caption{33-bus network.}
            \end{subfigure}
            \begin{subfigure}[h]{0.32\textwidth}
                  \centering
                  \includegraphics[width=\textwidth]{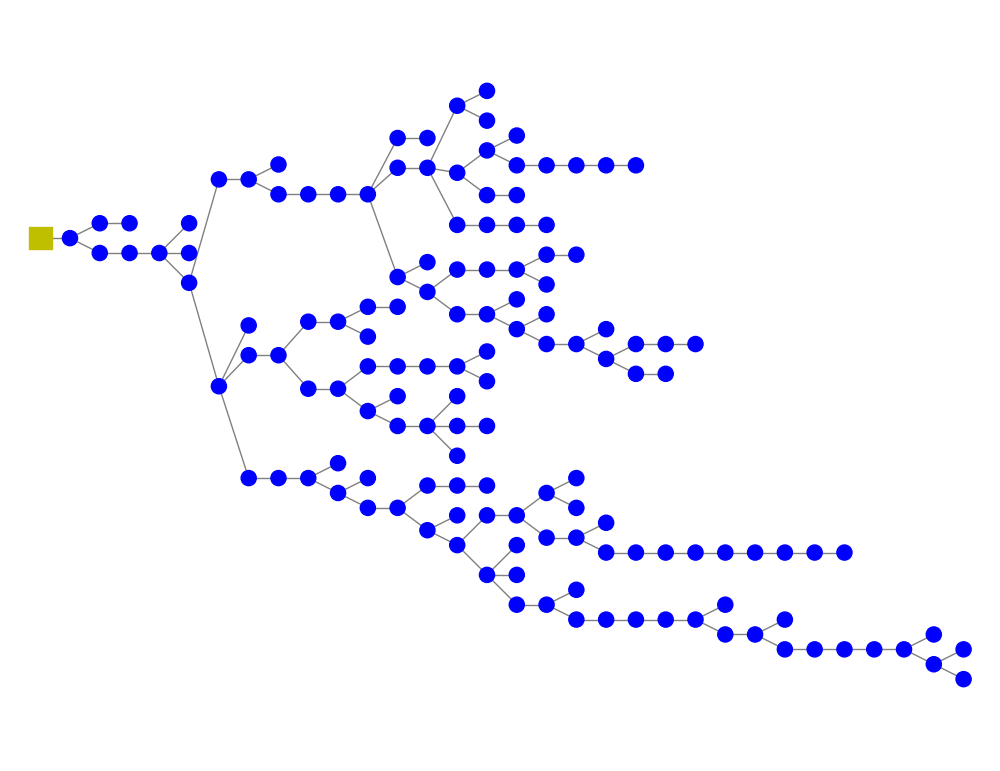}
                  \caption{141-bus network.}
            \end{subfigure}
            \begin{subfigure}[h]{0.32\textwidth}
                  \centering
                  \includegraphics[width=\textwidth]{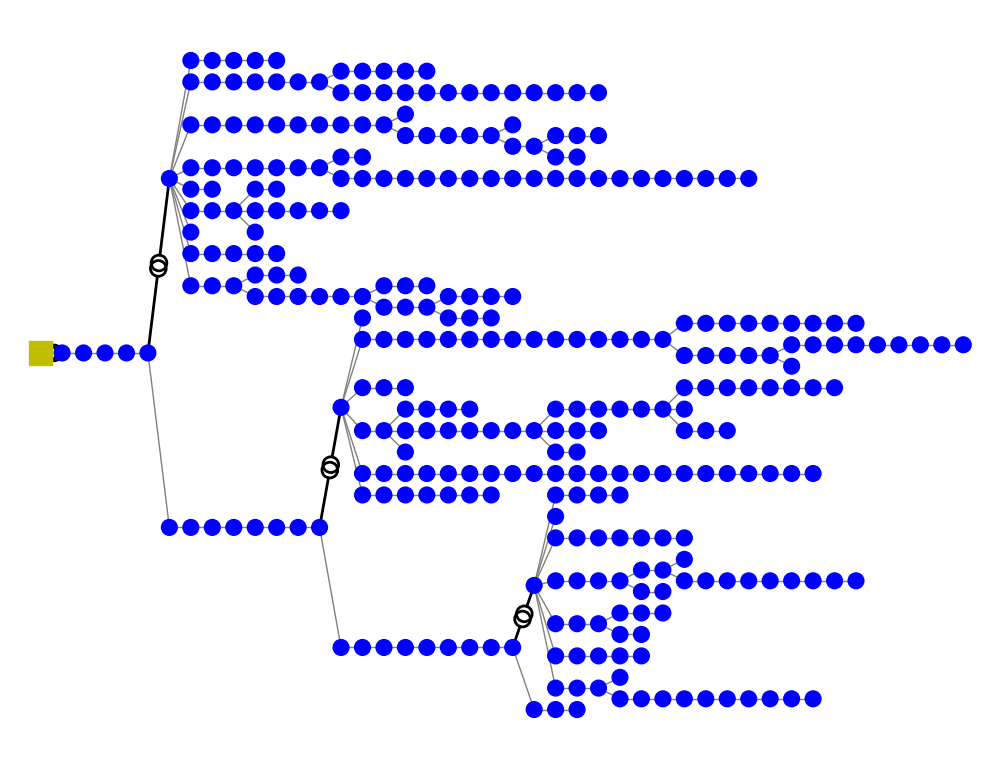}
                  \caption{322-bus network.}
            \end{subfigure}
            \caption{Topologies of power networks. The yellow square is the reference bus (a.k.a. the slack bus) and each blue circle is a non-reference bus. Transformers are highlighted as double-circles.}
            \label{fig:network_topology}
        \end{figure}
        
        \begin{table}[ht!]
            \centering
            \caption{Network specifications of all scenarios.}
            \begin{tabular}{ccccccc}
            \toprule
              & Rated Voltage & No. Loads & No. Regions & No. PVs & $p^{\scriptscriptstyle L}_{\max}$ & $p^{\scriptscriptstyle PV}_{\max}$ \\
            \midrule
                33-bus & 12.66 kV & 32 & 4 & 6 & 3.5 MW & 8.75 MW \\ 
            \midrule
                141-bus  & 12.5 kV & 84 & 9 & 22 & 20 MW & 80 MW  \\
            \midrule
                322-bus & 110-20-0.4 kV & 337 & 22 & 38 &1.5 MW & 3.75 MW \\
            \bottomrule
            \end{tabular}
            \label{tab:network_summary}
        \end{table}
        \paragraph{The 33-bus Network.} The 33-bus network is modified from the case33bw in MATPOWER \cite{zimmerman2010matpower} and PandaPower \cite{thurner2018pandapower}. To promise the tree structure, similar to \cite{liu2021online}, we drop lines 33-37 to avoid any loops. 6 PVs are added unevenly on bus 13 and 18 (zone 1), bus 22 (zone 2), bus 25 (zone 3), bus 29 and 33 (zone 4). The PV-load ratio is $PR=2.5$. 
        
        \paragraph{The 141-bus Network.} The 141-bus network is modified from the case141 in MATPOWER \cite{zimmerman2010matpower} as well. A similar procedure is followed as done for the 33-bus network.
        
        \paragraph{The 322-bus Network.} The proposed 322-bus network consists of an external 110-kV bus, a long medium-voltage (20 kV) line (25 buses in total) and 3 LV feeders (0.4 kV) representing rural (128 buses), semi-urban (110 buses), and urban (58 buses) areas defined by SimBench \cite{meinecke20simbench}. Areas with different voltage levels are connected though standard transformers defined in PandaPower \cite{thurner2018pandapower}. The rural area has the lowest power consumption level and some buses are with no loads, while more than one load are allowed to locate on a bus in the urban area, so the total number of loads is higher than the number of buses in the 322-bus network. The users can also generate their own synthetic networks by following out procedure. To simplify the setting, we aggregate the multiple loads at each bus into one.
        
    \subsection{Data Description}
    \label{subsec:data_descriptions}
        \paragraph{Load Profile.} The load profile of each network is modified based on the real-time Portuguese electricity consumption accounting for 232 consumers of 3 years.\footnote{\url{https://archive.ics.uci.edu/ml/datasets/ElectricityLoadDiagrams20112014}.} The original dataset contains 370 residential and industrial clients electricity usage from 2011 to 2014 in 15-min resolution. As some of the data record does not start at the beginning, we collect the data from 2012-01-01 00:15:00 and delete the locations that contain more than 20 missing data. The remaining missing data (mostly because of the winter time to daylight saving time switch) is interpolated linearly. The load data is then interpolated in 3-min resolution which is consistent with the real-time control period in the grid. The final data size is $526080\times232$ accounting for load profiles for 232 consumers of 1096 days (three years). We then remove the outliers that are outside $7\sigma$ against the mean value. For the 33-bus and the 141-bus networks, the 232 load profiles are randomly assigned to each bus. For the 322-bus network, repeated load profiles are allowed. In practice, the Gaussian noises are added to the load active power and the load reactive power.
        
        \paragraph{PV Profile.} Ten cities/regions/provinces in Belgium, Netherlands, and Luxembourg are considered to represent the distinct zonal solar radiation levels, including Antwerp, Brussels, Flemish-Brabant (a province of Belgium), Hainaut (a province of Belgium), Liege, Limburg (a province of Netherland), Luxembourg, Namur, Walloon-Brabant (a province of Belgium), and West-Flanders (a province of Belgium). The PV data is collected from Elia group,\footnote{\url{https://www.elia.be/en/grid-data/power-generation/solar-pv-power-generation-data}.} a Belgium’s power network operator. The PV data is also interpolated in 3-min resolution resulting in $526080\times10$ data in total. For the 33-bus (with 4 regions) and the 141-bus (with 9 regions) networks, the PV profiles are randomly assigned to each region. For the 322-bus (with 22 regions) network, different regions can have the same PV profile. Note that the PVs in the same control region share the same PV profiles as they are geometrically contiguous. In real time, we also add the Gaussian noise to the PV active power. 
        
        We summarise the load and PV profiles of different scales in Figure \ref{fig:33_bus_total}-\ref{fig:power_factor} below. Figure \ref{fig:33_bus_total} illustrates the total PV active power generation and active load consumption in the 33-bus network. Figure \ref{fig:33_bus_single} illustrates four distinct PV buses in the 33-bus network in January and July. Note that bus 13 and bus 18 are in the same region, so they have the same PV profiles. Figure \ref{fig:141_bus_total} illustrates the total PV active power generation and active load consumption in the 141-bus network. Figure \ref{fig:141_bus_single} illustrates four distinct PV buses in the 141-bus network in January and July. Note that bus 36 and bus 111 are in the same region, so they have the same PV profiles. Figure \ref{fig:322_bus_total} illustrates the total PV active power generation and active load consumption in the 322-bus network. Figure \ref{fig:322_bus_single} illustrates four distinct PV buses in the 322-bus network in January and July. Figure \ref{fig:power_factor} illustrates the power factors (PFs) of the three systems under test. THe high power factors ($>0.9$) usually represents the residential consumers, while the low power factors ($<0.5$) represent the industrial consumers.
        \begin{figure}[ht!]
            \centering
            \begin{subfigure}[h]{0.45\textwidth}
                  \centering
                  \includegraphics[width=\textwidth]{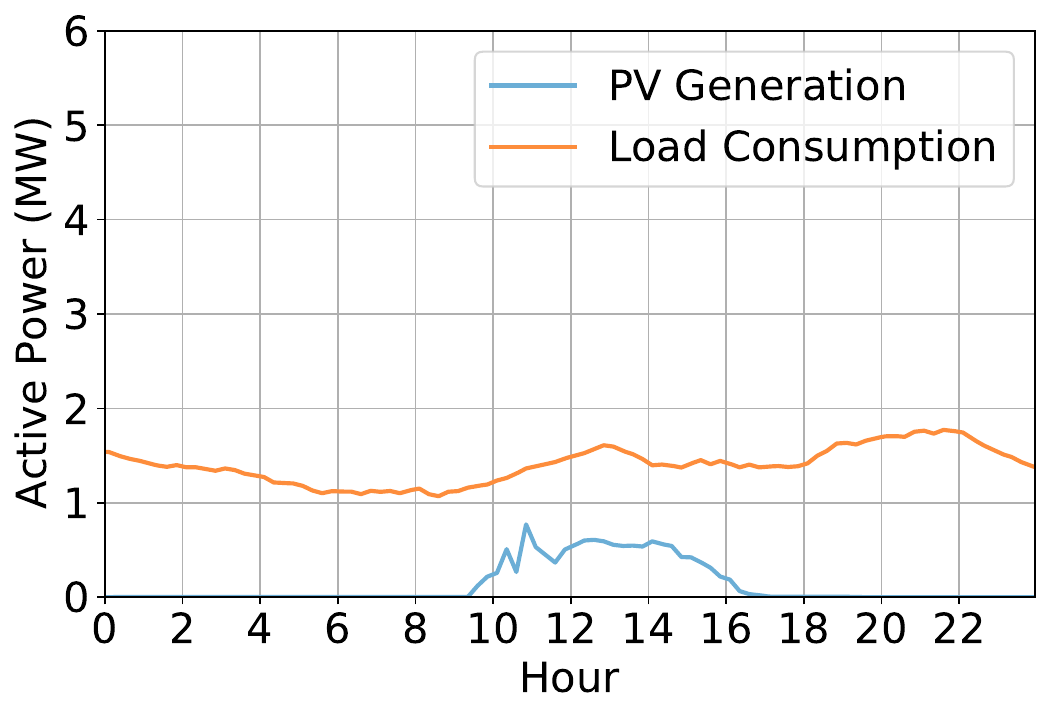}
                  \caption{A winter day.}
            \end{subfigure}
            \begin{subfigure}[h]{0.45\textwidth}
                  \centering
                  \includegraphics[width=\textwidth]{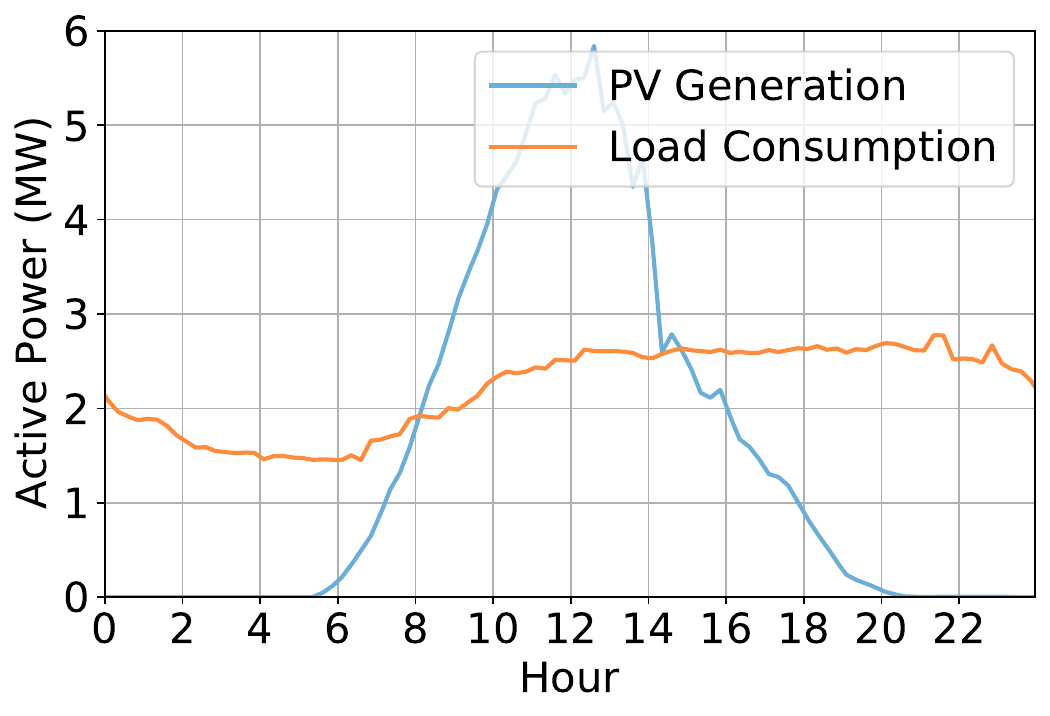}
                  \caption{A summer day.}
            \end{subfigure}
            \begin{subfigure}[h]{0.45\textwidth}
                  \centering
                  \includegraphics[width=\textwidth]{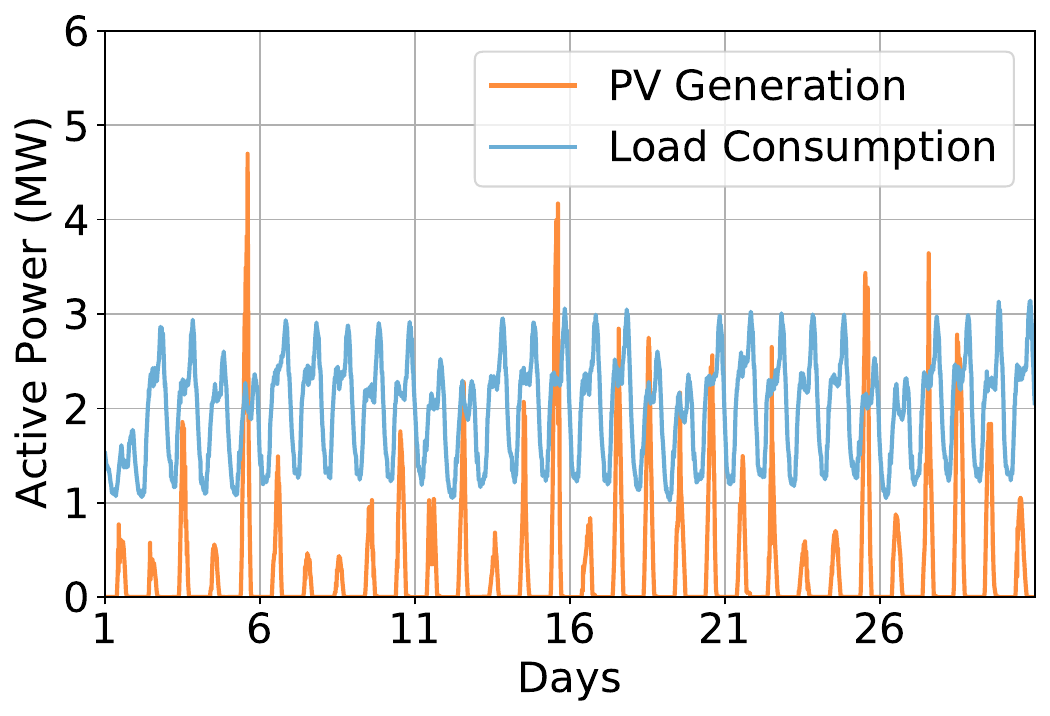}
                  \caption{A winter month (January).}
            \end{subfigure}
            \begin{subfigure}[h]{0.45\textwidth}
                  \centering
                  \includegraphics[width=\textwidth]{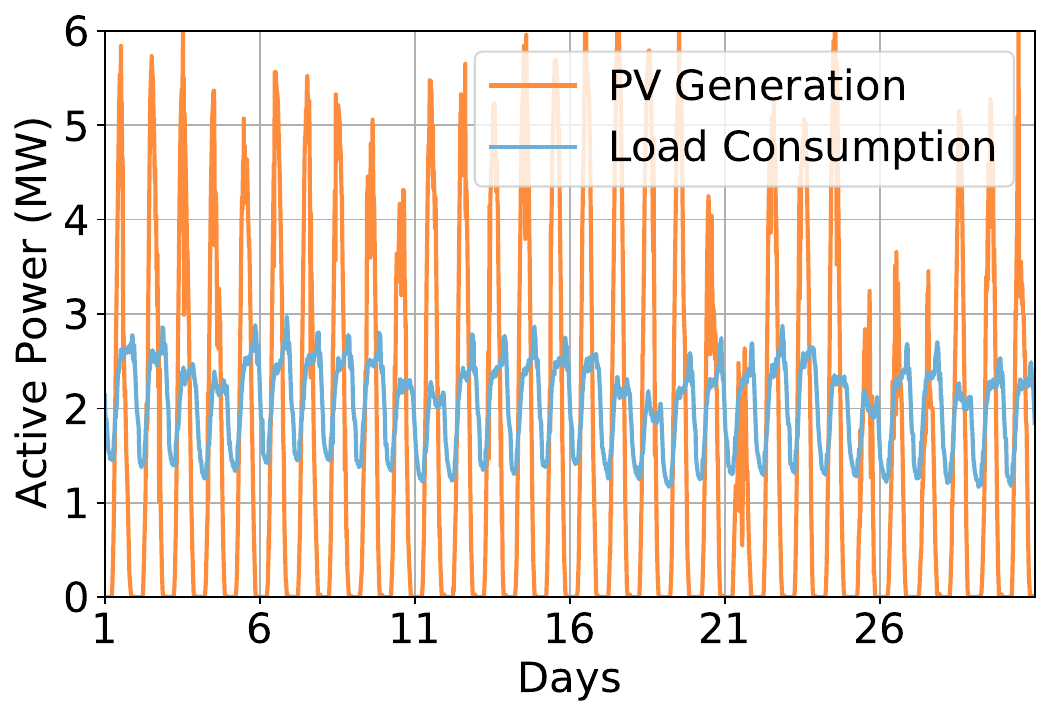}
                  \caption{A summer month (July).}
            \end{subfigure}
            \caption{Total power of the 33-bus network.}
            \label{fig:33_bus_total}
        \end{figure}
        
        \begin{figure}[ht!]
            \centering
            \begin{subfigure}[h]{0.45\textwidth}
                  \centering
                  \includegraphics[width=\textwidth]{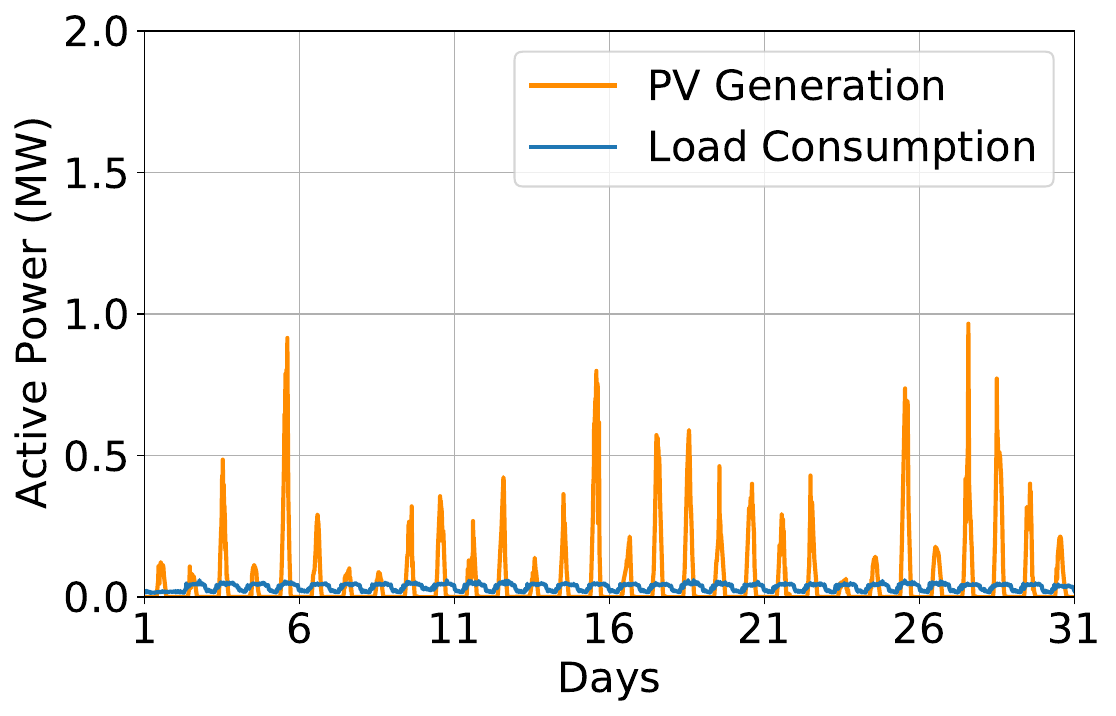}
                  \caption{A winter month of bus-13.}
            \end{subfigure}
            \begin{subfigure}[h]{0.45\textwidth}
                  \centering
                  \includegraphics[width=\textwidth]{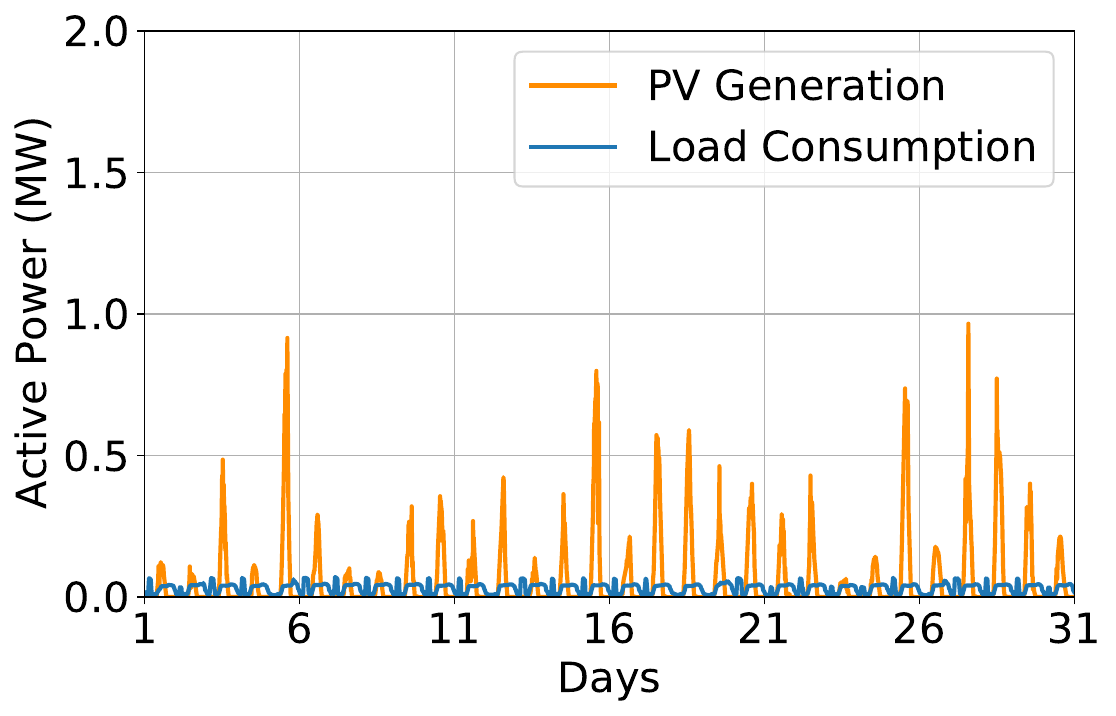}
                  \caption{A winter month of bus-18.}
            \end{subfigure}
            \begin{subfigure}[h]{0.45\textwidth}
                  \centering
                  \includegraphics[width=\textwidth]{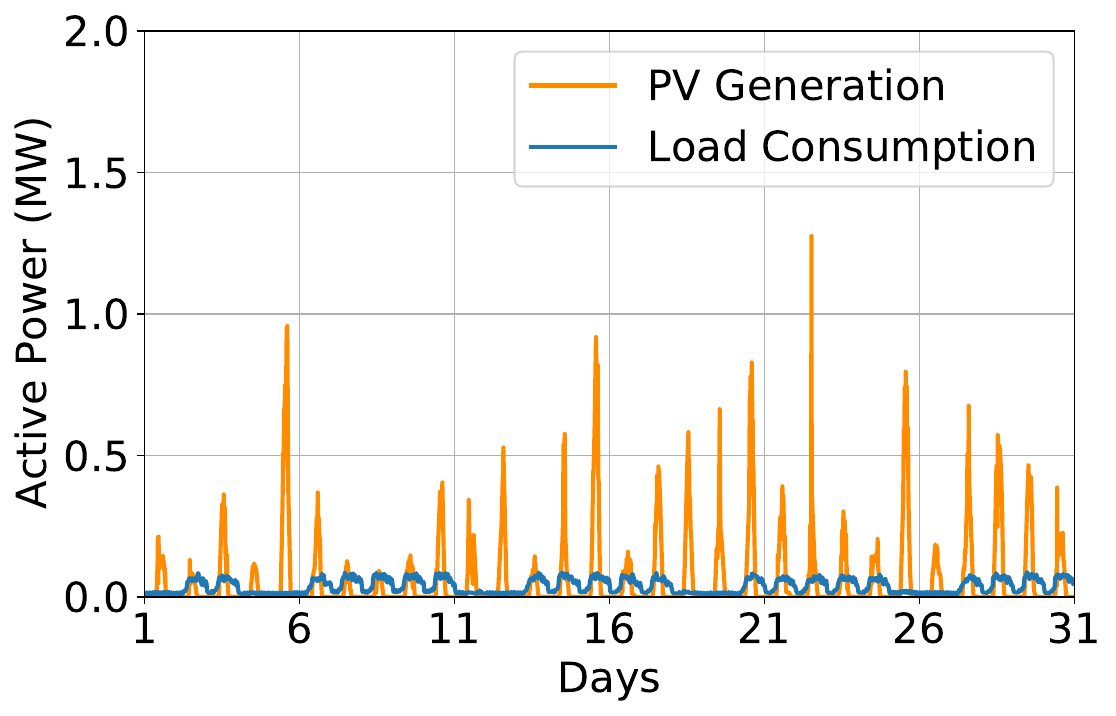}
                  \caption{A winter month of bus-22.}
            \end{subfigure}
            \begin{subfigure}[h]{0.45\textwidth}
                  \centering
                  \includegraphics[width=\textwidth]{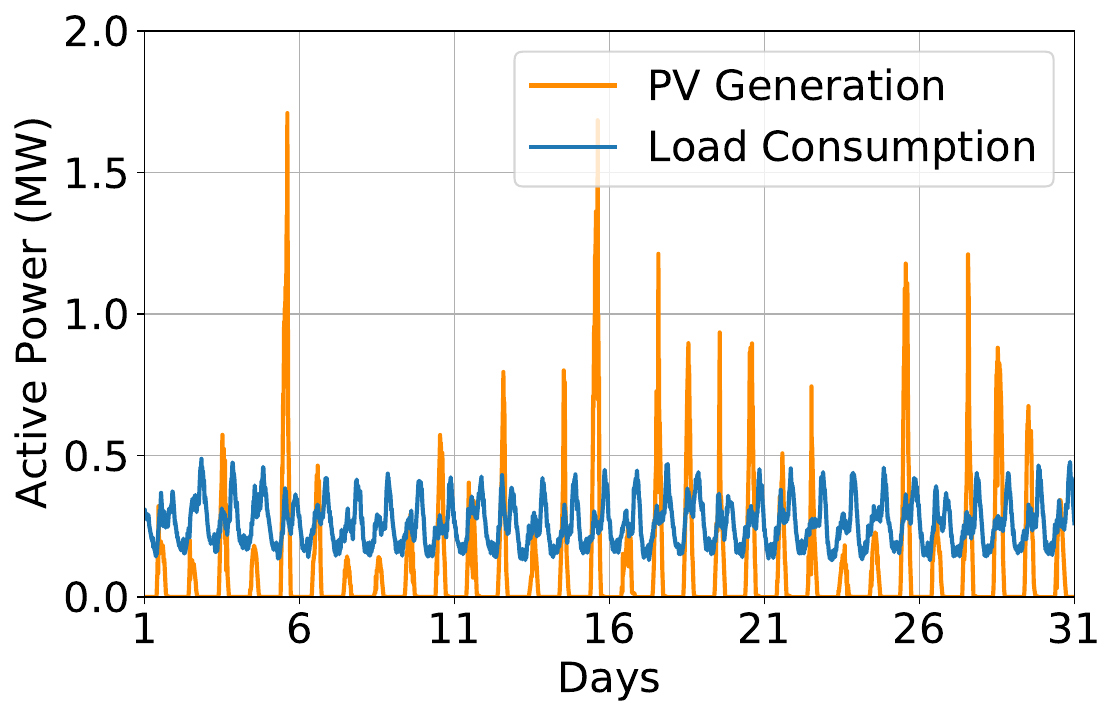}
                  \caption{A winter month of bus-25.}
            \end{subfigure}
            \begin{subfigure}[h]{0.45\textwidth}
                  \centering
                  \includegraphics[width=\textwidth]{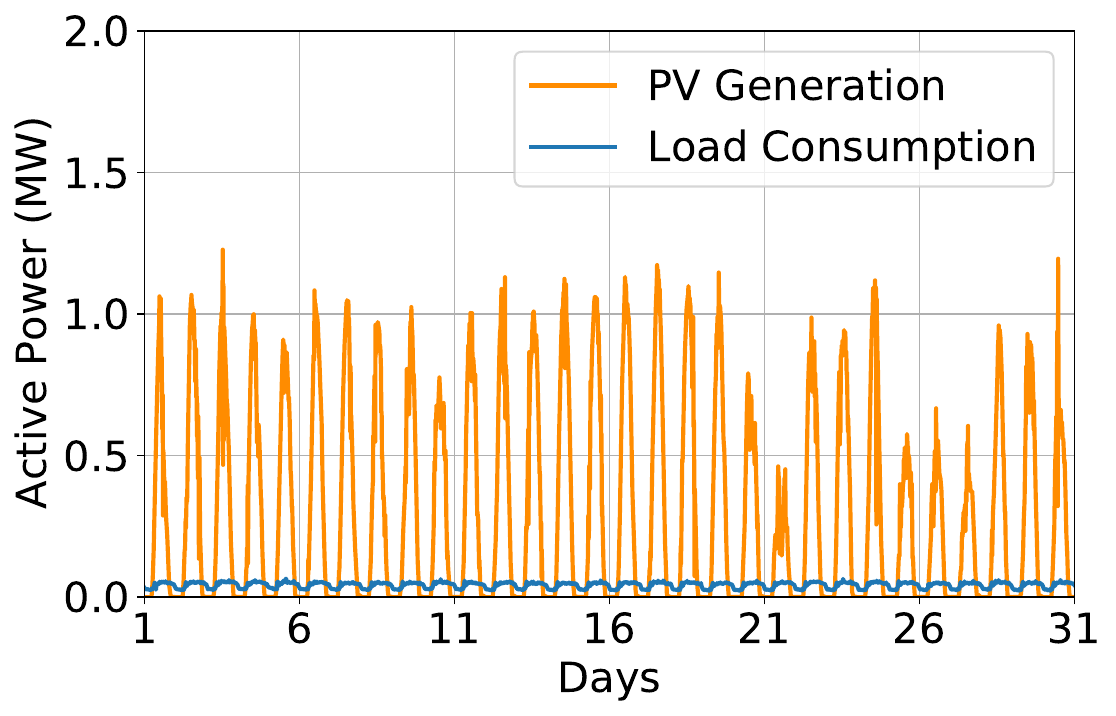}
                  \caption{A summer month of bus-13.}
            \end{subfigure}
            \begin{subfigure}[h]{0.45\textwidth}
                  \centering
                  \includegraphics[width=\textwidth]{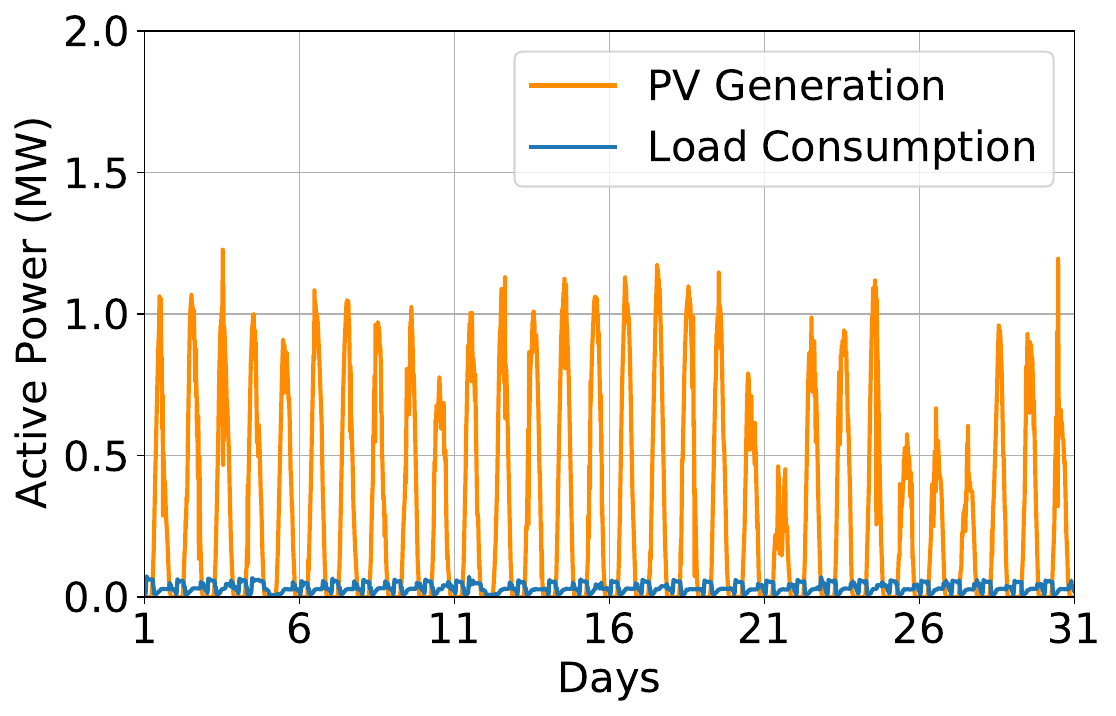}
                  \caption{A summer month of bus-18.}
            \end{subfigure}
            \begin{subfigure}[h]{0.45\textwidth}
                  \centering
                  \includegraphics[width=\textwidth]{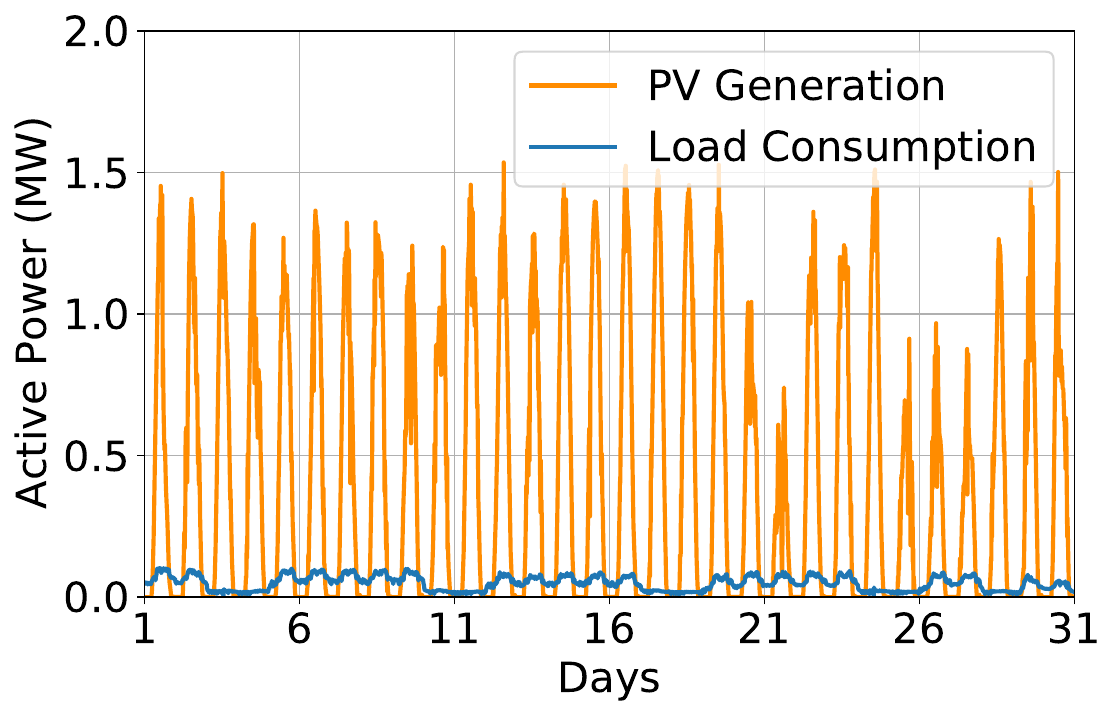}
                  \caption{A summer month of bus-22.}
            \end{subfigure}
            \begin{subfigure}[h]{0.45\textwidth}
                  \centering
                  \includegraphics[width=\textwidth]{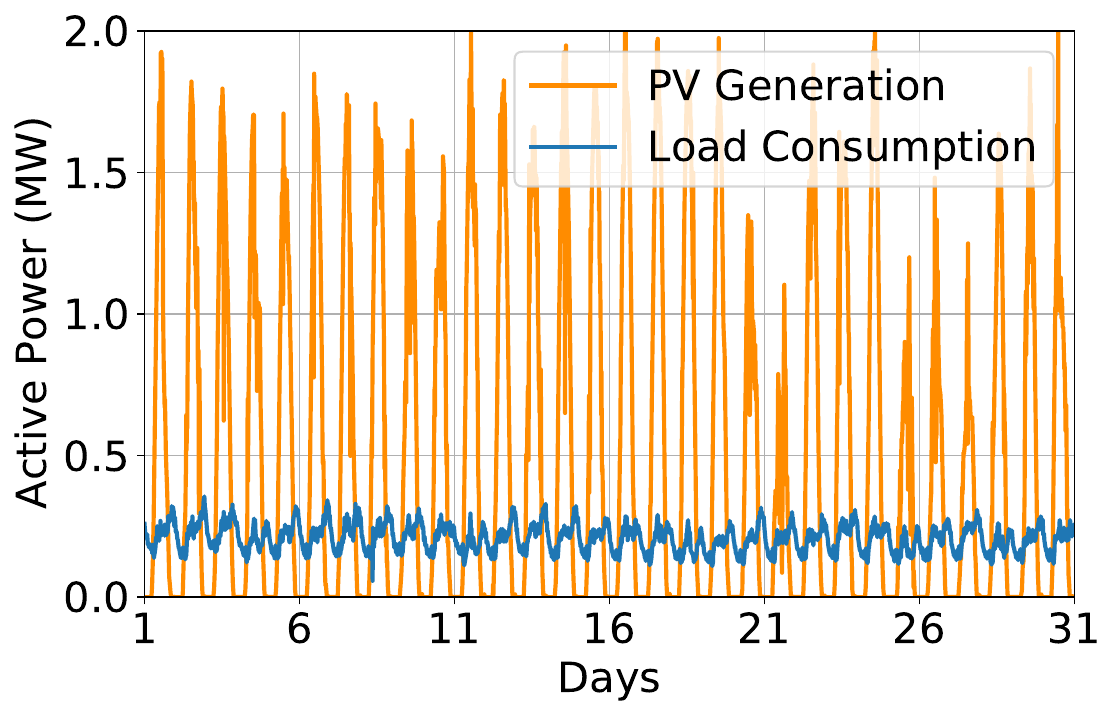}
                  \caption{A summer month of bus-25.}
            \end{subfigure}
            \caption{Daily power of the 33-bus network: active PV power generation and active load consumption for different buses in the 33-bus network.}
            \label{fig:33_bus_single}
        \end{figure}
        
        \begin{figure}[ht!]
            \centering
            \begin{subfigure}[h]{0.45\textwidth}
                  \centering
                  \includegraphics[width=\textwidth]{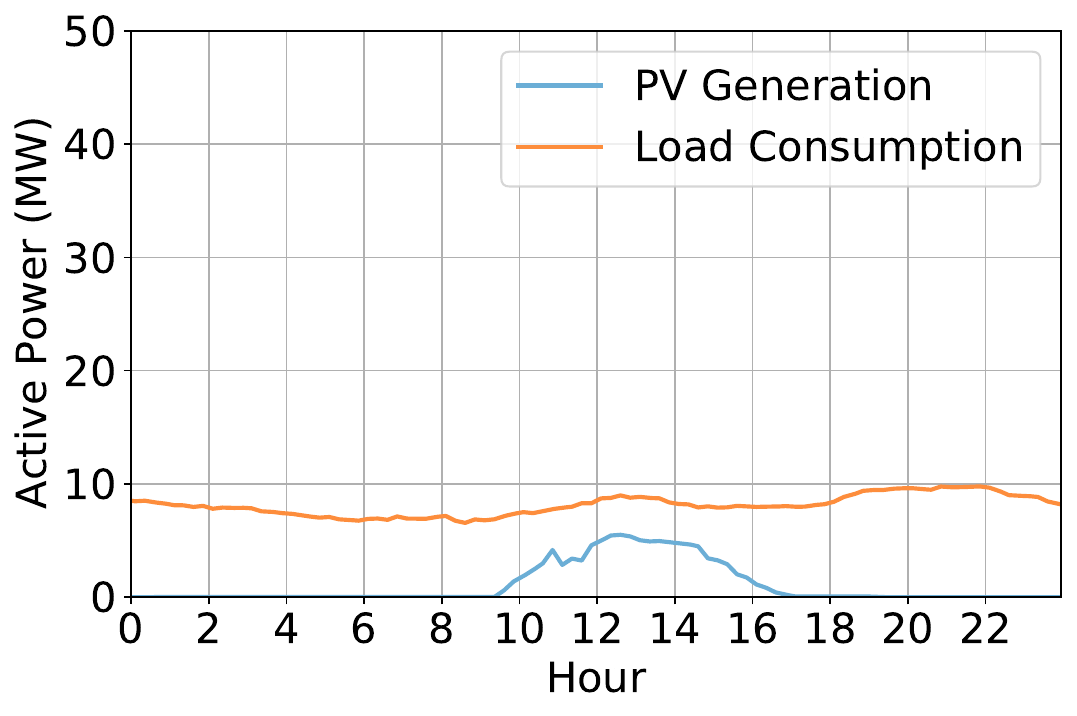}
                  \caption{A winter day.}
            \end{subfigure}
            \begin{subfigure}[h]{0.45\textwidth}
                  \centering
                  \includegraphics[width=\textwidth]{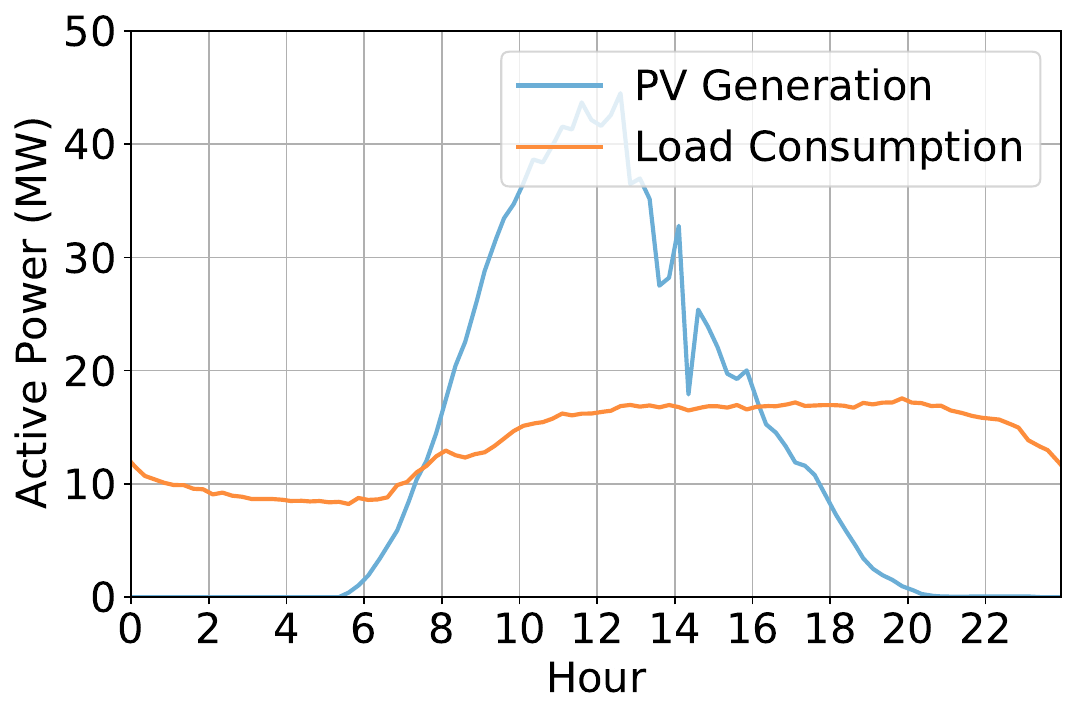}
                  \caption{A summer day.}
            \end{subfigure}
            \begin{subfigure}[h]{0.45\textwidth}
                  \centering
                  \includegraphics[width=\textwidth]{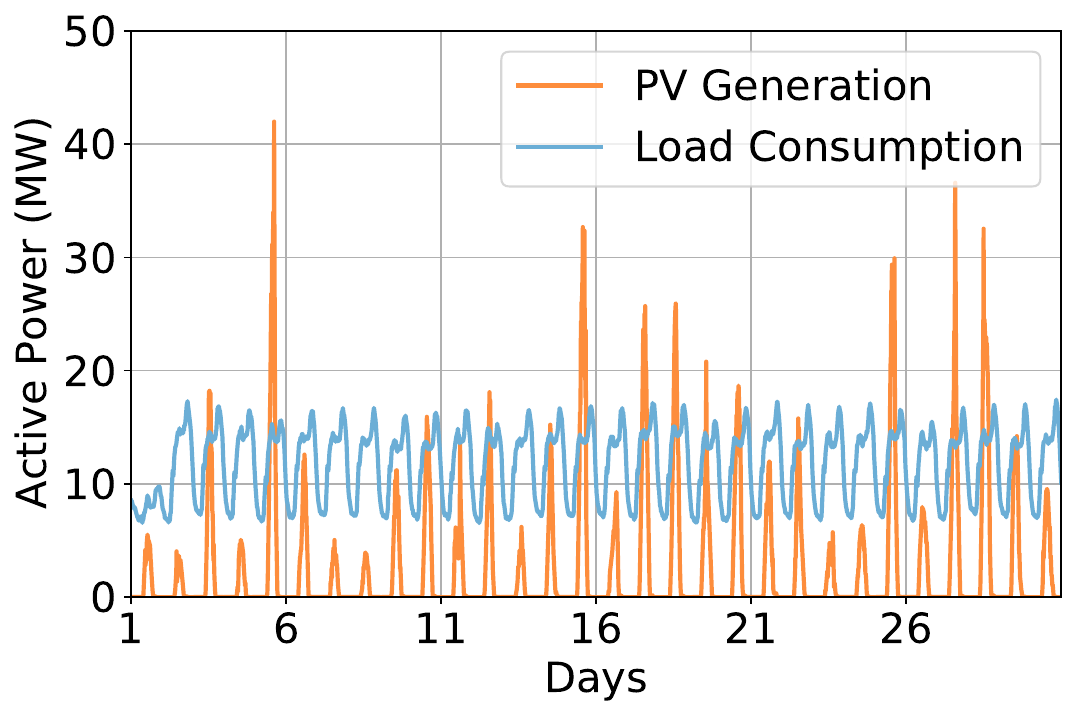}
                  \caption{A winter month (January).}
            \end{subfigure}
            \begin{subfigure}[h]{0.45\textwidth}
                  \centering
                  \includegraphics[width=\textwidth]{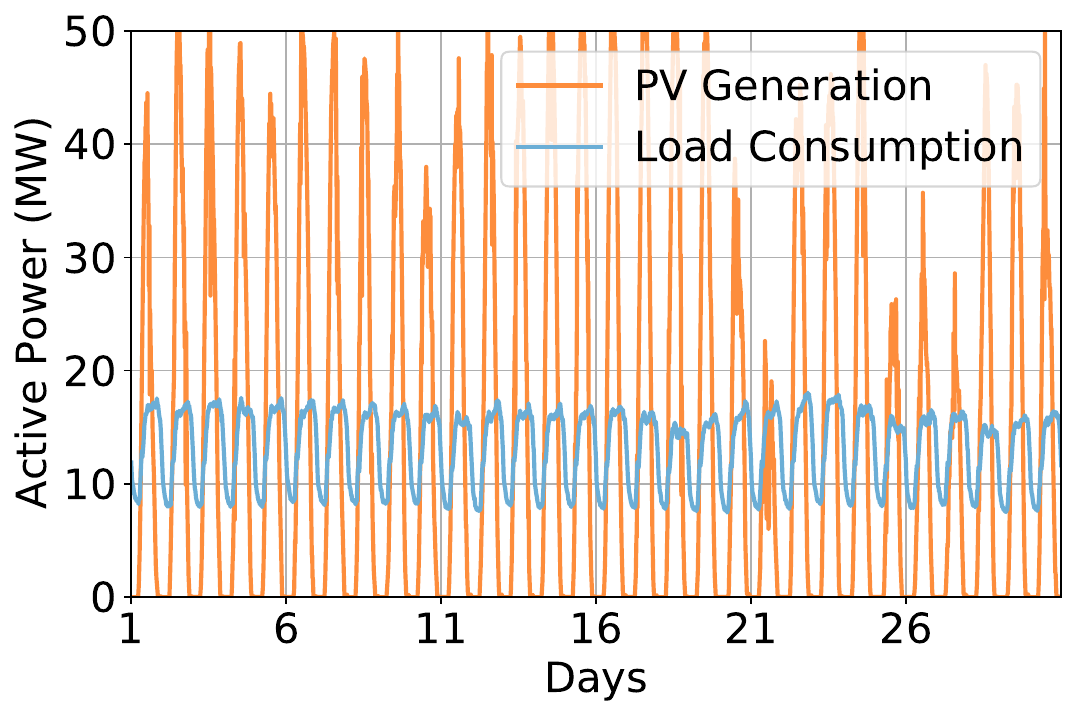}
                  \caption{A summer month (July).}
            \end{subfigure}
            \caption{Total power of the 141-bus network.}
            \label{fig:141_bus_total}
        \end{figure}
        
        \begin{figure}[ht!]
            \centering
            \begin{subfigure}[h]{0.45\textwidth}
                  \centering
                  \includegraphics[width=\textwidth]{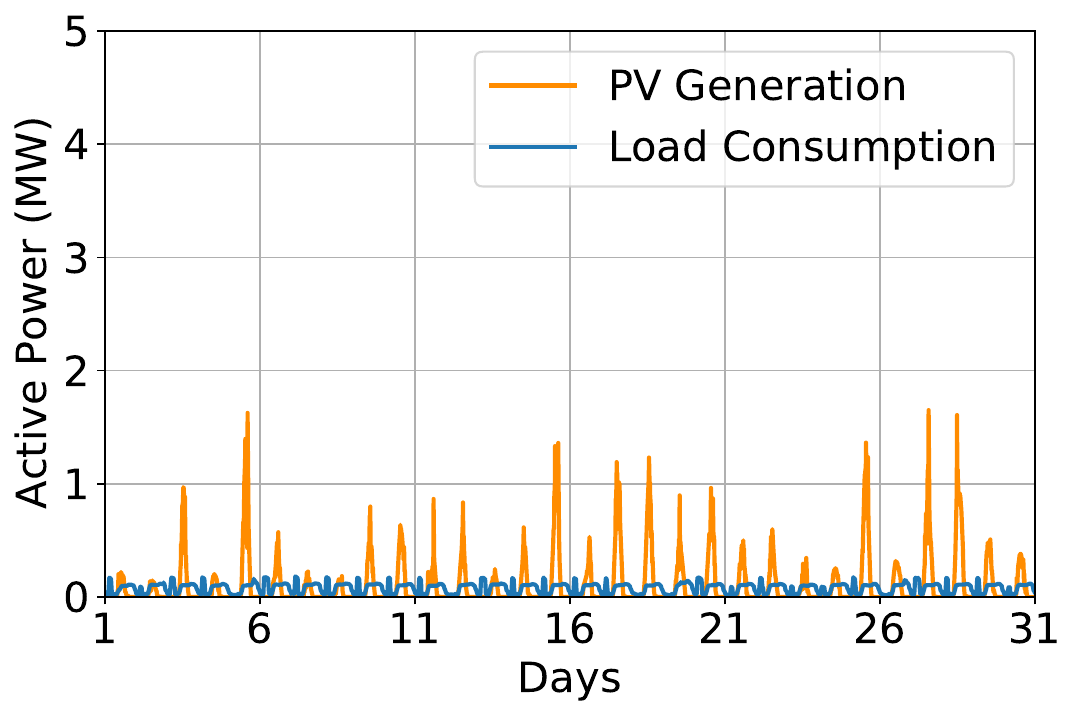}
                  \caption{A winter month of bus-36.}
            \end{subfigure}
            \begin{subfigure}[h]{0.45\textwidth}
                  \centering
                  \includegraphics[width=\textwidth]{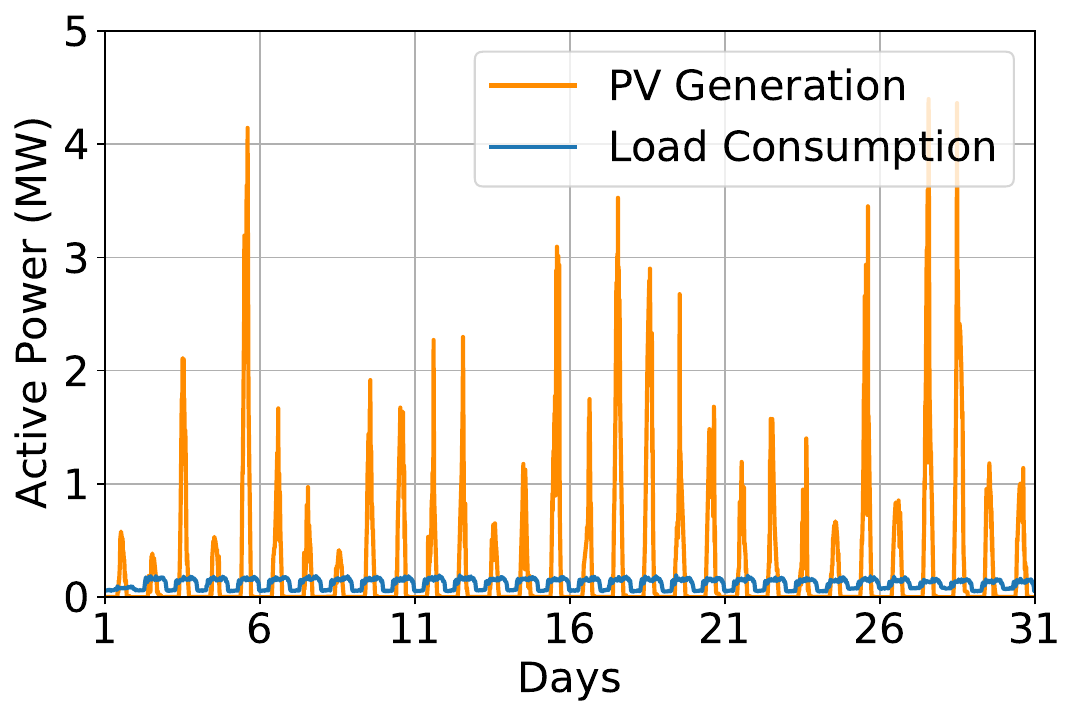}
                  \caption{A winter month of bus-77.}
            \end{subfigure}
            \begin{subfigure}[h]{0.45\textwidth}
                  \centering
                  \includegraphics[width=\textwidth]{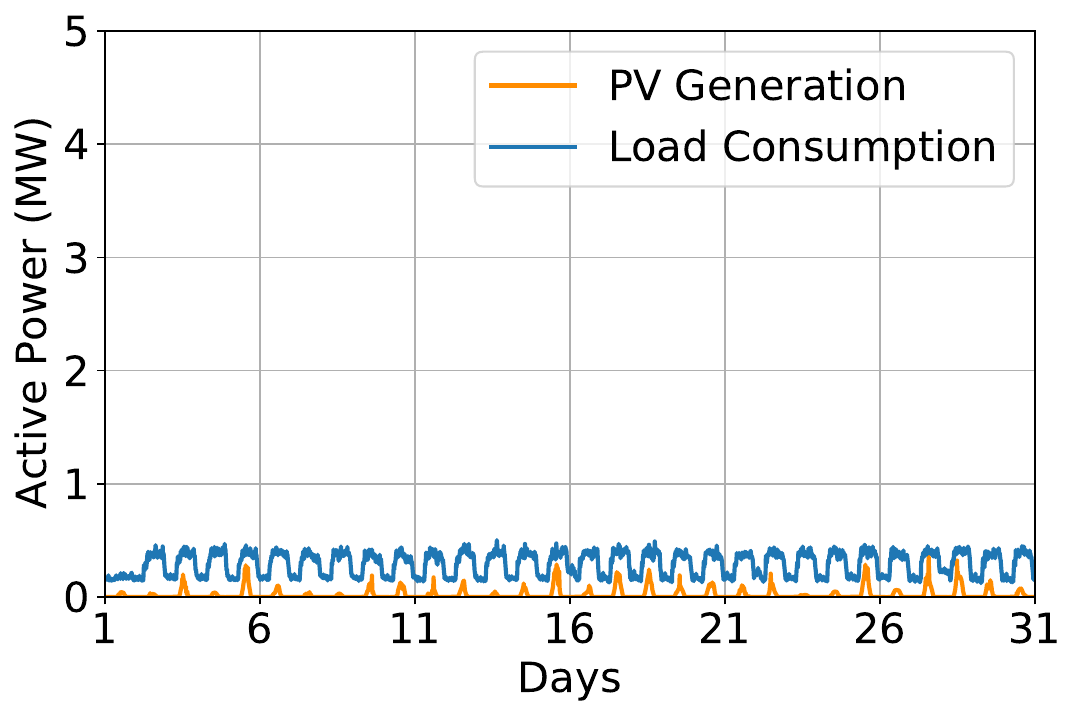}
                  \caption{A winter month of bus-100.}
            \end{subfigure}
            \begin{subfigure}[h]{0.45\textwidth}
                  \centering
                  \includegraphics[width=\textwidth]{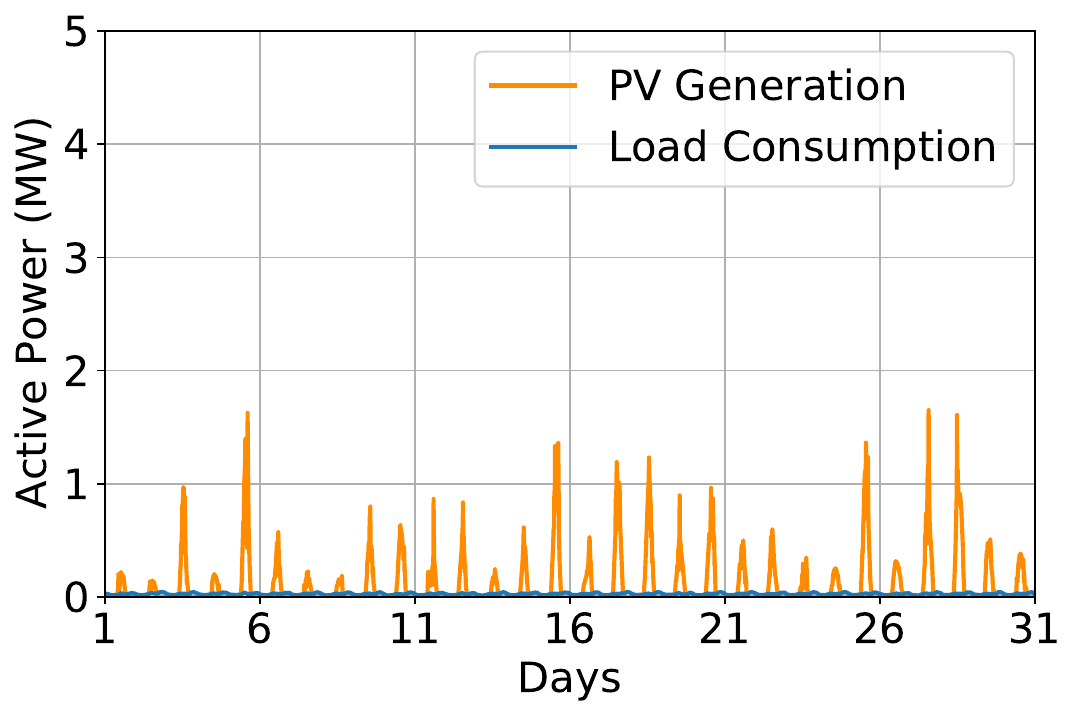}
                  \caption{A winter month of bus-111.}
            \end{subfigure}
            \begin{subfigure}[h]{0.45\textwidth}
                  \centering
                  \includegraphics[width=\textwidth]{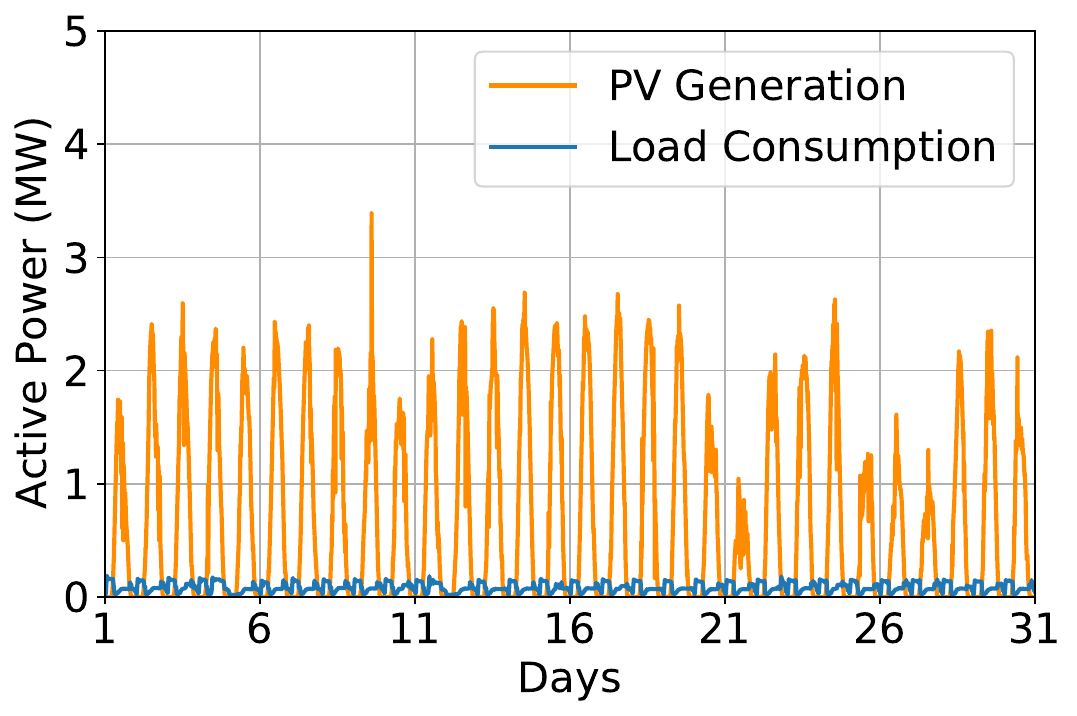}
                  \caption{A summer month of bus-36.}
            \end{subfigure}
            \begin{subfigure}[h]{0.45\textwidth}
                  \centering
                  \includegraphics[width=\textwidth]{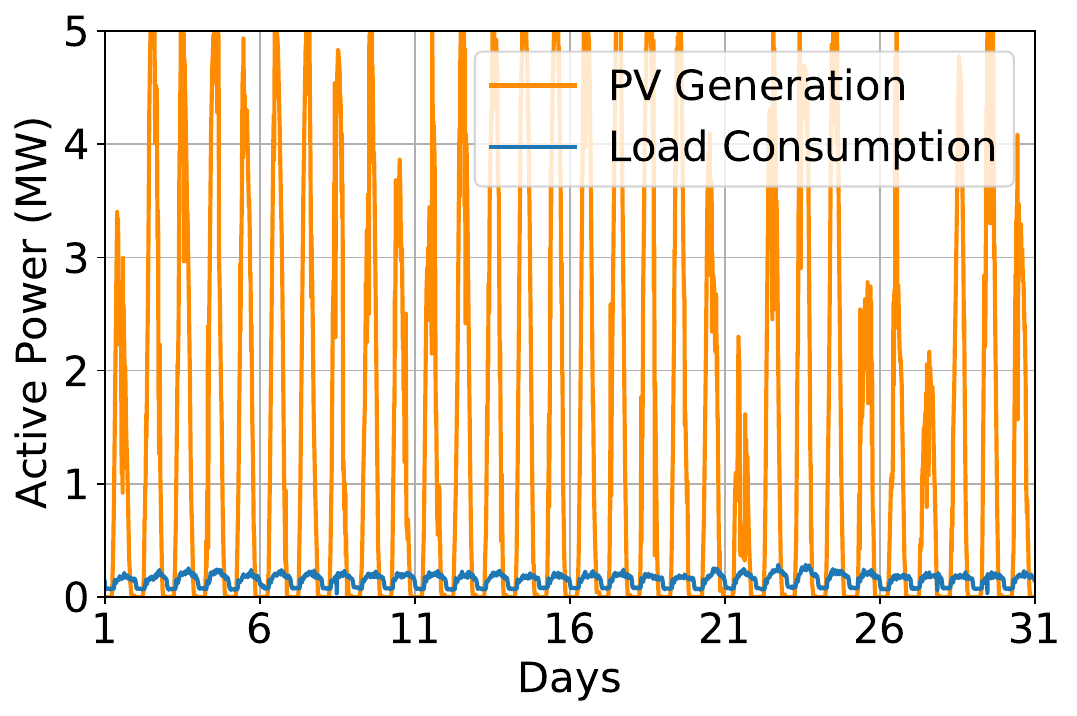}
                  \caption{A summer month of bus-77.}
            \end{subfigure}
            \begin{subfigure}[h]{0.45\textwidth}
                  \centering
                  \includegraphics[width=\textwidth]{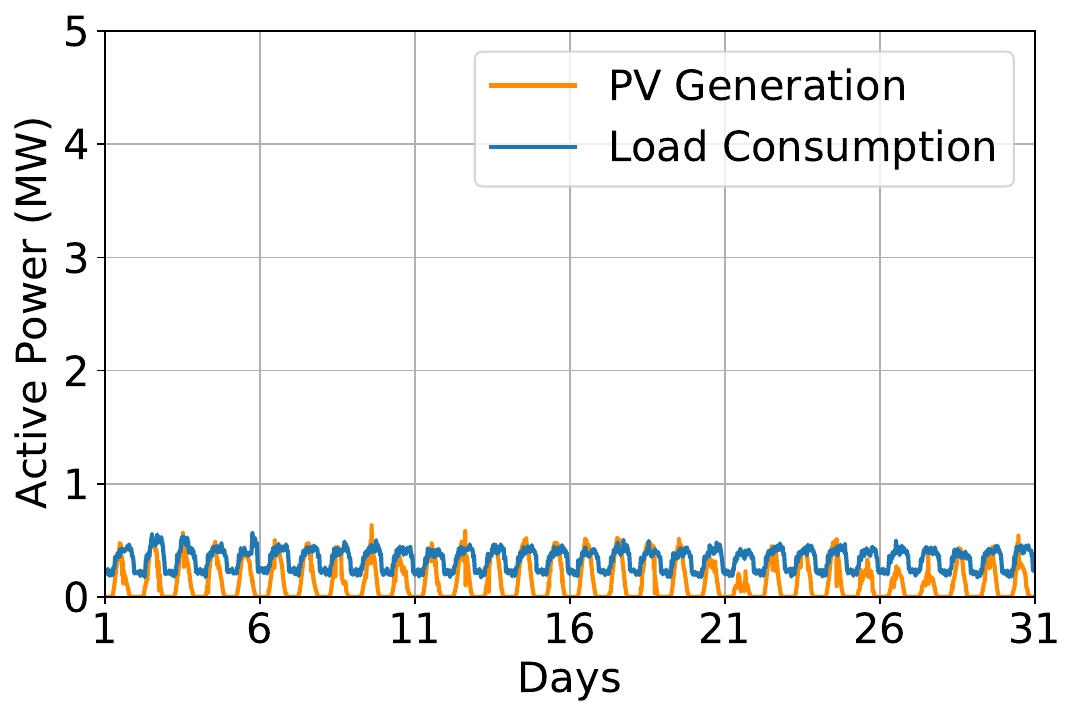}
                  \caption{A summer month of bus-100.}
            \end{subfigure}
            \begin{subfigure}[h]{0.45\textwidth}
                  \centering
                  \includegraphics[width=\textwidth]{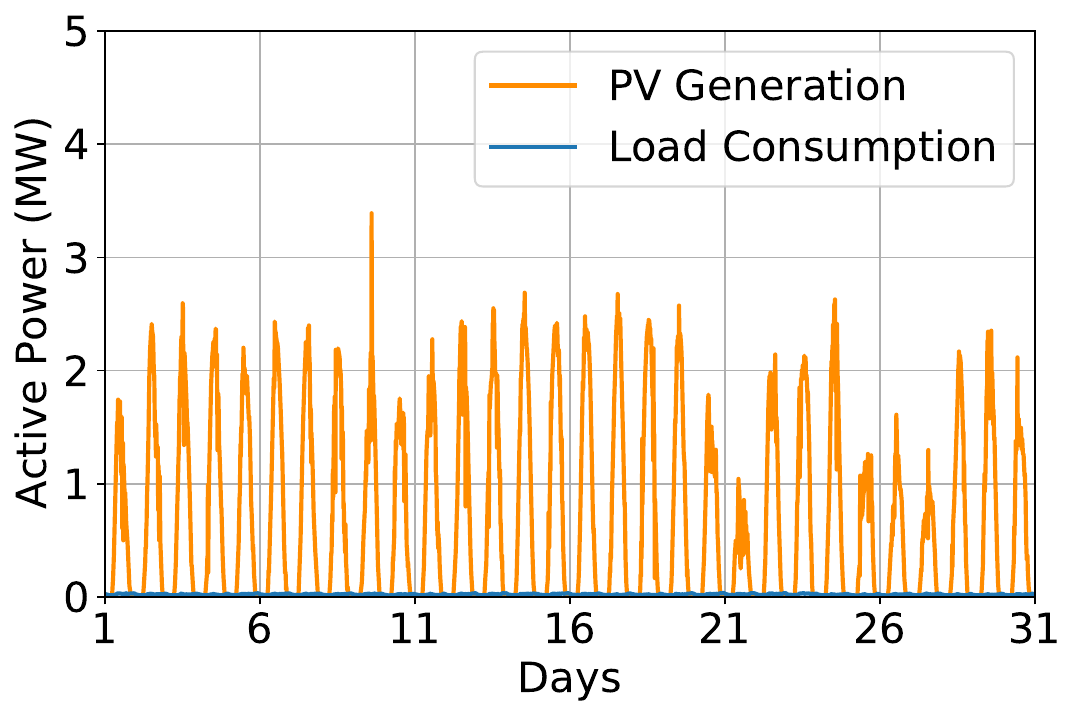}
                  \caption{A summer month of bus-111.}
            \end{subfigure}
            \caption{Daily power of the 141-bus network: active PV power generation and active load consumption for different buses in the 141-bus network.}
            \label{fig:141_bus_single}
        \end{figure}
        
        \begin{figure}[ht!]
            \centering
            \begin{subfigure}[h]{0.45\textwidth}
                  \centering
                  \includegraphics[width=\textwidth]{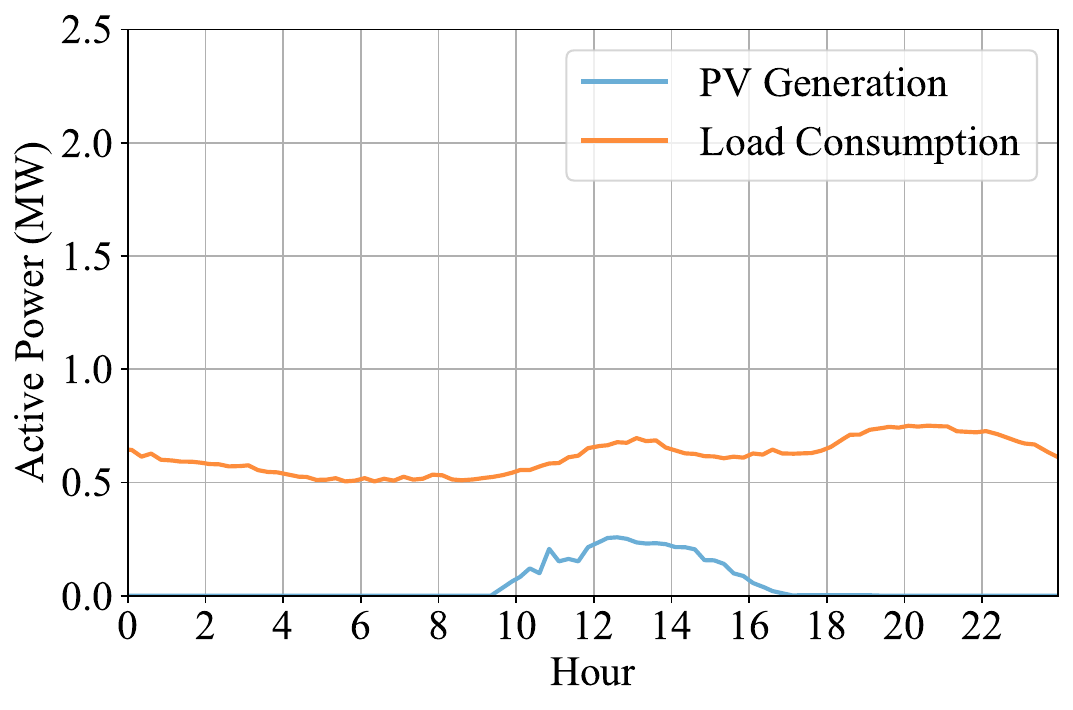}
                  \caption{A winter day.}
            \end{subfigure}
            \begin{subfigure}[h]{0.45\textwidth}
                  \centering
                  \includegraphics[width=\textwidth]{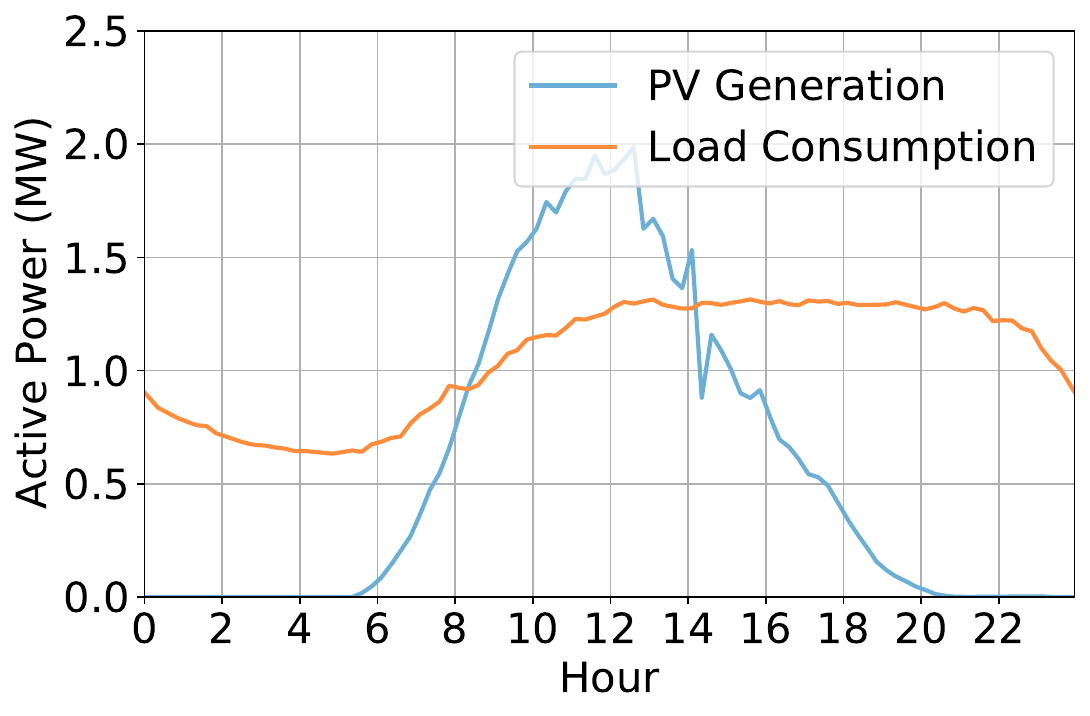}
                  \caption{A summer day.}
            \end{subfigure}
            \begin{subfigure}[h]{0.45\textwidth}
                  \centering
                  \includegraphics[width=\textwidth]{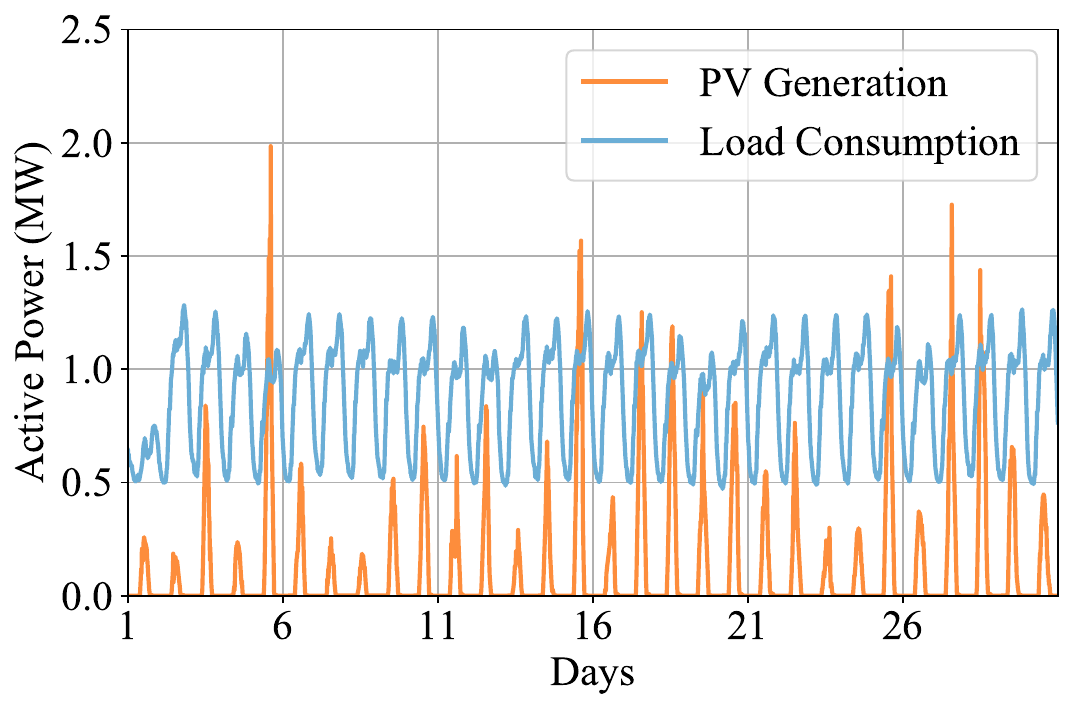}
                  \caption{A winter month (January).}
            \end{subfigure}
            \begin{subfigure}[h]{0.45\textwidth}
                  \centering
                  \includegraphics[width=\textwidth]{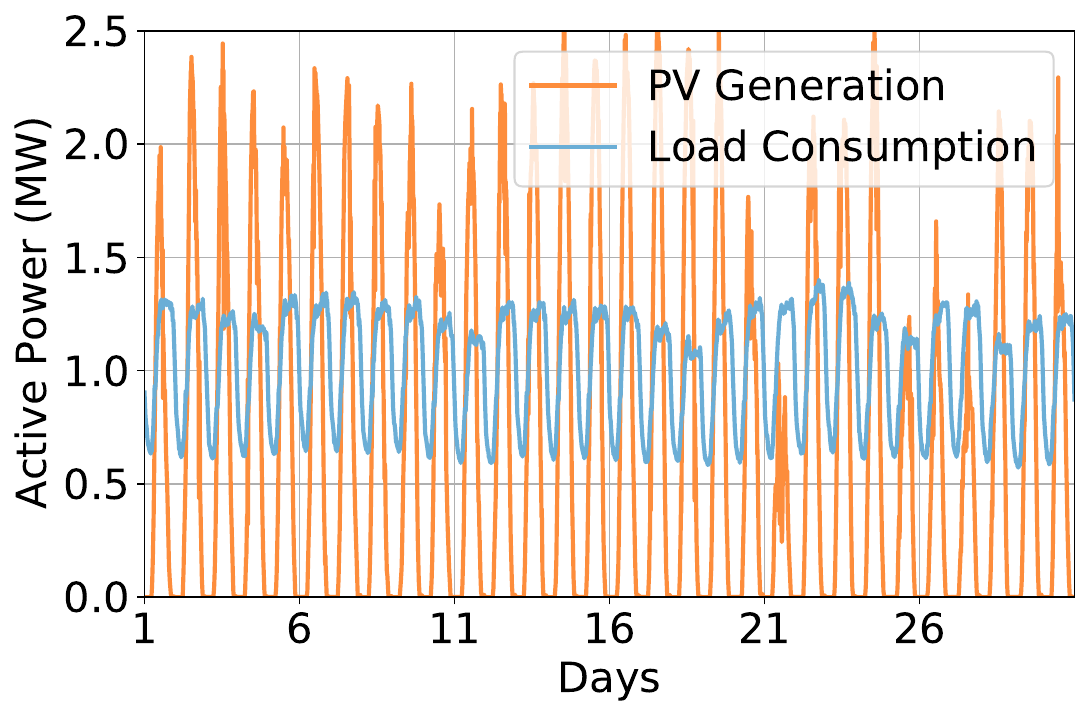}
                  \caption{A summer month (July).}
            \end{subfigure}
            \caption{Total power of the 322-bus network.}
            \label{fig:322_bus_total}
        \end{figure}
        
        \begin{figure}[ht!]
            \centering
            \begin{subfigure}[h]{0.45\textwidth}
                  \centering
                  \includegraphics[width=\textwidth]{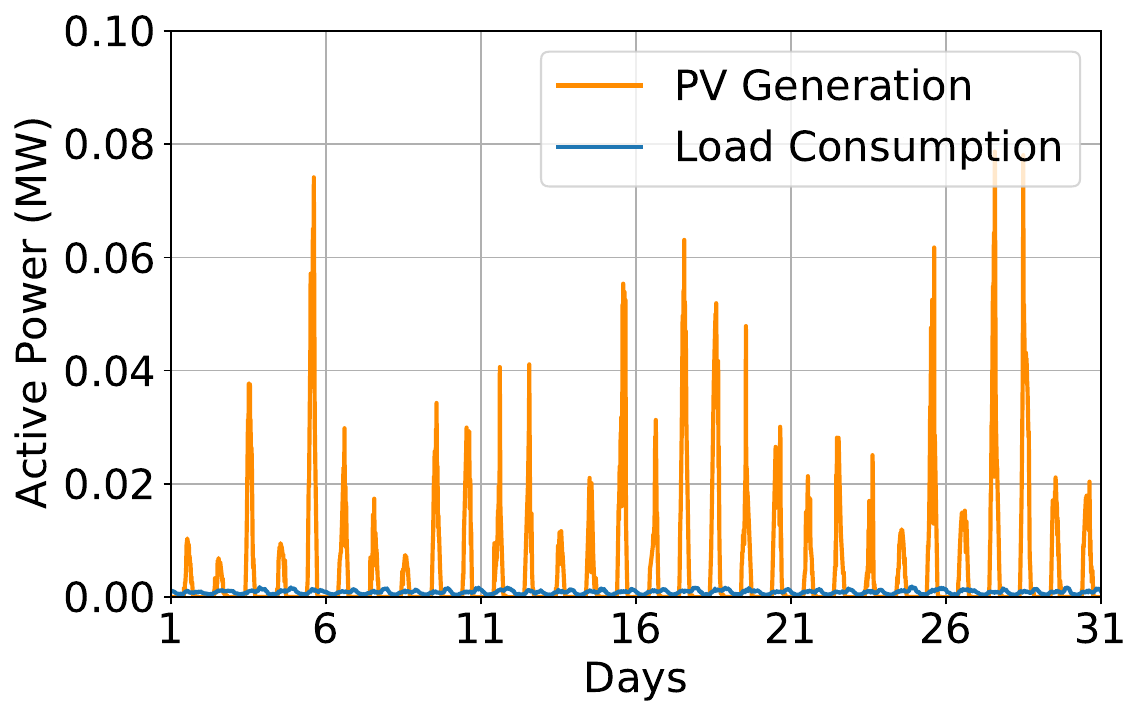}
                  \caption{A winter month of bus-54.}
            \end{subfigure}
            \begin{subfigure}[h]{0.45\textwidth}
                  \centering
                  \includegraphics[width=\textwidth]{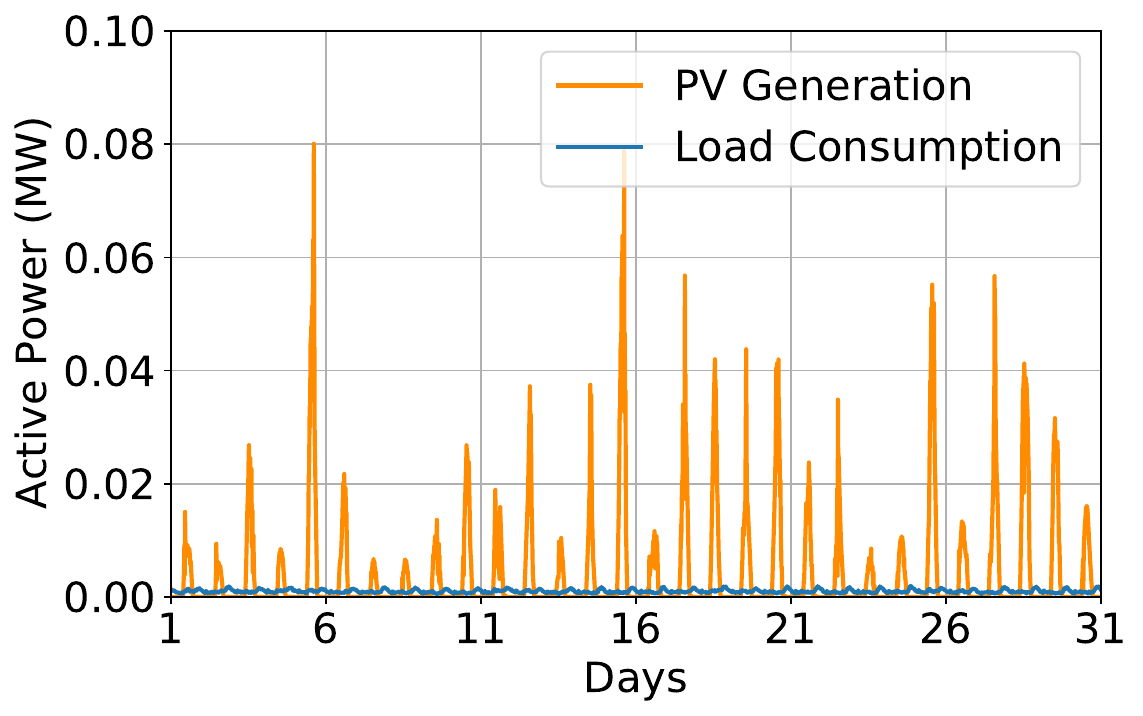}
                  \caption{A winter month of bus-147.}
            \end{subfigure}
            \begin{subfigure}[h]{0.45\textwidth}
                  \centering
                  \includegraphics[width=\textwidth]{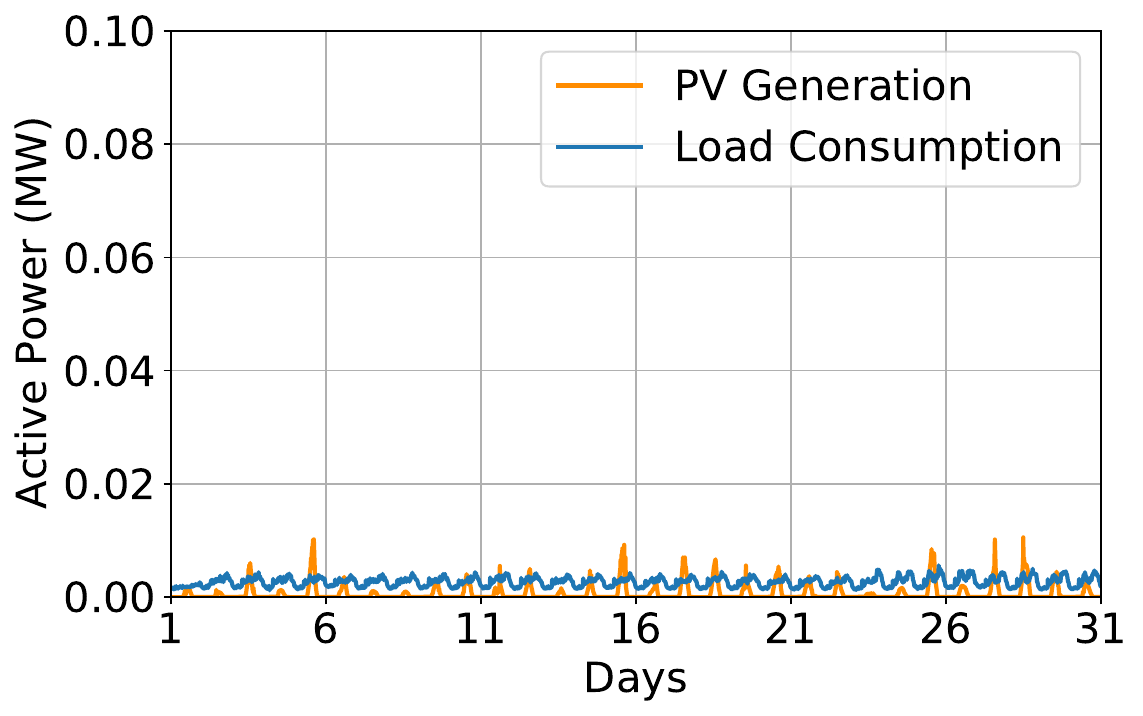}
                  \caption{A winter month of bus-297.}
            \end{subfigure}
            \begin{subfigure}[h]{0.45\textwidth}
                  \centering
                  \includegraphics[width=\textwidth]{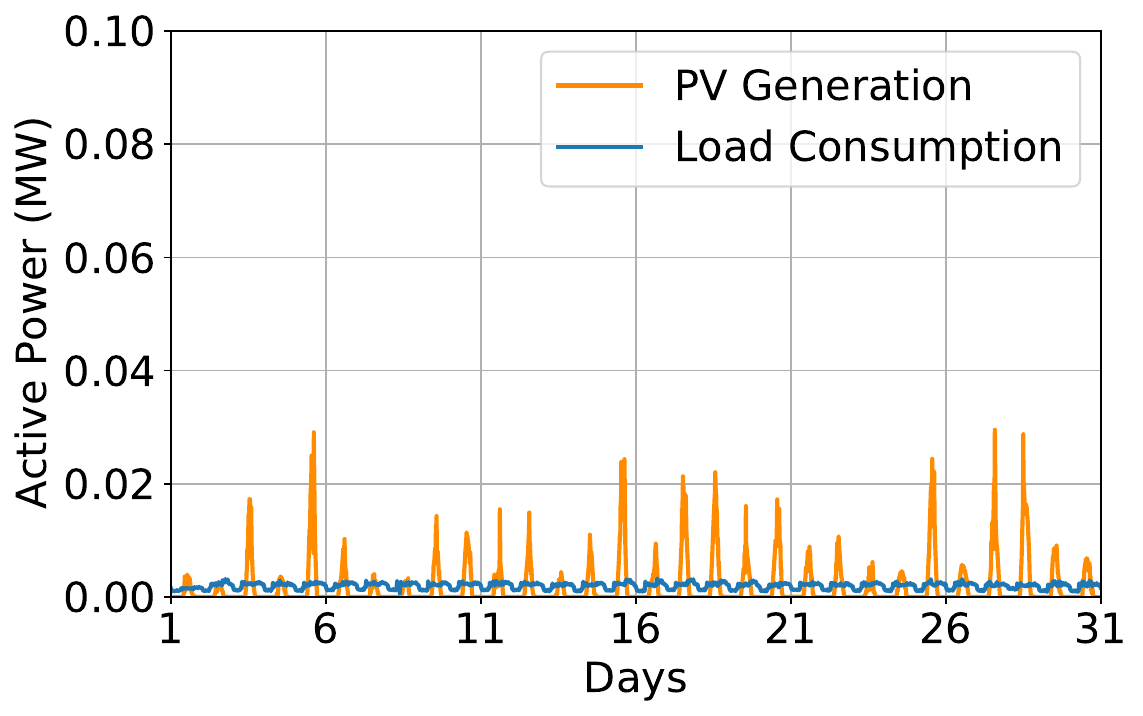}
                  \caption{A winter month of bus-322.}
            \end{subfigure}
            \begin{subfigure}[h]{0.45\textwidth}
                  \centering
                  \includegraphics[width=\textwidth]{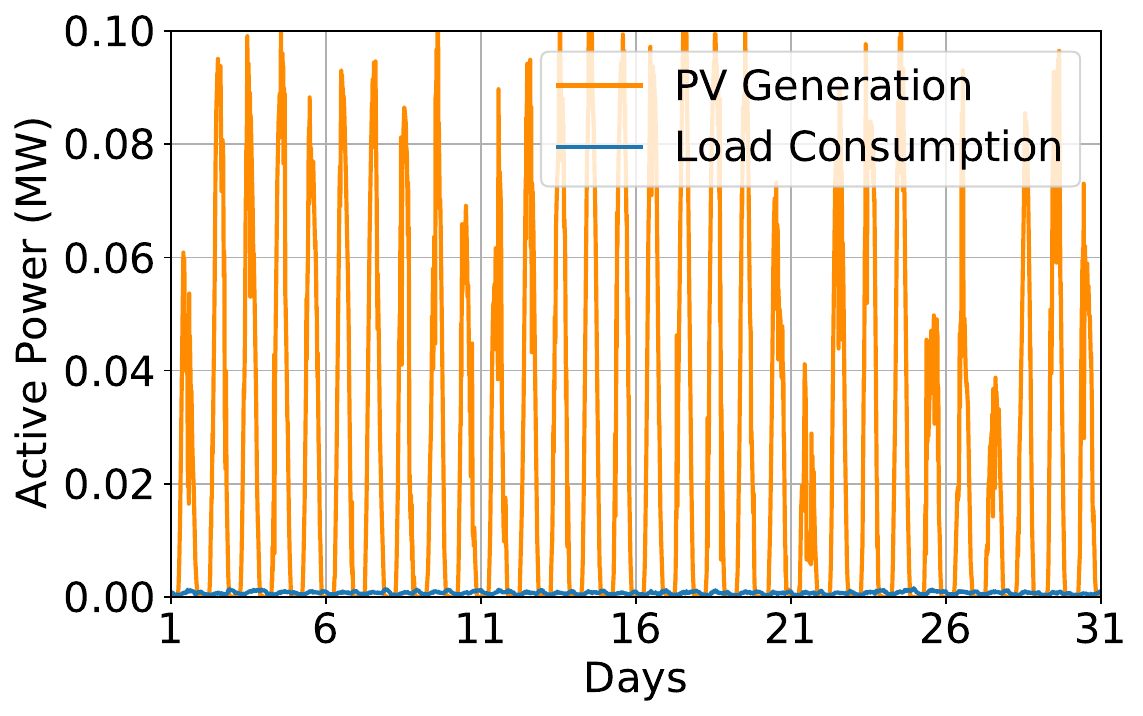}
                  \caption{A summer month of bus-54.}
            \end{subfigure}
            \begin{subfigure}[h]{0.45\textwidth}
                  \centering
                  \includegraphics[width=\textwidth]{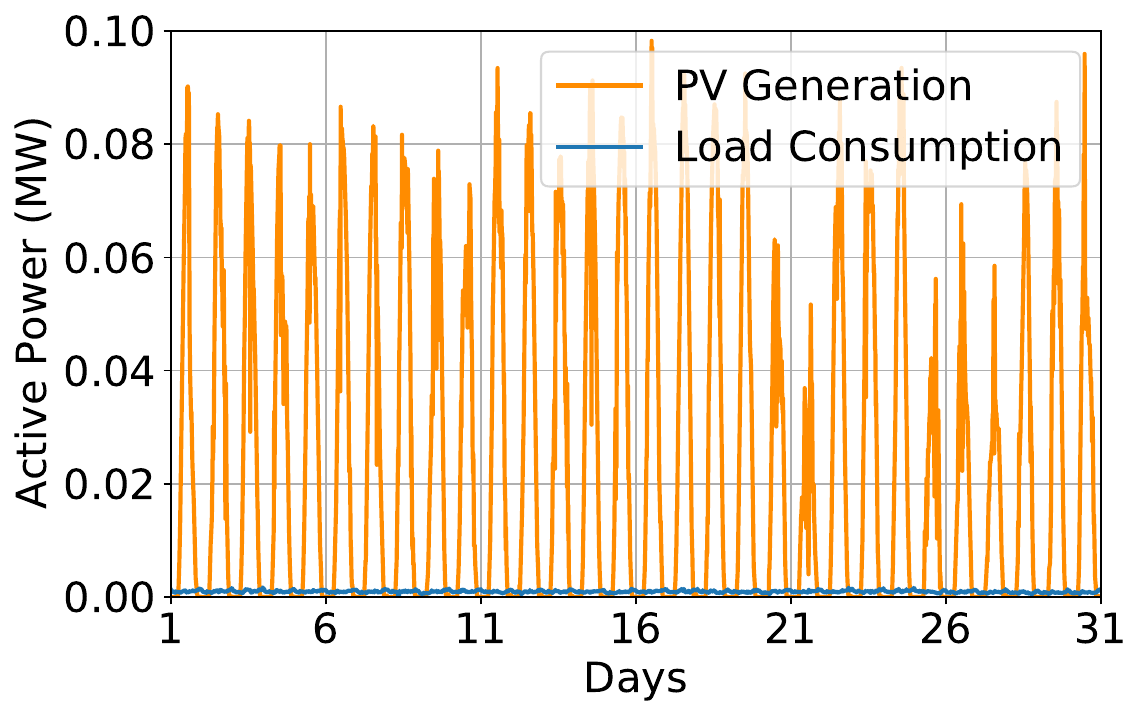}
                  \caption{A summer month of bus-147.}
            \end{subfigure}
            \begin{subfigure}[h]{0.45\textwidth}
                  \centering
                  \includegraphics[width=\textwidth]{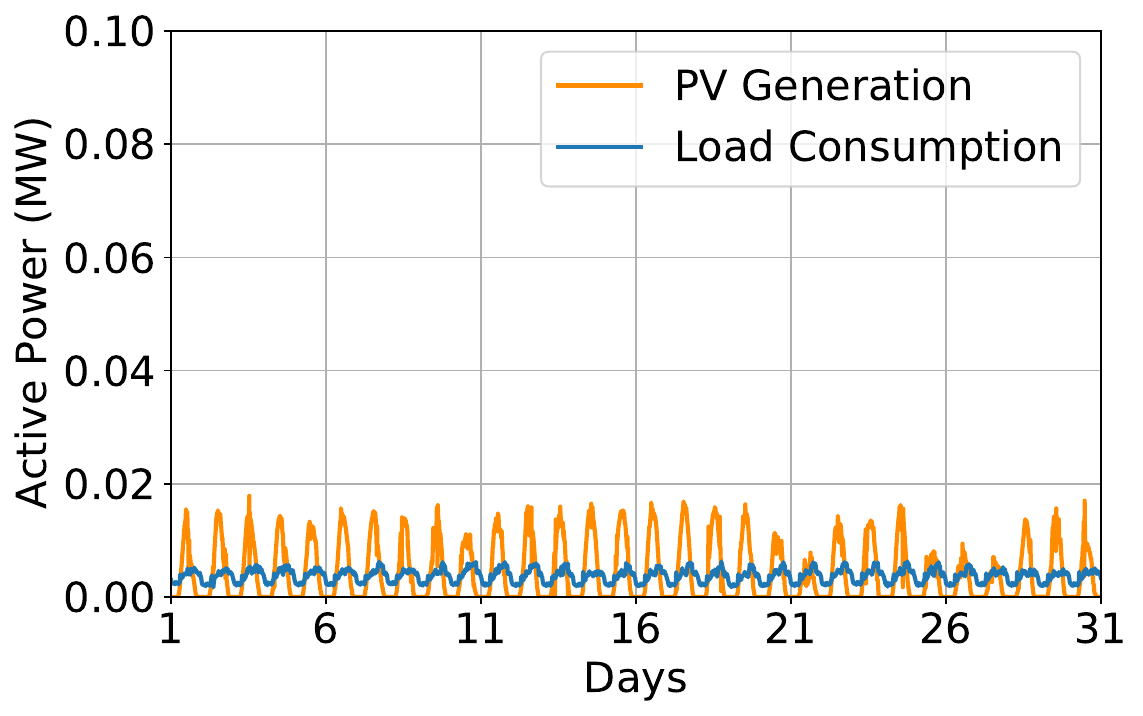}
                  \caption{A summer month of bus-297.}
            \end{subfigure}
            \begin{subfigure}[h]{0.45\textwidth}
                  \centering
                  \includegraphics[width=\textwidth]{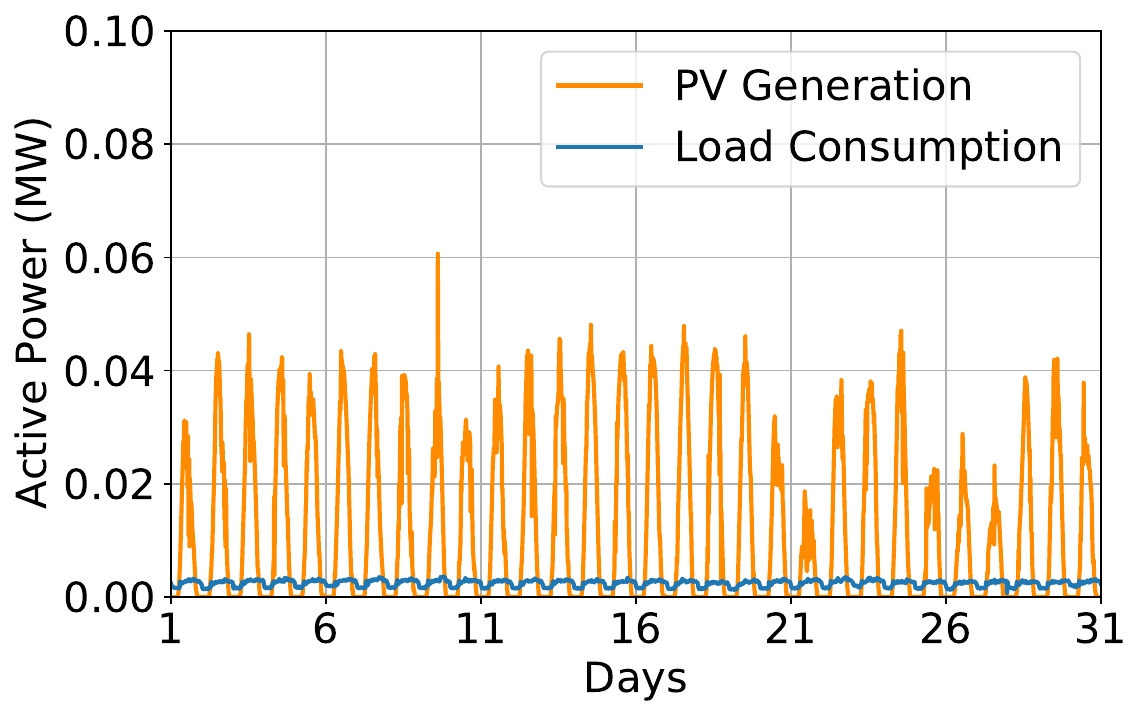}
                  \caption{A summer month of bus-322.}
            \end{subfigure}
            \caption{Daily power of the 322-bus network: active PV power generation and active load consumption for different buses in the 322-bus network.}
            \label{fig:322_bus_single}
        \end{figure}
        
        \begin{figure}[ht!]
            \centering
            \begin{subfigure}[h]{0.45\textwidth}
                  \centering
                  \includegraphics[width=\textwidth]{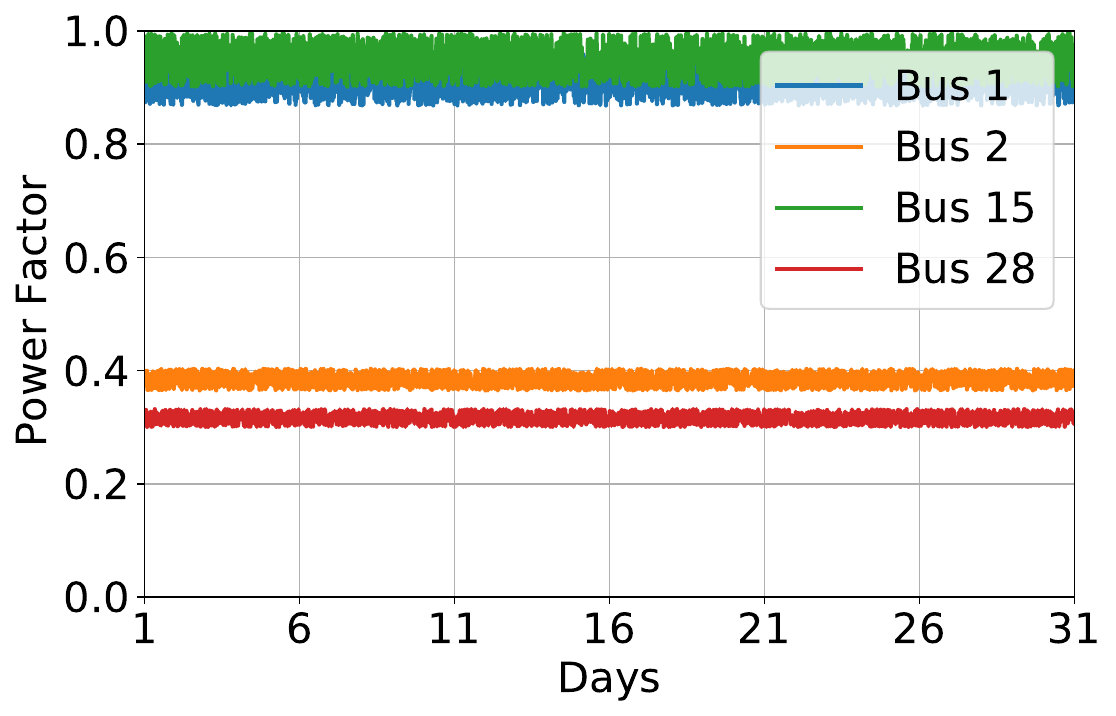}
                  \caption{The 33-bus network.}
            \end{subfigure}
            \begin{subfigure}[h]{0.45\textwidth}
                  \centering
                  \includegraphics[width=\textwidth]{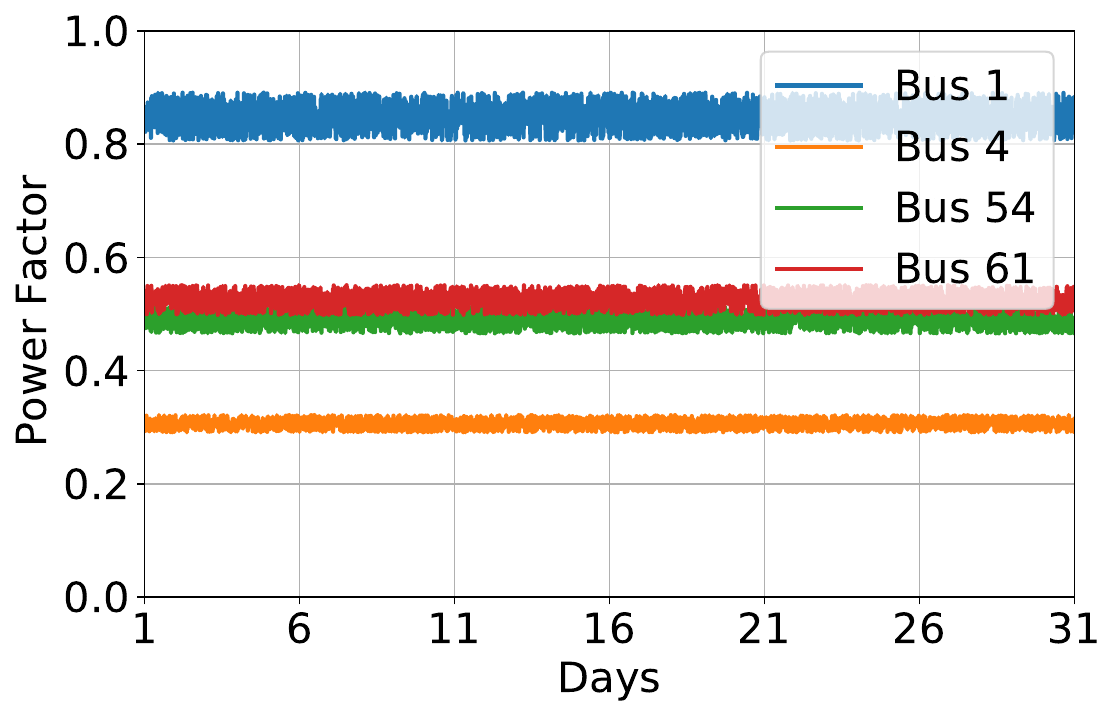}
                  \caption{The 141-bus network.}
            \end{subfigure}
            \begin{subfigure}[h]{0.45\textwidth}
                  \centering
                  \includegraphics[width=\textwidth]{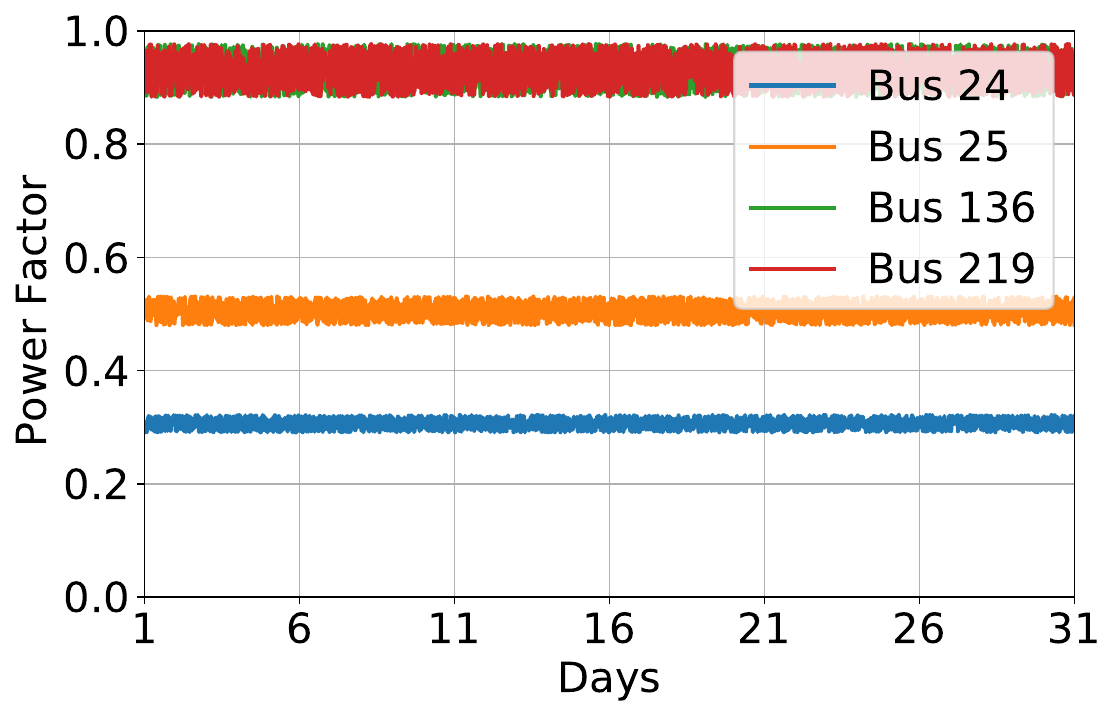}
                  \caption{The 322-bus network.}
            \end{subfigure}
            \caption{Power factors of four buses in all network topologies.}
            \label{fig:power_factor}
        \end{figure}

\addcontentsline{toc}{chapter}{Bibliography}
\bibliographystyle{IEEEtran}
\bibliography{bibliography}

\begin{thebibliography}{100}
\providecommand{\url}[1]{#1}
\csname url@samestyle\endcsname
\providecommand{\newblock}{\relax}
\providecommand{\bibinfo}[2]{#2}
\providecommand{\BIBentrySTDinterwordspacing}{\spaceskip=0pt\relax}
\providecommand{\BIBentryALTinterwordstretchfactor}{4}
\providecommand{\BIBentryALTinterwordspacing}{\spaceskip=\fontdimen2\font plus
\BIBentryALTinterwordstretchfactor\fontdimen3\font minus
  \fontdimen4\font\relax}
\providecommand{\BIBforeignlanguage}[2]{{%
\expandafter\ifx\csname l@#1\endcsname\relax
\typeout{** WARNING: IEEEtran.bst: No hyphenation pattern has been}%
\typeout{** loaded for the language `#1'. Using the pattern for}%
\typeout{** the default language instead.}%
\else
\language=\csname l@#1\endcsname
\fi
#2}}
\providecommand{\BIBdecl}{\relax}
\BIBdecl

\bibitem{koes2006constraint}
M.~Koes, I.~Nourbakhsh, and K.~Sycara, ``Constraint optimization coordination
  architecture for search and rescue robotics,'' in \emph{Proceedings 2006 IEEE
  International Conference on Robotics and Automation, 2006. ICRA 2006.}\hskip
  1em plus 0.5em minus 0.4em\relax IEEE, 2006, pp. 3977--3982.

\bibitem{ramchurn2010decentralized}
S.~D. Ramchurn, A.~Farinelli, K.~S. Macarthur, and N.~R. Jennings,
  ``Decentralized coordination in robocup rescue,'' \emph{The Computer
  Journal}, vol.~53, no.~9, pp. 1447--1461, 2010.

\bibitem{mannion2016experimental}
P.~Mannion, J.~Duggan, and E.~Howley, ``An experimental review of reinforcement
  learning algorithms for adaptive traffic signal control,'' \emph{Autonomic
  road transport support systems}, pp. 47--66, 2016.

\bibitem{cao2020reinforcement}
D.~Cao, W.~Hu, J.~Zhao, G.~Zhang, B.~Zhang, Z.~Liu, Z.~Chen, and F.~Blaabjerg,
  ``Reinforcement learning and its applications in modern power and energy
  systems: A review,'' \emph{Journal of Modern Power Systems and Clean Energy},
  vol.~8, no.~6, pp. 1029--1042, 2020.

\bibitem{silver2016mastering}
D.~Silver, A.~Huang, C.~J. Maddison, A.~Guez, L.~Sifre, G.~Van Den~Driessche,
  J.~Schrittwieser, I.~Antonoglou, V.~Panneershelvam, and M.~Lanctot,
  ``Mastering the game of go with deep neural networks and tree search,''
  \emph{nature}, vol. 529, no. 7587, p. 484, 2016.

\bibitem{heinrich2016deep}
J.~Heinrich and D.~Silver, ``Deep reinforcement learning from self-play in
  imperfect-information games,'' \emph{arXiv preprint arXiv:1603.01121}, 2016.

\bibitem{schulman2017proximal}
J.~Schulman, F.~Wolski, P.~Dhariwal, A.~Radford, and O.~Klimov, ``Proximal
  policy optimization algorithms,'' \emph{arXiv preprint arXiv:1707.06347},
  2017.

\bibitem{haarnoja2018soft}
T.~Haarnoja, A.~Zhou, K.~Hartikainen, G.~Tucker, S.~Ha, J.~Tan, V.~Kumar,
  H.~Zhu, A.~Gupta, P.~Abbeel \emph{et~al.}, ``Soft actor-critic algorithms and
  applications,'' \emph{arXiv preprint arXiv:1812.05905}, 2018.

\bibitem{tesauro2003extending}
G.~Tesauro, ``Extending q-learning to general adaptive multi-agent systems,''
  \emph{Advances in neural information processing systems}, vol.~16, 2003.

\bibitem{lanctot2017unified}
M.~Lanctot, V.~Zambaldi, A.~Gruslys, A.~Lazaridou, K.~Tuyls, J.~P{\'e}rolat,
  D.~Silver, and T.~Graepel, ``A unified game-theoretic approach to multiagent
  reinforcement learning,'' \emph{Advances in neural information processing
  systems}, vol.~30, 2017.

\bibitem{chang2004all}
Y.-H. Chang, T.~Ho, and L.~P. Kaelbling, ``All learning is local: Multi-agent
  learning in global reward games,'' in \emph{Advances in neural information
  processing systems}, 2004, pp. 807--814.

\bibitem{balch1997learning}
T.~Balch, ``Learning roles: Behavioral diversity in robot teams,''
  \emph{College of Computing Technical Report GIT-CC-97-12, Georgia Institute
  of Technology, Atlanta, Georgia}, vol.~73, 1997.

\bibitem{balch1999reward}
------, ``Reward and diversity in multirobot foraging,'' \emph{In IJCAI-99
  Workshop on Agents Learning About, From and With other Agents}, 1999.

\bibitem{foerster2018counterfactual}
J.~N. Foerster, G.~Farquhar, T.~Afouras, N.~Nardelli, and S.~Whiteson,
  ``Counterfactual multi-agent policy gradients,'' in \emph{Thirty-Second AAAI
  Conference on Artificial Intelligence}, 2018.

\bibitem{SunehagLGCZJLSL18}
P.~Sunehag, G.~Lever, A.~Gruslys, W.~M. Czarnecki, V.~F. Zambaldi,
  M.~Jaderberg, M.~Lanctot, N.~Sonnerat, J.~Z. Leibo, K.~Tuyls, and T.~Graepel,
  ``Value-decomposition networks for cooperative multi-agent learning based on
  team reward,'' in \emph{Proceedings of the 17th International Conference on
  Autonomous Agents and MultiAgent Systems, {AAMAS} 2018, Stockholm, Sweden,
  July 10-15, 2018}.\hskip 1em plus 0.5em minus 0.4em\relax International
  Foundation for Autonomous Agents and Multiagent Systems Richland, SC, {USA} /
  {ACM}, 2018, pp. 2085--2087.

\bibitem{shapley1953value}
L.~S. Shapley, ``A value for n-person games,'' \emph{Contributions to the
  Theory of Games}, vol.~2, no.~28, pp. 307--317, 1953.

\bibitem{lowe2017multi}
R.~Lowe, Y.~Wu, A.~Tamar, J.~Harb, O.~P. Abbeel, and I.~Mordatch, ``Multi-agent
  actor-critic for mixed cooperative-competitive environments,'' in
  \emph{Advances in Neural Information Processing Systems}, 2017, pp.
  6379--6390.

\bibitem{watkins1992q}
C.~J. Watkins and P.~Dayan, ``Q-learning,'' \emph{Machine learning}, vol.~8,
  no. 3-4, pp. 279--292, 1992.

\bibitem{silver2014deterministic}
D.~Silver, G.~Lever, N.~Heess, T.~Degris, D.~Wierstra, and M.~Riedmiller,
  ``Deterministic policy gradient algorithms,'' in \emph{ICML}, 2014.

\bibitem{busoniu2008comprehensive}
L.~Busoniu, R.~Babuska, and B.~De~Schutter, ``A comprehensive survey of
  multiagent reinforcement learning,'' \emph{IEEE Transactions on Systems, Man,
  and Cybernetics, Part C (Applications and Reviews)}, vol.~38, no.~2, pp.
  156--172, 2008.

\bibitem{stone2000multiagent}
P.~Stone and M.~Veloso, ``Multiagent systems: A survey from a machine learning
  perspective,'' \emph{Autonomous Robots}, vol.~8, no.~3, pp. 345--383, 2000.

\bibitem{wooldridge2009introduction}
M.~Wooldridge, \emph{An introduction to multiagent systems}.\hskip 1em plus
  0.5em minus 0.4em\relax John wiley \& sons, 2009.

\bibitem{wang2021multi}
J.~Wang, W.~Xu, Y.~Gu, W.~Song, and T.~Green, ``Multi-agent reinforcement
  learning for active voltage control on power distribution networks,''
  \emph{Advances in Neural Information Processing Systems}, vol.~34, 2021.

\bibitem{chalkiadakis2011computational}
G.~Chalkiadakis, E.~Elkind, and M.~Wooldridge, ``Computational aspects of
  cooperative game theory,'' \emph{Synthesis Lectures on Artificial
  Intelligence and Machine Learning}, vol.~5, no.~6, pp. 1--168, 2011.

\bibitem{hernandez2019survey}
P.~Hernandez-Leal, B.~Kartal, and M.~E. Taylor, ``A survey and critique of
  multiagent deep reinforcement learning,'' \emph{Autonomous Agents and
  Multi-Agent Systems}, vol.~33, no.~6, pp. 750--797, 2019.

\bibitem{perolat2017learning}
J.~P{\'e}rolat, F.~Strub, B.~Piot, and O.~Pietquin, ``Learning nash equilibrium
  for general-sum markov games from batch data,'' in \emph{Artificial
  Intelligence and Statistics}.\hskip 1em plus 0.5em minus 0.4em\relax PMLR,
  2017, pp. 232--241.

\bibitem{sen1999learning}
S.~Sen and G.~Weiss, ``Learning in multiagent systems,'' \emph{Multiagent
  systems: A modern approach to distributed artificial intelligence}, pp.
  259--298, 1999.

\bibitem{fudenberg1998theory}
D.~Fudenberg, F.~Drew, D.~K. Levine, and D.~K. Levine, \emph{The theory of
  learning in games}.\hskip 1em plus 0.5em minus 0.4em\relax MIT press, 1998,
  vol.~2.

\bibitem{osborne2004introduction}
M.~J. Osborne \emph{et~al.}, \emph{An introduction to game theory}.\hskip 1em
  plus 0.5em minus 0.4em\relax Oxford university press New York, 2004, vol.~3,
  no.~3.

\bibitem{monderer1996potential}
D.~Monderer and L.~S. Shapley, ``Potential games,'' \emph{Games and economic
  behavior}, vol.~14, no.~1, pp. 124--143, 1996.

\bibitem{monderer1996fictitious}
------, ``Fictitious play property for games with identical interests,''
  \emph{Journal of economic theory}, vol.~68, no.~1, pp. 258--265, 1996.

\bibitem{leslie2006generalised}
D.~S. Leslie and E.~J. Collins, ``Generalised weakened fictitious play,''
  \emph{Games and Economic Behavior}, vol.~56, no.~2, pp. 285--298, 2006.

\bibitem{marden2009joint}
J.~R. Marden, G.~Arslan, and J.~S. Shamma, ``Joint strategy fictitious play
  with inertia for potential games,'' \emph{IEEE Transactions on Automatic
  Control}, vol.~54, no.~2, pp. 208--220, 2009.

\bibitem{heinrich2015fictitious}
J.~Heinrich, M.~Lanctot, and D.~Silver, ``Fictitious self-play in
  extensive-form games,'' in \emph{International conference on machine
  learning}.\hskip 1em plus 0.5em minus 0.4em\relax PMLR, 2015, pp. 805--813.

\bibitem{gagniuc2017markov}
P.~A. Gagniuc, \emph{Markov chains: from theory to implementation and
  experimentation}.\hskip 1em plus 0.5em minus 0.4em\relax John Wiley \& Sons,
  2017.

\bibitem{box2015time}
G.~E. Box, G.~M. Jenkins, G.~C. Reinsel, and G.~M. Ljung, \emph{Time series
  analysis: forecasting and control}.\hskip 1em plus 0.5em minus 0.4em\relax
  John Wiley \& Sons, 2015.

\bibitem{bellman1952theory}
R.~Bellman, ``On the theory of dynamic programming,'' \emph{Proceedings of the
  National Academy of Sciences of the United States of America}, vol.~38,
  no.~8, p. 716, 1952.

\bibitem{sutton2018reinforcement}
R.~S. Sutton and A.~G. Barto, \emph{Reinforcement learning: An
  introduction}.\hskip 1em plus 0.5em minus 0.4em\relax MIT press, 2018.

\bibitem{bertsekas2019reinforcement}
D.~Bertsekas, \emph{Reinforcement learning and optimal control}.\hskip 1em plus
  0.5em minus 0.4em\relax Athena Scientific, 2019.

\bibitem{williams1992simple}
R.~J. Williams, ``Simple statistical gradient-following algorithms for
  connectionist reinforcement learning,'' \emph{Machine learning}, vol.~8, no.
  3-4, pp. 229--256, 1992.

\bibitem{melo2001convergence}
F.~S. Melo, ``Convergence of q-learning: A simple proof,'' \emph{Institute Of
  Systems and Robotics, Tech. Rep}, pp. 1--4, 2001.

\bibitem{konda2000actor}
V.~R. Konda and J.~N. Tsitsiklis, ``Actor-critic algorithms,'' in
  \emph{Advances in neural information processing systems}.\hskip 1em plus
  0.5em minus 0.4em\relax Citeseer, 2000, pp. 1008--1014.

\bibitem{mnih2015human}
V.~Mnih, K.~Kavukcuoglu, D.~Silver, A.~A. Rusu, J.~Veness, M.~G. Bellemare,
  A.~Graves, M.~Riedmiller, A.~K. Fidjeland, G.~Ostrovski \emph{et~al.},
  ``Human-level control through deep reinforcement learning,'' \emph{nature},
  vol. 518, no. 7540, pp. 529--533, 2015.

\bibitem{riedmiller2005neural}
M.~Riedmiller, ``Neural fitted q iteration--first experiences with a data
  efficient neural reinforcement learning method,'' in \emph{European
  conference on machine learning}.\hskip 1em plus 0.5em minus 0.4em\relax
  Springer, 2005, pp. 317--328.

\bibitem{lillicrap2015continuous}
T.~P. Lillicrap, J.~J. Hunt, A.~Pritzel, N.~Heess, T.~Erez, Y.~Tassa,
  D.~Silver, and D.~Wierstra, ``Continuous control with deep reinforcement
  learning,'' \emph{arXiv preprint arXiv:1509.02971}, 2015.

\bibitem{schulman2015trust}
J.~Schulman, S.~Levine, P.~Abbeel, M.~Jordan, and P.~Moritz, ``Trust region
  policy optimization,'' in \emph{International conference on machine
  learning}.\hskip 1em plus 0.5em minus 0.4em\relax PMLR, 2015, pp. 1889--1897.

\bibitem{kloek1978bayesian}
T.~Kloek and H.~K. Van~Dijk, ``Bayesian estimates of equation system
  parameters: an application of integration by monte carlo,''
  \emph{Econometrica: Journal of the Econometric Society}, pp. 1--19, 1978.

\bibitem{kullback1951information}
S.~Kullback and R.~A. Leibler, ``On information and sufficiency,'' \emph{The
  annals of mathematical statistics}, vol.~22, no.~1, pp. 79--86, 1951.

\bibitem{maskin2001markov}
E.~Maskin and J.~Tirole, ``Markov perfect equilibrium: I. observable actions,''
  \emph{Journal of Economic Theory}, vol. 100, no.~2, pp. 191--219, 2001.

\bibitem{powers2004new}
R.~Powers and Y.~Shoham, ``New criteria and a new algorithm for learning in
  multi-agent systems,'' \emph{Advances in neural information processing
  systems}, vol.~17, 2004.

\bibitem{bowling2002multiagent}
M.~Bowling and M.~Veloso, ``Multiagent learning using a variable learning
  rate,'' \emph{Artificial Intelligence}, vol. 136, no.~2, pp. 215--250, 2002.

\bibitem{bowling2004convergence}
M.~Bowling, ``Convergence and no-regret in multiagent learning,''
  \emph{Advances in neural information processing systems}, vol.~17, 2004.

\bibitem{bai2020provable}
Y.~Bai and C.~Jin, ``Provable self-play algorithms for competitive
  reinforcement learning,'' in \emph{International conference on machine
  learning}.\hskip 1em plus 0.5em minus 0.4em\relax PMLR, 2020, pp. 551--560.

\bibitem{krupnik2020multi}
O.~Krupnik, I.~Mordatch, and A.~Tamar, ``Multi-agent reinforcement learning
  with multi-step generative models,'' in \emph{Conference on Robot
  Learning}.\hskip 1em plus 0.5em minus 0.4em\relax PMLR, 2020, pp. 776--790.

\bibitem{willemsen2021mambpo}
D.~Willemsen, M.~Coppola, and G.~C. de~Croon, ``Mambpo: Sample-efficient
  multi-robot reinforcement learning using learned world models,'' in
  \emph{2021 IEEE/RSJ International Conference on Intelligent Robots and
  Systems (IROS)}.\hskip 1em plus 0.5em minus 0.4em\relax IEEE, 2021, pp.
  5635--5640.

\bibitem{bu2008comprehensive}
L.~Bu, R.~Babu, and B.~De~Schutter, ``A comprehensive survey of multiagent
  reinforcement learning,'' \emph{IEEE Transactions on Systems, Man, and
  Cybernetics, Part C (Applications and Reviews)}, vol.~38, no.~2, pp.
  156--172, 2008.

\bibitem{LyuXDA21}
X.~Lyu, Y.~Xiao, B.~Daley, and C.~Amato, ``Contrasting centralized and
  decentralized critics in multi-agent reinforcement learning,'' in
  \emph{{AAMAS} '21: 20th International Conference on Autonomous Agents and
  Multiagent Systems, Virtual Event, United Kingdom, May 3-7, 2021}, F.~Dignum,
  A.~Lomuscio, U.~Endriss, and A.~Now{\'{e}}, Eds.\hskip 1em plus 0.5em minus
  0.4em\relax {ACM}, 2021, pp. 844--852.

\bibitem{wolpert2002optimal}
D.~H. Wolpert and K.~Tumer, ``Optimal payoff functions for members of
  collectives,'' in \emph{Modeling complexity in economic and social
  systems}.\hskip 1em plus 0.5em minus 0.4em\relax World Scientific, 2002, pp.
  355--369.

\bibitem{guestrin2002coordinated}
C.~Guestrin, M.~Lagoudakis, and R.~Parr, ``Coordinated reinforcement
  learning,'' in \emph{ICML}, vol.~2.\hskip 1em plus 0.5em minus 0.4em\relax
  Citeseer, 2002, pp. 227--234.

\bibitem{kok2005using}
J.~R. Kok and N.~Vlassis, ``Using the max-plus algorithm for multiagent
  decision making in coordination graphs,'' in \emph{Robot Soccer World
  Cup}.\hskip 1em plus 0.5em minus 0.4em\relax Springer, 2005, pp. 1--12.

\bibitem{nguyen2018credit}
D.~T. Nguyen, A.~Kumar, and H.~C. Lau, ``Credit assignment for collective
  multiagent rl with global rewards,'' in \emph{Advances in Neural Information
  Processing Systems}, 2018, pp. 8102--8113.

\bibitem{RashidSWFFW18}
T.~Rashid, M.~Samvelyan, C.~S. de~Witt, G.~Farquhar, J.~N. Foerster, and
  S.~Whiteson, ``{QMIX:} monotonic value function factorisation for deep
  multi-agent reinforcement learning,'' in \emph{Proceedings of the 35th
  International Conference on Machine Learning, {ICML} 2018,
  Stockholmsm{\"{a}}ssan, Stockholm, Sweden, July 10-15, 2018}, ser.
  Proceedings of Machine Learning Research, vol.~80.\hskip 1em plus 0.5em minus
  0.4em\relax {PMLR}, 2018, pp. 4292--4301.

\bibitem{SonKKHY19}
K.~Son, D.~Kim, W.~J. Kang, D.~Hostallero, and Y.~Yi, ``{QTRAN:} learning to
  factorize with transformation for cooperative multi-agent reinforcement
  learning,'' in \emph{Proceedings of the 36th International Conference on
  Machine Learning, {ICML} 2019, 9-15 June 2019, Long Beach, California,
  {USA}}, ser. Proceedings of Machine Learning Research, vol.~97.\hskip 1em
  plus 0.5em minus 0.4em\relax {PMLR}, 2019, pp. 5887--5896.

\bibitem{axelrod1981evolution}
R.~Axelrod and W.~D. Hamilton, ``The evolution of cooperation,''
  \emph{science}, vol. 211, no. 4489, pp. 1390--1396, 1981.

\bibitem{peleg2007introduction}
B.~Peleg and P.~Sudh{\"o}lter, \emph{Introduction to the theory of cooperative
  games}.\hskip 1em plus 0.5em minus 0.4em\relax Springer Science \& Business
  Media, 2007, vol.~34.

\bibitem{shapley1971cores}
L.~S. Shapley, ``Cores of convex games,'' \emph{International journal of game
  theory}, vol.~1, no.~1, pp. 11--26, 1971.

\bibitem{gan2013optimal}
L.~Gan, N.~Li, U.~Topcu, and S.~H. Low, ``Optimal power flow in tree
  networks,'' in \emph{52nd IEEE Conference on Decision and Control}.\hskip 1em
  plus 0.5em minus 0.4em\relax IEEE, 2013, pp. 2313--2318.

\bibitem{saadat1999power}
H.~Saadat \emph{et~al.}, \emph{Power system analysis}.\hskip 1em plus 0.5em
  minus 0.4em\relax McGraw-hill, 1999, vol.~2.

\bibitem{agalgaonkar2013distribution}
Y.~P. Agalgaonkar, B.~C. Pal, and R.~A. Jabr, ``Distribution voltage control
  considering the impact of pv generation on tap changers and autonomous
  regulators,'' \emph{IEEE Transactions on Power Systems}, vol.~29, no.~1, pp.
  182--192, 2013.

\bibitem{masters2002voltage}
C.~Masters, ``Voltage rise: the big issue when connecting embedded generation
  to long 11 kv overhead lines,'' \emph{Power Engineering Journal}, vol.~16,
  no.~1, pp. 5--12, 2002.

\bibitem{yang2015voltage}
G.~Yang, F.~Marra, M.~Juamperez, S.~B. Kjær, S.~Hashemi, J.~Østergaard, H.~H.
  Ipsen, and K.~H.~B. Frederiksen, ``Voltage rise mitigation for solar pv
  integration at lv grids studies from pvnet. dk,'' \emph{Journal of Modern
  Power Systems and Clean Energy}, vol.~3, no.~3, pp. 411--421, 2015.

\bibitem{senjyu2008optimal}
T.~Senjyu, Y.~Miyazato, A.~Yona, N.~Urasaki, and T.~Funabashi, ``Optimal
  distribution voltage control and coordination with distributed generation,''
  \emph{IEEE Transactions on Power Delivery}, vol.~23, no.~2, pp. 1236--1242,
  2008.

\bibitem{singhal2019real}
A.~Singhal, V.~Ajjarapu, J.~Fuller, and J.~Hansen, ``Real-time local volt/var
  control under external disturbances with high pv penetration,'' \emph{IEEE
  Transactions on Smart Grid}, vol.~10, no.~4, pp. 3849--3859, 2019.

\bibitem{agalgaonkar2014distribution}
Y.~P. Agalgaonkar, B.~C. Pal, and R.~A. Jabr, ``Distribution voltage control
  considering the impact of pv generation on tap changers and autonomous
  regulators,'' \emph{IEEE Transactions on Power Systems}, vol.~29, no.~1, pp.
  182--192, 2014.

\bibitem{ieee2018ieee}
D.~G. Photovoltaics and E.~Storage, ``Ieee standard for interconnection and
  interoperability of distributed energy resources with associated electric
  power systems interfaces,'' \emph{IEEE Std 1547-2018 (Revision of IEEE Std
  1547-2003)}, pp. 1--138, 2018.

\bibitem{varma2018pv}
R.~K. Varma and E.~M. Siavashi, ``Pv-statcom: A new smart inverter for voltage
  control in distribution systems,'' \emph{IEEE Transactions on Sustainable
  Energy}, vol.~9, no.~4, pp. 1681--1691, 2018.

\bibitem{varma2009night}
R.~K. Varma, V.~Khadkikar, and R.~Seethapathy, ``Nighttime application of pv
  solar farm as statcom to regulate grid voltage,'' \emph{IEEE Transactions on
  Energy Conversion}, vol.~24, no.~4, pp. 983--985, 2009.

\bibitem{gomez2018electric}
A.~G{\'o}mez-Exp{\'o}sito, A.~J. Conejo, and C.~Ca{\~n}izares, \emph{Electric
  energy systems: analysis and operation}.\hskip 1em plus 0.5em minus
  0.4em\relax CRC press, 2018.

\bibitem{fusco2021decentralized}
G.~Fusco and M.~Russo, ``A decentralized approach for voltage control by
  multiple distributed energy resources,'' \emph{IEEE Transactions on Smart
  Grid}, pp. 1--1, 2021.

\bibitem{turistsyn2011options}
K.~Turitsyn, P.~Sulc, S.~Backhaus, and M.~Chertkov, ``Options for control of
  reactive power by distributed photovoltaic generators,'' \emph{Proceedings of
  the IEEE}, vol.~99, no.~6, pp. 1063--1073, 2011.

\bibitem{vovos2007centralized}
P.~N. Vovos, A.~E. Kiprakis, A.~R. Wallace, and G.~P. Harrison, ``Centralized
  and distributed voltage control: Impact on distributed generation
  penetration,'' \emph{IEEE Transactions on Power Systems}, vol.~22, no.~1, pp.
  476--483, 2007.

\bibitem{jahangiri2013distributed}
P.~Jahangiri and D.~C. Aliprantis, ``Distributed volt/var control by pv
  inverters,'' \emph{IEEE Transactions on Power Systems}, vol.~28, no.~3, pp.
  3429--3439, 2013.

\bibitem{schiffer2016voltage}
J.~Schiffer, T.~Seel, J.~Raisch, and T.~Sezi, ``Voltage stability and reactive
  power sharing in inverter-based microgrids with consensus-based distributed
  voltage control,'' \emph{IEEE Transactions on Control Systems Technology},
  vol.~24, no.~1, pp. 96--109, 2016.

\bibitem{xu2017multi}
Y.~Xu, Z.~Y. Dong, R.~Zhang, and D.~J. Hill, ``Multi-timescale coordinated
  voltage/var control of high renewable-penetrated distribution systems,''
  \emph{IEEE Transactions on Power Systems}, vol.~32, no.~6, pp. 4398--4408,
  2017.

\bibitem{anese2014optimal}
E.~Dall’Anese, S.~V. Dhople, and G.~B. Giannakis, ``Optimal dispatch of
  photovoltaic inverters in residential distribution systems,'' \emph{IEEE
  Transactions on Sustainable Energy}, vol.~5, no.~2, pp. 487--497, 2014.

\bibitem{zheng2016fully}
W.~Zheng, W.~Wu, B.~Zhang, H.~Sun, and Y.~Liu, ``A fully distributed reactive
  power optimization and control method for active distribution networks,''
  \emph{IEEE Transactions on Smart Grid}, vol.~7, no.~2, pp. 1021--1033, 2016.

\bibitem{tang2020distributed}
Z.~Tang, D.~J. Hill, and T.~Liu, ``Distributed coordinated reactive power
  control for voltage regulation in distribution networks,'' \emph{IEEE
  Transactions on Smart Grid}, vol.~12, no.~1, pp. 312--323, 2020.

\bibitem{sun2019review}
H.~Sun, Q.~Guo, J.~Qi, V.~Ajjarapu, R.~Bravo, J.~Chow, Z.~Li, R.~Moghe,
  E.~Nasr-Azadani, U.~Tamrakar, G.~N. Taranto, R.~Tonkoski, G.~Valverde, Q.~Wu,
  and G.~Yang, ``Review of challenges and research opportunities for voltage
  control in smart grids,'' \emph{IEEE Transactions on Power Systems}, vol.~34,
  no.~4, pp. 2790--2801, 2019.

\bibitem{zeraati2018voltage}
M.~Zeraati, M.~E.~H. Golshan, and J.~M. Guerrero, ``Voltage quality improvement
  in low voltage distribution networks using reactive power capability of
  single-phase pv inverters,'' \emph{IEEE transactions on smart grid}, vol.~10,
  no.~5, pp. 5057--5065, 2018.

\bibitem{cao2020multi}
D.~Cao, W.~Hu, J.~Zhao, Q.~Huang, Z.~Chen, and F.~Blaabjerg, ``A multi-agent
  deep reinforcement learning based voltage regulation using coordinated pv
  inverters,'' \emph{IEEE Transactions on Power Systems}, vol.~35, no.~5, pp.
  4120--4123, 2020.

\bibitem{liu2021online}
H.~Liu and W.~Wu, ``Online multi-agent reinforcement learning for decentralized
  inverter-based volt-var control,'' \emph{IEEE Transactions on Smart Grid},
  2021.

\bibitem{wang2020data}
S.~Wang, J.~Duan, D.~Shi, C.~Xu, H.~Li, R.~Diao, and Z.~Wang, ``A data-driven
  multi-agent autonomous voltage control framework using deep reinforcement
  learning,'' \emph{IEEE Transactions on Power Systems}, vol.~35, no.~6, pp.
  4644--4654, 2020.

\bibitem{cao2020distributed}
D.~Cao, J.~Zhao, W.~Hu, F.~Ding, Q.~Huang, and Z.~Chen, ``Distributed voltage
  regulation of active distribution system based on enhanced multi-agent deep
  reinforcement learning,'' \emph{arXiv preprint arXiv:2006.00546}, 2020.

\bibitem{cao2021data}
D.~Cao, J.~Zhao, W.~Hu, F.~Ding, Q.~Huang, Z.~Chen, and F.~Blaabjerg,
  ``Data-driven multi-agent deep reinforcement learning for distribution system
  decentralized voltage control with high penetration of pvs,'' \emph{IEEE
  Transactions on Smart Grid}, 2021.

\bibitem{burger2019restructuring1}
S.~P. Burger, J.~D. Jenkins, C.~Batlle, and I.~J. P{\'e}rez-Arriaga,
  ``Restructuring revisited part 1: Competition in electricity distribution
  systems,'' \emph{The Energy Journal}, vol.~40, no.~3, 2019.

\bibitem{burger2019restructuring2}
S.~P. Burger, J.~D. Jenkins, C.~Batlle, and I.~J. Perez-Arriaga,
  ``Restructuring revisited part 2: coordination in electricity distribution
  systems,'' \emph{The Energy Journal}, vol.~40, no.~3, 2019.

\bibitem{oliehoek2016concise}
F.~A. Oliehoek and C.~Amato, \emph{A concise introduction to decentralized
  POMDPs}.\hskip 1em plus 0.5em minus 0.4em\relax Springer, 2016.

\bibitem{wang2020towards}
J.~Wang, Z.~Ren, B.~Han, J.~Ye, and C.~Zhang, ``Towards understanding linear
  value decomposition in cooperative multi-agent q-learning,'' \emph{arXiv
  preprint arXiv:2006.00587}, 2020.

\bibitem{banach1922operations}
S.~Banach, ``Sur les op{\'e}rations dans les ensembles abstraits et leur
  application aux {\'e}quations int{\'e}grales,'' \emph{Fund. math}, vol.~3,
  no.~1, pp. 133--181, 1922.

\bibitem{jaakkola1994convergence}
T.~Jaakkola, M.~I. Jordan, and S.~P. Singh, ``On the convergence of stochastic
  iterative dynamic programming algorithms,'' \emph{Neural computation},
  vol.~6, no.~6, pp. 1185--1201, 1994.

\bibitem{oliehoek2012decentralized}
F.~A. Oliehoek, ``Decentralized pomdps,'' in \emph{Reinforcement
  Learning}.\hskip 1em plus 0.5em minus 0.4em\relax Springer, 2012, pp.
  471--503.

\bibitem{oliehoek2008optimal}
F.~A. Oliehoek, M.~T. Spaan, and N.~Vlassis, ``Optimal and approximate q-value
  functions for decentralized pomdps,'' \emph{Journal of Artificial
  Intelligence Research}, vol.~32, pp. 289--353, 2008.

\bibitem{chung2014empirical}
J.~Chung, C.~Gulcehre, K.~Cho, and Y.~Bengio, ``Empirical evaluation of gated
  recurrent neural networks on sequence modeling,'' \emph{arXiv preprint
  arXiv:1412.3555}, 2014.

\bibitem{ernst2005tree}
D.~Ernst, P.~Geurts, and L.~Wehenkel, ``Tree-based batch mode reinforcement
  learning,'' \emph{Journal of Machine Learning Research}, vol.~6, pp.
  503--556, 2005.

\bibitem{Wang_2020}
J.~Wang, Y.~Zhang, T.-K. Kim, and Y.~Gu, ``Shapley q-value: A local reward
  approach to solve global reward games,'' \emph{Proceedings of the AAAI
  Conference on Artificial Intelligence}, vol.~34, no.~05, p. 7285–7292, Apr
  2020.

\bibitem{claus1998dynamics}
C.~Claus and C.~Boutilier, ``The dynamics of reinforcement learning in
  cooperative multiagent systems,'' \emph{AAAI/IAAI}, vol. 1998, no. 746-752,
  p.~2, 1998.

\bibitem{PapoudakisC0A21}
G.~Papoudakis, F.~Christianos, L.~Sch{\"{a}}fer, and S.~V. Albrecht,
  ``Benchmarking multi-agent deep reinforcement learning algorithms in
  cooperative tasks,'' in \emph{Proceedings of the Neural Information
  Processing Systems Track on Datasets and Benchmarks 1, NeurIPS Datasets and
  Benchmarks 2021, December 2021, virtual}, J.~Vanschoren and S.~Yeung, Eds.,
  2021.

\bibitem{lauer2000algorithm}
M.~Lauer and M.~Riedmiller, ``An algorithm for distributed reinforcement
  learning in cooperative multi-agent systems,'' in \emph{In Proceedings of the
  Seventeenth International Conference on Machine Learning}.\hskip 1em plus
  0.5em minus 0.4em\relax Citeseer, 2000.

\bibitem{littman1995efficient}
M.~L. Littman, A.~R. Cassandra, and L.~P. Kaelbling, ``Efficient
  dynamic-programming updates in partially observable markov decision
  processes,'' Tech. Rep., 1995.

\bibitem{sondik1971optimal}
E.~J. Sondik, \emph{The optimal control of partially observable Markov
  processes}.\hskip 1em plus 0.5em minus 0.4em\relax Stanford University, 1971.

\bibitem{astrom1965optimal}
K.~J. Astrom, ``Optimal control of markov decision processes with incomplete
  state estimation,'' \emph{J. Math. Anal. Applic.}, vol.~10, pp. 174--205,
  1965.

\bibitem{aastrom1965optimal}
K.~J. {\AA}str{\"o}m, ``Optimal control of markov processes with incomplete
  state information,'' \emph{Journal of mathematical analysis and
  applications}, vol.~10, no.~1, pp. 174--205, 1965.

\bibitem{NIPS2016_6398}
S.~Sukhbaatar, a.~szlam, and R.~Fergus, ``Learning multiagent communication
  with backpropagation,'' in \emph{Advances in Neural Information Processing
  Systems 29}.\hskip 1em plus 0.5em minus 0.4em\relax Curran Associates, Inc.,
  2016, pp. 2244--2252.

\bibitem{kingma2014adam}
D.~P. Kingma and J.~Ba, ``Adam: A method for stochastic optimization,''
  \emph{arXiv preprint arXiv:1412.6980}, 2014.

\bibitem{das2018tarmac}
A.~Das, T.~Gervet, J.~Romoff, D.~Batra, D.~Parikh, M.~Rabbat, and J.~Pineau,
  ``Tarmac: Targeted multi-agent communication,'' \emph{arXiv preprint
  arXiv:1810.11187}, 2018.

\bibitem{patro2015normalization}
S.~Patro and K.~K. Sahu, ``Normalization: A preprocessing stage,'' \emph{arXiv
  preprint arXiv:1503.06462}, 2015.

\bibitem{pearson1895vii}
K.~Pearson, ``Vii. note on regression and inheritance in the case of two
  parents,'' \emph{proceedings of the royal society of London}, vol.~58, no.
  347-352, pp. 240--242, 1895.

\bibitem{bohmer2020deep}
W.~B{\"o}hmer, V.~Kurin, and S.~Whiteson, ``Deep coordination graphs,'' in
  \emph{International Conference on Machine Learning}.\hskip 1em plus 0.5em
  minus 0.4em\relax PMLR, 2020, pp. 980--991.

\bibitem{rashid2020weighted}
T.~Rashid, G.~Farquhar, B.~Peng, and S.~Whiteson, ``Weighted qmix: Expanding
  monotonic value function factorisation for deep multi-agent reinforcement
  learning,'' \emph{Advances in Neural Information Processing Systems},
  vol.~33, 2020.

\bibitem{iqbal2019actor}
S.~Iqbal and F.~Sha, ``Actor-attention-critic for multi-agent reinforcement
  learning,'' in \emph{International Conference on Machine Learning}.\hskip 1em
  plus 0.5em minus 0.4em\relax PMLR, 2019, pp. 2961--2970.

\bibitem{wang2020qplex}
J.~Wang, Z.~Ren, T.~Liu, Y.~Yu, and C.~Zhang, ``Qplex: Duplex dueling
  multi-agent q-learning,'' \emph{arXiv preprint arXiv:2008.01062}, 2020.

\bibitem{samvelyan2019starcraft}
M.~Samvelyan, T.~Rashid, C.~S. de~Witt, G.~Farquhar, N.~Nardelli, T.~G. Rudner,
  C.-M. Hung, P.~H. Torr, J.~Foerster, and S.~Whiteson, ``The starcraft
  multi-agent challenge,'' \emph{arXiv preprint arXiv:1902.04043}, 2019.

\bibitem{wei2016lenient}
E.~Wei and S.~Luke, ``Lenient learning in independent-learner stochastic
  cooperative games,'' \emph{The Journal of Machine Learning Research},
  vol.~17, no.~1, pp. 2914--2955, 2016.

\bibitem{wei2018multiagent}
E.~Wei, D.~Wicke, D.~Freelan, and S.~Luke, ``Multiagent soft q-learning,'' in
  \emph{2018 AAAI Spring Symposium Series}, 2018.

\bibitem{Byrd2012}
R.~H. Byrd, G.~M. Chin, J.~Nocedal, and Y.~Wu, ``Sample size selection in
  optimization methods for machine learning,'' \emph{Mathematical Programming},
  vol. 134, no.~1, pp. 127--155, 2012.

\bibitem{NIPS2015_effc299a}
T.~Hofmann, A.~Lucchi, S.~Lacoste-Julien, and B.~McWilliams, ``Variance reduced
  stochastic gradient descent with neighbors,'' in \emph{Advances in Neural
  Information Processing Systems}, C.~Cortes, N.~Lawrence, D.~Lee, M.~Sugiyama,
  and R.~Garnett, Eds., vol.~28.\hskip 1em plus 0.5em minus 0.4em\relax Curran
  Associates, Inc., 2015, pp. 2305--2313.

\bibitem{higgins2016beta}
I.~Higgins, L.~Matthey, A.~Pal, C.~P. Burgess, X.~Glorot, M.~M. Botvinick,
  S.~Mohamed, and A.~Lerchner, ``beta-vae: Learning basic visual concepts with
  a constrained variational framework,'' in \emph{5th International Conference
  on Learning Representations, {ICLR} 2017, Toulon, France, April 24-26, 2017,
  Conference Track Proceedings}.\hskip 1em plus 0.5em minus 0.4em\relax
  OpenReview.net, 2017.

\bibitem{baran1989network}
M.~Baran and F.~Wu, ``Network reconfiguration in distribution systems for loss
  reduction and load balancing,'' \emph{IEEE Transactions on Power Delivery},
  vol.~4, no.~2, pp. 1401--1407, 1989.

\bibitem{khodr2008maximum}
\BIBentryALTinterwordspacing
H.~Khodr, F.~Olsina, P.~D. O.-D. Jesus, and J.~Yusta, ``Maximum savings
  approach for location and sizing of capacitors in distribution systems,''
  \emph{Electric Power Systems Research}, vol.~78, no.~7, pp. 1192--1203, 2008.
  [Online]. Available:
  \url{https://www.sciencedirect.com/science/article/pii/S0378779607002143}
\BIBentrySTDinterwordspacing

\bibitem{zimmerman2010matpower}
R.~D. Zimmerman, C.~E. Murillo-S{\'a}nchez, and R.~J. Thomas, ``Matpower:
  Steady-state operations, planning, and analysis tools for power systems
  research and education,'' \emph{IEEE Transactions on power systems}, vol.~26,
  no.~1, pp. 12--19, 2010.

\bibitem{meinecke20simbench}
\BIBentryALTinterwordspacing
S.~Meinecke, D.~Sarajlić, S.~R. Drauz, A.~Klettke, L.-P. Lauven, C.~Rehtanz,
  A.~Moser, and M.~Braun, ``Simbench—a benchmark dataset of electric power
  systems to compare innovative solutions based on power flow analysis,''
  \emph{Energies}, vol.~13, no.~12, 2020. [Online]. Available:
  \url{https://www.mdpi.com/1996-1073/13/12/3290}
\BIBentrySTDinterwordspacing

\bibitem{LoweWTHAM17}
R.~Lowe, Y.~Wu, A.~Tamar, J.~Harb, P.~Abbeel, and I.~Mordatch, ``Multi-agent
  actor-critic for mixed cooperative-competitive environments,'' in
  \emph{Advances in Neural Information Processing Systems 30: Annual Conference
  on Neural Information Processing Systems 2017, December 4-9, 2017, Long
  Beach, CA, {USA}}, I.~Guyon, U.~von Luxburg, S.~Bengio, H.~M. Wallach,
  R.~Fergus, S.~V.~N. Vishwanathan, and R.~Garnett, Eds., 2017, pp. 6379--6390.

\bibitem{de2020independent}
C.~S. de~Witt, T.~Gupta, D.~Makoviichuk, V.~Makoviychuk, P.~H. Torr, M.~Sun,
  and S.~Whiteson, ``Is independent learning all you need in the starcraft
  multi-agent challenge?'' \emph{arXiv preprint arXiv:2011.09533}, 2020.

\bibitem{yu2021surprising}
C.~Yu, A.~Velu, E.~Vinitsky, Y.~Wang, A.~Bayen, and Y.~Wu, ``The surprising
  effectiveness of mappo in cooperative, multi-agent games,'' \emph{arXiv
  preprint arXiv:2103.01955}, 2021.

\bibitem{ackermann2019reducing}
J.~Ackermann, V.~Gabler, T.~Osa, and M.~Sugiyama, ``Reducing overestimation
  bias in multi-agent domains using double centralized critics,'' \emph{arXiv
  preprint arXiv:1910.01465}, 2019.

\bibitem{MajzoubiZCK0S20}
M.~Majzoubi, C.~Zhang, R.~Chari, A.~Krishnamurthy, J.~Langford, and
  A.~Slivkins, ``Efficient contextual bandits with continuous actions,'' in
  \emph{Advances in Neural Information Processing Systems 33: Annual Conference
  on Neural Information Processing Systems 2020, NeurIPS 2020, December 6-12,
  2020, virtual}, H.~Larochelle, M.~Ranzato, R.~Hadsell, M.~Balcan, and H.~Lin,
  Eds., 2020.

\bibitem{pearl2009causality}
J.~Pearl, \emph{Causality}.\hskip 1em plus 0.5em minus 0.4em\relax Cambridge
  university press, 2009.

\bibitem{wang2022shaq}
\BIBentryALTinterwordspacing
J.~Wang, Y.~Zhang, Y.~Gu, and T.-K. Kim, ``Shaq: Incorporating shapley value
  theory into multi-agent q-learning,'' 2021. [Online]. Available:
  \url{https://arxiv.org/abs/2105.15013}
\BIBentrySTDinterwordspacing

\bibitem{dales2003introduction}
H.~G. Dales, H.~G. Dales, P.~Aiena, J.~Eschmeier, K.~Laursen, and G.~A. Willis,
  \emph{Introduction to Banach algebras, operators, and harmonic
  analysis}.\hskip 1em plus 0.5em minus 0.4em\relax Cambridge University Press,
  2003, vol.~57.

\bibitem{smallwood1973optimal}
R.~D. Smallwood and E.~J. Sondik, ``The optimal control of partially observable
  markov processes over a finite horizon,'' \emph{Operations research},
  vol.~21, no.~5, pp. 1071--1088, 1973.

\bibitem{ha2017hyper}
D.~Ha, A.~M. Dai, and Q.~V. Le, ``Hypernetworks,'' in \emph{5th International
  Conference on Learning Representations, {ICLR} 2017, Toulon, France, April
  24-26, 2017, Conference Track Proceedings}.\hskip 1em plus 0.5em minus
  0.4em\relax OpenReview.net, 2017.

\bibitem{goodfellow2014generative}
I.~J. Goodfellow, J.~Pouget-Abadie, M.~Mirza, B.~Xu, D.~Warde-Farley, S.~Ozair,
  A.~Courville, and Y.~Bengio, ``Generative adversarial networks,'' \emph{arXiv
  preprint arXiv:1406.2661}, 2014.

\bibitem{tieleman2012lecture}
T.~Tieleman and G.~Hinton, ``Lecture 6.5-rmsprop: Divide the gradient by a
  running average of its recent magnitude,'' \emph{COURSERA: Neural networks
  for machine learning}, vol.~4, no.~2, pp. 26--31, 2012.

\bibitem{fujimoto2018addressing}
S.~Fujimoto, H.~Hoof, and D.~Meger, ``Addressing function approximation error
  in actor-critic methods,'' in \emph{International Conference on Machine
  Learning}.\hskip 1em plus 0.5em minus 0.4em\relax PMLR, 2018, pp. 1587--1596.

\bibitem{ba2016layer}
J.~L. Ba, J.~R. Kiros, and G.~E. Hinton, ``Layer normalization,'' \emph{arXiv
  preprint arXiv:1607.06450}, 2016.

\bibitem{engstrom2020implementation}
L.~Engstrom, A.~Ilyas, S.~Santurkar, D.~Tsipras, F.~Janoos, L.~Rudolph, and
  A.~Madry, ``Implementation matters in deep policy gradients: A case study on
  ppo and trpo,'' \emph{arXiv preprint arXiv:2005.12729}, 2020.

\bibitem{chulmanMLJA15}
J.~Schulman, P.~Moritz, S.~Levine, M.~I. Jordan, and P.~Abbeel,
  ``High-dimensional continuous control using generalized advantage
  estimation,'' in \emph{4th International Conference on Learning
  Representations, {ICLR} 2016, San Juan, Puerto Rico, May 2-4, 2016,
  Conference Track Proceedings}, Y.~Bengio and Y.~LeCun, Eds., 2016.

\bibitem{saxe2013exact}
A.~M. Saxe, J.~L. McClelland, and S.~Ganguli, ``Exact solutions to the
  nonlinear dynamics of learning in deep linear neural networks,'' \emph{arXiv
  preprint arXiv:1312.6120}, 2013.

\bibitem{thurner2018pandapower}
L.~Thurner, A.~Scheidler, F.~Schäfer, J.-H. Menke, J.~Dollichon, F.~Meier,
  S.~Meinecke, and M.~Braun, ``Pandapower—an open-source python tool for
  convenient modeling, analysis, and optimization of electric power systems,''
  \emph{IEEE Transactions on Power Systems}, vol.~33, no.~6, pp. 6510--6521,
  2018.

\bibitem{chai2018network}
Y.~Chai, L.~Guo, C.~Wang, Z.~Zhao, X.~Du, and J.~Pan, ``Network partition and
  voltage coordination control for distribution networks with high penetration
  of distributed pv units,'' \emph{IEEE Transactions on Power Systems},
  vol.~33, no.~3, pp. 3396--3407, 2018.

\end{thebibliography}

\end{document}